\documentclass[aip,superscriptaddress,nofootinbib,onecolumn,10pt]{revtex4-2} 

\usepackage{soul}
\usepackage[pdftex]{graphicx}
\usepackage{dcolumn}
\usepackage{bm}
\usepackage{epsfig}
\usepackage{amsmath}
\usepackage{amssymb} 
\usepackage{amsfonts}
\usepackage{pifont} 
\usepackage{multirow}
\usepackage{subfigure}
\usepackage{float}
\usepackage[utf8]{inputenc} 
\usepackage{epstopdf}
\usepackage{dsfont}
\usepackage{MnSymbol}
\usepackage{xspace}
\usepackage[most]{tcolorbox}
\tcbuselibrary{breakable}
\usepackage{multirow}
\usepackage{enumitem}
\usepackage{csquotes}
\usepackage{tabularray}
\usepackage{booktabs}
\usepackage[linesnumbered,lined,boxed,commentsnumbered,ruled,vlined,nofillcomment]{algorithm2e}
\usepackage{dashbox}
\usepackage[dvipsnames,svgnames]{xcolor}

\newcommand{\CYAN}[1]{\textcolor{cyan}{#1}}
\newcommand{\MAG}[1]{\textcolor{magenta}{#1}}

\usepackage{amsthm}
\newtheorem{theorem}{Theorem}
\newtheorem*{theorem*}{Theorem}
\newtheorem{lemma}[theorem]{Lemma}
\newtheorem*{lemma*}{Lemma}
\newtheorem{conjecture}[theorem]{Conjecture}
\newtheorem{corollary}[theorem]{Corollary}
\newtheorem*{corollary*}{Corollary}
\newtheorem{proposition}[theorem]{Proposition}
\newtheorem*{proposition*}{Proposition}
\newtheorem{definition}[theorem]{Definition}
\newtheorem*{definition*}{Definition}
\newtheorem{example}[theorem]{Example}
\newtheorem*{example*}{Example}



\newtcbtheorem[]{tcexample}{Example}{breakable, colback=orange!3!white,colframe=orange!85!black, fonttitle=\bfseries}{Exa}

\DeclareMathOperator{\spn}{span}

\newcommand{\g}{\mathfrak{g}}
\newcommand{\gh}{\mathfrak{h}}
\newcommand{\gr}{\mathfrak{r}}
\newcommand{\gn}{\mathfrak{n}}

\newcommand{\gs}{\mathfrak{s}}
\newcommand{\gd}{\mathfrak{d}}
\newcommand{\sch}{\mathfrak{sch}}
\newcommand{\aff}{\mathfrak{aff}}

\newcommand{\sso}[1]{\mathfrak{so}(#1)}
\def\C{\mathbb{C}}
\def\R{\mathbb{R}}
\newcommand{\su}[1]{\mathfrak{su}(#1)}
\newcommand{\slR}[1]{\mathfrak{sl}(#1,\R)}
\newcommand{\slC}[1]{\mathfrak{sl}(#1,\C)}
\newcommand{\spR}[1]{\mathfrak{sp}(#1,\R)}
\newcommand{\wh}{\mathfrak{wh}}
\def\Z{\mathbb{Z}}
\def\N{\mathbb{N}}

\def\dagg{\textup{$\dagger$}}

\newcommand{\Np}[1]{\N^{#1}_{\geq 0}}
\newcommand{\x}[2]{\chi_{#1}(#2)}

\newcommand{\chimap}[1]{\chi({#1})}

\newcommand{\mdeg}{\mathrm{mdeg}}

\newcommand{\ket}[1]{\ensuremath{| #1 \rangle}{}}

\newcommand{\lie}[1]{\langle #1 \rangle_{\mathrm{Lie}}{}}

\newcommand{\PP}[2]{\mathcal{P}_{#2}(#1)}

\newcommand{\Tp}{\mathrm{Tp}}

\newcommand{\plainrefs}[1]{{\mbox{\textup{[\!\!\citenum{#1}]}}}}

\usepackage{bbm}

\newcommand{\upto}[1]{\simeq_{#1}}

\newcommand{\graphlegend}[1]
{\begin{tikzpicture}[scale=0.75, every node/.style={transform shape}]
		\def\axy{0.3} 
		\def\bxy{0.325}   
		\coordinate (x) at (-2,0);
		\coordinate (z) at (2,0);	
		\node[ellipse, minimum width=0.6cm, minimum height=0.65cm, very thick,  draw = #1] (E2) at (x) {\large$x$};
		\node[ellipse, minimum width=0.6cm, minimum height=0.65cm, very thick, draw = #1] (E1) at (z) {\large$z$};
		\path (x) ++({\axy*cos(0)}, {\bxy*sin(0)}) coordinate (P1);
		\path (z) ++({\axy*cos(-180)}, {\bxy*sin(-180)}) coordinate (P2);
		\draw[-{Latex[length=3mm]}, very thick, draw=#1] (P1)-- (P2);
		\node at (0,0.3) {\large$y$};
		\node at (4,0) {\large$\hat{=}$};
		\node at (8.5,0) {\large$[x,y]=\kappa z$\quad with $\kappa\in\R\setminus\{0\}$.};
	\end{tikzpicture}
}

\usepackage[bookmarks=false,pdfstartview={FitH}]{hyperref}

\usepackage{tikz,xcolor}
\usetikzlibrary{patterns,patterns.meta,decorations.pathreplacing,angles,quotes, plotmarks, arrows.meta, calc, hobby}
\usetikzlibrary{shapes.geometric}
\usetikzlibrary{cd, babel, math} 
\definecolor{lime}{HTML}{A6CE39}
\DeclareRobustCommand{\orcidicon}{%
	\begin{tikzpicture}
		\draw[lime, fill=lime] (0,0) 
		circle [radius=0.16] 
		node[white] {{\fontfamily{qag}\selectfont \tiny ID}};
		\draw[white, fill=white] (-0.0625,0.095) 
		circle [radius=0.007];
	\end{tikzpicture}
	\hspace{-2mm}
}
\foreach \x in {A, ..., Z}{%
	\expandafter\xdef\csname orcid\x\endcsname{\noexpand\href{https://orcid.org/\csname orcidauthor\x\endcsname}{\noexpand\orcidicon}}
}

\begin{document}

	\title[Finite-dimensional Lie algebras in bosonic quantum dynamics: The single-mode case]{Finite-dimensional Lie algebras in bosonic quantum dynamics: The single-mode case}
	
	\author{Tim Heib\,\orcidC{}}
	\email{t.heib@fz-juelich.de}	
	\affiliation{Institute for Quantum Computing Analytics (PGI-12), Forschungszentrum J\"ulich, 52425 J\"ulich, Germany}
	\affiliation{Theoretical Physics, Universit\"at des Saarlandes, 66123 Saarbr\"ucken, Germany}
	\author{Andreea Silvia Goia\orcidE{}}
    \thanks{Current affiliation: Ludwig-Maximilians-Universität München, Theresienstrasse 37, 80333 München, Germany; Technische Universität München,  85748 Garching, Germany.}
	\affiliation{Institute for Quantum Computing Analytics (PGI-12), Forschungszentrum J\"ulich, 52425 J\"ulich, Germany}
	\affiliation{\mbox{University of Bucharest, Faculty of Physics, Măgurele, PO BOX MG11, 077125, Romania}}
	\author{Sona Baghiyan\orcidD{}}
    \thanks{Current affiliation: Physics Department, University of Wisconsin-Madison, Madison, Wisconsin, 53706-1390, USA.}
	\affiliation{Institute for Quantum Computing Analytics (PGI-12), Forschungszentrum J\"ulich, 52425 J\"ulich, Germany}
	\affiliation{St. Olaf College, Northfield, Minnesota, 55057, USA}
	\author{Robert Zeier\,\orcidA{}}
	\email{r.zeier@fz-juelich.de}
	\affiliation{Forschungszentrum Jülich GmbH, Peter Grünberg Institute, Quantum Control (PGI-8), 52425 Jülich, Germany}
	\author{David Edward Bruschi\,\orcidB{}}
	\email{david.edward.bruschi@posteo.net}	
	\email{d.e.bruschi@fz-juelich.de}	
	\affiliation{Institute for Quantum Computing Analytics (PGI-12), Forschungszentrum J\"ulich, 52425 J\"ulich, Germany}
	\affiliation{Theoretical Physics, Universit\"at des Saarlandes, 66123 Saarbr\"ucken, Germany}
	
	\begin{abstract}
		We study, classify, and explore the mathematical properties of finite-dimensional Lie algebras occurring in the quantum dynamics of single-mode and self-interacting bosonic systems.
		These Lie algebras are contained in the real skew-hermitian Weyl algebra $\hat{A}_1$, 
		defined as the real subalgebra of the Weyl algebra $A_1$ consisting of all skew-hermitian polynomials.
		A central aspect of our analysis is the choice of basis for $\hat{A}_1$, which is composed of skew-symmetric combinations of two elements of the Weyl algebra called monomials, namely strings of creation and annihilation operators combined with their hermitian conjugate.
		Motivated by the quest for analytical solutions in quantum optimal control and dynamics, we aim at answering the following three fundamental questions: (i) \emph{What are the finite-dimensional Lie subalgebras in $\hat{A}_1$ generated by monomials alone}?
		(ii)~\emph{What are the finite-dimensional Lie subalgebras in $\hat{A}_1$ that contain the free Hamiltonian}? (iii) \emph{What are the non-abelian and finite-dimensional Lie subalgebras that can be faithfully realized in $\hat{A}_1$}?
		We answer the first question by providing all possible realizations of all finite-dimensional non-abelian Lie algebras that are generated by monomials alone. We answer the second question by proving that any non-abelian and finite-dimensional subalgebra of $\hat{A}_1$ that contains a free Hamiltonian term must be a subalgebra of the Schr\"odinger algebra. We partially answer the third question by classifying all nilpotent and non-solvable Lie algebras that can be realized in $\hat{A}_1$, and comment on the remaining cases. 
		Finally, we also discuss the implications of our results for quantum control theory.
		Our work constitutes an important stepping stone to understanding
		quantum dynamics of bosonic systems in full generality.
	\end{abstract}

	\date{\today}

	\maketitle
	\thispagestyle{empty}
	
	\newpage
	
	\section*{Introduction\label{sec:introduction}}
	Completely determining the quantum dynamics of a physical system is key to fully understanding natural processes. In those cases where sufficient control can be achieved, it is then possible to steer the state of the system, starting from an initially given one, without the need for additional intermediate measurements.
	In general, an idealized setting is assumed where the state of a physical system---encoding all information of interest at any given time---is ultimately determined by a set of differential equations. Obtaining the analytical solution to such equations, which govern the time evolution, would, in principle, yield a complete understanding of the system's dynamics. However, in the vast majority of cases, analytical solutions to the full dynamics cannot be found, and this is true even for relatively simple setups. Examples include driven two-level quantum systems \plainrefs{Barnes:2012} and problems in many-body physics, such as interaction-driven metal-insulator transitions, superconductivity, or magnetism \plainrefs{LeBlanc:2015}, to mention a few. This limitation is not only confined to problems concerning quantum dynamics, but also extends to other areas of physics, such as the three-body problem \plainrefs{Musielak:2014} in classical mechanics or binary black hole coalescence \plainrefs{Pretorius:2009} in classical gravitational physics, thereby highlighting the fundamental challenge in fully understanding
	the dynamics of physical systems. In virtually all cases where a full solution is outstanding, one relies on numerical methods, such as Monte Carlo methods \plainrefs{Yokoyama:1987,Shi:2013,tenHaaf:1995}, dynamical mean-field methods \plainrefs{Aoki:2014,Anders:2011}, the density matrix embedding method \plainrefs{Knizia:2012}, Runge-Kutta methods \plainrefs{Blanes:2002}, or truncated Wigner approximations \plainrefs{Blakie:2008} to name a few. These are often used in conjunction with analytical methods such as considering the high-temperature limit \plainrefs{LeBlanc:2015}, or the limit of infinite dimensions \plainrefs{Metzner:1989,Georges:1996}, and yield approximate solutions. 
	Regardless of the approach taken, a complete predictive understanding of the system's evolution remains elusive.
	This fundamental gap naturally leads to the foundational question: \emph{Under which conditions can an analytical solution to the dynamics of a quantum system be, at least in principle, obtained?} 
	
	In quantum mechanics, the differential equation that determines the unitary time evolution of a quantum state $\rho(t)$ is the time-dependent von Neumann equation 
	\begin{align}
		\dot\rho(t)=\frac{\mathrm{d}}{\mathrm{d}t}\rho(t)=i [H(t),\rho(t)],
	\end{align}
	where $H(t)$ is the hermitian Hamiltonian operator encoding the specific dynamics of a given system. In this work, we use natural units with $\hbar=1$, unless explicitly stated. The solution to this equation is $\rho(t)=U(t)\rho(0)U^\dagg(t)$, where $\rho(0)$ is the initial state, and the problem is then moved to finding the unitary time evolution operator $U(t)$, which is a one-parameter family of operators determined by the differential equation $(\mathrm{d}/\mathrm{d}t)\, U(t)=-iH(t)U(t)$. Its formal solution is given by the time-ordered exponential
	\begin{align}\label{time:evolution:general:solution}
		U(t)=\overset{\leftarrow}{\mathcal{T}}\exp\left[-i\int_0^t H(\tau)\mathrm{d}\tau\right],
	\end{align}
	where $\overset{\leftarrow}{\mathcal{T}}$ is the time-ordering operator \plainrefs{Dyson:1949}.
	
	The expression \eqref{time:evolution:general:solution}, although formally correct, is almost always non-conductive to obtaining an exact closed-form expression for $\rho(t)$. This occurs because there is no operational way to apply this generic solution to an initial state $\rho(0)$ and explicitly obtain, for example, the statistical moments of such a state. Instead, one often needs to employ some approximate form of \eqref{time:evolution:general:solution} that can be concretely applied on the initial states and therefore allow for the computation of any quantity of interest. An example for obtaining an expression for the unitary time evolution operator is the Magnus expansion, which provides a formal series solution \plainrefs{Blanes:2009}. However, one often lacks analytical control over the series and resorts to truncating it and only computing the first few terms \plainrefs{Kirchhoff:2025,Blanes:2009}. 
	
	A powerful alternative approach has been pioneered by Wei and Norman \plainrefs{Wei:Norman:1963,Wei:Norman:1964}, which is nevertheless restricted to finite dimensions. This approach is based on the concept of \emph{factorization}, that is, on the ability to write the time evolution operator $U(t)=\prod_n U_n(t)$ as a product of unitary operators $U_n(t)$ each of which can be potentially applied in a direct fashion to the initial state of interest. The idea is that the operators $U_n(t)$  are determined by time-independent hermitian operators $G_n$ that are elements of a finite-dimensional Lie algebra \plainrefs{Wei:Norman:1963,Wei:Norman:1964}. This approach has generated renewed interest in the past years and has sparked new efforts in finding exact solutions to quantum dynamics of quantum systems, such as in the context of optomechanics \plainrefs{Qvarfort:Serafini:2020,Schneiter:Qvarfort:2020,martíneztibaduiza:2025}, and more general approaches \plainrefs{Qvarfort:2025,Bruschi:Lee:2013}. In particular, a novel direction has been recently proposed to tackle the problem as a whole by classifying finite-dimensional subalgebras of the so-called skew-hermitian Weyl algebra of skew-hermitian bosonic operators, providing the first set of general constraints on the elements of a Hamiltonian allowing for finite factorization \plainrefs{Bruschi:Xuereb:2024}. This has opened a novel direction in the broad area of dynamical systems, with promising implications.
	
	In this work we move on from the foundational work on Hamiltonian Lie algebras proposed in \plainrefs{Bruschi:Xuereb:2024}, and take a first step by focussing on algebraic aspects of single-mode bosonic systems. In general, bosonic systems (or, equivalently, systems of quantum harmonic oscillators) are described by non-commutative operators that are elements of the infinite-dimensional complex associative Weyl algebra $A_n$ that consists of all polynomials in the creation and annihilation operators of $n$ modes \plainrefs{Bruschi:Xuereb:2024}. We are interested in the single-mode case where $n=1$, and we note that a complete classification of all finite-dimensional algebras that can be realized in the complex algebra $A_1$ has recently been achieved \plainrefs{Tanasa:2005,TST:2006}. However, motivated by the hermitian nature of physical Hamiltonian operators we shift our focus to the real skew-hermitian Weyl algebra $\hat{A}_1$ and tackle the following key question: 
	\begin{quote}
		\textbf{Q}: \textit{Which Hamiltonians of one bosonic mode admit a finite-dimensional Lie algebra?}
	\end{quote}
	Identifying such Hamiltonians enables the application of the Wei-Norman method for all such cases, which provides in general local solutions in time. In the specific case of solvable algebras it also provides the exact analytical solution to the corresponding dynamics in terms of a product of simple operations for all times. Our study thereby contributes to the broader question of finding physical systems that admit closed-form (factorized) solution for their time evolution (as we will discuss later on). 
	
	The guiding principle of this work is the problem of quantum optimal control \plainrefs{Boscain:2021,Huang:1983}, a field gaining increasing attention due to recent advances in quantum information science \plainrefs{Google:2024,Werninghaus:2021,Kim:2023}.   
	The general idea of optimal control is the following: an initial state $\rho(0)$ is given
	together with a set of time-dependent control fields $u_j(t)$ that can be arbitrarily externally manipulated, and one aims at designing specific control fields in order to steer the system toward a desired state with high fidelity and robustness against noise. Our work can contribute towards this field by providing a deeper understanding of the relationship between controllability and the finiteness of subalgebras of the Weyl algebra $\hat{A}_1$, in particular in those cases when the free Hamiltonian cannot be steered (i.e., it is part of the drift). We also discuss the nullity
	of all finite-dimensional Lie subalgebras of the Weyl algebra $\hat{A}_1$, which is 
	defined as the minimal number of elements necessary to generate a given Lie algebra. Better understanding this concept can help to reduce the number of necessary operations that an external agent requires in order to steer a given system, thereby reducing cost and increasing efficiency of the processes of interest. In addition, we also discuss the set of reachable states for a given initial state. We find that nonlinear elements are required, where nonlinear here means that the operators are at least cubic in the creation and annihilation operators. Here, linear dynamics are given by Hamiltonians that are quadratic in the creation and annihilation operators \plainrefs{Adesso:Ragy:2014}.
	
	In this work, our main objective is to provide a complete list of non-abelian and finite-dimensional subalgebras of the real skew-hermitian Weyl algebra $\hat{A}_1$ that are either:
	\begin{enumerate}[label = (\alph*)]
		\item generated by a finite set of monomials of genuine bosonic creation and annihilation operators $a^\dagg$ and $a$, respectively. These monomials form a basis of the skew-hermitian Weyl algebra $\hat{A}_1$ and were introduced in \plainrefs{Bruschi:Xuereb:2024} (see equation \eqref{eqn:def:monomials});
		\item generated by a set of polynomials that includes a free Hamiltonian term - an operator that is of the form $i\omega a^\dagg a+ic$ with $c \in \R$;
		\item nilpotent or non-solvable.
	\end{enumerate}
	This aim connects our work to the broader and ongoing mathematical effort of providing a complete classification of all finite-dimensional Lie algebras, including solvable ones. The starting point of this classification-programme was the complete classification of semisimple Lie algebras, which laid the groundwork for further structural analysis \plainrefs{Cartan:1894,Killing:erster:1888,Killing:zweiter:1888,Killing:dritter:1889,Killing:vierter:1890}. Building on these preliminary endeavours, the field has seen significant progress \plainrefs{Bianchi:1903,Levi:1905,Matltsev:1942}, but the full classification is still an outstanding problem \plainrefs{Qi:2019,DelBarco:2025, Oeh:2023}. This classification programme encompasses a wide variety of algebraic structures, including nilpotent, solvable, non-reductive, and generally non-semisimple algebras.  Recent progress has been achieved in classifying mostly lower-dimensional algebras, even though often under restrictive conditions \plainrefs{DelBarco:2025,DeGraaf:2007,Ancocha:2015,Choriyeva:2024}.
	
	Our work is organized as follows. In Section~\ref{section:scope:of:the:work} we outline the scope of the work and provide a more detailed motivation for our investigation. Section~\ref{section:background:formalism} introduces the skew-hermitian Weyl algebra $\hat{A}_1$, summarizing essential definitions, conventions, and notational choices adopted from \plainrefs{Bruschi:Xuereb:2024}. In Section~\ref{section:monomial:generated:algebras} we present a comprehensive glossary of all finite-dimensional and non-abelian Lie subalgberas of the skew-hermitian Weyl algebra that are generated by monomials,
	where we also introduce key concepts, such as so-called generic Commutator Chains that aid in exploring specific structures of subalgebras of $\hat{A}_1$. We provide a list of all non-abelian finite-dimensional Lie subalgebras that contain a free Hamiltonian in Section~\ref{section:lie:algebras:with:free:hamiltonian}, where we do not constraint the algebras to be generated by monomials alone, but allow for arbitrary polynomials to be in the initial generating set. In Section~\ref{section:classification:nilpotent:nonsolvable}, we classify all finite-dimensional  Lie algebras that are either nilpotent or non-solvable and can faithfully be realized in $\hat{A}_1$. Finally, we discuss solvable subalgebras of $\hat{A}_1$, the Igusa condition \plainrefs{Igusa1981}, and the consequences for quantum optimal control theory in Section~\ref{section:considerations}. Section~\ref{sec:conclusion} contains the discussion of our results, as well as the outline of potential future directions.

	\section{Scope of this work\label{section:scope:of:the:work}}
	The main aspiration of this work is to apply the concepts and ideas presented in the foundational work on this topic \plainrefs{Bruschi:Xuereb:2024} to the case of a single bosonic mode, or the quantum harmonic oscillator. We here shed more light on the concrete motivation for this study.
	
	We focus our interest on physical systems composed solely by quantum harmonic oscillators that are each characterized by a frequency $\omega$ together with a pair of creation and annihilation operators $a^\dagg$ and $a$ respectively, which satisfy the canonical commutation relation $[a,a^\dagg]=1$. These systems appear as key elements across diverse domains of physics, such as the modes of the electromagnetic field in quantum field theory \plainrefs{Srednicki:2007}, phononic excitations in lattices in condensed matter physics \plainrefs{Ashcroft:Mermin:1976}, polaritons in dielectrics \plainrefs{Huttner:1992}, plasmons in ionic crystals \plainrefs{Bozhevolnyi:2017}, photonic excitations in cavity quantum electrodynamics \plainrefs{Haroche:Raimond:2006}, as well as phononic and photonic excitations in optomechanical systems \plainrefs{Aspelmeyer:Kippenberg:2014}. Consequently, it is undeniable that gaining insight into the underlying algebraic structure given by the skew-hermitian Weyl algebra is of significant interest. The relevance is further underscored by the recent advancements not only in quantum information science \plainrefs{Adesso:Illuminati:2007:v2,Weedbrook:Pirandola:2012,Adesso:Ragy:2014}, where bosonic systems are crucial to many key applications, but even to novel developments in the interdisciplinary field of relativistic quantum information \plainrefs{Mann:Ralph:2012}.
	
	The dynamics of bosonic systems are generally governed by a Hamiltonian $H$ that encodes the relevant interactions. As outlined in the introduction, closed-form solutions to the time evolution can, in certain cases, be obtained by factorizing the unitary time evolution operator. This factorization approach will be examined in greater detail in the following subsection. Additionally, we will also address the classification of finite-dimensional Lie algebras that can be realized within the complex Weyl algebra $A_{1}$, with particular attention to how these results translate to the real and skew-hermitian setting. This section ends with an outline of the specific goals of the present work, thereby providing a structured framework for the subsequent analysis.
	
	\subsection{Factorization of bosonic dynamics as a guiding principle}
	The dynamics of quantum mechanical systems is given by the one-parameter family $U(t)$ of unitary operators determined by the differential equation $(\mathrm{d}/\mathrm{d}t)\, U(t)=-iH(t)U(t)$ and the initial condition $U(0)=\mathds{1}$, where $H(t)$ denotes the potentially time-dependent hermitian operator known as the Hamiltonian of the system. Here, we want to consider, for sake of simplicity, only the finite-dimensional case. That is, $U(t)$ is a unitary operator on a finite-dimensional Hilbert space.  Note that an operator $O$ is called hermitian if $O^\dagg=O$, and skew-hermitian if  $O^\dagg=-O$. The Hamiltonian $H(t)$ can generally be decomposed into two parts: a time-independent component $H_{\mathrm{d}}$ and a time-dependent one, where the latter can be decomposed further in terms of a linear combination of time-independent hermitian operators $H_j$ with time-dependent real-valued coefficients $u_j(t)$, called \textit{control functions} in the parlance of quantum control. We therefore have
	\begin{align}
		H(t)=H_{\mathrm{d}}+\sum_{j\in\mathcal{J}} u_j(t) H_j,
	\end{align}
	where $\mathcal{J}$ is an index-set and $u_j(t)$ are control functions that are sometimes assumed to be piecewise constant. In our work, we only require that not all controls are always constant. The time-independent term $H_{\mathrm{d}}$ is normally referred to as the \emph{drift} Hamiltonian in the parlance of quantum control, and it represents the part of the system that cannot be externally manipulated. 
	
	The formal solution to the differential equation determining $U(t)$ is generally given by the time-orderd exponential given in \eqref{time:evolution:general:solution}. Rather than working directly with the hermitian Hamiltonian $H(t)$, it is often convenient to consider the associated skew-hermitian operator $iH(t)$ instead. This operator decomposes similarly into a linear combination of time-independent skew-hermitian operators $G_j$ as 
	\begin{align}
		iH(t)=G_{\mathrm{d}}+\sum_{j\in\mathcal{J}}u_j(t)G_j.
	\end{align}
	The choice between working with the skew-hermitian or the hermitian Hamiltonian is largely a matter of convention and does not affect the qualitative structure of the dynamics. In the physics-oriented literature, the Lie algebra $\g$ corresponding to a Lie group $G$ is typically defined as the set of elements $x\in \g$ that satisfy $\exp(ixt)\in G$ for all $t\in\R$, see \plainrefs{Georgi:1999}. In contrast, the more mathematically-oriented literature often defines $\g$ as the set of all elements $x\in \g$ that satisfy $\exp(xt)\in G$ for all $t\in\R$, see \plainrefs{Woit:2017}. Thus, these conventions are related by the identification $\g\leftrightarrow i\g$, which constitutes a Lie-algebra isomorphism with $\phi:\g\to i\g,\,x\mapsto ix$ and $[\phi(x),\phi(y)]_{i\g}=-i[\phi(x),\phi(y)]_\g$ for all $x,y\in\g$. We choose to follow the mathematical convention, and therefore work with skew-hermitian operators.
	
	The useful property that one wishes to have is that the set $\mathcal{G}:=\{G_{\mathrm{d}},G_j\}_{j\in\mathcal{J}}$ generates a finite-dimensional Lie algebra $\g:=\lie{\mathcal{G}}$. This is trivially obeyed in the finite-dimensional setting, but when extended to the infinite-dimensional case necessary, due to the requirements given in \plainrefs{Wei:Norman:1963,Wei:Norman:1964}. This space is constructed starting from the set of generators $\mathcal{G}$ and the bilinear, skew-symmetric, and Jacobi-identity satisfying operation $\mathfrak{G}\supseteq\mathcal{G}$ called \textit{Lie bracket}, and then adding all elements that can be obtained by performing a countable number of Lie brackets and linear combinations using the elements in the initial set $\mathcal{G}$. To be precise, and due to the importance of generating Lie algebras from initial sets of operators for this work, we deem it necessary to rigorously define the term \emph{Lie closure}.
    \begin{definition}[Lie closure]\label{def:Lie:closure}
        Let $\mathcal{G}\subseteq \mathfrak{G}$ be a subset of a Lie algebra $\mathfrak{G}$ equipped with a Lie bracket $[\cdot,\cdot]:\mathfrak{G}\times \mathfrak{G}\to \mathfrak{G}$. The \emph{Lie closure} of $\mathcal{G}$, denoted $\lie{\mathcal{G}}$, is defined as the Lie algebra $$\lie{\mathcal{G}}:=\spn\left\{\bigcup_{j\in\N_{\geq0}}\mathfrak{v}_j\right\},$$ where $\mathfrak{v}_0:=\spn\{\mathcal{G}\}$ and  $\mathfrak{v}_j:=[\mathfrak{v}_{f(j)},\mathfrak{v}_{g(j))}]$ for recursively defined functions $f,g:\N_{\geq1}\to\N_{\geq0}$ with initial values $f(1):=0=:g(1)$ and the update rules:
        \begin{align*}
            f(k+1):=&\left\{\begin{matrix}
                f(k)+1&,\text{ if }f(k)<g(k)\\
                0&,\text{otherwise}
            \end{matrix}\right.,\;&\; g(k+1):=&\left\{\begin{matrix}
                g(k)&,\text{ if }f(k)<g(k)\\
                g(k)+1&,\text{otherwise}
            \end{matrix}\right.,
        \end{align*}
        for all $k\in\N_{\geq1}$.
    \end{definition}
    Recall that the commutator of two subspaces $\mathfrak{v}_1,\mathfrak{v}_2\subseteq\mathfrak{G}$ of a Lie algebra $\mathfrak{G}$ is defined as $[\mathfrak{v}_1,\mathfrak{v}_2]:=\spn\{[x,y]\,\mid\,x\in\mathfrak{v}_1,\,y\in\mathfrak{v}_2\}$ \plainrefs{Knapp:1996}. Moreover, it is easy to verify that this definition is equivalent to commonly employed in the literature \plainrefs{Shore:2020}.

    Without loss of generality, we may now assume that the set $\mathcal{G}$ is linearly independent. Once the construction of the whole Lie algebra $\g$ has been performed, we can apply, e.g., the Gram-Schmidt procedure to construct a basis of $\g$ that is of the form $\{\mathcal{G}_{\mathrm{d}},G_k\}_{k\in\mathcal{K}}$, where $\mathcal{J}\subseteq\mathcal{K}$. In general this means that the number $|\mathcal{J}|$ of initial control elements in the Hamiltonian is less than the number $|\mathcal{K}|$ of basis elements of the generated algebra, since the Lie bracket among elements of $\mathcal{G}$ typically generates additional linearly independent elements. The condition that $\g$ is finite dimensional then translates to the requirement $|\mathcal{K}|<\infty$. Example~\ref{example:xy:interaction} helps better clarify this point.
	
	\begin{tcolorbox}[breakable, colback=Cerulean!3!white,colframe=Cerulean!85!black,title=Example: Hamiltonian with finite factorization]
		\begin{example}\label{example:xy:interaction}
			We consider here the Hamiltonian of two qubits interacting via an $\mathrm{XY}$-type coupling \plainrefs{Liangyu:2021}, also commonly refered to as $\mathrm{XX}$-type coupling \plainrefs{Coden:2021,Dimo:2022}. The Hamiltonian reads
			\begin{align*}
				H(t)=a_1(t)\sigma_{\mathrm{z},1}+a_2(t)\sigma_{\mathrm{z},2}+ J\left(\sigma_{\mathrm{x},1}\sigma_{\mathrm{x},2}+\sigma_{\mathrm{y}1}\sigma_{\mathrm{y},2}\right),
			\end{align*}
			where $\sigma_{j,k}$ are the Pauli operators acting on qubit $k$, and satisfying the commutation relations $[\sigma_{j,p},\sigma_{k,q}]=i\delta_{pq}\sum_{\ell}\epsilon_{jk\ell}\sigma_{\ell,p}$ and anticommutation relations $\{\sigma_{j,p},\sigma_{k,p}\}=2\delta_{jk}\mathds{1},$ with $\epsilon_{jk\ell}$ denoting the totally antisymmetric Levi-Civita symbol, $\delta_{jk}$ the Kronecker delta and $\mathds{1}$ the unity operator. Here, the interaction-operator $J(\sigma_{\mathrm{x},1}\sigma_{\mathrm{x},2}+\sigma_{\mathrm{y},1}\sigma_{\mathrm{y},2})$ constitutes the drift Hamiltonian. The generating set $\mathcal{G}$ of linearly independent skew-hermitian operators is
			\begin{align*}
				\mathcal{G}:=\left\{i\sigma_{\mathrm{z},1},i\sigma_{\mathrm{z},2},i\sigma_{\mathrm{x},1}\sigma_{\mathrm{x},2}+i\sigma_{\mathrm{y},1}\sigma_{\mathrm{y},2}\right\}.
			\end{align*}
			The induced Lie algebra $\g:=\lie{\mathcal{G}}$ is finite dimensional. A basis $\mathcal{F}$ for $\g$ as a vector space includes an additional element, and it reads
			\begin{align*}
				\mathcal{F}=\left\{i\sigma_{\mathrm{z},1},i\sigma_{\mathrm{z},2},i\sigma_{\mathrm{x},1}\sigma_{\mathrm{x},2}+i\sigma_{\mathrm{y},1}\sigma_{\mathrm{y},2},i\sigma_{\mathrm{x},1}\sigma_{\mathrm{y},2}-i\sigma_{\mathrm{y},1}\sigma_{\mathrm{x},2}\right\},
			\end{align*}
			demonstrating that $|\mathcal{G}|<|\mathcal{F}|$. It is straightforward to verify that the $\lie{\mathcal{G}}$ is isomorphic to $\mathfrak{u}(2) = \mathfrak{su}(2)\oplus\mathfrak{u}(1)$.
		\end{example}
	\end{tcolorbox}
	
	The factorization method pioneered by Wei and Norman enables the explicit construction of the time evolution operator $U(t)$, provided that the Lie algebra $\g$ is finite dimensional and the associated Lie group acts on a finite-dimensional space or a Banach space \plainrefs{Wei:Norman:1963,Wei:Norman:1964}. In such cases, one can express $U(t)$ as the product of exponentials involving time-dependent and real scalar functions $f_{\mathrm{d}}(t)$ and $f_k(t)$ for $k\in\mathcal{K}$ as 
	\begin{align}\label{eqn:time:evolution:operator:factorized}
		U(t)=e^{-f_{\mathrm{d}}(t) G_{\mathrm{d}}}\prod_{k\in\mathcal{K}}e^{-f_k(t) G_k}.
	\end{align}      
	In this work we will focus exclusively on finite sets $\mathcal{G}=\{G_j\mid\;0\leq j\leq m\}$ consisting of skew-hermitian polynomials constructed from the bosonic creation and annihilation operators $a^\dagg$ and $a$. A more detailed discussion can be found in Section~\ref{section:background:formalism}. This choice of physical systems introduces a significant difference in the mathematical setting, that is, the Hilbert space $\mathcal{H}$ on which these operators act on is infinite dimensional. This poses additional challenges as compared to the finite-dimensional case, in particular with regard to the controllability of the system. One major complication is the presence of \textit{unbounded operators}, namely linear operators $A$ for which there exists no constant $C$ such that $\| Av\|\leq C\|v\|$ for all $v\in\mathcal{H}$, see \plainrefs{Hall:Lie:quantum:theory:16}. While all linear operators are bounded in the case of finite-dimensional Hilbert spaces, unbounded linear operators are ubiquitous in an infinite-dimensional setting \plainrefs{Bogachev:2020}. A simple example of an unbounded linear operator is, for instance, the bosonic annihilation or creation operator \plainrefs{Bagarello:2007}. An illustrative case of the problem mentioned here can be found in Example~\ref{example:unbounded:operator} below.
	
	The controllability of quantum dynamics is highly dependent on whether the underlying Hilbert space is finite- or infinite-dimensional. In finite-dimensional settings, the Lie algebra $\g$ associated with any Hamiltonian will always be finite-dimensional, and the factorization \eqref{eqn:time:evolution:operator:factorized} can therefore always be obtained. This is not necessarily the case for systems with infinite-dimensional Hilbert spaces. In such settings, the Lie algebra associated with an arbitrary Hamiltonian often becomes infinite dimensional. Consider for example the Hamiltonian $H=\omega_1a_1^\dagg a_1+\omega_2 a_2^\dagg a_2+\omega_3 a_3^\dagg a_3+g_0(t)(a_1a_2a_3+a_1^\dagg a_2^\dagg a_3^\dagg)$, where the operators $a_j,a^\dagg_j$ satisfy the canonical commutation relations $[a_j,a_k^\dagg]=\delta_{jk}$. This Hamiltonian provides a toy model to explain the process of parametric down-conversion \plainrefs{Agust:SPDC:2020}. It is straightforward to verify that the corresponding generating set $\mathcal{G}=\{i\omega_1a_1^\dagg a_1+i\omega_2 a_2^\dagg a_2+i\omega_3 a_3^\dagg a_3,i(a_1a_2a_3+a_1^\dagg a_2^\dagg a_3^\dagg)\}$ generates an infinite-dimensional Lie algebra \plainrefs{Bruschi:Xuereb:2024}. All together, the previous observations should make it clear that addressing the problem of controllability becomes much more complicated in the case of infinite-dimensional Hilbert spaces, since the finiteness of the dimensionality of the Lie algebra must be supplemented with the appropriate subspaces on which operators can act \plainrefs{Huang:Tarn:1983,Lan:2005,Wu:Tarn:2006,Nelson:1959,Bloch:finite:controlability:2010} (refer also to 
	Section~\ref{sec:consequences:control}).
	These and related topics and issues have been extensively explored in the past. Notable results include, among others, conditions for controllability in finite-dimensional systems \plainrefs{JS:1972}, which depend significantly on the presence and form of the drift Hamiltonian. Therefore, since we do not tackle issues of domain and applicability of our work in the context of infinite-dimensional systems, we regard factorization as a motivation and restrict our scrutiny strictly to Lie subalgebras of $\hat{A}_1$.
	
	\begin{tcolorbox}[breakable, colback=Cerulean!3!white,colframe=Cerulean!85!black,title=Example: Unbounded operator]
		\begin{example}\label{example:unbounded:operator}
			Let $\mathcal{H}$ be the Hilbert space of a single quantum-mechanical harmonic oscillator. This space is spanned by the orthonormal basis $\{\ket{n}\}_{n\in\N_{\geq0}}$, defined by the equations $a\ket{n}=\sqrt{n}\ket{n-1}$ and $a^\dagg\ket{n}=\sqrt{n+1}\ket{n+1}$ for all $n\in\N_{\geq0}$ \plainrefs{Cohen:Tannoudji:2020}. Consider now the state
			\begin{align*}
				|\,\varphi\rangle:=\frac{\sqrt{6}}{\pi}\sum_{n=1}^\infty\frac{1}{n}|\,n\rangle\qquad\text{satisfying}\qquad\||\,\varphi\rangle\|^2=\langle\varphi\mid\varphi\rangle=\frac{6}{\pi^2}\sum_{n=1}^\infty\frac{1}{n^2}=1.
			\end{align*}
			This is clearly properly normalized. However, acting with the annihilation operator on it yields:
			\begin{align*}
				\| a|\,\varphi\rangle\|^2=\left\|\frac{\sqrt{6}}{\pi}\sum_{n=1}^\infty\frac{1}{\sqrt{n}}|\,n-1\rangle\right\|^2=\frac{6}{\pi^2}\sum_{n=1}^\infty\frac{1}{n},
			\end{align*}
			which diverges. This illustrates that the annihilation operator is unbounded, as it maps a normalized state to one with infinite norm.
		\end{example}
	\end{tcolorbox}

	\subsection{Factorization of quantum dynamics: a practical use case}
	Before we state the aim of this work we wish to provide a more concrete outlook on the physical motivation for studying this line of research, in line with the discussion present above. As argued in the previous subsection, factorization of the form \eqref{eqn:time:evolution:operator:factorized} has great potential when explicit expressions of the time evolution have to be obtained, or analytical properties of a quantum channel (such as the time evolution) have to be found. 
	
	The Wei-Norman method is the key tool to tackle this issue \plainrefs{Qvarfort:2025}. One starts from the time-evolution operator
	\begin{align}\label{eqn:appendix:time:ordered:exponential}
		U(t)=\overset{\leftarrow}{\mathcal{T}}\exp\left[-i\int_0^t H(\tau)\mathrm{d}\tau\right],
    \end{align}
    which yields the factorization ansatz
    \begin{align}
		\label{ansatz:Wei:Norman}
		U(t)&=\prod_{j=1}^n U_j(t),\qquad\text{ with }\qquad U_j(t)=e^{-f_j(t) g_j}.
	\end{align}
	Here, $g_j$ are linearly independent elements of a finite-dimensional Lie algebra $\g$. In our case, we have $\g\subseteq\hat{A}_1$. The real-valued and time-dependent functions $f_j(t)$ encode the dynamics of the system and need to be determined by the procedure explained below.
	
	We start by taking the time derivative of \eqref{ansatz:Wei:Norman}, we use the expression \eqref{eqn:appendix:time:ordered:exponential}, and perform some algebra as explained in Appendix~\ref{app:factorization} to find
	\begin{align}
		iH(t)=\sum_{j=1}^nu_j(t)g_j=\Dot{f}_1g_1+\Dot{f}_2(t)U_1(t) g_2U_1^\dagg(t)+\ldots+\Dot{f}_n(t) U_1(t)\ldots U_{n-1}(t)g_nU_{n-1}^\dagg(t)\ldots U_1^\dagg(t).
	\end{align}
	Finally, one must compute all similitude relations of the form $U_1(t)\ldots U_{k-1}(t)g_kU_{k-1}^\dagg(t)\ldots U_1^\dagg(t)$ with $k\in\{2,\ldots,n\}$ and equate the coefficients of identical generators $g_j$ on both sides which will be successful only if the Hamiltonian Lie algebra $\g=\lie{\{g_j\}_{j=1}^n}$ is finite dimensional. Our work ties into this factorization method precisely at this stage of the considerations: we are interested in obtaining conditions to determine a priori when the type of Lie algebras just mentioned are finite dimensional or not.
	
	In order to give a more concrete understanding of this problem, we illustrate the factorization method with two simple yet important examples, which will be central to our work. In particular we will consider in the Schr\"odinger algebra $\mathcal{S}$ and the Wigner-Heisenberg algebra $\wh_2$. Explicit and detailed computations can be found by the interested reader 
	in Appendix~\ref{app:factorization}, while we discuss the high-level results now.

	\subsubsection{The Schr\"odinger algebra}
	
	The first algebra that we consider is the Schr\"odinger algebra $\mathcal{S}$, which is a $6$-dimensional non-abelian Lie algebra that is spanned by the basis $\mathcal{G}:=\{g_1:=ia^\dagg a,g_2:=(a-a^\dagg),g_3:=i(a+a^\dagg),g_4:=a^2-(a^\dagg)^2,g_5:=i(a^2+(a^\dagg)^2),g_6:=i\}$ 
	and we point to Equation~\eqref{eqn:def:monomials} for the notation and
	Proposition~\ref{prop:schroedinger:algebra} for details about the Schr\"odinger algebra. 
	We employ the generic Hamiltonian
	\begin{align*}
		H(t)=u_1(t) a^{\dagg}a-i\,u_2(t)(a-a^{\dagg})+u_3(t)(a+a^{\dagg})-i\,u_4(t)(a^2-(a^{\dagg})^2)+u_5(t)(a^2+(a^{\dagg})^2),
	\end{align*}
	where $u_1(t)=\omega$ is standard choice for the harmonic oscillator.
	The ansatz for this case reads
	\begin{align*}
		U(t)=e^{-if_1(t)a^{\dagg}a}e^{-f_2(t)(a-a^{\dagg})}e^{-f_3(t)i(a+a^{\dagg})}e^{-f_4(t)(a^2-a^{\dagg}{}^2)}e^{-f_5(t)i(a^2+a^{\dagg}{}^2)}
	\end{align*}
	modulo an overall complex phase that is irrelevant.
	
	In Appendix~\ref{app:factorization} we compute the similitude relations necessary to obtain the differential equations for the functions $f_j(t)$ as required by the method described above. This is possible as the Schr\"odinger algebra is a finite-dimensional Lie algebra.
	We find: 
	\begin{subequations}\label{Schroedinger:factorization:equations}
		\begin{align}
			u_1(t) &= \dot{f}_1(t) - 2\dot{f}_5(t)\sinh(4f_4(t)), \\[6pt]
			u_2(t) &= \dot{f}_2(t)\cos(f_1(t)) - \dot{f}_3(t)\sin(f_1(t))   - 2\dot{f}_4(t)\bigl(f_2(t)\cos(f_1(t)) + f_3(t)\sin(f_1(t))\bigr) \nonumber \\
			&\quad + 2\dot{f}_5(t)\bigl(f_2(t)e^{-4f_4(t)}\sin(f_1(t)) - f_3(t)e^{4f_4(t)}\cos(f_1(t))\bigr), \\[6pt]
			u_3(t) &= \dot{f}_2(t)\sin(f_1(t)) + \dot{f}_3(t)\cos(f_1(t))   + 2\dot{f}_4(t)\bigl(f_3(t)\cos(f_1(t)) - f_2(t)\sin(f_1(t))\bigr) \nonumber \\
			&\quad - 2\dot{f}_5(t)\bigl(f_2(t)e^{-4f_4(t)}\cos(f_1(t)) + f_3(t)e^{4f_4(t)}\sin(f_1(t))\bigr), \\[6pt]
			u_4(t) &= \dot{f}_4(t)\cos(2f_1(t)) - \dot{f}_5(t)\cosh(4f_4(t))\sin(2f_1(t)), \\[6pt]
			u_5(t) &= \dot{f}_4(t)\sin(2f_1(t)) + \dot{f}_5(t)\cosh(4f_4(t))\cos(2f_1(t)).
		\end{align}
	\end{subequations}
	Finally, the inverted set of equations is:
	\begin{subequations}\label{Schroedinger:factorization:final:equations}
		\begin{align}
			\dot{f}_1(t) &= u_1(t)-2u_4(t)\sin(2f_1(t))\tanh(4f_4(t))+2u_5(t)\cos(2f_1(t))\tanh(4f_4(t)),\\[6pt]
			\dot{f}_2(t) &= u_2(t)\cos(f_1(t))+u_3(t)\sin(f_1(t))
			\nonumber \\
			&\quad-2u_4(t)\Big(f_3(t)\sin(2f_1(t))\big(1+\tanh(4f_4(t))\big)-f_2(t)\cos(2f_1(t))\Big)
			\nonumber \\
			&\quad+2u_5(t)\Big(f_3(t)\cos(2f_1(t))\big(1+\tanh(4f_4(t))\big)+f_2(t)\sin(2f_1(t))\Big), \\[6pt]
			\dot{f}_3(t) &= -u_2(t)\sin(f_1(t))+u_3(t)\cos(f_1(t))
			\nonumber \\
			&\quad-2u_4(t)\Big(f_2(t)\sin(2f_1(t))\big(1-\tanh(4f_4(t))\big)+f_3(t)\cos(2f_1(t))\Big)
			\nonumber \\
			&\quad+2u_5(t)\Big(f_2(t)\cos(2f_1(t))\big(1-\tanh(4f_4(t))\big)-f_3(t)\sin(2f_1(t))\Big), \\[6pt]
			\dot{f}_4(t) &= u_4(t)\cos(2f_1(t))+u_5(t)\sin(2f_1(t)), \\[6pt]
			\dot{f}_5(t) &= -u_4(t)\frac{\sin(2f_1(t))}{\cosh(4f_4(t))}+u_5(t)\frac{\cos(2f_1(t))}{\cosh(4f_4(t))}.
		\end{align}
	\end{subequations}
	These equations can then be solved numerically once the input Hamiltonian functions have been given.

	\subsubsection{The Wigner-Heisenberg algebra of second kind}
	We move on to the case of Wigner-Heisenberg algebra $\wh_2$, which is a $4$-dimensional Lie-algebra spanned by the basis $\mathcal{G}:=\{g_1:=ia^\dagg a,g_2:=a-a^\dagg,g_3:=i(a+a^\dagg),g_4:=i\}$. For more information about this algebra we refer the reader to Definition~\ref{def:Wigner:Heisenberg:algebra:2}. 
	We employ the generic Hamiltonian 
	\begin{align*}
		H(t)=u_1(t) a^\dagg a - i\,u_2(t)(a-a^\dagg)+u_3(t)(a+a^\dagg),
	\end{align*}
	where $u_1(t)=\omega$ is the standard choice for most scenarios.
	The ansatz for this case reads
	\begin{align*}
		U(t)=e^{-if_1(t)a^{\dagg}a}e^{-f_2(t)(a-a^{\dagg})}e^{-f_3(t)i(a+a^{\dagg})}
	\end{align*}
	modulo an irrelevant complex phase.
	
	As done for the previous case, in Appendix~\ref{app:factorization} we compute the similitude relations necessary to obtain the differential equations for the functions $f_j(t)$. We can obtain analytical solutions to the differential equations, thus obtaining 
	\begin{subequations}\label{Hesienberg:factorization:solution}
		\begin{align}
			f_1(t)&=\int_0^t u_1(\tau)\mathrm{d}\tau,\\
			f_2(t)&=\int_0^t\left(\cos(f_1(\tau)) u_2(\tau)+\sin(f_1(\tau)) u_3(\tau)\right) \mathrm{d}\tau,\\
			f_3(t)&=\int_0^t\left(\cos(f_1(\tau)) u_3(\tau)-\sin(f_1(\tau)) u_2(\tau)\right) \mathrm{d}\tau.
		\end{align}
	\end{subequations}
	Notice that it is immediate to see that one obtains \eqref{Hesienberg:factorization:solution} by inverting and integrating \eqref{Schroedinger:factorization:equations} with $u_4(t)=u_5(t)\equiv0$, which immediately also implies $f_4(t)=f_5(t)\equiv0$.

	\subsection{The real skew-hermitian Weyl algebra and the complex Weyl algebra of single mode}
	Beyond the physically motivated goal of identifying finite-dimensional Lie algebras this work is also driven by the need to refine the methods developed in \plainrefs{Bruschi:Xuereb:2024}, and apply them to cases of interest to modern research. The simplest, yet nontrivial, case of interest is that of the single-mode skew-hermitian Weyl algebra. Our aim in this work is to provide a complete classification of all finite-dimensional Lie subalgebras that can be faithfully realized within this algebra. A comparable classification has already been achieved for all non-abelian finite-dimensional Lie algebras that can be realized as a subalgebra of the complex Weyl algebra $A_1$ \plainrefs{Tanasa:2005,TST:2006}. The restriction to non-abelian subalgebras serves to focus attention on non-trivial structures, as any abelian finite-dimensional Lie algebra can trivially be realized in both the complex Weyl algebra $A_1$ and its skew-hermitian counterpart $\hat{A}_{1}$. The non-trivial finite-dimensional subalgebras of the complex Weyl algebra $A_1$ fall into the following six classes:
	
	\begin{enumerate}[label=(\roman*)]
		\item $\slC{2}$, also known as the complex special linear algebra (of $2$ modes),
		\item $\slC{2}\oplus\C$,
		\item $\slC{2}\ltimes\gh_{1,\C}$, where $\gh_{1,\C}$ denotes the complex three-dimensional Heisenberg algebra,
		\item $\mathcal{L}_{n,\C}$ for $n\geq2$, a family of nilpotent Lie algebras of dimension $n+1$,
		\item $\tilde{\mathcal{L}}_{n,\C}$ for $n\geq 2$, a family of solvable Lie algebras of dimension $n+2$,
		\item $\gr_\C(j_1,\ldots,j_n)$ for $0\leq j_1<\ldots<j_n$, a class of infinite families of solvable Lie algebras of dimension $n+1$.
	\end{enumerate}
	The nilpotent Lie algebras $\mathcal{L}_{n,\C}$ admit a basis $\{e_1,\ldots,e_n,e_{n+1}\}$ with non-zero commutation relations $[e_1,e_j]=e_{j+1}$ for $j\in\{2,\ldots,n\}$. Notably $\mathcal{L}_{2,\C}\cong\gh_{1,\C}$. The solvable algebras $\tilde{\mathcal{L}}_{n,\C}$ are isomorphic to a semidirect sum $\C\ltimes\mathcal{L}_{n,\C}$ and their first derived algebra is therefore isomorphic to $\mathcal{L}_{n,\C}$
	(where the Definition~\ref{def:semidirect:product} provides more details on semidirect sums). These algebras always admit a basis $\{e_0,e_1,\ldots,e_n,e_{n+1}\}$ with the non-vanishing commutation relations being $[e_0,e_1]=e_1$, $[e_0,e_j]=-(n+1-j)e_j$, and $[e_1,e_j]=e_{j+1}$ for $j\in\{2,\ldots,n\}$. The algebras $\gr_\C(j_1,\ldots,j_n)$ are isomorphic to the semidirect sum $\C\ltimes\C^n$ if $j_1\neq 0$ and $(\C\ltimes\C^n)\oplus\C$ otherwise. They admit for $j_1\neq 0$ a basis $\{e_0,e_1,\ldots,e_n\}$, where the only non-zero commutation relations are $[e_0,e_k]=-j_k e_k$. These algebras satisfy therefore the isomorphism $\gr_\C(j_1,\ldots,j_n)\cong\gr_\C(mj_1,\ldots,mj_n)$ for all $m\in\N_{>0}$, suggesting a connection to homogeneous coordinates in projective geometry \plainrefs{Graustein:1930}. The first derived algebra of each $\gr_\C(j_1,\ldots,j_n)$ is abelian, and therefore isomorphic to $\C^n$ if $j_1\neq 0$ and $\C^{n-1}$ otherwise.
	
	A natural assumption is that the same classification of finite-dimensional subalgebras also holds the real skew-hermitian case. Indeed, one can construct explicit real realizations of each of the algebra classes discussed above:
	\begin{align*}
		\slR{2}&\cong\left\langle\left\{2i\left(a^\dagg a+\frac{1}{2}\right),a^2-(a^\dagg)^2,i\left(a^2+(a^\dagg)^2\right)\right\}\right\rangle_{\textrm{Lie}},\\
		\slR{2}\oplus\R&\cong\left\langle\left\{2i,ia^\dagg a,a^2-(a^\dagg)^2,i\left(a^2+(a^\dagg)^2\right)\right\}\right\rangle_{\textrm{Lie}},\\
		\slR{2}\ltimes\gh_1&\cong\left\langle\left\{2i,ia^\dagg a,a-a^\dagg,i\left(a+a^\dagg\right),a^2-(a^\dagg)^2,i\left(a^2+(a^\dagg)^2\right)\right\}\right\rangle_{\textrm{Lie}},\\
		\mathcal{L}_n&\cong\left\langle\left\{2i,i\left(a+a^\dagg\right),a-a^\dagg,i\left(a-a^\dagg\right)^2,\ldots,i^{(n-1)\mod 2}\left(a-a^\dagg\right)^n\right\}\right\rangle_{\textrm{Lie}},\\
		\tilde{\mathcal{L}}_n&\cong\left\langle\left\{2i,i\left(a+a^\dagg\right),a^2-(a^\dagg)^2,a-a^\dagg,i\left(a-a^\dagg\right)^2,\ldots,i^{(n-1)\mod 2}\left(a-a^\dagg\right)^n\right\}\right\rangle_{\textrm{Lie}},\\
		\gr(j_1,\ldots,j_n)&\cong\left\langle\left\{a^2-(a^\dagg)^2,i\left(a+a^\dagg\right)^{j_1},\ldots,i\left(a+a^\dagg\right)^{j_n}\right\}\right\rangle_{\textrm{Lie}}.
	\end{align*}
	However, as we will demonstrate in the following, this naive assumption that classification of all real finite-dimensional skew-hermitian subalgebras yields the same algebras as in the complex case does not hold in general. The real skew-hermitian setting introduces subtleties that prevent a direct translation of the complex classification into the real one. The underlying reason for the failure of the ability to map one algebra to the other rests in the fact that the field of real numbers is not algebraically closed. This prevents the existence of certain Lie-algebra isomorphisms that are valid over algebraically closed fields such as the complex numbers $\C$. For instance, the proof of Theorem 4.17 from \plainrefs{TST:2006}---a theorem that classifies every non-nilpotent and solvable algebra that can be realized within $A_1$ and whose derived algebra isn't abelian---relies on the ability to diagonalize a linear map, which can be done over $\C$, but not over $\R$.

 This distinction can be illustrated by comparing the following two Lie algebras $\wh_1=\lie{\{2i,a^2-(a^\dagg)^2,i(a+a^\dagg),a-a^\dagg\}}$ and $\wh_2=\lie{2i,ia^\dagg a, i(a+a^\dagg),a-a^\dagg)\}}$ considered over both the real numbers and the complex numbers. Before analyzing these algebras directly, we first want to explain their qualitative difference. The aforementioned Theorem 4.17 from \plainrefs{TST:2006} states that any finite-dimensional non-nilpotent and solvable subalgebra of $A_1$, whose derived algebra is nilpotent, is itself isomorphic to $\tilde{\mathcal{L}}_{n,\C}$ for the appropriate $n$. To follow the proof in case of the real skew-hermitian Weyl Algebra $\hat{A}_1$, we require the existence of an element $h\in\wh_2$ such that the adjoint map $\operatorname{ad}(h)|_{\wh_2}$ is not nilpotent. Clearly, we can choose $h=ia^\dagg a$. Next, we need to diagonalize $\operatorname{ad}(ia^\dagg a)|_{\gh}=[ia^\dagg a,\cdot]|_{\gh}$. The eigenvalues of this endomorphism are $\lambda_1=\lambda_2=0$, $\lambda_3=i$, and $\lambda_4=-i$ with the corresponding eigenspaces: $\spn\{2i,ia^\dagg a\}$, $\spn\{a^\dagg\}$, and $\spn\{ a\}$. Thus, diagonalizing $\operatorname{ad}(ia^\dagg a)|_{\gh}$ is not possible over the field of real numbers, since some of the eigenvalues are not real, marking the point where the proofs for complex and real Lie algebras diverge.
	
	Here we present a result that clarifies the intuitive reasoning above.
	
	\begin{proposition}\label{prop:example:naive:assumption:wrong}
		The real skew-hermitian Lie algebras $\wh_1$ and $\wh_2$ are not isomorphic over the field of real numbers, but their complexifications are isomorphic.
	\end{proposition}
	
	\begin{proof}
		We begin with the second part of the statement. Consider the complexifications $\wh_{2,\C}=\wh_2+i\wh_2$ and $\wh_{1,\C}=\wh_1+i\wh_1$. Here, one can choose the following basis elements: For $\wh_{2,\C}$ the elements $e_1=1$, $e_2=a^\dagg a$, $e_3=a^\dagg$, and $e_4=a$; for $\wh_{1\C}$ the elements $b_1=-2i$, $b_2=\frac{1}{2}(a^2-(a^\dagg)^2)$, $b_3=-i(a+a^\dagg)$, and $b_4=a-a^\dagg$. The non-trivial commutator relations are:
		\begin{align*}
			[e_2,e_3]&=e_3,\;&\;[e_2,e_4]&=-e_4,\;&\;[e_3,e_4]&=-e_1,\\
			[b_2,b_3]&=b_3,\;&\;[b_2,b_4]&=-b_4,\;&\;[b_3,b_4]&=-b_1.
		\end{align*}
		This yields the immediate conclusion $\wh_{1,\C}\cong\wh_{2,\C}$. 
		
		To prove the first part of the statement (non-isomorphism over $\R$), one could attempt to construct an explicit Lie-algebra isomorphism $\phi:\wh_1\to\wh_2$ by solving the resulting system of equations. However, this approach is cumbersome and not particular illuminating. Instead, we refer to the literature: $\wh_2\cong A_{4,9}^{a}$ with $a=0$, and $\wh_1\cong A_{4,8}^{b}$ with $b=-1$, and it is known that then $A_{4,9}^{a}\not\cong A_{4,8}^{b}$, see \plainrefs{Popovych:2003}.  
	\end{proof}
	
	It is also worth noting that $\slC{2}$ has three real forms, namely $\su{2}$, $\slR{2}$, and $\slC{2}$ viewed as a real Lie algebra \plainrefs{Gorbatsevich:2017}. Among these, we expect that only $\slR{2}$ can be faithfully realized in $\hat{A}_1$. This expectation is supported by Theorem 4.4 from \plainrefs{Joseph:1972} (see also Theorem~\ref{thm:Joseph:4:4:full} in Appendix~\ref{App:Realizability:Recap}), which prohibits the realization of certain semisimple Lie algebras associated with compact Lie groups. An example of such a Lie algebra is $\su{2}$. Consequently, not only is $\su{2}$ excluded from faithful realizations in $\hat{A}_1$, but also $\slC{2}$, since $\su{2}$ is a subalgebra of $\slC{2}$.

	\subsection{Aim}
	
	The aim of this work is to investigate subalgebras, in particular finite-dimensional single-mode bosonic subalgebras, of the skew-hermitian Weyl algebra $\hat{A}_1$ from a perspective that encompasses ideas from both physics and more abstract mathematics. These algebras naturally arise in the context of quantum dynamics of coupled bosonic systems. Our treatment focuses on sets of monomials, their algebraic properties, and their commutation relations.
	
	The whole approach that we put forward revolves around three central questions:
	\begin{quote}
		\textbf{Q1}: \emph{Can we determine all finite-dimensional Lie algebras that are generated by a finite set $\mathcal{G}$ of monomials from the skew-hermitian Weyl algebra?}
	\end{quote}
	
	\begin{quote}
		\textbf{Q2}: \emph{Can we determine every finite-dimensional subalgebra of $\hat{A}_1$ that contains a free Hamiltonian term of the form $i(a^\dagg a+c)$, with $c\in\mathbb{R}$?}
	\end{quote}
	
	\begin{quote}
		\textbf{Q3}: \emph{Can we provide a complete classification of all non-abelian and finite-dimensional Lie algebras that can be faithfully realized in $\hat{A}_1$?}
	\end{quote}
	
	The first question is addressed comprehensively in Section~\ref{section:monomial:generated:algebras}, culminating in the key result presented in Theorem~\ref{thm:glossary:monomial:genereted}, which provides a complete list of every realization of every non-abelian, monomial-generated, and finite-dimensional Lie subalgebra of $\hat{A}_1$. It is important to note that a \enquote{monomial} in $\hat{A}_1$ is generally not a monomial in the complex Weyl algebra  $A_1$; for a detailed explanation of this distinction, see Section~\ref{section:background:formalism}. 
	The second question is physically motivated by the observation that the free Hamiltonian often appears as (the only) part of the drift in the Hamiltonian. Unlike in the case of Question~Q1, we do not restrict ourselves to specific sets of monomials, but also consider arbitrary polynomials in $\hat{A}_1$. This question is explored in  Section~\ref{section:lie:algebras:with:free:hamiltonian} and, together with the goals of Question~Q1, it provides insight into which conditions should the operators that define the Hamiltonian satisfy in order for the time evolution of the system to be controllable by a finite set of operations. Finally, the third question is inspired by the known classification of all non-abelian finite-dimensional Lie subalgebras of the complex Weyl algebra $A_1$.
	We tackle this question in Section~\ref{section:classification:nilpotent:nonsolvable}, where we succeed in classifying all nilpotent and all non-solvable finite-dimensional Lie algebras that can be faithfully realized in $\hat{A}_1$. However, a complete classification of solvable algebras remains open. Completing this classification would render the nilpotent case redundant, as any nilpotent algebra is, by definition, necessarily also solvable. A schematic overview of our classification strategy is provided in Figure~\ref{fig:sketch:approach}.
	\begin{figure}[tpb]
		\centering
		\includegraphics[width=0.85\linewidth]{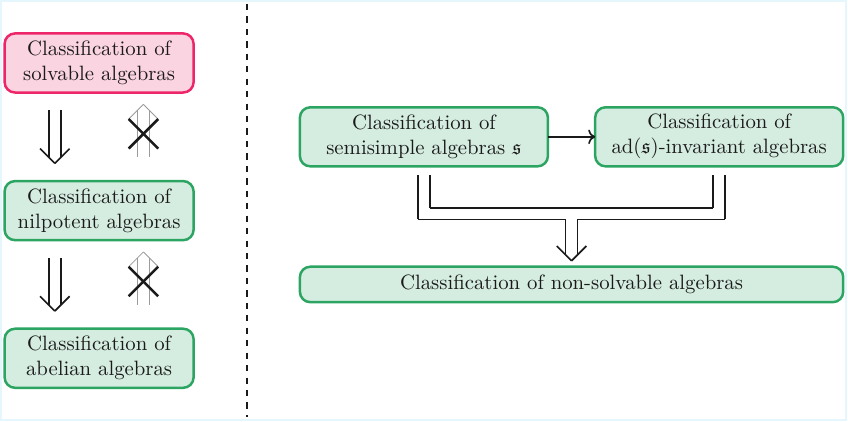}
		\caption{Schematic overview of the classification strategy for finite-dimensional Lie subalgebras of $\hat{A}_1$. The study is divided into two main branches: the classification of solvable Lie algebras (left), and the classification of the non-solvable Lie algebras (right). Double arrows indicate causal relations, while single arrows represent other structural dependencies. A double arrow is crossed out when one classification does not imply the other. Green boxes indicated classifications that have been achieved in this work, while the red boxes highlights those that remain outstanding. The flow-chart immediately shows that a complete classification of all solvable Lie algebras would immediately imply a classification of all nilpotent Lie algebras, since every nilpotent algebra is solvable. The reverse does not hold. Furthermore, there is no exploitable causal relationship between the classification of solvable and non-solvable Lie algebras. In order to achieve the classification of all non-solvable Lie algebras we start by identifying all semisimple Lie algebras $\gs$ that can be faithfully realized in $\hat{A}_1$. This informs the second step, namely, classifying all $\operatorname{ad}$-invariant algebras for such semisimple Lie algebras $\gs$. These two results together yield a complete classification of all non-solvable finite-dimensional Lie algebras that can be faithfully realized in $\hat{A}_1$ via the Levi-Mal'tsev theorem \plainrefs{Kuzmin:1977}.}
		\label{fig:sketch:approach}
	\end{figure}
	
	Investigating these three questions individually sheds light on the overarching question posed at the beginning of this work: \emph{What Hamiltonians admit a finite-dimensional Lie algebra?}

	\section{Background Formalism\label{section:background:formalism}}
	This work builds upon the foundational work on this topic \plainrefs{Bruschi:Xuereb:2024}. Therefore, we will employ the same notation here to maintain consistency. We summarize the most essential concepts and refer the interested reader to peruse the aforementioned work for further details.

	\subsection{The Weyl algebra}
	We consider a single-mode bosonic quantum system with creation and annihilation operators $a^\dagg$ and $a$ satisfying the canonical commutation relation $[a,a^\dagg]=1$. The set of operators $\{a,a^\dagg,1\}$ generate the complex three-dimensional \emph{Heisenberg algebra} $\gh_{1,\C}$ which induces the infinite-dimensional associative single-mode \emph{Weyl algebra} $A_1$ over the field of the complex numbers $\C$, see \plainrefs{Bjoerk:1979,Dixmier:1968,Dixmier:1977}. 
	
	The core ingredients for the following developments are listed here:
	
	\vspace{0.2cm}
	
	\noindent\textbf{Canonical basis}---The Poincaré-Birkhoff-Witt theorem  states that any element of $A_1$ can be expressed as a linear combination of \emph{canonical} monomials $(a^\dagg)^\alpha a^\beta$, see \plainrefs{PBW:Thm}. This choice of ordering, i.e., all creation operators are left of all annihilation operators, is known as the \textit{normal ordering} \plainrefs{Weinberg:1995}. We denote these monomials by $a^\gamma:=a^{(\alpha,\beta)}:=(a^\dagg)^\alpha a^\beta$, with the convention $a^0\equiv1$. The vectors $\gamma$ are covered in the next point.
	
	\noindent\textbf{Multi-indices}---We denote $\gamma=(\alpha,\beta)\in\Np{2}$ with components $\alpha,\beta\in\N_{\geq0}$ as a \emph{multi-index}, where $\Np{2}$ is the set of all non-negative integer vectors of length two. Given two multi-indices $\gamma=(\alpha,\beta)$ and $\gamma'=(\alpha',\beta')$, we define the order $\gamma>\gamma'$ if $\alpha>\alpha'$, or if $\alpha=\alpha'$ and $\beta>\beta'$, see \plainrefs{Bruschi:Xuereb:2024}. The relations $\gamma\geq\gamma'$, and $\gamma=\gamma'$ are defined analogously. A multi-index $\gamma\in\Np{2}$ is \emph{well-ordered} if $\alpha\geq \beta$. We use the notation $\tilde{\gamma}\in\Np{2}$ when $\tilde{\alpha}=\tilde{\beta}$, and write $\tilde{\gamma}\in\N_=^{2}$. We introduce furthermore the multi-indices $\iota_k\in\Np{2}$, defined by the coefficients $(\iota_k)_j=\delta_{jk}$, as well as the multi-index $\tau=\iota_1+\iota_2$\footnote{Note that for purpose of this work, we changed the convention for the multi-indices $\iota_k$. In our previous work, we defined $\iota_k$ as a multi-index from $\Np{n}$ \plainrefs{Bruschi:Xuereb:2024}. We change this convention now, such that $\iota_k$ is a multi-index from $\Np{2n}$ for $n=1$ instead. Thus, in the previous notation one had $\tau_k=(\iota_k,\iota_k)\in\Np{2n}$, whereas now one has $\tau_k=\iota_k+\iota_{n+k}$.}. 
	
	Let $m=c a^\gamma$ be a monomial with $c\in\C\setminus\{0\}$. The multi-degree $\mdeg(m)$ of $m$ is defined as the multi-index $\gamma$, while the degree $\deg(m)$ is given by $\deg(m)=|\gamma|=|\alpha|+|\beta|$, except in the case $\gamma=0$, for which we define $\deg(0)=-\infty$ as done in the literature \plainrefs{Fischer:2005}. The degree of a polynomial $p\in A_1$ is the maximal degree among all monomials in its unique expansion with respect to the canonical basis.
    
	\noindent\textbf{Dagger operation}---The linear anti-automorphism $(\cdot)^\dagg:A_1\to A_1,\;p\mapsto p^\dagg$ shall be defined on the basis elements by $1^\dagg=1$, $(a^\dagg)^\dagg=a$, $(a)^\dagg=a^\dagg$, and its property to reverse the order of multiplication, i.e.,$(m_1m_2)^\dagg=(m_2)^\dagg (m_1)^\dagg$. This induces the map $(\cdot)^\dagg:\Np{2}\to\Np{2},\;\gamma=(\alpha,\beta)\mapsto\gamma^\dagg:=(\beta,\alpha)$ on the set of multi-indices, since $(a^{(\alpha,\beta)})^\dagg=a^{(\beta,\alpha)}$.

	\subsection{The skew-hermitian Weyl algebra}
	The skew-hermitian Weyl algebra, denoted by $\hat{A}_1$, is the real subalgebra of the complex Weyl algebra $A_1$ consisting of all polynomials $p\in A_1$ that satisfy $p^\dagg=-p$. We denote the basis elements by $g^{\gamma _p}_{\sigma _p}$, and we also use the nomenclature \emph{monomials} for generators that are multiplied by an arbitrary non-zero real number $c$\footnote{Usually we employ $c=1$, except in those cases where $\gamma=(0,0,)$, $\gamma=(1,1)$, or explicitly stated. Then, we use other values of $c$. For instance, in the cases where $\gamma=(0,0)$ or $\gamma=(1,1)$, we usually use $c=1/2$.}, while 
	these are also referred to as \emph{generators} in the previous work \plainrefs{Bruschi:Xuereb:2024} and \emph{monomial generators} below. In addition, when multiplied by a non-trivial scalar, they are referred to as \emph{monomials} of $\hat{A}_1$. Note that these are not monomials in the sense of the complex algebra ${A}_{1}$, but rather with respect to the skew-hermitian algebra $\hat{A}_1$. Since we will exclusively consider the skew-hermitian Weyl algebra from this point onward, we will use the term \emph{monomial} to refer specificially to monomials in $\hat{A}_1$, unless otherwise specified. These monomials are defined as follows: 
	\begin{align}\label{eqn:def:monomials}
		g_+^{\gamma_p}:=g_{+}(a^{\gamma_p}):=&i((a^{\gamma_p})^{\dagg}+a^{\gamma_p})=i((a^{\dagg})^{\beta_p}a^{\alpha_p}+(a^{\dagg})^{\alpha_p}a^{\beta_p}),\nonumber\\
		g_-^{\gamma_p}:=g_{-}(a^{\gamma_p}):=&(a^{\gamma_p})^{\dagg}-a^{\gamma_p}=(a^{\dagg})^{\beta_p}a^{\alpha_p}-(a^{\dagg})^{\alpha_p}a^{\beta_p},
	\end{align}where the $a^{\gamma_p}$ are normal-ordered monomials in $A_n$. By construction, we demand $\alpha_p\geq \beta_p$ if $\sigma_p=+$ and $\alpha_p>\beta_p$ if $\sigma_p=-$ (while in the latter case the equality sign is not included since $g_-^{\tilde{\gamma}_p}=(a^{\tilde{\gamma}_p})^{\dagg}-a^{\tilde{\gamma}_p}=0$).
	
	The skew-hermitian Weyl algebra is closed under the Lie bracket, since for any two polynomials $p,p'\in\hat{A}_1$, it holds that  $[p,p']^\dagg=[p'^\dagg,p^\dagg]=-[p,p']$. Moreover, the skew-hermitian Weyl algebra $\hat{A}_1$ can be regarded as an inner product space when equipped with the scalar product $\langle\,\cdot\,|\,\cdot\,\rangle$ defined by the relation $\langle g_\sigma^{(\alpha,\beta)}|g_{\hat{\sigma}}^{(\hat{\alpha},\hat{\beta})}\rangle=\delta_{\sigma\hat{\sigma}}\delta_{\alpha\hat{\alpha}}\delta_{\beta\hat{\beta}}$ on the basis elements.

	\subsection{Decomposition of the skew-hermitian Weyl algebra\label{sec:dec:skew}}
	We now seek to decompose the skew-hermitian Weyl algebra $\hat{A}_1 $ into a direct sum of meaningful subspaces, which has already proven to be a winning strategy for the study of the finiteness of its subalgebras \plainrefs{Bruschi:Xuereb:2024}. We choose to decompose $\hat{A}_1$ according to the following:
	
	\begin{definition}\label{def:vector:spaces:a:1:k}
		Basis sets for the decomposition of $\hat{A}_1 $:
		\begin{align*}
			\mathcal{B}_1^0&:=\left\{i=\frac{1}{2}g_+^0, ia^{\dagg}a=\frac{1}{2}g_+^{\tau}\right\},\;&\;\mathcal{B}_1^1&:=\left\{i(a + a^{\dagg})=g_+^{\iota_1}, a - a^{\dagg}=g_-^{\iota_1}\right\},\;&\;\mathcal{B}_1^2&:=\left\{i(a^2 + (a^{\dagg})^2)=g_+^{2\iota_1}, a^2 - (a^{\dagg})^2=g_-^{2\iota_1}\right\},\\
			\mathcal{B}_1^=&:=\left\{2i(a^{\dagg})^k a^k =g_+^{k\tau} \mid k \geq 2\right\},\;&\;\mathcal{B}_1^\perp&:=\left\{g_{\sigma}^{(\alpha, \beta)} \mid \alpha > \beta \text{ and } \alpha + \beta \geq 3\right\}.
		\end{align*}
		We employ the sets listed here as bases of the vector spaces $\hat{A}_1^K$ for $K\in\{0,1,2,=,\perp\}$, defined by $\hat{A}_1^K:=\spn\{\mathcal{B}_1^K\}$. This, in turn, allows us to write the decomposition
		\begin{align*}
			\hat{A}_1=\hat{A}_1^0\oplus\hat{A}_1^1\oplus\hat{A}_1^2\oplus\hat{A}_1^=\oplus\hat{A}_1^\perp.
		\end{align*}
	\end{definition}
	
	The scalar product $\langle\,\cdot\,|\,\cdot\,\rangle$ defined earlier allows us to define orthogonal projections onto subspaces $\hat{W}\subseteq\hat{A}_1$. Let $\hat{W}$ be a subspace of $\hat{A}_1$ with orthonormal basis $\mathcal{B}_{\hat{W}}$. Then we define the orthogonal projection
	\begin{align*}
		\mathcal{P}_{\hat{W}}:\hat{A}\to\hat{W},
		\quad\text{with action}
		\quad
		v\mapsto\PP{v}{\hat{W}}:=\sum_{b\in\mathcal{B}}\langle v|b\rangle b.
	\end{align*}
	We say that a vector $v\in\hat{A}_1$ has non-vanishing support in a subspace $\hat{W}$ if $\PP{v}{\hat{W}}\neq0$. 

    Since we are interested in generating Lie algebras using only sets of monomials, it is convenient to provide a precise definition for those sets:

    \begin{definition}[Monomial generators]\label{def:gens:mon}
		Starting from distinct single monomials $g_{\sigma_p}^{\gamma_p}\in \hat{A}_1$ indexed by $p\in\mathcal{R}\subseteq\N_{\geq0}$ and a set of non-vanishing coefficients $\{c_p\}_{p\in\mathcal{P}}\subseteq\R\setminus\{0\}$, let
		$
		\mathcal{G}:=\{c_pg_{\sigma_p}^{\gamma_p}\mid\,p\in\mathcal{R}\subseteq\N_{\geq0}\}$ denote a set of \emph{monomial generators} for generating a Lie algebra. 
	\end{definition}

    Distinctness of monomials $g_{\sigma_p}^{\gamma_p}$ within a set of monomial generators $\mathcal{G}$ avoids unnecessary redundancies, since including the same monomial $g_{\sigma_p}^{\gamma_p}$ with different coefficients $c_p$ does clearly not alter the generated Lie algebra $\lie{\mathcal{G}}$. Moreover, the generated Lie algebra also does not rely on the specific non-trivial coefficients of the elements in $\mathcal{G}$. However, in certain cases it is convenient to use coefficients different from one. To address this coefficient ambiguity, first introduced in the definition of monomials themselves, we introduce the following equivalence relation:
    \begin{definition}\label{def:equivalence:of:sets:of:monomial:generators}
        Let $\mathcal{G}=\{c_pg_{\sigma_p}^{\gamma_p}\mid\,p\in\mathcal{P}\subseteq\N_{\geq0}\},\mathcal{G}'=\{\hat{c}_qg_{\hat{\sigma}_q}^{\hat{\gamma}_q}\mid\,q\in\mathcal{Q}\subseteq\N_{\geq0}\}\subseteq\hat{A}_1$ be two sets of monomial generators. We write $\mathcal{G}\sim\mathcal{G}'$ if and only if for every $p\in\mathcal{P}$ there exists exactly one $q\in\mathcal{Q}$ and vice versa such that $g_{\sigma_p}^{\gamma_p}=g_{\hat{\sigma}_q}^{\hat{\gamma}_q}$. 
    \end{definition}
    It is straightforward to verify that the relation \enquote{$\sim$} is reflexive, symmetric, and transitive, and hence an equivalence relation.

    \begin{definition}
        Let \enquote{$\sim$} be the equivalence relation from Definition~\ref{def:equivalence:of:sets:of:monomial:generators}, and let $[\mathcal{G}]_\sim,[\mathcal{G}']_\sim$ denote two equivalence classes with $\mathcal{G}=\{c_pg_{\hat{\sigma}_p}^{\gamma_p}\mid\,p\in\mathcal{P}\}$ and $\mathcal{G}'=\{\hat{c}_qg_{\hat{\sigma}_q}^{\hat{\gamma}_q}\mid\,q\in\mathcal{Q}\}$. We define the following symbols:
        \begin{enumerate}
            \item \textbf{Union:} The union of two equivalence classes is given by the equivalence class $[\mathcal{G}]_\sim\cup_\sim[\mathcal{G}']_\sim:=[\{g_{\sigma_p}^{\gamma_p}\,\mid\,p\in\mathcal{P}\}\cup\{g_{\hat{\sigma}_q}^{\hat{\gamma}_q}\mid\,q\in\mathcal{Q}\}]_\sim$.
            \item \textbf{Intersection:} The intersection of two equivalence classes is given by the equivalence class $[\mathcal{G}]_\sim\cap_\sim[\mathcal{G}']_\sim:=[\{g_{\sigma_p}^{\gamma_p}\,\mid\,p\in\mathcal{P}\}\cap\{g_{\hat{\sigma}_q}^{\hat{\gamma}_q}\mid\,q\in\mathcal{Q}\}]_\sim$. Moreover, the intersection of an equivalence class $[\mathcal{G}]_\sim$ with a subset $\mathcal{A}\subseteq\hat{A}_1$ is given by the equivalence class $[\mathcal{G}]_\sim\cap_{\sim}\mathcal{A}:=[\{g_{\sigma_p}^{\gamma_p}\,\mid\,p\in\mathcal{P}\}\cap\mathcal{A}]_\sim$.
            \item \textbf{Subset:} We write $[\mathcal{G}]_\sim\subseteq_\sim[\mathcal{G}']_\sim$ if and only if for every $p\in\mathcal{P}$, there exists an index $q\in\mathcal{Q}$ such that $g_{\sigma_p}^{\gamma_p}=g_{\hat{\sigma}_q}^{\hat{\gamma}_q}$.
            \item \textbf{Proper subset:} We write $[\mathcal{G}]_\sim\subsetneq_\sim[\mathcal{G}']_\sim$ if and only if for every $p\in\mathcal{P}$, there exists an index $q\in\mathcal{Q}$ such that $g_{\sigma_p}^{\gamma_p}=g_{\hat{\sigma}_q}^{\hat{\gamma}_q}$, and an index $q\in\mathcal{Q}$ such that $g_{\sigma_p}^{\gamma_p}\neq g_{\hat{\sigma}_q}^{\hat{\gamma}_q}$ for all $p\in\mathcal{P}$. 
            \item \textbf{Element:} We write $cg_\sigma^\gamma\in_\sim[\mathcal{G}]_\sim$ if and only if $c\neq 0$ and $g_\sigma^\gamma\in\{g_{\sigma_p}^{\gamma_p}\,\mid\,p\in\mathcal{P}\}$. 
        \end{enumerate}
        Finally, we define $\spn\{[\mathcal{G}]_\sim\}:=\spn\{\mathcal{G}\}$ and $\lie{[\mathcal{G}]}:=\lie{\mathcal{G}}$.
    \end{definition}
    Since each equivalence class $[\mathcal{G}]_\sim$ with $\mathcal{G}=\{c_pg_{\hat{\sigma}_p}^{\gamma_p}\mid\,p\in\mathcal{P}\subseteq\N_{\geq0}\}$ satisfies,  by Definitions~\ref{def:gens:mon} and~\ref{def:equivalence:of:sets:of:monomial:generators}, $\mathcal{G}\sim \{g_{\hat{\sigma}_p}^{\gamma_p}\mid\,p\in\mathcal{P}\subseteq\N_{\geq0}\}$, we often write $\mathcal{G}$ instead of $[\mathcal{G}]_\sim$ and omit the \enquote{$\sim$} indicator in operation, as context makes clear whether sets of equivalence classes are considered.

	\begin{definition}[Partioning monomial generators]\label{def:gens:mon:partition}
		Let
		$
		\mathcal{G}:=\{g_{\sigma_p}^{\gamma_p}\mid\,p\in\mathcal{R}\subseteq\N_{\geq0}\}$ be a set of monomial generators for generating a Lie algebra. The
		set $\mathcal{G}$ is partitioned into disjoint subsets
		\begin{align*}
			\mathcal{G}^K&:=\mathcal{G}\cap \hat{A}_1^K \,\text{ where }\,  K\in\{0,1,2,=,\perp\} \,\text{ and }\, \hat{A}_1^K \,\text{ is defined above.}
		\end{align*}
	\end{definition}

	Definition~\ref{def:gens:mon:partition} 
	partitions these monomial generators into classes with similar algebraic properties which will be exploited in the following sections. Single monomials from the sets $\mathcal{G}^0$ or $\mathcal{G}^=$ can always be written as $g_+^{k\tau}$ for some $k\in\N_{\geq0}$. We will sometimes write $g_+^{\tilde{\gamma}}$ instead of $g_+^{k\tau}$, where $\tilde{\gamma}$ denotes multi-indices form the set $\N_=^2:=\{k\tau\mid\,k\in\N_{\geq0}\}$. This tilde notation will be used exclusively for such multi-indices.

	Before proceeding we introduce the following important notation:
	
	\begin{definition}[Free Hamiltonian]\label{def:free:hamiltonian}
		Polynomials $p\in \hat{A}_1$ of the form $i\omega a^\dagg a+i\phi$, with $\omega>0$ and $\phi\in\R$, are called \emph{free Hamiltonians}. The real coefficient $\omega$ is referred to as \emph{frequency}, and both $\omega$ and $\phi$ may be time-dependent.
	\end{definition}
	
	We now provide a physical interpretation for the vector spaces introduced in Definition~\ref{def:vector:spaces:a:1:k} in order to better understand their role:
	\begin{itemize}
		\item[$\boldsymbol{\hat{A}_1^0}$:] This space is composed of polynomials made from linear combinations of the operators $i$ and $ia^\dagg a$. The former induces an irrelevant (i.e., unobservable) global phase factor in the time evolution operator \eqref{eqn:time:evolution:operator:factorized}, while the latter is the key component of the \emph{free Hamiltonian} as defined above in Definition~\ref{def:free:hamiltonian}.
		In the physics literature, the free Hamiltonian is often written as $H_0=\omega a^\dagg a+\phi$, i.e., in its hermitian form. However, we follow the mathematical convention of working with skew-hermitian Hamiltonians (as already mentioned before), thereby incorporating the factor $i$ from the von Neumann equation into the Hamiltonian itself. The free Hamiltonian drives the free evolution of a harmonic oscillator in the absence of interactions, and it is typically not tunable, i.e., $\omega$ is constant in time. This means that is often considered part of the \emph{drift} term in control theory.
		\item[$\boldsymbol{\hat{A}_1^1}$:] This space is spanned by two elements $i(a{+}a^\dagg)$ and $a{-}a^\dagg$. The exponentials of these operators generate the \emph{Weyl displacement operator} $D(\alpha{+}i\beta):=\exp(-\alpha(a{-}a^\dagg)+i\beta(a{+}a^\dagg))$, which result in the coherent states $|\alpha{+}i\beta\rangle:=D(\alpha{+}i\beta)\ket{0}$ when applied to the vacuum state $\ket{0}$ and is a central tool of quantum optics \plainrefs{Adesso:Ragy:2014,KEGGZ:2023}.
		\item[$\boldsymbol{\hat{A}_1^2}$:] This space is spanned by the two elements $i(a^2+(a^\dagg)^2)$ and $a^2-(a^\dagg)^2$, known as single-mode squeezing operators in quantum optics \plainrefs{Adesso:Ragy:2014}. These operators do not commute with the (skew-Hermitian) number operator $ia^\dagg a$, and are therefore classified as active transformations, since they induce a change in the total number of excitations present in the system.
		\item[$\boldsymbol{\hat{A}_1^=}$:] This space includes all higher-order nonlinearties of the form $i (a^\dagg)^n a^n$ with $n\geq2$, such as the Kerr-nonlinearity $i(a^\dagg)^2a^2$, which plays a crucial role in many physical systems including Josephson junctions  \plainrefs{Weiss:2015,Larson:2021}  or Fabry-Pérot cavities \plainrefs{Gong:2009}.
		\item[$\boldsymbol{\hat{A}_1^\perp}$:] This space includes all remaining nonlinearties, that is, those operators with $\gamma$ such that $|\gamma|\geq3$ and $\alpha\neq\beta$. An example is $g_-^{(3,0)}=a^3-(a^\dagg)^3$.
	\end{itemize}

	\subsection{Basic definitions}
	
	We conclude this section by introducing several basic maps and concepts that will be used extensively throughout this work. These are presented abstractly without immediate reference to their applications. Their significance will become apparent when properly contextualized. 
	
	\begin{definition}[Basic maps]\label{def:basic:maps}
		Let $\gamma=(\alpha,\beta)\in\Np{2}$ be a (not necessarily well-ordered) multi-index, where $\alpha,\beta$ are two non-negative integers. We introduce:
		\begin{enumerate}
			\item The binary operation $\circ:\Np{2}\times\Np{2}\to\Np{2}$ defined as the element-wise multiplication $\gamma\circ\gamma'=(\alpha,\beta)\circ(\alpha',\beta')\mapsto(\alpha\alpha',\beta\beta')$.
			\item The map $\Theta:\Np{2}\to\Np{2}$ defined by $\Theta(\gamma)=\gamma$ if $\gamma\geq\gamma^\dagg$ and $\Theta(\gamma)=\gamma^\dagg$ otherwise.
			\item The map $E:\Np{2}\to\{+1,-1\}$ defined by $E(\gamma)=1$ if $\gamma\geq\gamma^\dagg$ and $E(\gamma)=-1$ otherwise.
			\item The map $\chi:\Np{2}\to\Z,\gamma\mapsto\chimap{\gamma}=\chimap{(\alpha,\beta)}:=\alpha-\beta$.
		\end{enumerate}
	\end{definition}
	
	\begin{definition}[Center and centralizer]
		Let $T\subseteq\hat{A}_1$ and $S\subseteq\hat{A}_1$. The \emph{centralizer} of $T$ with respect to $S$ is denoted by $C_S(T)$ and is defined as $$C_S(T):=\{s\in S\,\mid\,[s,t]=0\text{ for all }t\in T\}.$$ If $T=\{t\}$ is a singleton, we write $C_S(t)$ instead of $C_S(\{t\})$. The \emph{center} of a subalgebra $\g\subseteq\hat{A}_1$, denoted $\mathcal{Z}(\g)$, is defined as: $\mathcal{Z}(\g):=C_\g(\g)$.
	\end{definition}

	\begin{definition}\label{def:equivalence:relation:polynomials}
		Let $x,y\in\hat{A}_1$ be two polynomials. We write $x\upto{d} y$ for $d\in\N_{\geq0}$ if and only if $\deg(x{-}y)<d$. This relation was introduced and extensively used in the foundational work   \plainrefs{Bruschi:Xuereb:2024}.
	\end{definition}
	
	It is straightforward to verify that Definition~\ref{def:equivalence:relation:polynomials} describes an equivalence relation that identifies two polynomials $x$ and $y$ that have the exact same elements of degree $\geq d$. Example~\ref{exa:highlighting:equivalence:relation} highlights the fact that two polynomials must not be equal to satisfy this equivalence relation. On the contrary, this equivalence relation allows us to add terms of lower degree to the terms without changing the validity of the relation.
	\begin{tcolorbox}[breakable, colback=Cerulean!3!white,colframe=Cerulean!85!black,title=Example: Equivalnce up to degree $d$]
		\begin{example}\label{exa:highlighting:equivalence:relation}
			Let $x= i(a^\dagg a)^2$, $y=i(a^\dagg)^2 a^2$, and $z=i(a^2+(a^\dagg)^2)$. Then $x\upto{4}y$ and $x\upto{3}y$, since $x=i(a^\dagg)^2a^2+ia^\dagg a$. One has furthermore $x\upto{4}x+z$, since $\deg(z)=2$, which highlights the fact that for any $p,q\in \hat{A}_1$, one has $p\upto{\deg(p)}p+q$ if $\deg(q)<\deg(p)$. 
		\end{example}
		
	\end{tcolorbox}
	
	We now present several properties of the objects defined above. The corresponding proofs are omitted, as they follow directly from standard arguments. We invite the interested reader to verify them independently.
	
	\begin{proposition}\label{prop:basic:maps}
		Let $\gamma,\gamma'\in\Np{2}$ be two well-ordered multi-indices. Then: (a) $\chi:\Np{2}\to\N_{\geq0}$ is a linear map; (b) $\N_=^2=\ker\chi$; (c) $\gamma+\gamma^\dagg\in \N_=^2$; (d) $\chimap{\gamma\circ\gamma'}=\chimap{\gamma'\circ\gamma}$; (e) $\chimap{\gamma^\dagg}=-\chimap{\gamma}$.
	\end{proposition}

	\begin{proposition}\label{prop:equivalence:relation:linear:combination}
		Let $x_1,x_2,y_1,y_2\in\hat{A}_1$ with $x_1\upto{d_1}y_1$ and $x_2\upto{d_2}y_2$. Then $x_1+x_2\upto{d}y_1+y_2$, where $d=\max\{d_1,d_2\}.$
	\end{proposition}
	
	\begin{proof}
		One computes $\deg(x_1+x_2-(y_1{+}y_2))=\max\{\deg(x_1{-}y_1),\deg(x_2{-}y_2)\}<\max\{d_1,d_2\}=d$. Hence, by Definition~\ref{def:equivalence:relation:polynomials}, one has $x_1+x_2\upto{d}y_1+y_2$.
	\end{proof}
	
	The property stated in Proposition~\ref{prop:equivalence:relation:linear:combination} is essential when working with polynomials, as it allows to deduce equivalence relations for entire polynomials from those of their individual monomial components. Due to its ubiquity and simplicity, this fact is often used without explicit reference. 
	
	\begin{proposition}\label{prop:equivalnce:relation:commutator}
		Let $g_\sigma^\gamma\in\hat{A}_1$ be a monomial, and let $u,u'\in\hat{A}_1$ be two arbitrary polynomials with unique decompositions 
		\begin{align*}
			u=\sum_{p\in\mathcal{P}}c_pg_{\sigma_p}^{\gamma_p}+\sum_{q\in\mathcal{Q}}\hat{c}_qg_{\hat{\sigma}_q}^{\hat{\gamma}_q}\quad\text{ and }\quad u'=\sum_{p'\in\mathcal{P}'}c_{p'}'g_{\sigma_{p'}'}^{\gamma_{p'}'}+\sum_{q'\in\mathcal{Q}'}\hat{c}_{q'}'g_{\hat{\sigma}_{q'}'}^{\hat{\gamma}_{q'}'},
		\end{align*}
		where $c_p,\hat{c}_q,c_{p'}',\hat{c}_{q'}'\neq0$ for all $p\in\mathcal{P},q\in\mathcal{Q},p'\in\mathcal{P}',q'\in\mathcal{Q}'$, $|\gamma_p|=x=\deg(u)>|\hat{\gamma}_q|$ for all $p\in \mathcal{P},q\in\mathcal{Q}$, and $|\gamma_{p'}'|=x'=\deg(u')>|\hat{\gamma}_{q'}'|$ for all $p'\in\mathcal{P}',q'\in\mathcal{Q}'$. Then:
		\begin{subequations}
			\begin{align}\label{first:claim:here}
				[u,g_\sigma^\gamma]\upto{d}\sum_{p\in\mathcal{P}}c_p[g_{\sigma_p}^{\gamma_p},g_\sigma^\gamma],
			\end{align}
			\begin{align}\label{second:claim:here}
				[u,u']\upto{f}\sum_{p\in\mathcal{P}}\sum_{p'\in\mathcal{P}'}c_pc_{p'}'[g_{\sigma_p}^{\gamma_p},g_{\sigma_{p'}}^{\gamma_{p'}'}],
			\end{align}
		\end{subequations}	
		where $d=|\gamma|+x-2$ and $f=x+x'-2$.
	\end{proposition}
	
	\begin{proof}
		The first claim \eqref{first:claim:here} follows from Definition~\ref{def:equivalence:relation:polynomials} and Proposition 13 in \plainrefs{Bruschi:Xuereb:2024}. We compute:
		\begin{align*}
			\deg\left([u,g_\sigma^\gamma]-\sum_{p\in\mathcal{P}}c_p[g_{\sigma_p}^{\gamma_p},g_\sigma^\gamma]\right)=\deg\left(\sum_{q\in\mathcal{Q}}\hat{c}_q[g_{\hat{\sigma}_q}^{\hat{\gamma}_q},g_\sigma^\gamma]\right)\leq\max_{q\in\mathcal{Q}}\{|\hat{\gamma}_q|\}+|\gamma|-2<x+|\gamma|-2=d.
		\end{align*}
		For convenience, the indices $p$, $q$, $p'$, and $q'$ shall now only be taken from the index-sets $\mathcal{P}$, $\mathcal{Q}$, $\mathcal{P}'$, and $\mathcal{Q}'$ respectively. To prove the second claim \eqref{second:claim:here} we compute:
		\begin{align*}
			[u,u']=\sum_pc_p[g_{\sigma_p}^{\gamma_p},u']+\sum_q\hat{c}_q[g_{\hat{\sigma}_q}^{\hat{\gamma}_q},u']\upto{f}\sum_{p,p'}c_pc_{p'}'[g_{\sigma_p}^{\gamma_p},g_{\sigma_{p'}}^{\gamma_{p'}'}]+\sum_{q,p'}\hat{c}_qc_{p'}'[g_{\hat{\sigma}_q}^{\hat{\gamma}_q},g_{\sigma_{p'}}^{\gamma_{p'}'}],
		\end{align*}
		where we used the linearity of the commutator and the first part of the statement, as well as Proposition~\ref{prop:equivalence:relation:linear:combination} by which $f=\max\{\max_{p,p'}\{|\gamma_p|+|\gamma_{p'}'|-2\},\max_{q,p'}\{|\hat{\gamma}_q|+|\gamma_{p'}'|-2\}\}=x+x'-2$. The simple computation
		\begin{align*}
			\deg\left(\sum_{p,p'}c_pc_{p'}'[g_{\sigma_p}^{\gamma_p},g_{\sigma_{p'}}^{\gamma_{p'}'}]-\sum_{p,p'}c_pc_{p'}'[g_{\sigma_p}^{\gamma_p},g_{\sigma_{p'}}^{\gamma_{p'}'}]-\sum_{q,p'}\hat{c}_qc_{p'}'[g_{\hat{\sigma}_q}^{\hat{\gamma}_q},g_{\sigma_{p'}}^{\gamma_{p'}'}]\right)&=\deg\left(\sum_{q,p'}\hat{c}_qc_{p'}'[g_{\hat{\sigma}_q}^{\hat{\gamma}_q},g_{\sigma_{p'}}^{\gamma_{p'}'}]\right)\leq\max_{q,p'}\{|\hat{\gamma}_q|+|\gamma_{p'}'|-2\}<f
		\end{align*}
		completes the proof.
	\end{proof}

	\subsection{Commutators of monomials}
	Before proceeding with the study of the Lie algebras, it is useful to derive an explicit expression for the commutators of two arbitrary monomials in $\hat{A}_1$.
	
	\begin{lemma}\label{lem:commutator:form}
		Let $g_\sigma^\gamma,g_{\hat{\sigma}}^{\hat{\gamma}}\in\hat{A}_1$ be two arbitrary monomials. Then, to leading order in the (maximal) degree $d=|\gamma|+|\hat{\gamma}|-2$, the commutator of $g_\sigma^\gamma$ and $g_{\hat{\sigma}}^{\hat{\gamma}}\in\hat{A}_1$ satisfies:
		\begin{align*}
			[g_+^{\hat{\gamma}},g_-^{\gamma}]&\upto{d}\chimap{\hat{\gamma}\circ\gamma^\dagg}g_{+}^{\gamma+\hat{\gamma}-\tau}-\chimap{\hat{\gamma}\circ\gamma}g_+^{\Theta(\gamma+\hat{\gamma}^\dagg)-\tau},\;&\;[g_+^{\hat{\gamma}},g_+^{\hat{\gamma}}]&\upto{d}-\chimap{\hat{\gamma}\circ\gamma^\dagg}g_-^{\gamma+\hat{\gamma}-\tau}-\chimap{\hat{\gamma}\circ\gamma}E(\hat{\gamma}+\gamma^\dagg)g_-^{\Theta(\gamma+\hat{\gamma}^\dagg)-\tau},\\
			[g_-^{\hat{\gamma}},g_+^\gamma]&\upto{d}\chimap{\hat{\gamma}\circ\gamma^\dagg}g_+^{\gamma+\hat{\gamma}^\dagg-\tau}+\chimap{\hat{\gamma}\circ\gamma}g_+^{\Theta(\gamma+\hat{\gamma}^\dagg)-\tau},\;&\;[g_-^{\hat{\gamma}},g_-^\gamma]&\upto{d}\chimap{\hat{\gamma}\circ\gamma^\dagg}g_-^{\gamma+\hat{\gamma}-\tau}-\chimap{\hat{\gamma}\circ\gamma}E(\hat{\gamma}+\gamma^\dagg)g_-^{\Theta(\gamma+\hat{\gamma}^\dagg)-\tau}.
		\end{align*}        
		These four possible outcomes can be summarized in the following expression:
		\begin{align*}
			[g_{\hat{\sigma}}^{\hat{\gamma}},g_\sigma^\gamma]\upto{d}\hat{\epsilon}\left(\hat{\alpha}\beta-\hat{\beta}\alpha\right)g_{-\sigma\hat{\sigma}}^{\gamma+\hat{\gamma}-\tau}+\hat{\epsilon}'\left(\hat{\alpha}\alpha-\hat{\beta}\beta\right)g_{-\sigma\hat{\sigma}}^{\Theta(\gamma+\hat{\gamma}^\dagg)-\tau}=\hat{\epsilon}\chimap{\hat{\gamma}\circ\gamma^\dagg}g_{-\sigma\hat{\sigma}}^{\gamma+\hat{\gamma}-\tau}+\hat{\epsilon}'\chimap{\hat{\gamma}\circ\gamma}g_{-\sigma\hat{\sigma}}^{\Theta(\gamma+\hat{\gamma}^\dagg)-\tau},
		\end{align*}
		$\hat{\epsilon},\hat{\epsilon}'\in\{\pm\}$ are sign-factors depending on $\sigma$, $\hat{\sigma}$, and the result of $\Theta(\gamma+\hat{\gamma}^\dagg)$.\footnote{The general form of the expression above for $\hat{A}_n$ has been presented in the literature \plainrefs{Bruschi:Xuereb:2024}.}
	\end{lemma}

	\begin{proof}
		This is an immediate application of Lemma 35 from \plainrefs{Bruschi:Xuereb:2024}, restricted to the single-mode case. Note that we changed the multi-index of the second monomial from $\Theta(\hat{\gamma}+\gamma^\dagg)-\tau$, as written in \plainrefs{Bruschi:Xuereb:2024}, to $\Theta(\gamma+\hat{\gamma}^\dagg)-\tau$ for the sake of convenience. This modification is justified by the observation $\Theta(\gamma')=\Theta(\gamma'{}^\dagg)$ for any well-ordered multi-index $\gamma'\in\Np{2}$.
	\end{proof}
	
	\begin{lemma}\label{important:monomial:commutator:lemma}
		Let $g_\sigma^\gamma,g_{\hat{\sigma}}^{\hat{\gamma}}\in\hat{A}_1$ be two non-vanishing monomials. Then, either (a) $[g_{\hat{\sigma}}^{\hat{\gamma}},g_\sigma^\gamma]=0$ or (b) $\deg([g_{\hat{\sigma}}^{\hat{\gamma}},g_\sigma^\gamma])=|\gamma|+|\hat{\gamma}|-2$, but not both.
	\end{lemma}
	
	\begin{proof}
		This is a reduced version of Lemma 36 from \plainrefs{Bruschi:Xuereb:2024}.
	\end{proof}
	While the original version of Lemma~\ref{important:monomial:commutator:lemma} presented in \plainrefs{Bruschi:Xuereb:2024} is correct, a subtle aspect that we present here was overlooked in its proof. When considering the right-hand side of the expressions (a)-(d) from Lemma~35, the monomials {$g_-^{\gamma+\lambda-\tau_\nu}$} and {$g_-^{\Theta(\gamma+\lambda^\dagg)-\tau_\nu}$} may, in principle, be ill-defined if they are of the form {$g_-^{\tilde{\gamma}}$} with $\tilde{\gamma}\in\N_{=}^{2n}$. To address this, we first note that by extending the definition of $g_-^{\gamma}$ in \plainrefs{Bruschi:Xuereb:2024} to also include multi-indices $\gamma\in\N_=^{2n}$, the monomials with these multi-indices vanish. Thus, one must verify that the monomials in the expressions (a)-(d) from Lemma~35 only vanish when their prefactors also vanish, an issue that we address here. Assume without loss of generality {$g_-^{\gamma+\lambda-\tau_\nu}$}$\,=0$ (the case for {$ g_-^{\Theta(\gamma+\lambda^\dagg)-\tau_\nu}$}$\,=0$ is analogous). This implies $\gamma+\lambda\in \N_{=}^{2n}$, since $\tau_\nu\in \N_{=}^{2n}$ for all $\nu\in\{1,\ldots,n\}$. Consequently, $\alpha_\nu+\mu_\nu=\beta_\nu+\nu_\nu$ for all $\nu\in\{1,\ldots,n\}$. Given that the multi-indices $\gamma$ and $\lambda$ are well-ordered, if $\gamma\notin \Np{=}$, then there exists an index $k$ such that $\alpha_k>\beta_k$ and $\alpha_{k'}=\beta_{k'}$ for all $k'<k$. To satisfy the equation $\alpha_k+\beta_k=\mu_k+\nu_k$, one must have $\mu_k<\nu_k$, implying $\lambda\notin\N_{=}^{2n}$. But then there exists an index $j<k$ such that $\mu_{j}>\hat{\nu}_{j}$, which would require $\alpha_{j}<\beta_{j}$, a contradiction. Thus, $\gamma\in\N_{=}^{2n}$, and by symmetry, $\lambda\in\N_{=}^{2n}$ as well. In this case, however, the prefactor $\alpha_\nu\nu_\nu-\beta_\nu\mu_\nu$ identically vanishes, ensuring consistency.
	
	\section{Finite-dimensional Lie algebras generated by sets of monomials\label{section:monomial:generated:algebras}}
	In this section we proceed with the first main part of our work. We examine all possible subalgebras of the skew-hermitian Weyl algebra $\hat{A}_1$ that are generated by a finite set $\mathcal{G}:=\{g^{\gamma _p}_{\sigma _p}\,\mid\,p\in \mathcal{R}, \,|\mathcal{R}|<\infty\}$ of single monomials $g^{\gamma _p}_{\sigma _p}$ with the aim of determining whether the resulting Lie algebra is finite dimensional. Our goal is to provide a complete list of all non-abelian and finite-dimensional Lie algebras $\g$ that are generated in this fashion.

	\subsection{Characteristic examples and crucial tools}
	We begin by providing preliminary results and tools to tackle the main question of this section. In particular, we analyze the result of commuting elements in two different sets of monomial generators, see Lemma~\ref{lem:commutator:form}. This might help to determine whether the Lie algebras generated by these sets are finite dimensional. We also identify the key subalgebras with well-known Lie algebras from the literature, such as the Schr\"odinger algebra $\mathcal{S}\cong\slR{2}\ltimes\gh_1$ or the special linear algebra $\slR{2}$. 
	We start by defining some of the most well-known Lie algebras
	
	\begin{definition}[Heisenberg algebra]
		The real \emph{Heisenberg algebra} $\gh_1$ is the three-dimensional Lie algebra with basis $\{q,p,z\}$, defined by the commutation relations $[q,p]=z$ and $[z,\gh_1]=0$, see \plainrefs{Kac:1990}.
	\end{definition}
	
	\begin{definition}[Wigner-Heisenberg algebra $\wh_2$]\label{def:Wigner:Heisenberg:algebra:2}
		The real \emph{Wigner-Heisenberg algebra of the second kind}, denoted $\wh_2$, is the four-dimensional Lie algebra with basis $\{h,q,p,z\}$, defined by the commutation relations $[h,q]=-p$, $[h,p]=q$ $[q,p]=z$, and $[z,\wh_2]=0$, see \plainrefs{Yang:1951}.
	\end{definition}
	
	The Wigner-Heisenberg algebra sometimes appears in the literature with other names, one example being $A_{4,9}^{a=0}$, see \plainrefs{Popovych:2003}. In the original formulation \plainrefs{Wigner:1950}, the commutation relations were expresses differently as: $[H,ix]=v$, $[H,v]=ix$, $[v,ix]=1$. By identifying $e_4\leftrightarrow H$, $e_3\leftrightarrow ix+v$, $e_2\leftrightarrow ix-v$, and $e_1\leftrightarrow -2$, one observes that this algebra, over the field of real numbers, is not isomorphic to $\wh_2$ and clearly refers to a different algebra then its modern incarnation \plainrefs{Chung:2023}. That is the reason why we add the descriptive \enquote{of second kind} to differentiate it from the original algebra envisioned by E. Wigner \plainrefs{Wigner:1950}. In this work we shall omit, for notational convenience, the descriptive notation \enquote{of second kind}, and shall write instead \textit{The Wigner-Heisenberg algebra $\wh_2$}. For completeness, we also define the original Wigner-Heisenberg algebra here.
	
	\begin{definition}[Wigner-Heisenberg algebra $\wh_1$]\label{def:Wigner:Heisenberg:algebra:1}
		The real \emph{Wigner-Heisenberg algebra of the first kind}, denoted by $\wh_1$, is the four-dimensional Lie algebra with basis $\{e_1,e_2,e_3,e_4\}$, defined by the commutation relations $[e_2,e_3]=e_1$, $[e_2,e_4]=e_2$ $[e_3,e_4]=-e_3$ and $[e_1,\wh_1]=0$, see \plainrefs{Wigner:1950}.
	\end{definition}
	
	This algebra is sometimes referred to as $A_{4,8}^{b=-1}$ in the literature \plainrefs{Popovych:2003}. As shown in Proposition~\ref{prop:example:naive:assumption:wrong}, the complexifications of the two algebras $\wh_1$ and $\wh_2$ are isomorphic,
	i.e., $\wh_{1,\C}\cong\wh_{2,\C}\equiv\wh_\C$. The complex Wigner-Heisenberg algebra $\wh_\C$ is commonly used when describing the dynamics of a single quantum harmonic oscillator with a linear driving term as done, for example, in linear optics \plainrefs{Qvarfort:2025}, although often not explicitly named. 
	
	\begin{definition}[$\slR{2}$-algebra]\label{def:sl2:algebra}
		The real special linear Lie algebra $\slR{2}$ is the three-dimensional Lie algebra with basis $\{h,x,y\}$, defined by the commutation relations $[h,x]=2x$, $[h,y]=-2y$, and $[x,y]=h$.
	\end{definition}
	
	This algebra is isomorphic to several other well-known real Lie algebras, such as $\slR{2}\cong\su{1,1}\cong\sso{2,1}\cong\spR{2}$, see \plainrefs{helgason2024differential}. However, it is not isomorphic to the compact real Lie algebras $\su{2}\cong\sso{3}$, see \plainrefs{Knapp:1986}.
	We are now ready prove a few initial results related to these concepts. Following Definition~\ref{def:gens:mon:partition}, $\mathcal{G}$
	denotes a set of monomial generators with disjoint sets 
	$\mathcal{G}^K=\mathcal{G}\cap \hat{A}_1^K$ where $K\in\{0,1,2,=,\perp\}$ and $\hat{A}_1^K$ is defined in Sec.~\ref{sec:dec:skew}.
	
	\begin{proposition}\label{prop:commutator:g0:g=:g0:g=}
		A set $\mathcal{G}$  of monomial generators with $\mathcal{G}=\mathcal{G}^0\cup\mathcal{G}^=$
		generates an abelian Lie algebra. That is, $[\mathcal{G}, \mathcal{G}] \subseteq[\hat{A}_1^0\oplus\hat{A}_1^=,\hat{A}_1^0\oplus\hat{A}_1^=]\subseteq \{0\}$ and $\langle \mathcal{G}\rangle_{\mathrm{Lie}}=\spn(\mathcal{G})$ is an abelian Lie algebra.
	\end{proposition}	
	
	\begin{proof}
		This follows directly from Lemma 16 (a4) in \plainrefs{Bruschi:Xuereb:2024} and the bi-linearity of the commutator.
	\end{proof}
	
	\begin{proposition}\label{prop:A11:algbera:heisenberg}
		The Lie algebra generated by a set 
		$\mathcal{G}$ of monomial generators with $\mathcal{G}=\mathcal{G}^1\subseteq\hat{A}_1^1$ is isomorphic to a subalgebra of the real Heisenberg algebra $\gh_1$. That is, $\langle\mathcal{G}\rangle_{\mathrm{Lie}} \subseteq \langle \hat{A}_1^1\rangle_{\mathrm{Lie}}= \spn\{2i, i(a+a^{\dagg}), a-a^{\dagg}\}\cong\gh_1$. 
	\end{proposition}
	
	\begin{proof}
		We compute the commutator $[a-a^{\dagg}, i(a+a^{\dagg})]=2i$. The element $2i$ is clearly in the Lie center of $\hat{A}_1$. This shows $\langle\mathcal{G}^1\rangle_{\mathrm{Lie}} \subseteq \langle \hat{A}_1^1\rangle_{\mathrm{Lie}}= \spn\{2i, i(a+a^{\dagg}), a-a^{\dagg}\}$. Identifying, $2i\leftrightarrow z$, $a-a^\dagg\leftrightarrow q$, and $i(a+a^\dagg)\leftrightarrow p$, we recover the defining relations of the real Heisenberg algebra $\gh_1$. Hence, $\lie{\hat{A}_1^1}\cong\gh_1$.
	\end{proof}
	
	\begin{definition}
		To avoid cumbersome notation in the subsequent proofs, we introduce the following spaces: $\hat{A}_1^{\leq1}:=\hat{A}_1^0\oplus\hat{A}_1^{1}$ and $\hat{A}_1^{\leq 2}:=\hat{A}_1^{\leq 1}\oplus\hat{A}_1^2=\hat{A}_1^0\oplus\hat{A}_1^{1}\oplus\hat{A}_1^2$. Moreover, starting from a set $\mathcal{G}$ of monomial generators, we define the subsets $\mathcal{G}^{\leq1}$ and $\mathcal{G}^{\leq 2}$ via the intersections $\mathcal{G}\cap \hat{A}_1^{\leq1}$ and $\mathcal{G}\cap\hat{A}_1^{\leq 2}$ respectively. 
	\end{definition}
	
	\begin{proposition}\label{prop:algebra:go:g1:space}
		The Lie algebra generated by a set $\mathcal{G}$ of monomial generators with 
		$\mathcal{G}=\mathcal{G}^0\cup\mathcal{G}^1\subseteq\hat{A}_1^0\oplus\hat{A}_1^1 = \hat{A}_1^{\leq 1}$ is isomorphic to a subalgebra of the real Wigner-Heisenberg algebra $\wh_2$. That is, $\langle\mathcal{G}\rangle_{\mathrm{Lie}} \subseteq \langle \hat{A}_1^{\leq1}\rangle_{\mathrm{Lie}}= \spn\{i,ia^\dagg a, i(a{+}a^{\dagg}), a{-}a^{\dagg}\}\cong\wh_2$.
	\end{proposition}
	
	\begin{proof}
		We compute the relevant commutators. We have: $[a-a^{\dagg}, ia^{\dagg}a]=i(a{+}a^{\dagg})$, $[i(a{+}a^{\dagg}), ia^{\dagg}a]=-(a{-}a^{\dagg})$, and $[a{-}a^\dagg,i(a{+}a^\dagg)]=2i$. The remaining commutators either vanish trivially or follow from the antisymmetric property of the commutators. This implies $\langle\mathcal{G}\rangle_{\mathrm{Lie}} \subseteq \langle \hat{A}_1^{\leq1}\rangle_{\mathrm{Lie}}= \spn\{i,ia^\dagg a, i(a{+}a^{\dagg}), a{-}a^{\dagg}\}$.
		Identifying the monomials as: $h\leftrightarrow ia^\dagg a$, $q\leftrightarrow(a{-}a^\dagg)$, $p=i(a{+}a^\dagg)$, and $z\leftrightarrow 2i$, we recover the defining relations of the Wigner-Heisenberg algebra $\wh_2$ as given in Definition~\ref{def:Wigner:Heisenberg:algebra:2}, therefore implying $\langle\hat{A}_1^{\leq1}\rangle_{\mathrm{Lie}}\cong\wh_2$.
	\end{proof}
	
	\begin{proposition}\label{prop:algebra:g2:space}
		The Lie algebra generated by a set $\mathcal{G}$ of monomial generators with $\mathcal{G}=\mathcal{G}^2\subseteq\hat{A}_1^2$ is isomorphic to a subalgebra of the real special linear algebra $\slR{2}$. That is, $\langle\mathcal{G}\rangle_{\mathrm{Lie}} \subseteq \langle\hat{A}_1^2\rangle_{\mathrm{Lie}}=\spn\{i(a^{\dagg}a+\frac{1}{2}), i(a^2 + (a^{\dagg})^2), a^2 - (a^{\dagg})^2\}\cong\slR{2}$.
	\end{proposition}
	
	Before proceeding with the proof of this claim, we emphasize the following important observation: The Lie algebra generated by the space $\hat{A}_n^2$ for $n$-modes is the real symplectic Lie algebra $\mathfrak{sp}(2n,\R)$ \plainrefs{Bruschi:Xuereb:2024,Sattinger:1986,Ripka:1986}. In the single-mode case, where $n=1$, this Lie algebra is isomorphic to the special linear algebra $\slR{2}$ \plainrefs{helgason2024differential}. Consequently, it would be natural to denote the Lie algebra generated by $\hat{A}_n^2$ with $\mathfrak{sp}(2,\R)$ rather than $\slR{2}$. Nevertheless, in Proposition~\ref{prop:realizable:semisimple:algebras}, we demonstrate that the only finite-dimensional Lie subalgebras of $\hat{A}_1$ that are semisimple must be three-dimensional. Therefore, we adopt the conventional notation prevalent in the literature, referring to the two real three-dimensional semisimple Lie algebras as the special unitary Lie algebra $\mathfrak{su}(2)$ and the special linear algebra $\slR{2}$.
	
	\begin{proof}
		We compute the relevant commutators: $[i(a^2{+}(a^{\dagg})^2),a^2{-}(a^{\dagg})^2]=-8i(a^{\dagg}a{+}\frac{1}{2})$, $[i(a^{\dagg}a{+}\frac{1}{2}),i(a^2{+}(a^{\dagg})^2)]=2(a^2{-}(a^{\dagg})^2)$, and $[i(a^{\dagg}a{+}\frac{1}{2}),a^2{-}(a^{\dagg})^2]=-2i(a^2{+}(a^{\dagg})^2)$. This proves the assertion $\langle\mathcal{G}\rangle_{\mathrm{Lie}} \subseteq \langle\hat{A}_1^2\rangle_{\mathrm{Lie}}$ where $\langle\hat{A}_1^2\rangle_{\mathrm{Lie}}=\spn\{i(a^{\dagg}a{+}\frac{1}{2}), i(a^2 {+} (a^{\dagg})^2), a^2 {-} (a^{\dagg})^2\}$. This leaves proving $\langle\hat{A}_1^2\rangle_{\mathrm{Lie}}\cong\slR{2}$. We can clearly choose the basis
		\begin{align*}
			\mathcal{B}=\left\{H:=i\left(a^\dagg a+\frac{1}{2}\right)+\frac{1}{2}\left(g_+^{2\iota_1}+g_-^{2\iota_1}\right),\, X:=\frac{1}{2\sqrt{2}}\left(g_+^{2\iota_1}+2i\left(a^\dagg a+\frac{1}{2}\right)\right),\;Y:=-\frac{1}{2\sqrt{2}}\left(g_-^{2\iota_1}+2i\left(a^\dagg a+\frac{1}{2}\right)\right)\right\}\subseteq\langle\hat{A}_1^2\rangle_{\mathrm{Lie}}.
		\end{align*}
		It is immediate to verify that $[H,X]=2X$, $[H,Y]=-2Y$ and $[X,Y]=H$. The details of this calculation are left as an exercise to the reader.
		The identification $H\leftrightarrow h$, $X\leftrightarrow x$, and $Y\leftrightarrow y$ concludes this proof.
	\end{proof}
	
	We have seen that the commutation relations among the basis elements of $\hat{A}_1^0\oplus\hat{A}_1^1\oplus\hat{A}_1^2:=\hat{A}_1^{\leq2}$, i.e., the monomials $i=\frac{1}{2}g_+^0$, $ia^\dagg a=\frac{1}{2}g_+^{\tau}$, $i(a+a^\dagg)=g_+^{\iota_1}$, $a-a^\dagg=g_-^{\iota_1}$, $i(a^2+(a^\dagg)^2)=g_+^{2\iota_1}$, and $a^2-(a^\dagg)^2=g_-^{2\iota_1}$, are key for understanding the interplay of the sets $\mathcal{G}^0$, $\mathcal{G}^1$, and $\mathcal{G}^2$. For completeness, we present the full set of commutation relations between these elements in Appendix~\ref{App:commutation:relations}. The explicit expressions for the individual commutators are listed in Table~\ref{tab:Full:Commutator:Algebra:(pesudo):schroedinger:algebra}.
	
	\begin{tcolorbox}[breakable, colback=Cerulean!3!white,colframe=Cerulean!85!black,title=Example: Different realization of $\slR{2}$ in $\hat{A}_1$]
		We note the following:  There exist multiple distinct subspaces of the skew-hermitian Weyl algebra $\hat{A}_1$ that are isomorphic to $\slR{2}$, when equipped with the usual commutation relations.
		\begin{example}\label{exa:different:realizations:of:sl2}
			We have seen that $\langle\hat{A}_1^2\rangle_{\mathrm{Lie}}$ is a realization of $\slR{2}$ as a subalgebra of $\hat{A}_1$. We have furthermore seen that it admits the basis
			\begin{align*}
				\mathcal{B}_1=\left\{H:=i\left(a^\dagg a{+}\frac{1}{2}\right)+\frac{1}{2}\left(g_+^{2\iota_1}{+}g_-^{2\iota_1}\right),\, X:=\frac{1}{2\sqrt{2}}\left(g_+^{2\iota_1}{+}2i\left(a^\dagg a{+}\frac{1}{2}\right)\right),\;Y:=-\frac{1}{2\sqrt{2}}\left(g_-^{2\iota_1}{+}2i\left(a^\dagg a{+}\frac{1}{2}\right)\right)\right\},
			\end{align*}
			which satisfy the commutation relations $[H,X]=2X$, $[H,Y]=2Y$, and $[X,Y]=H$. We can now furthermore compute $[H,g_+^{\iota_1}]=g_+^{\iota_1}$, $[H,g_-^{\iota_1}]=-g_-^{\iota_1}$, $[X,g_+^{\iota_1}]=0$, $[X,g_-^{\iota_1}]=-\sqrt{2}g_+^{\iota_1}$, $[Y,g_+^{\iota_1}]=-\frac{1}{\sqrt{2}}\left(g_-^{\iota_1}{+}g_+^{\iota_1}\right)$, and $[Y,g_-^{\iota_1}]=\frac{1}{\sqrt{2}}\left(g_-^{\iota_1}{+}g_+^{\iota_1}\right)$.
			Here, we used the commutation relations listed in Table~\ref{tab:Full:Commutator:Algebra:(pesudo):schroedinger:algebra}. Thus, one has $[\tilde{H},\tilde{X}]=2\tilde{X}$, $\tilde{H},\tilde{Y}=-2\tilde{Y}$, and $[\tilde{X},\tilde{Y}]=\tilde{H}$, where we have introduced
			the new elements $\tilde{H},\tilde{X},\tilde{Y}$ that form the basis
			\begin{align*}
				\mathcal{B}_2:=\left\{\tilde{H}:=H+\frac{1}{\sqrt{2}}\left(g_+^{\iota_1}+g_-^{\iota_1}\right),\Tilde{X}:=X+g_+^{\iota_1}+\frac{1}{\sqrt{2}}i,\Tilde{Y}:=Y\right\}.
			\end{align*}
			This basis generates a three dimensional algebra $\langle\mathcal{B}_2\rangle_{\mathrm{Lie}}$ that is, by the identification $\tilde{H}\leftrightarrow h$, $\tilde{X}\leftrightarrow x$, and $\tilde{Y}\leftrightarrow y$, isomorphic to $\slR{2}$. This example illustrates that multiple subspaces of the skew-hermitian Weyl algebra can realize the same Lie algebraic structure. Note, furthermore that the vector space $\spn\{\mathcal{B}_1\cup\mathcal{B}_2\}$ is five-dimensional and generates the six-dimensional Lie algebra $\langle\hat{A}_1^1\oplus\hat{A}_1^2\rangle_{\mathrm{Lie}}$. This is an easy example of the fact that combining two vector spaces associated to semisimple Lie algebras can generate a non-semisimple algebra (more about this is discussed with Proposition~\ref{prop:algebra:g0:g2:space}).
		\end{example}
	\end{tcolorbox}
	
	\begin{proposition}\label{prop:algebra:g0:g2:space}
		The Lie algebra generated by a set $\mathcal{G}$ of monomial generators with  $\mathcal{G}=\mathcal{G}^0\cup\mathcal{G}^2$ is isomorphic to a subalgebra of the Lie algebra $\slR{2}\oplus\R$. That is, $\langle\mathcal{G}\rangle_{\mathrm{Lie}} \subseteq \langle\hat{A}_1^0\oplus\hat{A}_1^2\rangle_{\mathrm{Lie}}=\spn\{i, ia^{\dagg}a, i(a^2 {+} (a^{\dagg})^2), a^2 {-} (a^{\dagg})^2\}\cong \slR{2}\oplus\R$.
	\end{proposition}
	
	\begin{proof}
		We compute the relevant commutators: $[ia^{\dagg}a,i(a^2+(a^{\dagg})^2)]=2(a^2-(a^{\dagg})^2)$, $[ia^{\dagg}a,a^2-(a^{\dagg})^2]=-2i(a^2+(a^{\dagg})^2)$, and $[i(a^2{+}(a^{\dagg})^2),a^2{-}(a^{\dagg})^2]=-8ia^{\dagg}a-4i$. Since $i$ commutes with any element in $\hat{A}_1$, one has now shown that $\langle\mathcal{G}\rangle_{\mathrm{Lie}} \subseteq \langle\hat{A}_1^1\oplus\hat{A}_1^2\rangle_{\mathrm{Lie}}=\spn\{i, ia^{\dagg}a, i(a^2 {+} (a^{\dagg})^2), a^2 {-} (a^{\dagg})^2\}$. Thus, one simply needs to show that $\langle\hat{A}_1^0\oplus\hat{A}_1^2\rangle_{\mathrm{Lie}}$ is isomorphic to the outer direct sum $\slR{2}\oplus\R$. We use the basis $\{H,X,Y\}$ of $\lie{\hat{A}_1^2}$ introduced in the proof of Proposition~\ref{prop:algebra:g2:space}. Identifying $H\leftrightarrow(H,0)$, $X\leftrightarrow (X,0)$, $Y\leftrightarrow (Y,0)$, $2i\leftrightarrow(0,2i)$ and recalling that the Lie bracket on $\slR{2}\oplus\R$ is defined as $[(e_1,b_1),(e_2,b_2)]=([e_1,e_2],[b_1,b_2])$ yields clearly $\lie{\hat{A}_1^0\oplus\hat{A}_1^2}\cong\slR{2}\oplus\R$.
	\end{proof}
	
	\begin{proposition}\label{prop:algebra:g0:gperp:divergence}
		Given a set $\mathcal{G}$ of monomial generators with $\mathcal{G}=\mathcal{G}^0 \cup \mathcal{G}^{\perp}$ with $|\mathcal{G}^\perp|\leq 1$.
		Then the Lie algebra $\langle\mathcal{G}\rangle_{\mathrm{Lie}}$ is finite dimensional if and only if $\mathcal{G}^0\subseteq\{i\}$.
	\end{proposition}
	
	\begin{proof}
		Assume, one has $\mathcal{G}^0\subseteq\{i\}$ and $|\mathcal{G}^\perp|\leq 1$. In this case, the set $\mathcal{G}^0\cup\mathcal{G}^\perp$ generates an abelian algebra, and is henceforth finite dimensional. This trivially follows from the fact that: (a) $[i, g_{\sigma}^{\gamma}]=0$ for any $g_\sigma^\gamma\in\mathcal{G}^\perp$; and (b) $\mathcal{G}^\perp$ contains at most one element. 
		
		Suppose, on the other hand, $ia^\dagg a\in\mathcal{G}^0$ and $\mathcal{G}^\perp\neq\emptyset$. Then, by Theorem 25 from \plainrefs{Bruschi:Xuereb:2024}, both monomials $g_+^\gamma,g_-^\gamma$ lie in the Lie algebra $\langle \mathcal{G}^0\cup\mathcal{G}^\perp\rangle_\mathrm{Lie}$ for any $g_\sigma^\gamma\in\mathcal{G}^\perp$. Such two operators allow for the construction of a Commutator Chain of type II (see Definition 47 in \plainrefs{Bruschi:Xuereb:2024}) that contains elements of strictly increasing degree (due to Theorem 53 from \plainrefs{Bruschi:Xuereb:2024}). Hence, since the whole Commutator Chain is contained in the algebra this in turn makes $\langle\mathcal{G}\rangle_{\mathrm{Lie}}$ infinite-dimensional in this case. The claim follows from the observation that $|\mathcal{G}^\perp|=1$ is a subcase of $\mathcal{G}^\perp\neq\emptyset$.
	\end{proof}	
	
	Before proceeding, we introduce a graphic-representation tool for algebras that proves very useful in this work.

	\begin{definition}[Labeled, directed commutator graph for Lie algebras]\label{directed:graphs}
		We introduce here the concept of labeled directed graphs for Lie algebras, similarly introduced in \plainrefs{Kazi:2025}. A graph is composed of vertices and edges that connect pairs of vertices. We choose the edges and vertices as follows: let $\g$  be a Lie algebra. Then, the vertices and edges of the graph represent elements of $\g$, and the labelling is such that the target vertex is obtained as the commutator between the starting vertex with the element that labels the edge, each item appearing modulo an overall multiplication constant. This means that $x\in\g$ and $c\cdot x\in\g$ label the same vertices and edges for all $x\in\g$ and $c\in\mathbb{R}\setminus\{0\}$. The basic building block reads as
		
		\vspace{0.1cm}
		{{\color{white}test}\;\hspace{3.6cm}\graphlegend{Black}}
		
		\vspace{0.1cm}
	\end{definition}
	
	We now move on defining the \textit{Schr\"odinger algebra}, which is of particular importance to our work.
	
	\begin{definition}[Schr\"odinger algebra]\label{def:Schroedinger:algebra}
		The real \emph{Schr\"odinger algebra} $\mathcal{S}$ is the six-dimensional Lie algebra with basis $\mathcal{G}:=\{h,x,y,q,p,z\}$, defined by the commutation relations:
		\begin{align*}
			[h,x]&=2x,\;&\;[h,y]&=-2y,\;&\;[x,y]&=h,\;&\;[h,q]&=q,\;&\;[h,p]&=-p,\\
			[x,q]&=0,\;&\;[x,p]&=q,\;&\;[y,q]&=p,\;&\;[y,p]&=0,\;&\;[q,p]&=z,
		\end{align*}
		while $[z,\mathcal{S}]=0$. The Schr\"odinger algebra $\mathcal{S}$ contains notably the subalgebras $\slR{2}=\langle\{h,x,y\}\rangle_{\mathrm{Lie}}$ and $\gh_1=\langle\{q,p,z\}\rangle_{\mathrm{Lie}}$, and it is the semidirect sum (see also Def.~\ref{def:semidirect:product})
		$\mathcal{S}\cong\slR{2}\ltimes\gh_1\cong\lie{\{h,x,y,q,p,z\}}$ \plainrefs{Yang:2018}.
	\end{definition}
	
	\begin{figure}[htpb]
		\centering
		\includegraphics[width=1.0\linewidth]{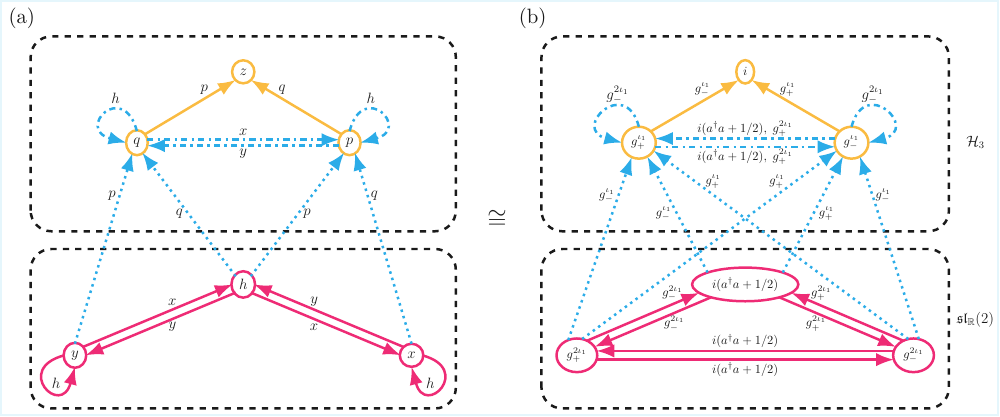}
		\caption[Depiction of the Schr\"odinger algebra using a directed graphs.]{Depiction of the Schr\"odinger algebra using a directed graphs as presented in Definition~\ref{directed:graphs}. 
			In (a) the chosen basis is the one found in Definition~\ref{def:Schroedinger:algebra}; in (b) the chosen basis is $\{i,g_+^{\iota_1},g_-^{\iota_1},i(a^\dagg a+1/2), g_+^{2\iota_1},g_-^{2\iota_1}\}$, inspired by the basis of $\hat{A}_1^0\oplus\hat{A}_1^1\oplus\hat{A}_1^2$.
			The red edges and vertices belong to the special linear algebra $\slR{2}$, the yellow ones to the Heisenberg algebra $\gh_1$, and the blue dotted edges represent commutation relations between elements from $\slR{2}$ and $\gh_1$. This highlights the structure of the Schr\"odinger algebra as a semidirect sum of the special linear algebra and the Heisenberg algebra, further emphasized by the dotted rectangles encircling the basis elements of the respective subalgebras. }
		\label{clone:trooper}
	\end{figure}
	
	We have to remark that the term \enquote{Schr\"odinger algebra} is used in the literature to refer to different but related algebraic structures \plainrefs{Dobrev:1997,Nikitin:2020,Duval:1994,Aizawa:2011,Tao:2022}.  
	The Schr\"odinger algebra, as defined here, can be regarded as the fundamental version of any of the algebras used in the literature described by this name. Generally, the Schr\"odinger algebra is related to the symmetries of the Schr\"odinger equation of a free particle. 
	
	There are now two trains of thought on extending the fundamental Schr\"odinger algebra: (i) One idea is it to extend the group of symmetries of the free Schr\"odinger equation from a $(1{+}1)$-dimensional flat spacetime to an $(n{+}1)$-dimensional spacetime. One usually writes $\sch(1,3)$ if one refers to the 13-dimensional algebra in a $(3{+}1)$-dimensional spacetime \plainrefs{Nikitin:2020}. Other sources write alternatively $\widetilde{\sch}(n)$ \plainrefs{Duval:1994} or $\gs(n)$ \plainrefs{Aizawa:2011}, and call these algebras an \textit{extended} or \textit{super Schr\"odinger algebra}, since these algebra can be obtained by extending the \textit{Galileo algebra} $\g(1,n)$ \plainrefs{Nikitin:2020}. The Schr\"odinger algebra given in Definition~\ref{def:Schroedinger:algebra} is, using this notion, the $(1{+}1)$-dimensional Schr\"odinger algebra \plainrefs{Aizawa:2011}; (ii) The other idea is to extend the fundamental Schr\"odinger algebra, also called the first Schr\"odinger algebra $\sch_1$ \plainrefs{Tao:2022}, from Definition~\ref{def:Schroedinger:algebra} to the $n$-th Schr\"odinger algebra $\sch_n$ by taking the semidirect sum of $\slR{2}$ and the $n$-th Heisenberg algebra $\gh_n$. This algebra is then spanned by the elements $\{h,x,y,q_1,\ldots,q_n,p_1,\ldots,p_n,z\}$ which satisfy the following commutation relations:
	\begin{align*}
		[h,x]&=2x,\;&\;[h,y]&=-2y,\;&\;[x,y]&=h,\;&\;[h,q_j]&=q_j,\;&\;[h,p_j]&=-p_j,\\
		[x,q_j]&=0,\;&\;[x,p_j]&=q_j,\;&\;[y,q_j]&=p_j,\;&\;[y,p_j]&=0,\;&\;[q_j,p_k]&=z\delta_{jk},
	\end{align*}
	while $[z,\sch_n]=0$, see \plainrefs{Tao:2022}. 
	
	Note, furthermore, that in the context of this work, the natural extension of the Schr\"odinger algebra to multiple bosonic modes would be the Lie algebra $\lie{\hat{A}_n^0\oplus\hat{A}_n^1\oplus\hat{A}_n^2}=\lie{\hat{A}_n^{\leq 2}}$. This algebra, however, is $(2n^2{+}3n{+}1)$-dimensional, since it contains $n$ linearly independent free Hamiltonian terms $ia_j^\dagg a_j$, the central element $i$, $2n$ single-mode squeezing operators $g_\sigma^{2\iota_j}$, $n(n{-}1)$ two-mode squeezing operators $g_\sigma^{\iota_j+\iota_k}$, $n(n{-}1)$ mode-mixing operators $g_\sigma^{\iota_j+\iota_{n+k}}$, and finally $2n$ quadrature operators $g_\sigma^{\iota_k}$. Thus, for two modes, the dimension of $\lie{\hat{A}_2^{\leq2}}$ is already fifteen, which does not agree with the other two conventions.
	
	A visualization of the Schrödinger algebra, using labeled, directed graphs, as described in Definitions~\ref{directed:graphs} can be found in Figure~\ref{clone:trooper}. This graph-representation of Lie algebras visually illustrates structural properties. For instance, $\slR{2}$ is clearly non-solvable, as every vertex $v$ contained within the graph of $\slR{2}$ can be obtained from at least one other vertex contained within using edges pointing towards $v$. This implies that its first derived algebra is $\slR{2}$ itself. Calculating the derived series corresponds to iteratively removing every vertex to which no edge points, along all associated edges, i.e., all edges that are either labeled by the removed edges or originate therefrom. It is therefore clear that the presence of a closed directed walk, where all the vertices are connected via edges labeled by vertex-elements of the walk itself, indicates that the algebra is non-solvable. This graph structure also facilitates the identification of nilpotent algebras, since computing the lower central series involves removing, at each step, every vertex to which no edge points, along with all edges that originate from those vertices, but not the edges that are labeled with these vertices. Thus, we can see that $\gh_1$ is nilpotent, as the first derived algebra is represented solely by basis element $i$, and no edge points from this vertex to another one, implying that the second element of the lower series is already the trivial one. 
	
	\begin{proposition}\label{prop:algebra:g1:g2:space}
		The Lie algebra generated by a set
		$\mathcal{G}$ of monomial generators with
		$\mathcal{G}=\mathcal{G}^1\cup\mathcal{G}^2$ is a subspace of $\hat{A}_1^0\oplus\hat{A}_1^1\oplus\hat{A}_1^2 =\hat{A}_1^{\leq 2}$. That is, $\langle\mathcal{G}\rangle_{\mathrm{Lie}} \subseteq \langle\hat{A}_1^1\oplus\hat{A}_1^2\rangle_{\mathrm{Lie}}= \spn\{ i, ia^{\dagg}a, i(a{+}a^{\dagg}), a{-}a^{\dagg}, i(a^2{+}(a^{\dagg})^2), a^2{-}(a^{\dagg})^2\}$. One has furthermore $\langle\mathcal{G}^{\leq 2}\rangle_{\mathrm{Lie}} \subseteq \langle\hat{A}_1^{\leq 2}\rangle_{\mathrm{Lie}}= \langle\hat{A}_1^1\oplus\hat{A}_1^2\rangle_{\mathrm{Lie}}$, where $\mathcal{G}^0\cup\mathcal{G}^1\cup \mathcal{G}^2:=\mathcal{G}^{\leq 2}$. 
	\end{proposition}
	
	\begin{proof}
		The relevant commutation relations are listed in Table~\ref{tab:Full:Commutator:Algebra:(pesudo):schroedinger:algebra}. These relations imply $[\mathcal{G}^1,\mathcal{G}^2]\subseteq\hat{A}_1^1$ and more generally $[\hat{A}_1^1,\hat{A}_1^2]\subseteq\hat{A}_1^1$. From Propositions~\ref{prop:algebra:g0:g2:space} and~\ref{prop:A11:algbera:heisenberg}, we also know: $[\mathcal{G}^2,\mathcal{G}^2]\subseteq\langle\mathcal{G}^0\cup \mathcal{G}^2\rangle_{\mathrm{Lie}}\subseteq\hat{A}_1^0\oplus\hat{A}_1^2$ and $\langle\mathcal{G}^1\rangle_{\mathrm{Lie}}\subseteq\hat{A}_1^0\oplus\hat{A}_1^1$. Combining these, we conclude:  $\langle\mathcal{G}\rangle\subseteq\hat{A}_1^{\leq 2}=\spn\{ i, ia^{\dagg}a, i(a{+}a^{\dagg}), a{-}a^{\dagg}, i(a^2{+}(a^{\dagg})^2), a^2{-}(a^{\dagg})^2\}$ and furthermore $\langle\hat{A}_1^1\oplus\hat{A}_1^2\rangle_{\mathrm{Lie}}=\langle\hat{A}_1^{\leq 2}\rangle_{\mathrm{Lie}}= \spn\{ i, ia^{\dagg}a, i(a{+}a^{\dagg}), a{-}a^{\dagg}, i(a^2{+}(a^{\dagg})^2), a^2{-}(a^{\dagg})^2\}$.
	\end{proof}
	
	\begin{corollary}\label{cor:decomposition:a0:a1:a2:space}
		The space $\lie{\hat{A}_1^0\oplus\hat{A}_1^1\oplus\hat{A}_1^2}=\lie{\hat{A}_1^{\leq 2}}$ can be decomposed into a direct sum of the spaces $\lie{\hat{A}_1^1}$ and $\lie{\hat{A}_1^2}$, i.e., $\lie{\hat{A}_1^{\leq 2}}=\lie{\hat{A}_1^1}\oplus\lie{\hat{A}_1^2}$.
	\end{corollary}
	
	\begin{proof}
		By Proposition~\ref{prop:algebra:g1:g2:space}, we obtain $\langle\hat{A}_1^{\leq 2}\rangle_{\mathrm{Lie}}= \spn\{ i, ia^{\dagg}a, i(a{+}a^{\dagg}), a{-}a^{\dagg}, i(a^2{+}(a^{\dagg})^2), a^2{-}(a^{\dagg})^2\}$. Furthermore Propositions~\ref{prop:A11:algbera:heisenberg} and~\ref{prop:algebra:g2:space} state that $\lie{\hat{A}_1^1}=\spn\{i, i(a+a^{\dagg}), a-a^{\dagg}\}$ and $\lie{\hat{A}_1^2}=\spn\{i(a^\dagg a+1/2),i(a^2+(a^{\dagg})^2), a^2-(a^{\dagg})^2\}$. It is evident that $\lie{\hat{A}_1^1}\cap \lie{\hat{A}_1^2}=\{0\}$, and that the union of the two spans cover the full space: $\spn\{i, i(a+a^{\dagg}), a-a^{\dagg}\}\oplus\spn\{i(a^\dagg a+1/2),i(a^2+(a^{\dagg})^2), a^2-(a^{\dagg})^2\}=\spn\{ i, ia^{\dagg}a, i(a+a^{\dagg}), a-a^{\dagg}, i(a^2+(a^{\dagg})^2), a^2-(a^{\dagg})^2\}$. Hence, we conclude: $\lie{\hat{A}_1^{\leq 2}}=\lie{\hat{A}_1^1}\oplus\lie{\hat{A}_1^2}$.
	\end{proof}

	\begin{proposition}\label{prop:schroedinger:algebra}
		The Lie algebra generated by the vector space $\hat{A}_1^0\oplus\hat{A}_1^1\oplus\hat{A}_1^2=\hat{A}_1^{\leq2}$ is isomorphic to the Schr\"odinger algebra $\mathcal{S}$, that is, $\langle\hat{A}_1^{\leq2}\rangle_{\mathrm{Lie}}\cong\mathcal{S}=\slR{2}\ltimes\gh_1$.
	\end{proposition}
	
	\begin{proof}
		In Proposition~\ref{prop:algebra:g1:g2:space}, we have seen that $\langle\hat{A}_1^{\leq2}\rangle_{\mathrm{Lie}}=\lie{\hat{A}_1^1\oplus\hat{A}_1^2}=\spn\{i,ia^\dagg a,i(a+a^\dagg),a-a^\dagg, i(a^2)+(a^\dagg)^2),a^2-(a^\dagg)^2\}$. We define now the elements $\tilde{h}:=(a^2-(a^\dagg)^2)/2=g_-^{2\iota_1}/2$, $\tilde{x}:=-i(a+a^\dagg)^2/4=i (g_+^{\iota_1})^2/4 $, $\tilde{y}:=-i(a-a^\dagg)^2/4=-i (g_-^{\iota_1})^2/4$, $\tilde{q}:=g_+^{\iota_1}$, $\tilde{p}:=g_-^{\iota_1}$, and $\tilde{z}:=-2i$. Table~\ref{tab:Full:Commutator:Algebra:(pesudo):schroedinger:algebra} allows us to write:
		\begin{align*}
			[\tilde{h},\tilde{x}]&=2\tilde{x},\;&\;
			[\tilde{h},\tilde{y}]&=-2\tilde{x},\;&\;
			[\tilde{x},\tilde{y}]&=\tilde{h},\;&\;
			[\tilde{h},\tilde{q}]&=g_+^{\iota_1}=\tilde{q},\;&\;
			[\tilde{h},\tilde{p}]&=-g_-^{\iota_1}=-\tilde{p},\\
			[\tilde{q},\tilde{p}]&=-2i=\tilde{z},\;&\;
			[\tilde{x},\tilde{q}]&=0,\;&\;
			[\tilde{x},\tilde{p}]&=\tilde{q},\;&\;
			[\tilde{y},\tilde{p}]&=0,\;&\;
			[\tilde{y},\tilde{q}]&=\tilde{p},
		\end{align*}
		while $\tilde{z}=-2i$ commutes with every element. Thus $\lie{\{\tilde{h},\tilde{x},\tilde{y},\tilde{q},\tilde{p},\tilde{z}\}}\cong\mathcal{S}$ by Definition~\ref{def:Schroedinger:algebra}. One has furthermore $\tilde{x}=-(2i(a^\dagg a{+}1/2)+g_+^{2\iota_1})/4$ and $\tilde{y}=(2i(a^\dagg a{+}1/2)-g_+^{2\iota_1})/4$ and therefore $\lie{\{\tilde{h},\tilde{x},\tilde{y},\tilde{q},\tilde{p},\tilde{z}\}}=\lie{\hat{A}_1^1\oplus\hat{A}_1^2}$.
	\end{proof}
	
	In the case presented above, it was straightforward to choose a basis that rendered establishing a Lie algebra-isomorphism between $\lie{\hat{A}_1^0\oplus\hat{A}_1^1\oplus\hat{A}_1^2}$ and the Schr\"odinger algebra $\mathcal{S}$ a trivial task. However, this task is generally quite challenging when the two algebras $\g$ and $\g'$ under consideration share the same invariants, such as the dimension, the dimension of the center, or the dimension of the first derived algebra. In such a case, one is typically required to solve a system of $m^2(m{+}1)/2 +1$ non-linear equations with $m^2+1$ variables, where $m$ is the dimension of the two algebras. For a more detailed discussion, see Appendix~\ref{App:Testing:Isomorphisms}.  
	
	\begin{proposition}\label{prop:algebra:g1:g2:g=:space}
		The Lie algebra generated by a set $\mathcal{G}$ of monomial generators with
		$\mathcal{G}=\mathcal{G}^1 \cup \mathcal{G}^2 \cup \mathcal{G}^=$ is infinite dimensional if and only if  both $\mathcal{G}^1 \cup \mathcal{G}^2 \neq \varnothing$ and $\mathcal{G}^= \neq \varnothing$ hold.
	\end{proposition}
	
	\begin{proof}
		Suppose there exists an element $g_+^{\tilde{\gamma}}=2i(a^{\dagg})^ka^k \in \mathcal{G}^=$ for some $k\geq2$ and another monomial $g_\sigma^\gamma\in\mathcal{G}^1\cup\mathcal{G}^2$ with $\gamma=(\alpha,\beta)\in\Np{2}$. These two elements can be used to construct a Commutator Chain of type I (see Def. 39 in \plainrefs{Bruschi:Xuereb:2024}). Moreover, since $\alpha\neq \beta$ by Definition~\ref{def:vector:spaces:a:1:k}, Theorem 46 from \plainrefs{Bruschi:Xuereb:2024} applies. According to this theorem the chain elements $z^{(\ell)}$ have degree $\deg(z^{(\ell)})=2\ell(k{-}1)+1$, which is a strictly increasing function of $\ell$, and are therefore all linearly independent. The Lie algebra $\langle\mathcal{G}^1 \cup \mathcal{G}^2 \cup \mathcal{G}^=\rangle_{\mathrm{Lie}}$ therefore contains all elements $z^{(\ell)}$ of the chain and must consequently be infinite dimensional. If, on the other hand, either $\mathcal{G}^1\cup\mathcal{G}^2=\emptyset$, or $\mathcal{G}^==\emptyset$, then $\lie{\mathcal{G}}$ is finite dimensional by Proposition~\ref{prop:commutator:g0:g=:g0:g=} or Proposition~\ref{prop:algebra:g1:g2:space} respectively.
	\end{proof}

	\begin{proposition}\label{prop:algebra:g=:gperp:space}
		The Lie algebra 
		generated by a set $\mathcal{G}$ of monomial generators with
		$\mathcal{G}= \mathcal{G}^=\cup\mathcal{G}^{\bot}$ is infinite dimensional if both sets $\mathcal{G}^=$ and $\mathcal{G}^\perp$ are non-empty.
	\end{proposition}
	
	\begin{proof}
		Suppose there exists two elements $g_{\sigma}^{\gamma} \in \mathcal{G}^{\bot}$ and $g_{+}^{\tilde{\gamma}} \in \mathcal{G}^{=}$ with $\gamma=(\alpha,\beta)\in\Np{2}$ and $\tilde{\gamma}=(\tilde{\alpha},\tilde{\alpha})\in\N_=^2$. Then, since $\alpha_1 \neq \beta_1$ and $\tilde{\alpha}_2 \geq2$ by Definition~\ref{def:vector:spaces:a:1:k}, the monomials $g_{\sigma}^\gamma$ and $g_+^{\tilde{\gamma}}$ satisfy, according to Theorem 46 from \plainrefs{Bruschi:Xuereb:2024}, the necessary conditions to generate a commutator chain of type I that lies completely in $\lie{\mathcal{G}^=\cup\mathcal{G}^\perp}$ and contains (an infinite amount of) elements of strictly increasing degree. Hence, the Lie algebra $\lie{\mathcal{G}^=\cup\mathcal{G}^\perp}$ is infinite dimensional.
	\end{proof}

	\subsection{Generic Commutator Chains in $\hat{A}_1$}
	We now focus on the objects called Commutator Chains (according to the terminology introduced in \plainrefs{Bruschi:Xuereb:2024}), which are particular sequences of nested commutators. In particular, we seek to generalize the original concept by introducing a broader framework to characterize such nested commutators. As we will see later, this generalization allows for greater flexibility in the study of the finiteness of subalgebras of $\hat{A}_1$, in particular when the generating set $\mathcal{G}$ contains elements that are linear combinations of monomials $g_\sigma^\gamma$. Note that this can be directly generalized to $\hat{A}_n$ by simply replacing $\hat{A}_1$ with $\hat{A}_n$. Nevertheless, we choose to restrict the definitions below to the case of interest here to avoid introducing unnecessary points of confusion.

	\begin{definition}\label{Commutator:Chain:General}
		Let $u^{(0)}\in\hat{A}_1$ be a fixed element, referred to as \emph{seed element}, and $\{s^{(\ell)}\}_{\ell\in\N_{\geq0}}\subseteq\hat{A}_1$ a sequence of elements, referred to as \emph{auxiliary sequence}. We define the \emph{generic Commutator Chain} $C^{\mathrm{gen}}\subseteq\hat{A}_1$ generated by $u^{(0)}$ and $\{s^{(\ell)}\}_{\ell\in\N_{\geq0}}$ as the sequence
		\begin{align}
			C^{\mathrm{gen}}(u^{(0)},\{s^{(\ell)}\}_{\ell\in\N_{\geq0}})\equiv C^{\mathrm{gen}}:=\{u^{(\ell)}\in\hat{A}_1\mid\,\,u^{(\ell+1)}:=[u^{(\ell)},s^{(\ell)}]\,\,\text{for all integers}\,\,\ell\in\N_{\geq0},\,\,\text{with}\,\,s^{(\ell)}\in\hat{A}_1\}.
		\end{align}
	\end{definition}
	
	\begin{figure}[t]
		\centering
		\includegraphics[width=0.65\linewidth]{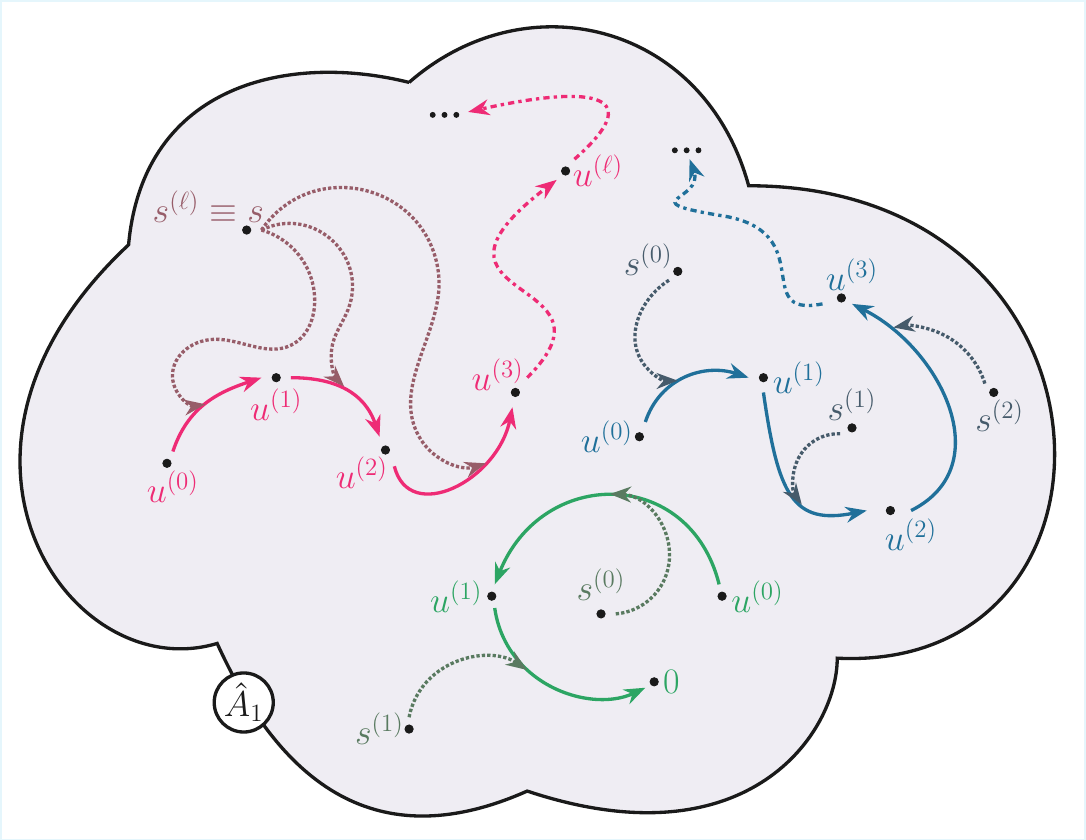}
		\caption{Depiction of three different generic Commutator Chains $C^{\mathrm{gen}}$. The red chain utilizes the same auxiliary element at each step, while the blue chain employs a different auxiliary element at every step. The green chain illustrates a chain that eventually terminates.}
		\label{fig:depict:generic:commutator:chains}
	\end{figure}

	\begin{definition}\label{Termination:Chain:General}
		Let $C^{\mathrm{gen}}(u^{(0)},\{s^{(\ell)}\}_{\ell\in\N_{\geq0}})$ be a generic Commutator Chain. We say that the chain \emph{terminates} if there exists an $\ell_*\geq1$ such that $u^{(\ell_*)}=0$.
	\end{definition}

	A depiction of several generic Commutator Chains is provided in Figure~\ref{fig:depict:generic:commutator:chains}. Notably, the Commutator Chains of Type I and II, as introduced in the literature \plainrefs{Bruschi:Xuereb:2024}, are present as special cases within this more general framework.

	\begin{lemma}\label{lem:generic:chain:behavior}
		Let $g_\sigma^\gamma,g_{\hat{\sigma}}^{\hat{\gamma}}\in\hat{A}_1$ be two monomials with $\hat{\gamma}>\hat{\gamma}^\dagg$ and $C^\mathrm{gen}$ a generic Commutator Chain with seed element $g_{\hat{\sigma}}^{\hat{\gamma}}=u^{(0)}$ and auxiliary sequence $g_\sigma^\gamma=s^{(\ell)}$ for all $\ell\in\N_{\geq0}$. One has then:
		\begin{align}\label{eqn:generic:chain:help}
			u^{(\ell)}\upto{d^{(\ell)}}\sum_{p\in\mathcal{P}^{(\ell)}}C_p^{(\ell)}g_{(-\sigma)^\ell\hat{\sigma}}^{\gamma_p^{(\ell)}}+\sum_{q\in\mathcal{Q}^{(\ell)}}\hat{C}_q^{(\ell)}g_{(-\sigma)^\ell\hat{\sigma}}^{\hat{\gamma}_q^{(\ell)}},
		\end{align}
		where $d^{(\ell)}=|\hat{\gamma}|+\ell(|\gamma|-2)$, $\gamma_p^{(\ell)}=\hat{\gamma}+p\gamma+(\ell-p)\gamma^\dagg-\ell\tau$, $\hat{\gamma}_q^{(\ell)}=\hat{\gamma}^\dagg+q\gamma+(\ell-q)\gamma^\dagg-\ell\tau$, and $\mathcal{P}^{(\ell)}$ as well as $\mathcal{Q}^{(\ell)}$ are suitable index sets defined such that $\gamma_p^{(\ell)}\geq(\gamma_p^{(\ell)})^\dagg$ and $\hat{\gamma}_q^{(\ell)}\geq(\hat{\gamma}_q^{(\ell)})^\dagg$.
	\end{lemma}
	
	We want to emphasize that Lemma~\ref{lem:generic:chain:behavior} has two important implications: (a) that the result of the commutator above has exactly the expression \eqref{eqn:generic:chain:help}; and (b) that $\deg(u^{(\ell)})\leq d^{(\ell)}$, while equality is not in general guaranteed.
	
	\begin{proof}
		We prove the claim by induction. The base case $\ell=0$ is clearly given by $\mathcal{P}^{(0)}=\{0\}$, $\mathcal{Q}^{(0)}=\emptyset$, $C_0^{(0)}=1$, and $\gamma_0^{(0)}=\hat{\gamma}$. Note that $\mathcal{Q}^{(0)}=\emptyset$ is given by the conditions $\hat{\gamma}>\hat{\gamma}^\dagg$ and $\hat{\gamma}_q^{(\ell)}\geq(\hat{\gamma}_q^{(\ell)})^\dagg$.
		Thus, we can now assume that $u^{(\ell)}$ is of the form \eqref{eqn:generic:chain:help} and we proceed to compute
		\begin{align*}
			u^{(\ell+1)}=[u^{(\ell)},g_\sigma^\gamma]\upto{d^{(\ell+1)}}\sum_{p\in\mathcal{P}^{(\ell)}}C_p^{(\ell)}[g_{(-\sigma)^\ell\hat{\sigma}}^{\gamma_p^{(\ell)}},g_\sigma^\gamma]+\sum_{q\in\mathcal{Q}^{(\ell)}}\hat{C}_q^{(\ell)}[g_{(-\sigma)^\ell\hat{\sigma}}^{\hat{\gamma}_q^{(\ell)}},g_\sigma^\gamma].
		\end{align*}
		Here, we have applied Proposition~\ref{prop:equivalnce:relation:commutator} and we now apply Lemma~\ref{lem:commutator:form} to compute
		\begin{align*}
			[g_{(-\sigma)^\ell\hat{\sigma}}^{\gamma_p^{(\ell)}},g_\sigma^\gamma]&\upto{d^{(\ell+1)}}\hat{\epsilon}_1\left(\alpha{\beta}_p^{(\ell)}{-}\beta{\alpha}_p^{(\ell)}\right)g_{(-\sigma)^{\ell+1}\hat{\sigma}}^{\gamma_p^{(\ell)}+\gamma-\tau}+\hat{\epsilon}_2\left(\alpha{\alpha}_p^{(\ell)}{-}\beta{\beta}_p^{(\ell)}\right)g_{(-\sigma)^{\ell+1}\hat{\sigma}}^{\Theta(\gamma+(\gamma_p^{(\ell)})^\dagg)-\tau},\\
			[g_{(-\sigma)^\ell\hat{\sigma}}^{\hat{\gamma}_q^{(\ell)}},g_\sigma^\gamma]&\upto{d^{(\ell+1)}}\hat{\epsilon}_3\left(\alpha\hat{\beta}_q^{(\ell)}{-}\beta\hat{\alpha}_q^{(\ell)}\right)g_{(-\sigma)^{\ell+1}\hat{\sigma}}^{\hat{\gamma}_q^{(\ell)}+\gamma-\tau}+\hat{\epsilon}_4\left(\alpha\hat{\alpha}_p^{(\ell)}{-}\beta\hat{\beta}_q^{(\ell)}\right)g_{(-\sigma)^{\ell+1}\hat{\sigma}}^{\Theta(\gamma+(\hat{\gamma}_q^{(\ell)})^\dagg)-\tau}.
		\end{align*}
		It is now immediate to verify that $u^{(\ell+1)}$ is of the desired form, with the correct multi-indices $\gamma_p^{(\ell+1)}$, $\hat{\gamma}_q^{(\ell+1)}$ and appropriate index sets $\mathcal{P}^{(\ell+1)}$ and $\mathcal{Q}^{(\ell+1)}$, which concludes this proof.
	\end{proof}	
	
	Before proceeding, we want to extend the equivalence relation \enquote{$\upto{d}$} from Definition~\ref{def:equivalence:relation:polynomials} to generic Commutator Chains.
	
	\begin{definition}[Equivalence of generic Commutator Chains]\label{def:equivalnce:generic:comm:chains}
		Let $C^{\mathrm{gen}}_1$ and $C^{\mathrm{gen}}_2$ be two generic Commutator Chains with seed elements $u_1^{(0)}$, $u^{(0)}_2$, and auxiliary sequences $s^{(\ell)}_1$, $s^{(\ell)}_2$ for all $\ell\in\N_0$, respectively. We say that the two chains are equivalent up to order $d^{(\ell)}$, written $C^{\mathrm{gen}}_1\upto{d^{(\ell)}}C^{\mathrm{gen}}_2$ for $\{d^{(\ell)}\}_{\ell \in\N_{\geq0}}$ with $d^{(\ell)}\in\N_{\geq0}$ for all $\ell\in\N_{\geq0}$, if and only if $u_1^{(\ell)}\upto{d^{(\ell)}}u^{(\ell)}_2$ for all $\ell\in\N_{\geq0}$. 
	\end{definition}
	
	It is immediate to verify that $C^{\mathrm{gen}}_1\upto{d^{(\ell)}}C^{\mathrm{gen}}_2$ is an equivalence relation. Note that any pair of generic Commutator Chains $C^{\mathrm{gen}}_1$ and $C^{\mathrm{gen}}_2$ is trivially equivalent in the sense $C^{\mathrm{gen}}_1\upto{d^{(\ell)}}C^{\mathrm{gen}}_2$ if $d^{(\ell)}=\infty$ for all $\ell$, since any pair of elements $a_1,a_2$ in $\hat{A}_1$ are equivalent in the sense $a_1\upto{\infty}a_2$.
	
	\begin{proposition}\label{prop:equivalnce:generic:com:chains}
		Consider two generic Commutator Chains $C^{\mathrm{gen}}_1$ and $C^{\mathrm{gen}}_2$  with seed elements $u_1^{(0)}$, $u^{(0)}_2$, and auxiliary sequences $s^{(\ell)}_1\equiv s_1$, $s^{(\ell)}_2\equiv s_2$ for all $\ell\in\N_{\geq0}$, respectively. Suppose $u_1^{(0)}\upto{d_1}u^{(0)}_2$ and $s_1\upto{d_2}s_2$ with $d_1=\deg(u_1^{(0)})\in\N_{\geq0}$ and $d_2=\deg(s_1)\in\N_{\geq0}$. Then the two chains have the property that $C^{\mathrm{gen}}_1\upto{d^{(\ell)}}C^{\mathrm{gen}}_2$  with $d^{(\ell)}=d_1+\ell(d_2-2)$. 
	\end{proposition}
	
	\begin{proof}
		This can be shown using induction. In the base case $\ell=0$, we have $u_1^{(0)}\upto{d^{(0)}}u_2^{(0)}$ by the initial assumptions, with $d^{(0)}=d_1$. Thus, the base case holds, and we can assume that $u_1^{(\ell)}\upto{d^{(\ell)}}u_2^{(\ell)}$ with $d^{(\ell)}=d_1+\ell(d_2-2)$ for some $\ell\in\N_{\geq0}$. Proposition~\ref{prop:equivalnce:relation:commutator} and Lemma~\ref{lem:generic:chain:behavior} yield:
		\begin{align*}
			u_1^{(\ell+1)}=[u_1^{(\ell)},s_1]\upto{d^{(\ell+1)}}[u_2^{(\ell)},s_2]=u_2^{(\ell+1)},
		\end{align*}
		since $\deg(u_1^{(\ell)})\leq d^{(\ell)}$ and $\deg(u_2^{(\ell)})\leq d^{(\ell)}$, and $d^{(\ell+1)}=d^{(\ell)}+d_2-2$. Together with the claim of induction this completes the proof. 
	\end{proof}
	
	\subsection{Classification of finite-dimensional Lie algebras}
	We are now in the position to prove the main results of this section.
	
	\begin{theorem}\label{thm:andreea:sona}
		Let $g_{\hat{\sigma}}^{\hat{\gamma}}\in \hat{A}_1^1\oplus \hat{A}_1^2\oplus\hat{A}_1^\perp$ and $g_\sigma^\gamma\in \hat{A}_1^\perp$ be two different monomials. Then, the Lie algebra $\g:=\langle\{g_{\hat{\sigma}}^{\hat{\gamma}},g_\sigma^\gamma\}\rangle_\mathrm{Lie}$ is infinite dimensional. 
	\end{theorem}
	
	The proof of this theorem relies on constructing generic Commutator Chains that contain elements of strictly increasing degree. Since monomials of different degrees are linearly independent, the presence of such a Commutator Chain implies that the Lie algebra generated by the two initial monomials contains an infinite number of linearly-independent terms, which in turn implies that the whole generated algebra is infinite dimensional since the Commutator Chain is a subset of the algebra itself. The relevant generic Commutator Chains used in the proof are illustrated schematically in Figure~\ref{fig:sketch:chains:thm:sona:andreea}.
	
	\begin{figure}[t]
		\centering
		\includegraphics[width=0.8\linewidth]{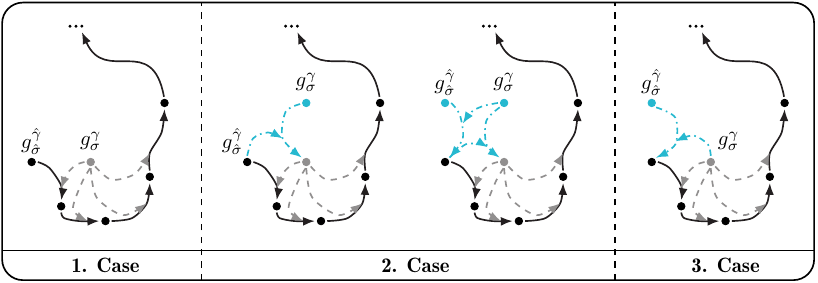}
		\caption{Schematic overview of the generic Commutator Chains used in the proof of Theorem~\ref{thm:andreea:sona}. The central chain (Case 1) corresponds to the situation where the two initial monomials $g_{\hat{\sigma}}^{\hat{\gamma}}\in\hat{A}_1^1\oplus\hat{A}_1^2\oplus\hat{A}_1^\perp$ and $g_\sigma^\gamma\in\hat{A}_1^\perp$ can be directly used to construct a chain whose elements have a simple and easily verifiable property: in the expansion in terms of monomials of every chain element $u^{(\ell)}$ there exists a unique non-vanishing monomial with highest degree. This makes it straightforward to prove that the degrees of the elements in the chain increase strictly with each step. In contrast, the remaining Cases 2 and 3 do not have this property. However, each of them can be reduced to Case 1 by constructing at most two auxiliary elements (indicated by the end of the light-blue lines) that can in turn be used so that the resulting chain generated by such two elements can be treated analogously to Case 1.}
		\label{fig:sketch:chains:thm:sona:andreea}
	\end{figure}
	
	\begin{proof}
		We prove the claim by considering the generic Commutator Chain $C^\mathrm{gen}$ with seed element $g_{\hat{\sigma}}^{\hat{\gamma}}=u^{(0)}$ and auxiliary sequence $s^{(\ell)}\equiv g_\sigma^\gamma$ for all $\ell\in\N_{\geq0}$ and thereby proving that $\deg(u^{(\ell)})=d^{(\ell)}=|\hat{\gamma}|+\ell(|\gamma|-2)$ for all $\ell\in\N_{\geq0}$. Therefore, the degree of the chain elements grows strictly as a function of $\ell$, since $|\gamma|>3$ by Lemma 7 from \plainrefs{Bruschi:Xuereb:2024}, ultimately implying that the Lie algebra $\g=\lie{\{g_\sigma^\gamma,g_{\hat{\sigma}}^{\hat{\gamma}}\}}$ is infinite dimensional. We obtain our goal by extending the statement of Lemma~\ref{lem:generic:chain:behavior} and dividing the proof into three distinct cases:
		
		\par\noindent
		\textbf{Case 1}: $\boldsymbol{\alpha\hat{\beta}\neq\beta\hat{\alpha}}$. Here, we extend the statement of Lemma~\ref{lem:generic:chain:behavior} to state not only that $u^{(\ell)}$ is given by the expression \eqref{eqn:generic:chain:help} but also that, in the expansion of $u^{(\ell)}$ to highest degree $d^{(\ell)}$, there exists always exactly one unique monomial with multi-index of the specific form $\gamma_\ell^{(\ell)}=\hat{\gamma}+\ell(\gamma{-}\tau)$ and non-vanishing prefactor.
		
		Let us see why this occurs. Lemma~\ref{lem:commutator:form} implies that
		\begin{align}\label{eqn:expansion:u_0:thm:sona:andreea}
			u^{(1)}=[u^{(0)},s^{(0)}]=[g_{\hat{\sigma}}^{\hat{\gamma}},g_\sigma^\gamma]\upto{d^{(1)}}\hat{\epsilon}\left(\alpha\hat{\beta}{-}\beta\hat{\alpha}\right)g_{-\sigma\hat{\sigma}}^{\gamma+\hat{\gamma}-\tau}+\hat{\epsilon}'\left(\alpha\hat{\alpha}{-}\beta\hat{\beta}\right)g_{-\sigma\hat{\sigma}}^{\Theta(\gamma+\hat{\gamma}^\dagg)-\tau}.
		\end{align}
		Since $\alpha\hat{\beta}-\beta\hat{\alpha}\neq0$, we can conclude that in the expansion of $u^{(1)}$ there exists at least one monomial with non-vanishing prefactor, and it has multi-index $\gamma_1^{(1)}=\gamma+\hat{\gamma}-\tau$. We now need to show that there exists no other such monomial in the expansion of $u^{(1)}$. To obtain this we simply need to consider the equation $\gamma+\hat{\gamma}-\tau=\Theta(\gamma+\hat{\gamma}^\dagg)-\tau$, since these are the multi-indices of the two monomials appearing in expansion \eqref{eqn:expansion:u_0:thm:sona:andreea}. Now recall that $\Theta(\gamma)=\gamma$ if $\gamma\geq \gamma^\dagg$ and $\Theta(\gamma)=\gamma^\dagg$ otherwise. Thus, one only needs to consider the following two possibilities: Either $\Theta(\gamma+\hat{\gamma}^\dagg)=\gamma+\hat{\gamma}^\dagg$ or $\Theta(\gamma+\hat{\gamma}^\dagg)=\gamma^\dagg+\hat{\gamma}$. In the former case, one has $\gamma+\hat{\gamma}-\tau=\Theta(\gamma+\hat{\gamma}^\dagg)-\tau=\gamma+\hat{\gamma}^\dagg-\tau$ and consequently $\hat{\gamma}=\hat{\gamma}^\dagg$. In the latter case, one has $\gamma+\hat{\gamma}-\tau=\Theta(\gamma+\hat{\gamma}^\dagg)-\tau=\gamma^\dagg+\hat{\gamma}-\tau$ and consequently $\gamma=\gamma^\dagg$. To summarize $\gamma+\hat{\gamma}-\tau=\Theta(\gamma+\hat{\gamma}^\dagg)-\tau$, implies either $\gamma=\gamma^\dagg$ or $\hat{\gamma}=\hat{\gamma}^\dagg$, which is excluded by the initial assumption $g_{\hat{\sigma}}^{\hat{\gamma}}\in \hat{A}_1^1\oplus \hat{A}_1^2\oplus\hat{A}_1^\perp$ and $g_\sigma^\gamma\in \hat{A}_1^\perp$. At last, we need to check that the monomial with multi-index $\gamma_1^{(1)}$ appearing in the expansion of $u^{(1)}$ does itself not vanish, as it would be the case if it were of the form $g_-^{\tilde{\gamma}}$ with $\tilde{\gamma}\in\N_=^2$. Here, it is enough to observe that $\gamma+\hat{\gamma}-\tau>(\gamma+\hat{\gamma}-\tau)^\dagg$ showing that $\gamma_1^{(1)}\notin\N_=^2$. This implies consequently $\deg(u^{(1)})=d^{(1)}=|\hat{\gamma}|+|\gamma|-2$ and concludes the base case for the proof by induction. 
		
		Now we assume that the induction hypothesis holds for $\ell\in\N_{\geq0}$ Concretely, $u^{(\ell)}$ can be written in the form given in \eqref{eqn:generic:chain:help} and there exists exactly one monomial with multi-index $\gamma_\ell^{(\ell)}=\hat{\gamma}+\ell(\gamma-\tau)$ and non-vanishing prefactor $C_\ell^{(\ell)}$.
		
		Thus, we know that in the expansion of $u^{(\ell+1)}$, there exists a term proportional to the commutator
		\begin{align*}
			[g_{(-\sigma)^\ell\hat{\sigma}}^{\gamma_\ell^{(\ell)}},g_\sigma^\gamma]&\upto{d^{(\ell+1)}}\hat{\epsilon}_1\left(\alpha{\beta}_\ell^{(\ell)}-\beta{\alpha}_\ell^{(\ell)}\right)g_{(-\sigma)^{\ell+1}\hat{\sigma}}^{\gamma_\ell^{(\ell)}+\gamma-\tau}+\hat{\epsilon}_2\left(\alpha{\alpha}_\ell^{(\ell)}-\beta{\beta}_\ell^{(\ell)}\right)g_{(-\sigma)^{\ell+1}\hat{\sigma}}^{\Theta(\gamma+(\gamma_\ell^{(\ell)})^\dagg)-\tau},
		\end{align*}
		since $u^{(\ell+1)}=[u^{(\ell)},g_{\sigma}^\gamma]$ and since in the expansion of $u^{(\ell)}$, given by equation \eqref{eqn:generic:chain:help}, there exists the monomial $g_{(-\sigma)^\ell\hat{\sigma}}^{\gamma_\ell^{(\ell)}}$ with a non-vanishing prefactor.
		We see that the first term in the expression above is a multiple of the monomial with multi-index $\gamma_{\ell+1}^{(\ell+1)}=\gamma_\ell^{(\ell)}+\gamma-\tau$. The prefactor of this monomial is $\hat{\epsilon}_1C_\ell^{(\ell)}(\alpha\beta_\ell^{(\ell)}-\beta{\alpha}_\ell^{(\ell)})$, where $\hat{\epsilon}_1\in\{\pm1\}$ and $C_\ell^{(\ell)}\neq0$ by the induction hypothesis. Hence, to show that this prefactor does not vanish, one must show that $\alpha{\beta}_\ell^{(\ell)}-\beta{\alpha}_\ell^{(\ell)}\neq0$. We have
		\begin{align*}
			\alpha{\beta}_\ell^{(\ell)}-\beta{\alpha}_\ell^{(\ell)}&=\alpha(\hat{\beta}+\ell(\beta-1))-\beta(\hat{\alpha}+\ell(\alpha-1))=-\ell(\alpha-\beta)+\alpha\hat{\beta}-\hat{\alpha}\beta,
		\end{align*}
		We will now show that the expression above is, without loss of generality, always smaller than zero and therefore non-zero. The first essential observation is that $\alpha>\beta$ holds, by Definition~\ref{def:vector:spaces:a:1:k}, for every monomial $g_\sigma^\gamma\in\hat{A}_n^\perp$. This implies $-\ell(\alpha-\beta)<0$ for all $\ell\in\N_{\geq1}$. Next, we have to argue why, without loss of generality, $\alpha\hat{\beta}-\hat{\alpha}\beta\leq 0$ holds in all cases. Now recall that the monomial $g_{\hat{\sigma}}^{\hat{\gamma}}$ is assumed to belong to the space $\hat{A}_n^1\oplus\hat{A}_n^2\oplus\hat{A}_n^\perp$. We must consequently consider the three possibilities (a) $g_{\hat{\sigma}}^{\hat{\gamma}}\in\hat{A}_n^1$, (b) $g_{\hat{\sigma}}^{\hat{\gamma}}\in\hat{A}_n^2$, and (c) $g_{\hat{\sigma}}^{\hat{\gamma}}\in\hat{A}_n^\perp$, which we tackle individually below:

		\begin{enumerate}[label = (\alph*)]
			\item Here one has $\boldsymbol{g_{\hat{\sigma}}^{\hat{\gamma}}\in\hat{A}_1^1}$. By Definition~\ref{def:vector:spaces:a:1:k}, this implies $\hat{\gamma}=(\hat{\alpha},\hat{\beta})=\iota_1=(1,0)$, and therefore $\alpha\hat{\beta}-\hat{\alpha}\beta=-\beta\leq 0$. 
			\item Here one has $\boldsymbol{g_{\hat{\sigma}}^{\hat{\gamma}}\in\hat{A}_1^2}$. By Definition~\ref{def:vector:spaces:a:1:k}, this implies $\hat{\gamma}=(\hat{\alpha},\hat{\beta})=2\iota_1=(2,0)$, and therefore $\alpha\hat{\beta}-\hat{\alpha}\beta=-2\beta\leq 0$. 
			\item Here one has $\boldsymbol{g_{\hat{\sigma}}^{\hat{\gamma}}\in\hat{A}_1^\perp}$. Recall now $g_\sigma^\gamma$ is also a monomial that lies in the space $\hat{A}_1^\perp$. Hence, one can simply exchange the roles of $g_{\hat{\sigma}}^{\hat{\gamma}}$ and $g_\sigma^\gamma$ if $\alpha\hat{\beta}-\hat{\alpha}\beta>0$, allowing to assume without loss of generality that $\alpha\hat{\beta}-\hat{\alpha}\beta\leq 0$. 
		\end{enumerate}
		Thus, for every $\ell\in\N_{\geq1}$, the term $\alpha{\beta}_\ell^{(\ell)}-\beta{\alpha}_\ell^{(\ell)}$ is the sum of a negative and a non-positive term, thus implying that it is non-zero. It follows therefore that there exists at least one monomial in the expansion of $u^{(\ell+1)}$ with non-vanishing prefactor and multi-index $\gamma_{\ell+1}^{(\ell+1)}$. Next, we need to consider the case that there exist other monomials in the expansion of $u^{(\ell+1)}$ with the same multi-index, and prove that this cannot occur if we wish to ensure its claimed uniqueness. To obtain our goal we need to explicitly compute the expansion of $u^{(\ell+1)}$ to the highest degree. This can be done using Propositions~\ref{prop:equivalnce:relation:commutator} and \ref{lem:commutator:form}, as well as the induction hypothesis that $u^{(\ell)}$ has the form by the expression \eqref{eqn:generic:chain:help}. We have
		\begin{align}
			u^{(\ell+1)}=[u^{(\ell)},g_\sigma^\gamma]\upto{d^{(\ell+1)}}\sum_{p\in\mathcal{P}^{(\ell)}}C_p^{(\ell)}[g_{(-\sigma)^\ell\hat{\sigma}}^{\gamma_p^{(\ell)}},g_\sigma^\gamma]+\sum_{q\in\mathcal{Q}^{(\ell)}}\hat{C}_q^{(\ell)}[g_{(-\sigma)^\ell\hat{\sigma}}^{\hat{\gamma}_q^{(\ell)}},g_\sigma^\gamma],\label{eqn:help:thm:40:u:ell:plus:one}
		\end{align}
		where the individual commutators are given by
		\begin{subequations}
			\begin{align}
				[g_{(-\sigma)^\ell\hat{\sigma}}^{\gamma_p^{(\ell)}},g_\sigma^\gamma]&\upto{d^{(\ell+1)}}\hat{\epsilon}_1\left(\alpha{\beta}_p^{(\ell)}-\beta{\alpha}_p^{(\ell)}\right)g_{(-\sigma)^{\ell+1}\hat{\sigma}}^{\gamma_p^{(\ell)}+\gamma-\tau}+\hat{\epsilon}_2\left(\alpha{\alpha}_p^{(\ell)}-\beta{\beta}_p^{(\ell)}\right)g_{(-\sigma)^{\ell+1}\hat{\sigma}}^{\Theta(\gamma+(\gamma_p^{(\ell)})^\dagg)-\tau},\label{eqn:help:thm:40:u:ell:plus:one:first:commutator}\\
				[g_{(-\sigma)^\ell\hat{\sigma}}^{\hat{\gamma}_q^{(\ell)}},g_\sigma^\gamma]&\upto{d^{(\ell+1)}}\hat{\epsilon}_3\left(\alpha\hat{\beta}_q^{(\ell)}-\beta\hat{\alpha}_q^{(\ell)}\right)g_{(-\sigma)^{\ell+1}\hat{\sigma}}^{\hat{\gamma}_q^{(\ell)}+\gamma-\tau}+\hat{\epsilon}_4\left(\alpha\hat{\alpha}_p^{(\ell)}-\beta\hat{\beta}_q^{(\ell)}\right)g_{(-\sigma)^{\ell+1}\hat{\sigma}}^{\Theta(\gamma+(\hat{\gamma}_q^{(\ell)})^\dagg)-\tau}.\label{eqn:help:thm:40:u:ell:plus:one:second:commutator}
			\end{align}
		\end{subequations}
		Thus, to verify that there exists only one monomial with the multi-index $\gamma_{\ell+1}^{(\ell+1)}$ in the expansion of $u^{(\ell+1)}$ found in equation \eqref{eqn:help:thm:40:u:ell:plus:one}, we need to examine the possible multi-indices of the monomials obtained in equations \eqref{eqn:help:thm:40:u:ell:plus:one:first:commutator} and \eqref{eqn:help:thm:40:u:ell:plus:one:second:commutator}. The first monomial in the right-hand side of equation \eqref{eqn:help:thm:40:u:ell:plus:one:first:commutator} is $\gamma_p^{(\ell)}+\gamma-\tau=\hat{\gamma}+(p+1)\gamma+(\ell+1-(p+1))\gamma^\dagg-(\ell+1)\tau$, while the second one is $\Theta(\gamma+(\gamma_p^{(\ell)})^\dagg)-\tau$ which takes either the form $\hat{\gamma}+p\gamma+(\ell+1-p)\gamma^\dagg-(\ell+1)\tau$ or $\hat{\gamma}^\dagg+(\ell+1-p)\gamma+p\gamma^\dagg-(\ell+1)\tau$ depending on the result of the $\Theta(\cdot)$ function. The first monomial in the right-hand side of equation \eqref{eqn:help:thm:40:u:ell:plus:one:second:commutator} is similarly $\hat{\gamma}_q^{(\ell)}+\gamma-\tau=\hat{\gamma}^\dagg+(q+1)\gamma+(\ell+1-(q+1))\gamma^\dagg-(\ell+1)\tau$ and the second one is $\Theta(\gamma+(\hat{\gamma}_q^{(\ell)})^\dagg)-\tau$ which either takes the form $\hat{\gamma}^\dagg+q\gamma+(\ell+1-q)\gamma^\dagg-(\ell+1)\tau$ or $\hat{\gamma}+(\ell+1-q)\gamma+q\gamma^\dagg-(\ell+1)\tau$ depending the result of the $\Theta(\cdot)$ function. To prove that no other monomial in the expansion of $u^{(\ell+1)}$ possesses the multi-index $\gamma_{\ell+1}^{(\ell+1)}$, we need to show the multi-indices listed before coincide with $\gamma_{\ell+1}^{(\ell+1)}$ only if they parametrize the same monomial. We can spell the conditions out as follows:
		\begin{subequations}
			\begin{align}
				\hat{\gamma}+(\ell+1)\gamma-(\ell+1)\tau&=\gamma_p^{(\ell)}+\gamma-\tau=\hat{\gamma}+(p+1)\gamma+(\ell+1-(p+1))\gamma^\dagg-(\ell+1)\tau,\label{eqn:thm:sona:andreea:case:1:1}\\
				\hat{\gamma}+(\ell+1)\gamma-(\ell+1)\tau&=\gamma_p^{(\ell)}+\gamma^\dagg-\tau=\hat{\gamma}+p\gamma+(\ell+1-p)\gamma^\dagg-(\ell+1)\tau,\label{eqn:thm:sona:andreea:case:1:2}\\
				\hat{\gamma}+(\ell+1)\gamma-(\ell+1)\tau&=(\gamma_p^{(\ell)})^\dagg+\gamma-\tau=\hat{\gamma}^\dagg+(\ell+1-p)\gamma+p\gamma^\dagg-(\ell+1)\tau,\label{eqn:thm:sona:andreea:case:1:3}\\
				\hat{\gamma}+(\ell+1)\gamma-(\ell+1)\tau&=\hat{\gamma}_q^{(\ell)}+\gamma-\tau=\hat{\gamma}^\dagg+(q+1)\gamma+(\ell+1-(q+1))\gamma^\dagg-(\ell+1)\tau,\label{eqn:thm:sona:andreea:case:1:4}\\
				\hat{\gamma}+(\ell+1)\gamma-(\ell+1)\tau&=\hat{\gamma}_q^{(\ell)}+\gamma^\dagg-\tau=\hat{\gamma}^\dagg+q\gamma+(\ell+1-q)\gamma^\dagg-(\ell+1)\tau,\label{eqn:thm:sona:andreea:case:1:5}\\
				\hat{\gamma}+(\ell+1)\gamma-(\ell+1)\tau&=(\hat{\gamma}_q^{(\ell)})^\dagg+\gamma-\tau=\hat{\gamma}+(\ell+1-q)\gamma+q\gamma^\dagg-(\ell+1)\tau\label{eqn:thm:sona:andreea:case:1:6}.
			\end{align}
		\end{subequations}
		These equations have a similar structure and we can consequently combine them into the single equation $z(\hat{\gamma}-\hat{\gamma}^\dagg)+(\ell+1-y)(\gamma-\gamma^\dagg)=0$ with $(z,y)\in\{(0,p+1),(0,p), (1,\ell+1-p), (1,q+1), (1,q), (0,\ell+1-q)\}$, where the pair $(z,y)=(0,p+1)$ corresponds for example to equation \eqref{eqn:thm:sona:andreea:case:1:1}, while the pair $(z,y)=(0,\ell+1-q))$ corresponds to equation \eqref{eqn:thm:sona:andreea:case:1:6}. 
		
		If $z=0$, one is concerned with the equation $(\ell+1-y)(\alpha-\beta)=0$ which implies $y=\ell+1$, since $\alpha>\beta$. Thus, one has either $p=\ell$ in Equation~\eqref{eqn:thm:sona:andreea:case:1:1} and is therefore left with the same monomial in the expansion of $u^{(\ell+1)}$, or one has $p=\ell+1$ in Equation~\eqref{eqn:thm:sona:andreea:case:1:2}. This is not possible, since one always has $p\leq\ell$. The only other equation with $z=0$ is Equation~\eqref{eqn:thm:sona:andreea:case:1:6}, where one would need $q=0$. This, however, is not possible as this monomial in the expansion of $u^{(\ell+1)}$ would originate from a commutator of a monomial in the expansion of $u^{(\ell)}$ with the multi-index $\hat{\gamma}_0^{(\ell)}$ and the monomial $g_\sigma^\gamma$. The multi-index $\hat{\gamma}_0^{(\ell)}$ violates clearly the condition $\hat{\gamma}_q^{(\ell)}\geq (\hat{\gamma}_q^{(\ell)})^\dagg$ and can therefore not appear in the expansion of $u^{(\ell)}$. If now $z=1$, one is concerned with the equation $\hat{\alpha}-\hat{\beta}+(\ell+1-y)(\alpha-\beta)=0$. To satisfy this equation one needs $y>\ell+1$, since $\alpha>\beta\geq0$ and $\hat{\alpha}>\hat{\beta}\geq0$. However, $y$ can clearly never be larger than $\ell+1$ making this case impossible. All together, this multi-index analysis shows that there exists exactly one monomial in the expansion of $u^{(\ell+1)}$ with multi-index $\gamma_{\ell+1}^{(\ell+1)}$ and non-vanishing prefactor. The observation $\gamma_{\ell+1}^{(\ell+1)}>(\gamma_{\ell+1}^{(\ell+1)})^\dagg$ implies furthermore $g_{(-\sigma)^\ell\hat{\sigma}}^{\gamma_{\ell+1}^{(\ell+1)}}\neq0$. This shows consequently $\deg(u^{(\ell+1)})=d^{(\ell+1)}$, since $\deg(u^{(\ell)})\leq d^{(\ell)}$  for all $\ell\in \N_{\geq0}$ by Lemma~\ref{lem:generic:chain:behavior}, and $|\gamma_{\ell+1}^{(\ell+1)}|=d^{(\ell+1)}$, concluding the proof by induction in this first case. 
		\par\noindent
		\textbf{Case 2}: $\boldsymbol{\beta=0=\hat{\beta}}$. By Lemma 7 from \plainrefs{Bruschi:Xuereb:2024}, this implies $\alpha\geq 3$. Thus, one can choose without loss of generality under appropriate relabeling $\alpha\geq 3>\hat{\alpha}>0$ for $g_{\hat{\sigma}}^{\hat{\gamma}}\in\hat{A}_1^1\oplus\hat{A}_1^2$ and $\alpha\geq\hat{\alpha}\geq 3$ for $g_{\hat{\sigma}}^{\hat{\gamma}}\in\hat{A}_1^\perp$. Here, we simply need to compute the commutator of these two monomials:
		\begin{align*}
			u^{(1)}=[g_{\hat{\sigma}}^{\hat{\gamma}},g_\sigma^\gamma]\upto{d^{(1)}}\hat{\epsilon}\left(\alpha\hat{\beta}-\beta\hat{\alpha}\right)g_{-\sigma\hat{\sigma}}^{\gamma+\hat{\gamma}-\tau}+\hat{\epsilon}'\left(\alpha\hat{\alpha}-\beta\hat{\beta}\right)g_{-\sigma\hat{\sigma}}^{\Theta(\gamma+\hat{\gamma}^\dagg)-\tau}=\hat{\epsilon}'\alpha\hat{\alpha}g_{-\sigma\hat{\sigma}}^{\gamma+\hat{\gamma}^\dagg-\tau}.
		\end{align*}
		The prefactor $\hat{\epsilon}'\alpha\hat{\alpha}$ of the latter monomial is non-vanishing since $\alpha,\hat{\alpha}>0$. This term can therefore only vanish if $\gamma+\hat{\gamma}^\dagg-\tau=(\alpha-1,\hat{\alpha}-1)\in \N_=^2$ and $-\sigma\hat{\sigma}=-1$. In order for this to happen we would have to have  $\gamma=\hat{\gamma}$ and $\sigma=\hat{\sigma}$, i.e., $g_\sigma^\gamma=g_{\hat{\sigma}}^{\hat{\gamma}}$. However, this case is excluded by the initial assumption that one considers two distinct monomials. Hence, $\hat{\epsilon}'\alpha\hat{\alpha}g_{-\sigma\hat{\sigma}}^{\gamma+\hat{\gamma}^\dagg-\tau}$ does not vanish.
		
		We now are led to consider the following three subcases $g_{\hat{\sigma}}^{\hat{\gamma}}\in\hat{A}_1^1$, $g_{\hat{\sigma}}^{\hat{\gamma}}\in\hat{A}_1^2$, and $g_{\hat{\sigma}}^{\hat{\gamma}}\in\hat{A}_1^\perp$ separately:
		\begin{itemize}
			\item[$\boldsymbol{\hat{A}_1^2}$:] Here one has $g_{\hat{\sigma}}^{\hat{\gamma}}\in\hat{A}_1^2$, and $\gamma+\hat{\gamma}^\dagg-\tau=(\alpha-1,1)$, since $\hat{\gamma}=2\iota_1$. One can then simply consider the generic Commutator Chain with seed element $g_{\hat{\sigma}}^{\hat{\gamma}}=u^{(0)}$ and auxiliary sequence $g_{-\sigma\hat{\sigma}}^{\gamma+\hat{\gamma}^\dagg-\tau}\equiv s^{(\ell)}\in\hat{A}_1^\perp$ for all $\ell\in\N_{\geq0}$ instead of the auxiliary sequence $g_\sigma^\gamma=s^{(\ell)}$, which can be treated by the first case, since $\chimap{\hat{\gamma}\circ(\gamma+\hat{\gamma}^\dagg-\tau)^\dagg}=2\neq0$. This generic Commutator Chain will not necessarily lie in $\g=\lie{\{g_\sigma^\gamma,g_{\hat{\sigma}}^{\hat{\gamma}}\}}$ but, according to Proposition~\ref{prop:equivalnce:generic:com:chains}, it will generate to highest degree $d^{(\ell)}$ the same elements as those found in the generic Commutator Chain with seed element $g_{\hat{\sigma}}^{\hat{\gamma}}$ and auxiliary sequence $s^{(\ell)}\equiv u^{(1)}/(\hat{\epsilon}'\alpha\hat{\alpha})\upto{d^{(1)}}g_{-\sigma\hat{\sigma}}^{\gamma+\hat{\gamma}^\dagg-\tau}$ for all $\ell\in\N_{\geq0}$, thereby implying that $\g$ is infinite dimensional.
			\item[$\boldsymbol{\hat{A}_1^\perp}$:] Here one has $g_{\hat{\sigma}}^{\hat{\gamma}}\in\hat{A}_1^\perp$, and $\gamma+\hat{\gamma}^\dagg-\tau=(\alpha-1,\hat{\alpha}-1)$. This multi-index is only in $\N_=^2$ if $\alpha=\hat{\alpha}$. Suppose this is the case. Then, to be concerned  with two different monomials $g_{\hat{\sigma}}^{\hat{\gamma}}$ and $g_\sigma^\gamma$, one must have consequently $\hat{\sigma}=-\sigma$. This allows however to create a Commutator Chain of type II \plainrefs{Bruschi:Xuereb:2024} that is wholly contained in $\g$, and is composed by elements of strictly increasing degree due to Theorem 53 from \plainrefs{Bruschi:Xuereb:2024}, thereby implying that $\g$ is an infinite-dimensional Lie algebra. Therefore, we can assume without loss of generality $\alpha>\hat{\alpha}$ and can replace, by the same argument as before, $g_\sigma^\gamma$ with $g_{-\sigma\hat{\sigma}}^{\hat{\gamma}+\gamma-\tau}\in\hat{A}_1^\perp$ as the auxiliary sequence in the generic Commutator Chain of interest. In such case, we would then be once again in the remit of Case 1, since $\gamma+\hat{\gamma}^\dagg-\tau$ satisfies the necessary conditions for the first case given that $\chimap{\hat{\gamma}\circ(\gamma+\hat{\gamma}^\dagg-\tau)^\dagg}=\hat{\alpha}(\hat{\alpha}-1)\neq0$.  Proposition~\ref{prop:equivalnce:generic:com:chains} guarantees that the generic Commutator Chain obtained this way generates, to highest degree $d^{(\ell)}$, the same elements as the generic Commutator Chain with seed element $g_{\hat{\sigma}}^{\hat{\gamma}}$ and auxiliary sequence $s^{(\ell)}\equiv u^{(1)}/(\hat{\epsilon}'\alpha\hat{\alpha})\upto{d^{(1)}}g_{-\sigma\hat{\sigma}}^{\hat{\gamma}+\gamma-\tau}$ for all $\ell\in\N_{\geq0}$. All of these elements lie in $\g$, which is therefore infinite dimensional.
			\item[$\boldsymbol{\hat{A}_1^1}$:] Here one has $g_{\hat{\sigma}}^{\hat{\gamma}}\in\hat{A}_1^1$, and $g_{-\sigma\hat{\sigma}}^{\gamma+\hat{\gamma}^\dagg-\tau}\in\hat{A}_1^2$ or $g_{-\sigma\hat{\sigma}}^{\gamma+\hat{\gamma}^\dagg-\tau}\in\hat{A}_1^\perp$ and can, by the same logic as in the previous two cases, simply replace $g_{\hat{\sigma}}^{\hat{\gamma}}$ with $u^{(1)}/(\hat{\epsilon}'\alpha\hat{\alpha})\upto{d^{(1)}}g_{-\sigma\hat{\sigma}}^{\gamma+\hat{\gamma}^\dagg-\tau}$ as the seed element $u^{(0)}$ in the generic Commutator Chain of interest. Thus, one will simply be concerned with one of the previous two cases, since $\gamma+\hat{\gamma}^\dagg-\tau=(\alpha-1,0)$.
		\end{itemize}
		\textbf{Case 3}: $\boldsymbol{\alpha\hat{\beta}=\beta\hat{\alpha}\neq0}$. This requires clearly $\beta,\hat{\beta}>0$ and consequently $g_{\hat{\sigma}}^{\hat{\gamma}}\in\hat{A}_1^\perp$. First, note that in the case $\gamma=\hat{\gamma}$, we can create a Commutator Chain of type II , which lies entirely in $\g$ and contains elements of strictly increasing degree due to Theorem 53 of \plainrefs{Bruschi:Xuereb:2024}, thereby implying that $\g$ is an infinite-dimensional Lie algebra. This is the case, as one assumes $g_\sigma^\gamma\neq g_{\hat{\sigma}}^{\hat{\gamma}}=g_{\hat{\sigma}}^{\gamma}$, implying $\sigma=-\hat{\sigma}$. We are left considering the case $\gamma\neq\hat{\gamma}$ and can assume without loss of generality $\Theta(\gamma+\hat{\gamma}^\dagg)=\gamma+\hat{\gamma}^\dagg$. If this is not the case, one can simply switch the roles of $g_{\hat{\sigma}}^{\hat{\gamma}}$ and $g_\sigma^\gamma$. Here, we have now
		\begin{align*}
			u^{(1)}\upto{d^{(1)}}\hat{\epsilon}\left(\alpha\hat{\beta}-\beta\hat{\alpha}\right)g_{-\sigma\hat{\sigma}}^{\gamma+\hat{\gamma}-\tau}+\hat{\epsilon}'\left(\alpha\hat{\alpha}-\beta\hat{\beta}\right)g_{-\sigma\hat{\sigma}}^{\Theta(\gamma+\hat{\gamma}^\dagg)-\tau}=\hat{\epsilon}'\left(\alpha\hat{\alpha}-\beta\hat{\beta}\right)g_{-\sigma\hat{\sigma}}^{\gamma+\hat{\gamma}^\dagg-\tau},
		\end{align*}
		since we assume that $\alpha\hat{\beta}=\beta\hat{\alpha}$ in this third case, thereby prompting the first term $\hat{\epsilon}(\alpha\hat{\beta}-\beta\hat{\alpha})g_{-\sigma\hat{\sigma}}^{\gamma+\hat{\gamma}-\tau}$ to vanish. The remaining term in the expression above does not vanish, since $\alpha>\beta$ and $\hat{\alpha}>\hat{\beta}$ implies $\alpha\hat{\alpha}>\beta\hat{\beta}$, and because $\gamma\neq\hat{\gamma}$ implies $\gamma+\hat{\gamma}^\dagg-\tau\notin \N_=^2$, guaranteeing $g_{-\sigma\hat{\sigma}}^{\gamma+\hat{\gamma}^\dagg-\tau}\neq0$. Now, we simply need to compute 
		\begin{align*}
			-\chimap{(\gamma+\hat{\gamma}^\dagg-\tau)\circ\gamma^\dagg}=\alpha(\beta+\hat{\alpha}-1)-\beta(\alpha+\hat{\beta}-1)=\alpha(\hat{\alpha}-1)-\beta(\hat{\beta}-1)>0,
		\end{align*}
		since $\alpha>\beta>0$ and $\hat{\alpha}-1>\hat{\beta}-1\geq0$. Therefore, at this point we can replace the seed element $g_{\hat{\sigma}}^{\hat{\gamma}}$ in the original generic Commutator Chain with $g_{-\sigma\hat{\sigma}}^{\gamma+\hat{\gamma}^\dagg-\tau}\in\hat{A}_1^\perp$, which is treated by the first case. This yields, according to Proposition~\ref{prop:equivalnce:generic:com:chains}, a generic Commutator Chain that contains to highest degree the same elements as the generic Commutator Chain with seed element $u^{(1)}/(\hat{\epsilon}'(\alpha\hat{\alpha}-\beta\hat{\beta}))\upto{d^{(1)}}g_{-\sigma\hat{\sigma}}^{\gamma+\hat{\gamma}^\dagg-\tau}$ and auxiliary sequence $g_\sigma^\gamma\equiv s^{(\ell)}$ for all $\ell\in\N_{\geq0}$, and this Commutator Chain lies in $\g$.
		
		Finally, given any two monomials $g_{\hat{\sigma}}^{\hat{\gamma}}\in\hat{A}_1^1\oplus\hat{A}_1^2\oplus\hat{A}_1^\perp$ and $g_\sigma^\gamma\in\hat{A}_1^\perp$ with $g_\sigma^\gamma\neq g_{\hat{\sigma}}^{\hat{\gamma}}$, one can construct a generic Commutator Chain that lies completely in $\g$ and contains elements of strictly increasing degree, which in turn implies that the Lie algebra $\g$ is infinite dimensional.
	\end{proof}
	
	\begin{proposition}   
		Given an  $n$-dimensional Lie algebra $\g$ with vector space basis $\mathcal{B}=\{g_1,\ldots,g_n\}$, then Algorithm~\ref{alg:generating:all:subalgebras:from:subset:of:basis} generates every possible realization of every possible subalgebra of $\g$ that is generated by a subset of~$\mathcal{B}$.
		{
			\begin{algorithm}[t]
				\DontPrintSemicolon
				\KwData{A list $\mathcal{B}$ of elements $g_1,\ldots,g_n$ that form a basis of the finite-dimensional Lie algebra $\g$}
				\KwResult{List $\mathcal{L}$ of all realizations of finite-dimensional Lie algebras generated by a subset of $\mathcal{B}$}
				\SetKwData{Left}{left}\SetKwData{This}{this}\SetKwData{Up}{up}
				\SetKwFunction{Union}{Union}\SetKwFunction{FindCompress}{FindCompress}
				\SetKwInOut{Input}{input}\SetKwInOut{Output}{output}
				\BlankLine
					$\mathcal{L}\leftarrow\{\spn\{g_1\},\ldots,\spn\{g_n\},\g\}$\tcc*[h]{Include all trivial subalgebras}\;
					\BlankLine
					\SetKwFunction{Fctf}{Fct1}
					\SetKwProg{Fn}{Function}{:}{end} 
					\Fn(\tcc*[h]{Generate a basis for $\lie{\mathcal{S}\cup\mathcal{A}}$}){\Fctf{$\mathcal{S},\mathcal{A}$}}{
						\uIf{$\spn\{\mathcal{S}\cup\mathcal{A}\}\neq\g$}{ 
							$\mathcal{H}\leftarrow\emptyset$\;
							\For{$(s,a)\in\mathcal{S}\times\mathcal{A}$}{
								\lIf{$[s,a]\notin\spn\{\mathcal{S}\cup       \mathcal{A}\}$}{
									$\mathcal{H}\leftarrow \mathcal{H}\cup\{[s,a]\}$
								}
							}
							\For{$(a_j,a_k)\in\mathcal{A}\times\mathcal{A}$ \textbf{with} $j<k$}{
								\lIf{$[a_j,a_k]\notin\spn\{\mathcal{S}\cup\mathcal{A}\}$}{
									$\mathcal{H}\leftarrow\mathcal{H}\cup\{[a_j,a_k]\}$
								}
							}
							\lIf{$\mathcal{H}=\emptyset$}{
								\Return $\mathcal{S}\cup\mathcal{A}$
							}
							\lElse{
								\Fctf{$\mathcal{S}\cup\mathcal{A}$, Basis of $\spn\{\mathcal{H}\}$}
							}
						}
						\lElse{
							\Return $\mathcal{B}$
						}
					}
					\BlankLine
					\SetKwFunction{Fctg}{Proc1} 
					\SetKwProg{Fn}{Procedure}{:}{end} 
					\Fn(\tcc*[h]{All realizations of every subalgebra generated by $\mathcal{S}$ and $g_j$ for $j\geq k$}){\Fctg{$\mathcal{S},k$}}{
						\If{$k\leq n\;\wedge\;\spn\{\mathcal{S}\cup\{g_k\}\}\neq\g$}{
							\If{$g_k\in\spn\{\mathcal{S}\}$}{
								\Fctg{$\mathcal{S},k+1$}
							}
							\Else{
								\For{$\ell\in\{k,\ldots,n\}$}{
									$\mathcal{T}\leftarrow$ \Fctf{$\mathcal{S},\{g_\ell\}$}\;
									\If{$\spn\{\mathcal{T}\}\notin\mathcal{L}$}{
                                        $\mathcal{L}\leftarrow\mathcal{L}\cup\{\spn\{\mathcal{T}\}\}$\tcc*[h]{append the space $\spn\{\mathcal{T}\}$ to the list/set of Lie algebras $\mathcal{L}$}\;
										\Fctg{$\mathcal{T},\ell+1$}\;
									}
								}
							}
						}
					}
					\BlankLine
					\tcc*[h]{All possible realizations of every possible subalgebra of $\g$ using only subsets of $\mathcal{B}$ as initial monomials}\;
					\For{$j\in\{1,\ldots,n-1\}$}{
						\Fctg{$\{g_j\},j+1$}\;
					}
				\caption{Generating subalgebras from set of basis elements}\label{alg:generating:all:subalgebras:from:subset:of:basis}
		\end{algorithm}}
	\end{proposition}
	
	\begin{proof}
		Let $\mathcal{B}=\{g_1,\ldots,g_n\}$ an $n$-dimensional basis of the Lie algebra $\g$ as a vector space. The proof is given in two parts: first, we show that the recursive procedure \Fctf{}, defined in Algorithm~\ref{alg:generating:all:subalgebras:from:subset:of:basis}, correctly constructs a subalgebra of $\g$; second, we show that the recursive iteration \Fctg{}, also defined in Algorithm~\ref{alg:generating:all:subalgebras:from:subset:of:basis}, exhausts all subsets of $\mathcal{B}$ that generate distinct realizations of subalgebras of $\g$.
		
		\textbf{Part 1}: We aim at showing that \Fctf{$\mathcal{S},\{g_k\}$} yields a subalgebra of $\g$ if $\mathcal{S}$ is a basis of a subalgebra of $\g$ and $g_k\notin\spn\{\mathcal{S}\}$. The first step is to check if $\spn\{\mathcal{S}\cup\{g_k\}\}=\g$. If this is true, \Fctf{} can simply return the basis $\mathcal{B}$, since $\mathcal{S}\cup\{g_k\}$ generates $\g$, and since $\g$ is a Lie algebra. If $\spn\{\mathcal{S}\cup\{g_k\}\}\neq\g$, one continues with the second step: One computes the commutator of every basis element $s$ in $\mathcal{S}$ with $g_k$. Note that one does not need to compute the commutator of the different basis elements in $\mathcal{S}$, since $\spn\{\mathcal{S}\}=\lie{\mathcal{S}}$. One collects every new element in the set $\mathcal{H}$. If $\mathcal{H}$ is empty, one has $\spn\{\mathcal{S}\cup\{g_k\}\}=\lie{\mathcal{S}\cup\{g_k\}}$ and \Fctf{} returns $\mathcal{S}\cup\{g_k\}$, since $\mathcal{S}\cup\{g_k\}$ is a basis of the vector space $\lie{\mathcal{S}\cup\{g_k\}}$. If $\mathcal{H}$ is non-empty, one chooses a basis $\mathcal{B}_{\mathcal{H}}$ of $\spn\{\mathcal{H}\}$. By construction, no element in $\mathcal{B}_{\mathcal{H}}$ is also in the vector space $\spn\{\mathcal{S}\cup\{g_k\}\}$. Thus one has that $\spn\{\mathcal{S}\cup\{g_k\}\cup\mathcal{B}_{\mathcal{H}}\}$ has the basis $\mathcal{S}\cup\{g_k\}\cup\mathcal{B}_{\mathcal{H}}$, since every element in $\mathcal{H}$ is not in $\spn\{\mathcal{S}\}$ and $g_k$ is also neither in $\spn\{\mathcal{S}\}$ nor in $\spn\{\mathcal{B}_{\mathcal{H}}\}$. We now repeat the process with $\mathcal{S}\cup\{g_k\}$ instead of $\mathcal{S}$ and $\mathcal{A}=\mathcal{B}_{\mathcal{H}}$ instead of $\{g_k\}$. The only new aspect is that one also needs to compute the commutator of every pair $(a_1,a_2)\in\mathcal{A}\times\mathcal{A}$, as one does not know if also such commutators are elements of $\spn\{\mathcal{S}\cup\mathcal{A}\}$. One can now use the anti-symmetric property of the Lie bracket to reduce the number of commutators one has to compute. Iterating this process yields consequently a set $\mathcal{S}'$ that satisfies either the property $\lie{\mathcal{S}'}=\spn\{\mathcal{S}'\}$ and is therefore a basis of the vector space determining $\lie{\mathcal{S}'}\neq\g$, or $\lie{\mathcal{S}'}=\g$. The procedure $\Fctf{}$ returns in the first case $\mathcal{S}'$ and in the second $\mathcal{B}$. The finiteness of $\g$ guarantees furthermore that this process will eventually terminate.
		
		\textbf{Part 2}: Now we have to show that performing \Fctg{$\{s_j\},j+1$} for every $j\in\{1,\ldots,n-1\}$ generates every possible realization of every possible subalgebra of $\g$ that is generated by a subset of $\mathcal{B}$. Note that one initializes the set $\mathcal{L}$ containing every realization with the trivial subalgebras $\lie{g_1},\ldots,\lie{g_n}$ and $\g$ which poses the bases $\{g_1\},\ldots,\{g_n\}$ and $\mathcal{B}$ respectively. Note that $\mathcal{B}$ contains $2^n-1$ non-empty distinct subsets that all span different vector spaces. The naive approach would therefore be to compute the Lie algebra that each subset of $\mathcal{B}$ generates. Here, we do this with a recursive algorithm, which generally reduces the number of Lie closures one has to compute. The process is the following: we start with the set $\mathcal{S}_j=\{g_j\}$ for $j\in\{1,\ldots,n-1\}$ and generate the corresponding Lie algebra $\lie{\{g_j\}}$. This initial step is skipped, since $\spn\{g_j\}=\lie{\{g_j\}}$ allows us to include these algebras in the list of all possible realizations $\mathcal{L}$ from the start. Then, we construct the sets $\mathcal{S}_{j,k_1}=\{g_j,g_{k_1}\}$ by adding to each set $\mathcal{S}_j$ all elements $g_{k_1}$ satisfying $k_1> j$, and we subsequently generate the corresponding Lie algebras $\lie{\mathcal{S}_{j,k_1}}$ of these sets. The next step consists of constructing the sets $\mathcal{S}_{j,k_1,k_2}=\{g_j,g_{k_1},g_{k_2}\}$ by adding to each set $\mathcal{S}_{j,k_1}$ all elements $g_{k_2}$ that satisfy $k_2>k_1$. Repeating this procedure generates clearly every possible subset of $\mathcal{B}$ containing the element $g_j$. The ordering condition $j<k_{1}<\ldots<k_{m}\leq n$ guarantees furthermore that every set $\mathcal{S}_{j,k_1,\ldots,k_m}$ is unique. We do not have to perform \Fctg{$\{g_n\},n+1$}, since $\{g_n\}$ is the only set that can be created this way and the corresponding vector space $\lie{\{g_n\}}$ is already included in the set $\mathcal{L}$ because $n+1>n$. Thus, one has to perform \Fctg{$\{g_j\},j+1$} only for all $j\in\{1,\ldots,n-1\}$.
		
		We do now have to consider every step in the algorithm, where following the above construction is skipped, to verify that this does not change the set containing every possible realization $\mathcal{L}$: Let us consider \Fctg{$\mathcal{S},k$}, where the input of \Fctg{} is $\mathcal{S}$, a basis of the Lie algebra generated by $\mathcal{S}_{j,k_1,\ldots,k_m}$ and an integer $k>m$. Here, one performs the following steps and checks: 
		\begin{enumerate}
			\item If $k>n$ (Algorithm~\ref{alg:generating:all:subalgebras:from:subset:of:basis} line 13): One has already reached the last element in $\mathcal{B}$ and cannot add $g_{k}$ to $\mathcal{S}$, since $g_k$ does not exists. Thus $\mathcal{S}$ is already the largest possible set given the construction rules.
			\item If $\spn\{\mathcal{S}\cup\{g_k\}\}=\g$ (Algorithm~\ref{alg:generating:all:subalgebras:from:subset:of:basis} line 13): Adding any element $g_{\ell}$ to $\mathcal{S}\cup\{g_k\}$ cannot generate a different algebra than $\g$, so one can skip all of the sets $\mathcal{S}_{j,k_1,\ldots,k_m,k,\ell_1,\ldots,\ell_p}$ for $k<\ell_1<\ell_2<\ldots<\ell_p$, since they would all generate $\g$ which is already included in $\mathcal{L}$. 
			\item If $g_k\in\spn\{\mathcal{S}\}$ (Algorithm~\ref{alg:generating:all:subalgebras:from:subset:of:basis} line 14): Adding $g_k$ to $\mathcal{S}$ will not generate a new algebra. Thus, one can skip the set $\mathcal{S}_{j,k_1,\ldots,k_m,k}$ as it generates the same Lie algebra as $\mathcal{S}_{j,k_1,\ldots,k_m}$, and one needs to try to add the next element $g_{k+1}$ to $\mathcal{S}$.
			\item If none of the above conditions hold (Algorithm~\ref{alg:generating:all:subalgebras:from:subset:of:basis} lines 15-20): one computes a basis $\mathcal{T}$ for the Lie algebra $\lie{\mathcal{S}\cup\{g_\ell\}}$ for every $\ell\in\{k,\ldots,n\}$ (Algorithm~\ref{alg:generating:all:subalgebras:from:subset:of:basis} lines 16 and 17).
			Then, one checks if $\spn\{\mathcal{T}\}$ is already included in $\mathcal{L}$ (see Algorithm~\ref{alg:generating:all:subalgebras:from:subset:of:basis} line 18). If it is not, one repeats this process with $\mathcal{T}$ instead of $\mathcal{S}$ and $\ell+1$ instead of $k$ (see Algorithm~\ref{alg:generating:all:subalgebras:from:subset:of:basis} lines 20 and 21) and adds $\spn\{\mathcal{T}\}$ to $\mathcal{L}$ (Algorithm~\ref{alg:generating:all:subalgebras:from:subset:of:basis} line 19). But if it is already included in $\mathcal{L}$, one does not repeat this process. I.e., one does not compute the algebra generated by the sets $\mathcal{S}\cup\{g_{\ell},g_{\ell_2},\ldots,g_{\ell_p}\}$ for $\ell<\ell_1<\ldots<\ell_p\leq n$ and $p\geq 1$. The reason for this is simple: If the Lie algebra $\spn\{\mathcal{T}\}$ has been obtained in a previous step, one did already compute all algebras generated by the successive sets that span the same vector spaces as $\mathcal{S}\cup\{g_{\ell},g_{\ell_2},\ldots,g_{\ell_m}\}$ and generate therefore the same Lie algebras. Note that the initialization of $\mathcal{L}$ does not interfere with this, since one only considers Lie algebras that are generated by sets with at least two elements.
		\end{enumerate}
		These final steps conclude the proof.
	\end{proof}
	
	In order to help the reader better understand our Algorithm~\ref{alg:generating:all:subalgebras:from:subset:of:basis}, we apply it to an example where the Lie algebra is $\g=\lie{\{s_1=i,s_2=ia^\dagg a, s_3=g_+^{\iota_1},s_4=g_-^{\iota_1},s_5=g_+^{2\iota_1},s_6=g_-^{2\iota_1}\}}$. We illustrate each step in which a new Lie subalgebra is added to the set $\mathcal{L}$ of realizations of subalgebras, also indicating the parent algebra through appropriate indentation. The corresponding illustration of the execution of the algorithm is shown in Figure~\ref{fig:example:algorithm}. We are now in the position to state one of our main results.
	
	\begin{figure}
		\centering
		\includegraphics[width=0.85\linewidth]{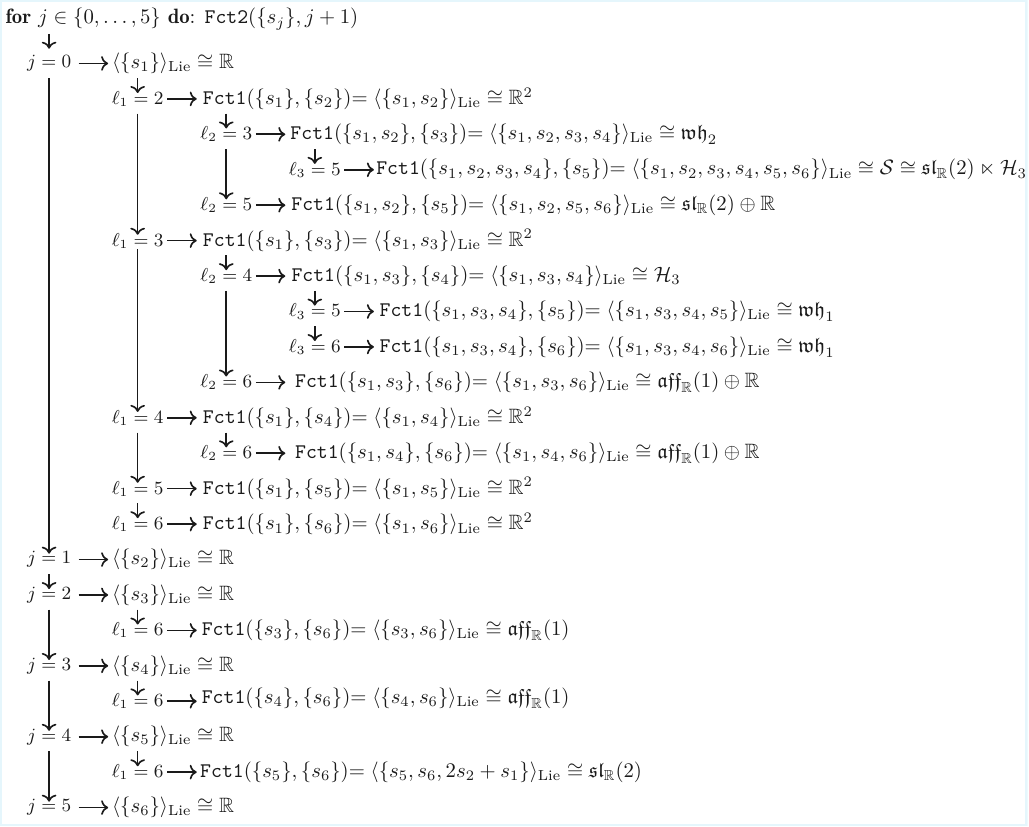}
		\caption{Schematic visualization of Algorithm~\ref{alg:generating:all:subalgebras:from:subset:of:basis} applied to the basis $\mathcal{B}=\{s_1=i,s_2=ia^\dagg a, s_3=g_+^{\iota_1},s_4=g_-^{\iota_1},s_5=g_+^{2\iota_1},s_6=g_-^{2\iota_1}\}$. Here we omit to include the trivial Lie subalgebras $\lie{s_1},\ldots,\lie{s_6}$, and $\lie{\mathcal{B}}$ generated by subsets of $\mathcal{B}$ in the initial $\mathcal{L}$-set of all such subalgebras. Instead, we initialize $\mathcal{L}$ as the empty set and generate these subalgebras explicitly. This is achieved by performing the initial \texttt{for}-loop over the indices $\{0,1,\ldots,5\}$, defining $\{s_0\}:=\emptyset$, and not skipping the generation step when $\lie{\mathcal{S}\cup \{s_k\}}=\lie{\mathcal{B}}$ occurs the first time. Only steps in which a new realization of a Lie subalgebra is generated and added to $\mathcal{L}$ are visualized. Horizontal lines indicate points at which the function \texttt{Fct1} is executed, while vertical lines represent \texttt{for}-loops initiated by the recursive procedure \texttt{Proc2}, where new elements are added to the basis of the previously generated subalgebra. For each subalgebra generated via \texttt{Fct1}, we also provide its isomorphism class.}
		\label{fig:example:algorithm}
	\end{figure}

	\begin{theorem}[Subalgebra classification]\label{thm:glossary:monomial:genereted}
		Any non-abelian finite-dimensional subalgebra $\g$ of the skew-hermitian Weyl algebra $\hat{A}_1$ that is generated by a finite set of single monomials $\mathcal{G}:=\{g^{\gamma _p}_{\sigma _p}\,\mid\, p\in \mathcal{R},\, |\mathcal{R}|<\infty\}$ is one of eight possible algebras listed in Figure~\ref{fig:list:non:abelian:finite:algebras:generated:by:single:g:terms}. This figure enumerates, furthermore, all possible realizations under the given constraints. 
		
		\begin{figure}[htpb]
			\centering
			\includegraphics[width=1.0\linewidth]{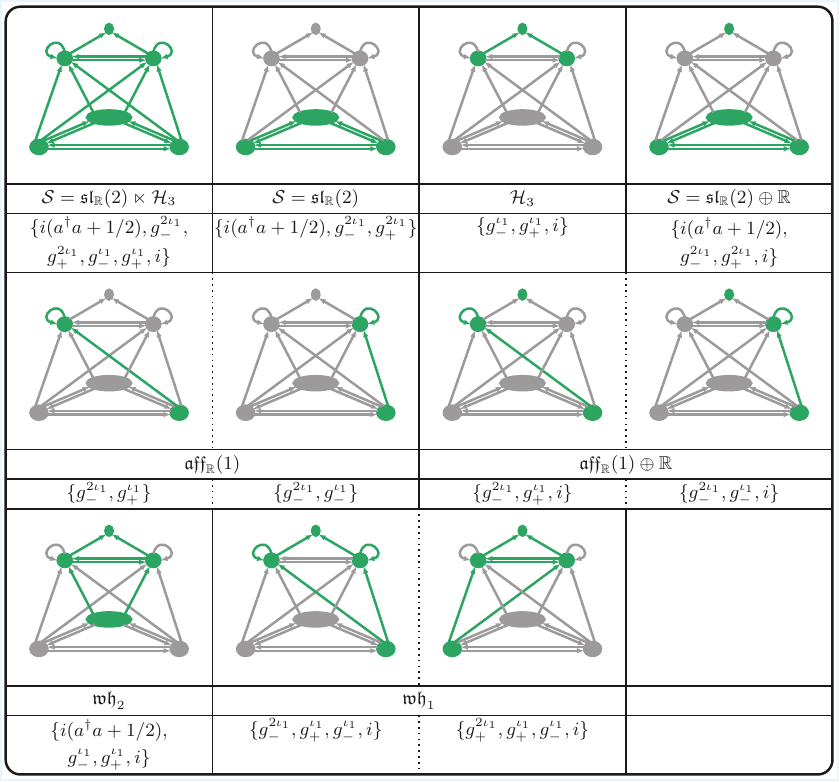}
			\caption{Glossary of all non-abelian and finite-dimensional Lie algebras that can be realized in $\hat{A}_1$ and are only generated from monomials. The directed graphs use the same basis as the depiction of the Schr\"odinger algebra in Figure~\ref{clone:trooper}(b). The realizations of the specific subalgebra under consideration is highlighted in green, while the remaining edges and vertices are shown in gray. The dimension of each realization is given by the number of green vertices.}
			\label{fig:list:non:abelian:finite:algebras:generated:by:single:g:terms}
		\end{figure}
	\end{theorem}
	
	\begin{proof}
		Let $\mathcal{G}:=\{g^{\gamma _p}_{\sigma _p}\,\mid\, p\in \mathcal{R},\, |\mathcal{R}|<\infty\}$ be set of single monomials generating the Lie algebra $\g:=\langle\mathcal{G}\rangle_{\mathrm{Lie}}$.
		By Theorem~\ref{thm:andreea:sona}, Proposition~\ref{prop:algebra:g0:gperp:divergence}, and Proposition~\ref{prop:algebra:g=:gperp:space}, the set $\mathcal{G}$ can contain at most one monomial $g_\sigma^\gamma\in\hat{A}_1^\perp$ and if this is the case, one must have $\mathcal{G}\subseteq\{ci,g_\sigma^\gamma\}$. Thus, any finite-dimensional Lie algebra $\g$ generated by a set of single monomials $\mathcal{G}$ that includes an element $g_\sigma^\gamma\in\hat{A}_1^\perp$ must be abelian. Similarly, Propositions~\ref{prop:commutator:g0:g=:g0:g=},~\ref{prop:algebra:g1:g2:g=:space}, and~\ref{prop:algebra:g=:gperp:space} state that any $\mathcal{G}$ generating a finite-dimensional Lie algebra and containing a monomial $g_+^{\tilde{\gamma}}\in\hat{A}_1^=$ must be a subset of $\hat{A}_1^0\oplus\hat{A}_1^=$, and consequently generates an abelian Lie algebra by virtue of Proposition~\ref{prop:commutator:g0:g=:g0:g=}. Given these observations we conclude that any set $\mathcal{G}$ that generates a non-abelian finite-dimensional Lie algebra must be a subset of the set $\mathcal{G}_{\mathcal{S}}=\{i,ia^\dagg a,g_+^{\iota_1},g_-^{\iota_1},g_+^{2\iota_1},g_-^{2\iota_1}\}$. One can apply Algorithm~\ref{alg:generating:all:subalgebras:from:subset:of:basis} with the basis $\mathcal{B}=\{s_1=i,s_2=ia^\dagg a, s_3=g_+^{\iota_1},s_4=g_-^{\iota_1},s_5=g_+^{2\iota_1},s_6=g_-^{2\iota_1}\}$ to determine all possible subsets $\mathcal{G}\subseteq\mathcal{G}_{\mathcal{S}}$ that generate finite-dimensional Lie algebras. Algorithm~\ref{alg:generating:all:subalgebras:from:subset:of:basis} yields the following realizations in terms of their dimensionality $d$:
		\begin{enumerate}
			\item Here we have $\boldsymbol{d=1}$. The following six realizations are found: 
            \begin{align*}
                \lie{\{i\}},\quad\lie{\{ia^\dagg a\}},\quad\lie{\{g_+^{\iota_1}\}},\quad\lie{\{g_-^{\iota_1}\}},\quad\lie{\{g_+^{2\iota_1}\}},\quad\lie{\{g_-^{2\iota_1}\}}.
            \end{align*}
            These are clearly all abelian Lie algebras and hence mutually isomorphic as one-dimensional Lie algebras, and therefore they are also isomorphic to $\mathbb{R}$.
			\item Here we have $\boldsymbol{d=2}$. The following seven realizations are found:
            \begin{align*}
                \lie{\{i,ia^\dagg a\}},\quad \lie{\{i,g_+^{\iota_1}\}},\quad \lie{\{i,g_-^{\iota_1}\}},\quad \lie{\{i,g_+^{2\iota_1}\}},\quad \lie{\{i,g_-^{2\iota_1}\}},\quad \lie{\{g_+^{\iota_1},g_-^{2\iota_1}\}},\quad \lie{\{g_-^{\iota_1},g_-^{2\iota_1}\}}.
            \end{align*}
            The first five algebras are clearly abelian and therefore isomorphic to $\R^2$. The last two are non-abelian and therefore both isomorphic to $\aff_\R(1)$, since $\R^2$ and $\aff_\R(1)$ are the only two-dimensional Lie algebras up to isomorphism\plainrefs{Andrada:2005,DelBarco:2025}.
			\item Here we have $\boldsymbol{d=3}$. The following four realizations are found: 
            \begin{align*}
                \lie{\{i,g_+^{\iota_1},g_-^{2\iota_1}\}},\quad \lie{\{i,g_-^{\iota_1},g_-^{2\iota_1}\}},\quad \lie{\{i,g_+^{\iota_1},g_-^{\iota_1}\}}, \quad\lie{\{g_+^{2\iota_1},g_-^{2\iota_1},i(a^\dagg a+1/2)\}}.
            \end{align*}
            We now proceed to classify these.
			
			We start with the algebra $\lie{\{i,g_+^{\iota_1},g_-^{2\iota_1}\}}$. Here we can define the basis elements $e_1=g_-^{2\iota_1}/2$, $e_2=g_+^{\iota_1}$ and $e_3=i$. Using Table~\ref{tab:Full:Commutator:Algebra:(pesudo):schroedinger:algebra}, the commutation relations are: $[e_1,e_2]=e_2$, $[e_1,e_3]=0$, and $[e_2,e_3]=0$. Therefore $\lie{\{i,g_+^{\iota_1},g_-^{2\iota_1}\}}\cong\R\oplus\aff_\R(1)$ \plainrefs{Andrada:2005,DelBarco:2025}.
			Similarly, for the algebra $\lie{\{i,g_-^{\iota_1},g_-^{2\iota_1}\}}$, define the basis elements $e_1=-g_-^{2\iota_1}/2$, $e_2=g_-^{\iota_1}$, and $e_3=i$. Computing the commutation relation, allows to conclude $\lie{\{i,g_-^{\iota_1},g_-^{2\iota_1}\}}\cong\R\oplus\aff_\R(1)$. 
			Proposition~\ref{prop:A11:algbera:heisenberg} states that $\lie{\{i,g_+^{\iota_1},g_-^{\iota_1}\}}\cong\gh_1$ and Proposition~\ref{prop:algebra:g2:space} asserts that $\lie{\{g_+^{2\iota_1},g_-^{2\iota_1},i(a^\dagg a+1/2)\}}\cong\slR{2}$.
			According to Theorem 1.1 in \plainrefs{Andrada:2005}, the Lie algebras $\slR{2}$, $\gh_1$, and $\aff_\R(1)\oplus\R$ are pairwise non-isomorphic.
			\item Here we have $\boldsymbol{d=4}$. The following four realizations are found: 
            \begin{align*}
                 \lie{\{i,ia^\dagg a,g_+^{\iota_1},g_-^{\iota_1}\}},\quad\lie{\{i,ia^\dagg a,g_+^{2\iota_1},g_-^{2\iota_1}\}},\quad\lie{\{i,g_+^{\iota_1},g_-^{\iota_1},g_+^{2\iota_1}\}}, \quad\lie{\{i,g_+^{\iota_1},g_-^{\iota_1},g_-^{2\iota_1}\}}.
            \end{align*}
            We conclude, using Proposition~\ref{prop:algebra:go:g1:space}, that $\lie{\{i,ia^\dagg a,g_+^{\iota_1},g_-^{\iota_1}\}}\cong\wh_2$, while we conclude, using Proposition~\ref{prop:algebra:g0:g2:space}, that $\lie{\{i,ia^\dagg a,g_+^{2\iota_1},g_-^{2\iota_1}\}}\cong\slR{2}\oplus\R$. We now proceed to classify the remaining two algebras. 
			
			Consider the algebra $\lie{\{i,g_+^{\iota_1},g_-^{\iota_1},g_+^{2\iota_1}\}}$. Here, we can choose the basis elements $e_1=4i$, $e_2=g_+^{\iota_1}+g_-^{\iota_1}$, $e_3=g_+^{\iota_1}-g_-^{\iota_1}$, and $e_4=g_+^{2\iota_1}/2$. Using Table~\ref{tab:Full:Commutator:Algebra:(pesudo):schroedinger:algebra}, it is easy to confirm that, according to Definition~\ref{def:Wigner:Heisenberg:algebra:1}, $\lie{\{i,g_+^{\iota_1},g_-^{\iota_1},g_+^{2\iota_1}\}}\cong\wh_1$, since $[e_2,_3]=e_1$, $[e_2,e_4]=e_2$, and $[e_3,e_4]=-e_3$, while $e_1$ commutes with every other element. 
			We are left with considering the algebra $\lie{\{i,g_+^{\iota_1},g_-^{\iota_1},g_-^{2\iota_1}\}}$. Here, we can choose the following basis elements: $e_1=2i$, $e_2=g_-^{\iota_1}$, $e_3=g_+^{\iota_1}$, and $e_4=g_-^{2\iota_1}/2$. Once again, the commutation relations are: $[e_2,e_3]=e_1$, $[e_2,e_4]=e_2$, and $[e_3,e_4]=-e_3$ while $e_1$ commutes with every other element. Hence, one finds $\lie{\{i,g_+^{\iota_1},g_-^{\iota_1},g_-^{2\iota_1}\}}\cong\wh_1$. 
			\item Here we have $\boldsymbol{d=6}$. The only realization in this case is $\lie{\{i,ia^\dagg a,g_+^{\iota_1},g_-^{\iota_1},g_+^{2\iota_1},g_-^{2\iota_1}\}}$ which, by Proposition~\ref{prop:schroedinger:algebra}, is isomorphic to the Schr\"odinger algebra~$\mathcal{S}\cong\slR{2}\ltimes\gh_1$. \qedhere
		\end{enumerate}
	\end{proof}
	
	We believe that it may be helpful to highlight alternative conventions and names that are commonly used in the literature for the algebras obtained in Theorem~\ref{thm:glossary:monomial:genereted}. The Heisenberg algebra $\gh_1$ is also known as the reduced classical Galilei algebra $\mathrm{A}\Bar{\mathrm{G}}_1(1)$ \plainrefs{Nesterenko:2016}, which is sometimes also denoted by $A_{3,1}$, $\g_{3,1}$, $L_{3,2}$, or $L_0^4$ \plainrefs{Popovych:2003,Mubarakzyanov:1963,DeGraaf:2007,DeGraaf:2004}. Instead of $\gh_1$ it is also common to denote the Heisenberg algebra by $\mathcal{H}_3$, $H_1$, or $\gh_3$ \plainrefs{TST:2006,Vinet:2011,Andrada:2005}. The Schr\"odinger algebra $\mathcal{S}\cong\slR{2}\ltimes\gh_1$ is also referred to as the special Galileo algebra $\mathrm{AG}_3(1)$ \plainrefs{Nesterenko:2016}. The Wigner-Heisenberg algebra $\wh_1$ is sometimes denoted $A_{4,8}^{b=-1}$ or $\gd_4$, while the Wigner-Heisenberg algebra $\wh_2$ is denoted by $A_{4,9}^{a=0}$ or $\gd_{4,0}'$ \plainrefs{Popovych:2003,Andrada:2005}. By identifying $\tilde{e}_0=-e_4$, $\tilde{e}_1=e_2$, $\tilde{e}_2=e_3$, and $\tilde{e}_3=e_1$, one sees that $\wh_1\cong\gd_4$ \plainrefs{Andrada:2005}. Moreover, it is immediate that $\gd_{4,0}^{\prime}$ from \plainrefs{Andrada:2005} is isomorphic to the Wigner-Heisenberg algebra $\wh_2$. Moreover, the affine Lie algebra $\aff_\R(1)$ is also denoted by $\aff(\R)$, $A_{2,1}$, or $\g_{2,1}$ \plainrefs{Andrada:2005,Popovych:2003,Mubarakzyanov:1963}

	\section{Finite-dimensional Lie algebras containing a free Hamiltonian\label{section:lie:algebras:with:free:hamiltonian}}
	In this section we investigate the structure of finite-dimensional subalgebras of the skew-hermitian Weyl algebra $\hat{A}_1$ that include the free Hamiltonian term of the form $i(a^\dagg a+c)$, where $c\in\R$. The presence of such a term is physically motivated since it typically represents the free component of the Hamiltonian that induces free time-evolution of the Hamiltonian, and is normally part of the drift component when studying the system from the perspective of quantum control.
	Our goal is to determine necessary and sufficient conditions under which the inclusion of a free Hamiltonian term into a set of monomials still allows for the generated Lie algebra to be finite dimensional. We build on the previous section, which was concerned with the classification of monomial-generated algebras, and extend the analysis to arbitrary polynomials.
	This section culminates in the core structural result Theorem~\ref{thm:core:result}, which characterizes all finite-dimensional subalgebras of $\hat{A}_1$ that contain the term $i(a^\dagg a+c)$. Theorem~\ref{thm:all:physically:relevant:algebras:list} provides furthermore an exhaustive list of all such Lie algebras that are non-abelian.

	\subsection{Preliminary results}
	
	We start by observing that Propositions~\ref{prop:algebra:g0:gperp:divergence}, ~\ref{prop:algebra:g1:g2:g=:space}, and~\ref{prop:algebra:g=:gperp:space}, as well as Theorem~\ref{thm:andreea:sona}, remain valid even when considering the sets of monomials mentioned within and including terms of lower degree. This motivates the following lemma, which generalizes these results to linear combinations of monomials.
	
	\begin{lemma}\label{lem:adding:lower:degree:terms:extension}
		Let $e_1,e_2\in\hat{A}_1$ be two polynomials with degrees $d_1$, and $d_2$, respectively, such that $e_1\upto{d_1}c_1g_{\sigma_1}^{\gamma_1}$, and $e_2\upto{d_2} c_2 g_{\sigma_2}^{\gamma_2}$ with $c_1,c_2\in\R\setminus\{0\}$. Then, the Lie algebra $\g:=\lie{\{e_1,e_2\}}$ is infinite dimensional if either one of the following holds:
		\begin{enumerate}[label = (\alph*)]
			\item $g_{\sigma_1}^{\gamma_1}\in\hat{A}_1^1\oplus\hat{A}_1^2\oplus\hat{A}_1^\perp$ and $g_{\sigma_2}^{\gamma_2}=g_+^{\tilde{\gamma}}=2i (a^\dagg)^{\tilde{\alpha}}a^{\tilde{\alpha}}$ with $\tilde{\alpha}\in \N_{\geq2}$, or vice versa;
			\item $g_{\sigma_1}^{\gamma_1}\in \hat{A}_1^0\oplus\hat{A}_1^1\oplus\hat{A}_1^2$ and $g_{\sigma_2}^{\gamma_2}\in\hat{A}_1^\perp$ with $g_{\sigma_1}^{\gamma_1}\neq 2i$, or vice versa.
		\end{enumerate}
	\end{lemma}
	
	\begin{proof}
		The proof of this claim is a straightforward extension of Propositions~\ref{prop:algebra:g0:gperp:divergence}, ~\ref{prop:algebra:g1:g2:g=:space}, and~\ref{prop:algebra:g=:gperp:space}, as well as Theorem~\ref{thm:andreea:sona}, using Proposition~\ref{prop:equivalnce:generic:com:chains}. It is simply sufficient to repeat the proof of each claim by constructing generic Commutator Chains $C^{\mathrm{gen}}_1$ using the single monomials $g_{\sigma_1}^{\gamma_1}$ and $g_{\sigma_2}^{\gamma_2}$. These Commutator Chains contain elements of strictly increasing $d^{(\ell)}$ as a function of $\ell$. Proposition~\ref{prop:equivalnce:generic:com:chains} implies then that these generic Commutator Chains are equivalent to generic Commutator Chains $C^{\mathrm{gen}}_2$ that are constructed in the same fashion but using $e_1/c_1$ and $e_2/c_2$ instead of $g_{\sigma_1}^{\gamma_1}$ and $g_{\sigma_2}^{\gamma_2}$, with the property $C^{\mathrm{gen}}_1\upto{d^{(\ell)}} C^{\mathrm{gen}}_2$. Therefore, $C^{\mathrm{gen}}_2$ does also contain elements of strictly increasing degree, thus implying that the Lie algebra $\g:=\lie{\{e_1,e_2\}}$ is infinite dimensional.
	\end{proof}  
	
	\begin{lemma}\label{lem:linear:two:combinations}
		Consider a set $\mathcal{G}\subseteq\hat{A}_1$ that contains at least the two elements $ia^\dagg a+ic$ and $e_1\upto{d_0}c_1g_{\sigma_1}^{\gamma_1}+c_2g_{\sigma_2}^{\gamma_2}\neq0$ with $c,c_1,c_2\in\R$, $\gamma_1,\gamma_2\in\N_{\geq0}^2\setminus\N_=^{2}$ and $|\gamma_1|=|\gamma_2|=d_0=\deg(e_1)\geq 3$. Then the Lie algebra $\g:=\langle\mathcal{G}\rangle_\mathrm{Lie}$ is infinite dimensional.
	\end{lemma}
	
	\begin{proof}
		We start with the observation that $[ia^\dagg a+ic,p]=[ia^\dagg a,p]$ holds for any polynomial $p\in\hat{A}_1$, since $i$ commutes with any element in $\hat{A}_1$. Thus, it is enough to consider the case $ia^\dagg a\in\mathcal{G}$, the extension to elements of the form $ia^\dagg a+ic$ is immediate.
		 
		We start analyzing a preliminary special case:
		If $g_{\sigma_1}^{\gamma_1}=g_{\sigma_2}^{\gamma_2}$, one must have $c_1\neq -c_2$ to satisfy the condition $e\not\simeq_{d_0}0$. This leaves $e\upto{d_0}(c_1+c_2)g_{\sigma_1}^{\gamma_1}\in\hat{A}_1^\perp$, to the highest degree, to be a constant multiple of a single monomial. Note that $g_{\sigma_1}^{\gamma_1}\in\hat{A}_1^\perp$, since $\gamma_1\in \Np{2}\setminus\N_=^{2}$ and $|\gamma_1|\geq3$. Lemma~\ref{lem:adding:lower:degree:terms:extension} then implies that $\g$ is infinite dimensional. 
		This case is trivial and will therefore be excluded from the following discussion. We also exclude similarly the case that either $c_1=0$ or $c_2=0$, since they reduce to the one above. 
		This leaves us with two remaining cases that we now consider separately.
			In the first case, we have 	$\sigma_1=-\sigma_2\equiv\sigma$ and $	\gamma_1=\gamma_2\equiv\gamma$. We compute
			\begin{align*}
				e_2:=[ia^\dagg a,e_1]\upto{d_0}c_1[ia^\dagg a,g_{\sigma}^{\gamma}]+c_2[ia^\dagg a,g_{-\sigma}^{\gamma}]\upto{d_0}c_1\sigma\chimap{\gamma}g_{-\sigma}^{\gamma}-c_2\sigma\chimap{\gamma}g_{\sigma}^\gamma=\sigma\chimap{\gamma}\left(c_1 g_{-\sigma}^{\gamma}-c_2g_{\sigma}^{\gamma}\right)
			\end{align*}
			using Lemma~\ref{lem:commutator:form}. We then introduce the family of elements
			\begin{align*}
				e_\xi:=e_1+\xi e_2\upto{d_0}\left(c_1-\sigma\xi\chimap{\gamma}c_2\right)g_{\sigma}^\gamma+\left(c_2+\sigma\xi\chimap{\gamma}c_1\right)g_{-\sigma}^\gamma,
			\end{align*}
			where $\xi\in\R$, and note that, crucially, $e_\xi\in\g$ for all $\xi$. The choice $\xi_*=-c_2/(\sigma\chimap{\gamma}c_1)$ implies $e_{\xi_*}\upto{d_0}(c_1^2+c_2^2)/c_1 g_\sigma^\gamma\in\hat{A}_1^\perp$ with $e_{\xi_*}\in\g$, and this choice can be made since $c_1\neq0$ and $\chi(\gamma)\neq0$. One has furthermore $e_{\xi_*}\neq0$, since $c_1\neq0\neq c_2$. 
			By Lemma~\ref{lem:adding:lower:degree:terms:extension} and the same argument of the preliminary case, we conclude that $\g$ is infinite dimensional.
			
			In the second case, we have $\gamma_1\neq\gamma_2$. Here, we compute
			\begin{align*}
				e_2:=[ia^\dagg a,e_1]\upto{d_0}c_1[ia^\dagg a, g_{\sigma_1}^{\gamma_1}]+c_2[ia^\dagg a,g_{\sigma_2}^{\gamma_2}]\upto{d_0}\sigma_1\chimap{\gamma_1}c_1g_{-\sigma_1}^{\gamma_1}+\sigma_2\chimap{\gamma_2}c_2g_{-\sigma_2}^{\gamma_2},
			\end{align*}
			using Lemma~\ref{lem:commutator:form}, Proposition~\ref{prop:equivalnce:relation:commutator} and the fact that $\chimap{\gamma\circ \tau}=\chimap{\tau\circ\gamma}=\chimap{\gamma}$, see Proposition~\ref{prop:basic:maps} and Definition~\ref{def:basic:maps}.
			We now compute
			\begin{align*}
				e_3:=[ia^\dagg a,e_2]\upto{d_0}-\chimap{\gamma_1}^2c_1g_{\sigma_1}^{\gamma_1}-\chimap{\gamma_2}^2c_2g_{\sigma_2}^{\gamma_2},
			\end{align*}
			We now note that $e_1$ and $e_3$ are linearly independent, and $e_3\in\g$ since it is obtained as commutators of elements of the algebra. In fact, in order to be linearly dependent, they would have to have $\chimap{\gamma_1}^2c_1=-\kappa\,c_1$ and $\chimap{\gamma_2}^2c_2=-\kappa\,c_2$, which trivially implies $\chimap{\gamma_1}^2=\chimap{\gamma_2}^2$, and therefore either (i) $\chimap{\gamma_1}=\chimap{\gamma_2}$, or (ii) $\chimap{\gamma_1}=-\chimap{\gamma_1}$. Now, $\chimap{\gamma_1}=\alpha_1-\beta_1$ and $\chimap{\gamma_2}=\alpha_2-\beta_2$, and we also know that (iii) $\alpha_1+\beta_1=\alpha_2+\beta_2$ by the assumption $|\gamma_1|=|\gamma_2|$. On the one hand, constraints (i) and (iii) together imply $\gamma_1=\gamma_2$, which leads to a contradiction. On the other hand, constraints (ii) and (iii) together imply $\gamma_2=\gamma_1^\dagg$. This is not possible since $\gamma_1>\gamma_1^\dagg$ and $\gamma_2>\gamma_2^\dagg$ by assumption, which leads to a contradiction. 
			
			Since $e_1$ and $e_3$ are linearly independent, and so are $g_{\sigma_1}^{\gamma_1}$ and $g_{\sigma_2}^{\gamma_2}$, we can then take suitable linear combinations of $e_1$ and $e_3$ to have $\xi_1 e_1+\xi_3 e_3\upto{d_0}g_{\sigma_1}^{\gamma_1}$ for appropriate $\xi_1,\xi_3\in\R$ (the same can be done for $g_{\sigma_2}^{\gamma_2}$). Now, the fact that $e_1\in\g$ and $e_3\in\g$ implies that $\xi_1 e_1+\xi_3 e_3\in\g$. Thus, we have reduced the problem to the preliminary case, and the algebra generated is infinite dimensional. 
		\end{proof}

	\begin{proposition}\label{prop:thm:25:equivalent}
		Let $e\in\hat{A}_1\setminus\{0\}$ be an arbitrary non-vanishing element of the skew-hermitian Weyl algebra $\hat{A}_1$ with a unique expansion $e=\sum_{p\in\mathcal{P}}c_pg_{\sigma_p}^{\gamma_p}$, where $\mathcal{P}\subseteq\N$ is a finite index set and $c_p\neq0$ for all $p\in\mathcal{P}$. Let furthermore $\tilde{\gamma}=(\tilde{\alpha},\tilde{\alpha})\in \N_=^2$ with $\tilde{\alpha}\geq 1$. Then:
        \begin{enumerate}[label = (\alph*)]
            \item $[ia^\dagg a,e]=0$ if and only if $e\in\hat{A}_1^0\oplus\hat{A}_1^=$;
            \item $\deg([ia^\dagg a,e])=\max\{|\gamma_p|\,\mid\,p\in\mathcal{P}\text{ and }\chimap{\gamma_p}\neq0\}$ if $[ia^\dagg a,e]\neq0$;
            \item $[i(a^\dagg)^{\tilde{\alpha}}a^{\tilde{\alpha}},g_\sigma^\gamma]\upto{d}\tilde{\alpha}\sigma\chimap{\gamma}g_{-\sigma}^{\gamma+\tilde{\gamma}-\tau}$ with $d=|\gamma|+|\tilde{\gamma}|-2$;
            \item $[i(a^\dagg)^{\tilde{\alpha}} a^{\tilde{\alpha}},e]=0$ for $\tilde{\alpha}\in\N_{\geq1}$ if and only if $e\in\hat{A}_1^0\oplus\hat{A}_1^=$;
            \item $\deg([i(a^\dagg)^{\tilde{\alpha}}a^{\tilde{\alpha}},e])=\max\{|\gamma_p|+|\tilde{\gamma}|-2\mid\,p\in\mathcal{P}\text{ and }\chimap{\gamma_p}\neq0\}$ if $[ia^\dagg a,e]\neq0$;
            \item if $e_2\upto{|\tilde{\gamma}|} ic(a^\dagg)^{\tilde{\alpha}} a^{\tilde{\alpha}}$ with $c\in\R\setminus\{0\}$ and  $\PP{e}{\hat{A}_1^0\oplus\hat{A}_1^=}=0$, then $\deg([e_2,e])=\deg([i(a^\dagg)^{\tilde{\alpha}}a^{\tilde{\alpha}},e])=|\tilde{\gamma}|+\deg(e)-2:=f$ and $[e_2,e]\upto{f}\tilde{\alpha}c\sum_{p\in\mathcal{P}}c_p\sigma_p\chimap{\gamma_p}g_{-\sigma_p}^{\gamma_p+\tilde{\gamma}-\tau}$.
        \end{enumerate}
	\end{proposition}
	
	\begin{proof}
		Let us start with claim (a): Theorem 25 from \plainrefs{Bruschi:Xuereb:2024} and the linearity of the commutator imply:
		\begin{align*}
			[ia^\dagg a,e]=\sum_{p\in\mathcal{P}}c_p[ia^\dagg a,g_{\sigma_p}^{\gamma_p}]=\sum_{p\in\mathcal{P}}\sigma_p\chimap{\gamma_p}c_p g_{-\sigma_p}^{\gamma_p}.
		\end{align*}
		The individual monomials in the equation above are all linearly independent, if they do not vanish identically, as we assume that $\sum_{p\in\mathcal{P}}c_pg_{\sigma_p}^{\gamma_p}$ is a unique expansion of $e$. Thus, 
		$ia^\dagg a$ and $e$ can only commute 
		if either $\chimap{\gamma_p}=0$ or $\gamma_p\in\N_=^2$ and $\sigma=-$ hold
		for every $p\in\mathcal{P}$, as $g_-^{\tilde{\gamma}}=0$ for any $\tilde{\gamma}\in\N_=^2$. The converse is trivially also true.
		
		In order to complete the proof of (a), it is now convenient to introduce the following claims:
		\begin{enumerate}[label = {\arabic*.}]
			\item \textbf{Claim $A$}: \enquote{$ia^\dagg a$ and $e$ commute},
			\item \textbf{Claim $B_p$}: \enquote{$\chimap{\gamma_p}=0$},
			\item \textbf{Claim $C_p$}: \enquote{$\gamma_p\in \N_=^2$ and $\sigma=-$},
		\end{enumerate}
		which we first use to re-state what we have proven so far here. We have:  
		\begin{align*}
			A\Leftrightarrow \left(\forall\, p\in\mathcal{P}:\,B_p\vee C_p\right).
		\end{align*}
		However, if $\gamma_p\in\N_=^2$, one has trivially $\chimap{\gamma_p}=\alpha_p-\alpha_p=0$, i.e., $C_p\implies B_p$ for each $p\in\mathcal{P}$. We observe now that the following logical statement is correct:
		\begin{align*}
			\left(\left(A\Leftrightarrow \left(\forall\, p\in\mathcal{P}:\,B_p\vee C_p\right)\right)\wedge\left(\forall\, p\in\mathcal{P}:\,C_p\implies B_p\right)\right)\Leftrightarrow\left(A\Leftrightarrow\left(\forall\, p\in\mathcal{P}:\,B_p\right)\right).
		\end{align*}
		This can be straightforwardly verified using a truth table. Therefore, we have proven that $[ia^\dagg a,e]=0$ if and only if $\chimap{\gamma_p}=0$ for all $p\in\mathcal{P}$. The condition $\chimap{\gamma_p}=0$ for all $p\in\mathcal{P}$ implies then that $\alpha_p=\beta_p$ for all $p\in\mathcal{P}$ and consequently $e\in\hat{A}_1^0\oplus\hat{A}_1^=$, since $\PP{e}{\hat{A}_1^1\oplus\hat{A}_1^2\oplus\hat{A}_1^\perp}\neq 0$ is now prohibited by Definition~\ref{def:gens:mon:partition}, which proves (a).
		
		We now move on to claim (b). This claim follows from Lemma~\ref{important:monomial:commutator:lemma} and the observation that, if $[ia^\dagg a,e]\neq0$, then the degree of $[ia^\dagg a,e]$ equals the largest $|\gamma_p|$ for which $\chimap{\gamma_p}\neq0$.
		
		This leaves us with proving the claims (c), (d), (e), and (f). To do this, we introduce the coefficient $\lambda_\sigma$ that is defined by $\lambda_+:=i$ and $\lambda_-:=1$ and compute the commutator of $i(a^\dagg)^{\tilde{\alpha}}a^{\tilde{\alpha}}$ and an arbitrary monomial $g_\sigma^\gamma\in\hat{A}_1$:
		\begin{align*}
			[i(a^\dagg)^{\tilde{\alpha}} a^{\tilde{\alpha}},g_{\sigma}^{\gamma}]&=i\lambda_{\sigma}[(a^\dagg)^{\tilde{\alpha}}a^{\tilde{\alpha}}, (a^\dagg)^{\beta}a^{\alpha}+\sigma(a^\dagg)^{\alpha}a^{\beta}].
		\end{align*}
		Employing the bilinearity of the commutator, the identity $[AB,CD]=A[B,C]D+AC[B,D]+[A,C]DB+C[A,D]B$, as well as the observation that the canonical commutation relations imply $[a^\alpha,a^\beta]=0=[(a^\dagg)^\alpha,(a^\dagg)^\beta]$ for all $\alpha,\beta\in \N_{\geq0}$, we find:
		\begin{align*}
			[i(a^\dagg)^{\tilde{\alpha}} a^{\tilde{\alpha}},g_{\sigma}^{\gamma}]&=i\lambda_{\sigma}\left((a^\dagg)^{\tilde{\alpha}}[a^{\tilde{\alpha}},(a^\dagg)^{\beta}]a^{\alpha}+(a^\dagg)^{\beta}[(a^\dagg)^{\tilde{\alpha}},a^{\alpha}]a^{\tilde{\alpha}}+\sigma(a^\dagg)^{\tilde{\alpha}}[a^{\tilde{\alpha}},(a^\dagg)^{\alpha}]a^{\beta}+\sigma(a^\dagg)^{\alpha}[(a^\dagg)^{\tilde{\alpha}},a^{\beta}]a^{\tilde{\alpha}}\right).
		\end{align*}
		The individual commutators can be evaluated using Lemma 16 from \plainrefs{Bruschi:Xuereb:2024}, and we obtain:
		\begin{align*}
			[i(a^\dagg)^{\tilde{\alpha}} a^{\tilde{\alpha}},g_{\sigma}^{\gamma}]&=i\lambda_{\sigma}\sum_{j=1}^{\min\{\tilde{\alpha},\beta\}}j!\binom{\tilde{\alpha}}{j}\binom{\beta}{j}\left((a^\dagg)^{\tilde{\alpha}+\beta-j}a^{\tilde{\alpha}+\alpha-j}-\sigma(a^\dagg)^{\tilde{\alpha}+\alpha-j}a^{\tilde{\alpha}+\beta-j}\right)\\
			&\quad-i\lambda_\sigma\sum_{j=1}^{\min\{\tilde{\alpha},\alpha\}}j!\binom{\tilde{\alpha}}{j}\binom{\alpha}{j}\left((a^\dagg)^{\tilde{\alpha}+\alpha-j}a^{\tilde{\alpha}+\beta-j}-\sigma(a^\dagg)^{\tilde{\alpha}+\beta-j}a^{\tilde{\alpha}+\alpha-j}\right).
		\end{align*}
		Recalling that $g_\sigma^\gamma=g_\sigma^{(\alpha,\beta)}=\lambda_\sigma((a^\dagg)^\beta a^\alpha+\sigma(a^\dagg)^\alpha a^\beta)$, where any well-ordered multi-index $\gamma\in\N_{\geq0}^2$ satisfies $\alpha\geq \beta$, we can rewrite the expression above as:
		\begin{align}
			[i(a^\dagg)^{\tilde{\alpha}} a^{\tilde{\alpha}},g_{\sigma}^{\gamma}]&=-\sigma\left[\sum_{j=1}^{\min\{\tilde{\alpha},\beta\}}j!\binom{\tilde{\alpha}}{j}\binom{\beta}{j}g_{-\sigma}^{\gamma+\tilde{\gamma}-j\tau}\right]+\sigma\left[\sum_{j=1}^{\min\{\tilde{\alpha},\alpha\}}j!\binom{\tilde{\alpha}}{j}\binom{\alpha}{j}g_{-\sigma}^{\gamma+\tilde{\gamma}-j\tau}\right].\label{eqn:help:general:commutator:multi:adagga}
		\end{align}
		Reordering the sums finally yields:
		\begin{align*}
			[i(a^\dagg)^{\tilde{\alpha}} a^{\tilde{\alpha}},g_{\sigma}^{\gamma}]&=\sigma\left[\sum_{j=1}^{\min\{\tilde{\alpha},\alpha,\beta\}}j!\binom{\tilde{\alpha}}{j}\left(\binom{\alpha}{j}-\binom{\beta}{j}\right)g_{-\sigma}^{\gamma+\tilde{\gamma}-j\tau}\right]+\sigma\left[\sum_{j=\min\{\tilde{\alpha},\alpha,\beta\}+1}^{\min\{\tilde{\alpha},\alpha\}}j!\binom{\tilde{\alpha}}{j}\binom{\alpha}{j}g_{-\sigma}^{\gamma+\tilde{\gamma}-j\tau}\right].
		\end{align*}
		Note now that for $\tilde{\alpha}=1$, this reduces to $[i(a^\dagg)^{\tilde{\alpha}} a^{\tilde{\alpha}},g_{\sigma}^{\gamma}]=[ia^\dagg a,g_\sigma^\gamma]=\sigma\chimap{\gamma}g_{-\sigma}^{\gamma}$, which is consistent with Theorem 25 from \plainrefs{Bruschi:Xuereb:2024} and claim (a). We can furthermore conclude
		\begin{align*}
			[i(a^\dagg)^{\tilde{\alpha}} a^{\tilde{\alpha}},g_{\sigma}^{\gamma}]&\upto{d}\sigma\tilde{\alpha}\chimap{\gamma}g_{-\sigma}^{\gamma+\tilde{\gamma}-\tau},
		\end{align*} 
		where $d=|\gamma|+|\tilde{\gamma}|-2$. This shows claim (c).
		
		Next, we want to prove claims (d) and (e) combined. Note that any monomial $g_{{-\sigma}}^{\gamma+\tilde{\gamma}-j\tau}$ vanishes identically if $\gamma=\gamma^\dagg$ and {$\sigma=+$,} i.e.,if $g_\sigma^\gamma\in\hat{A}_1^0\oplus\hat{A}_1^=$. One does not need to consider the case $\gamma=\gamma^\dagg$ and {$\sigma=-$,} as then $g_-^\gamma=0$ by definition. On the other hand, if $\gamma>\gamma^\dagg$ it follows that the monomials $g_{-\sigma}^{\gamma+\tilde{\gamma}-j\tau}$ are all linearly independent for each $j\in\{1,\ldots,\min\{\tilde{\alpha},\alpha\}\}$. Thus, a necessary condition for the commutator $[i(a^\dagg)^{\tilde{\alpha}} a^{\tilde{\alpha}},g_\sigma^\gamma]$ to vanish is 
		\begin{align*}
			[i(a^\dagg)^{\tilde{\alpha}} a^{\tilde{\alpha}},g_\sigma^\gamma]\upto{d}\sigma\tilde{\alpha}\chimap{\gamma}g_{-\sigma}^{\gamma+\tilde{\gamma}-\tau}=0.
		\end{align*}
		This implies $\alpha=\beta$ and therefore $\gamma=\gamma^\dagg$, which in turn implies $[i(a^\dagg)^{\tilde{\alpha}} a^{\tilde{\alpha}},g_\sigma^\gamma]=0$ identically at all degrees (see the expansion above of the commutator). Nevertheless, the condition $\gamma=\gamma^\dagg$ is ruled out for this case. We can conclude that $[i(a^\dagg)^{\tilde{\alpha}} a^{\tilde{\alpha}},g_\sigma^\gamma]=0$ if and only if $g_\sigma^\gamma\in\hat{A}_n^0\oplus\hat{A}_n^=$, due to claim (a). 
		
		We can now split the index-set $\mathcal{P}$ into the two disjoint index-sets $\mathcal{P}^=$ and $\mathcal{P}^{\neq}$ satisfying $\mathcal{P}^=\cup\mathcal{P}^{\neq}=\mathcal{P}$ and $\mathcal{P}^=\cap \mathcal{P}^\neq=\emptyset$, defined by the property that $\gamma_p\in\N_=^2$ if $p\in\mathcal{P}^=$ and $\gamma_p\in\Np{2}\setminus\N_=^2$ if $p\in\mathcal{P}^\neq$. The commutator of $i(a^\dagg)^{\tilde{\alpha}}a^{\tilde{\alpha}}$ and $e$ can be computed by employing the bilinearity of the commutator. This allows us to compute the individual commutators of $i(a^\dagg)^{\tilde{\alpha}}a^{\tilde{\alpha}}$ with the monomials in the expansion of $e$, which can furthermore be split into two separate sums, one where the monomials possess multi-indices that lie in the space $\N_=^2$, and one where the monomials possess multi-indices that lie in the space $\N_{\geq0}^2\setminus\N_=^2$. Namely, we have
		\begin{align*}
			[i(a^\dagg)^{\tilde{\alpha}}a^{\tilde{\alpha}},e]&=\sum_{p\in\mathcal{P}^=}c_p[i(a^\dagg)^{\tilde{\alpha}}a^{\tilde{\alpha}},g_+^{\gamma_p}]+\sum_{p\in\mathcal{P}^{\neq}}c_p[i(a^\dagg)^{\tilde{\alpha}}a^{\tilde{\alpha}},g_{\sigma_p}^{\gamma_p}].
		\end{align*}
		The first sum vanishes identically. This follows directly from equation \eqref{eqn:help:general:commutator:multi:adagga},  and the fact that $\alpha=\beta$. One is left with:
		\begin{align*}
			[i(a^\dagg)^{\tilde{\alpha}}a^{\tilde{\alpha}},e]&=\sum_{p\in\mathcal{P}^{\neq}}c_p\left(\sigma\sum_{j=1}^{\min\{\tilde{\alpha},\alpha_p,\beta_p\}}j!\binom{\tilde{\alpha}}{j}\left(\binom{\alpha_p}{j}-\binom{\beta_p}{j}\right)g_{-\sigma}^{\gamma_p+\tilde{\gamma}-j\tau}+\sigma\sum_{j=\min\{\tilde{\alpha},\alpha_p,\beta_p\}+1}^{\min\{\tilde{\alpha},\alpha_p\}}j!\binom{\tilde{\alpha}}{j}\binom{\alpha_p}{j}g_{-\sigma}^{\gamma_p+\tilde{\gamma}-j\tau}\right).
		\end{align*}
		Our goal is to prove here that at least all monomials in the expansion above that are of maximal possible degree are unique. In other words, there exist no two monomials $g_{-\sigma}^{\gamma_p+\tilde{\gamma}-j\tau}$ and $g_{-\sigma'}^{\gamma_{p'}+\tilde{\gamma}-j'\tau}$ that coincide for $(\gamma_p,j)\neq(\gamma_{p'},j')$, where either $|\gamma_p|=\max\{|\gamma_p|\mid\,p\in\mathcal{P}^{\neq}\}:=d$ or $|\gamma_{p'}|=d$ and $j=1$ or $j'=1$ respectively. Hence, we study the properties of monomials with multi-indices $\gamma_{p_*}+\tilde{\gamma}-\tau$ in the expansion above, i.e., those that are uniquely described by the tuple $(\gamma_{p_*},1)$, where the index $p_*\in\mathcal{P}^\neq$ shall be chosen such that $|\gamma_{p_*}|=d$. Since non-vanishing monomials can only be linearly dependent if they possess the same multi-index, we must now consider the equation $\gamma_{p_*}+\tilde{\gamma}-\tau=\gamma_{p'}+\tilde{\gamma}-j'\tau$. This can be re-written as $\gamma_{p_*}-\gamma_{p'}=(1-j')\tau$, which can only be satisfied if $\gamma_{p_*}-\gamma_{p'}\in\Z_=^2$ and $\gamma_{p_*}\leq\gamma_{p'}$, since $j'\geq1$ and $\tau \in \N_=^2$. If $\gamma_{p_*}=\gamma_{p'}$, one must also have $j'=1$ and one is therefore considering the same multi-index on both sides of the equation $\gamma_{p_*}+\tilde{\gamma}-\tau=\gamma_{p'}+\tilde{\gamma}-j'\tau$. If, on the other hand $\gamma_{p_*}<\gamma_{p'}$, one must consequently have $j'>1$ and it can therefore never occur that $\gamma_{p_*}+\tilde{\gamma}-\tau=\gamma_{p'}+\tilde{\gamma}-j'\tau$. To see this, observe that two non-vanishing monomials can only be linearly dependent if they have the same degree, but here, one has $|\gamma_{p_*}+\tilde{\gamma}-\tau|=d+|\tilde{\gamma}-2>|d|+|\tilde{\gamma}-2j'$, since $d=\max\{|\gamma_p|\mid\,p\in\mathcal{P}^{\neq}\}\geq |\gamma_{p'}|$. Thus, we conclude that all monomials with multi-index $\gamma_{p_*}+\tilde{\gamma}-\tau$ are unique in the expansion above and have non-vanishing prefactors, since $\alpha_{p_*}>\beta_{p*}$. We can therefore conclude that, $i(a^\dagg)^{\tilde{\alpha}}a^{\tilde{\alpha}}$ and $e$ only commute if $\mathcal{P}^{\neq}=\emptyset$ or equivalently if $e\in\hat{A}_1^0\oplus\hat{A}_1^=$. One has furthermore $\deg([i(a^\dagg)^{\tilde{\alpha}}a^{\tilde{\alpha}},e])=d+|\tilde{\gamma}|-2$ if $\mathcal{P}^{\neq}\neq\emptyset$. This shows (d) and (e). 
		
		Let us turn to the final claim (f). Here, we consider an element $e_2\in\hat{A}_1$ with $e_2\upto{|\tilde{\gamma}|}ic(a^\dagg)^{\tilde{\alpha}}a^{\tilde{\alpha}}=c g_+^{\tilde{\gamma}}/2$ with $c\in\R\setminus\{0\}$ and a second element $e\in\hat{A}_1$ for which we assume $\PP{e}{\hat{A}_1^0\oplus\hat{A}_1^=}=0$. This implies in the previous notation for $e$: $\mathcal{P}^==\emptyset$ and consequently $\deg(e)=d=\max\{|\gamma_p|\,\mid\,p\in\mathcal{P}^{\neq}\}$. Furthemore, $\mathcal{P}\equiv\mathcal{P}^\neq$. We can exploit Proposition~\ref{prop:equivalnce:relation:commutator}, which implies that $[e_2,e]\upto{f}c[i(a^\dagg)^{\tilde{\alpha}}a^{\tilde{\alpha}},e]$ with $f=|\tilde{\gamma}|+d-2$. Furthermore, Proposition~\ref{prop:equivalnce:relation:commutator} and claim (c) imply that
		\begin{align*}
			[i(a^\dagg)^{\tilde{\alpha}}a^{\tilde{\alpha}},e]\upto{f}\tilde{\alpha}\sum_{p\in\mathcal{P}}c_p\sigma_p\chimap{\gamma_p}g_{-\sigma_p}^{\gamma_p+\tilde{\gamma}-\tau},
		\end{align*}
		as the previous result states that $[i(a^\dagg)^{\tilde{\alpha}}a^{\tilde{\alpha}},g_{\sigma_p}^{\gamma_p}]\upto{f}\sigma_p\tilde{\alpha}\chimap{\gamma_p}g_{-\sigma_p}^{\gamma_p+\tilde{\gamma}-\tau}$ for any $|\gamma_p|=d$ modulo any irrelevant term with a smaller degree which can be added to the relation above without changing its validity. We can combine these results to find
		\begin{align*}
			[e_2,e]\upto{f}\tilde{\alpha}c\sum_{p\in\mathcal{P}}c_p\sigma_p\chimap{\gamma_p}g_{-\sigma_p}^{\gamma_p+\tilde{\gamma}-\tau}\quad\text{with}\quad\deg\left(c\sum_{p\in\mathcal{P}}c_p\sigma_p\chimap{\gamma_p}g_{-\sigma_p}^{\gamma_p+\tilde{\gamma}-\tau}\right)=f.
		\end{align*}
		Thus, the previous results imply that   $\deg([e_2,e])=\deg([i(a^\dagg)^{\tilde{\alpha}}a^{\tilde{\alpha}},e])=f$.
	\end{proof}
	
	\begin{proposition}\label{prop:e:iaa:e:commutator:equal:element}
		Consider a set $\mathcal{G}$ that contains at least the elements $ia^\dagg a+ic$ and $e=c_1g_\sigma^\gamma+c_2 g_{-\sigma}^\gamma$, with $c,c_1,c_2\in\R$. Then, one has
		\begin{align*}
			e_1:=[e,[ia^\dagg a+ic,e]]\upto{d}-\left(c_1^2+c_2^2\right)\chimap{\gamma}\chimap{\gamma\circ\gamma}g_+^{\gamma+\gamma^\dagg-\tau},
		\end{align*}
		with $d=2(|\gamma|-1)$ and, crucially, $e_1\in\g:=\lie{\mathcal{G}}$. Furthermore, $d=\deg(e_1)$ if $\gamma\in\Np{2}\setminus \N_=^2$ and $e\neq0$.
	\end{proposition}
	
	\begin{proof}
		The proof is straightforward and, for later convenience, we note that $e_1:=[e,[ia^\dagg a+ic,e]]=[e,[ia^\dagg a,e]]$. We start by computing:
		\begin{align*}
			[ia^\dagg a,e]=c_1[ia^\dagg a,g_{\sigma}^\gamma]+c_2[ia^\dagg a,g_{-\sigma}^\gamma]=\sigma\chimap{\gamma}\left(c_1g_{-\sigma}^\gamma-c_2g_\sigma^\gamma\right),
		\end{align*}
		where we have used Theorem 25 from \plainrefs{Bruschi:Xuereb:2024}. The application of Lemma~\ref{lem:commutator:form} as well as the linearity and anti-symmetric property of the commutator yield then
		\begin{align*}
			[e,[ia^\dagg a,e]]=\sigma\chimap{\gamma}[c_1g_\sigma^\gamma+c_2g_{-\sigma}^\gamma,c_1g_{-\sigma}^\gamma-c_2g_\sigma^\gamma]=\sigma\chimap{\gamma}\left(c_1^2+c_2^2\right)[g_\sigma^\gamma,g_{-\sigma}^\gamma]\upto{d}-\left(c_1^2+c_2^2\right)\chimap{\gamma}\chimap{\gamma\circ\gamma}g_+^{\gamma+\gamma^\dagg-\tau},
		\end{align*}
		where $d=2(|\gamma|-1)$. If now $\gamma>\gamma^\dagg$, one has $\alpha>\beta$ and $\alpha^2>\beta^2$. This implies $\chimap{\gamma}\neq0\neq\chimap{\gamma\circ\gamma}$. Thus, we consequently have $\deg([e,[ia^\dagg a,e]])=d$ when $e\neq0$, since the prefactors cannot vanish, $g_+^{\gamma+\gamma^\dagg-\tau}$ cannot vanish, and the latter is a polynomial of degree $d$. The extension from $ia^\dagg a$ to $ia^\dagg a+ic$ is trivial.
	\end{proof}    
	
	\begin{lemma}\label{lem:help:lem:infinite:algebra:containg:support:perp:and:equal:elements}
		Let $g_+^{\tilde{\gamma}}\in \hat{A}_1^=$ be a monomial and $e=\sum_{p\in\mathcal{P}}c_pg_{\sigma_p}^{\gamma_p}\in \hat{A}_1$ a linear combination of distinct monomials with a non-empty index set $\mathcal{P}\subseteq\N_{\geq0}$. Suppose furthermore $\mathcal{P}$ contains an index $p_*$ for which $c_{p_*}\neq0$, $\gamma_{p_*}>\gamma_{p_*}^\dagg$, and $|\gamma_{p_*}|=\deg(e)$. Then, the generic Commutator Chain $C^{\mathrm{gen}}$ with seed element $u^{(0)}=e$ and auxiliary sequence $s^{(\ell)}\equiv g_+^{\tilde{\gamma}}\in\hat{A}_1^=$ for all $\ell\in\N_{\geq0}$ satisfies
		\begin{align*}
			u^{(\ell)}\upto{d^{(\ell)}}\sum_{p\in\mathcal{P}}c_p^{(\ell)} g_{(-1)^\ell\sigma_p}^{\gamma_p^{(\ell)}},
		\end{align*}
		where $d^{(\ell)}=\deg(e)+\ell(|\tilde{\gamma}|-2)=\deg(u^{(\ell)})$, $\gamma_p^{(\ell)}=\gamma_p+\ell(\tilde{\gamma}-\tau)$, and $c_{p_*}^{(\ell)}\neq 0$ for all $\ell\in\N_{\geq0}$. Moreover, all monomials appearing in the expansion above are unique if they are not identically zero.
	\end{lemma}
	
	\begin{proof}
		We will use induction to prove this claim. Let us start with the base case $\ell=1$. Proposition~\ref{prop:thm:25:equivalent} (c) states $[g_\sigma^\gamma,g_+^{\tilde{\gamma}}]\upto{d}2\hat{\epsilon}\tilde{\alpha}(\alpha-\beta) g_{-\sigma}^{\gamma+\tilde{\gamma}-\tau}$ with $\hat{\epsilon}\in\{\pm\}$ and $d=|\gamma|+|\tilde{\gamma}|-2$, since $g_+^{\tilde{\gamma}}=g_+^{(\tilde{\alpha},\tilde{\alpha})}=2i(a^\dagg)^{\tilde{\alpha}}a^{\tilde{\alpha}}$. One has consequently:
		\begin{align*}
			u^{(1)}=[u^{(0)},s^{(0)}]=[e,g_+^{\tilde{\gamma}}]=\sum_{p\in\mathcal{P}} c_p[g_{\sigma_p}^{\gamma_p},g_+^{\tilde{\gamma}}]\upto{d^{(1)}}2\tilde{\alpha}\sum_{p\in\mathcal{P}}c_p\hat{\epsilon}_p(\alpha_p-\beta_p)g_{-\sigma_p}^{\gamma_p+\tilde{\gamma}-\tau},
		\end{align*}
		with $d^{(1)}=\deg(e)+|\tilde{\gamma}|-2$. A close inspection of the above shows that $u^{(1)}$ is of the desired form with $c_p^{(1)}=2\tilde{\alpha}c_p\hat{\epsilon}_p(\alpha_p-\beta_p)$ and $\gamma_p+\tilde{\gamma}-\tau=\gamma_p^{(1)}$. Proposition~\ref{prop:thm:25:equivalent} implies furthermore $\deg(u^{(1)})=d^{(1)}$. It is furthermore clear that all monomials in the expansion above are unique if they do not vanish. This can easily be seen by assuming the converse: One has $\gamma_p+\tilde{\gamma}-\tau=\gamma_{p'}+\tilde{\gamma}-\tau$, which is only possible if $\gamma_p=\gamma_{p'}$. Since the monomials in the expansion of $e$ are assumed to be distinct, one must consequently have already $p=p'$, implying that each monomial in the expansion of $u^{(1)}$ is unique. Note furthermore that $2\hat{\epsilon}_{p_*}\tilde{\alpha}(\alpha_{p_*}-\beta_{p_*})\neq0$, since $\alpha_{p_*}>\beta_{p_*}$. This also implies that every monomial in the expansion of $u^{(1)}$ with a multi-index $\gamma_p+\tilde{\gamma}-\tau$ has a non-vanishing prefactor, if $\gamma_p>\gamma_p^\dagg$ and $|\gamma_p|=\deg(e)$, as this is the only requirement for $\gamma_{p_*}$.
		
		We continue with the induction step and assume $u^{(\ell)}$ is of the desired form. We compute
		\begin{align*}
			u^{(\ell+1)}&=[u^{(\ell)},g_+^{\tilde{\gamma}}]\upto{d^{(\ell+1)}}\sum_{p\in\mathcal{P}} c_p^{(\ell)}[g_{(-1)^\ell\sigma_p}^{\gamma_p^{(\ell)}},g_+^{\tilde{\gamma}}]\upto{d^{(\ell+1)}}2\sum_{p\in\mathcal{P}}c_p^{(\ell)}\hat{\epsilon}_p\tilde{\alpha}((\alpha_p+\ell(\tilde{\alpha}-1))-(\beta_p+\ell(\tilde{\alpha}-1)))g_{(-1)^{\ell+1}\sigma_p}^{\gamma_p^{(\ell)}+\tilde{\gamma}-\tau}\\
			&\upto{d^{(\ell+1)}}2\sum_{p\in\mathcal{P}}c_p^{(\ell)}\hat{\epsilon}_p\tilde{\alpha}(\alpha_p-\beta_p)g_{(-1)^{\ell+1}\sigma_p}^{\gamma_p^{(\ell)}+\tilde{\gamma}-\tau}.
		\end{align*}
		Noticing that $\gamma_p^{(\ell)}+\tilde{\gamma}-\tau=\gamma_p^{(\ell+1)}$ allows us to conclude that, in the expansion above, all monomials have the desired multi-indices. They are furthermore unique by an analogous discussion as for the base case $\ell=1$. One has furthermore by the induction hypothesis that $c_{p_*}^{(\ell)}\neq0$ and consequently $c_{p_*}^{(\ell+1)}=2\hat{\epsilon}_{p_*}\tilde{\alpha}(\alpha_{p_*}-\beta_{{p_*}})c_{p_*}^{(\ell)}\neq 0 $ since $\alpha_{p_*}>\beta_{p_*}$. Finally, this allows us to use the same argument above to also conclude that $\deg(u^{(\ell+1)})=d^{(\ell+1)}$, thereby completing the proof  by induction.
	\end{proof}
	\begin{lemma}\label{lem:infinite:algebra:containg:support:perp:and:equal:elements}
		Consider a set of polynomials $\mathcal{G}$ containing two elements $e_1=\sum_{p\in\mathcal{P}}c_pg_{\sigma_p}^{\gamma_p}\in\mathcal{G}$ and $e_2=\sum_{q\in\mathcal{Q}}\hat{c}_qg_{\hat{\sigma}_q}^{\hat{\gamma}_q}\in\mathcal{G}$. Then, these two elements generate an infinite-dimensional Lie algebra $\g=\langle\mathcal{G}\rangle_{\mathrm{Lie}}$ if all of the following conditions hold:
		\begin{enumerate}[label = (\arabic*)]
			\item $e_1=c_1g_+^{\tilde{\gamma}}+e_1^{\neq}+e_1^=$, where $c_1\neq0$, $\deg(e_1)=|\tilde{\gamma}|>\deg(e_1^=+e_1^{\neq})$, $\PP{e_1^{\neq}}{\hat{A}_1^0\oplus\hat{A}_1^=}=0$, $e_1^=\in\hat{A}_1^0\oplus\hat{A}_1^=$, and $|\tilde{\gamma}|\geq 4$. 
			\item $e_2=\hat{c}_1g_+^{\tilde{\lambda}}+e_2^{\neq}+e_2^=$, where $\tilde{\lambda}=\tilde{\mu}\tau$, $\hat{c}_1\in\R$, $e_2^{\neq}\neq0$, $\PP{e_2^{\neq}}{\hat{A}_1^0\oplus\hat{A}_1^=}=0$, $|\tilde{\lambda}|>\deg(e_2^=)$, and $e_2^=\in\hat{A}_1^0\oplus\hat{A}_1^=$.
			\item One does not have simultaneously: (3a) $\hat{c}_1\neq0$; (3b) $|\tilde{\lambda}|+\deg(e_2^\neq)=|\tilde{\gamma}|+\deg(e_1^\neq)$; and (3c) the unique expansions $e_1^{\neq}\upto{\deg(e_1^{\neq})}\sum_{p\in\mathcal{P}_1^{\neq}}c_pg_{\sigma_p}^{\gamma_p}$ and $e_2^{\neq}\upto{\deg(e_2^{\neq})}\sum_{q\in\mathcal{Q}_1^{\neq}}\hat{c}_qg_{\hat{\sigma}_q}^{\hat{\gamma}_q}$ with $|\gamma_p|=\deg(e_1^{\neq})$ for all $p\in\mathcal{P}^{\neq}_1$, $|\hat{\gamma}_q|=\deg(e_2^{\neq})$ for all $q\in\mathcal{Q}_1^{\neq}$, and $|\mathcal{P}_1^{\neq}|=|\mathcal{Q}_1^{\neq}|$ satisfy the following property: for every $p\in\mathcal{P}_1^{\neq}$ there exists exactly one $q\in\mathcal{Q}_1^{\neq}$ and vice versa, such that $\gamma_p+\tilde{\gamma}=\hat{\gamma}_q+\tilde{\lambda}$, $\sigma_p=\hat{\sigma}_q$, and $\hat{\mu}\hat{c}_q\chimap{\hat{\gamma}_q}=\hat{\alpha}c_p\chimap{\gamma_p}$, where $\tilde{\gamma}=(\tilde{\alpha},\tilde{\alpha})$ and $\tilde{\lambda}=(\tilde{\mu},\tilde{\mu})$.
		\end{enumerate}
	\end{lemma}
	
	The idea of this proof is simple. We start by considering in an initial case that $e_1=g_+^{\tilde{\gamma}}\in\hat{A}_1^=$ and $e_2\in \hat{A}_1$ is a polynomial that, to highest degree, contains monomials with multi-indices from $\N_{\geq0}^2\setminus\N_=^2$. This allows us to show that the Lie algebra $\lie{\{e_1,e_2\}}$ is infinite dimensional. We will then consecutively add terms to the two elements such that the modified elements cover all possible polynomials that satisfy all of the conditions (1), (2), and (3) from the claim above.
	
	\begin{proof}   
		Let $g_+^{\tilde{\gamma}}\in \hat{A}_1^=$ and $e=\sum_{p\in\mathcal{P}}c_pg_{\sigma_p}^{\gamma_p}\in \hat{A}_1$ a linear combination of distinct monomials with a non-empty index set $\mathcal{P}\subseteq\N_{\geq0}$ and an index $p_*\in\mathcal{P}$ such that $c_{p_*}\neq0$, $\gamma_{p_*}>\gamma_{p_*}^\dagg$, and $|\gamma_{p_*}|=\deg(e)$. Consider now the generic Commutator Chain $C^{\mathrm{gen}}$ with seed element $u^{(0)}=e$ and auxiliary sequence $s^{(\ell)}\equiv g_+^{\tilde{\gamma}}\in\hat{A}_1^=$ for all $\ell\in\N_{\geq0}$. It follows from Lemma~\ref{lem:help:lem:infinite:algebra:containg:support:perp:and:equal:elements} that the chain elements $u^{(\ell)}$ satisfy
		\begin{align*}
			u^{(\ell)}\upto{d^{(\ell)}}\sum_{p\in\mathcal{P}}c_p^{(\ell)} g_{(-1)^\ell\sigma_p}^{\gamma_p^{(\ell)}},
		\end{align*}
		where $d^{(\ell)}=\deg(e)+\ell(|\tilde{\gamma}|-2)$, $\gamma_p^{(\ell)}=\gamma_p+\ell(\tilde{\gamma}-\tau)$, and $c_{p_*}^{(\ell)}\neq 0$ for all $\ell\in\N_{\geq0}$. Moreover, all monomials appearing in the expansion above are unique if they are not identically zero. It follows that any set of polynomials $\mathcal{G}$ containing $g_+^{\tilde{\gamma}}$ and $e$ generates an infinite-dimensional Lie algebra $\g=\lie{\mathcal{G}}$, since all elements $u^{(\ell)}$ of the chain that they generate lie within the algebra and are all linearly independent.
		
		Next, we make the following observation: The monomial $g_+^{\tilde{\gamma}}$ and the restricted polynomial $e_2^\neq := \sum_{p\in\mathcal{P}^\neq}c_pg_{\sigma_p}^{\gamma_p}=\PP{e}{\hat{A}_1^1\oplus\hat{A}_1^2\oplus\hat{A}_1^\perp}$ with the index set $\mathcal{P}^{\neq}\subseteq\mathcal{P}$ defined by all $p\in\mathcal{P}$ for which $\gamma_p>\gamma_p^\dagg$ does hold, still generate an infinite-dimensional Lie algebra. This follows from the observation that, by Proposition~\ref{prop:thm:25:equivalent}, $u_1^{(\ell)}\upto{d^{(\ell)}}u^{(\ell)}$ for all $u_1^{(\ell)}\in C_1^{\mathrm{gen}}$ with $\ell\geq 1$, where $C_1^{\mathrm{gen}}$ is the generic commutator chain with seed elements $u^{(0)}=e_2^\neq$ and auxiliary sequence $s_1^{(\ell)}\equiv g_+^{\tilde{\gamma}}$ for all $\ell\in\N_{\geq0}$.
		
		In the next step, we extend the elements $g_+^{\tilde{\gamma}}$ and $e_2^\neq$ by adding appropriate terms. Hence, we define the polynomial $e_1:=g_+^{\tilde{\gamma}}+e_1^=+e_1^\neq$, where $e_1\upto{|\tilde{\gamma}|}g_+^{\tilde{\gamma}}\in\hat{A}_1^=$ and $e_1^=$ only has non-trivial support in the space $\hat{A}_1^0\oplus\hat{A}_1^=$, while $e_1^\neq$ only has non-trivial support in the space $\hat{A}_1^1\oplus\hat{A}_1^2\oplus\hat{A}_1^\perp$. That is $\deg(e_1)=|\tilde{\gamma}|>\deg(e_1^=+e_1^\neq)$, $\PP{e_1^{\neq}}{\hat{A}_1^0\oplus\hat{A}_1^=}=0$, $e_1^=\in\hat{A}_1^0\oplus\hat{A}_1^=$, and $|\tilde{\gamma}|\geq 4$. We extend similarly the polynomial $e_2^\neq$ by adding a monomial $g_+^{\tilde{\lambda}}\in\hat{A}_n^=$ and a polynomial $e_2^=\in \hat{A}_1^0\oplus\hat{A}_1^=$ that has a smaller degree than $g_+^{\tilde{\lambda}}$, i.e., $\deg(g_+^{\tilde{\lambda}})=|\tilde{\gamma}|>\deg(e_2^=)$. The extension of $e_2^\neq$ to $e^\neq$, obtained by adding $g_+^{\tilde{\lambda}}$ and $e_2^=$, is consequently given by $e_2:=e_2^\neq+e_2^=+g_+^{\tilde{\lambda}}.$ Thus, the elements $e_1$ and $e_2$ satisfy the conditions (1) and (2) respectively, modulo an appropriate overall scalar factor.
		
		Let us now consider the generic Commutator Chain $C_2^{\mathrm{gen}}$ with seed element $u_2^{(0)}=e_2$ and auxiliary chain $s_2^{(\ell)}\equiv e_1\equiv s_2$ for all $\ell\in\N_{\geq0}$, where $e_1^=,e_2^=,g_+^{\tilde{\gamma}},g_+^{\tilde{\lambda}}\in\hat{A}_1^0\oplus\hat{A}_1^=$, $\deg(e_{2}^=)<|\tilde{\lambda}|$, $\deg(e_{1}^=)<|\tilde{\gamma}|$, $\deg(e_1^\neq)<|\tilde{\gamma}|$, and $\PP{e_1^\neq}{\hat{A}_1^0\oplus\hat{A}_1^=}=0$. Proposition~\ref{prop:thm:25:equivalent} allows us to write:
		\begin{align*}
			[u_2^{(0)},s_2]\upto{f^{(1)}}[e_2^\neq,g_+^{\tilde{\gamma}}]-[e_1^\neq,g_+^{\tilde{\lambda}}],
		\end{align*}
		where $f^{(1)}=\max\{|\tilde{\lambda}|+\deg(e_1^\neq)-2,|\tilde{\gamma}|+\deg(e_2^\neq)-2\}$. The case in which the two terms $|\tilde{\lambda}|+\deg(e_1^\neq)-2$ and $|\tilde{\gamma}|+\deg(e_2^\neq)-2$ are equal has to be treated differently from the case where they are not, as in the former both commutators on the right-hand side in the expression above have the same degree, while in the later they do not, which allows to disregard the one with lower degree. Let us consider these cases separately.
		
		The first case we want to treat shall be the case where $|\tilde{\lambda}|+\deg(e_1^\neq)\neq|\tilde{\gamma}|+\deg(e_2^\neq)$.
		Here, we can choose without loss of generality $|\tilde{\gamma}|+\deg(e_2^\neq)>|\tilde{\lambda}|+\deg(e_1^\neq)$, since $|\tilde{\gamma}|+\deg(e_2^\neq)<|\tilde{\lambda}|+\deg(e_1^\neq)$ and $|\tilde{\gamma}|>\deg(e_1^\neq)$ implies $|\tilde{\lambda}|>\deg(e_2^\neq)$. Thus, the polynomial $e_2$ satisfies the exact same structural properties as $e_1$, allowing a simple relabelling of the quantities, which brings us to the aforementioned condition that $|\tilde{\gamma}|+\deg(e_2^\neq)>|\tilde{\lambda}|+\deg(e_1^\neq)$. Note that, when performing such relabelling, one has to repeat the derivation of the generic Commutator Chain $C_1^{\mathrm{gen}}$, where the qualitative result does not change. 		
		
		We proceed here employing induction and showing that $u_2^{(\ell)}\upto{d^{(\ell)}}u_1^{(\ell)}$ for all $\ell\in\N_{\geq1}$. The base case $\ell=1$ is immediately verified by observing that $f^{(1)}=d^{(1)}$ and $\deg([e_1^\neq,g_+^{\tilde{\lambda}}])\leq\deg(e_1^\neq)+|\tilde{\lambda}|-2<|\tilde{\gamma}|+\deg(e_2^\neq)-2=d^{(1)}$. This implies immediately $u_2^{(1)}\upto{d^{(1)}}u_1^{(1)}$.		
		We then proceed with the induction argument and assume that $u_2^{(\ell)}\upto{d^{(\ell)}}u_1^{(\ell)}$ for an $\ell\in\N_{\geq1}$, which requires us to consider $u_2^{(\ell+1)}=[u_2^{(\ell)},s_2]=[u_2^{(\ell)},g_+^{\tilde{\gamma}}]+[u_2^{(\ell)},e_1^\neq+e_1^=]$. We here have $[u_2^{(\ell)},g_+^{\tilde{\gamma}}]\upto{d^{(\ell+1)}}[u_1^{(\ell)},g_+^{\tilde{\gamma}}]=u_1^{(\ell+1)}$ by the induction hypothesis and Proposition~\ref{prop:equivalnce:relation:commutator}. Furthermore, we also have $\deg([u_2^{(\ell)},e_1^\neq+e_1^=])=d^{(\ell)}+\max\{\deg(e_1^\neq),\deg(e_1^=)\}-2<d^{(\ell)}+|\tilde{\gamma}|-2=d^{(\ell+1)}$ by the induction hypothesis and the initial assumptions. This concludes the proof by induction, and one has $\deg(u_2^{(\ell)})=d^{(\ell)}$, since $\deg(u_1^{(\ell)})=d^{(\ell)}$. This shows that if in condition (3), claim (3b) does not hold, $e_1$ and $e_2$ generate an infinite-dimensional Lie algebra.
		
		In the second case, we have $|\tilde{\lambda}|+\deg(e_1^\neq)=|\tilde{\gamma}|+\deg(e_2^\neq)$.
		Here, we have $|\tilde{\gamma}|>\deg(e_1^\neq)$ by the initial assumption and consequently also $|\tilde{\lambda}|>\deg(e_2^\neq)$. We can furthermore write $e_1^\neq$ as the unique decomposition $e_1^\neq=\sum_{q\in\mathcal{Q}^{\neq}}\hat{c}_qg_{\hat{\sigma}_q}^{\hat{\gamma}_q}$, where $\hat{c}_q\neq0$ and $\gamma_q>\gamma_q^\dagg$ for all $q\in\mathcal{Q}^{\neq}$, since $e_1^\neq$ lies completely in the space $\hat{A}_1^1\oplus\hat{A}_1^2\oplus\hat{A}_1^\perp$ and is therefore formed only of monomials with multi-indices from $\N_{\geq0}^2\setminus\N_=^2$. Proposition~\ref{prop:thm:25:equivalent} now yields
		\begin{align*}
			u_2^{(1)}=[u_2^{(0)},s_2]=[e_2,e_1]\upto{f^{(1)}}[e_2^\neq,g_+^{\tilde{\gamma}}]-[e_1^\neq,g_+^{\tilde{\lambda}}]\upto{d^{(1)}}-2\tilde{\alpha}\sum_{p\in\mathcal{P}^{\neq}}\sigma_pc_p\chimap{\gamma_p}g_{-\sigma_p}^{\gamma_p+\tilde{\gamma}-\tau}+2\tilde{\mu}\sum_{q\in\mathcal{Q}^{\neq}}{\hat{\sigma}}_q\hat{c}_q\chimap{\hat{\gamma}_q}g_{-\hat{\sigma}_q}^{\hat{\gamma}_q+\tilde{\lambda}-\tau},
		\end{align*}
		where $d^{(1)}=f^{(1)}=|\tilde{\lambda}|+\deg(e_1^\neq)-2=|\tilde{\gamma}|+\deg(e_2^\neq)-2$.
		We clearly see that the sets $\{g_{-\sigma_p}^{\gamma_p+\tilde{\gamma}-\tau}\}_{p\in\mathcal{P}^{\neq}}$ and $\{g_{-\hat{\sigma}_q}^{\hat{\gamma}_q+\tilde{\lambda}-\tau}\}_{q\in\mathcal{Q}^{\neq}}$ are, by similar arguments as before, two sets of mutually linearly independent monomials. Thus assuming $\deg(u_2^{(1)})<d^{(1)}$ implies, that for every $p\in\mathcal{P}^{\neq}$ with $|\gamma_p|=\deg(e_1)$, there must exist exactly one $q\in\mathcal{Q}^{\neq}$ such that $\gamma_p+\tilde{\gamma}-\tau=\hat{\gamma}_q+\tilde{\lambda}-\tau$, $\sigma_p=\hat{\sigma}_q$, and $\hat{\mu}\hat{c}_q\chimap{\hat{\gamma}_q}-\tilde{\alpha}c_p\chimap{\gamma_p}=0$. Let us explain this in detail. In the expansion above none of the generator monomials vanishes identically, nor do the prefactors, since $\chimap{\gamma}\neq 0$ if $\gamma>\gamma^\dagg$. Moreover, a monomial $-2\alpha \sigma_p c_p\chimap{\gamma_p}g_{-\sigma_p}^{\gamma_p+\tilde{\gamma}-\tau}$ can never be canceled by any number of other monomials $-2\alpha \sigma_{p'} c_{p'}\chimap{\gamma_{p'}}g_{-\sigma_{p'}}^{\gamma_{p'}+\tilde{\gamma}-\tau}$ with $p,p'\in\mathcal{P}^\neq$, since $\{g_{-\sigma_p}^{\gamma_p+\tilde{\gamma}-\tau}\}_{p\in\mathcal{P}^{\neq}}$ is a linearly independent set. Similarly no monomial $2\tilde{\mu}\hat{\sigma}_q\hat{c}_q\chimap{\hat{\gamma}_q}g_{-\hat{\sigma}_q}^{\hat{\gamma}_q+\tilde{\lambda}-\tau}$ can ever be canceled by any number of monomials $2\tilde{\mu}\hat{\sigma}_{q'}\hat{c}_{q'}\chimap{\hat{\gamma}_{q'}}g_{-\hat{\sigma}_{q'}}^{\hat{\gamma}_{q'}+\tilde{\lambda}-\tau}$ for $q,q'\in\mathcal{Q}^\neq$. However, for $\deg(u_2^{(1)})<d^{(1)}$ to be valid all of these monomials in the expansion of $u_2^{(1)}$ must, to the highest order, be canceled. Due to the linear independence of the monomials $g_\sigma^\gamma$, a monomial $-2\alpha \sigma_p c_p\chimap{\gamma_p}g_{-\sigma_p}^{\gamma_p+\tilde{\gamma}-\tau}$ can therefore only be canceled by a monomial $2\tilde{\mu}\hat{\sigma}_{q}\hat{c}_{q}\chimap{\hat{\gamma}_{q}}g_{-\hat{\sigma}_{q}}^{\hat{\gamma}_{q}+\tilde{\lambda}-\tau}$ and vice versa. However, since all generator monomials in the two sets $\{g_{-\sigma_p}^{\gamma_p+\tilde{\gamma}-\tau}\}_{p\in\mathcal{P}^{\neq}}$ and $\{g_{-\hat{\sigma}_q}^{\hat{\gamma}_q+\tilde{\lambda}-\tau}\}_{q\in\mathcal{Q}^{\neq}}$ are unique, there can always exists only one monomial $2\tilde{\mu}\hat{\sigma}_{q}\hat{c}_{q}\chimap{\hat{\gamma}_{q}}g_{-\hat{\sigma}_{q}}^{\hat{\gamma}_{q}+\tilde{\lambda}-\tau}$ that cancels a monomial $-2\alpha \sigma_p c_p\chimap{\gamma_p}g_{-\sigma_p}^{\gamma_p+\tilde{\gamma}-\tau}$ and vice. This is only possible if $-\sigma_p=-\hat{\sigma}_q$, $\gamma_p+\tilde{\gamma}-\tau=\hat{\gamma}_q+\tilde{\lambda}-\tau$ and $2(\tilde{\mu}\hat{\sigma}_q\hat{c}_q\chimap{\hat{\gamma}_q}-\tilde{\alpha}\sigma_pc_p\chimap{\gamma_p})=0$ for all $p\in\mathcal{P}^\neq$ and $q\in\mathcal{Q}^\neq$.

		If this condition is violated, that is, if both conditions (3a) and (3c) are violated, one has $\deg(u_2^{(1)})=d^{(1)}$, and it is straightforward to verify that all chain elements $u_2^{(\ell)}$ possess the strictly increasing degree $\deg(u_2^{(\ell)})=d^{(1)}+(\ell-1)(|\tilde{\gamma}|-2)$.
		Thus, the Lie algebra generated is infinite dimensional.
	
		To summarize, we have:
		\begin{align*}
			\left(\left(\left((1)\,\wedge\,(2)\right)\,\wedge\,\left(\neg (3b)\right)\right)\,\vee\,\left(\left((1)\,\wedge\,(2)\right)\,\wedge\,\left(\neg\left((3a)\,\wedge\,(3c)\right)\right)\right)\right)\implies \lie{\{e_1,e_2\}}\text{ is infinite dimensional.}
		\end{align*}
		This can compactly be written as
		\begin{align*}
			\left(\left((1)\,\wedge\,(2)\right)\,\wedge\,\left(\neg\left((3a)\,\wedge\,(3b)\,\wedge\,(3c)\right)\right)\right)\implies \left(\lie{\{e_1,e_2\}}\text{ is infinite dimensional.}\right),
		\end{align*}
		which is equivalent to 
		\begin{align*}
			\left(\left((1)\,\wedge\,(2)\,\wedge\,(3)\right)\implies \left(\lie{\{e_1,e_2\}}\text{ is infinite dimensional.}\right)\right)
		\end{align*}
		thereby concluding the proof.
	\end{proof}
	
	\begin{proposition}\label{prop:form:of:multiindex:with:same:degree}
		Let $d\in\N_{\geq0}$ and let $\Gamma_d=\{\gamma=(\alpha,\beta)\in\N_{\geq0}^2\,\mid\,|\gamma|=d\}$ be the set of every well-ordered multi-index $\gamma\in\N_{\geq0}^2$ that satisfies $|\gamma|=d$. Then, $\Gamma=\{(d-\ell,\ell)\,\mid\,\ell\in\N_{\geq0}\text{ such that }\ell\leq d/2\}$.   
	\end{proposition}
	
	\begin{proof}
		This claim follows immediately from the two conditions $\alpha\geq \beta$ and $d=\alpha+\beta$, as one only needs to identify $\ell=\beta$, since $\alpha=d-\beta$ and $\beta=2\beta/2\leq (\alpha+\beta)/2=d/2$.
	\end{proof}
	
	\begin{theorem}\label{thm:divergence:iadaga:perp}
		If $\{ia^\dagg a+ic, e\}\subseteq\mathcal{G}$ with $c\in\R$ and $\PP{e}{\hat{A}_1^\perp}\neq\{0\}$, then $\g:=\langle\mathcal{G}\rangle_{\mathrm{Lie}}$ is infinite dimensional.
	\end{theorem}
	
	We want to extend Lemma~\ref{lem:linear:two:combinations} to be applicable to linear combinations of arbitrary size, which have non-zero support in $\hat{A}_1^\perp$. Note that we chose without loss of generality $c=0$, since $i$ is in the center of $\hat{A}_1$ and commuting $ia^\dagg a +ic$ with any element $\tilde{e}\in\hat{A}_1$ yields the same result as simply commuting $ia^\dagg a$ with $\tilde{e}$.
	
	The idea of the proof below is as follows: In a first step, we show that any pair $\{e,i(a^\dagg a+c)\}\subseteq\hat{A}_1$ for which $\PP{e}{\hat{A}_1^\perp}\neq 0$ generates an infinite-dimensional Lie algebra if the reduced pair $\{\PP{e}{\hat{A}_1^1\oplus\hat{A}_1^2\oplus\hat{A}_1^\perp},ia^\dagg a\}$ generates an infinite-dimensional Lie algebra. This is shown using the double commutator $[ia^\dagg a,[e,i a^\dagg a]]$, as this preserves the non-trivial support in $\hat{A}_1^1\oplus\hat{A}_1^2\oplus\hat{A}_1^\perp$ but eliminates all components in $\hat{A}_1^0\oplus\hat{A}_1^=$ from $e$. Thus, we can assume that, without loss of generality, $e$ lies in $\hat{A}_1^1\oplus\hat{A}_1^2\oplus\hat{A}_1^\perp$.
	
	The second part of the proof investigates the nested commutator $[e,[ia^\dagg a,e]]$. The result is, to highest degree, decomposed into terms that lie in $\hat{A}_1^=$ and terms that lie in $\hat{A}_1^\perp$. Using structural properties of the monomials, as well as previous results, it is shown that either the nested commutators yield new elements of strictly increasing degree, allowing a recursive construction of infinitely many linearly independent elements in $\lie{\{ia^\dagg a,e\}}$, or the recursive process leads to a configuration covered by a previous result, which is then used to also conclude that the algebra is infinite dimensional.

	\begin{proof}
		Let $e$ be an element in $\mathcal{G}$ that satisfies $\PP{e}{\hat{A}_1^\perp}\neq 0$. This element can now be decomposed as $e=e^=+e^\neq$ where $e^=\in\hat{A}_1^0\oplus\hat{A}_1^=$ is the component of $e$ that belongs to $\hat{A}_1^0\oplus\hat{A}_1^=$, i.e.,$\PP{e}{\hat{A}_1^0\oplus\hat{A}_1^=}=e^=$, and $e^\neq$ the component of $e$ that belongs to $\hat{A}_1^1\oplus\hat{A}_1^2\oplus\hat{A}_1^\perp$, i.e.,$\PP{e}{\hat{A}_1^1\oplus\hat{A}_1^2\oplus\hat{A}_1^\perp}=e^\neq$. We can furthermore decompose $e^\neq$ uniquely into a linear combination of monomials that only possess multi-indices from $\N_{\geq0}^2\setminus\N_=^2$. Since commuting a monomial $g_\sigma^\gamma$ with $ia^\dagg a$ yields a monomial with the same multi-index but an exchange of $\sigma$ to $-\sigma$, we can write
		\begin{align*}
			e^\neq\upto{d_0}\sum_{p\in\mathcal{P}}\left(c_pg_+^{\gamma_p}+\hat{c}_pg_-^{\gamma_p}\right)=\sum_{p\in\mathcal{P}}e_p\;\text{ with }\; e_p:=c_pg_+^{\gamma_p}+\hat{c}_pg_-^{\gamma_p},
		\end{align*}
		where for every $p,q\in\mathcal{P}$ with $p\neq q$ one has $\gamma_p\neq\gamma_q$, $\gamma_p>\gamma_p^\dagg$, $(c_p,\hat{c}_p)\neq0$, and $|\gamma_p|=d_0=\deg(e^\neq)$. The condition $(c_p,\hat{c}_p)\neq 0$ reflects that at least one monomial of the form $g_\sigma^{\gamma_p}$ appears in the expansion of $e^\neq$. Thus, if for instance only $g_+^{\gamma_p}$ appears in the expansion but not $g_-^{\gamma_p}$, we have $\hat{c}_p=0$. Nevertheless, we added these terms for later convenience. Note that the size of $\mathcal{P}$ is bounded by $|\mathcal{P}|<\deg(e^\neq)/2$, i.e., by the degree of $e^\neq$, since any well-ordered multi-index $\gamma\in\Np{2}\setminus\N_=^2$ with $|\gamma|=\deg(e^\neq)$ must, analogous to Proposition~\ref{prop:form:of:multiindex:with:same:degree}, be of the form $(\deg(e^\neq)-j)\iota_1+j\iota_2$ with $\deg(e)-j> j$.
		
		We can now compute the following nested commutator:
		\begin{align*}
			[ia^\dagg a,[e,ia^\dagg a]]=[ia^\dagg a,[e^\neq ,ia^\dagg a]]\upto{d_0}\sum_{p\in\mathcal{P}}(\chimap{\gamma_p})^2\left(c_p g_+^{\gamma_p}+\hat{c}_p g_-^{\gamma_p}\right).
		\end{align*}
		In the first step, we used the fact that Proposition~\ref{prop:thm:25:equivalent} implies $[ia^\dagg a,e^=]=0$. Then, we used the bilinearity of the Lie bracket and the identity $[ia^\dagg a,g_\sigma^\gamma]=\sigma\chimap{\gamma}g_{-\sigma}^\gamma$. Since, for every $p\in \mathcal{P}$, one has $\gamma_p>\gamma_p^\dagg$ and consequently $\chimap{\gamma_p}\neq 0$. Thus $\PP{[ia^\dagg a,[e,ia^\dagg a]]}{\hat{A}_1^\perp}\neq 0$ if $\PP{e}{\hat{A}_1^\perp}\neq 0$, but to highest degree $[ia^\dagg a,[e,ia^\dagg a]]$ has no support in $\hat{A}_1^0\oplus\hat{A}_1^=$. It is straightforward to verify that $[ia^\dagg a,[e,ia^\dagg a]]$ has only non-trivial support in the space $\hat{A}_1^1\oplus\hat{A}_1^2\oplus\hat{A}_1^\perp$, by considering lower degree terms and applying the same logic.
		
		We are now interested in computing the Lie closure $\lie{\mathcal{G}}$, which is obtained by iteratively computing commutators, as formalized in Definition~\ref{def:Lie:closure}.
		Thus, if an element $c\in \hat{A}_1$ is obtained by the commutator $[a,b]=c$ of two elements $a,b\in\mathcal{G}$, one can always assume without loss of generality that $c$ does also belong to $\mathcal{G}$, as including it in the set of generators $\mathcal{G}$ does not change the result of computing the Lie closure $\lie{\mathcal{G}}$. Thus, when $\mathcal{G}$ contains $ia^\dagg a$ and an element $e$ that satisfies $\PP{e}{\hat{A}_1^\perp}\neq0$, we can assume without loss of generality, that $[ia^\dagg a,[e,ia^\dagg a]]\in\mathcal{G}$, where $\PP{[ia^\dagg a,[e,ia^\dagg a]]}{\hat{A}_1^0\oplus\hat{A}_1^=}=0$ and $\PP{[ia^\dagg a,[e,ia^\dagg a]]}{\hat{A}_1^\perp}\neq 0$. Therefore, we can also assume without loss of generality that an element $e\in\{ia^\dagg a, e\}\subseteq\mathcal{G}$ that satisfies $\PP{e}{\hat{A}_1^\perp}\neq 0$ does always belong to the space $\hat{A}_1^1\oplus\hat{A}_1^2\oplus\hat{A}_1^\perp$, as otherwise, one can simply replace it with the element $[ia^\dagg a,[e,ia^\dagg a]]$ that satisfies these conditions.
		
		We now proceed with computing $[e,[ia^\dagg a,e]]$ to highest degree, which will allow us to conclude that $\lie{\{e,ia^\dagg a\}}$ is infinite dimensional if $\PP{e}{\hat{A}_1^\perp}\neq 0$ using previous results. Let now $e\in\mathcal{G}$ be an element with non-trivial support in $\hat{A}_1^\perp$. By the discussion above, we can assume that 
		\begin{align*}
			e=e^\neq\upto{d_0}\sum_{p\in\mathcal{P}}\left(c_pg_+^{\gamma_p}+\hat{c}_pg_-^{\gamma_p}\right)=\sum_{p\in\mathcal{P}}e_p\;\text{ with }\; e_p:=c_pg_+^{\gamma_p}+\hat{c}_pg_-^{\gamma_p},
		\end{align*}
		where for every $p,q\in\mathcal{P}$ with $p\neq q$ one has $\gamma_p\neq\gamma_q$, $\gamma_p>\gamma_p^\dagg$, $(c_p,\hat{c}_p)\neq0$, and $|\gamma_p|=d_0=\deg(e)$.
		First, using Proposition~\ref{prop:thm:25:equivalent}, we calculate:
		\begin{align*}
			[ia^\dagg a,e]\upto{d_0}\sum_{p\in\mathcal{P}}\chimap{\gamma_p}\left(c_pg_-^{\gamma_p}-\hat{c}_pg_+^{\gamma_p}\right).
		\end{align*}
		Then, we wish to compute
		\begin{align*}
			[e,[ia^\dagg a,e]]&\upto{d^{(1)}}\sum_{p,q\in\mathcal{P}}\chimap{\gamma_q}[c_pg_+^{\gamma_p}+\hat{c}_pg_-^{\gamma_p},c_qg_-^{\gamma_q}-\hat{c}_qg_+^{\gamma_q}].
		\end{align*}
		To tackle this commutator, we split the sum in the expression above into two separate contributions: one for which $p=q$ and one for which $p\neq q$ holds. 
		
		We begin tackling the case $p=q$. Here, we employ Proposition~\ref{prop:e:iaa:e:commutator:equal:element}, which yields:
		\begin{align*}
			\sum_{p\in\mathcal{P}}[e_p,[ia^\dagg a,e_p]]&\upto{d^{(1)}}-\sum_{p\in\mathcal{P}}\left(c_p^2+\hat{c}_p^2\right)\chimap{\gamma_p}\chimap{\gamma_p\circ\gamma_p}g_+^{\gamma_p+\gamma_p^\dagg-\tau}
			=-\sum_{p\in\mathcal{P}}\left(c_p^2+\hat{c}_p^2\right)\chimap{\gamma_p}\chimap{\gamma_p\circ\gamma_p}g_+^{(d_0-1)\tau}.
		\end{align*}
		Next, we  tackle the case $p\neq q$, and find
		\begin{align*}
			\sum_{\underset{p\neq q}{p,q\in\mathcal{P}}}[e_p,[ia^\dagg a,e_q]]&=\sum_{\underset{p\neq q}{p,q\in\mathcal{P}}}\chimap{\gamma_q}\left(c_pc_q[g_+^{\gamma_p},g_-^{\gamma_q}]-\hat{c}_p\hat{c}_q[g_-^{\gamma_p},g_+^{\gamma_q}]-c_p\hat{c}_q[g_+^{\gamma_p},g_+^{\gamma_q}]+\hat{c}_pc_q[g_-^{\gamma_p},g_-^{\gamma_q}]\right).
		\end{align*}
		We define for convenience the terms $p_+(p,q):=c_pc_q[g_+^{\gamma_p},g_-^{\gamma_q}]-\hat{c}_p\hat{c}_q[g_-^{\gamma_p},g_+^{\gamma_q}]$ and $p_-(p,q):=-c_p\hat{c}_q[g_+^{\gamma_p},g_+^{\gamma_q}]+\hat{c}_pc_q[g_-^{\gamma_p},g_-^{\gamma_q}]$ and compute:
		\begin{align*}
			p_+(p,q)+p_+(q,p)&=c_pc_q[g_+^{\gamma_p},g_-^{\gamma_q}]+c_qc_p[g_+^{\gamma_q},g_-^{\gamma_p}]-\hat{c}_p\hat{c}_q[g_-^{\gamma_p},g_+^{\gamma_q}]-\hat{c}_q\hat{c}_p[g_-^{\gamma_q},g_+^{\gamma_p}]\\
			&=\left(c_pc_q+\hat{c}_p\hat{c}_q\right)\left([g_+^{\gamma_p},g_-^{\gamma_q}]-[g_-^{\gamma_p},g_+^{\gamma_q}]\right).
		\end{align*}
		Furthermore, Lemma~\ref{lem:commutator:form} guarantees that
		\begin{align*}
			[g_+^{\gamma_p},g_-^{\gamma_q}]-[g_-^{\gamma_p},g_+^{\gamma_q}]&\upto{d^{(1)}}\chimap{\gamma_p\circ\gamma_q^\dagg}g_+^{\gamma_p+\gamma_q^\dagg}-\chimap{\gamma_p\circ\gamma_q}g_+^{\Theta(\gamma_p+\gamma_q^\dagg)-\tau}-\chimap{\gamma_p\circ\gamma_q^\dagg}g_+^{\gamma_p+\gamma_q^\dagg-\tau}-\chimap{\gamma_p\circ\gamma_q}g_+^{\Theta(\gamma_p+\gamma_q^\dagg)-\tau}\\
			&=-2\chimap{\gamma_p\circ\gamma_q^\dagg}g_+^{\Theta(\gamma_p+\gamma_q^\dagg)-\tau}.
		\end{align*}
		We then  similarly compute
		\begin{align*}
			p_-(p,q)+p_-(q,p)&=-c_p\hat{c}_q[g_+^{\gamma_p},g_+^{\gamma_q}]-c_q\hat{c}_p[g_+^{\gamma_q},g_+^{\gamma_p}]+\hat{c}_pc_q[g_-^{\gamma_p},g_-^{\gamma_q}]+\hat{c}_qc_p[g_-^{\gamma_q},g_-^{\gamma_p}]\\
			&=\left(\hat{c}_pc_q-c_p\hat{c}_q\right)\left([g_-^{\gamma_p},g_-^{\gamma_q}]-[g_+^{\gamma_p},g_+^{\gamma_q}]\right).
		\end{align*}
		Once again, Lemma~\ref{lem:commutator:form} guarantees that
		\begin{align*}
			[g_-^{\gamma_p},g_-^{\gamma_q}]-[g_+^{\gamma_p},g_-^{\gamma_q}]&\upto{d^{(1)}}\chimap{\gamma_p\circ\gamma_q^\dagg}g_-^{\gamma_p+\gamma_q-\tau}-\epsilon_{pq}\chimap{\gamma_p\circ\gamma_q}g_-^{\Theta(\gamma_p+\gamma_q^\dagg)-\tau}\\
			&+\chimap{\gamma_p\circ\gamma_q^\dagg}g_-^{\gamma_p+\gamma_q-\tau}+\epsilon_{pq}\chimap{\gamma_p\circ\gamma_q}g_-^{\Theta(\gamma_p+\gamma_q^\dagg)-\tau}\\
			&=2\chimap{\gamma_p\circ\gamma_q^\dagg}g_-^{\gamma_p+\gamma_q-\tau}.
		\end{align*}
		Here, it is important to note that $\gamma_p+\gamma_q-\tau$ is always in $\Np{2}$ for $p\neq q$ since $\gamma_p\neq\gamma_q$, and therefore, there exists only one multi-index $\gamma_p$ such that $\gamma_p=d_0\iota_1$, while all others must have $\beta\neq0$. This concludes treating the case $p\neq q$.
		
		Hence, combing the results for the two cases $p=q$ and $p\neq q$, we find:
		\begin{align}\label{eqn:expansion:e:iaa:e:lin:comb:lem}
			[e,[ia^\dagg a,e]]&\upto{d^{(1)}}-\sum_{p\in\mathcal{P}}\left(c_p^2+\hat{c}_p^2\right)\chimap{\gamma_p}\chimap{\gamma_p\circ\gamma_p}g_+^{(d_0-1)\tau}\nonumber\\
			&\quad+2\sum_{\underset{p<q}{p,q\in\mathcal{P}}}\chimap{\gamma_q}\left(\left(\hat{c}_pc_q-c_p\hat{c}_q\right)\chimap{\gamma_p\circ\gamma_q^\dagg}g_-^{\gamma_p+\gamma_q-\tau}-\left(c_pc_q+\hat{c}_p\hat{c}_q\right)\chimap{\gamma_p\circ\gamma_q^\dagg}g_+^{\Theta(\gamma_p+\gamma_q^\dagg)-\tau}\right).
		\end{align}
		Proposition~\ref{prop:form:of:multiindex:with:same:degree} allows us to write $\gamma_p:=(d_0-{x_p})\iota_1+{x_p}\iota_2$, where $x_p$ are integers smaller than $(d_0+1)/2$. We can furthermore choose the ordering of the multi-indices such that $p<q$ implies $x_p<x_q$. This implies furthermore $\Theta(\gamma_p+\gamma_q^\dagg)=\gamma_p+\gamma_q^\dagg=(d_0-{(x_p-x_q)})\iota_1+(d_0+{(x_p-x_q)})\iota_2$, since $p<q$. In turn, this implies $\gamma_p+\gamma_q^\dagg-\tau>\gamma_p^\dagg-\gamma_q-\tau$, as equality is only possible if {$x_p=x_q$,} which is forbidden by the condition $p<q$.  It is immediate to verify $\gamma_p+\gamma_q-\tau>\gamma_p^\dagg+\gamma_q^\dagg-\tau$, since $\gamma_{p'}>\gamma_{p'}^\dagg$ for all $p'\in\mathcal{P}$.  One recognizes furthermore that $(c_p^2+\hat{c}_p^2)\chimap{\gamma_p}\chimap{\gamma_p\circ\gamma_p}>0$, since at least one of the coefficients $c_p$ or $\hat{c}_p$ does not vanish and $\gamma_p>\gamma_p^\dagg$ for all $p\in\mathcal{P}$. Thus, one has consequently $-\sum_p(c_p^2+\hat{c}_p^2)\chimap{\gamma_p}\chimap{\gamma_p\circ\gamma_p}\neq0$, with $g_+^{(d_0-1)-\tau}$ being a unique monomial in the expansion \eqref{eqn:expansion:e:iaa:e:lin:comb:lem}, since $g_+^{(d_0-1)\tau}\in\hat{A}_1^=$ and $g_-^{\gamma_p+\gamma_q-\tau},g_+^{\gamma_p+\gamma_q^\dagg-\tau}\in\hat{A}_n^\perp$, as $d_0\geq3$.
		
		We now want to conclude that $\g:=\lie{\mathcal{G}}$ is infinite dimensional. This can be shown in two ways. One option is to find a recursive construction for which the elements are of strictly increasing degree and all belong to $\lie{\mathcal{G}}$. This would imply that $\g$ is infinite dimensional, as all these elements would be linearly independent. The other possibility is to observe that the situation at hand is already covered by a previous result. To prove all of this claim we start by noting that if the second sum in the expression \eqref{eqn:expansion:e:iaa:e:lin:comb:lem} does not vanish one is left again with a situation similar to the initial one with $[ia^\dagg a,[e,[ia^\dagg a,e]]]$ satisfying the same conditions as $e$. That is $\PP{[ia^\dagg a,[e,[ia^\dagg a,e]]]}{\hat{A}_1^0\oplus\hat{A}_1^=}=0$ and $\PP{[ia^\dagg a,[e,[ia^\dagg a,e]]]}{\hat{A}_1^\perp}\neq 0$. The key difference between the two situations is that $\deg([ia^\dagg a,[e,[ia^\dagg a,e]]])>\deg(e)$. Thus, we can simply replace $e$ with $[ia^\dagg a,[e,[ia^\dagg a,e]]]$ and iterate this procedure. This will either yield a sequence of elements of strictly increasing degree, rendering $\g$ infinite dimensional, or one will eventually reach a step $\ell_*$ where the second sum in the expression \eqref{eqn:expansion:e:iaa:e:lin:comb:lem} vanishes identically. In such case, one is left in the situation covered by Lemma~\ref{lem:infinite:algebra:containg:support:perp:and:equal:elements} with the new elements $e$ and $[e,[ia^\dagg a,e]]$ for the step $\ell_*$. Lemma~\ref{lem:infinite:algebra:containg:support:perp:and:equal:elements} implies then that $\g$ is infinite dimensional.
	\end{proof}
	
	\begin{lemma}
		If $\{e_1^=,e_2\}\subseteq\mathcal{G}$ with $e_1^=\in\hat{A}_1^0\oplus\hat{A}_1^=$, $e_2\in\hat{A}_1^0\oplus\hat{A}_1^1\oplus\hat{A}_1^2$, $\PP{e_1^=}{\hat{A}_1^=}\neq0$, and $\PP{e_2}{\hat{A}_1^1\oplus\hat{A}_1^2}\neq0$, then $\g:=\langle\mathcal{G}\rangle_{\mathrm{Lie}}$ is infinite dimensional.
	\end{lemma}
	
	\begin{proof}
		Let $e_1^=$ and $e_2$ be defined as in the statement. Then, we can write $e_1^=\upto{d_1}cg_+^{\tilde{\gamma}}$ with $|\tilde{\gamma}|=\deg(e_1^=)=:d_1\geq 4$ and $c\neq0$, since $e_1^=$ has non-zero support in $\hat{A}_1^=$ due to the assumption $\PP{e_1^=}{\hat{A}_1^=}\neq0$ and there can be only one monomial $g_+^{\tilde{\gamma}}\in\hat{A}_1^=$ for each degree $d$. Furthermore, we can write
		\begin{align*}
			e_2=\sum_{p\in\mathcal{P}}c_pg_{\sigma_p}^{\gamma_p}=c_0g_+^{0} + c_1g_+^{\iota_1} + c_2g_-^{\iota_1} + c_3g_+^{\tau} + c_4g_+^{2\iota_1} + c_5g_-^{2\iota_1},
		\end{align*}
		where at least one of the coefficients $c_1,c_2,c_4,c_5$ does not vanish. We now note that $[e_1^=,e_2]=[e_1^=,e_2^{\neq}]$ with $e_2^{\neq}:=c_1g_+^{\iota_1} + c_2g_-^{\iota_1} + c_4g_+^{2\iota_1} + c_5g_-^{2\iota_1}$ due to Proposition~\ref{prop:thm:25:equivalent}, since the commutator with the remaining terms in $e_2$ with elements in $\hat{A}_1^0\oplus\hat{A}_1^=$ vanishes. The same proposition implies furthermore that $\deg([e_1^=,e_2^{\neq}])=|\tilde{\gamma}|+\deg(e_2^{\neq})-2\geq 3$ and $\PP{[e_1^=,e_2^{\neq}]}{\hat{A}_1^\perp}\neq\{0\}$, since $\deg(e_2^\neq)\geq 1$. Lemma~\ref{lem:infinite:algebra:containg:support:perp:and:equal:elements} implies then that $\g$ is infinite dimensional.
	\end{proof}

	\subsection{Core result}
	\begin{theorem}[Core result]\label{thm:core:result}
		Any finite set $\mathcal{G}\subseteq\hat{A}_1$ that contains the element $ia^\dagg a+ic$ generates a finite-dimensional Lie algebra $\g=\langle\mathcal{G}\rangle_{\mathrm{Lie}}$ if and only if the conditions (i) $\PP{\mathcal{G}}{\hat{A}_1^\perp}=\{0\}$, and (ii) $\PP{\mathcal{G}}{\hat{A}_1^1\oplus\hat{A}_1^2}=\{0\}$ when $\PP{\mathcal{G}}{\hat{A}_1^=}\neq\{0\}$ hold.
	\end{theorem}
	
	\begin{proof}
		Let us start by assuming the conditions (i) and (ii) hold. Then, one has one of the following two cases: (a) $\mathcal{G}\subseteq\hat{A}_1^0\oplus\hat{A}_1^1\oplus\hat{A}_1^2=\hat{A}_1^{\leq 2}$, or (b) every element $e\in\mathcal{G}$ is an element of the space $\hat{A}_1^0\oplus\hat{A}_1^=$. 
		Consider case (a) first: Here, one has $\g\subseteq\langle\hat{A}_1^{\leq 2}\rangle_{\mathrm{Lie}}\cong\slR{2}\ltimes\gh_1$, by Proposition~\ref{prop:schroedinger:algebra} which is a six-dimensional Lie algebra. In case (b), one has $[\mathcal{G},\mathcal{G}]=\{0\}$ by Proposition~\ref{prop:commutator:g0:g=:g0:g=}, yielding $\g=\spn\{\mathcal{G}\}$ which is finite dimensional since $\mathcal{G}$ is finite.
		
		We prove the necessary implication by showing the opposite case, i.e., violating either condition (i) or (ii) implies that $\g$ is an infinite-dimensional Lie algebra. If $\PP{\mathcal{G}}{\hat{A}_1^\perp}\neq\{0\}$, then there must exist an element $e\in\mathcal{G}$ with $\PP{e}{\hat{A}_1^\perp}\neq\{0\}$, since $\mathcal{G}$ is non-empty and $\PP{ia^\dagg a +c}{\hat{A}_1^\perp}=0$. Thus, Theorem~\ref{thm:divergence:iadaga:perp} implies that $\g$ is infinite dimensional. This leaves assuming the case that $\PP{\mathcal{G}}{\hat{A}_1^=}\neq\{0\}$ and $\PP{\mathcal{G}}{\hat{A}_1^1\oplus\hat{A}_1^2}\neq\{0\}$. We can disregard the case $\PP{\mathcal{G}}{\hat{A}_1^\perp}\neq\{0\}$ because it reduces to the previous one. This leaves us with two possibilities: either there exists at least one element $e\in \mathcal{G}$ satisfies simultaneously both $\PP{e}{\hat{A}_1^=}\neq0$ and $\PP{e}{\hat{A}_1^1\oplus\hat{A}_1^2}\neq0$ or there exist two distinct elements in $e_1,e_2\in\mathcal{G}$ such that one satisfies $\PP{e_1}{\hat{A}_1^=}\neq0$ while the other one satisfies $\PP{e_2}{\hat{A}_1^1\oplus\hat{A}_1^2}\neq0$ but not $\PP{e_2}{\hat{A}_1^=}\neq0$. Let us consider these separately: 
		\begin{enumerate}
			\item Suppose $\mathcal{G}$ contains the element $p_1$ that satisfies simultaneously $\PP{p_1}{\hat{A}_1^=}\neq0$, $\PP{p_1}{\hat{A}_1^1\oplus\hat{A}_1^2}\neq0$, and $\PP{p_1}{\hat{A}_1^\perp}=0$. This allows us to decompose $p_1$ as the sum $p_1=e_1^=+e_1^{\neq}$, where $e_1^=\in\hat{A}_1^0\oplus\hat{A}_1^=$, $e_1^{\neq}\in\hat{A}_1^1\oplus\hat{A}_1^2$, and $\deg(e_1^=)\geq 4$. Then, Proposition~\ref{prop:thm:25:equivalent} and Table~\ref{tab:Full:Commutator:Algebra:(pesudo):schroedinger:algebra} imply $[ia^\dagg a +ic,p_1]=[ia^\dagg a, e_1^{\neq}]\in\hat{A}_1^1\oplus\hat{A}_1^2$. This leaves us with the upcoming case that $\mathcal{G}$ contains two distinct elements $e_1,e_2\in\mathcal{G}$ such that $\PP{e_1}{\hat{A}_1^=}\neq0$, $\PP{e_2}{\hat{A}_1^1\oplus\hat{A}_1^2}\neq0$, $\PP{e_2}{\hat{A}_1^=}=0$, and $\PP{\{e_1,e_2\}}{\hat{A}_1^\perp}=\{0\}$, since we can assume, without loss of generality, that $[ia^\dagg a +ic,p_1]$ belongs to $\mathcal{G}$.
			\item Suppose $\mathcal{G}$ contains two distinct elements $p_1,p_2$ such that $\PP{p_1}{\hat{A}_1^=}\neq0$, $\PP{p_2}{\hat{A}_1^1\oplus\hat{A}_1^2}\neq0$, $\PP{p_2}{\hat{A}_1^=}=0$, and $\PP{\{p_1,p_2\}}{\hat{A}_1^\perp}=\{0\}$. Here, we can decompose $p_1$ as the sum $p_1=e_1^=+e_1^{\neq}$, where $e_1^=\in\hat{A}_1^0\oplus\hat{A}_1^=$, $e_1^{\neq}\in\hat{A}_1^1\oplus\hat{A}_1^2$, and $\deg(e_1^=)\geq 4$. The element $p_2$ furthermore belongs  to the space $\hat{A}_1^{\leq2}$ by the assumption $\PP{p_2}{\hat{A}_1^=\oplus\hat{A}_1^\perp}=0$. Lemma~\ref{lem:infinite:algebra:containg:support:perp:and:equal:elements} implies then that $\g$ is infinite dimensional.
		\end{enumerate}
		To summarize, we have now shown that violating condition (i) or condition (ii) always leads to an infinite-dimensional Lie algebra. Thus, in order for $\g$ to be finite-dimensional, both conditions must hold.
	\end{proof}

	\begin{corollary}\label{cor:all:physically:relevant:finite:algebras}
		All finite-dimensional non-abelian Lie algebras $\g$ that can be realized in the Weyl algebra $\hat{A}_1$ and contain a free term $ia^\dagg a +ic\in\g$ are subalgebras of the Schr\"odinger algebra $\mathcal{S}\cong\slR{2}\ltimes\gh_1$, and have at least dimension three. 
	\end{corollary}
	
	\begin{proof}
		Theorem~\ref{thm:core:result} implies that any finite-dimensional Lie algebra $\g\subseteq\hat{A}_1$ that contains a free term $ia^\dagg a +ic\in\g$ must itself be a subalgebra of either $\langle\hat{A}_1^0\oplus\hat{A}_1^=\rangle_{\mathrm{Lie}}$, or $\langle\hat{A}_1^0\oplus\hat{A}_1^1\oplus\hat{A}_1^2\rangle_{\mathrm{Lie}}=\lie{\hat{A}_1^{\leq 2}}$. Proposition~\ref{prop:commutator:g0:g=:g0:g=} implies that $\langle\hat{A}_1^0\oplus\hat{A}_1^=\rangle_{\mathrm{Lie}}$ is abelian. Therefore, any non-abelian algebra $\g$ must be a subalgebra of $ \langle\hat{A}_1^{\leq 2}\rangle_{\mathrm{Lie}}$, which is by Proposition~\ref{prop:schroedinger:algebra} isomorphic to the Schr\"odinger algebra $ \langle\hat{A}_1^{\leq2}\rangle_{\mathrm{Lie}}\cong\mathcal{S}\cong\slR{2}\ltimes\gh_1$.
		
		Next, we have to show that there exist no non-abelian subalgebras $\g$ of  $\slR{2}\ltimes\gh_1$ of dimension one or two that contain the free term $ia^\dagg a+ic$. We can disregard the case that $\g$ is one-dimensional, as it would read $\g=\lie{\{ia^\dagg a+ic\}}$ and then be trivially abelian. Now, we assume that there exists a non-abelian Lie algebra $\g$ of dimension two that contains the element $ia^\dagg a+ic$ and a second element $e\in \g$ such that $[ia^\dagg a+ic,e]=\kappa_1e+\kappa_2(i a^\dagg a+ic)\neq0$ for coefficients  $\kappa_1,\kappa_2\in\R$. Since $\g\subseteq\langle\hat{A}_1^{\leq2}\rangle_{\mathrm{Lie}}=\hat{A}_1^{\leq2}\cong\slR{2}\ltimes\gh_1$, we can express $e$ as a linear combination of the basis elements of $\hat{A}_1^{\leq 2}$ as follows:
		\begin{align*}
			e=\sum_{p\in\mathcal{P}}c_pg_{\sigma_p}^{\gamma_p}=c_0g_+^{0} + c_1g_+^{\iota_1} + c_2g_-^{\iota_1} + c_3g_+^{\tau} + c_4g_+^{2\iota_1} + c_5g_-^{2\iota_1},
		\end{align*}
		and compute, using Proposition~\ref{prop:thm:25:equivalent}, the following commutator 
		\begin{align*}
			[ia^\dagg a+ic,e]=c_1g_-^{\iota_1} - c_2g_+^{\iota_1} + 2c_4g_-^{2\iota_1} - c_5g_+^{2\iota_1}=\kappa_1e+\kappa_2(ia^\dagg a+ic).
		\end{align*}
		We match the linearly independent monomials with the expression $[ia^\dagg a+ic,e]=\kappa_1e+\kappa_2(i a^\dagg a+ic)$ obtained before and obtain the following system of equations:
		\begin{align*}
			c_1-\kappa_1c_2&=0,\;&\;-c_2-\kappa_1c_1&=0,\;&\;2c_4-\kappa_1c_5&=0,\;&\;-2c_5-\kappa_1c_4&=0.
		\end{align*}
		The only solution to this equation is $c_1=c_2=c_3=c_4=0$. This, however, implies $e\in\hat{A}_1^0$. Proposition~\ref{prop:thm:25:equivalent} implies then $[ia^\dagg a+ic,e]=0$, which is contradicts the assumption that $\g$ is a non-abelian two-dimensional Lie algebra. The smallest possible dimension for a non-abelian subalgebra containing a free term is therefore three, since non-abelian three-dimensional subalgebras of the Schr\"odinger algebra, such as the Heisenberg algebra, can be realized as subalgebras of $\hat{A}_1$, as shown in Proposition~\ref{prop:A11:algbera:heisenberg}.
	\end{proof}
	
	\begin{proposition}\label{prop:schroedinger:algebra:no:nice:properties}
		The Schr\"odinger algebra $\mathcal{S}\cong\langle\hat{A}_1^0\oplus\hat{A}_1^1\oplus\hat{A}_1^2\rangle_{\mathrm{Lie}}=\lie{\hat{A}_1^{\leq2}}$ is neither solvable, nilpotent, simple, semisimple, nor reductive.
	\end{proposition}
	
	\begin{proof}
		We know that $\mathcal{S}\cong\langle\hat{A}_1^{\leq 2}\rangle_{\mathrm{Lie}}$ by virtue of  Proposition~\ref{prop:schroedinger:algebra}. This result is well known in the literature \plainrefs{Yang:2018}, but we want to discuss this property in more detail using the results obtained here. 
		
		Let us consider the basis $\mathcal{B}=\{i,ia^\dagg a, i(a+a^\dagg),a-a^\dagg,i(a^2+(a^\dagg)^2),a^2-(a^\dagg)^2\}$ of $\mathcal{S}$. Using the commutation relations of the basis elements (see Table~\ref{tab:Full:Commutator:Algebra:(pesudo):schroedinger:algebra}), we observe $[\mathcal{S},\mathcal{S}]=\mathcal{S}$ implying that $\mathcal{S}$ is neither solvable nor nilpotent \plainrefs{Knapp:1996}. 
		The center $\mathcal{Z}(\mathcal{S})$ of $\slR{2}\ltimes\gh_1$ is clearly the one-dimensional vector space generated by $i$, which is isomorphic to $\mathbb{R}$. Thus, $\mathcal{S}$ is not semisimple and consequently not simple since semisimplicity is a prerequisite for simplicity \plainrefs{fuchs2003symmetries}.
		Finally, we see that the Heisenberg algebra $\gh_1\cong\spn\{i,i(a+a^\dagg),a-a^\dagg\}$ is nilpotent and hence solvable. The commutation relation of the basis elements yields $[\slR{2}\ltimes\gh_1,\gh_1]\subseteq\gh_1$ (see Table~\ref{tab:Full:Commutator:Algebra:(pesudo):schroedinger:algebra}), which implies that $\gh_1$ is a nontrivial non-abelian ideal. Therefore, the dimension of the radical $\operatorname{rad}(\slR{2}\ltimes\gh_1)$ is at least three. One concludes the it is never possible to have $\mathcal{Z}(\slR{2}\ltimes\gh_1)=\operatorname{rad}(\slR{2}\ltimes\gh_1)$ making $\slR{2}\ltimes\gh_1$ non-reductive \plainrefs{fuchs2003symmetries}.
	\end{proof}
	
	\begin{theorem}[Classification of Lie subalgebras containing the free Hamiltonian]\label{thm:all:physically:relevant:algebras:list}
		Let $\g$ be a non-abelian finite-dimensional subalgebra of the skew-hermitian Weyl algebra $\hat{A}_1$ that contains a free Hamiltonian term $i(a^\dagg a+c)\in\g$ for some $c\in\R$. Then $\g$ is isomorphic to one of the following four Lie algebras: (a) the Schr\"odinger algebra $\mathcal{S}\cong\slR{2}\ltimes\gh_1$, (b) the algebra $\slR{2}\oplus\R$, (c) the special linear algebra $\slR{2}$, or (d) the solvable Wigner-Heisenberg algebra $\wh_2$.
	\end{theorem}
	
	The idea of this proof is to consider the free Hamiltonian $i(a^\dagg a+c)$ and a second element $e\in\g$ and compute the Lie closure $\lie{\{e,i(a^\dagg a+c)\}}$ of these two elements. Then, we add a linearly independent third element and again compute the Lie closure of these three elements. Repeating this procedure and demanding finiteness and the fact that $\g$ is non-abelian will eventually yield all possible Lie subalgebras satisfying the desired constraints.
	
	\begin{proof}   
		Let $\g\subseteq\hat{A}_1$ be a non-abelian finite-dimensional Lie algebra that contains the free Hamiltonian term $i(a^\dagg a+c)$ for some arbitrary $c\in\R$. We want to start by deriving some constraints on an element $e\in\g$ that never coincides with a scalar multiple of a linear combination of the free Hamiltonian term $i(a^\dagg a+c)$ and the central element $i$. Suppose therefore that every $e\in\g$ commutes with $i(a^\dagg a+c)$. Then, by Proposition~\ref{prop:thm:25:equivalent}, $e\in \hat{A}_1^0\oplus\hat{A}_1^=$, and consequently $\g\subseteq\hat{A}_1^0\oplus\hat{A}_1^=$. This leads to a contradiction, as Proposition~\ref{prop:commutator:g0:g=:g0:g=} implies that $\g$ would then be abelian. Thus, there exists some element $e_1\in\g$ such that $[e_1,i(a^\dagg a+c)]\neq0$. By virtue of Theorem~\ref{thm:core:result} and Proposition~\ref{prop:schroedinger:algebra} we have $\g\subseteq\langle\hat{A}_1^0\oplus\hat{A}_1^1\oplus\hat{A}_1^2\rangle_{\mathrm{Lie}}\cong\mathcal{S}$. We can therefore write  $e_1=ic_1+ia^\dagg a c_2+c_3g_+^{\iota_1}+c_4g_-^{\iota_1}+c_5g_+^{2\iota_1}+c_6g_-^{2\iota_1}$, where $c_k\in\R$ for all $k$. Proposition~\ref{prop:thm:25:equivalent} states that at least one of the coefficients $c_3,c_4,c_5,c_6$ must be non-zero, as otherwise $e_1$ would commute with $i(a^\dagg a+c)$. 
		
		Next, we want check whether $\spn\{i(a^\dagg a,+c),e_1\}$ is closed under commutation. Here, Proposition~\ref{prop:thm:25:equivalent} implies that the result of the commutator of $i(a^\dagg a+c)$ and $e_1$, denoted by $e_2$ is given by $e_2:=[i(a^\dagg a+c),e_1]=c_3g_-^{\iota_1}-c_4g_+^{\iota_1}+2c_5g_-^{2\iota_1}-2c_6g_+^{2\iota_1}$. Suppose $e_2=i(a^\dagg a+c)\lambda_1+\lambda_2e_1$ for some scalars $\lambda_1,\lambda_2\in\R$, i.e., suppose $\spn\{i(a^\dagg a,+c),e_1\}=\lie{\{i(a^\dagg a,+c),e_1\}}$. Comparing the prefactors of the basis elements $i$, $ia^\dagg a$, $g_+^{\iota_1}$, $g_-^{\iota_1}$, $g_+^{2\iota_1}$, and $g_-^{2\iota_1}$ would imply the following system of equations:
		\begin{align*}
			\lambda_1c+\lambda_2 c_1&=0,\;&\;\lambda_1 +\lambda_2 c_2&=0,\;&\;\lambda_2 c_3&=-c_4,\;&\;\lambda_2c_4&=c_3,\;&\lambda_2 c_5&=-2c_6,\;&\;\lambda_2 c_6=2c_5.
		\end{align*}
		Simple manipulations yield:
		\begin{align*}
			\lambda_2(c_1-cc_2)&=0,\;&\;(1+\lambda_2^2)c_3&=0,\;&\;(1+\lambda_2^2)c_4&=0,\;&\;(4+\lambda_2^2)c_5&=0,\;&\;(4+\lambda_2^2)c_6&=0.
		\end{align*}
		Since $\lambda_2^2\geq 0$, the last four equations can only be satisfied if $c_3=c_4=c_5=c_6=0$, which contradicts the earlier assumption that at least one of these coefficients is non-zero. Hence, $\spn\{i(a^\dagg a+c),e_1,e_2\}$ is three-dimensional and $\lie{\{i(a^\dagg a+c),e_1\}}$ is at least three-dimensional. 
		
		Now, we want to investigate whether $\lie{\{i(a^\dagg a+c),e_1\}}$ can be three dimensional and if so which conditions need to be met. For this, we compute $e_3:=[e_2,i(a^\dagg a+c)]=c_3g_+^{\iota_1}+c_4g_-^{\iota_1}+4c_5g_+^{2\iota_1}+4c_6g_-^{2\iota_1}$ and suppose $e_3\in\spn\{i(a^\dagg a+c),e_1,e_2\}$. This would imply $e_3=i(a^\dagg a+c)\lambda_1+\lambda_2e_1+\lambda_3e_2$ for some real coefficients $\lambda_1,\lambda_2,\lambda_3\in\R$. Comparing the coefficients of the basis elements $ia^\dagg a$ and $i$, we obtain the condition \textbf{(C0)} that $\boldsymbol{\lambda_1c+\lambda_2 c_1=0}$, and $\boldsymbol{\lambda_1 +\lambda_2 c_2=0}$ holds. This is only possible if either: \textbf{(C1)} $\boldsymbol{c_1=cc_2}$, or \textbf{(C2)} $\boldsymbol{\lambda_1=\lambda_2=0}$. If now condition (C2) holds, one would be left with $e_3=\lambda_3 e_2$, and therefore
		\begin{align*}
			c_3(1+\lambda_3^2)&=0,\;&\;c_4(1+\lambda_3^2)&=0,\;&\;c_5(4+\lambda_3^2)&=0,\;&\;c_6(4+\lambda_3^2)&=0.
		\end{align*}
		This is not possible, for the same reason that $e_2=\lambda_2 e_1$ was previously ruled out: it would require all coefficients $c_3,c_4,c_5,c_6$ to vanish, contradicting the assumption that at least one of them is non-zero. 
		
		We are left considering the case (C1), i.e., $c_1=cc_2$. Here, by condition (C0), one must have $\lambda_1=-\lambda_2 c_2$, and the following relations hold:
		\begin{align*}
			c_3&=\lambda_2c_3-\lambda_3c_4,\;&\;c_4&=\lambda_2 c_4+\lambda_3 c_3,\;&\;4c_5&=\lambda_2c_5-2\lambda_3c_6,\;&\; 4c_6&=\lambda_2c_6+2\lambda_3c_5.
		\end{align*}
		We can rewrite these equations as follows:
		\begin{align}
			(1-\lambda_2)(c_3+c_4)&=\lambda_3(c_3-c_4),\;&\;(1-\lambda_2)(c_3-c_4)&=-\lambda_3(c_3+c_4),\nonumber\\
			(4-\lambda_2)(c_5+c_6)&=2\lambda_3(c_5-c_6),\;&\;(4-\lambda_2)(c_5-c_6)&=-2\lambda_3(c_5+c_6).\label{eqn:help:algebraic:conditions:C1:1}
		\end{align}
		To derive constraints on the coefficients $c_j$, we start supposing $\boldsymbol{(c_3,c_4)\neq0}$ and $\boldsymbol{(c_5,c_6)\neq 0}$ and denote this condition by \textbf{(C1a)}. 
		
		If we assume $\lambda_2=1$, then the first pair of equations in \eqref{eqn:help:algebraic:conditions:C1:1} read $\lambda_3(c_3-c_4)=0=-\lambda_3(c_3+c_4)$. This can only be satisfied if either $\lambda_3=0$ or $c_3=c_4=0$. The latter is prohibited by conditions (C1a), which we are currently considering, and it follows that $\lambda_3=0$. The last two equations in \eqref{eqn:help:algebraic:conditions:C1:1} read then $3(c_5-c_6)=0=3(c_5-c_6)$, which is only possible if $c_5=0=c_6$. However, this is again prohibited by condition (C1a) and therefore contradicts the assumption that $e_3\in\spn\{i(a^\dagg a+c),e_1,e_2\}$. Thus, we may safely assume $\lambda_2\neq1$ and rewrite the first two equations in \eqref{eqn:help:algebraic:conditions:C1:1} as: $((1-\lambda_2)^2+\lambda_3^2)(c_3-c_4)=0=((1-\lambda_2)^2+\lambda_3^2)(c_3+c_4)$. This can only be satisfied if $c_3=c_4=0$, which contradicts the assumption (C1a). Hence, the space $\spn\{i(a^\dagg a+c),e_1,e_2,e_3\}$ must be at least four-dimensional when $(c_3,c_4)\neq0$ and $(c_5,c_6)\neq 0$.
		
		We consider now the two remaining cases where $e_3\in\spn\{i(a^\dagg a+c),e_1,e_2\}$, namely the case \textbf{(C1b)} given by the constraints $\boldsymbol{c_1=cc_2}$ and $\boldsymbol{c_5=c_6=0}$, and the case \textbf{(C1c)} given by the constraints $\boldsymbol{c_1=cc_2}$ and $\boldsymbol{c_3=c_4=0}$. Let us consider them separately:
		
		\noindent For case (C1b), the constraints are $\boldsymbol{c_1=cc_2}$ and $\boldsymbol{c_5=c_6=0}$. Here, one has $e_3\in\spn\{i(a^\dagg a+c),e_1,e_2\}$. This can be verified by substituting $c_5=0=c_6$ and $\lambda_2=1$ into equation \eqref{eqn:help:algebraic:conditions:C1:1} and recalling that condition (C0) demands $\lambda_1=-\lambda_2c_2$. Combined, one has $e_3=-i(a^\dagg a+c)c_2+e_1$. We aim to check whether $\spn\{i(a^\dagg a+c),e_1,e_2\}$ is closed under the commutator. So far, we have seen that, under the given constraints, $e_2=[i(a^\dagg a+c),e_1]$ and $e_3=[i(a^\dagg a+c),e_2]$ both belong to $\spn\{i(a^\dagg a+c),e_1,e_2\}$. By the antisymmetry and the bilinearity of the commutator, we are left computing $[e_2,e_1]$. Here, one has $e_4:=-[e_1,e_2]=[e_1,[e_1,i(a^\dagg a+c)]]=[c_1i+c_2ia^\dagg a+c_3g_+^{\iota_1}+c_4g_-^{\iota_1},c_4g_+^{\iota_1}-c_3g_-^{\iota_1}]\in\g$, using Table~\ref{tab:Full:Commutator:Algebra:(pesudo):schroedinger:algebra}. The result is $e_4=c_2[ia^\dagg a,c_4g_+^{\iota_1}-c_3g_-^{\iota_1}]+(c_3^2+c_4^2)[g_-^{\iota_1},g_+^{\iota_1}]=c_2(c_3g_+^{\iota_1}+c_4g_-^{\iota_1})+ 2i(c_3^2+c_4^2)$, which is non-zero as, by the initial discussion on non-abelian elements, we need to assume that at least one of the coefficients $c_3,c_4$ does not vanish. This is clearly a linearly independent element from $\spn\{i(a^\dagg a+c),e_1,e_2\}$, which can be verified by identifying $(1,0,0,0)^\mathrm{Tp}\leftrightarrow i$, $(0,1,0,0)^\mathrm{Tp}\leftrightarrow ia^\dagg a$, $(0,0,1,0)^\mathrm{Tp}\leftrightarrow g_+^{\iota_1}$, and $(0,0,0,1)^\mathrm{Tp}\leftrightarrow g_-^{\iota_1}$. Then, solving the equation for linear independence $c_2(c_3g_+^{\iota_1}+c_4g_-^{\iota_1})+2i(c_3^2+c_4^2)=\lambda_1 i(a^\dagg a+c)+\lambda_2 e_1+\lambda_3 e_2$ with $\lambda_1,\lambda_2,\lambda_3\in\R$ is equivalent to verifying that the determinant of the matrix
		\begin{align*}
			\boldsymbol{M}_{\mathrm{C1b}}:=\begin{pmatrix}
				2(c_3^2+c_4^2) & c & c_1 & 0\\
				0 & 1 & c_2 & 0 \\
				c_2c_3 & 0 & c_3 & -c_4\\
				c_2c_4 & 0 & c_4 & c_3
			\end{pmatrix}
		\end{align*}
		is non-zero, since $e_2=[ia^\dagg a,e_1]=c_3g_+^{\iota_1}-c_4 g_+^{\iota_1}$. It is immediate to compute $\operatorname{det}(\boldsymbol{M}_{\mathrm{C1b}})=(c_3^2+c_4^2)(2c_3^2+2c_4^2+c_2(cc_2-c_1))$. Now recall that we are currently in case (C1) which demands $c_1=cc_2$, implying that $\operatorname{det}(\boldsymbol{M}_{\mathrm{C1b}})=2(c_3^2+c_4^2)^2$ and verifying that $e_4\notin\spn\{i(a^\dagg a+c),e_1,e_2\}$. It is furthermore easy to verify that $\langle\{i(a^\dagg a+c),e_1\}\rangle_{\mathrm{Lie}}=\langle\{i,ia^\dagg a,g_+^{\iota_1},g_-^{\iota_1}\}\rangle_{\mathrm{Lie}}$. 
		
		\noindent For case (C1c), the constraints are $\boldsymbol{c_1=cc_2}$ and $\boldsymbol{c_3=c_4=0}$. Here, one has $e_3\in\spn\{i(a^\dagg a+c),e_1,e_2\}$. This can be verified by substituting $c_3=0=c_4$ and $\lambda_2=4$ into equation \eqref{eqn:help:algebraic:conditions:C1:1} and recalling that condition (C0) demands $\lambda_1=-\lambda_2c_2$. Combined, one has $e_3=-4i(a^\dagg a+c)c_2+4e_1$. Similar to the previous case, we aim to check whether $\spn\{i(a^\dagg a+c),e_1,e_2\}$ is closed under the commutator, and are consequently left computing $[e_2,e_1]$. Here, one has $e_4:=-[e_1,e_2]/2=[e_1,[e_1,i(a^\dagg a+c)]/2]=[c_1i+c_2ia^\dagg a+c_5g_+^{2\iota_1}+c_6g_-^{2\iota_1},c_6g_+^{2\iota_1}-c_5g_-^{2\iota_1}]\in\g$, using Table~\ref{tab:Full:Commutator:Algebra:(pesudo):schroedinger:algebra}. The result is $e_4=c_2[ia^\dagg a,c_6g_+^{2\iota_1}-c_5g_-^{2\iota_1}]+(c_5^2+c_6^2)[g_-^{2\iota_1},g_+^{2\iota_1}]=2c_2(c_5g_+^{2\iota_1}+c_6g_-^{2\iota_1})+8i(c_5^2+c_6^2)(a^\dagg a+1/2)\neq0$. To check for linear independence, we follow the same strategy as the preceding case. That is, we identify $(1,0,0,0)^{\mathrm{Tp}}\leftrightarrow i$,  $(0,1,0,0)^\mathrm{Tp}\leftrightarrow ia^\dagg a$, $(0,0,1,0)^\mathrm{Tp}\leftrightarrow g_+^{2\iota_1}$, and $(0,0,0,1)^\mathrm{Tp}\leftrightarrow g_-^{2\iota_1}$ and compute the determinant of the matrix
		\begin{align*}
			\boldsymbol{M}_{\mathrm{C1c}}:=\begin{pmatrix}
				4(c_5^2+c_6^2)&c&cc_2&0\\
				8(c_5^2+c_6^2)&1&c_2&0\\
				2c_2c_5&0&c_5&c_6\\
				2c_2c_6&0&c_6&-c_5
			\end{pmatrix}.
		\end{align*}
		Note that we substituted here already the condition $c_1=cc_2$ given by the constraint (C1). One finds: $\operatorname{det}(\boldsymbol{M}_\mathrm{C1c})=4(c_5^2+c_6^2)^2(2c-1)$ and we can conclude that $e_4\in \spn\{i(a^\dagg a+c),e_1,e_2\}$ if $c=1/2$ and $e_4\notin\spn\{i(a^\dagg a+c),e_1,e_2\}$ otherwise. It is easy to verify that $\langle\{i(a^\dagg a+c),e_1\}\rangle_{\mathrm{Lie}}=\langle\{i(a^\dagg a+1/2),g_+^{2\iota_1},g_-^{2\iota_1}\}\rangle_{\mathrm{Lie}}$ if $c=1/2$ and $\langle\{i(a^\dagg a+c),e_1\}\rangle_{\mathrm{Lie}}=\langle\{i,ia^\dagg a,g_+^{2\iota_1},g_-^{2\iota_1}\}\rangle_{\mathrm{Lie}}$ otherwise.
		
		This concludes treating the case (C1). To summarize, we found that the only three-dimensional Lie subalgebra that can be generated by the free Hamiltonian term and an element $e_1=c_1i+c_2ia^\dagg a+c_3g_+^{\iota_1}+c_4g_-^{\iota_1}+c_5g_+^{2\iota_1}+c_6g_-^{2\iota_1}$ that does not commute with $i(a^\dagg a+c)$ is given by the case $c=1/2$, $c_1=c_2/2$ and $c_3=c_4=0$. All other ones are at least four-dimensional. We can now proceed with the remaining case \textbf{(C3)}, where $e_3\notin\spn\{i(a^\dagg a+c),e_1,e_2\}$; i.e., either $c_1\neq cc_2$ or $(c_3,c_4)\neq0$ and $(c_5,c_6)\neq 0$ simultaneously. We split this discussion into separate subcases and investigate whether $\spn\{i(a^\dagg a+c),e_1,e_2,e_3\}$ equals the Lie closure of the set $\{i(a^\dagg a+c),e_1\}$:
		\begin{enumerate}
			\item \textbf{Case (C3a)}: Here, at least one of the coefficients from each of the pairs $(c_3,c_4)$ and $(c_5,c_6)$ shall not vanish. Moreover, we assume that $\boldsymbol{c_1=cc_2}$. We now want to investigate whether $\spn\{i(a^\dagg a+c),e_1,e_2,e_3\}$ is closed under the commutator and construct therefore the two elements $e_4:=(i(a^\dagg a+c)c_2+e_3-e_1)/3=c_5g_+^{2\iota_1}+c_6 g_-^{2\iota_1}$, and $e_5:=e_1-i(a^\dagg a+c)c_2-e_4=c_3g_+^{\iota_1}+c_4g_-^{\iota_1}$ that clearly satisfy the condition $e_4,e_5\in\spn\{i(a^\dagg a+c),e_1,e_2,e_3\}$. One computes now: $[e_4,[e_4,i(a^\dagg a+c)]]=2(c_5^2+c_6^2)[g_-^{2\iota_1},g_+^{2\iota_1}]=16i(c_5^2+c_6^2)(a^\dagg a+1/2)\neq0$, and $[e_5,[e_5,i(a^\dagg a+c)]]=(c_3^2+c_4^2)[g_-^{\iota_1},g_+^{\iota_1}]=2i(c_3^2+c_4^2)\neq0$. Thus, both basis elements $i$ and $ia^\dagg a$ lie in $\langle\{i(a^\dagg a+c),e_1\}\rangle_{\mathrm{Lie}}$. Furthermore, it is straightforward to confirm that $\spn\{e_4,e_5,[ia^\dagg a,e_4],[ia^\dagg a,e_5]\}=\spn\{g_+^{\iota_1},g_-^{\iota_1},g_+^{2\iota_1},g_-^{2\iota_1}\}$. Consequently, $\langle\{i(a^\dagg a+c),e_1\}\rangle_{\mathrm{Lie}}=\langle\{i,ia^\dagg a,g_+^{\iota_1},g_-^{\iota_1},g_+^{2\iota_1},g_-^{2\iota_1}\}\rangle_{\mathrm{Lie}}\cong\mathcal{S}\cong\slR{2}\ltimes\gh_1$. 
			\item \textbf{Case (C3b)}: Here, the constraints shall be $\boldsymbol{c_1\neq c c_2}$, $\boldsymbol{c_3=c_4=0}$, and $\boldsymbol{(c_5,c_6)\neq0}$. In this case, we first note that at least one of the coefficients $c_1$ or $c_2$ must be non-zero; otherwise the condition $c_1\neq c c_2$ would be violated. Suppose exactly one of the coefficients $c_1$ or $c_2$ vanishes. Then, similar to case (C1c), it is straightforward to verify that $\spn\{i(a^\dagg a+c),e_1,e_2,e_3\}=\spn\{i,ia^\dagg a,g_+^{2\iota_1},g_-^{2\iota_1}\}$ and $\langle\{i(a^\dagg a+c),e_1\}\rangle_{\mathrm{Lie}}=\langle \{i,ia^\dagg a,g_+^{2\iota_1},g_-^{2\iota_1}\}\rangle_{\mathrm{Lie}}=\spn\{i,ia^\dagg a,g_+^{2\iota_1},g_-^{2\iota_1}\}$. Now suppose $c_1\neq0$ and $c_2\neq0$. We compute: $[e_2,e_3]=8(c_5^2+c_6^2)[g_-^{2\iota_1},g_+^{2\iota_1}]=64i(c_5^2+c_6^2)(a^\dagg a+1/2)\neq0$. If $c\neq 1/2$, then $i(a^\dagg a+1/2)$ and $i(a^\dagg a+c)$ are clearly independent, and we conclude $\langle\{i(a^\dagg a+c),e_1\}\rangle_{\mathrm{Lie}}=\langle \{i,ia^\dagg a,g_+^{2\iota_1},g_-^{2\iota_1}\}\rangle_{\mathrm{Lie}}$. If $c=1/2$, then $i(a^\dagg a+c)=i(a^\dagg a+1/2)$, and one observes $i(a^\dagg a+c)-(e_1-e_3/4)/c_2=i(cc_2-c_1)/c_2$. This does never vanish by the assumption $c_1\neq c c_2$ implying again $\langle\{i(a^\dagg a+c),e_1\}\rangle_{\mathrm{Lie}}=\langle \{i,ia^\dagg a,g_+^{2\iota_1},g_-^{2\iota_1}\}\rangle_{\mathrm{Lie}}=\spn\{i,ia^\dagg a,g_+^{2\iota_1},g_-^{2\iota_1}\}$.
			\item \textbf{Case (C3c)}: Here, the constraints shall be $\boldsymbol{c_1\neq c c_2$, $c_5=c_6=0}$, and $\boldsymbol{(c_3,c_4)\neq0}$. In this case, we can simply compute $[e_2,e_3]=(c_3^2+c_4^2)[g_-^{\iota_1},g_+^{\iota_1}]=2i(c_3^2+c_4^2)\neq0$. Hence, its is easy to verify that $\langle\{i(a^\dagg a+c),e_1\}\rangle_{\mathrm{Lie}}=\langle \{i,ia^\dagg a,g_+^{\iota_1},g_-^{\iota_1}\}\rangle_{\mathrm{Lie}}=\spn\{i,ia^\dagg a,g_+^{\iota_1},g_-^{\iota_1}\}$ holds.
			\item \textbf{Case (C3d)}: Here, the constraints shall be $\boldsymbol{c_1\neq c c_2}$, and $\boldsymbol{(c_3,c_4)\neq0}$, and $\boldsymbol{(c_5,c_6)\neq0}$. In this case, we can start, by introducing the element $h:=[[i(a^\dagg a+c),e_3],i(a^\dagg a+c)]$, as well as the four elements $e_4:=(h-e_3)/12=c_5g_+^{2\iota_1}+c_6 g_-^{2\iota_1}$, $e_5:=[i(a^\dagg a+c),e_4/2]=c_5 g_-^{2\iota_1}-c_6 g_+^{2\iota_1}$, $e_6:=(4e_3-h)/3=c_3g_+^{\iota_1}+c_4g_-^{\iota_1}$, and $e_7:=[i(a^\dagg a+c),e_6]=c_3g_-^{\iota_1}-c_4g_+^{\iota_1}$. Clearly $e_4,e_5,e_6,e_7,h\in\langle\{i(a^\dagg a+c),e_1\}\rangle_{\mathrm{Lie}}$ by construction. It is straightforward to verify that $\spn\{e_4,e_5,e_6,e_7\}=\spn\{g_+^{\iota_1},g_-^{\iota_1},g_+^{2\iota_1},g_-^{2\iota_1}\}$, by computing $c_5e_4-c_6e_5=(c_5^2+c_6^2)g_+^{2\iota_1}$, $c_6e_4+c_5e_5=(c_5^2+c_6^2)g_-^{2\iota_1}$, $c_3e_6-c_4e_7=(c_3^2+c_4^2)g_+^{\iota_1}$, and $c_4e_6+c_3e_7=(c_3^2+c_4^2)g_-^{\iota_1}$. Moreover, using the commutation relations $[g_-^{\iota_1},g_+^{\iota_1}]=2i$ and $[g_-^{2\iota_1},g_+^{2\iota_1}]=8i(a^\dagg a+1/2)$ we conclude that both $i$ and $ia^\dagg a$ lie in $\lie{\{i(a^\dagg a+c),e_1\}}$. Therefore $\langle\{i(a^\dagg a+c),e_1\}\rangle_{\mathrm{Lie}}=\langle\{i,ia^\dagg a,g_+^{\iota_1},g_-^{\iota_1},g_+^{2\iota_1},g_-^{2\iota_1}\}\rangle_{\mathrm{Lie}}\cong\mathcal{S}\cong\slR{2}\ltimes\gh_1$. 
		\end{enumerate}
		It is immediate to verify that cases (C3a) - (C3d) exhaust all possibilities to satisfy condition (C3). That is
		\begin{align*}
			\left(\mathrm{(C3a)}\,\vee\mathrm{(C3b)}\,\vee\mathrm{(C3c)}\,\vee\mathrm{(C3d)}\right)\Leftrightarrow\left(c_1=cc_2\,\vee\,\left((c_3,4_4)\neq 0\,\wedge\,(c_5,c_6)\neq 0\right)\right).
		\end{align*}
		We have shown that the Lie algebra generated by $i(a^\dagg a+c)$ and an element $e_1\in\langle\hat{A}_1^0\oplus\hat{A}_1^1\oplus\hat{A}_1^2\rangle_{\mathrm{Lie}}$ that does not commute with $ia^\dagg a$ is one of the following four Lie algebras: (i) $\langle\{i,ia^\dagg a,g_+^{\iota_1},g_-^{\iota_1}\}\rangle_{\mathrm{Lie}}$ (cf. Cases (C1b) and (C3c)), (ii) $\langle\{i(a^\dagg a+1/2),g_+^{2\iota_1},g_-^{2\iota_1}\}\rangle_{\mathrm{Lie}}$ (cf. Case (C1c)), (iii) $\langle\{i,ia^\dagg a,g_+^{2\iota_1},g_-^{2\iota_1}\}\rangle_{\mathrm{Lie}}$ (cf. Cases (C1c) and (C3b)), or (iv) $\langle\{i,ia^\dagg a,g_+^{\iota_1},g_-^{\iota_1},g_+^{2\iota_1},g_-^{2\iota_1}\}\rangle_{\mathrm{Lie}}$ (cf. cases (C3a) and (C3d)). We do now have to consider the cases that there exists an element $\tilde{e}_1\in \g$ that does not belong to $\langle\{i(a^\dagg a+c),e_1\}\rangle_{\mathrm{Lie}}$. This can only occur in the cases (i) - (iii). Consider these separately and investigate the Lie closure of the corresponding sets $\{i(a^\dagg a+c),e_1,\tilde{e}_1\}$:
		\begin{enumerate}[label = (\roman*)]
			\item $\boldsymbol{\lie{\{i(a^\dagg a+c),e_1\}}=\langle\{i,ia^\dagg a,g_+^{\iota_1},g_-^{\iota_1}\}\rangle_{\mathrm{Lie}}}$: We may assume without loss of generality that $\tilde{e}_1=\tilde{c}_5 g_+^{2\iota_1}+\tilde{c}_6g_-^{2\iota_1}$, with $\tilde{c}_5\neq 0$ or $\tilde{c}_6\neq 0$, since $\tilde{e}_1$ must have non-vanishing support in $\hat{A}_1^2$ and any component in $\spn\{i,ia^\dagg a,g_+^{\iota_1},g_-^{\iota_1}\}$ can be subtracted off. Let us consider the case $\tilde{c}_5\neq0$; the case $\tilde{c}_6\neq0$ can be treated analogously: Define $\tilde{e}_2:=2\tilde{c}_5\tilde{e}_1-\tilde{c}_6[ia^\dagg a,\tilde{e}_1]=2(\tilde{c}_5^2+\tilde{c}_6^2) g_+^{2\iota_1}$. Thus, $g_+^{2\iota_1}\in\g$ and since $g_-^{2\iota_1}=[ia^\dagg a,g_+^{2\iota_1}]/2$, we conclude that $g_-^{2\iota_1}\in\g$ as well. Therefore, $\g=\langle\{i,ia^\dagg a,g_+^{\iota_1},g_-^{\iota_1},g_+^{2\iota_1},g_-^{2\iota_1}\}\rangle_{\mathrm{Lie}}\cong\mathcal{S}\cong \slR{2} 
			\ltimes\gh_1$.
			\item $\boldsymbol{\lie{\{i(a^\dagg a+c),e_1\}}=\langle\{i(a^\dagg a+1/2),g_+^{2\iota_1},g_-^{2\iota_1}\}\rangle_{\mathrm{Lie}}}$: We may assume without loss of generality that $\tilde{e}_1=i\tilde{c}_1+ia^\dagg a\tilde{c}_2+\tilde{c}_3 g_+^{\iota_1}+\tilde{c}_4g_-^{\iota_1}$, where at least one of the coefficients $\tilde{c}_1,\tilde{c}_2,\tilde{c}_3,\tilde{c}_4\in\R$ does not vanish. Furthermore, if $\tilde{c}_3=\tilde{c}_4=0$, then $2\tilde{c}_1\neq \tilde{c}_2$, and it is easy to verify that in this case $\langle\{i(a^\dagg a+c),e_1,\tilde{e}_1\}\rangle_{\mathrm{Lie}}=\langle\{i,ia^\dagg a, g_+^{2\iota_1},g_-^{2\iota_1}\}\rangle_{\mathrm{Lie}}$. If there exists now a third element $\tilde{e}_2\in \g$ such that $\tilde{e}_2\notin \langle\{i,ia^\dagg a, g_+^{2\iota_1},g_-^{2\iota_1}\}\rangle_{\mathrm{Lie}}$, one needs to consider the upcoming case (iii). 
			
			Now suppose that at least one of the coefficients $\tilde{c}_3$ or $\tilde{c}_4$ does not vanish. Without loss of generality, assume $\tilde{c}_3\neq 0$; the case $\tilde{c}_4\neq 0$ can be treated analogously. One realizes that $(\tilde{c}_3^2+\tilde{c}_4^2)g_+^{\iota_1}=\tilde{c}_3[[i(a^\dagg a+1/2),\tilde{e}_1],i(a^\dagg a+1/2)]-\tilde{c}_4[i(a^\dagg a+1/2),\tilde{e}_1]\neq 0$. Thus, $g_+^{\iota_1}\in\g$, and since $g_-^{\iota_1}=[i(a^\dagg a+1/2),g_+^{\iota_1}]$, we conclude $g_-^{\iota_1}\in\g$, and $2i=[g_-^{\iota_1},g_+^{\iota_1}]\in\g$. Henceforth $\g=\langle\{i,ia^\dagg a,g_+^{\iota_1},g_-^{\iota_1},g_+^{2\iota_1},g_-^{2\iota_1}\}\rangle_{\mathrm{Lie}}\cong\mathcal{S}\cong\slR{2}\ltimes\gh_1$.
			\item $\boldsymbol{\lie{\{i(a^\dagg a+c),e_1\}}=\langle\{i,ia^\dagg a,g_+^{2\iota_1},g_-^{2\iota_1}\}\rangle_{\mathrm{Lie}}}$: This case can be treated analogously to the previous case (ii), when one assumed that either $\tilde{c}_3\neq0$ or $\tilde{c}_4\neq0$.
		\end{enumerate}
		We have now considered all possible cases in which a non-abelian finite-dimensional Lie algebra $g\subseteq\hat{A}_1$ contains a free Hamiltonian term $i(a^\dagg a+c)$. Our results obtained above allow us, together with Propositions~\ref{prop:algebra:go:g1:space},~\ref{prop:algebra:g2:space}, and~\ref{prop:algebra:g0:g2:space}, to conclude that $\g$ is isomorphic to one of the following Lie algebras: (a) $\langle\{i,ia^\dagg a,g_+^{\iota_1}, g_-^{\iota_1}, g_+^{2\iota_1}, g_-^{2\iota_1}\}\rangle_{\mathrm{Lie}}=\langle\hat{A}_1^0\oplus\hat{A}_1^1\oplus\hat{A}_1^2\rangle_{\mathrm{Lie}}\cong \mathcal{S}=\slR{2}\ltimes\gh_1$, (b) $\langle\{i,ia^\dagg a,g_+^{2\iota_1},g_-^{2\iota_1}\}\rangle_{\mathrm{Lie}}=\langle\hat{A}_1^0\oplus\hat{A}_1^2\rangle_{\mathrm{Lie}}=\slR{2}\oplus\R$, (c) $\langle\{i(a^\dagg a+1/2),g_+^{2\iota_1},g_-^{2\iota_1}\}\rangle_{\mathrm{Lie}}=\langle\hat{A}_1^2\rangle_{\mathrm{Lie}}\cong\slR{2}$, or $\langle\{i,ia^\dagg a,g_+^{\iota_1},g_-^{\iota_1}\}\rangle_{\mathrm{Lie}}=\langle\hat{A}_1^0\oplus\hat{A}_1^1\rangle_{\mathrm{Lie}}\cong \wh_2$.  
	\end{proof}

	\section{Nilpotent and non-solvable finite-dimensional Lie algebras realizable in $\hat{A}_1$\label{section:classification:nilpotent:nonsolvable}}
	
	In this section, we aim to classify all finite-dimensional nilpotent and non-solvable Lie algebras that can be faithfully realized as a subalgebra of the skew-hermitian Weyl algebra $\hat{A}_1$. While the classification of such algebras in the complex Weyl algebra $A_1$ has already been achieved \plainrefs{TST:2006,Tanasa:2005}, our focus here is on the real skew-hermitian setting, which introduces additional structural constraints.
	Rather than relying on weight-based methods as done in previous work \plainrefs{TST:2006,Tanasa:2005}, we develop the classification using the formalism introduced earlier combined with key results from the literature. This approach allows us to reach our goal.

	\subsection{Basic tools}
	
	We start our new endeavour by recalling key and relevant concepts relating to Lie algebras.
	First, we want to introduce the concept of Lie-algebra homomorphism, which formalizes the notion of structure preserving maps between Lie algebras. This concept is fundamental to our analysis, as it allows us to rigorously define when an abstract Lie algebra can be faithfully represented in the real skew-hermitian Weyl algebra $\hat{A}_1$.
	
	\begin{definition}[Lie-algebra homorphism]\label{def:Lie:algebra:homomorphism}
		Let $\g$ and $\gh$ be two real Lie algebras equipped with Lie brackets $[\cdot,\cdot]_\g$ and $[\cdot,\cdot]_{\gh}$ respectively. A \emph{Lie-algebra homomorphism} $\phi:\g\to\gh$ is a linear map that satisfies
		\begin{align}\label{eqn:def:Lie:algebra:isomorphism}
			\phi([x,y]_\g)=[\phi(x),\phi(y)]_{\gh}\quad\text{for all }x,y\in\g.
		\end{align}
		If  $\g=\gh$, $\phi$ is a called \emph{Lie-algebra endomorphism}, written $\phi\in\operatorname{End}(\g)$. If $\phi$ is invertible (or bijective) it is called a \emph{Lie-algebra isomorphism}. An endomorphic Lie-algebra isomorphism is denoted \emph{Lie-algebra automorphism} \plainrefs{Knapp:1996}. If there exists a Lie-algebra isomorphism $\phi$ between $\g$ and $\gh$, we write $\g\cong\gh$. 
	\end{definition}
	
	The notion of a realization refines the idea of a Lie-algebra homomorphism by specifying the target space. In our case, the target space is the skew-hermitian Weyl algebra $\hat{A}_1$. Realizations are central to our classification program, as they determine whether a given abstract Lie algebra can be concretely represented by a set of particular physical operators. Here, faithful realizations are of interest as they guarantee a one-to-one correspondence between the abstract generators and physical operators.
	\begin{definition}[Realization]
		Let $\g$ be a finite-dimensional real Lie algebra. A map $\Phi:\g\to\Phi(\g)\subseteq\hat{A}_1$ is called a \emph{realization} if $\Phi$ is a Lie-algebra homomorphism. A realization $\Phi$ is called \emph{faithful} if $\Phi$ is a Lie-algebra isomorphism~\plainrefs{Popovych:2003}.
	\end{definition}
	
	The next two definitions are essential to the concept of the semidirect sum, which is, by the Levi-Mal'tsev theorem \plainrefs{Kuzmin:1977}, a fundamental tool for understanding the structure of any finite-dimensional Lie algebra. We first introduce the derivation of a Lie algebra, which is the generalization of the standard derivation operator. That is, a derivation is a linear operator that satisfies the Leibniz identity. In our setting, it is important, as it ensures compatibility of Lie-algebra homomorphisms with the Lie bracket and is therefore crucial for characterizing how elements of one algebra act on another. We then introduce  representations, which generalize the action of a Lie algebra on other algebraic structures and are therefore essential for studying composite algebras, where one algebra acts on another. Here, in this work, they allow us to describe how subalgebras of $\hat{A}_1$ act reciprocally on each other.
	
	\begin{definition}[Derivation]\label{def:derivation}
		Let $\g$ be a Lie algebra. The map $\pi:\g\to\g$ is called a \emph{derivation} if and only if it is a linear map, satisfying the Leibniz rule $\pi([x,y])=[\pi(x),y]+[x,\pi(y)]$ for all $x,y\in\g$ \plainrefs{Knapp:1996}. 
	\end{definition}

	\begin{definition}[Representation]\label{def:representation}
		Let $\g$ and $\gh$ be two Lie algebras and $\pi:\g\to\operatorname{End}(\gh)$ a map that satisfies the conditions: (i) $\pi$ is linear, (ii) $\pi([x,y])=[\pi(x),\pi(y)]$ for all $x,y\in\g$, and (iii) $\pi(x)$ is a derivation on $\gh$ for all $x\in\g$. 
		Then $\pi$ is called a \emph{representation of $\g$ on $\gh$} \plainrefs{Knapp:1996}.
	\end{definition}

	\begin{definition}[Semidirect sum]\label{def:semidirect:product}
		The (outer) direct sum $\g\oplus\gh$ with the Lie bracket $[\cdot,\cdot]_{\pi}$ is called the semidirect sum $\g\ltimes_\pi\gh$ if and only if $\pi$ is a representation of $\g$ on $\gh$ and the Lie bracket $[\cdot,\cdot]_{\pi}$ is defined as
		\begin{align*}
			[(x_1,y_1),(x_2,y_2)]:=([x_1,x_2],[y_1,y_2]+\pi(x_1)y_2-\pi(x_2)y_1)
		\end{align*}
		for all $x_1,x_2\in\g$ and $y_1,y_2\in\gh$ \plainrefs{Barati:2020}. 
	\end{definition}
	
	It is easy to verify the following: (a) $\g\ltimes_\pi\gh$ is a Lie algebra, (b) the restriction of the Lie bracket $[\cdot,\cdot]_\pi$  to the algebras $\g$ and $\gh$ agrees with the Lie bracket on the original Lie algebras, and (iii) $[(x,0),(0,y)]=(0,\pi(x)y)$ for every $x\in\g$ and $y\in\gh$. We'll omit the representation indicator when $\pi(x)=[x,\cdot]$ and $\pi(x)(\gh)\neq\{0\}$, writing $\g\ltimes \gh$ instead of $\g\ltimes_\pi\gh$. We'll furthermore write $\g\oplus\gh$ instead of $\g\ltimes_\pi\gh$ if $\pi$ is a trivial derivation, i.e.,satisfies $\pi(x)y=0$ for all $x\in \g$ and $y\in\gh$. Note that, in the finite-dimensional setting, the terms \emph{semidirect sum} and \emph{semidirect product} are often used interchangeably in the literature \plainrefs{Knapp:1996,Woit:2017:Semidirect}.
	
	\begin{tcolorbox}[breakable, colback=Cerulean!3!white,colframe=Cerulean!85!black,title=Example: Semidirect sums]
		We want to clarify the Definition~\ref{def:semidirect:product} by presenting a few examples, and we want to justify the introduction of the convention to write $\ltimes$ instead of $\ltimes_{[\,\cdot\,,\,\cdot\,]}$, as done in the previous sections.
		\begin{example}\label{exa:semidirect:product:direct:sum:derivation:trivial}
			We here validate the notation $\g\oplus\gh$ for the semidirect sum $\g\ltimes_\pi\gh$ if $\pi(x)y=[x,y]=0$ for every $x\in\g$ and $y\in\gh$: Let $\g$ and $\gh$ be two finite-dimensional Lie algebras with bases $\mathcal{B}_\g=\{x_j\}_{j\in\mathcal{J}}$ and $\mathcal{B}_{\gh}=\{y_k\}_{k\in\mathcal{K}}$ respectively. Then, one can introduce the map $\Phi:\g\ltimes_\pi\gh\to\g\oplus\gh$, as a linear map defined by the identifications $(x_j,0)\leftrightarrow x_j$ and $(0,y_k)\leftrightarrow y_k$. It is clear that
			\begin{align*}
				\Phi([(x_j,y_k),(x_{j'},y_{k'})])&=\Phi([x_j,x_{j'}],[y_k,y_{k'}])=\Phi\left(\left(\sum_{\ell\in\mathcal{J}}f_{jj'\ell}x_\ell,\sum_{\ell'\in\mathcal{K}}\hat{f}_{kk'\ell'}y_{\ell'}\right)\right)=\sum_{\ell\in\mathcal{J}}f_{jj'\ell} x_\ell+\sum_{\ell'\in\mathcal{K}}\hat{f}_{kk'\ell'}y_{\ell'}\\
				&=[x_j,x_{j'}]+[y_k,y_{k'}]=[x_j,x_{j'}]+[x_j,y_{k'}]+[y_k,x_{j'}]+[y_k,y_{k'}]\\
				&=[x_j+y_k,x_{j'}+y_{k'}]=[\Phi(x_j,y_k),\Phi(x_{j'},y_{k'})],
			\end{align*}
			where $f_{jk\ell}$ and $\hat{f}_{jk\ell}$ are the structure constants of the algebras $\g$ and $\gh$ respectively. Thus, $\Phi$ is a Lie-algebra endomorphism. The one-to-one identification furthermore guarantees the invertibility of $\Phi$, making $\Phi$ a Lie-algebra isomorphism. Thus $\g\ltimes_\pi\gh\cong \g\oplus\gh$ if $\pi=[\,\cdot\,,\,\cdot\,]$ and $C_\g(\gh)=\g$.
		\end{example}
		
		We also want to verify that the Lie bracket $[\cdot,\cdot]$ constitutes a valid map satisfying the three conditions listed in Definition~\ref{def:semidirect:product}:
		\begin{example}
			Let $\g$ and $\gh$ be two finite-dimensional subalgebras of the skew-hermitian Weyl algebra $\hat{A}_1$. Then, one has by definition that the Lie bracket $[\cdot,\cdot]$ is bilinear. One has, furthermore, by the Jacobi identity
			\begin{align*}
				[[x,y],\cdot]=-[[y,\cdot],x]-[[\cdot,x],y]=[x,[y,\cdot]]-[y,[x,\cdot]]\quad\text{for all }x,y\in\g.
			\end{align*}
			Thus $[\cdot,\cdot]$ satisfies condition (ii) of Definition~\ref{def:representation}. This leaves computing
			\begin{align*}
				[x,[y,z]]=-[y,[z,x]]-[z,[x,y]]=[[x,y],z]+[y[x,z]]\quad\text{for all }x\in\g\text{ and }y,z\in\gh.
			\end{align*}
			Here, we used again the Jacobi identity. Thus $[x,\cdot]$ is a derivation on $\gh$ for every $x\in\g$ verifying condition (iii) of Definition~\ref{def:representation}. Hence, it is valid to write $\g\ltimes_{[\cdot,\cdot]}\gh=\g\ltimes\gh$.
		\end{example}
	\end{tcolorbox}

	\subsection{The abelian centralizer condition}
	
	The goal of this section is to prove the following statement: \emph{The centralizer of any non-scalar element of the skew-hermitian Weyl algebra $\hat{A}_1$ is an abelian subalgebra of $\hat{A}_1$}. We refer to this property as the \emph{abelian centralizer condition}. This condition is also known as the \emph{commutative centralizer condition} \plainrefs{TST:2006}. Alternatively, a Lie algebra satisfying this property is called \emph{commutative transitive} \plainrefs{Gorbatsevich:2017}.
	
	\begin{proposition}
		Let $\g$ be a Lie algebra and $x\in\g$. Then the centralizer $C_\g(x)$ is a subalgebra of $\g$.
	\end{proposition}
	
	\begin{proof}
		This follows directly from the bilinearity of the Lie bracket and the Jacobi identity.
	\end{proof}
	
	\begin{proposition}\label{prop:center:of:A1}
		The center $\mathcal{Z}(\hat{A}_1)$ of the skew-hermitian Weyl algebra $\hat{A}_1$ is isomorphic to the real numbers $\R$, i.e., it only consists of scalar multiples of the identity, referred to as scalar elements.
	\end{proposition}
	
	\begin{proof}
		Let $x\in \mathcal{Z}(\hat{A}_1)$. Then, we can write as a unique decomoposition $x=\sum_{p\in\mathcal{P}}c_pg_{\sigma_p}^{\gamma_p}$, where $c_p\neq0$ for all $p\in\mathcal{P}$. Since $x\in\mathcal{Z}(\hat{A}_1)$,  it must commute with all elements in $\hat{A}_1$. In particular, it must commute with the two elements $ia^\dagg a\in\hat{A}_1$ and $a-a^\dagg\in\hat{A}_1$. Thus, one must have $[ia^\dagg a,x]=0$. Proposition~\ref{prop:thm:25:equivalent} implies then that $x=\sum_{p\in\mathcal{P}}c_pg^{p\tau}\in\hat{A}_1^0\oplus\hat{A}_1^=$. One must furthermore have $[x,a-a^\dagg]=0$. Here, again Proposition~\ref{prop:thm:25:equivalent} implies that $x=2i c$ for some $c\in\R$, i.e.,$x$ is a scalar multiple of the identity. The trivial observation $\langle \{2ic\mid\,c\in\R\}\rangle_{\mathrm{Lie}}=i\R\cong\R$ concludes the proof. 
	\end{proof}

	\begin{theorem}[Centralizer condition]\label{thm:commutative:centralizer:condition}
		Let $x\in\hat{A}_1\setminus\mathcal{Z}(\hat{A}_1)$. Then, the centralizer $C_{\hat{A}_1}(x)$ is an abelian subalgebra of~$\hat{A}_1$.
	\end{theorem}
	
	\begin{proof}
		It is known that for any non-scalar element $x$ in the real single-mode Weyl algebra, denoted $A_{1,\R}$, or in the complex Weyl algebra $A_1$, the centralizer $C_{A_{1,\R}}(x)$, or respectively $C_{A_1}(x)$ forms an abelian Lie algebra \plainrefs{Dixmier:1968,Amitsur1958CommutativeLD}. Note that the real single-mode Weyl $A_{1,\R}$ algebra is defined analogously to the complex Weyl algebra $A_1$: both are the infinite-dimensional associative Lie algebras generated by the creation and annihilation operators $a^\dagg$ and $a$, which satisfy the canonical commutation relation $[a,a^\dagg]=1$, with the distinction that $A_{1,\R}$ is defined over the real numbers rather than the complex numbers. We elaborate on the proof below.
		
		Let now $x\in \hat{A}_1\setminus\mathcal{Z}(\hat{A}_1)$. By Proposition~\ref{prop:center:of:A1}, $x$ is not a scalar. Consider now the complexification of the centralizer $(C_{\hat{A}_1}(x))_\C=C_{\hat{A}_1}(x)+i C_{\hat{A}_1}(x)$. We want to show now that $(C_{\hat{A}_1}(x))_\C\subseteq A_1$ holds. We can write any element $z$ in the complexification of $C_{\hat{A}_1}(x)$ as $z=z_1+iz_2$, where $z_1,z_2\in C_{\hat{A}_1}(x)$. Then: $[x,z]=[x,z_1]+i[x,z_2]=0$, since both $z_1$ and $z_2$ commute with $x$. Thus, the complexification of the centralizer of any $x\in\hat{A}_1\setminus\mathcal{Z}(\hat{A}_1)$ is a subset of the centralizer of the same $x$ over the complex Weyl algebra $A_1$. Now we proceed by contradiction: Suppose that $C_{\hat{A}_1}(x)$ is not abelian. Then there exist two elements $y,z\in C_{\hat{A}_1}(x)$ such that $[x,y]\neq0$. One must now obviously have $x,y\in (C_{\hat{A}_1}(x))_\C\subseteq C_{A_1}(x)$. We know, however, that this centralizer $C_{A_1}(x)$ is an abelian algebra. This is a contradiction, showing that $C_{\hat{A}_1}(x)$ must already be abelian. Since any abelian set is closed under the Lie bracket, $C_{\hat{A}_1}(x)$ forms an abelian Lie algebra.
	\end{proof}   
	
	\begin{tcolorbox}[breakable, colback=Cerulean!3!white,colframe=Cerulean!85!black,title=Example: abelian centralizer condition in the two-mode case]
		We want to emphasize the importance of Theorem~\ref{thm:commutative:centralizer:condition} by highlighting a qualitative difference between the case for the single-mode skew-hermitian Weyl algebra $\hat{A}_1$ and the case for the multi-mode skew-hermitian Weyl algebras $\hat{A}_n$. We do so by demonstrating that the abelian centralizer condition does only hold in the single-mode case.
		\begin{example}\label{exa:abelian:centralizer:condition:two:modes}
			Let $\hat{A}_2$ denote the two-mode skew-hermitian Weyl algebra defined as the real skew-hermitian subalgebra of the complex universal enveloping algebra \plainrefs{Hall:Lie:groups:16}, generated by the elements $a_1,a_2,a_1^\dagg,a_2^\dagg,1$ with the canonical commutation relation $[a_j,a_k^\dagg]=\delta_{jk}$ while all other commutators vanish. We want to consider the centralizer $C_{\hat{A}_2}(x)$ of the element $x:=i(a_1^\dagg a_1+a_2^\dagg a_2)\in\hat{A}_2$. It is straightforward to confirm that two elements $y:=i(a_1a_2^\dagg+a_1^\dagg a_2)\in\hat{A}_2$ and $z:=a_1a_2^\dagg-a_1^\dagg a_2\in\hat{A_2}$ lie in the centralizer $ C_{\hat{A}_2}(x)$. However, their commutator is $[y,z]=2i(a_1^\dagg a_1-a_2^\dagg a_2)\neq0$. Furthermore, since $\hat{A}_2$ is trivially a subalgebra of $\hat{A}_n$, this finally implies that Theorem~\ref{thm:commutative:centralizer:condition} cannot be extended beyond the single-mode case.
		\end{example}
	\end{tcolorbox}

	\subsection{Finite-dimensional nilpotent subalgebras of the skew-hermitian Weyl algebra}
	
	We now demonstrate the significant implications of the abelian centralizer condition by proving a structural theorem for nilpotent subalgebras of $\hat{A}_1$. The ideas underlying the proofs of Propositions~\ref{prop:center:of:non:abelian:subalgebras}, ~\ref{prop:intersection:of:two:centralizer} and~\ref{prop:index:of:non:abelian:nilpotent:algebra} below are based on work in the literature \plainrefs{Dixmier:1968,Tanasa:2005}, and they have been adopted to the context of this study.
	
	\begin{theorem}\label{thm:classification:of:all:nilpotent:subalgberas}
		Let $\gn$ be an $n$-dimensional nilpotent subalgbera of $\hat{A}_1$. Then $\gn$ is either abelian or admits a basis $\{y_0,\ldots,y_{n-2},x\}$, such that $[y_j,y_k]=0$ for every $j,k\in\{0,1,\ldots,n-2\}$, and $[x,y_{j}]=y_{j-1}$ for all $j\in\{1,\ldots,n-2\}$.
	\end{theorem}
	
	To prove this theorem one needs to make several important observations. Thus, to avoid an lengthy and non-illuminating proof, we first want to present some essential propositions that significantly reduce the complexity of the proof for Theorem~\ref{thm:classification:of:all:nilpotent:subalgberas}.
	
	\begin{proposition}\label{prop:center:of:non:abelian:subalgebras}
		Let $\g$ be a non-abelian subalgebra of $\hat{A}_1$. Then its center satisfies $\mathcal{Z}(\g)\subseteq\mathcal{Z}(\hat{A}_1)$, and one either has $\mathcal{Z}(\g)=i\R$ or $\mathcal{Z}(\g)=\{0\}$.
	\end{proposition}

	\begin{proof}
		Suppose there exists a non-scalar $z\in \mathcal{Z}(\g)$. Then one has clearly $\g\subseteq C_{\hat{A}_1}(z)$. But, by Theorem~\ref{thm:commutative:centralizer:condition}, $C_{\hat{A}_1}(z)$ is abelian, and so must every subalgebra of it. However, $\g$ is a non-abelian subalgebra of $C_{\hat{A}_1}(z)$, which is a contradiction. Thus, $z$ must have been a scalar and one has either $\mathcal{Z}(\g)=i\R\cong \R$ or $\mathcal{Z}(\g)=\{0\}$. Hence $\mathcal{Z}(\g)\subseteq\mathcal{Z}(\hat{A}_1)$ if $\g$ is a non-abelian subalgebra of $\hat{A}_1$.
	\end{proof}

	\begin{proposition}\label{prop:intersection:of:two:centralizer}
		Let $e_1,e_2\in\hat{A}_1\setminus\mathcal{Z}(\hat{A}_1)$. Then, one has $C_{\hat{A}_1}(\{e_1,e_2\})=C_{\hat{A}_1}(e_1)$ if $[e_1,e_2]=0$ and $C_{\hat{A}_1}(\{e_1,e_2\})=\mathcal{Z}(\hat{A}_1)$ otherwise.
	\end{proposition}
	
	\begin{proof}
		Let $e_1,e_2\in\hat{A}_1\setminus\mathcal{Z}(\hat{A}_1)$. Thus, by Proposition~\ref{prop:center:of:A1}, they are not scalar multiples of $i$. If $[e_1,e_2]=0$, one has $e_1\in C_{\hat{A}_1}(e_2)$ and $e_2\in C_{\hat{A}_1}(e_1)$, since $e\in C_{\hat{A}_1}(e)$ holds for every $e\in\hat{A}_1$ by the antisymmetry of the commutator. Therefore $e_1,e_2\in C_{\hat{A}_1}(e_1)\cap C_{\hat{A}_1}(e_2)$. Let now $e_3\in C_{\hat{A}_1}(e_2)$. Then $[e_1,e_3]=0$ by Theorem~\ref{thm:commutative:centralizer:condition}, and consequently $e_3\in C_{\hat{A}_1}(e_1)$. Henceforth, $C_{\hat{A}_1}(\{e_1,e_2\})=C_{\hat{A}_1}(e_1)$.
		
		Consider now the case $[e_1,e_2]\neq 0$. Before proceeding, one needs to realise that for any $x\in\hat{A}_1\setminus\mathcal{Z}(\hat{A}_1)$ and $y\in C_{\hat{A}_1}(x)\setminus\mathcal{Z}(\hat{A}_1)$, one has always $C_{\hat{A}_1}(x)=C_{\hat{A}_1}(y)$. This is a simple consequence of Theorem~\ref{thm:commutative:centralizer:condition}. In fact, for every $z\in C_{\hat{A}_1}(x)$ one has $[y,z]=0$ and consequently $z\in C_{\hat{A}_1}(y)$. Similarly, for every $z\in C_{\hat{A}_1}(y)$ one has $[x,z]=0$ and consequently $C_{\hat{A}_1}(x)= C_{\hat{A}_1}(y)$. Now, suppose $e_3\in C_{\hat{A}_1}(\{e_1,e_2\})\setminus\mathcal{Z}(\hat{A}_1)$, i.e., $e_3\in C_{\hat{A}_1}(x)\setminus\mathcal{Z}(\hat{A}_1)$ and $e_3\in C_{\hat{A}_1}(y)\setminus\mathcal{Z}(\hat{A}_1)$. Then, one would have $C_{\hat{A}_1}(e_1)=C_{\hat{A}_1}(e_3)=C_{\hat{A}_1}(e_2)$ and consequently $[e_1,e_2]=0$ due to Theorem~\ref{thm:commutative:centralizer:condition}. This is a contradiction to the assumption $[e_1,e_2]\neq 0$. Thus  $C_{\hat{A}_1}(\{e_1,e_2\})\subseteq\mathcal{Z}(\hat{A}_1)$, since $ C_{\hat{A}_1}(\{e_1,e_2\})\setminus\mathcal{Z}(\hat{A}_1)$ must be empty from the start. The trivial observation that $i\R\in C_{\hat{A}_1}(\{x,y\})$ concludes the proof that $[e_1,e_2]\neq0$ implies $C_{\hat{A}_1}(\{e_1,e_2\})=\mathcal{Z}(\hat{A}_1)$.
	\end{proof}
	
	\begin{corollary}\label{cor:intersection:of:two:centralizer}
		Let $\g$ be a non-abelian subalgebra of $\hat{A}_1$, and let $e_1,e_2\in\g\setminus\mathcal{Z}(\g)$. Then, one has $C_\g(\{e_1,e_2\})=C_\g(e_1)$ if $[e_1,e_2]=0$ and $C_\g(\{e_1,e_2\})=\mathcal{Z}(\g)$ otherwise.
	\end{corollary}
	
	\begin{proof}
		If $\g$ is abelian. Then $\g\setminus\mathcal{Z}(\g)=\emptyset$, and the above statement is not applicable. Thus, it is necessary to assume that $\g$ is a non-abelian subalgebra of $\hat{A}_1$.
		
		Now let $e_1,e_2\in\g\setminus\mathcal{Z}(\g)$. By virtue of Propositions~\ref{prop:center:of:A1} and~\ref{prop:center:of:non:abelian:subalgebras}, these two elements are non-scalar elements. If $e_1$ and $e_2$ commute, one has $e_1\in C_\g(e_2)\subseteq C_{\hat{A}_1}(e_2)$ and $e_2\in C_\g(e_1)\subseteq C_{\hat{A}_1}(e_1)$. Let now $e_3\in C_\g(e_2)$. Then one has $e_3\in C_{\hat{A}_1}(e_2)$ and consequently $[e_1,e_3]=0$ due to Theorem~\ref{thm:commutative:centralizer:condition}. This implies $e_3\in C_\g(e_1)$ and consequently $C_\g(\{e_1,e_2\})=C_\g(e_1)$.
		
		Consider an element $x\in \g\setminus\mathcal{Z}(\g)$ and a second element $y\in C_\g(x)\setminus\mathcal{Z}(\g)$. Let $z\in C_\g(x)\subseteq C_{\hat{A}_1}(x)$. Then one has $[z,y]=0$, due to Theorem~\ref{thm:commutative:centralizer:condition}, and consequently $z\in C_\g(y)$. One has similarly $[z,x]=0$ for every $z\in C_\g(y)$. Hence $C_{\g}(x)=C_\g(y)$ if $x,y\in \g\setminus\mathcal{Z}(\g)$ and $[x,y]=0$.
		
		Suppose now $e_3\in C_\g(\{e_1,e_2\})\setminus\mathcal{Z}(\g)$. Then, one has $C_\g(e_1)=C_\g(e_3)=C_\g(e_2)$ due to the previous observation, since $[e_1,e_3]=0=[e_2,e_3]$ and $e_1,e_2,e_3\in\g\setminus\mathcal{Z}(\g)$. This yields a contradiction, since $e_1\in C_\g(e_1)$ and $e_2\in C_\g(e_2)$ implies, in conjunction with the property $C_\g(e_1)=C_\g(e_2)$, that $[e_1,e_2]=0$. Thus $C_\g(\{e_1,e_2\})\setminus\mathcal{Z}(\g)=\emptyset$ and consequently $C_\g(\{e_1,e_2\})\subseteq\mathcal{Z}(\g)$. Let now $z\in \mathcal{Z}(\g)$. Then one trivially has $z\in C_\g(\{e_1,e_2\})$ implying $C_\g(\{e_1,e_2\})=\mathcal{Z}(\g)$ if $[e_1,e_2]=0$.
	\end{proof}
	
	\begin{proposition}\label{prop:index:of:non:abelian:nilpotent:algebra}
		Let $\gn$ be a non-abelian $n$-dimensional nilpotent subalgebra of $\hat{A}_1$. Then the nilpotency index of $\gn$ is $n-1$, i.e., the lower central series terminates at zero after $n-1$ steps. Moreover, one has:
		\begin{enumerate}[label = (\alph*)]
			\item $\dim(\gn_j)=n-j-1$ for all $j\in\{1,\ldots,n-1\}$, where $\gn_{j+1}:=[\gn,\gn_{j}]$ is the $(j+1)$-th Lie algebra in the lower central series and $\gn_0:=\gn$;
			\item The center of $\gn$ is $\mathcal{Z}(\gn)=\gn_{n-2}=i\R$;
			\item The first derived algebra (coinciding with the first Lie algebra in the lower central series) $\gn_1=[\gn,\gn]$ is abelian;
			\item $\gn=\spn\{e\}\oplus \gr$, where $\gr$ is abelian and $e$ is an element in $\gn$ such that $[e,\gr]\neq0$ and $\gn_1\subseteq\gr$.
		\end{enumerate} 
	\end{proposition}
	
	\begin{proof}
		Let $\gn$ be a non-abelian $n$-dimensional nilpotent subalgebra of $\hat{A}_1$ that is of nilpotency index $k$. This means the following: let $\{\gn_j\}_{j\in\N_{\geq0}}$ be the lower central series defined by $\gn_0:=\gn$ and $\gn_{j+1}:=[\gn,\gn_j]$ for all $j\in\N_{>0}$, then $\gn_{k-1}\neq \{0\}$ while $\gn_{\ell}=\{0\}$ for all $\ell\geq k$.
		
		First, we want to show that the center of $\mathfrak{n}$ is $\mathcal{Z}(\gn)=\gn_{k-1}=i\R$. Let therefore $z\in\mathcal{Z}(\gn)$. Then $z=ic$ for $c\in\R$ by virtue of Propositions~\ref{prop:center:of:A1} and~\ref{prop:center:of:non:abelian:subalgebras}, since $\gn$ is non-abelian by assumption and a finite-dimensional Lie subalgebra of $\hat{A}_1$. One has furthermore $\gn_{k-1}\subseteq\mathcal{Z}(\gn)$, since $[\gn,\gn_{k-1}]=\gn_k=\{0\}$. Moreover, one has $\gn_{k-1}\subseteq\mathcal{Z}(\hat{A}_1)$ by Proposition~\ref{prop:center:of:non:abelian:subalgebras}. Since $\mathcal{Z}(\hat{A}_1)=i\R$ and $\gn_{k-1}\neq\{0\}$, one must have $\gn_{k-1}=\mathcal{Z}(\hat{A}_1)=i\R=\mathcal{Z}(\gn)$, thus proving claim (b), if the nilpotency index of $\gn$ is $n-1$, which will be established below.
		
		Let us proceed with showing that $\gn_1$ is abelian. By the previous discussion, for every $c\in\R$ there must exist an element $\hat{e}_3\in\gn_{k-2}$ and $\hat{e}_1\in\gn$ such that $[\hat{e}_3,\hat{e}_1]=ic$. In particular, there exist two elements $e_3\in\gn_{k-2}$ and $e_1\in\gn$ such that $[e_1,e_3]=i$. Note that $e_3$ is not a scalar element, as it would otherwise commute with any element in $\gn$. Consider now two elements $\tilde{e}_1,\tilde{e}_2\in\gn$. These satisfy respectively $[\tilde{e}_1,e_3]=ic_1$ and $[\tilde{e}_2,e_3]=ic_2$, due to $[\gn,\gn_{k-2}]=i\R$. Using the Jacobi identity, one computes
		\begin{align*}
			[[\tilde{e}_1,\tilde{e}_2],e_3]=-[[\tilde{e}_2,e_3],\tilde{e}_1]-[[e_3,\tilde{e}_1],\tilde{e}_2]=[e_1,ic_1]+[ic_1,e_2]=0.
		\end{align*}
		Since $\tilde{e}_1$ and $\tilde{e}_2$ were arbitrarily chosen, and their commutator lies in $\gn_1=[\gn,\gn]$, the Lie algebra $\gn_1$ lies in the centralizer $C_{\gn}(e_3)$ which is abelian by Theorem~\ref{thm:commutative:centralizer:condition}, making $\gn_1$ itself also abelian and proving claim (c). 
		
		Next, we want to show that $\gn$ can be decomposed as $\gn=\spn\{e\}\oplus \gr$, where $\gr$ is abelian and $e$ is an element in $\gn$ such that $[e,\gr]\neq0$ and $\gn_1\subseteq\gr$. For this, suppose that the two coefficients $c_1$ and $c_2$ defined by the commutators $[\tilde{e}_1,e_3]=ic_1$ and $[\tilde{e}_2,e_3]=ic_2$ respectively satisfy $c_1\neq 0\neq c_2$. One can then simply linearly combine the equations $[\tilde{e}_1,e_3]=ic_1$ and $[\tilde{e}_2,e_3]=ic_2$ as follows: $c_2[\tilde{e}_1,e_3]-c_1[\tilde{e}_2,e_3]=[c_2\tilde{e}_1-c_1\tilde{e}_2,e_3]=0$. Hence, $c_2\tilde{e}_1-c_1\tilde{e}_2\in C_{\gn}(e_3)$. This shows that for any two elements $\hat{e}_1,\hat{e}_2\in\gn$ that do not commute with $e_3$ there exist two real non-zero coefficients $c_1,c_2\in\R$ such that $c_1\tilde{e}_1+c_2\tilde{e}_2\in C_{\gn}(e_3)$. This allows us to conclude that $\gn\setminus\mathcal{Z}(\gn)$ is one dimensional, and we can write $\gn=\spn\{e\}\oplus \gr$, where $\gr\subseteq C_{\gn}(e_3)$ and $e$ is some particular element that does not commute with $e_3$ and has no non-scalar support in $C_{\gn}(e_3)$. One has furthermore that $\gr$ is abelian, due to Theorem~\ref{thm:commutative:centralizer:condition}, showing claim (d). 
		
		This leaves showing that $k=n-1$ and $\dim(\gn_j)=n-j-1$ for all $j\in\{1,\ldots,k\}$. To prove this, one notices that
		\begin{align*}
			\gn_1&=[\gn,\gn]=[\spn\{e\},\gr],\;&\;\gn_2&=[\spn\{e\},[\spn\{e\},\gr]],\;&\;\ldots,\;&\;\;&\;\;&\;\gn_j=[\spn\{e\}[\ldots,[\spn\{e\},\gr]]],
		\end{align*}
		since $\gn_{j+1}\subseteq \gn_j$ and $\gn_1\subseteq \gr$. The last claim can easily be verified by assuming the converse. In such case, there would exist an element $\tilde{e}$ that can, without loss of generality, be written as $\tilde{e}=e+\tilde{r}$, where $\tilde{r}\in\gr$. However, $\tilde{e}$ cannot belong to $C_\gn(e_3)$, since $e$ does not commute with $e_3$ while $\tilde{r}$ does. Thus, we can define the linear map $\Phi_j:\gn_j\to\gn_{j+1},e_1\mapsto\Phi_j(e_1):=[e,e_1]$. Let now $e_4\in \gn_j$ with $j\geq 1$. Then, one has $e_4\in C_{\gn}(e_3)$, since $e_4\in\gn_j\subseteq\gn_1\subseteq C_{\gn}(e_3)$. Suppose now $e_4\in \ker\Phi_j$. Then, one has $[e,e_4]=0$ and consequently $e_4\in C_{\gn}(\{e,e_3\})\subseteq C_{\hat{A}_1}(\{e,e_3\})$. By Corollary~\ref{cor:intersection:of:two:centralizer}, $e_4$ must be in $\mathcal{Z}(\gn)$ and consequently $e_4=ic$ for some $c\in\R$. Hence $\ker\Phi_j\subseteq i\R$ and clearly also $\ker\Phi_j=i\R$ if $j<k$. The rank-nullity theorem implies that $\dim(\gn_j)=\dim(\ker\Phi_j)+\dim(\operatorname{Im}\Phi_j)=1+\dim(\gn_{j+1})$. Since $\gn_k=\{0\}$ and $\gn_{k-1}\neq\{0\}$, one has $\dim(\gn_{k-j})=j$ for all $j\in\{0,\ldots,k-1\}$, which proves claim (a). This leaves considering the linear map $\Phi_0:\gn\to\gn_1,e_1\mapsto\Phi_0(e_1)=[e,e_1]$. Here, one has now $\ker\Phi_0=\spn\{e,i\}$ and consequently $n=\dim(\gn)=\dim(\ker\Phi_0)+\dim(\operatorname{Im}\Phi_0)=2+k-1$. This implies $k = n-1$ and completes this proof.
	\end{proof}
	
	We can now return to Theorem~\ref{thm:classification:of:all:nilpotent:subalgberas} and proceed with its proof employing the tools developed up to here.
	
	\begin{proof}
		The idea of this proof is to employ claims (a) - (d) from Proposition~\ref{prop:index:of:non:abelian:nilpotent:algebra} to recursively define the desired basis. Thus, let $\gn$ be a non-abelian $n$-dimensional nilpotent subalgebra of $\hat{A}_1$.  Let $\{\gn_j\}_{j\in\N_{\geq0}}$ be the lower central series of $\gn$ defined by $\gn_0:=\gn$ and $\gn_{j+1}:=[\gn,\gn_j]$ for all $j\in\N_{\geq0}$. Proposition~\ref{prop:index:of:non:abelian:nilpotent:algebra} implies that the nilpotency index of $\gn$ is $n-1$, and $\dim(\gn_j)=n-j-1$ for $j\in\{1,\ldots,n-1\}$. Furthermore, it also states that $\gn_{n-2}=i\R$. Thus, for every $c\in\R$ there must exist an element $\tilde{y}\in\gn_{n-3}$ and $\tilde{x}\in\gn$ such that $[\tilde{x},\tilde{y}]=ic$. In particular, there exist two elements $y_1\in\gn_{n-3}$ and $x\in\gn$ such that $[x,y_1]=i=:y_0$. Note that $y_1$ is not a scalar element, as it would otherwise commute with any element in $\gn$. 
		
		We can now construct the desired basis inductively. That is, we show that $\gn$ admits the basis $\{y_0,\ldots,y_{n-2},x\}$, where $[y_j,y_k]=0$ for every $j\in\{0,1,\ldots,n-2\}$ and $[x,y_{j}]=y_{j-1}$ for every $j\in\{1,\ldots,n-2\}$. For the base case, we have to consider the following two cases: 
		\begin{enumerate}
			\item One has $\boldsymbol{n>3}$. Then $x\notin\gn_{n-3}$, since $y_1\in\gn_{n-3}\subseteq \gn_1$ which is an abelian algebra due to Proposition~\ref{prop:index:of:non:abelian:nilpotent:algebra}. Thus, one has $\spn\{y_1,y_0\}\subseteq\gn_{n-3}$, where $\dim(\spn\{y_1,y_0\})=2$. Proposition~\ref{prop:index:of:non:abelian:nilpotent:algebra} implies that $\dim(\gn_{n-3})=2$, implying $\gn_{n-3}=\spn\{y_0,y_1\}$. 
			\item One has $\boldsymbol{n=3}$. Then $x,y_1\in\gn$ and consequently $\gn=\spn\{y_0,y_1,x\}$ which is clearly isomorphic to the Heisenberg algebra $\gh_1$. 
		\end{enumerate}
		
		We can now iterate this process for $n>3$, however, we first note that Proposition~\ref{prop:index:of:non:abelian:nilpotent:algebra} states that $\gn=\spn\{e\}\oplus \gr$, where $\gr$ is abelian and $e\in\gn\setminus\gn_1$. Since $[x,y_1]\neq 0$ and $y_1\in\gn_{n-3}\subseteq\gn_1\subseteq\gr$, we know that $x\notin \gr$ and one can assume without loss of generality that $e=x$ since any non-trivial support $x$ has in $\gr$ can simply be disregarded because it commutes with any element in $\gr$ (alternatively, once can construct $e$ as a linear combination of $x$ and appropriate elements in $\gr$ such that $e$ has no non-zero support in $\gr$). 
		
		We proceed with the induction step and assume that the Lie algebra $\gn_{n-j}$ admits the basis $\{y_0,y_1,\ldots,y_{j-2}\}$ for some index $j\in\{2,\ldots,n-1\}$. One has again two cases:
		\begin{enumerate}
			\item One has $\boldsymbol{n>j+1}$, and we know that for every $z\in\spn\{y_0,y_1,\ldots,y_{j-2}\}$ there must exist two elements $\tilde{x}\in\gn$ and $\tilde{y}_{j-1}\in\gn_{n-j}$ such that $[\tilde{x},\tilde{y}_{j-1}]=z$. In particular, two elements $\tilde{x}'\in\gn$ and $\tilde{y}_{j-1}'\in\gn_{n-j}$ such that $[\tilde{x}',\tilde{y}_{j-1}']=y_{j-2}$. Clearly, one has $\tilde{x}'\notin \gr$ and therefore $\tilde{x}'=c x+ r$, where $c\in\R\setminus\{0\}$ and $r\in \gr$. One can now choose without loss of generality $x=\tilde{x}'$, since $[r,\tilde{y}_{j-1}']=0$. Thus, there exists some element $y_{j-1}\in\gn_{n-j+1}$ such that $[x,y_{j-1}]=y_{j-2}$. The condition $\dim(\gn_{n-j+1})=j$ implies then $\spn\{y_0,y_1,\ldots,y_{j-1}\}=\gn_{n-j+1}$. Here it is important to note that $y_{j-1}\notin\spn\{y_0,y_1,\ldots,y_{j-2}\}$ as otherwise $\gn$ wouldn't be nilpotent. The condition $\gn_{n-j-1}\subseteq\gn_1\subseteq\gr$ along with the assumption $n>j+1$ implies furthermore that $\spn\{y_0,y_1,\ldots,y_{j-1}\}$ is abelian. 
			\item One has $\boldsymbol{n=j}$. Here, we can still follow the logic of the previous case with the difference being that $x,y_{n-2}\in\gn$. In this case one sees $\gn=\spn\{y_0,y_1,\ldots,y_{n-2},x\}$, where $[y_j,y_k]=0$ for every $j,k\in\{0,1,\ldots,n-2\}$ and $[x,y_{j}]=y_{j-1}$ for every $j\in\{1,\ldots,n-2\}$.
		\end{enumerate}
		This completes the proof by induction and we conclude that iterating this procedure yields $\gn=\spn\{y_0,\ldots,y_{n-2},x\}$, where $[y_j,y_k]=0$ for every $j\in\{0,1,\ldots,n-2\}$ and $[x,y_{j}]=y_{j-1}$ for every $j\in\{1,\ldots,n-2\}$, where the finiteness of $\gn$ guarantees that this procedure terminates eventually after finitely many steps. Thus, every non-abelian finite-dimensional nilpotent subalgebra of $\hat{A}_1$ admits the claimed basis.
	\end{proof}
	
	\begin{corollary}
		Let $\g$ be a finite-dimensional solvable subalgebra of $\hat{A}_1$. Then its derived algebra $[\g,\g]$ is either (a)~abelian, or (b) nilpotent and has the basis $\{y_0,\ldots,y_{n-1},x\}$, where $[y_j,y_k]=0$ for every $j,k\in\{0,1,\ldots,n-2\}$ and $[x,y_{j}]=y_{j-1}$ for every $j\in\{1,\ldots,n-2\}$.
	\end{corollary}
	
	\begin{proof}
		It is a standard result in Lie theory that the derived algebra $[\g,\g]$ of any real finite-dimensional solvable algebra $\g$ is nilpotent \plainrefs{Knapp:1996} (this may not hold for fields with non-zero characteristic \plainrefs{Jacobson:1979}). Since $\g\subseteq\hat{A}_1$  it follows that its derived algebra $[\g,\g]$ is a finite-dimensional nilpotent subalgebra of $\hat{A}_1$. Theorem~\ref{thm:classification:of:all:nilpotent:subalgberas} implies then that any such subalgebra is either abelian or admits the desired basis, which concludes the proof. 
	\end{proof}

	\subsection{Non-solvable realizable Lie algebras}
	
	\begin{proposition}\label{prop:complexification:Invariant:properties}
		Let $\g$ be a finite-dimensional real Lie algebra and $\g_\C:=\g+i\g$ its complexification. Then one has: (a) $\g$ is nilpotent with nilpotency index $k$ if and only if $\g_\C$ is nilpotent with nilpotency index $k$; (b) $\g$ is solvable if and only if $\g_\C$ is solvable; (c) $\g$ is semisimple if and only if $\g_\C$ is semisimple. 
	\end{proposition}
	
	\begin{proof}
		We treat each case separately:
		\begin{enumerate}[label = (\alph*)]
			\item By Proposition 1.17 from \plainrefs{Knapp:1996}, one as $[\g,\g]_\C=[\g_\C,\g_\C]$. Consequently $\g$  is nilpotent with nilpotency index $k$ if and only if $\g_\C$ is nilpotent with nilpotency index $k$. 
			\item This can be shown analogously to the previous claim. 
			\item By Corollary 1.50 from \plainrefs{Knapp:1996}, $\g$ is semisimple if and only if $\g_\C$ is semisimple. \qedhere
		\end{enumerate}
	\end{proof}
	
	Before proceeding, it is necessary to comment on the faithful realizability of different Lie algebras under the constraint that the faithful realized algebras satisfy some specified structural property. Let $\g$ and $\gh$ be two finite-dimensional Lie algebras that can each be faithfully realized in $\hat{A}_1$ in multiple ways. That is, there exist at least four Lie-algebra homomorphisms $\Phi_1,\Phi_2:\g\to\hat{A}_1$ and $\Psi_1,\Psi_2:\gh\to\hat{A}_1$ that are, restricted to their images, isomorphic. 
	Suppose, for example, that these realizations satisfy the relation $[\Phi_j(\g),\Psi_j(\gh)]\subseteq \Psi_j(\gh)$ for each $j\in\{1,2\}$. We denote this relation by $\Phi_j(\g)\sim \Psi_j(\gh)$. 
	
	Let now $\{g_k\}_{k\in\mathcal{K}}$ and $\{h_\ell\}_{\ell\in\mathcal{L}}$ be bases of $\g$ and $\gh$ respectively. The sets $\{\Phi_1(g_k)\}_{k\in\mathcal{K}}$ and $\{\Psi_1(h_\ell)\}_{\ell\in\mathcal{L}}$ are then bases of $\Phi_1(\g)$ and $\Psi_1(\gh)$ respectively. We can now find a basis $\mathcal{B}_1$ of $\Phi_1(\g)\cup\Psi_1(\gh)$ by taking the set $\{\Psi_1(h_\ell)\}_{\ell\in\mathcal{L}}$ and adding all $\Phi_1(g_k)$ from $\{\Phi_1(g_k)\}_{k\in\mathcal{K}}$ that are linearly independent from the elements in $\{\Psi_1(h_\ell)\}_{\ell\in\mathcal{L}}$.
	
	We can now define the linear map $\Delta:\Phi_1(\g)\cup\Psi_1(\gh)\to\Phi_2(\g)\cup\Psi_2(\gh)$ via the identifications $\Delta(\Phi_1(g_k)):=\Phi_2(g_k)$ if $\Phi_1(g_k)\in\mathcal{B}_1$ and $\Delta(\Psi_1(h_\ell))=\Psi_2(h_\ell)$ for all $\ell\in\mathcal{L}$.
	It is straightforward to verify that $\Delta$ preserves the relation $\Phi_j(\g)\sim \Psi_j(\gh)$: 
	\begin{align*}
		\Delta([\Phi_1(g_k),\Psi_1(h_\ell)])&=\Delta\left(\sum_{j}c_{k\ell j}\Psi_j(h_j)\right)=\sum_j c_{k\ell j}\Delta(\Psi_1(h_j))=\sum_j c_{k\ell j}\Psi_2(h_j)\in\Psi_2(\gh).
	\end{align*}
	If $\Phi_1(g_k)$ is not in $\mathcal{B}_1$, one has clearly $[\Delta(\Phi_1(g_k)),\Delta(\Psi_1(h_\ell))]\in\Psi_2(\gh)$, otherwise, one finds:
	\begin{align*}
		[\Delta(\Phi_1(g_k)),\Delta(\Psi_1(h_\ell))]&=[\Phi_2(g_k),\Psi_2(h_\ell)]=\sum_j\tilde{c}_{k\ell j}\Psi_2(h_j)\in\Psi_2(\gh).
	\end{align*}
	
	Moreover, the restriction of $\Delta$ to $\Psi_1(\gh)$, written $\Delta|_{\Psi_1(\gh)}:\Psi_1(\gh)\to\Psi_2(\gh)$ and denoted $\Delta|_{\gh}$, is a Lie-algebra isomorphism. This follows from its linearity, invertability, and the fact that it preserves the Lie bracket:
	\begin{align*}
		\Delta|_{\Psi_1(\gh)}([\Psi_1(h_j),\Psi_2(h_k)])&=\Delta|_{\Psi_1(\gh)}(\Psi_1([h_j,h_k]))=\Delta|_{\Psi_1(\gh)}\left(\Psi_1\left(\sum_\ell f_{jk\ell} h_\ell\right)\right)=\Delta|_{\Psi_1(\gh)}\left(\sum_\ell f_{jk\ell}\Psi_1(h_\ell)\right)\\
		&=\sum_\ell f_{jk\ell}\Delta|_{\Psi_1(\gh)}(\Psi_1(h_\ell))=\sum_\ell f_{jk\ell}\Psi_2(h_\ell)=\Psi_2\left(\sum_\ell f_{jk\ell}h_\ell\right)=\Psi_2([h_j,h_k])\\
		&=[\Psi_2(h_k),\Psi_2(h_j)]=[\Delta|_{\Psi_1(\gh)}(\Psi_1(h_j)),\Delta|_{\Psi_1(\gh)}(\Psi_1(h_k))].
	\end{align*}
	Therefore, if we fix a particular faithful realization $\Phi_1$ of $\g$ and demand that a faithful realization $\Psi_1$ of $\gh$ satisfies the relation $\Phi_1(\g)\sim\Psi_1(\gh)$, it follows that for any other faithful realization $\Phi_2$ of $\g$, there exists a corresponding faithful realization $\Psi_2$ of $\gh$ that is isomorphic to $\Psi_1(\gh)$. We can depict the situation in the following diagram.
	\begin{figure}[H]\label{useful:diagram}
		\centering
		\begin{tikzpicture}
			\node[align = center] at (-3,0) {$\g$};
			\draw[dashed] (-1.7,0.8) rectangle (1.7,1.8);
			\draw[dashed] (-1.7,-0.8) rectangle (1.7,-1.8);
			\draw[<->] (-2.8,0.2) -- (-1.5,1.2);
			\draw[<->] (-2.8,-0.2) -- (-1.5,-1.2);
			\draw[<->] (2.8,0.2) -- (1.5,1.2);
			\draw[<->] (2.8,-0.2) -- (1.5,-1.2);
			\draw[<->] (0,-0.7) -- (0,0.7);
			\draw[<->] (1,-1) -- (1,1);
			\node[align = center] at (0+0.25,0) {$\Delta$};
			\node[align = center] at (1+0.4,0) {$\Delta|_{\gh}$};
			\node[align = center] at (2.4,0.8) {$\Psi_1$};
			\node[align = center] at (2.4,-0.8) {$\Psi_2$};
			\node[align = center] at (-2.4,0.8) {$\Phi_1$};
			\node[align = center] at (-2.4,-0.8) {$\Phi_2$};
			\node[align = center] at (0,1.3) {$\Phi_1(\g)\;\;\;\sim\;\;\;\Psi_1(\gh)$};
			\node[align = center] at (3,0) {$\gh$};
			\node[align = center] at (0,-1.3) {$\Phi_2(\g)\;\;\;\sim\;\;\;\Psi_2(\gh)$};
		\end{tikzpicture}
	\end{figure}\noindent
	The same reasoning applies when considering the relation $[\Phi_j(\g),\Psi_j(\gh)]\subseteq\{0\}$.
	
	\begin{proposition}\label{prop:realizable:semisimple:algebras}
		The only semisimple Lie algebra that can be faithfully realized in the skew-hermitian Weyl algebra $\hat{A}_1$ is $\slR{2}$.
	\end{proposition}
	
	\begin{proof}
		Let $\g$ be a semisimple Lie algebra that can be faithfully realized in $\hat{A}_1$. Then $\g_\C$, its complexification, can be faithfully realized in the complex Weyl algebra $A_1$. By virtue of  Proposition~\ref{prop:complexification:Invariant:properties}, it follows that $\g_\C$ is semisimple. It is known that the only complex semisimple Lie algebra that can be realized in $A_1$ is three-dimensional  \plainrefs{Klep:2010}. Furthermore, it is known that any real semisimple Lie algebra $\g$ satisfying the abelian centralizer condition is isomorphic to $\su{2}$, $\slR{2}$, or the real Lie algebra $\slC{2}$, see \plainrefs{Gorbatsevich:2017}. That is, $\slC{2}$ is seen as the real six-dimensional Lie algebra defined by the basis elements $\{h,x,y,ih,ix,iy\}$ and the commutation relations:
		\begin{align}\label{eqn:basis:for:sl:2:C}
			[h,x]&=2x,\;&\;[h,ix]&=2ix,\;&\;[h,y]&=-2y,\;&\;[h,iy]&=-2iy,\;&\;[x,y]&=h,\;&\;[x,iy]&=ih,\nonumber\\
			[ih,x]&=2ix,\;&\;[ih,ix]&=-2x,\;&\;[ih,y]&=-2iy,\;&\;[ih,iy]&=2y,\;&\;[ix,y]&=ih,\;&\;[ix,iy]&=-h,
		\end{align}
		while all other ones vanish.
		Let us begin by considering the three-dimensional cases. A classification of all three-dimensional Lie algebras has been provided over a century ago by Bianchi \plainrefs{Bianchi:1903}. The only semisimple candidates among the algebras provided by Bianchi are the real Lie algebras $\slR{2}$ and $\su{2}$. However, since $\su{2}$ corresponds to the compact Lie group $\operatorname{SU}(2)$ and is of rank one, it cannot be realized in $\hat{A}_1$, due to Theorem 4.4 from \plainrefs{Joseph:1972}, which prohibits the realization of semisimple Lie algebras with compact Lie groups that are of rank one. 
		
		Thus, we need to consider the possibility that $\slC{2}$ can be faithfully realized in $\hat{A}_1$. We can express $\slC{2}$ using the basis $\{h,ih,x,ix,y,iy\}$ defined by the commutation relations \eqref{eqn:basis:for:sl:2:C}.
		It follows that $\su{2}\subseteq\slC{2}$, since
		\begin{align*}
			\left[-\frac{1}{2}ih,\frac{1}{2}(ix+iy)\right]&=\frac{1}{2}(x-y),\;&\;\left[\frac{1}{2}(ix+iy),\frac{1}{2}(x-y)\right]&=-\frac{1}{2}ih,\;&\;\left[\frac{1}{2}(x-y),-\frac{1}{2}ih\right]&=\frac{1}{2}(ix+iy),
		\end{align*}
		which shows that $\lie{\{-ih/2,(ix+iy)/2,(x-y)/2\}}\cong\su{2}$. This, however, leads to a contradiction, as we have already seen that $\su{2}$ cannot be faithfully realized in $\hat{A}_1$. Hence, a faithful realization of $\slC{2}$ in $\hat{A}_1$ is also impossible.
		
		This implies that the only semisimple Lie algebra that can be realized in $\hat{A}_1$ is $\slR{2}$. Finally, Proposition~\ref{prop:algebra:g2:space} guarantees that $\slR{2}$ can be faithfully realized in $\hat{A}_1$.
	\end{proof}

	Note that Theorem 4.2 from \plainrefs{Joseph:1972} implies that no semisimple Lie algebra of rank greater than one can be realized in $\hat{A}_1$. Furthermore, it is important to mention that Simoni and Zaccaria have established in \plainrefs{simoni:1969} results equivalent to Joseph's Theorems 4.2 and 4.3 \plainrefs{Joseph:1972}. Given the significance of these results in the context of this section, we present a more detailed discussion in Appendix~\ref{App:Realizability:Recap}, where the proofs are reformulated using the adapted formalism from the foundational work \plainrefs{Bruschi:Xuereb:2024} on which the present one is based.
	
	\begin{theorem}
		Let $\g$ be a reductive Lie algebra that can be faithfully realized in $\hat{A}_1$. Then $\g\cong\slR{2}\oplus\R$ or $\g\cong \slR{2}$.
	\end{theorem}
	
	\begin{proof}
		If $\g$ is a reductive Lie algebra that can be faithfully realized in $\hat{A}_1$ then, by definition, it is the direct sum of a semisimple Lie algebra $\gh$ and its center $\mathcal{Z}(\g)$, which is an abelian Lie algebra \plainrefs{fuchs2003symmetries}. The algebra $\gh$ must itself be realizable in $\hat{A}_1$. By Proposition~\ref{prop:realizable:semisimple:algebras}, one must have $\gh\cong\slR{2}$. Due to the discussion above, we can choose a particular faithful realization of $\slR{2}$. We choose $\gh=\langle\{i(a^\dagg a+\frac{1}{2}),i(a^2+(a^\dagg)^2,a^2-(a^\dagg)^2\}\rangle_{\mathrm{Lie}}$. Let $z\in\mathcal{Z}(\g)$ an arbitrary polynomial of degree $d$. We can write
		\begin{align*}
			z=\sum_{m=1}^{d}\sum_{j=0}^{\lceil\frac{m-1}{2}\rceil}c_{mj}g_+^{(m-j)\iota_1+j\iota_2}+\sum_{m=1}^{d}\sum_{j=0}^{\lfloor\frac{m-1}{2}\rfloor}\hat{c}_{mj}g_-^{(m-j)\iota_1+j\iota_2}+c_{00}i.
		\end{align*}
		We now apply Proposition~\ref{prop:thm:25:equivalent} and the condition that $\mathcal{Z}(\g)$ is the center of $\g$, i.e., $[x,z]=0$ for all $x\in\gh$, to find 
		\begin{align*}
			0=\left[i\left(a^\dagg a+\frac{1}{2}\right),z\right]=\sum_{m=1}^{d}\sum_{j=0}^{\lceil\frac{m-1}{2}\rceil-1}c_{mj}(m-2j)g_-^{(m-j)\iota_1+j\iota_2}+\sum_{m=1}^{d}\sum_{j=0}^{\lfloor\frac{m-1}{2}\rfloor}\hat{c}_{mj}(m-2j)g_+^{(m-j)\iota_1+j\iota_2}.
		\end{align*}
		Since the non-vanishing monomials appearing in the expansion above are linearly independent, one must have: (i) $\hat{c}_{jm}=0$ for all $j,m$, and (ii) $c_{jm}=0$ for all $j,m$ unless $m=2j$. Furthermore, $d$ must be even, as otherwise all coefficients $c_{dj}$ vanish, implying $\deg(z)<d$. This allows us to write
		\begin{align*}
			z=\sum_{j=0}^{d/2}c_jg_+^{j\tau}\in\hat{A}_1^0\oplus\hat{A}_1^=.
		\end{align*}
		We can now similarly compute $[a^2-(a^\dagg)^2,z]=0$ and apply Proposition~\ref{prop:thm:25:equivalent} to conclude that $z$ must be a scalar multiple of $i$. Thus, by example~\ref{exa:semidirect:product:direct:sum:derivation:trivial}, one has $\g\cong\slR{2}\oplus\R$ or $\g\cong \slR{2}$.
	\end{proof}
	
	\begin{definition}
		Let $\g$ and $\gh$ be two Lie algebras with faithful realizations $\Phi$ and $\Psi$ in $\hat{A}_1$. We say $\gh$ is \emph{stable under the actions} $[x,\cdot]$ for every $x\in \g$, if and only if $[\Phi(x),\Psi(y)]\in\Psi(\gh)$ for every $x\in\g$ and $y\in\gh$.
	\end{definition}

	\begin{proposition}\label{prop:stable:solvable:Lie:algebras}
		Let $\g$ be a finite-dimensional Lie algebra that can be faithfully realized in the skew-hermitian Weyl algebra $\hat{A}_1$ and is stable under the actions $[x,\cdot]$ for every $x\in\slR{2}$. Then $\g$ is isomorphic to  a subalgebra of $\langle\hat{A}_1^0\oplus\hat{A}_1^1\oplus\hat{A}_1^2\rangle_{\mathrm{Lie}}$. Furthermore, if $\g$ is solvable then it is isomorphic to any of the three algebras: $\{0\}$, $\R$, or the Heisenberg algebra $\gh_1$.
	\end{proposition}
	
	\begin{proof}
		Let $\g$ be as stated in the claim. We begin by choosing, without loss of generality, the faithful realization of $\slR{2}$ in $\hat{A}_1$, given by $\slR{2}\cong\langle\{i(a^\dagg a+\frac{1}{2}),i(a^2+(a^\dagg)^2),a^2-(a^\dagg)^2\}\rangle_{\mathrm{Lie}}$. This is possible, due to the comment below Proposition~\ref{prop:complexification:Invariant:properties}. We omit now to denote the faithful realization of $\g$ with a respective isomorphism and employ $\g$ for its realization  in order to evade a cumbersome notation. Since $\g$ is finite dimensional and stable under actions $[e,\cdot]$ for all $e\in\langle\{i(a^\dagg a+\frac{1}{2}),i(a^2+(a^\dagg)^2),a^2-(a^\dagg)^2\}\rangle_{\mathrm{Lie}}$, the Lie algebra $\langle\{i(a^\dagg a+{1}/{2}),x\rangle_{\mathrm{Lie}}$ must also be a finite-dimensional subalgebra of $\g$ for every $x\in\g$. By virtue of Theorem~\ref{thm:core:result} and Corollary~\ref{cor:all:physically:relevant:finite:algebras}, $\g$ must be isomorphic to a subalgebra of the Lie algebra $\langle\hat{A}_1^0\oplus\hat{A}_1^1\oplus\hat{A}_1^2\rangle_{\mathrm{Lie}}\cong\slR{2}\ltimes\gh_1$. This already shows the first part of this claim. 
		
		Now, let us proceed with the second one and assume that $\g$ is solvable. The idea behind the proof for this claim is to first show that $\g$ cannot have any nontrivial support in $\lie{\hat{A}_1^2}$. After establishing this, we simply investigate the remaining possibilities. We begin by supposing that $\g$ is solvable and has non-trivial support in $\langle\hat{A}_1^2\rangle_{\mathrm{Lie}}$. Then, Corollary~\ref{cor:decomposition:a0:a1:a2:space} implies that there exists an element $x\in\g$ such that $x=y+z$, where $y\in\langle\hat{A}_1^1\rangle_{\mathrm{Lie}}$ and $z\in\langle\hat{A}_1^2\rangle_{\mathrm{Lie}}$ is nonzero. One can now write $z=2i(a^\dagg a+1/2)c_1+g_+^{2\iota_1}c_2+g_-^{2\iota_1}c_3$, $y=ic_4+g_+^{\iota_1}c_5+g_-^{\iota_1}c_6$ and compute the commutators of $y$ and $z$ with the three basis elements in $\slR{2}$ realized within $\hat{A}_1$, given by $\slR{2}\cong\langle\{i(a^\dagg a+\frac{1}{2}),i(a^2+(a^\dagg)^2),a^2-(a^\dagg)^2\}\rangle_{\mathrm{Lie}}$. For the commutators involving $z$, one finds{\small
			\begin{align*}
				\left[z,2i\left(a^\dagg a+\frac{1}{2}\right)\right]&=\left[2i\left(a^\dagg a+\frac{1}{2}\right)c_1+i\left(a^2+(a^\dagg)^2\right)c_2+\left(a^2-(a^\dagg)^2\right)c_3,2i\left(a^\dagg a+\frac{1}{2}\right)\right]=4\left(i\left(a^2+(a^\dagg)^2\right)c_3-\left(a^2-(a^\dagg)^2\right)c_2\right),\\
				\left[z,i\left(a^2+(a^\dagg)^2\right)\right]&=\left[2i\left(a^\dagg a+\frac{1}{2}\right)c_1+i\left(a^2+(a^\dagg)^2\right)c_2+\left(a^2-(a^\dagg)^2\right)c_3,i\left(a^2+(a^\dagg)^2\right)\right]=4\left(2i\left( a^\dagg a+\frac{1}{2}\right)c_3+\left(a^2-(a^\dagg)^2\right)c_1\right),\\
				\left[z,\left(a^2-(a^\dagg)^2\right)\right]&=\left[2i\left(a^\dagg a+\frac{1}{2}\right)c_1+i\left(a^2+(a^\dagg)^2\right)c_2+\left(a^2-(a^\dagg)^2\right)c_3,\left(a^2-(a^\dagg)^2\right)\right]=-4\left(2i\left(a^\dagg a+\frac{1}{2}\right)c_2+i\left(a^2+(a^\dagg)^2\right)c_1\right).
		\end{align*}}\noindent
		For those involving $y$, one computes furthermore:
		\begin{align*}
			\left[y,2i\left(a^\dagg a+\frac{1}{2}\right)\right]&=\left[i c_4+ i(a+a^\dagg)c_5+(a-a^\dagg)c_6,2i\left(a^\dagg a+\frac{1}{2}\right)\right]=2\left(i(a+a^\dagg)c_6-(a-a^\dagg)c_5\right),\\
			\left[y,i\left(a^2+(a^\dagg)^2\right)\right]&=\left[i c_4+ i(a+a^\dagg)c_5+(a-a^\dagg)c_6,i\left(a^2+(a^\dagg)^2\right)\right]=2\left(i(a+a^\dagg)c_6+(a-a^\dagg)c_5\right),\\
			\left[y,\left(a^2-(a^\dagg)^2\right)\right]&=\left[i c_4+ i(a+a^\dagg)c_5+(a-a^\dagg)c_6,\left(a^2-(a^\dagg)^2\right)\right]=-2\left(i(a+a^\dagg)c_5-(a-a^\dagg)c_6\right).
		\end{align*}
		The condition that $\g$ is stable under the actions $[x,\cdot]$ for every $x\in \langle\{i(a^\dagg a+\frac{1}{2}),i(a^2+(a^\dagg)^2,a^2-(a^\dagg)^2\}\rangle_{\mathrm{Lie}}$ implies that the space spanned by the elements $x$ and $[x,e]$, for $e\in \{2i(a^\dagg a+\frac{1}{2}),i(a^2+(a^\dagg)^2),a^2-(a^\dagg)^2\}$, is a subspace of $\g$. That is, combining the previous results, one must have
		\begin{align*}
			V:=\spn\Biggl\{&x,g_+^{2\iota_1}c_3-g_-^{2\iota_1}c_2+\frac{1}{2}(g_+^{\iota_1}c_6-g_-^{\iota_1}c_5),2i\left(a^\dagg a+\frac{1}{2}\right)c_3+g_-^{2\iota_1}c_1+\frac{1}{2}(g_+^{\iota_1}c_6+g_-^{\iota_1}c_5),\\
			&2i\left(a^\dagg a+\frac{1}{2}\right)c_2+g_+^{2\iota_1}c_1+\frac{1}{2}(g_+^{\iota_1}c_5-g_-^{\iota_1}c_6)\Biggr\}\subseteq\g
		\end{align*}
		It is easy to confirm that this vector space, under the condition $(c_1,c_2,c_3)\in\R^3\setminus\{0\}$, is generally at least three-dimensional, except if $c_1\neq 0$ and $c_1^2=c_2^2+c_3^2$, in which case it is at least two-dimensional. To verify this, one may identify
		\begin{align*}
			\begin{pmatrix}
				1\\
				0\\
				0\\
				0\\
				0\\
				0
			\end{pmatrix}&\leftrightarrow 2i\left(a^\dagg a+\frac{1}{2}\right),\;&\;\begin{pmatrix}
				0\\
				1\\
				0\\
				0\\
				0\\
				0
			\end{pmatrix}&\leftrightarrow g_+^{2\iota_1},\;&\;\begin{pmatrix}
				0\\
				0\\
				1\\
				0\\
				0\\
				0
			\end{pmatrix}&\leftrightarrow g_-^{2\iota_1},\;&\;\begin{pmatrix}
				0\\
				0\\
				0\\
				1\\
				0\\
				0
			\end{pmatrix}&\leftrightarrow i,\;&\,\begin{pmatrix}
				0\\
				0\\
				0\\
				0\\
				1\\
				0
			\end{pmatrix}&\leftrightarrow g_+^{\iota_1},\;&\;\begin{pmatrix}
				0\\
				0\\
				0\\
				0\\
				0\\
				1
			\end{pmatrix}&\leftrightarrow g_-^{\iota_1},
		\end{align*}
		and compute the rank of the matrix
		\begin{align*}
			\boldsymbol{M}:=\begin{pmatrix}
				c_1 & 0 & c_3 & c_2 \\
				c_2 & c_3 & 0 & c_1\\
				c_3 & -c_2 & c_1 & 0\\
				c_4 & 0 & 0 & 0\\
				c_5 & \frac{1}{2}c_6 & \frac{1}{2}c_6 & \frac{1}{2}c_5\\
				c_6 & -\frac{1}{2}c_5 & \frac{1}{2}c_5 & \frac{1}{2}c_5
			\end{pmatrix}
		\end{align*}
		under the constraint that at least one of the coefficients $c_1$, $c_2$, or $c_3$ does not vanish.
		We consider the two cases that $\dim(V)\geq 3$ and $\dim(V)\geq 2$ separately.
		
		\textbf{First case}---Here we have $\dim(V)\geq3$ since either $c_1= 0$ or $c_1^2\neq c_2^2+c_3^2$.
		In this case it is clear that, by appropriate linear combinations of $x$, $[x,2i(a^\dagg a+1/2)]$, $[x,g_+^{2\iota_1}]$, and $[x,g_-^{2\iota_1}]$, the three linearly independent elements $\{2i(a^\dagg a+{1}/{2})+r_1,i(a^2+(a^\dagg)^2)+r_2,a^2-(a^\dagg)^2+r_3\}$ must belong to $\g$, where $r_1,r_2,r_3\in\langle\hat{A}_1^1\rangle_{\mathrm{Lie}}$. We now want to demonstrate that $\g$ is not solvable in this case, thereby contradicting the initial assumption. To achieve this, we compute terms from the first derived algebra $[\g,\g]$. Application of Table~\ref{tab:Full:Commutator:Algebra:(pesudo):schroedinger:algebra} allows us to conclude that the following relations hold:
		\begin{align*}
			\left[2i\left(a^\dagg a+\frac{1}{2}\right)+r_1,i\left(a^2+(a^\dagg)^2\right)+r_2\right]&=4\left(a^2-(a^\dagg)^2\right)+\tilde{r}_3,\\
			\left[2i\left(a^\dagg a+\frac{1}{2}\right)+r_1,\left(a^2-(a^\dagg)^2\right)+r_3\right]&=-4i\left(a^2+(a^\dagg)^2\right)+\tilde{r}_2,\\
			\left[i\left(a^2+(a^\dagg)^2\right)+r_2,\left(a^2-(a^\dagg)^2\right)+r_3\right]&=-8i\left(a^\dagg a+\frac{1}{2}\right)+\tilde{r}_1,
		\end{align*}
		where $\tilde{r}_1,\tilde{r}_2,\tilde{r}_3\in\langle\hat{A}_1^1\rangle_{\mathrm{Lie}}$. Now recall that the derived series $\mathcal{D}^\ell\g$ is defined by a recursive construction \plainrefs{Knapp:1996}, which we present here for ease of presentation:
		\begin{align*}
			\mathcal{D}^{\ell+1}\g:=[\mathcal{D}^\ell\g,\mathcal{D}^\ell\g]\qquad\text{for all }\ell\in\N_{\geq0},\qquad\text{where}\qquad \mathcal{D}^0\g:=\g.
		\end{align*}
		Now note that we have just demonstrated that for every $\ell\in \N_{\geq0}$ three linearly independent elements of the form $2i(a^\dagg a+{1}/{2})+r_1^{(\ell)}$, $i(a^2+(a^\dagg)^2)+r_2^{(\ell)}$, and $a^2-(a^\dagg)^2+r_3^{(\ell)}$ lie in each $\mathcal{D}^\ell\g$, where $r_1^{(\ell)},r_2^{(\ell)},r_3^{(\ell)}$ are some appropriate terms form $\lie{\hat{A}_1^1}$.
		Thus, $\g$ is not solvable, since the derived series can never terminate.
		
		\textbf{Second case}---Here we have $\dim(V)\geq2$. Recall that this is only possible if $c_1\neq 0$ and $c_1^2=c_2^2+c_3^2$. This can be combined to the condition $c_1^2=c_2^2+c_3^2\neq0$. By the stability condition, $\g$ contains the two distinct elements $z_1:=2i(a^\dagg a+1/2)c_3+(a^2-(a^\dagg)^2)c_1+r_1=[z,i(a^2+(a^\dagg)^2)]/4$, $z_2:= 2i(a^\dagg a+1/2)c_2+i(a^2+(a^\dagg)^2)c_1+r_2=[(a^2-(a^\dagg)^2),z]/4$, where $r_1$ and $r_2$ are suitable elements from $\langle\hat{A}_1^1\rangle_{\mathrm{Lie}}$. It is straightforward to verify that $z_1$ and $z_2$ are linearly independent if $c_1^2=c_2^2+c_3^2\neq 0$. One can apply the stability condition once more. To be more precise, since $\g$ is stable under the action $[e,\cdot]$ for every $e\in \langle\{i(a^\dagg a+\frac{1}{2}),i(a^2+(a^\dagg)^2),a^2-(a^\dagg)^2\}\rangle_{\mathrm{Lie}}$, it must also contain the two elements $[z_1, 2i(a^\dagg a+1/2)]$ and $[z_2,2i(a^\dagg a+1/2)]$. By the antisymmetry of the commutator it is enough to compute: 
		\begin{align*}
			\left[2i\left(a^\dagg a+\frac{1}{2}\right)c_3+g_-^{2\iota_1}c_1+r_1,2i\left(a^\dagg a+\frac{1}{2}\right)\right]&=4 g_+^{2\iota_1}c_1+\tilde{r}_1,\\
			\left[2i\left(a^\dagg a+\frac{1}{2}\right),2i\left(a^\dagg a+\frac{1}{2}\right)c_2+g_+^{2\iota_1}c_1+r_2\right]&=4 g_-^{2\iota_1}c_1+\tilde{r}_2,
		\end{align*}
		where $\tilde{r}_1,\tilde{r}_2\in\langle\hat{A}_1^1\rangle_{\mathrm{Lie}}$. Since $c_1\neq 0$, one has $i(a^2+(a^\dagg)^2)+\hat{r}_2,a^2-(a^\dagg)^2+\hat{r}_3\in\g$, for appropriate $\hat{r}_2,\hat{r}_3\in\langle\hat{A}_1^1\rangle_{\mathrm{Lie}}$. Employing Table~\ref{tab:Full:Commutator:Algebra:(pesudo):schroedinger:algebra} and computing the commutator of $[z_1, 2i(a^\dagg a+1/2)]$ and $[z_2,2i(a^\dagg a+1/2)]$, one can conclude that $\g$ contains the three linearly independent elements $2i(a^\dagg a+1/2)+r_1',i(a^2+(a^\dagg)^2)+r_2',(a^2-(a^\dagg)^2)+r_3'$ for suitable terms $r_1',r_2',r_3'\in\lie{\hat{A}_1^1}$. Hence, $\g$ is not solvable by the same argument of the previous case.
		
		This discussion shows that $\g$ can only be solvable and stable under the actions $[e,\cdot]$ for every $e\in \slR{2}\cong\langle\{i(a^\dagg a+\frac{1}{2}),i(a^2+(a^\dagg)^2),a^2-(a^\dagg)^2\}\rangle_{\mathrm{Lie}}$ if it belongs to $\lie{\hat{A}_1^{\leq2}}$ but has only trivial support in $\lie{\hat{A}_1^2}$. Hence, $\g$ must be a subalgebra of $\langle\{i,i(a+a^\dagg),(a-a^\dagg)\}\rangle_{\mathrm{Lie}}=\lie{\hat{A}_1^1}\cong\gh_1$. Suppose $\g\subseteq\lie{\hat{A}_1^1}$ is two-dimensional. From Proposition~\ref{prop:center:of:non:abelian:subalgebras} it follows then that $\g$ contains the non-vanishing element $x=2i c_1+g_+^{\iota_1}c_2+g_-^{\iota_1}c_3\in\g$ with $(c_2,c_3)\neq 0$. The commutators of $x$ with the three basis elements $\{2i(a^\dagg a+\frac{1}{2}),i(a^2+(a^\dagg)^2),a^2-(a^\dagg)^2\}$ read:
		\begin{align*}
			\left[x,2i\left(a^\dagg a+\frac{1}{2}\right)\right]&=\left[2i c_1+g_+^{\iota_1}c_2+g_-^{\iota_1}c_3,2i\left(a^\dagg a+\frac{1}{2}\right)\right]=2\left(g_+^{\iota_1}c_3-g_-^{\iota_1}c_2\right),\\
			\left[x,i\left(a^2+(a^\dagg)^2\right)\right]&=\left[2i c_1+g_+^{\iota_1}c_2+g_-^{\iota_1}c_3,i\left(a^2+(a^\dagg)^2\right)\right]=2\left(g_+^{\iota_1}c_3+g_-^{\iota_1}c_2\right),\\
			\left[x,\left(a^2-(a^\dagg)^2\right)\right]&=\left[2i c_1+g_+^{\iota_1}c_2+g_-^{\iota_1}c_3,\left(a^2-(a^\dagg)^2\right)\right]=-2\left(g_+^{\iota_1}c_2-g_-^{\iota_1}c_3\right).
		\end{align*}
		Since $\g$ is stable under the actions $[x,\cdot]$ for every $x\in\langle\{i(a^\dagg a+\frac{1}{2}),i(a^2+(a^\dagg)^2,a^2-(a^\dagg)^2\}\rangle_{\mathrm{Lie}}$, this implies that the vector space spanned by $x$ and the three elements above must be a subspace of $\g$. That is
		\begin{align*}
			W:=\spn\left\{2i c_1+g_+^{\iota_1}c_2+g_-^{\iota_1}c_3,g_+^{\iota_1}c_3-g_-^{\iota_1}c_2, g_+^{\iota_1}c_3+g_-^{\iota_1}c_2,g_+^{\iota_1}c_2-g_-^{\iota_1}c_3\right\}\subseteq\g.
		\end{align*}
		The assumption that $\g$ is two dimensional, implies that $W$ is also at most two-dimensional. By the same unit vector identification and rank-computing strategy used multiple times in this work, it is easy to see that $W$ is one-dimensional if $c_1\neq0$ and $c_2=c_3=0$, two-dimensional if $c_1=0$ and either $c_2\neq0$ or $c_3\neq0$, and three-dimensional in every other case. Note that the condition $x\neq0$ assures that $c_1=c_2=c_3=0$ is prohibited. We can therefore focus on the two cases in which $W$ is one- or two-dimensional.
		
		In the first case we consider the case that $c_1=0$ and $(c_2,c_3)\neq 0$. That is the case $\dim(W)=2$, where it is easy to verify that $W$ has the basis $\mathcal{B}=\{g_+^{\iota_1},g_-^{\iota_1}\}$. However, one can compute $\langle \mathcal{B}\rangle_{\mathrm{Lie}}$ utilizing Table~\ref{tab:Full:Commutator:Algebra:(pesudo):schroedinger:algebra}, and therefore see that $\dim(\langle\mathcal{B}\rangle_{\mathrm{Lie}})=\dim(\gh_1)=3$, which is impossible because $\langle \mathcal{B}\rangle_{\mathrm{Lie}}$ must be a subalgebra of the Lie algebra $\g$ (which is assumed to be two-dimensional).		
		
		In the second case we consider $c_1\neq 0$ and $c_2=0=c_3$. We have shown that this also implies that $\dim(W)=1$ and $W=i\R$.
		We assumed that $\g$ is two-dimensional. Thus, there exists at least one element $y\in\g$ with support in $\spn\{g_+^{\iota_1},g_-^{\iota_1}\}$. Repeating the procedure with $y$ instead of $x$ will consequently lead either to the case $\dim(W)=3$, which is forbidden by the assumption $\dim(\g)=2$ and $W\subseteq \g$, or the previous case.
		
		This shows that $\g$ cannot be two-dimensional and we are left with considering the case that $\g$ is either zero, one, or three dimensional. By the same argument as presented before, we have that if $\g$ is one-dimensional, it must follow that $\g=i\R\cong\R$ and if $\g$ is three dimensional, it must be $\g=\lie{\hat{A}_1^1}\cong\mathfrak{h}_1$. This concludes this proof. 
	\end{proof}
	
	With these structural constraints now established - including the role of semisimple components and the behavior of solvable radicals - we classify all finite-dimensional non-solvable Lie algebras that admit faithful realizations in $\hat{A}_1$.

	\begin{theorem}\label{thm:nonsolvable:algebras:realizable}
		Let $\g$ be a finite-dimensional non-solvable Lie algebra. Then $\g$ can be faithfully realized as a subalgebra of the skew-hermitian Weyl algebra $\hat{A}_1$ if and only if $\g$ is isomorphic to one the the following three algebras: (i)~$\slR{2}$, (ii)~$\slR{2}\oplus\R$, and (iii) $\slR{2}\ltimes_\pi\gh_1$, for some non-trivial suitable representation $\pi:\slR{2}\to\operatorname{End}(\gh_1)$. 
	\end{theorem}
	
	\begin{proof}
		The Levi-Mal'tsev theorem guarantees that any finite-dimensional Lie algebra $\g$ can be decomposed into a semidirect sum of a semisimple Lie algebra $\g_1$ and its radical $\gr:=\operatorname{rad}(\g)$ for a suitable representation $\pi:\g_1\to\operatorname{End}(\gr)$ \plainrefs{Kuzmin:1977}. According to Proposition~\ref{prop:realizable:semisimple:algebras}, the only semisimple Lie algebra that can be faithfully realized in $\hat{A}_1$ is $\slR{2}$. Furthermore, Proposition~\ref{prop:stable:solvable:Lie:algebras} shows that the only solvable Lie algebras that are stable under the actions $[x,\cdot]$ for every $x\in\slR{2}$ are $\{0\}$, $\R$, and the Heisenberg algebra $\gh_1$. Thus, the radical $\gr$ is one of these three algebras. One must now have that $\g$ can be realized in $\hat{A}_1$ if and only if $\g$ is isomorphic to one of the three following algebras: (i)~$\slR{2}\ltimes_\pi\{0\}$, (ii) $\slR{2}\ltimes_\pi\R$, or (iii) $\slR{2}\ltimes_\pi\gh_1$, for a suitable representation $\pi:\slR{2}\to\operatorname{End}(\gr)$. Consider these three cases separately:
		\begin{enumerate}[label = (\roman*)]
			\item $\boldsymbol{\g=\slR{2}\ltimes_\pi\{0\}}$: Since $\pi$ is a representation on the zero vector space it must be trivial, and the semidirect sum reduces to a direct sum as shown in Example~\ref{exa:semidirect:product:direct:sum:derivation:trivial}. Therefore, as shown in Proposition~\ref{prop:algebra:g2:space}, $\slR{2}\cong\slR{2}\oplus\{0\}$ can be faithfully realized in $\hat{A}_1$, for example via the subalgebra $\langle\hat{A}_1^2\rangle_{\mathrm{Lie}}$.
			\item $\boldsymbol{\g=\slR{2}\ltimes_\pi\R}$: Since $\R$ is a one-dimensional field, the linear endomorphism $\pi(u)\in\operatorname{End}(\R)$ must act as a scalar multiplication for every $u\in\slR{2}$. In other words, one must have $\pi(u)(v)=v\pi(u)(1)=\pi(u)(1) v\in\R$ for all $u\in\slR{2}$ and $v\in\R$, where $\pi(u)(1)\in\R$. Thus, one has by Definition~\ref{def:representation} that $\pi([u,v])(w)=[\pi(u),\pi(v)](w)=\pi(u)(\pi(v)(w))-\pi(v)(\pi(u)(w))=\pi(u)(w\pi(v)(1))-\pi(v)(w\pi(u)(1))=w\pi(v)(1)\pi(u)(1)-w\pi(u)(1)\pi(v)(1)=0$ for every $u,v\in\slR{2}$ and $w\in\R$, since $z,\pi(x)(1),\pi(y)(1)$ are real numbers that commute. We can now choose the standard basis $\{h,x,y\}$ of $\slR{2}$ and find $\pi(h)(z)=\pi([x,y])(z)=0$, $\pi(x)(z)=\pi([h,\frac{1}{2}x])(z)=0$, $\pi(y)(z)=\pi([\frac{1}{2}y,h])(z)=0$ for every $z\in\R$. Hence, by linearity, $\pi(g)(z)=0$ for every $g\in\slR{2}$ and $z\in\R$, thereby showing that the representation $\pi$ is trivial. The semidirect sum reduces therefore to a direct sum and one has $\slR{2}\ltimes_\pi\R=\slR{2}\oplus\R$.
			Therefore, as shown in Proposition~\ref{prop:algebra:g0:g2:space} this algebra can be faithfully realized in $\hat{A}_1$, for example as $\langle\hat{A}_1^0\oplus\hat{A}_1^2\rangle_{\mathrm{Lie}}$.
			\item $\boldsymbol{\g=\slR{2}\ltimes_\pi\gh_1}$: Assume $\pi$ is the trivial representation, i.e., $\pi(u)(v)=0$ for every $u\in \slR{2}$ and $v\in\gh_1$. We'll now use the standard basis $\{h,x,y\}$ of $\slR{2}$ and $\{p,q,z\}$ of $\gh_1$ and the standard Lie bracket on $\slR{2}\ltimes_\pi\mathcal{H}$ introduced in Definition~\ref{def:semidirect:product}. One must now have $(0,q),(0,p)\in C_{\slR{2}\ltimes_\pi\gh_1}((h,0))$ since $[(h,0),(0,g)]=([h,0],[0,g]+\pi(h)(g)-\pi(0)(0))=(0,0)$ for every $g\in\gh_1$ if $\pi$ is trivial. But then $[(0,q),(0,p)]=(0,[q,p])=(0,z)\neq0$. This is prohibited by the abelian centralizer condition, see Theorem~\ref{thm:commutative:centralizer:condition}. Therefore, as shown in Proposition~\ref{prop:algebra:g1:g2:space}, Theorem~\ref{thm:core:result}, and Proposition~\ref{prop:schroedinger:algebra}, the Schr\"odinger algebra $\mathcal{S}=\slR{2}\ltimes\gh_1$ can be faithfully realized in $\hat{A}_1$, for example via the subalgebra $\lie{\hat{A}_1^1\oplus\hat{A}_1^2}$. In this example, the representation is not trivial, but it is given by the commutator $\pi\equiv [\,\cdot\,,\,\cdot\,]$. 
		\end{enumerate}
		This concludes the proof.
	\end{proof}
	
	Notably, in Theorem~\ref{thm:nonsolvable:algebras:realizable}, it is shown that the Lie algebra $\slR{2}\ltimes_\pi\mathfrak{h}_1$ can be faithfully realized in $\hat{A}_1$. However, the only explicit constraint provided for the representation $\pi$, which determines the structure of the semidirect sum, is that it is nontrivial; no further conditions are specified.
	
	\begin{corollary}
		Every non-solvable Lie algebra $\g\subseteq\hat{A}_1$ that contains a free term $i(a^\dagg a+c)$ is either three-, four-, or six-dimensional.
	\end{corollary}
	
	\begin{proof}
		This is a simple consequence of Corollary~\ref{cor:all:physically:relevant:finite:algebras}, Theorem~\ref{thm:all:physically:relevant:algebras:list} and Theorem~\ref{thm:nonsolvable:algebras:realizable}.
	\end{proof}
	
	\begin{conjecture}
		The Schr\"odinger algebra $\mathcal{S}=\slR{2}\ltimes\gh_1$ admits a unique faithful realization as a subalgebra of the skew-hermitian Weyl algebra $\hat{A}_1$, up to automorphism. This realization is the one explicitly constructed in Proposition~\ref{prop:schroedinger:algebra}, no other subspace of $\hat{A}_1$ is isomorphic to $\mathcal{S}$, when equipped with the canonical commutation relation $[a,a^\dagg]=1$.
	\end{conjecture}

	\section{Considerations\label{section:considerations}}
	We now provide final considerations regarding important questions that are left open, the implications for quantum dynamics, and the potential future scope of our work.

	\subsection{Non-nilpotent solvable subalgebras of the skew-hermitian Weyl algebra}
	In the previous classification of finite-dimensional subalgebras of $\hat{A}_1$ we did not provide sufficient results to classify non-nilpotent solvable Lie algebras within $\hat{A}_1$. In Section~\ref{section:scope:of:the:work} we showed that there exist at least two classes of non-nilpotent solvable Lie algebras that can be faithfully realized in $\hat{A}_1$. One family is denoted by $\tilde{\mathcal{L}}_n$ for $n\geq 2$, and it has the property that its first derived algebra is non-abelian. Notably, $\tilde{\mathcal{L}}_2\cong\wh_1$. The other family is denoted by $\gr(j_1,\ldots,j_n)$ with integers $0\leq j_1<\ldots< j_n$ and $n\geq 1$, and the first derived algebra is, however, abelian. Possible realizations include:
	\begin{align*}
		\tilde{\mathcal{L}}_n(n\geq 2)&\cong\left\langle\left\{\frac{1}{2}g_-^{2\iota_1},-\frac{1}{2}g_+^{\iota_1},i,g_-^{\iota_1},i\frac{1}{2!}(g_-^{\iota_1})^2,\frac{1}{3!}(g_-^{\iota_1})^3,\ldots,i^{n\mod 2}\frac{1}{(n-1)!}(g_-^{\iota_1})^{n-1}\right\}\right\rangle_{\textrm{Lie}},\\
		\gr(j_1,\ldots,j_n)&\cong\left\langle\left\{-\frac{1}{2}g_-^{2\iota_1},i^{j_{1-1}\mod2}(g_+^{\iota_1})^{j_1},\ldots,i^{j_{n-1}\mod 2}(g_+^{\iota_1})^{j_n}\right\}\right\rangle_{\textrm{Lie}}.
	\end{align*}
	However, we have also also seen that the Wigner-Heisenberg algebra $\wh_2$ is not isomorphic to the Wigner-Heisenberg algebra $\wh_1$. Since $[\wh_2,\wh_2]\cong\gh_1$, which is not abelian, it follows that $\wh_2$ is not isomorphic to any $\gr(j_1,j_2,j_3)$. Thus, there exists at least one finite-dimensional, non-nilpotent, and solvable Lie subalgebra of $\hat{A}_1$ that does not belong to either of the classes $\tilde{\mathcal{L}}_n$ or $\gr(j_1,\ldots,j_n)$. A natural question now arises, which we have not answered:
	\begin{quote}
		\textbf{Q}: \emph{Is $\wh_2$ the only additional finite-dimensional, non-nilpotent, and solvable Lie subalgebra of $\hat{A}_1$ that does not belong to either of the classes $\tilde{\mathcal{L}}_n$ or $\gr(j_1,\ldots,j_n)$? If there are more of such algebras, are they sporadic, or do they form an infinite family?}
	\end{quote}
	In the case of three-dimensional or two-dimensional Lie subalgebras that are solvable but not nilpotent it is straightforward to verify that those algebras possess an abelian first derived algebra, since $\dim_\R(\g)\geq \dim_\C(\g_\C)$ and every three-dimensional or two-dimensional non-nilpotent solvable Lie subalgebra of the complex Weyl algebra admits an abelian first derived algebra. We were furthermore unable to construct any additional four-dimensional solvable Lie algebras beyond those already discussed, nor were we able to identify such algebras in higher-dimensional cases. This leads us to propose the following conjecture:
	
	\begin{conjecture}
		The only finite-dimensional, non-nilpotent, and solvable Lie subalgebra of the skew-hermitian Weyl algebra $\hat{A}_1$ that is not covered by the two families $\tilde{\mathcal{L}}_n$ or $\gr(j_1,\ldots,j_n)$ is the sporadic Wigner-Heisenberg algebra $\wh_2$.
	\end{conjecture}
	\noindent
	We leave it to future work to address the open questions put forward above.

	\subsection{Arbitrary linear combinations of single monomials: the Igusa condition}
	Determining all finite-dimensional, non-nilpotent, solvable Lie subalgebras of $\hat{A}_1$ is a major quest. An aid to such endeavour is to identify a list of sufficient or necessary conditions under which two elements from $\hat{A}_1$ generate an infinite-dimensional Lie algebra. A few dacades ago, Igusa proved a result that yields such a sufficient condition for two elements of the real or complex Weyl algebra of $n$-modes to generate an infinite-dimensional Lie algebra \plainrefs{Igusa1981}. The proof of that result can be adapted with minor modifications for the skew-hermitian Weyl algebra, which is the case of interest here. Details on this extension can be found in Appendix~\ref{App:Igusa}. Here we provide a summary of the procedure to obtain the aforementioned extension. The idea is simple: one first determines a condition or criterion for two elements to generate a generic Commutator Chain that contains elements of strictly increasing degree. Then, one generalizes this conditions by considering all elements that can be transformed into the desired from using a suitable Lie-algebra isomorphism. This leads to the following result:
	
	\begin{theorem}[Igusa condition]\label{Igusa:condition}
		Let $e_1,e_2\in\hat{A}_1$ with degrees $d_1:=\deg(e_1)>2$ and $d_2:=\deg(e_2)>2$, where the two elements are both of the form:
		\begin{align*}
			e_j&\upto{d_j}\sum_{k=0}^{\lceil\frac{d_j-1}{2}\rceil}c_k^{(j)}g_+^{(d_j-k)\iota_1+k\iota_2}+\sum_{k=0}^{\lfloor\frac{d_j-1}{2}\rfloor}\hat{c}_k^{(j)}g_-^{(d_j-k)\iota_k+k\iota_2}.
		\end{align*}
		Then, the Lie algebra $\lie{\{e_1,e_2\}}$ is infinite dimensional if there exist two scalars $\lambda_1,\,\lambda_2$, and a complex symplectic matrix $\boldsymbol{\sigma}\in\operatorname{Sp}(2,\C)$ such that the following two equations have at least one solution:
		\begin{align*}
			1-\lambda_1\alpha_0^{(1)}\alpha_0^{(2)}&=0,\;&\;1-\lambda_2\left(d_2 \alpha_1^{(1)}\alpha_0^{(2)}-d_1\alpha_0^{(1)}\alpha_1^{(0)}\right)&=0,
		\end{align*}
		where the coefficients are defined as:
		{\small
			\begin{align*}
				\alpha_0^{(j)}&:=\sum_{k=0}^{d_j} f_k^{(j)}\sigma_{12}^k\sigma_{22}^{d_j-k},\;&\;\alpha_1^{(j)}&:=\sum_{k=0}^{d_j}\left(k \sigma_{11}\sigma_{12}^{k-1}\sigma_{22}^{d_j-k}+(d_j-k)\sigma_{21}\sigma_{12}^k\sigma_{22}^{d_j-k-1}\right),
				\quad
				\text{with}
				\quad
				f_k^{(j)}&:=\left\{\begin{matrix}
					ic_k^{(j)}+\hat{c}_k^{(j)}& \text{ if }k<d_j/2\\
					ic_k^{(j)}& \text{ if }k=d_j/2\\
					ic_k^{(j)}-\hat{c}_k^{(j)}& \text{ if }k>d_j/2
				\end{matrix}\right..
			\end{align*}
		}
	\end{theorem}
	
	Note that Theorem~\ref{Igusa:condition} is a reformulated version of the original Igusa condition. A version more in line with the original one is given by Theorem~\ref{thm:Igusa:1} in Appendix~\ref{App:Igusa}.
	
	\begin{proof}
		Let $e_1,e_2\in\hat{A}_1$ be as stated in the claim. Then, each element can be expressed as:
		\begin{align*}
			e_j\upto{d_j}\left(ic_0^{(j)}+\hat{c}_0^{(j)}\right)a^{d_j}+\left(ic_1^{(j)}+\hat{c}_1^{(j)}\right)a^\dagg a^{d_j-1}+\ldots+\left(ic_1^{(j)}-\hat{c}_1^{(j)}\right)(a^\dagg)^{d_j-1}a+\left(ic_0^{(j)}-\hat{c}_0^{(j)}\right)(a^\dagg)^{d_j}.
		\end{align*}
		For convenience, we define $\hat{c}_{d_j/2}^{(j)}:=0$ when $d_j$ is even. This allows us to write:
		\begin{align}\label{eqn:def:f:k:j:help:Igusa}
			e_j\upto{d_j}\sum_{k=0}^{d_j}f_k^{(j)}(a^\dagg)^k a^{d_j-k}\qquad\text{with}\qquad f_k^{(j)}:=\left\{\begin{matrix}
				ic_k^{(j)}+\hat{c}_k^{(j)}&,\text{ if }k\leq d_j/2\\
				ic_k^{(j)}-\hat{c}_k^{(j)}&,\text{ if }k>d_j/2
			\end{matrix}\right..
		\end{align}
		Example~\ref{exa:Igusa:lemma:applied} in Appendix~\ref{App:Igusa} informs us that two elements of the full algebra with the properties just listed above generate an infinite-dimensional Lie algebra if the set of requirements
		\begin{align*}
			\hat{c}_0^{(1)}\hat{c}_0^{(2)}&\neq c_0^{(1)}c_0^{(2)},\;&\;\hat{c}_0^{(2)}c_0^{(1)}&\neq -\hat{c}_0^{(1)}c_0^{(2)},\\
			d_1\left(\hat{c}_0^{(1)}\hat{c}_1^{(2)}-c_0^{(1)}c_1^{(2)}\right)&\neq d_2\left(\hat{c}_1^{(1)}\hat{c}_0^{(2)}-c_1^{(1)}c_0^{(2)}\right),\;&\;d_1\left(c_0^{(1)}\hat{c}_1^{(2)}+\hat{c}_0^{(1)}c_1^{(2)}\right)&\neq d_2\left(c_1^{(1)}\hat{c}_0^{(2)}+\hat{c}_1^{(1)}c_0^{(2)}\right)
		\end{align*}
		is met.
		The set of steps necessary to obtain these equation is long and not illuminating. Therefore, we refer to Appendix~\ref{App:Igusa} for a detailed derivation where the interested reader can also find the clarifying Example~\ref{exa:Igusa:lemma:applied}, which is a reformulation of Lemma~\ref{lem:Igusa:5} to the current notation. These conditions can be reformulated, using the coefficients $f_k^{(j)}$ defined in \eqref{eqn:def:f:k:j:help:Igusa}, as:
		\begin{align*}
			1-\lambda_1 f_0^{(1)}f_0^{(2)}&=0,\;&\;1-\lambda_2\left(d_2 f_1^{(1)} f_0^{(2)}-d_1 f_0^{(1)} f_1^{(2)}\right)&= 0
		\end{align*}
		for some scalars $\lambda_1,\lambda_2$. If these equations have a solution, then the Lie algebra 
		$\g:=\lie{\{e_1,e_2\}}$ is infinite dimensional (cf. Lemma~\ref{lem:Igusa:5}).  By virtue of Corollary~\ref{cor:extension:isomorphism:enveloping:algebra} any symplectic matrix $\boldsymbol{\sigma}\in\operatorname{Sp}(2,\C)$ induces a Lie-algebra isomorphism for any subalgebra of the Weyl algebra. Consider therefore the linear map $\phi:\spn\{a^\dagg,a,1\}\to\spn\{a^\dagg,a,1\}$ defined by the action $\phi(a^\dagg):=\sigma_{11}a^\dagg+\sigma_{12}a$, $\phi(a):=\sigma_{21}a^\dagg+\sigma_{22}a$, and $\phi(1):=1$.
		Denote the extended associative endomorphism by $\Phi$. For more details, see Lemmata~\ref{lem:extension:automorphism:enveloping:algebra} and ~\ref{lem:Igusa:0}, as well as Corollary~\ref{cor:extension:isomorphism:enveloping:algebra}. It follows that
		\begin{align*}
			\Phi\left((a^\dagg)^{\ell}a^{d_j-\ell}\right)&=\left(\sigma_{11} a^\dagg+\sigma_{12} a\right)^\ell\left(\sigma_{21}a^\dagg+\sigma_{22}a\right)^{d_j-\ell}\nonumber\\
			&\upto{d_j}\sigma_{12}^\ell\sigma_{22}^{d_j-\ell}a^{d_j}+\left(\ell\sigma_{11}\sigma_{12}^{\ell-1}\sigma_{22}^{d_j-\ell}+(d_j-\ell)\sigma_{21}\sigma_{12}^\ell \sigma_{22}^{d_j-\ell-1}\right) a^\dagg a^{d_j-1}+\ldots,
		\end{align*}
		where we omitted to write any term that is of the form $c(a^\dagg)^x a^y$ with $x\geq2$. Hence, by the linearity of $\Phi$ we have
		\begin{align*}
			\Phi(e_j)\upto{d_j}\sum_{k=0}^{d_j} f_k^{(j)} \sigma_{12}^k\sigma_{22}^{d_j-k}a^{d_j}+\sum_{k=0}^{d_j}\left(k\sigma_{11}\sigma_{12}^{k-1}\sigma_{22}^{d_j-k}+(d_j-k)\sigma_{21}\sigma_{12}^k \sigma_{22}^{d_j-k-1}\right) a^\dagg a^{d_j-1}+\ldots.
		\end{align*}
		Since $\Phi$ is an isomorphism into another (not necessarily skew-hermitian) real subalgebra of the complex Weyl algebra, the same conditions for infinite-dimensionality apply now to $\Phi(e_1)$ and $\Phi(e_2)$. This concludes this proof.
	\end{proof}
	
	We have previously established several sufficient conditions under which two elements from the skew-hermitian Weyl algebra $\hat{A}_1$ generate an infinite-dimensional Lie subalgebra, given for example in Lemma~\ref{lem:infinite:algebra:containg:support:perp:and:equal:elements}, Proposition~\ref{prop:algebra:g=:gperp:space}, or Proposition~\ref{prop:algebra:g1:g2:g=:space}. These results suggest that the same conceptual framework as the one employed within the Igusa condition can be applied here. Let us, for example, consider a simplified version of Lemma~\ref{lem:infinite:algebra:containg:support:perp:and:equal:elements}.
	
	\begin{corollary}
		Let $\mathcal{G}$ be a set containing two elements $e_1=\sum_{p\in\mathcal{P}}c_pg_{\sigma_p}^{\gamma_p}\in\mathcal{G}$ and $e_2=\sum_{q\in\mathcal{Q}}\hat{c}_qg_{\hat{\sigma}_q}^{\hat{\gamma}_q}\in\mathcal{G}$. Then the Lie algebra $\g=\langle\mathcal{G}\rangle_{\mathrm{Lie}}$ infinite dimensional, if the following structural conditions are satisfied:
		\begin{itemize}
			\item[$\boldsymbol{e_1}$] \textbf{structure}: the element $e_1$ can be decomposed as $e_1=c_1g_+^{\tilde{\gamma}}+e_1^{\neq}+e_1^=$, where $c_1\neq0$, $\deg(e_1)=|\tilde{\gamma}|>\deg(e_1^=+e_1^{\neq})$, $\PP{e_1^{\neq}}{\hat{A}_1^0\oplus\hat{A}_1^=}=0$, $e_1^=\in\hat{A}_1^0\oplus\hat{A}_1^=$, and $|\tilde{\gamma}|\geq 4$.
			\item[$\boldsymbol{e_2}$]\textbf{structure}: the element $e_2$ can be decomposed as $e_2=e_2^{\neq}+e_2^=$, where $\deg(e^\neq)\geq \deg(e^=)$, $e_2^{\neq}\neq0$, $\PP{e_2^{\neq}}{\hat{A}_1^0\oplus\hat{A}_1^=}=0$, and $e_2^=\in\hat{A}_1^0\oplus\hat{A}_1^=$.
		\end{itemize}
	\end{corollary}
	
	\begin{proof}
		This is a direct consequence of Lemma~\ref{lem:infinite:algebra:containg:support:perp:and:equal:elements}.
	\end{proof}
	
	These conditions can equivalently be reformulated as follows. Let $e_1,e_2\in\hat{A}_1$ be two polynomials satisfying:
	\begin{enumerate}
		\item $d_1:=\deg(e_1)=2d_0$ with $d_0\in\N_{\geq2}$, and $e_1\upto{d_1} c_1 g_+^{d_0\tau}$ for the first element $e_1$.
		\item $e_2\upto{d_2}\sum_{k=0}^{d_2} f_k^{(2)} (a^\dagg)^k a^{d_2-k}$, with $d_2:=\deg(e_2)>0$ for the second element $e_2$, where $f_k^{(2)}\neq 0$ for at least one $k\in\{0,\ldots,d_2\}\setminus\{d_2/2\}$.
	\end{enumerate}
	Then, the Lie algebra $\g:=\lie{\{e_1,e_2\}}$ is infinite dimensional. 
	
	We now proceed and apply the central idea behind the Igusa condition by considering any pair of elements $\tilde{e}_1,\tilde{e}_2\in\hat{A}_1$ such that there exists a sympletic matrix $\boldsymbol{\sigma}\in\operatorname{Sp}(2,\C)$ that induces an isomorphism $\Phi$ for which $\Phi(\tilde{e}_1)$ and $\Phi(\tilde{e}_2)$ satisfy the conditions stated above. Then, the Lie algebra $\lie{\{\tilde{e}_1,\tilde{e}_2\}}$ is infinite dimensional. We can write this formally as follows:
	
	\begin{corollary}
		Let $\tilde{e}_1\in\hat{A}_1$ be an element of degree $d_1:=\deg(\tilde{e}_1)=2d_0$ with $d_0\in\N_{\geq2}$, and $\tilde{e}_2\in\hat{A}_1$ be another element with positive degree. Suppose there exists a symplectic matrix $\boldsymbol{\sigma}\in\operatorname{Sp}(2,\C)$ such that the induced Lie-algebra isomorphism $\Phi((a^\dagg)^\ell a^{d-\ell}):=(\sigma_{11} a^\dagg +\sigma_{12} a)^\ell(\sigma_{21}a^\dagg+\sigma_{22})^{d-\ell}$ satisfies $\Phi(\lie{\{\tilde{e}_1,\tilde{e}_2\}})\subseteq\hat{A}_1$, and the following hold:
		\begin{enumerate}
			\item The first element $\tilde{e}_1$ satisfies: $\Phi(\tilde{e}_1)\upto{d_1}ic_1 (a^\dagg)^{d_0}a^{d_0}$ with $c_1\in \R$.
			\item The second element $\tilde{e}_2$ satisfies $\Phi(\tilde{e}_2)\upto{d_2}\sum_{k=0}^{d_2}\hat{c}_k (a^\dagg)^k a^{d_2-k}$, where $\hat{c}_k\neq0$ for at least one $k\in\{0,\ldots,d_2\}\setminus\{d_2/2\}$.
		\end{enumerate}
		Then, the Lie algebra $\lie{\{\tilde{e}_1,\tilde{e}_2\}}$ is infinite dimensional.
	\end{corollary}
	
	The specific equations that have to be solved by following the Igusa procedure can be computed explicitly when the need arises.
	Other sufficient conditions, such as those presented in Propositions~\ref{prop:algebra:g1:g2:g=:space} or~\ref{prop:algebra:g=:gperp:space}, can also be generalized in a similar fashion.

	\subsection{Nullity of finite-dimensional subalgebras}
	
	Another concept of interest in the context of Lie algebras is that of the nullity. We recall that the \textit{nullity} $\operatorname{Nul}(\g)$ of a Lie algebra $\g$ is defined as the minimal number of elements required to generate $\g$ \plainrefs{Knebelman:1935}. To be more precise:
	\begin{definition}[nullity]
		Let $\g$ be a finite-dimensional Lie algebra and $\mathcal{G}\subseteq\g$ a finite subset such that the Lie closure of $\mathcal{G}$ is $\g$, i.e., $\lie{\mathcal{G}}=\g$. Then,  we introduce the \emph{nullity} of $\g$, which is a key concept in the theory of Lie algebras \plainrefs{Alburquerque:1992}, defined as 
		\begin{align*}
			\operatorname{Nul}(\g):=\min\left\{|\mathcal{G}|\,\middle\mid\,\mathcal{G}\subseteq\g\text{ is a finite subset and }\lie{\mathcal{G}}=\g\right\}.
		\end{align*}
	\end{definition}
	Note that the nullity of a finite-dimensional Lie algebra is always a finite integer, as any basis $\mathcal{B}$ of $\g$ satisfies necessarily $\lie{\mathcal{B}}=\spn\{\mathcal{B}\}=\g$ and $|\mathcal{B}|=\dim(n)$. This implies that the nullity of a Lie algebra is bounded by its dimension, i.e., $\operatorname{Nul}(\g)\leq\dim(\g)$.
	We are particularly interested in finding the nullity of all non-abelian Lie subalgebras of $\hat{A}_1$, which are a particular class of algebras that we investigated in this work. We are able to obtain the following result:
	
	\begin{theorem}[Nullity of non-abelian subalgebras of $\hat{A}_1$]\label{thm:nullity}
		The Lie algebras $\slR{2}$, $\slR{2}\oplus\R$, $\slR{2}\ltimes\gh_1$, $\mathcal{L}_n$, $\wh_1$, $\wh_2$, $\tilde{\mathcal{L}}_n$ and $\gr(j_1,\ldots,j_n)$ all have nullity two.
	\end{theorem}

	\begin{proof}
		First, observe that any Lie algebra of dimension greater than one has a nullity of at least two, due to the antisymmetry of the Lie bracket. To show this, one supposes the converse: Let $\g$ be an $n$-dimensional Lie algebra that is generated by the singleton set $\mathcal{S}:=\{e\}\subseteq\g$, and let $n\geq 2$. To compute the Lie closure $\lie{\mathcal{S}}$, we can start by computing the vector space spanned by $e$. One finds $\spn\{e\}\subseteq \lie{\mathcal{S}}$. Next, we need to compute the Lie bracket of every pair of elements in $\spn\{e\}$. By the antisymmetry of the Lie bracket, one has $[\kappa_1e,\kappa_2e]=0\in\spn\{e\}$ for all $\kappa_1,\kappa_2\in\R$. This shows that $\spn\{e\}$ is closed under the Lie bracket and it follows that $\lie{\mathcal{S}}=\spn\{e\}$. By assumption, one has $\lie{\mathcal{S}}=\g$ and consequently $1=\dim(\lie{\mathcal{S}})=\dim(\g)=n\geq 2$, which is a contradiction. Hence, the nullity of a finite-dimensional Lie algebra is at least two if $\dim(\g)\geq 2$. Thus, if we can bound the nullity of a given Lie algebra of dimension two or more from above by two, we can conclude that the nullity is two. We now consider each Lie algebra individually:
		\begin{enumerate}[label = (\roman*)]
			\item $\boldsymbol{\slR{2}}$: This Lie algebra is simple, and thus, by Theorem 6 from \plainrefs{Kuranishi:1951}, one has $\operatorname{Nul}(\slR{2})=2$.
			\item $\boldsymbol{\slR{2}\oplus\R}$: Let $z$ denote the central element, and use the basis $\{x,y,h\}$ of $\slR{2}$ from Definition~\ref{def:sl2:algebra}. It is straightforward to verify that $\lie{\{x,y+z\}}\cong\slR{2}\oplus\R$. Hence, $\operatorname{Nul}(\slR{2}\oplus\R)=2$.
			\item $\boldsymbol{\slR{2}\ltimes\gh_1}$: Consider the basis $\{x,y,h,p,q,z\}$ from Definition~\ref{def:Schroedinger:algebra}. We verify now that $\lie{\{y,h+x+p\}}\cong\slR{2}\ltimes\gh_1$ and start with computing $[h+x+p,y]=-2y+h$, which implies $\spn\{h,y,x+p\}\subseteq\lie{\{y,h+x+p\}}$ via appropriate linear combinations. Next, compute $[h,x+p]=2x-p$. Therefore, one must also have $\spn\{x,y,h,p\}\subseteq\lie{\{y,h+x+p\}}$. Finally, compute $[x,p]=q$ and $[q,p]=z$, confirming $\lie{\{y,h+x+p\}}\cong\slR{2}\ltimes\gh_1$. Hence, $\operatorname{Nul}(\slR{2}\ltimes\gh_1)=2$.
			\item $\boldsymbol{\mathcal{L}_n}$: These algebras have a basis $\{e_1,e_2,\ldots,e_n,e_{n+1}\}$ satisfying $[e_1,e_j]=e_{j+1}$ for all $j\in\{2,\ldots,n\}$. Clearly, $\lie{\{e_1,e_2\}}\cong\mathcal{L}_2$. Hence, $\operatorname{Nul}(\mathcal{L}_n)=2$ for all $n\geq 2$.
			\item $\boldsymbol{\wh_1}$: Using the basis $\{e_1,e_2,e_3,e_4\}$ from Definition~\ref{def:Wigner:Heisenberg:algebra:1}, it is straightforward to verify that $\lie{\{e_2+e_3,e_4\}}\cong\wh_1$. Hence, $\operatorname{Nul}(\wh_1)=2$.
			\item $\boldsymbol{\wh_2}$: With the basis $\{h,q,p,z\}$ from Definition~\ref{def:Wigner:Heisenberg:algebra:2}, it is easy to confirm that $\lie{\{h,q\}}\cong\wh_2$. Hence, $\operatorname{Nul}(\wh_2)=2$.
			\item $\boldsymbol{\tilde{\mathcal{L}}_n}$: These algebras admit a basis $\{e_0,e_1,\ldots,e_n,e_{n+1}\}$ with $[e_0,e_1]=e_1$, $[e_0,e_j]=-(n+1-j)e_j$, and $[e_1,e_j]=e_{j+1}$ for $j\in\{2,\ldots,n\}$. It is easy to verify that $\lie{\{e_0,e_1+e_2\}}\cong\tilde{\mathcal{L}}_n$. Hence, $\operatorname{Nul}(\tilde{\mathcal{L}}_n)=2$ for all $n\geq 2$.
			\item $\boldsymbol{\gr(j_1,\ldots,j_n)}$: These algebras admit a basis $\{e_0,e_1,\ldots,e_n\}$ with $[e_0,e_k]=j_ke_k$. Now compute $\operatorname{ad}_{e_0}^m\left(\sum_{k=1}^ne_k\right)=\sum_{k=1}^n j_k^m e_k$.
			Thus, one has $\lie{\{e_0,e_1+e_2+\ldots+e_n\}}\cong\gr(j_1,\ldots,j_n)$, since
			\begin{align*}
				\begin{pmatrix}
					1 & j_1 & j_1^2 & \ldots & j_1^{n-1}\\
					1 & j_2 & j_2^2 & \ldots & j_2^{n-1}\\
					1 & j_3 & j_3^2 & \ldots & j_3^{n-1}\\
					\vdots & \vdots & \vdots &&\vdots\\
					1& j_n & j_n^2 & \ldots & j_n^{n-1}
				\end{pmatrix}
			\end{align*}
			is a regular Vandermonde matrix \plainrefs{Macon:1958}, and $0\leq j_1<j_2<\ldots j_n$. Hence, $\operatorname{Nul}(\gr(j_1,\ldots,j_n))=2$.
		\end{enumerate}
		Note that none of the proofs for the individual Lie algebras presented above depend on any specific realization, making this a general statement.
	\end{proof}
	
	\begin{corollary}
		The nullity of every non-abelian finite-dimensional subalgebra of the complex Weyl algebra is two.
	\end{corollary}
	
	\begin{proof}
		The algebras in questions are $\slC{2}$, $\slC{2}\oplus\C$, $\slC{2}\ltimes \gh_{1,\C}$, $\mathcal{L}_{2,\C}$, $\tilde{\mathcal{L}}_{n,\C}$ and $\gr_\C(j_1,\ldots,j_n)$ \plainrefs{TST:2006}. The statement above is a direct consequence of Theorem~\ref{thm:nullity} and Theorem 1 from \plainrefs{Sato:1974}.
	\end{proof}
	
	The observation that all these non-abelian finite-dimensional Lie algebras have nullity two has physically relevant consequences. In the remit of quantum control theory, where one aims at steering an initial state to desired final target states, the nullity of the Lie algebra corresponds to the minimum number of control Hamiltonians required to generate the full dynamical Lie algebra governing the time evolution of the system, assuming that no drift term is present. Thus, a nullity of two implies that only two linearly independent control fields are sufficient to generate the whole Lie algebra (via nested commutators), even for solvable Lie algebras of arbitrarily high-dimension, such as $\tilde{\mathcal{L}}_n$. We can expect that this will have positive implications for the design of experimental setups or classical simulation methods.

	\subsection{Applications to physical context}
	
	We now proceed to comment on the implications and applications within physics and physical theories. The general idea has been spelled out in detail in the introductory chapters of this work, where we clearly motivated our work by the desire to control
	quantum dynamics, in particular, through the process of dynamics factorization.

	\subsubsection{Perspective from control theory\label{sec:consequences:control}}
	
	We shortly complement our discussion on the dichotomy between
	Lie algebras of finite and infinite dimension by commenting from the perspective of control theory.
	Even if the generated Lie algebra in the skew-hermitian Weyl algebra $\hat{A}_n$ is finite dimensional,
	we need to fulfill further analyticity conditions to obtain a \emph{restricted controllability}
	\plainrefs{Huang:Tarn:1983,Lan:2005,Wu:Tarn:2006,Nelson:1959,Bloch:finite:controlability:2010}.
	Moreover, we usually do not observe exact
	controllability in infinite-dimensional quantum systems \plainrefs{Ball1982,BCC:2019,Dirr2022,Huang:Tarn:1983,Lan:2005,KZSH:2014,Keyl:2019,ABFH:2018,KEGGZ:2023}
	and we expect only weaker conditions such as \emph{accessibility} (see, e.g., \plainrefs{Bloch2015}), \emph{pure-state controllability}, or 
	\emph{approximate controllability} 
	in terms of the strong topology
	\plainrefs{KZSH:2014,Keyl:2019}
	to hold. Furthermore controllability questions in infinite dimensions
	can no longer be resolved using an infinitesimal 
	perspective and one needs to switch to a group picture \plainrefs{KZSH:2014,Keyl:2019}.
	As discussed in Section~\ref{section:scope:of:the:work}, the existence of
	unbounded operators \plainrefs{Messiah:1999,meise1997,hall2013quantum,schmuedgen2012}
	also contributes to the complexity of infinite-dimensional controllability
	\plainrefs{Huang:Tarn:1983,KZSH:2014,KEGGZ:2023,Keyl:2019}.

	In this work, as well as the one on which it is based \plainrefs{Bruschi:Xuereb:2024}, we take a more formal point of view in order
	to build up a structural understanding 
	while, at the same time, mostly ignoring the complexities in infinite dimensions just mentioned above. It is seems curious
	that so few possibilities for finite-dimensional Lie algebras in the skew-hermitian Weyl algebra $\hat{A}_n$ occur.
	There, in the future, it might be interesting to also explore a characterization 
	of different infinite-dimensional Lie algebras that are realizable in the Weyl algebra.

	\begin{tcolorbox}[breakable, colback=Cerulean!3!white,colframe=Cerulean!85!black,title=State reachability in $\hat{A}_1$ ]
		In the remit of the theory of quantum control, one of the main quests is to determine if any state in the Hilbert space can be reached from a given initial state, once the dynamics have been chosen. This question can be settled in the context of finite-dimensional systems while, as we have just discussed in this section, the answer in the case of infinite-dimensional systems is more subtle and remains unsolved. Here, we present an example for the current single-mode case $\hat{A}_1$, where the dynamics do not allow for a full exploration of the Hilbert space starting from any initial state. This example clearly indicates that, for the chosen dynamics, full state reachability cannot be obtained via Hamiltonians with finite-dimensional Lie algebras. 
		\begin{example}	
			Assume the control system induces the abelian Lie algebra $\g$ containing the free term $ia^\dagg a$. Then, one has no control over the system, even with an infinite number of control terms:
			Let $H(t)=H_0+H_I(t)=H_0+\sum_{j\in\mathcal{J}}u_j(t) G_j$ such that $\g:=\lie{\{iH_0,iG_j\}}$ is abelian and contains the term $ia^\dagg a$. We know that $C_{\hat{A}_1}(ia^\dagg a)=\hat{A}_1^0\oplus\hat{A}_1^=$ and consequently $H(t)=\sum_{n=0}^\infty f_j(t) (a^\dagg)^n a^n$. Thus $U(t)\ket{0}=\ket{0}$ for all times $t$. Hence $|1-\langle 1|U(t)|0\rangle|^2=1$.
			Thus, if $H(t)$ contains a free term, a Kerr nonlinearity, and induces an abelian Lie algebra, one cannot reach all states in the Hilbert space. 
		\end{example}
	\end{tcolorbox}

	\subsubsection{Implications for quantum dynamics}
	As a final contribution to this commentary we add a brief discussion about the role of this work in the broad effort to control quantum dynamics. As already abundantly stated above, the ultimate goal is to be able to have an explicit solution to the dynamics of an arbitrary quantum system. Achieving this goal is prevented by the seemingly impossible-to-pierce complexity of the space of bosonic linear operators that act on the Hilbert space of bosonic quantum systems, also called the skew-hermitian Weyl algebra $\hat{A}_n$. The new approach pioneered in the literature has provided a fresh perspective on possible ways to systematically categorize such a space of operators \plainrefs{Bruschi:Xuereb:2024}, an endeavour that is at its beginnings. 
	
	Our current work has provided an in depth dive into the study of the properties of $\hat{A}_1$, that is, the simplest possible skew-hermitian Weyl algebra. This is meaningful because it is necessary to study all aspects relevant of the overarching ambition in this new direction of research. We have been able to provide significant results in the remit of Lie theory concerning this algebra, while we have left open a few questions that were not answered here.  
	
	An important aspect that must be highlighted here is that fundamental differences between $\hat{A}_1$ and $\hat{A}_n$ with $n\geq2$ are of special interest. In fact, although quantum dynamics of a single bosonic mode is of limited interest for practical applications, fundamental properties of $\hat{A}_1$ are of great interest in the quest of understanding the fundamental difference between single and multi-mode systems. One immediate observation 
	is that $\hat{A}_1$ is much simpler than the general case of $\hat{A}_n$ where the additional subset $\hat{A}_n^\mathrm{om}$
	need to be considered, while  $\hat{A}_1^\mathrm{om}$ is empty. The space $\hat{A}_n^\mathrm{om}$ is the space spanned by $g_\sigma^\gamma=g_\sigma^{\iota_j}a^{\tilde{\gamma}}$ with $|\tilde{\gamma}|\geq 2$ and $\tilde{\gamma}_j=\tilde{\gamma}_{n+j}=0$ \plainrefs{Bruschi:Xuereb:2024}. That is 
	\begin{equation*}
		\hat{A}_n^\mathrm{om}:=\spn\{g_\sigma^\gamma=g_\sigma^{(\alpha,\beta)}\,\mid\,|\gamma|=|\alpha|+|\beta|\geq 3,\,\alpha_k+\beta_k=1\text{ for a unique index }k,\,\alpha_j=\beta_j\text{ for }j\neq k\}. 
	\end{equation*}
	The \enquote{om} stands for \emph{optomechanics} where a typical optomechanical Hamiltonian is given by $H=\omega_1a_1^\dagg a_1+\omega_2 a_2^\dagg a_2 +ga_1^\dagg a_1(a_2+a_2^\dagg)$ \plainrefs{Aspelmeyer:Kippenberg:2014}.
	This is a key difference that will be explored in depth in future work, where the properties of skew-hermitian Weyl algebras with two or more modes will be studied.

	\section{Conclusions\label{sec:conclusion}}
	
	This work addresses the overarching question of determining finite factorizations for quantum dynamics of bosonic systems by characterizing the space of all possible Hamiltonian Lie algebras of one bosonic mode, also known as the skew-hermitian Weyl algebra $\hat{A}_1$. In particular, we focused on determining the structure and classification of finite-dimensional Lie subalgebras and achieved our goal following a threefold path. 
	
	First, we provided a complete and exhaustive glossary of all finite-dimensional Lie subalgebras $\g\subseteq\hat{A}_1$ that are generated by a finite set of single skew-hermitian monomials $g_{\sigma}^{\gamma}$, based on the decomposition of $\hat{A}_1$ into meaningful subspaces provided in the literature \plainrefs{Bruschi:Xuereb:2024}. This also includes a classification of all realizations of such (non-abelian) algebras that was achieved by means of a constructive and algorithmic approach. The resulting list includes, among others, the Wigner-Heisenberg algebras $\wh_1$ and $\wh_2$, the special linear algebra $\slR{2}$, and the Schr\"odinger algebra $\mathcal{S}\cong\slR{2}\ltimes\gh_1$.
	Second, we established necessary and sufficient conditions for the finiteness of Lie algebras containing a free Hamiltonian term $e_0=i(a^\dagg a+c)$, where $c\in\mathbb{R}$. This term appears (in its hermitian form) within all Hamiltonians that encode the dynamics of single-mode bosonic systems. We proved that any such non-abelian algebra must be a subalgebra of the Schr\"odinger algebra $\mathcal{S}$, and must be isomorphic to one of the following known algebras: the Wigner-Heisenberg $\wh_2$, the special linear $\slR{2}$, the algebra $\slR{2}\oplus\R$, or  the Schr\"odinger $\mathcal{S}$ itself. 
	Third, we extended our classification to include all nilpotent and non-solvable finite-dimensional Lie algebras. While the results found in this context have structures that resemble those known from the classification of the complex Weyl algebra $A_1$, the classification of all solvable but not nilpotent Lie algebras remains incomplete. We identified two infinite families of the latter algebras, namely $\tilde{\mathcal{L}}_n$ and $\gr(j_1,\ldots,j_n)$, as well as at least one sporadic example, namely the Wigner-Heisenberg algebra $\wh_2$, which does not belong to these classes $\tilde{\mathcal{L}}_n$ and $\gr(j_1,\ldots,j_n)$ in the real case although it does in the complex setting.
	
	Among important aspects of our work we note that a key technical contribution of this work was the generalization of the Commutator Chains (i.e., specific sequences of nested commutators) introduced in the literature \plainrefs{Bruschi:Xuereb:2024}, which allow us to control the maximal degree of meaningful nested commutators and therefore to derive precise conditions under which a given set of monomials yields an infinite-dimensional Lie algebra. This framework enabled us to relax constraints on the initial generating set as compared to those put forward in the foundational work \plainrefs{Bruschi:Xuereb:2024}, and therefore to explore more general conditions for algebraic finiteness. Another important aspect is the fact that our analysis was restricted to the case of a single bosonic mode, and therefore the skew-hermitian Weyl algebra $\hat{A}_1$. Nevertheless, the insights gained here constitute an important step in the understanding of the whole multi-mode framework within the remit of skew-hermitian Weyl algebra $\hat{A}_n$ of $n$ bosonic modes, thereby extending the ideas, concepts and results presented in the foundational work \plainrefs{Bruschi:Xuereb:2024}. A particularly important task was that of identifying less strict conditions under which a set of monomials yields a finite- or infinite-dimensional Lie algebra. Achieving this goal is generally more complicated in the multi-mode case because, among other things, the abelian centralizer condition does only hold for the single-mode case. The classification of all finite-dimensional solvable subalgebras remains an open problem and will be tackled in future work. The methods developed in this work are expected to aid in addressing this challenge.
	
	The results of our work have significance beyond their immediate mathematical implications, in particular within a broad part of physics. Among others, these results can inform on mathematical aspects of the theory in fields such as quantum control, quantum computing, quantum simulation, or quantum annealing, where understanding the algebraic structure of quantum dynamics is essential for the design of efficient, reliable and robust quantum protocols. Furthermore, a deeper understanding of the formalism can allow for the prediction of new phenomena, a task that would be otherwise difficult, if not impossible, through perturbation theory or numerical methods alone. Finally, our work contributes to the overarching ambition of understanding the fundamental differences that arise in the transition between physical systems with few (or one) constituents and those with large numbers of constituents.

	\section*{Data availability statement}
	
	Data sharing is not applicable to this work as no datasets are used to support the presented results.

	

	\section*{Conflict of interest}
	
	The authors have no relevant financial or non-financial interest to disclose.

	\section*{Author contributions}
	DEB and RZ conceived the idea of studying the properties of the single-mode skew-hermitian Weyl algebra. ASG and SB obtained the initial insights in this work. In particular, ASG obtained a first proof of Theorem~\ref{thm:andreea:sona} and related results, while SB studied the properties of the algebras using the symbolic software \textsc{Magma} \plainrefs{Bosma:MAGMA:1997}. TCH produced the proof of all other results, as well as the figures. DEB and RZ supervised the project. All authors contributed towards the preparation of the manuscript.

	\section*{Funding}
	ASG and SB were funded by the Deutsche Forschungsgemeinschaft (DFG, German Research Foundation) under Germany’s Excellence Strategy – \href{https://ml4q.de/}{Cluster of Excellence Matter and Light for Quantum Computing (ML4Q)} EXC 2004/1 – 390534769. SB was also funded by St. Olaf College's Johnson Family Opportunity Fund. DEB and TCH
	acknowledge support from the joint project No. 13N15685 ``German Quantum Computer based on Superconducting
	Qubits (GeQCoS)'' sponsored by the German Federal Ministry of Education and Research (BMBF) under the
	\href{https://www.quantentechnologien.de/fileadmin/public/Redaktion/Dokumente/PDF/Publikationen/Federal-Government-Framework-Programme-Quantum-technologies-2018-bf-C1.pdf}{framework programme “Quantum technologies -- from basic research to the market”}. DEB
	acknowledges furthermore support from the German Federal Ministry of Education and Research via
	the \href{https://www.quantentechnologien.de/fileadmin/public/Redaktion/Dokumente/PDF/Publikationen/Federal-Government-Framework-Programme-Quantum-technologies-2018-bf-C1.pdf}{framework programme “Quantum technologies -- from basic research to the market”}
	under contract number 13N16210 ``Diamond spin-photon-based quantum computer (SPINNING)''. RZ acknowledges funding
	under Horizon Europe programme HORIZON-CL4-2022-QUANTUM-02-SGA via the project
	\href{https://doi.org/10.3030/101113690}{101113690} (PASQuanS2.1).


	\appendix

	\section{The commutation relations in $\lie{\hat{A}_0^1\oplus\hat{A}_1^1\oplus\hat{A}_1^2}$\label{App:commutation:relations}}
	
	Here we provide lists of useful commutation relations to aid the proof of many claims of this work. Table~\ref{tab:Full:Commutator:Algebra:(pesudo):schroedinger:algebra} lists all non-trivial commutators among the six fundamental monomials of the Lie algebra $\lie{\hat{A}_0^1\oplus\hat{A}_1^1\oplus\hat{A}_1^2}$. This algebra is isomorphic to the Schr\"odinger algebra $\mathcal{S}=\slR{2}\ltimes\gh_1$, as shown in Proposition~\ref{prop:schroedinger:algebra}, altough the basis here differs from the standard one typically employed in the literature and shown in Definition~\ref{def:Schroedinger:algebra}. The chosen basis is adapted to the structure of the skew-hermitian Weyl algebra and emphasizes its decomposition into physically and algebraically meaningful subspaces, as defined in the literature \plainrefs{Bruschi:Xuereb:2024}. The monomials of interest are the monomials $g_+^{0}/2=i$, $g_+^\tau/2=ia^\dagg a$, $g_+^{\iota_1}$, $g_-^{\iota_1}$, $g_+^{2\iota_1}$, and $g_-^{2\iota_1}$.
	
	\begin{table}[h!]
		\centering
		{\small
			\begin{tblr}{ |c||c|c||c|c||c|c|  }
				\hline
				\SetCell[c=7]{c}{full commutator algebra of $\langle\hat{A}_1^0\oplus\hat{A}_1^1\oplus\hat{A}_1^2\rangle_{\textrm{Lie}}$} \\
				\hline
				\hline
				$[\downarrow,\rightarrow]$   & $i$ & $ia^\dagg a$ & $g_+^{\iota_1}$  & $g_-^{\iota_1}$  & $g_+^{2\iota_1}$  & $g_-^{2\iota_1}$  \\
				\hline
				\hline
				$i$    &    &  0  & 0 & 0 & 0 & 0  \\
				\hline
				$ia^\dagg a$   &    &    &  $g_-^{\iota_1}$ & $-g_+^{\iota_1}$ & $2g_-^{2\iota_1}$ & $-2g_+^{2\iota_1}$  \\
				\hline
				\hline
				$g_+^{\iota_1}$   &    &    &  & $-2i$ & $2g_-^{\iota_1}$ & $-2g_+^{\iota_1}$  \\
				\hline
				$g_-^{\iota_1}$ &  &  &  & & $2g_+^{\iota_1}$ & $2g_-^{\iota_1}$ \\
				\hline
				\hline
				$g_+^{2\iota_1}$ &  &  &  & & & $-8i(a^\dagg a+\frac{1}{2})$\\
				\hline
				$g_-^{2\iota_1}$ &  &  &  & & & \\
				\hline
			\end{tblr}
		}
		\caption{All commutators of the monomials of the Lie algebra  $\langle\hat{A}_1^0\oplus\hat{A}_1^1\oplus\hat{A}_1^2\rangle_\textrm{Lie}$. All omitted entries are either zero or follow by antisymmetry $[A,B]=-[B,A]$ of the commutator.}
		\label{tab:Full:Commutator:Algebra:(pesudo):schroedinger:algebra}
	\end{table}\noindent

	\section{Algebraic conditions for Lie-algebra isomorphisms\label{App:Testing:Isomorphisms}}
	
	We proceed to give algebraic conditions for Lie-algebra isomorphisms and restrict our attention, for simplicity, to the real case. It is immediate to see, due to  Definition~\ref{def:Lie:algebra:homomorphism}, that two Lie algebras can be isomorphic only if they have the same dimension. Consider two $m$-dimensional real Lie algebras $\g$ and $\gh$ with respective bases $\mathcal{B}_\g=\{x_j\}_{j=1}^m$ and $\mathcal{B}_{\gh}=\{y_j\}_{j=1}^m$ equipped with the commutation relations
	\begin{align*}
		[x_j,x_k]=\sum_{\ell=1}^mf_{jk\ell}x_\ell,\quad\text{and}\quad[y_j,y_k]=\sum_{\ell=1}^m\hat{f}_{jk\ell}y_\ell,
	\end{align*}
	where $f_{jk\ell}$ and $\hat{f}_{jk\ell}$ are the structure constants of the Lie algebras $\g$ and $\gh$ respectively.
	Suppose now that there exists a Lie-algebra isomorphism $\phi:\g\to\gh$. Then, by the linearity of $\phi$ we have
	\begin{align*}
		\phi(x_j)=\sum_{k=1}^ma_{jk}y_k\quad\text{for all }j\in\{1,\ldots,m\},
	\end{align*}
	where $a_{jk}$ are some real coefficients.
	The following identity must hold in order to satisfy the homomorphism condition from equation \eqref{eqn:def:Lie:algebra:isomorphism}, namely:
	\begin{align*}
		0&=[\phi(x_j),\phi(x_k)]-\phi([x_j,x_k])=\sum_{p,q=1}^ma_{jp}a_{kq}[y_p,y_q]-\sum_{p=1}^mf_{jkp}\phi(x_p)\\
		&=\sum_{p,q,\ell=1}^ma_{jp}a_{kq}\hat{f}_{pq\ell}y_\ell-\sum_{p,\ell=1}^ma_{p\ell}f_{jkp}y_\ell=\sum_{\ell=1}^m\left(\sum_{p=1}^m\left(a_{jp}\sum_{q=1}^ma_{kq}\hat{f}_{pq\ell}-a_{p\ell}f_{jkp}\right)\right)y_\ell.
	\end{align*}
	This condition only needs to be verified for the indices $1\leq j<k\leq m$, due to the antisymmetry of the commutator. We can reduce this to the condition
	\begin{align}
		\sum_{p=1}^m\left(a_{jp}\sum_{q=1}^ma_{kp}\hat{f}_{pq\ell}-a_{p\ell}f_{jkp}\right)&=0,
	\end{align}
	since $\mathcal{B}_{\gh}$ is a basis of linear independent elements $y_j$. The invertibility condition of $\phi$ can furthermore expressed by the condition $\det(\phi)\neq 0$ or equivalently $1-c\det(\phi)=0$ for $c\in\R$. Thus, testing whether $\g$ and $\gh$ are isomorphic amounts to solving a system of $m^2(m-1)/2+1$ equations with $m^2+1$ variables $\{a_{jk}\}_{j,k=1}^m$ and $c$, where all but one equations are quadratic.

	\section{Revisiting the Igusa condition\label{App:Igusa}}

	This appendix refines a criterion, originally proposed by Igusa \plainrefs{Igusa1981}, to determine when two elements in the Weyl algebra generate an infinite-dimensional Lie algebra. The core idea of this criterion is straightforward:  one identifies conditions under which two elements $x,y\in A_1$ generate a generic Commutator Chain whose elements are of strictly increasing degree. This ensures that the Lie algebra $\g:=\lie{\{x,y\}}$ is infinite dimensional. One then generalizes this criterion by characterizing all isomorphisms that map two elements $\tilde{x},\tilde{y}\in A_1$ to elements satisfying the necessary conditions that have been identified. Our goal is to provide a re-derivation of the Igusa condition using the formalism introduced in the literature on which our work is based \plainrefs{Bruschi:Xuereb:2024}.
	We start by introducing the following notations and concepts: 
	
	\noindent\textbf{Skew-hermitian polynomial space}: Let $(\mathbb{S}[X,Y],+,\cdot)\equiv\mathbb{S}[X,Y]$ denote the vector space of commutative polynomials over the complex numbers $\C$, equipped with the usual addition \enquote{+} and scalar multiplication \enquote{$\cdot$}, such that replacing the variables $X$ and $Y$ with the operators $a^\dagg$ and $a$ respectively yields skew-hermitian operators when the polynomials are written in the canonical form.
	We want to emphasize this: This operation is only well-defined if we require the polynomials $p\in\mathbb{S}[X,Y]$ to be in the canonical form $p(X,Y)=\sum_{j,k=0}^ma_{jk}X^j Y^k$, due to the non-abelian property of the Weyl algebra. We denote polynomials in $\mathbb{S}[X,Y]$ by symbols such as $x(X,Y)$, $y(X,Y)$, and define $d_x:=\deg(x(X,Y))$. Any polynomial $x\in\mathbb{S}[X,Y]$ can be expresses as
	\begin{align*}
		x(X,Y)=\sum_{m=0}^{d_x}\sum_{j=0}^mc_{mj}X^{m-j}Y^j,
	\end{align*}
	where $c_{mj}$ are complex coefficients. For convenience, we may write $x\equiv x(X,Y)$ when the context is clear.
	Note the following important observations: the space $\mathbb{S}[X,Y]$ is not closed under multiplication, and the product of two such polynomials lies in the commutative ring of complex polynomials $\C[X,Y]$.
	
	\noindent\textbf{Vectors in} $\boldsymbol{\C^2}$: Vectors in $\C^2$ are denoted by $\boldsymbol{\xi}$ and $\boldsymbol{\eta}$.
	
	\noindent\textbf{Symplectic structure}: We introduce the sympletic matrix $\boldsymbol{J}$ along with the symplectic form $\omega:\C^2\times\C^2\to\C$ defined by
	\begin{align*}
		\boldsymbol{J}:=\begin{pmatrix}
			0&1\\
			-1&0
		\end{pmatrix}\quad\text{and}\quad\omega(\boldsymbol{\xi},\boldsymbol{\eta})=\boldsymbol{\xi}^\Tp\boldsymbol{J}\boldsymbol{\eta}.
	\end{align*}
	These can be used to define the \emph{real symplectic Lie group} $\operatorname{Sp}(2,\C):=\{\boldsymbol{A}\in\C^{2\times2}\mid\,\boldsymbol{A}^\Tp\boldsymbol{J}\boldsymbol{A}\}$ as the matrix group containing all real $2\times 2$ matrices that preserve the symplectic form $\omega$. The complex symplectic Lie group $\operatorname{Sp}(2,\C)$ is the group of complex matrices with determinant one and therefore isomorphic to the Lie group $\operatorname{SL}(2,\C)$. Note the importance of this symplectic structure in the context of linear quantum optics \plainrefs{Adesso:Ragy:2014}.
	
	\noindent\textbf{Reduction map}: We introduce the reduction map $R:\mathbb{S}[X,Y]\to\mathbb{S}[X,Y]$. This map reduces a polynomial $x$ to its homogeneous component of highest degree \plainrefs{Cox:2013}. That is, given some polynomial $x=\sum_{m=0}^{d_x}\sum_{j=0}^mc_{mj}X^{m-j}Y^j$, the reduction map removes every monomial component in its decomposition that is not of the highest degree, i.e.: 
	\begin{align*}
		R(x)(X,Y)\equiv R_x(X,Y):=\sum_{j=0}^{d_x}c_{d_x j}X^{d_x-j}Y^{j},\quad\text{where}\quad x=\sum_{m=0}^{d_x}\sum_{j=0}^mc_{mj}X^{m-j}Y^j.
	\end{align*}
    The reduction map on the polynomial ring $\C[X,Y]$, i.e., the map $R:\mathbb{C}[X,Y]\to\mathbb{C}[X,Y]$, is analogously defined.
 	
	\noindent\textbf{Partial derivatives and gradient}: We define the two derivative maps $\partial_X,\partial_Y:\mathbb{S}[X,Y]\to\C[X,Y]$ by:
	\begin{align*}
		\partial_X\sum_{m=0}^{d_x}\sum_{j=0}^mc_{mj}X^{m-j}Y^j&=\sum_{m=0}^{d_x}{\sum_{j=0}^{m-1}}c_{mj}(m-j)X^{m-j-1}Y^j,\;&\;\partial_Y\sum_{m=0}^{d_x}\sum_{j=0}^mc_{mj}X^{m-j}Y^j&=\sum_{m=0}^{d_x}{\sum_{j=1}^{m}}c_{mj}jX^{m-j}Y^j.
	\end{align*}
	We define the gradient map
	\begin{align*}
		\nabla:\mathbb{S}[X,Y]\to\C[X,Y]\times\C[X,Y]
		\quad
		\text{by its action}
		\quad
		x(X,Y)\mapsto\begin{pmatrix}
			\partial_X(x(X,Y))\\
			\partial_Y (x(X,Y))
		\end{pmatrix}.
	\end{align*}
	These maps acting on the polynomial ring $\C[X,Y]$, i.e., the maps $\partial_X,\partial_Y:\mathbb{S}[X,Y]\to\C[X,Y]$ and $\nabla:\mathbb{C}[X,Y]\to\C[X,Y]\times\C[X,Y]$, are analogously defined. For notational convenience, we write $\nabla x(a^\dagg,a)$ to mean $\nabla x(X,Y)|_{(X,Y)=(a^\dagg ,a)}$, and similarly for $\partial_X (a^\dagg, a)$ and $\partial_Y x(a^\dagg ,a)$ instead of $\partial_X x(X,Y)|_{(X,Y)=(a^\dagg,a)}$ or $\partial_Y x(X,Y)|_{(X,Y)=(a^\dagg,a)}$ respectively. 
	
	\noindent\textbf{Symplectic action on polynomials}: Let $\boldsymbol{\sigma}\in\operatorname{Sp}(2,\C)$ be an arbitrary symplectic matrix. It can be used to define the map
	\begin{align*}
		\sigma:\mathbb{S}[X,Y]\to\C[X,Y]
		\quad
		\text{by its action}
		\quad
		x(X,Y)\mapsto (\sigma x)(X,Y):=x\left(\boldsymbol{\sigma}\begin{pmatrix}
			X\\
			Y
		\end{pmatrix}\right)=x(\sigma_{11}X+\sigma_{12}Y,\sigma_{21}X+\sigma_{22}Y).
	\end{align*}
	Given the definitions above we can now prove a set of useful results regarding the extension of the Igusa condition to the framework used in this work.
	
	\begin{proposition}\label{prop:help:commutator:q:p}
		Let $r,s\in\N_{\geq0}$ be two non-negative integers. Then one has: (a) $[(a^\dagg)^r,a^s]\upto{r+s-2}-rs(a^\dagg)^{r-1}a^{s-1}$, (b) $[(a^\dagg)^r,a^s]\upto{r+s-2}-rsa^{s-1}(a^\dagg)^{r-1}$, and (c) $(a^\dagg)^ra^s\upto{r+s}a^s (a^\dagg)^r$.
	\end{proposition}
	
	\begin{proof}
		For claim (a), see Lemma 16 from \plainrefs{Bruschi:Xuereb:2024}. Claims (b) and (c) follow immediately from (a). 
	\end{proof}
	
	Proposition~\ref{prop:help:commutator:q:p} guarantees that the operator $o_1\in A_1$ obtained by replacing the variables $X$ and $Y$ from a canonically ordered polyonomial $p(X,Y)=\sum_{j,k=1}^ma_{jk}X^j Y^k$ with the creation and annihilation operators $a^\dagg$ and $a$ respectively satisfies $o_1\upto{\deg(p)} o_2$, where $o_2\in A_1$ is any operator obtained by replacing the variables $X$ and $Y$ from the polynomial $p(X,Y)$ also with the creation and annihilation operators $a^\dagg$ and $a$ respectively irrespective of the ordering of the variables.
	
	\begin{proposition}\label{prop:product:rule}
		Let $r,s\in\N_{\geq0}$ and $\sigma_{11},\sigma_{12},\sigma_{21},\sigma_{21}\in\C$. Then, the following identities hold:
		\begin{align*}
			\partial_X(\sigma_{11}X+\sigma_{12}Y)^r(\sigma_{21}X+\sigma_{22}Y)^s&=r\sigma_{11}(\sigma_{11}X+\sigma_{12}Y)^{r-1}(\sigma_{21}X+\sigma_{22}Y)^s+s\sigma_{21}(\sigma_{11}X+\sigma_{12}Y)^r(\sigma_{21}X+\sigma_{22}Y)^{s-1},\\
			\partial_Y(\sigma_{11}X+\sigma_{12}Y)^r(\sigma_{21}X+\sigma_{22}Y)^s&=r\sigma_{12}(\sigma_{11}X+\sigma_{12}Y)^{r-1}(\sigma_{21}X+\sigma_{22}Y)^s+s\sigma_{22}(\sigma_{11}X+\sigma_{12}Y)^r(\sigma_{21}X+\sigma_{22}Y)^{s-1}.
		\end{align*}
	\end{proposition}
	
	Note that we evade the problem of differentiating non-commutative products by performing the derivation on the commutative ring $\mathbb{C}[X,Y]$ and not the non-abelian polynomials $p(a^\dagg,a)$.

	\begin{proof}
		The proof is straightforward using the Binomial theorem. One only needs to recall that $\mathbb{S}[X,Y]$ is a subspace of the commutative ring $\C[X,Y]$, and that the following identity
		\begin{align*}
			\binom{r-1}{k}\binom{s}{k'}+\binom{r}{k}\binom{s-1}{k'}=\frac{r+s-(k+k')}{rs}\binom{r}{k}\binom{s}{k'}
		\end{align*}
		holds.
	\end{proof}
	
	\begin{lemma}\label{lem:Igusa:1}
		Let $x,y\in\mathbb{S}[X,Y]$ be two arbitrary polynomials with degrees $d_x$ and $d_y$ respectively and ${\sigma}\in\operatorname{Sp}(2,\C)$. Let furthermore $d:=d_x+d_y-2$. Then:
		\begin{enumerate}[label = (\alph*)]
			\item The reduction map and the symplectic action commute. That is: $(\sigma R_x)(X,Y)=R_{\sigma x}(X,Y)$;
			\item The commutator of the two polynomials $x(a^\dagg,a), y(a^\dagg ,a)\in\hat{A}_1$ is, to highest degree, uniquely determined by evaluating the symplectic form $\omega$ at the gradients of the reduction maps of $x$ and $y$ when $X$ and $Y$ are replaced with $a^\dagg$ and $a$ respectively. That is: $R_{[x,y]}(a^\dagg,a)\upto{d}-\omega(\nabla R_x(a^\dagg,a),\nabla R_y(a^\dagg,a))$;
			\item One has: $\omega(\nabla R_{\sigma x}(a^\dagg,a),\nabla R_{\sigma y}(a^\dagg,a))\upto{d}\sigma\omega(\nabla R_x(a^\dagg,a),\nabla R_y(a^\dagg,a))$.
		\end{enumerate}
	\end{lemma}
	
	\begin{proof}
		Let $x,y\in\mathbb{S}[X,Y]$ be two arbitrary polynomials and $\boldsymbol{\sigma}\in\operatorname{Sp}(2,\C)$ a symplectic matrix. We will now show that claim (a) follows directly from the definition of the reduction map and the symplectic action on a polynomial. For this, start by recalling that the symplectic action is defined by:
		\begin{align*}
			(\sigma x)(X,Y)= \sum_{m=0}^{d_x}\sum_{j=0}^mc_{mj}(\sigma_{11}X+\sigma_{12}Y)^{m-j}(\sigma_{21}X+\sigma_{22}Y)^j.
		\end{align*}
		Thus, calculating the reduction map of the polynomial above yields:
		\begin{align}\label{eqn:help:lem:Igusa:1:1}
			R_{\sigma x}(X,Y)=\sum_{j=0}^{d_x}c_{d_xj}(\sigma_{11}X+\sigma_{12}Y)^{d_x-j}(\sigma_{21}X+\sigma_{22}Y)^j.
		\end{align}
		Now, we observe that the reduction map of $x(X,Y)$ is given by $R_x(X,Y)=\sum_{j=0}^{d_x}c_{d_xj}X^{d_x-j}Y^j$. Acting with the symplectic action $\sigma$ on this polynomial reads now:
		\begin{align}\label{eqn:help:lem:Igusa:1:2}
			(\sigma R_x)(X,Y)=R_x(\sigma_{11}X+\sigma_{12}Y,\sigma_{21}X+\sigma_{22}Y)=\sum_{j=0}^{d_x}c_{d_xj}(\sigma_{11}X+\sigma_{12}Y)^{d_x-j}(\sigma_{21}X+\sigma_{22}Y)^j.
		\end{align}
		Comparing equations \eqref{eqn:help:lem:Igusa:1:1} and \eqref{eqn:help:lem:Igusa:1:2} concludes the proof of claim (a).
		
		To prove claim (b), one starts by computing the commutator $[x(a^\dagg,a),y(a^\dagg,a)]$. The bilineartity of the commutator allows us to write:
		\begin{align*}
			[x(a^\dagg,a),y(a^\dagg,a)]&=\sum_{m=0}^{d_x}\sum_{j=0}^m\sum_{m'=0}^{d_y}\sum_{j'=0}^{m'}c_{mj}\hat{c}_{m'j'}[(a^\dagg)^{m-j}a^j,(a^\dagg)^{m'-j'}a^{j'}].
		\end{align*}
		Next, we can observe that $[a^\alpha,a^\beta]=0=[(a^\dagg)^\alpha,(a^\dagg)^\beta]$ for all $\alpha,\beta\in\N_{\geq0}$. Hence, the identity $[AB,CD]=A[B,C]D+AC[B,D]+[A,C]DB+C[A,D]B$ allows us to find:
		\begin{align*}
			[x(a^\dagg,a),y(a^\dagg,a)]&=\sum_{m=0}^{d_x}\sum_{j=0}^m\sum_{m'=0}^{d_y}\sum_{j'=0}^{m'}c_{mj}\hat{c}_{m'j'}\left((a^\dagg)^{m'-j'}[(a^\dagg)^{m-j},a^{j'}]a^j-(a^\dagg)^{m-j}[(a^\dagg)^{m'-j'},a^j]a^{j'}\right).
		\end{align*}
		Proposition~\ref{prop:help:commutator:q:p} now implies
		\begin{align*}
			[x(a^\dagg,a),y(a^\dagg,a)]&\upto{d}-\sum_{m=0}^{d_x}\sum_{j=0}^m\sum_{m'=0}^{d_y}\sum_{j'=0}^{m'}c_{mj}\hat{c}_{m'j'}\left((m-j)j'-j(m'-j')\right)(a^\dagg)^{m+m'-(j+j')-1}a^{j+j'-1},
		\end{align*}
		where $d=d_x+d_y-2$.
		Now, recall that terms of degree smaller than $d$ do not contribute to the validity of the relation \enquote{$\upto{d}$}. Thus, one can discard terms that are not of maximal degree and one has
		\begin{align*}
			[x(a^\dagg,a),y(a^\dagg,a)]&\upto{d}-\sum_{j=0}^{d_x}\sum_{j'=0}^{d_y}c_{d_xj}\hat{c}_{d_yj'}((d_x-j)j'-j(d_y-j'))(a^\dagg)^{d_x+d_y-(j+j')-1}a^{j+j'-1}=R_{[x,y]}(a^\dagg,a).
		\end{align*}
		Next, we want to compute $\omega(\nabla R_x(a^\dagg,a),\nabla R_y(a^\dagg,a))$. This can be done using the definition of $\omega$ and utilizing Proposition~\ref{prop:product:rule}:
			\begin{align*}
				\omega(\nabla R_x(a^\dagg,a),\nabla R_y(a^\dagg,a))&=(\partial_X R_x(a^\dagg,a))(\partial_Y R_y(a^\dagg,a))-(\partial_Y R_x(a^\dagg,a))(\partial_X R_y(a^\dagg,a))\\
				&=\sum_{j=0}^{d_x}(d_x-j)c_{d_xj}(a^\dagg)^{d_x-j-1}a^j\sum_{j'=0}^{d_y}j'\hat{c}_{d_yj'}(a^\dagg)^{d_y-j'}a^{j'-1}\\
				&\qquad-\sum_{j=0}^{d_x}jc_{d_xj}(a^\dagg)^{d_x-j}a^{j-1}\sum_{j'=0}^{d_y}(d_y-j')\hat{c}_{d_yj'}(a^\dagg)^{d_y-j'-1}a^{j'}.
			\end{align*}
			Now we can use Proposition~\ref{prop:help:commutator:q:p} to furthermore compute
			\begin{align*}
				\omega(\nabla R_x(a^\dagg,a),\nabla R_y(a^\dagg,a))&\upto{d}\sum_{j=0}^{d_x}\sum_{j'=0}^{d_y}c_{d_xj}\hat{c}_{d_yj'}((d_x-j)j'-j(d_y-j'))(a^\dagg)^{d_x+d_y-(j+j')-1}a^{j+j'-1}.
			\end{align*}\noindent
			We can conclude that $R_{[x,y]}(a^\dagg,a)\upto{d}-\omega(\nabla R_x(a^\dagg,a),\nabla R_y(a^\dagg,a))$, proving claim (b). 
			
			Finally, we want to show claim (c). To achieve this, we compute $\omega(\nabla R_{\sigma x}(a^\dagg,a),\nabla R_{\sigma y}(a^\dagg,a))$ using the definition of the symplectic form $\omega$:
			\begin{align*}
				&\omega(\nabla R_{\sigma x}(a^\dagg,a),\nabla R_{\sigma y}(a^\dagg,a))=(\partial_X R_{\sigma x}(a^\dagg,a))(\partial_Y R_{\sigma y}(a^\dagg,a))-(\partial_Y R_{\sigma x}(a^\dagg,a))(\partial_X R_{\sigma y} (a^\dagg,a))\\
				&\qquad=\sum_{j=0}^{d_x}\sum_{j'=0}^{d_y}c_{d_xj}\hat{c}_{d_yj'}\Bigl(\left(\partial_X(\sigma_{11}X+\sigma_{12}Y)^{d_x-j}(\sigma_{21}X+\sigma_{22}Y)^j\right)\left(\partial_Y(\sigma_{11}X+\sigma_{12}Y)^{d_y-j'}(\sigma_{21}X+\sigma_{22}Y)^{j'}\right)\\
				&\qquad\qquad-\left(\partial_Y(\sigma_{11}X+\sigma_{12}Y)^{d_x-j}(\sigma_{21}X+\sigma_{22}Y)^j\right)\left(\partial_X(\sigma_{11}X+\sigma_{12}Y)^{d_y-j'}(\sigma_{21}X+\sigma_{22}Y)^{j'}\right)\Bigr)\Biggr|_{(X,Y)=(a^\dagg,a)}.
			\end{align*}
			Now, we can employ Proposition~\ref{prop:product:rule} to calculate the derivations $\partial_X$ and $\partial_Y$. Substituting $a^\dagg$ for $X$ and $a$ for $Y$ yields:
			\begin{align*}
				&\quad\omega(\nabla R_{\sigma x}(a^\dagg,a),\nabla R_{\sigma y}(a^\dagg,a))\\&=\sum_{j=0}^{d_x}\sum_{j'=0}^{d_y}c_{d_xj}\hat{c}_{d_yj'}\Bigl(\sigma_{11}\sigma_{12}(d_x-j)(d_y-j')(\sigma_{11}a^\dagg+\sigma_{12}a)^{d_x-j-1}(\sigma_{21}a^\dagg+\sigma_{22}a)^j(\sigma_{11}a^\dagg+\sigma_{12}a)^{d_y-j'-1}(\sigma_{21}a^\dagg+\sigma_{22}a)^{j'}\\
				&\qquad\qquad+\sigma_{11}\sigma_{22}(d_x-j)j'(\sigma_{11}a^\dagg+\sigma_{12}a)^{d_x-j-1}(\sigma_{21}a^\dagg+\sigma_{22}a)^j(\sigma_{11}a^\dagg+\sigma_{12}a)^{d_y-j'}(\sigma_{21}a^\dagg+\sigma_{22}a)^{j'-1}\\
				&\qquad\qquad+\sigma_{21}\sigma_{12}j(d_y-j')(\sigma_{11}a^\dagg+\sigma_{12}a)^{d_x-j}(\sigma_{21}a^\dagg+\sigma_{22}a)^{j-1}(\sigma_{11}a^\dagg+\sigma_{12}a)^{d_y-j'-1}(\sigma_{21}a^\dagg+\sigma_{22}a)^{j'}\\
				&\qquad\qquad+\sigma_{21}\sigma_{22}jj'(\sigma_{11}a^\dagg+\sigma_{12}a)^{d_x-j}(\sigma_{21}a^\dagg+\sigma_{22}a)^{-1}(\sigma_{11}a^\dagg+\sigma_{12}a)^{d_y-j'}(\sigma_{21}a^\dagg+\sigma_{22}a)^{j'-1}\\
				& \qquad\qquad-\sigma_{12}\sigma_{11}(d_x-j)(d_y-j')(\sigma_{11}a^\dagg+\sigma_{12}a)^{d_x-j-1}(\sigma_{21}a^\dagg+\sigma_{22}a)^j(\sigma_{11}a^\dagg+\sigma_{12}a)^{d_y-j'-1}(\sigma_{21}a^\dagg+\sigma_{22}a)^{j'}\\
				&\qquad\qquad-\sigma_{12}\sigma_{21}(d_x-j)j'(\sigma_{11}a^\dagg+\sigma_{12}a)^{d_x-j-1}(\sigma_{21}a^\dagg+\sigma_{22}a)^j(\sigma_{11}a^\dagg+\sigma_{12}a)^{d_y-j'}(\sigma_{21}a^\dagg+\sigma_{22}a)^{j'-1}\\
				&\qquad\qquad-\sigma_{22}\sigma_{11}j(d_y-j')(\sigma_{11}a^\dagg+\sigma_{12}a)^{d_x-j}(\sigma_{21}a^\dagg+\sigma_{22}a)^{-1}(\sigma_{11}a^\dagg+\sigma_{12}a)^{d_y-j'-1}(\sigma_{21}a^\dagg+\sigma_{22}a)^{j'}\\
				&\qquad\qquad-\sigma_{22}\sigma_{21}jj'(\sigma_{11}a^\dagg+\sigma_{12}a)^{d_x-j}(\sigma_{21}a^\dagg+\sigma_{22}a)^{j-1}(\sigma_{11}a^\dagg+\sigma_{12}a)^{d_y-j'}(\sigma_{21}a^\dagg+\sigma_{22}a)^{j'-1}\Bigr).
			\end{align*}
			These terms can now be recombined by recalling that $\det(\boldsymbol{\sigma})=\sigma_{11}\sigma_{22}-\sigma_{12}\sigma_{21}=1$ is the defining property of a symplectic matrix $\boldsymbol{\sigma}\in\operatorname{Sp}(2,\C)$. Hence:
			\begin{align*}
				&\quad\omega(\nabla R_{\sigma x}(a^\dagg,a),\nabla R_{\sigma y}(a^\dagg,a))\\
				&=\sum_{j=0}^{d_x}\sum_{j'=0}^{d_y}c_{d_xj}\hat{c}_{d_yj'}\Bigl((d_x-j)j'(\sigma_{11}a^\dagg+\sigma_{12}a)^{d_x-j-1}(\sigma_{21}a^\dagg+\sigma_{22}a)^j(\sigma_{11}a^\dagg+\sigma_{12}a)^{d_y-j'}(\sigma_{21}a^\dagg+\sigma_{22}a)^{j'-1}\\
				&\qquad\qquad-j(d_y-j')(\sigma_{11}a^\dagg+\sigma_{12}a)^{d_x-j}(\sigma_{21}a^\dagg+\sigma_{22}a)^{j-1}(\sigma_{11}a^\dagg+\sigma_{12}a)^{d_y-j'-1}(\sigma_{21}a^\dagg+\sigma_{22}a)^{j'}\Bigr).
			\end{align*}
			Utilizing Proposition~\ref{prop:help:commutator:q:p} yields finally
			{\small
				\begin{align*}
					\omega(\nabla R_{\sigma x}(a^\dagg,a),\nabla R_{\sigma y}(a^\dagg,a))&\upto{d}\sum_{j=0}^{d_x}\sum_{j'=0}^{d_y}c_{d_xj}\hat{c}_{d_yj'}\left((d_x-j)j'-j(d_y-j')\right)(\sigma_{11}a^\dagg+\sigma_{12}a)^{d_x+d_y-(j+j')-1}(\sigma_{21}a^\dagg+\sigma_{22}a)^{j+j'-1}\\
					&\upto{d}\sigma\omega(\nabla R_x(a^\dagg,a),\nabla R_y(a^\dagg,a)),
			\end{align*}}\noindent
			which is the desired result, and therefore completes the proof of claim (c).
		\end{proof}
		
		\begin{lemma}\label{lem:Igusa:2}
			Let $x\in\mathbb{S}[X,Y]$ be a polynomial of degree $d_x$ satisfying $x-R_x=0$, and let $\boldsymbol{\sigma}=(\boldsymbol{\xi}|\boldsymbol{\eta})\in\operatorname{Sp}(2,\C)$ be an arbitrary symplectic matrix. Then, one has:
			\begin{align*}
				(\sigma x)(a^\dagg,a)\upto{d_x}x(\boldsymbol{\eta})a^{d_x}+( \nabla x(\boldsymbol{\eta}))^\Tp\boldsymbol{\xi}a^\dagg a^{d_x-1}+\mathcal{O}\left((a^\dagg)^2a^{d_x-2}+\ldots+(a^\dagg)^{d_x}\right),
			\end{align*}
			where the remainder $\mathcal{O}((a^\dagg)^2a^{d_x-2}+\ldots+(a^\dagg)^{d_x})$ contains only terms of the form $(a^\dagg)^j a^{d_x-j}$ with $j\in\{2,\ldots,d_x\}$.
		\end{lemma}
		
		\begin{proof}
			Let $x\in\mathbb{S}[X,Y]$ be a polynomial of degree $d_x$ satisfying $x-R_x=0$, i.e., let $x$ be a homogeneous polynomial of degree $d_x$. Let furthermore $\boldsymbol{\sigma}=(\boldsymbol{\xi}|\boldsymbol{\eta})\in\operatorname{Sp}(2,\C)$. Then, by the definition of the symplectic action on polynomials, one has:
			\begin{align*}
				(\sigma x)(a^\dagg,a)&=\sum_{j=0}^{d_x}c_{d_xj}\left(\xi_1a^\dagg+\eta_1a\right)^{d_x-j}\left(\xi_2a^\dagg+\eta_2a\right)^j.
			\end{align*}
			We can expand the two terms $(\xi_1a^\dagg+\eta_1a)^{d_x-j}$ and $(\xi_2a^\dagg+\eta_2a)^j$ individually, and then group terms with the same number of creation and annihilation operators together. This yields: 
			\begin{align*}
				(\sigma x)(a^\dagg,a)&=\sum_{j=0}^{d_x}c_{d_xj}\left(\xi_1^{d_x-j}(a^\dagg)^{d_x-j}+\ldots+\xi_1\eta_1^{d_x-j-1}\left(a^\dagg a^{d_x-j-1}+\ldots+a^{d_x-j-1}a^\dagg\right)+\eta_1^{d_x-j}a^{d_x-j}\right)\\
				&\qquad\qquad\times\left(\xi_2^j(a^\dagg)^j+\ldots+\xi_2\eta_2^{j-1}\left(a^\dagg a^{j-1}+\ldots+a^{j-1}a^\dagg\right)+\eta_2^ja^j\right).
			\end{align*}
			Applying Proposition~\ref{prop:help:commutator:q:p}, one can show that $a^\dagg a^\ell+a a^\dagg a^{\ell-1}+\ldots+a^{\ell-1}a^\dagg a+a^\ell a^\dagg\upto{\ell+1}(\ell+1)a^\dagg a^\ell$. It follows that
			{\small
				\begin{align*}
					(\sigma x)(a^\dagg,a)&\upto{d_x}\sum_{j=0}^{d_x}c_{d_xj}\left(\xi_1^{d_x-j}(a^\dagg)^{d_x-j}+\ldots+(d_x-j)\xi_1\eta_1^{d_x-j-1}a^\dagg a^{d_x-j-1}+\eta_1^{d_x-j}a^{d_x-j}\right)\left(\xi_2^j(a^\dagg)^j+\ldots+j\xi_2\eta_2^{j-1}a^\dagg a^{j-1}+\eta_2^{j}a^j\right).
			\end{align*}}\noindent
			We can now multiply the two brackets together, group terms with the same number of creation and annihilation operators together, and use Proposition~\ref{prop:help:commutator:q:p} to normal order the operators. We find:
			\begin{align}
				(\sigma x)(a^\dagg,a)&\upto{d_x}\sum_{j=0}^{d_x}c_{d_xj}\left(\xi_1^{d_x-j}\xi_2^j(a^\dagg)^{d_x}+\ldots+\left((d_x-j)\xi_1\eta_1^{d_x-j-1}\eta_2^{j}+j\xi_2\eta_1^{d_x-j}\eta_2^{j-1}\right)a^\dagg a^{d_x-1}+\eta_1^{d_x-j}\eta_2^ja^{d_x}\right).\label{eqn:help:Lem:Igusa:2:1}
			\end{align}\noindent
			One observes furthermore
			\begin{align*}
				x(\boldsymbol{\eta})a^{d_x}=\sum_{j=0}^{d_x}c_{d_xj}\eta_1^{d_x-j}\eta_2^{j}a^{d_x},\;\;\text{and}\;\;\nabla x(\boldsymbol{\eta})&=\sum_{j=0}^{d_x}c_{d_xj}\left.\begin{pmatrix}
					(d_x-j)X^{d_x-j-1}Y^j\\
					jX^{d_x-j}Y^{j-1}
				\end{pmatrix}\right|_{(X,Y)=\boldsymbol{\eta}}=\sum_{j=0}^{d_x}c_{d_xj}\eta_1^{d_x-j-1}\eta_2^{j-1}\begin{pmatrix}
					(d_x-j)\eta_2\\
					j\eta_1
				\end{pmatrix}.
			\end{align*}
			Thus, one has:
			\begin{align*}
				(\nabla x(\boldsymbol{\eta}))^\Tp\boldsymbol{\xi}a^\dagg a^{d_x-1}&=\sum_{j=0}^{d_x}c_{d_xj}\left((d_x-j)\xi_1\eta_1^{d_x-j-1}\eta_2^j+j\xi_2\eta_1^{d_x-j}\eta_2^{j-1}\right)a^\dagg a^{d_x-1},
			\end{align*}
			and henceforth, by combining the results for $x(\boldsymbol{\eta})$ and $( \nabla x(\boldsymbol{\eta}))^\Tp\boldsymbol{\xi}a^\dagg a^{d_x-1}$, we find:
			\begin{align*}
				x(\boldsymbol{\eta})a^{d_x}+( \nabla x(\boldsymbol{\eta}))^\Tp\boldsymbol{\xi}a^\dagg a^{d_x-1}=\sum_{j=0}^{d_x}c_{d_xj}\left(\left((d_x-j)\xi_1\eta_1^{d_x-j-1}\eta_2^{j}+j\xi_2\eta_1^{d_x-j}\eta_2^{j-1}\right)a^\dagg a^{d_x-1}+\eta_1^{d_x-j}\eta_2^ja^{d_x}\right).
			\end{align*}
			Comparing this with equation \eqref{eqn:help:Lem:Igusa:2:1}, yields
			\begin{align*}
				(\sigma x)(a^\dagg,a)\upto{d_x}x(\boldsymbol{\eta})a^{d_x}+( \nabla x(\boldsymbol{\eta}))^\Tp\boldsymbol{\xi}a^\dagg a^{d_x-1}+\mathcal{O}\left((a^\dagg)^2a^{d_x-2}+\ldots+(a^\dagg)^{d_x}\right),
			\end{align*}
			as claimed.
		\end{proof}
		
		\begin{lemma}\label{lem:Igusa:3}
			Let $x,y\in\mathbb{S}[X,Y]$ be two polynomials of degree greater than zero. If there exists a vector $\boldsymbol{\eta}\in\C^2$ such that the following conditions apply:
			\begin{enumerate}[label =(\roman*)]
				\item $R_x(\boldsymbol{\eta})R_y(\boldsymbol{\eta})\neq0$,
				\item $\omega(\nabla R_x,\nabla R_y)(\boldsymbol{\eta})\neq0$,
				\item The set $\{\boldsymbol{J}(\nabla R_x(\boldsymbol{\eta})),\boldsymbol{J}(\nabla R_y(\boldsymbol{\eta})),\boldsymbol{\eta}\}$ is linearly dependent.
			\end{enumerate}
			Then, there exists a symplectic matrix $\boldsymbol{\sigma}=(\boldsymbol{\xi}|\boldsymbol{\eta})\in\operatorname{Sp}(2,\C)$ such that
			\begin{align*}
				\sigma R_x(a^\dagg,a)&\upto{d_x}a_0a^{d_x}+a_1a^\dagg a^{d_x-1}+\mathcal{O}\left((a^\dagg)^2a^{d_x-2}+\ldots+(a^\dagg)^{d_x}\right),\\
				\sigma R_y(a^\dagg,a)&\upto{d_y}b_0a^{d_x}+b_1a^\dagg a^{d_x-1}+\mathcal{O}\left((a^\dagg)^2a^{d_x-2}+\ldots+(a^\dagg)^{d_x}\right),
			\end{align*}
			where $a_0b_0\neq0$ and $d_xa_0b_1-d_ya_1b_0\neq0$.
		\end{lemma}
		
		\begin{proof}
			Let $x,y\in\mathbb{S}[X,Y]$ be two polynomials satisfying the three conditions above for some fixed $\boldsymbol{\eta}\in\C^2$. The goal is to construct a symplectic matrix $\boldsymbol{\sigma}=(\boldsymbol{\xi}|\boldsymbol{\eta})\in\operatorname{Sp}(2,\C)$ explicitly such that the transformed polynomials $\sigma R_x$ and $\sigma R_y$, when evaluated at $(a^\dagg,a)$, have the desired leading-order structure. 
			
			We start by noting that condition (i) implies $R_x(\boldsymbol{\eta})\neq0$. Thus, we can define the vector
			\begin{align*}
				\boldsymbol{\xi}:=\frac{1}{d_x R_x(\boldsymbol{\eta})}\boldsymbol{J}(\nabla R_x(\boldsymbol{\eta})),
			\end{align*}
			and compute:
			\begin{align*}
				\omega(\boldsymbol{\xi},\boldsymbol{\eta})&=\boldsymbol{\xi}^{\Tp}\boldsymbol{J}\boldsymbol{\eta}=\frac{1}{d_x R_x(\boldsymbol{\eta})}(\nabla R_x(\boldsymbol{\eta}))^{\Tp}\boldsymbol{J}^{\Tp}\boldsymbol{J}\boldsymbol{\eta}=\frac{1}{d_x R_x(\boldsymbol{\eta})}\begin{pmatrix}
					\sum_{j=0}^{d_x}c_{d_xj}(d_x-j)\eta_1^{d_x-j-1}\eta_2^j\\
					\sum_{j=0}^{d_x}c_{d_xj} j\eta_1^{d_x-j}\eta_2^{j-1}
				\end{pmatrix}^{\Tp}\begin{pmatrix}
					\eta_1\\
					\eta_2
				\end{pmatrix}\\
				&=\frac{1}{d_x R_x(\boldsymbol{\eta})}\sum_{j=0}^{d_x}c_{d_xj}\eta_1^{d_x-j}\eta_2^j(d_x-j+j)=1.
			\end{align*}
			The combined matrix $\boldsymbol{\sigma}:=(\boldsymbol{\xi}|\boldsymbol{\eta})$ is consequently a symplectic matrix, since $\boldsymbol{\sigma}^{\Tp}\boldsymbol{J}\boldsymbol{\sigma}=\boldsymbol{J}\omega(\boldsymbol{\xi},\boldsymbol{\eta})=\boldsymbol{J}$. Condition (iii) implies now that $\boldsymbol{\xi}$, $\boldsymbol{\eta}$, and $\boldsymbol{J}(\nabla R_y(\boldsymbol{\eta}))$ are linearly dependent. Thus, one can write $\boldsymbol{J}(\nabla R_y(\boldsymbol{\eta}))=\alpha\boldsymbol{\xi}+\beta\boldsymbol{\eta}$ for some coefficients $\alpha,\beta\in\C$. To compute those coefficients, we need to calculate $\boldsymbol{J}(\nabla R_y(\boldsymbol{\eta}))$ explicitly:
			\begin{align}
				\boldsymbol{J}(\nabla R_y(\boldsymbol{\eta}))=\begin{pmatrix}
					0&1\\
					-1&0
				\end{pmatrix}\begin{pmatrix}
					\sum_{j=0}^{d_y}\hat{c}_{d_yj}(d_y-j)\eta_1^{d_y-j-1}\eta_2^j\\
					\sum_{j=0}^{d_y}\hat{c}_{d_yj}j\eta_1^{d_y-j}\eta_2^{j-1}
				\end{pmatrix}=\sum_{j=0}^{d_y}\hat{c}_{d_yj}\eta_1^{d_y-j-1}\eta_2^{j-1}\begin{pmatrix}
					j\eta_1\\
					(j-d_y)\eta_2
				\end{pmatrix}.\label{eqn:help:lem:Igusa:3:1}
			\end{align}
			To express $\boldsymbol{J}(\nabla R_y(\boldsymbol{\eta}))$ in terms of $\boldsymbol{\xi}$ and $\boldsymbol{\eta}$, recall that $\boldsymbol{\xi}$ is proportional to $\boldsymbol{J}(\nabla R_x(\boldsymbol{\eta}))$. As will be apparent later, it is helpful to compute:
			\begin{subequations}
				\begin{alignat}{2}
					R_y(\boldsymbol{\eta})\boldsymbol{J}(\nabla R_x(\boldsymbol{\eta}))&=\sum_{j=0}^{d_x}\sum_{j'=0}^{d_y}c_{d_xj}\hat{c}_{d_yj'}\eta_1^{d_y-j'}\eta_2^{j'}\eta_1^{d_x-j-1}\eta_2^{j-1}\begin{pmatrix}
						j\eta_1\\
						(j-d_x)\eta_2
					\end{pmatrix}\nonumber\\
					&=\sum_{j=0}^{d_x}c_{d_xj}\eta_1^{d_x-j}\eta_2^{j}\sum_{j'=0}^{d_y}\hat{c}_{d_yj'}\eta_1^{d_y-j'-1}\eta_2^{j'-1}\begin{pmatrix}
						j\eta_1\\
						(j-d_x)\eta_2
					\end{pmatrix},\label{eqn:help:lem:Igusa:3:2}\\
					(\nabla R_x(\boldsymbol{\eta}))^{\Tp}\boldsymbol{J}(\nabla R_y(\boldsymbol{\eta}))&=\sum_{j=0}^{d_x}\sum_{j'=0}^{d_y}c_{d_xj}\hat{c}_{d_yj'}\eta_1^{d_x-j-1}\eta_2^{j-1}\eta_1^{d_y-j'-1}\eta_2^{j'-1}\begin{pmatrix}
						(d_x-j)\eta_2\\
						j\eta_1
					\end{pmatrix}^{\Tp}\begin{pmatrix}
						j'\eta_1\\
						(d_y-j')\eta_2
					\end{pmatrix}\nonumber\\
					&=\sum_{j=0}^{d_x}c_{d_xj}\eta_1^{d_x-j}\eta_2^{j}\sum_{j'=0}^{d_y}\hat{c}_{d_yj'}\eta_1^{d_y-j'-1}\eta_2^{j'-1}\left((d_x-j)j'-j(d_y-j')\right).\label{eqn:help:lem:Igusa:3:3}
				\end{alignat}
			\end{subequations}	
			Combining \eqref{eqn:help:lem:Igusa:3:2} and \eqref{eqn:help:lem:Igusa:3:3} by multiplying appropriate terms yields:
			\begin{align*}
				d_yR_y(\boldsymbol{\eta})\boldsymbol{J}(\nabla R_x(\boldsymbol{\eta}))+(\nabla R_x(\boldsymbol{\eta}))^{\Tp}\boldsymbol{J}(\nabla R_y(\boldsymbol{\eta}))\boldsymbol{\eta}&=d_x\sum_{j=0}^{d_x}c_{d_xj}\eta_1^{d_x-j}\eta_2^{j}\sum_{j'=0}^{d_y}\hat{c}_{d_yj'}\eta_1^{d_y-j'-1}\eta_2^{j'-1}\begin{pmatrix}
					j'\eta_1\\
					(j'-d_y)\eta_2
				\end{pmatrix}
			\end{align*}
			Comparing this with equation \eqref{eqn:help:lem:Igusa:3:1} allows us to write
			\begin{align*}
				d_yR_y(\boldsymbol{\eta})\boldsymbol{J}(\nabla R_x(\boldsymbol{\eta}))+(\nabla R_x(\boldsymbol{\eta}))^{\Tp}\boldsymbol{J}(\nabla R_y(\boldsymbol{\eta}))\boldsymbol{\eta}&=d_x R_x(\boldsymbol{\eta})\boldsymbol{J}(\nabla R_y(\boldsymbol{\eta}))
			\end{align*}
			and consequently:
			\begin{align*}
				\boldsymbol{J}(\nabla R_y(\boldsymbol{\eta}))=\frac{d_yR_y(\boldsymbol{\eta})\boldsymbol{J}(\nabla R_x(\boldsymbol{\eta}))+(\nabla R_x(\boldsymbol{\eta}))^{\Tp}\boldsymbol{J}(\nabla R_y(\boldsymbol{\eta}))\boldsymbol{\eta}}{d_xR_x(\boldsymbol{\eta})}=d_y R_y(\boldsymbol{\eta}) \boldsymbol{\xi}+\frac{1}{d_xR_x(\boldsymbol{\eta})}\omega(\nabla R_x(\boldsymbol{\eta}),\nabla R_y(\boldsymbol{\eta}))\boldsymbol{\eta}.
			\end{align*}
			Thus, we can make the identifications
			\begin{align*}
				\alpha=d_yR_y(\boldsymbol{\eta}),\qquad\text{and}\qquad\beta=\frac{1}{d_xR_x(\boldsymbol{\eta})}\omega(\nabla R_x(\boldsymbol{\eta}),\nabla R_y(\boldsymbol{\eta})).
			\end{align*}
			Lemma~\ref{lem:Igusa:2} now implies
			\begin{align*}
				\sigma R_x(a^\dagg,a)&\upto{d_x}R_x(\boldsymbol{\eta})a^{d_x}+(\nabla R_x(\boldsymbol{\eta}))^{\Tp}\boldsymbol{\xi}a^\dagg a^{d_x-1}+\mathcal{O}\left((a^\dagg)^2a^{d_x-2}+\ldots+(a^\dagg)^{d_x}\right),\\
				\sigma R_y(a^\dagg,a)&\upto{d_y}R_y(\boldsymbol{\eta})a^{d_y}+(\nabla R_y(\boldsymbol{\eta}))^{\Tp}\boldsymbol{\xi}a^\dagg a^{d_y-1}+\mathcal{O}\left((a^\dagg)^2a^{d_y-2}+\ldots+(a^\dagg)^{d_y}\right).
			\end{align*}
			Condition (i) also implies that $a_0:=R_x(\boldsymbol{\eta})\neq0\neq R_y(\boldsymbol{\eta}):=b_0$. The anti-symmetric property of the symplectic form $\omega$ (i.e., $\omega(\boldsymbol{\xi},\boldsymbol{\eta})=-\omega(\boldsymbol{\eta},\boldsymbol{\xi})$) allows us furthermore to identify:
			\begin{align*}
				(\nabla R_y(\boldsymbol{\eta}))^{\Tp}\boldsymbol{\xi}&=\frac{1}{d_xR_x(\boldsymbol{\eta})}(\nabla R_y(\boldsymbol{\eta}))^{\Tp}\boldsymbol{J}(\nabla R_x(\boldsymbol{\eta}))=\frac{\omega(\nabla R_y(\boldsymbol{\eta}),\nabla R_x(\boldsymbol{\eta}))}{d_x R_x(\boldsymbol{\eta})}=-\frac{\omega(\nabla R_x(\boldsymbol{\eta}),\nabla R_y(\boldsymbol{\eta}))}{d_xR_x(\boldsymbol{\eta})}=-\beta=:b_1,\\
				(\nabla R_x(\boldsymbol{\eta}))^{\Tp}\boldsymbol{\xi}&=\frac{1}{d_x R_x(\boldsymbol{\eta})}(\nabla R_x(\boldsymbol{\eta}))^{\Tp}\boldsymbol{J}(\nabla R_x(\boldsymbol{\eta}))=\frac{\omega(\nabla R_x(\boldsymbol{\eta}),\nabla R_x(\boldsymbol{\eta}))}{d_xR_x(\boldsymbol{\eta})}=0=:a_1.
			\end{align*}
			Conditions (i) and (ii) guarantee now that $\beta$ is a non-zero coefficient. Thus, one has
			\begin{align*}
				\sigma R_x(a^\dagg,a)&\upto{d_x}a_0a^{d_x}+a_1a^\dagg a^{d_x-1}+\mathcal{O}\left((a^\dagg)^2a^{d_x-2}+\ldots+(a^\dagg)^{d_x}\right),\\
				\sigma R_y(a^\dagg,a)&\upto{d_y}b_0a^{d_x}+b_1a^\dagg a^{d_x-1}+\mathcal{O}\left((a^\dagg)^2a^{d_x-2}+\ldots+(a^\dagg)^{d_x}\right),
			\end{align*}
			where $a_0b_0\neq0$ and $d_xa_0b_1-d_ya_1b_0=d_xa_0b_1=-\omega(\nabla R_x(\boldsymbol{\eta}),\nabla R_y(\boldsymbol{\eta}))\neq0$, due to condition (ii).
		\end{proof}
		
		In Lemma~\ref{lem:Igusa:3}, we considered two polynomials $x,y\in\mathbb{S}[X,Y]$ and showed that, under certain constraints, they can be transformed into a specific form that will be significant for the results that follow. A key property of these transformed polynomials is that the inequality $d_xa_0b_1-d_ya_1b_0\neq 0$ holds, where the coefficients $d_x,a_0,a_1,d_y,b_0,b_1$ are determined by the transformed polynomials. This inspires the following definition:
		
		\begin{definition}\label{def:delta:function:appendix}
			Let $x,y\in\C[X,Y]$ be two polynomials of degree $d_x$ and $d_y$ respectively, satisfying:
			\begin{align*}
				x(a^\dagg,a)\upto{d_x}\sum_{m=0}^{d_x}a_m (a^\dagg)^m a^{d_x-m},\qquad\text{and}\qquad y(a^\dagg,a)\upto{d_y}\sum_{m=0}^{d_y}b_m (a^\dagg)^m a^{d_x-m}.
			\end{align*}
			We define the function $\delta:\C[X,Y]\times\C[X,Y]\to\C,$ via the action $(x,y)\mapsto \delta(x,y):= d_ya_1b_0-d_xa_0b_1$.
		\end{definition}
		
		\begin{lemma}\label{lem:Igusa:4}
			Let $x,y\in\C[X,Y]$ be two polynomials of degree $d_x$ and $d_y$ respectively, that admit the following expansions
			\begin{subequations}
				\begin{align}
					R_x(a^\dagg,a)&\upto{d_x}a_0a^{d_x}+a_1a^\dagg a^{d_x-1}+\mathcal{O}\left((a^\dagg)^2a^{d_x-2}+\ldots (a^\dagg)^{d_x}\right),\label{eqn:lem:Igusa_4:1}
				\end{align}
				\begin{align}
					R_y(a^\dagg,a)&\upto{d_y}b_0a^{d_y}+b_1a^\dagg a^{d_y-1}+\mathcal{O}\left((a^\dagg)^2a^{d_y-2}+\ldots+(a^\dagg)^{d_y}\right).\label{eqn:lem:Igusa_4:2}
				\end{align}
			\end{subequations}
			Let the function $\delta:\C[X,Y]\times\C[X,Y]\to\C$ be defined as in Definitions~\ref{def:delta:function:appendix}. Then, the commutator satisfies:
			\begin{align*}
				[x,y](a^\dagg,a)\upto{d}-\delta(x,y)a^{d}+\mathcal{O}\left(a^\dagg a^{d-1}+\ldots (a^\dagg)^{d}\right),
			\end{align*}
			where $d:=d_x+d_y-2$.
		\end{lemma}
		
		\begin{proof}
			This claim can be verified by a simple calculation. First note that for any polynomial $z(a^\dagg,a)$, one has clearly $z(a^\dagg,a)\upto{\deg(z)}R_z(a^\dagg,a)$. Thus, for two polynomials $x,y\in\C[X,Y]$ of degree $d_x$ and $d_y$ respectively that posses the expansions \eqref{eqn:lem:Igusa_4:1} and \eqref{eqn:lem:Igusa_4:2}, we can write $[x,y](a^\dagg,a)\upto{d}R_{[x,y]}(a^\dagg,a)$, since Theorem 1.1 from \plainrefs{Coutinho:1995} implies $\deg([x,y](a^\dagg,a)])\leq \deg(x(a^\dagg,a))+\deg(y(a^\dagg,a))-2=d_x+d_y-2=d$, and $z_1(a^\dagg,a)\upto{f_1}z_2(a^\dagg,a)$ for polynomials $z_1,z_2\in\C[X,Y]$ implies $z_1(a^\dagg,a)\upto{f_2}z_2(a^\dagg,a)$ for all $f_2\geq f_1$. The application of Lemma~\ref{lem:Igusa:1} yields:
			\begin{align*}
				[x,y](a^\dagg,a)&\upto{d}R_{[x,y]}(a^\dagg,a)\upto{d}-\omega(\nabla R_x(a^\dagg,a),\nabla R_y(a^\dagg,a)).\\
			\end{align*}
			The definition of the symplectic form $\omega$ and the gradient allows us to find:
			\begin{align*}
				[x,y](a^\dagg,a)&\upto{d}-(\nabla R_x(a^\dagg,a))^{\Tp}\boldsymbol{J}(\nabla R_y(a^\dagg,a))\\
				&=-\sum_{j=0}^{d_x}\sum_{j'=0}^{d_y}c_{d_xj}\hat{c}_{d_yj'}\begin{pmatrix}
					(d_x-j)(a^\dagg)^{d_x-j-1}a^j\\
					j(a^\dagg)^{d_x-j}a^{j-1}
				\end{pmatrix}^{\Tp}\begin{pmatrix}
					0&1\\
					-1&0
				\end{pmatrix}\begin{pmatrix}
					(d_y-j')(a^\dagg)^{d_y-j'-1}a^{j'}\\
					j'(a^\dagg)^{d_y-j'}a^{j'-1}
				\end{pmatrix}.
			\end{align*}
			Performing the matrix multiplication and employing Proposition~\ref{prop:help:commutator:q:p} allows us to obtain
			\begin{align*}
				[x,y](a^\dagg,a)&\upto{d}-\sum_{j=0}^{d_x}\sum_{j'=0}^{d_y}c_{d_xj}\hat{c}_{d_yj'}\left((d_x-j)j'-j(d_y-j')\right)(a^\dagg)^{d_x+d_y-(j+j')-1}a^{j+j'-1}.
			\end{align*}
			We can focus on the terms proportional to $a^d$. In other words every term in the sum above that is proportional to a term proportional $(a^\dagg)^\ell a^{d-\ell}$ with $\ell\geq 1$ can be discarded into $\mathcal{O}\left(a^\dagg a^{d-1}+\ldots+(a^\dagg)^d\right)$. This leaves us with 
			\begin{align*}
				[x,y](a^\dagg,a)&\upto{d}-(c_{d_x,d_x-1}\hat{c}_{d_yd_y}d_y-c_{d_xd_x}\hat{c}_{d_y,d_y-1}d_x)a^{d}+\mathcal{O}\left(a^\dagg a^{d-1}+\ldots+(a^\dagg)^d\right).
			\end{align*}
			The coefficients $c_{d_x,d_x-1}$, $\hat{c}_{d_yd_y}$, $c_{d_xd_x}$, and $\hat{c}_{d_y,d_y-1}$ are, by equations \eqref{eqn:lem:Igusa_4:1} and \eqref{eqn:lem:Igusa_4:2}, defined as $c_{d_x,d_x-1}\equiv a_1$, $\hat{c}_{d_yd_y}\equiv b_0$, $c_{d_xd_x}\equiv a_0$, and $\hat{c}_{d_y,d_y-1}\equiv b_1$, which prompts the result:
			\begin{align*}
				[x,y](a^\dagg,a)&\upto{d}-(a_1b_0d_y-a_0b_1d_x)a^{d}+\mathcal{O}\left(a^\dagg a^{d-1}+\ldots+(a^\dagg)^d\right).
			\end{align*}
			Recalling the definition of $\delta$ we have
			\begin{align*}
				[x,y](a^\dagg,a)\upto{d}-\delta(x,y)a^d+\mathcal{O}\left(a^\dagg a^{d-1}+\ldots+(a^\dagg)^d\right),
			\end{align*}
			which concludes the proof.
		\end{proof}
		
		\begin{lemma}\label{lem:Igusa:5}
			Let $x,y\in\C[X,Y]$ be two polynomials of degree $d_x>2$ and $d_y>2$ respectively that admit the following expansions:
			\begin{align*}
				R_x(a^\dagg,a)&\upto{d_x}a_0a^{d_x}+a_1a^\dagg a^{d_x-1}+\mathcal{O}\left((a^\dagg)^2a^{d_x-2}+\ldots (a^\dagg)^{d_x}\right),\\
				R_y(a^\dagg,a)&\upto{d_y}b_0a^{d_y}+b_1a^\dagg a^{d_y-1}+\mathcal{O}\left((a^\dagg)^2a^{d_y-2}+\ldots+a^{d_y}\right),
			\end{align*}
			where $a_0b_0\neq0$ and $\delta(x,y)\neq0$. Then $x$ and $y$ generate an infinite-dimensional Lie algebra.
		\end{lemma}
		
		\begin{proof}
			Let $x,y$ be as stated in the lemma. We now consider the generic Commutator Chain $C^{\mathrm{gen}}$ with seed element $u^{(0)}=y(a^\dagg,a)$ and auxiliary chain
			\begin{align*}
				s^{(\ell)}(a^\dagg,a)&=\left\{\begin{matrix}
					x(a^\dagg,a),&\quad&\text{if }\delta(u^{(\ell)},x)\neq0,\\
					y(a^\dagg,a),&\quad&\text{if }\delta(u^{(\ell)},x)=0\text{ and }\delta(u^{(\ell)},y)\neq0,\\
					0,&\quad&\text{otherwise}
				\end{matrix}\right.\qquad\text{ for all $\ell\in\N_{\geq0}$}.
			\end{align*}
			Note that one can associate to each chain element $u^{(\ell)}\in C^{\mathrm{gen}}$ a polynomial $u^{(\ell)}(X,Y)\in\C[X,Y]$ by replacing $a^\dagg$ with the variable $X$ and $a$ with the variable $Y$. 
			
			The goal is now to verify that the chain elements $u^{(\ell)}$ are of strictly increasing degree $d^{(\ell)}$. This would imply that $\lie{\{x(a^\dagg,a),y(a^\dagg,a)\}}$ is infinite dimensional since all $u^{(\ell)}$ would be linearly independent and lie within this Lie algebra.
			To do this, we will now demonstrate inductively that the chain elements $u^{(\ell)}$ satisfy
			\begin{align*}
				R_{u^{(\ell)}}\upto{d^{(\ell)}}-\delta(u^{(\ell-1)},s^{(\ell-1)})a^{d^{(\ell)}}+c_1^{(\ell)}a^\dagg a^{d^{(\ell)}-1}+\mathcal{O}\left((a^\dagg)^2a^{d^{(1)}-2}+\ldots+(a^\dagg)^{d^{(1)}}\right),
			\end{align*}
			where $c_1^{(\ell)}\in\C$, $s^{(\ell-1)}\in \{x,y\}$, $\delta(u^{(\ell-1)},s^{(\ell-1)})\neq 0$, and $d^{(\ell)}=d^{(\ell-1)}+\deg(s^{(\ell-1)})-2>d^{(\ell-1)}$ for all $\ell\in\N_{\geq0}$.
			
			Let us start with the base case $\ell=1$. By assumption, we have $\delta(u^{(0)},x)=\delta(y,x)=-\delta(x,y)\neq 0$ and consequently $s^{(0)}(a^\dagg,a)=x(a^\dagg,a)$, since $\delta$ is antisymmetric. Thus, by virtue of Lemma~\ref{lem:Igusa:4}, we have
			\begin{align*}
				R_{u^{(1)}}=R_{[y,x]}(a^\dagg,a)\upto{d^{(1)}}-\delta(y,x)a^{d^{(1)}}+c_1a^\dagg a^{d^{(1)}-1}+\mathcal{O}\left((a^\dagg)^2a^{d^{(1)}-2}+\ldots+(a^\dagg)^{d^{(1)}}\right),
			\end{align*}
			where $d^{(1)}=d_x+d_y-2$ and $c_1\in\C$.
			
			With the base case established, let us proceed with the induction step. Thus, assume for one $\ell\in\N_{\geq0}$, the chain element $u^{(\ell)}$ is such that 
			\begin{align*}
				R_{u^{(\ell)}}\upto{d^{(\ell)}}-\delta(u^{(\ell-1)},s^{(\ell-1)})a^{d^{(\ell)}}+c_1^{(\ell)}a^\dagg a^{d^{(\ell)}-1}+\mathcal{O}\left((a^\dagg)^2a^{d^{(1)}-2}+\ldots+(a^\dagg)^{d^{(1)}}\right),
			\end{align*}
			where $c_1^{(\ell)}\in\C$, $s^{(\ell-1)}\in \{x,y\}$, $\delta(u^{(\ell-1)},s^{(\ell-1)})\neq 0$, and $d^{(\ell)}=d^{(\ell-1)}+\deg(s^{(\ell-1)})-2>d^{(\ell-1)}$.  We start by showing that $s^{(\ell)}\neq 0$. Thus, by the definition of the auxiliary chain elements $s^{(\ell)}$, we need to compute
			\begin{align}
				\delta(u^{(\ell)},x)=d_xa_0c_1^{(\ell)}-d^{(\ell)}\delta(s^{(\ell-1)},u^{(\ell-1)})a_1,
				\qquad
				\text{and}
				\qquad
				\delta(u^{(\ell)},y)=d_yb_0c_1^{(\ell)}-d^{(\ell)}\delta(s^{(\ell-1)},u^{(\ell-1)})b_1.\label{eqn:help:lem:Igusa:5:1}
			\end{align}
			To suppose the converse, assume $s^{(\ell)}=0$, i.e., both expressions above vanish. This allows us to obtain, by subtracting and adding these equations together, the following constraints
			\begin{align*}
				0=\left(d_xa_0+d_yb_0\right)c_1^{(\ell)}-d^{(\ell)}\delta(s^{(\ell-1)},u^{(\ell-1)})\left(a_1+b_1\right)\qquad\text{and}\qquad0=\left(d_xa_0-d_yb_0\right)c_1^{(\ell)}-d^{(\ell)}\delta(s^{(\ell-1)},u^{(\ell-1)})\left(a_1-b_1\right). 
			\end{align*}
			We aim to show that these constraints yield a contradiction to the initial assumption. To accomplish this task, we first need to show that these constraints demand $c_1^{(\ell)}\neq 0$ to hold. To prove this claim, suppose $c_1^{(\ell)}$ vanishes. The induction hypothesis $\delta(u^{(\ell-1)},s^{(\ell-1)})\neq 0$ and $d^{(\ell)}>d^{(\ell-1)}>\ldots>d^{(1)}>d^{(0)}=d_y\neq 0$ imply now $a_1=b_1$ and $a_1=-b_1$. Consequently $a_1=b_1=0$ which violates the initial assumption $\delta(x,y)=d_y a_1b_0-d_xa_0b_1\neq 0$. Thus $c_1^{(\ell)}\neq 0$ and we can rewrite equations \eqref{eqn:help:lem:Igusa:5:1}, under the assumption that $\delta(u^{(\ell)},s^{(\ell)})$ vanishes, as
			\begin{align*}
				a_0=\frac{d^{(\ell)}}{d_x}\frac{\delta(s^{(\ell-1)},u^{(\ell-1)})}{c_1^{(\ell)}}a_1\qquad\text{and}\qquad b_0=\frac{d^{(\ell)}}{d_y}\frac{\delta(s^{(\ell-1)},u^{(\ell-1)})}{c_1^{(\ell)}}b_1,
			\end{align*}
			since $d_x>2$ and $d_y>2$ by the initial assumption. Multiplying these expressions yields
			\begin{align*}
				a_0b_0=(d^{(\ell)})^2\frac{(\delta(s^{(\ell-1)},u^{(\ell-1)}))^2}{(c_1^{(\ell)})^2}\frac{a_1b_1}{d_xd_y}.
			\end{align*}
			The initial assumption $a_0b_0\neq0$, together with the induction hypothesis $\delta(s^{(\ell-1)},u^{(\ell-1)})\neq 0\neq d^{(\ell)}$, consequently imply $a_1b_1\neq0$ and thus $a_1\neq0\neq b_1$. This allows us to write
			\begin{align*}
				a_0d_xb_1=d^{(\ell)}\frac{\delta(s^{(\ell-1)},u^{(\ell-1)})}{c_1^{(\ell)}}a_1b_1\qquad\text{and}\qquad b_0d_y a_1=d^{(\ell)}\frac{\delta(s^{(\ell-1)},u^{(\ell-1)})}{c_1^{(\ell)}}a_1b_1.
			\end{align*}
			Subtracting these equations from each other yields
			\begin{align*}
				0=d_ya_1b_0-d_xa_0b_1=\delta(x,y),
			\end{align*}
			which cannot occur because of the initial assumption $\delta(x,y)\neq0$. Hence, at least one of $\delta(u^{(\ell)},x)$ or $\delta(u^{(\ell)},y)$ must be nonzero. Therefore, by definition, $s^{(\ell)}\neq0$ and  $s^{(\ell)}\in\{x,y\}$. Employing Lemma~\ref{lem:Igusa:4} allows us to obtain
			\begin{align*}
				R_{u^{(\ell+1)}}=R_{[u^{(\ell)},s^{(\ell)}]}\upto{d^{(\ell+1)}}-\delta(u^{(\ell)},s^{(\ell)})a^{d^{(\ell+1)}}+c_1^{(\ell+1)}a^\dagg a^{d^{(\ell+1)}-1}+\mathcal{O}\left((a^\dagg)^2a^{d^{(\ell+1)}-2}+\ldots+(a^\dagg)^{d^{(\ell+1)}}\right),
			\end{align*}
			where $d^{(\ell+1)}=\deg(u^{(\ell)})+\deg(s^{(\ell)})-2=d^{(\ell)}+\deg(s^{(\ell)})-2$. Now recall that generally $s^{(\ell)}\in\{x,y,0\}$ and to show that $s^{(\ell)}\neq0$, we have assumed the converse, i.e., that $s^{(\ell)}=0$, which consequently implied $\delta(u^{(\ell)},s^{(\ell)})=0$, which led to a contradiction. Thus, we can also conclude that $\delta(u^{(\ell)},s^{(\ell)})\neq 0$ by the same reasoning for which now $s^{(\ell)}$ cannot vanish and  therefore satisfies $s^{(\ell)}\in\{x,y\}$. That is, showing $\delta(u^{(\ell)},s^{(\ell)})\neq 0$ can be done by simply repeating the same steps as before, which we can therefore skip.  Furthermore, we also have that $\deg(u^{(\ell+1)})=d^{(\ell+1)}$ which is larger than $d^{(\ell)}$, since $s^{(\ell)}\in \{x,y\}$ and $\deg(x)>2<\deg(y)$ by the initial assumption. This completes the proof by induction and therefore also the proof establishing that the Lie algebra $\lie{x(a^\dagg,a),y(a^\dagg,a)\}}$ generated by the two polynomials $x(a^\dagg,a)$ and $y(a^\dagg,a)$ is infinite dimensional.

		\end{proof}
		
		It is now helpful to illustrate the importance and implications of Lemma~\ref{lem:Igusa:5} with a simple example, which we present below.
		
		\vspace{0.5cm}
		
		\begin{tcolorbox}[breakable, colback=Cerulean!3!white,colframe=Cerulean!85!black,title=Example: Application of Lemma~\ref{lem:Igusa:5}]
			
			Let us consider two normal-ordered polynomials $e_1,e_2\in\hat{A}_1$. We aim now at translating the conditions given in Lemma~\ref{lem:Igusa:5} into conditions of the coefficients of $e_1,e_2$. Such translation provides a more concrete understanding of the abstract criteria in Lemma~\ref{lem:Igusa:5}, and it also illustrates how they manifest in the structure of specific elements of the skew-hermitian Weyl algebra.
			\begin{example}\label{exa:Igusa:lemma:applied}
				Let $e_1,e_2\in\hat{A}_1$, and define $d_1:=\deg(e_1)$, $d_2:=\deg(e_2)$. Then, one can write:
				\begin{align*}
					e_j&\upto{d_j}\sum_{k=0}^{\lceil\frac{d_j-1}{2}\rceil}c_k^{(j)}g_+^{(d_j-k)\iota_1+k\iota_2}+\sum_{k=0}^{\lfloor\frac{d_j-1}{2}\rfloor}\hat{c}_k^{(j)}g_-^{(d_j-k)\iota_1+k\iota_2}\\
					&\upto{d_j}\left(ic_0^{(j)}+\hat{c}_0^{(j)}\right)a^{d_j}+\left(ic_1^{(j)}+\hat{c}_1^{(j)}\right)a^\dagg a^{d_j-1}+\ldots+\left(ic_1^{(j)}-\hat{c}_1\right)(a^\dagg)^{d_j-1}a+\left(ic_0^{(j)}-\hat{c}_0^{(j)}\right)(a^\dagg)^{d_j}.
				\end{align*}
				The condition $a_0b_0\neq0$ from Lemma~\ref{lem:Igusa:5} translates into the requirement that: 
				\begin{align*}
					\hat{c}_0^{(1)}\hat{c}_0^{(2)}\neq c_0^{(1)}c_0^{(2)},\qquad\hat{c}_0^{(2)}c_0^{(1)}\neq-\hat{c}_0^{(1)}c_0^{(2)}.
				\end{align*}
				The second condition, $\delta(x,y)\neq0$, becomes the following requirement:
				\begin{align*}
					d_1\left(\hat{c}_0^{(1)}\hat{c}_1^{(2)}-c_0^{(1)}c_1^{(2)}\right)\neq d_2\left(\hat{c}_1^{(1)}\hat{c}_0^{(2)}-c_1^{(1)}c_0^{(2)}\right),\qquad d_1\left(c_0^{(1)}\hat{c}_1^{(2)}+\hat{c}_0^{(1)}c_1^{(2)}\right)\neq d_2\left(c_1^{(1)}\hat{c}_0^{(2)}+\hat{c}_1^{(1)}c_0^{(2)}\right).
				\end{align*}
				This condition, in its current form, does not capture all pairs of elements in $\hat{A}_1$ that generate infinite-dimensional Lie algebras. For instance, consider the elements $e_1=g_-^{3\iota_1}$, and $e_2=g_+^{3\iota_1}$. These two can be used to construct a Commutator Chain of Type II, which contains elements of strictly increasing degree (as shown in \plainrefs{Bruschi:Xuereb:2024}), thereby implying that the Lie algebra $\lie{\{e_1,e_2\}}$ that contains them is infinite dimensional.
				However, if one attempts to verify this using the conditions from Lemma~\ref{lem:Igusa:5}, one finds: $\hat{c}_0^{(1)}=1=c_0^{(2)}$, while all other coefficients vanish. The first condition is not satisfied, since $\hat{c}_0^{(1)}\hat{c}_0^{(2)}=0=c_0^{(1)}c_0^{(2)}$. The second condition is also not satisfied, since 
				$c_1^{(1)}=\hat{c}_1^{(1)}=c_1^{(2)}=\hat{c}_1^{(2)}=0$. Thus, even arbitrary linear combinations of $e_1$ and $e_2$ will never satisfy the second condition. This highlights the need for a generalization of Lemma~\ref{lem:Igusa:5}.
			\end{example}
		\end{tcolorbox}
		
		\begin{lemma}\label{lem:extension:automorphism:enveloping:algebra}
			Let $\phi:\g\to\g$ be a Lie-algebra automorphism of a finite-dimensional Lie algebra $\g$, where the Lie bracket is given by the commutator. Then, $\phi$ induces a Lie-algebra automorphism $\Phi:U(\g)\to U(\g)$ on the universal enveloping algebra $U(\g)$, such that the following applies: let $\{x_j\}_{j\in\mathcal{J}}\subseteq\g$ be a basis of $\g$, and let $\{\alpha_j\}_{j\in\mathcal{J}}\subseteq\N_{\geq0}^{|\mathcal{J}|}$ be a set of non-negative integers. The map $\Phi:U(\g)\to U(\g)$ is defined on basis elements by the action:
			\begin{align*}
				x=\prod_{j\in\mathcal{J}}x_j^{\alpha_j}\mapsto\Phi(x):=\prod_{j\in\mathcal{J}}\phi(x_j)^{\alpha_j},
			\end{align*}
			and extended linearly to all of $U(\g)$, i.e., $\Phi(y+\lambda z)=\Phi(y)+\lambda\Phi(z)$ for arbitrary $y,z\in U(\g)$ and scalar $\lambda$. The map $\Phi$ is an associative Lie-algebra automorphism.
		\end{lemma}
		
		This lemma is a direct consequence of well known results in the literature. For instance, Proposition 6.2.3 from \plainrefs{Bourbaki:algebra:1974} established that the extension map $\Phi$ exists and is unique. Moreover, the proof of Proposition 2.2.1 from \plainrefs{Bourbaki:Lie:algebra:1971} shows that the map $\Phi:U(\g)\to U(\g)$ is given by the composition of the canonical mapping $\sigma_0$ from $\g$ to $U(\g)$ and the automorphism $\phi:\g\to\g$, i.e., $\Phi=\sigma_0\circ\phi$. Thus $\Phi$ satisfies exactly the properties stated in the lemma above. For readers unfamiliar with the theory of universal enveloping algebras, we provide a comprehensive proof blow. This proof may be skipped if the explanation above is sufficient.

		\begin{proof}
			We start by observing that both the codomain and domain of $\Phi$ are the universal enveloping algebra $U(\g)$. This follows directly from the assumption that $\phi$ is an automorphism and that $\Phi$ satisfies, by construction, $\Phi(y+\lambda z)=\Phi(x)+\lambda \Phi(z)$ for all $y,z\in U(\g)$ and scalar coefficients $\lambda$. To establish that $\Phi$ is an associative Lie-algebra automorphism, we must verify the following four properties:
			\begin{enumerate}
				\item \textbf{Linearity}: The map $\Phi$ is linear by construction.
				\item \textbf{Associativity}:  
				We now prove inductively that $\Phi(x_{j_1}\ldots x_{j_m})=\Phi(x_{j_1})\ldots\Phi(x_{j_m})$ holds for an arbitrary product $x_{j_1}\ldots x_{j_m}$ of basis elements $x_j\in\g$.
				
				\textit{Base case} ($m=2$): For any two indices $j,k\in \{1,\ldots,n\}$ with $j\leq k$, the product $x_jx_k$ is already in the canonical order, and we have $\Phi(x_j x_k)=\Phi(x_j)\Phi(x_k)$ by construction. Suppose now $j>k$. By the definition of the commutator, one has $x_jx_k=x_kx_j+[x_j,x_k]$. The linearity of $\Phi$ implies then $\Phi(x_jx_k)=\Phi(x_kx_j+[x_j,x_k])=\Phi(x_kx_j)+\Phi([x_j,x_k])$. Let now $f_{jk\ell}$ be the structure constants of $\g$. That is, the commutator of two basis elements $x_j$ and $x_k$ can be written as $[x_j,x_k]=\sum_\ell f_{jk\ell} x_\ell$. Thus, by the definition of $\Phi$, one has $\Phi(x_jx_k)=\phi(x_k)\phi(x_j)+\Phi\left(\sum_{\ell}f_{jk\ell}x_\ell\right)=\phi(x_k)\phi(x_j)+\sum_{\ell}f_{jk\ell}\phi(x_\ell)$, since $x_kx_j$ and $x_\ell$ are basis elements, by the Poincaré-Birkhoff-Witt Theorem \plainrefs{PBW:Thm} and the assumption that $j>k$. Recalling that $\phi$ is a Lie-algebra automorphism, i.e., it preserves the Lie bracket, we can write $\Phi(x_jx_k)=\phi(x_k)\phi(x_j)+[\phi(x_j),\phi(x_k)]=\phi(x_j)\phi(x_k)=\Phi(x_j)\Phi(x_k)$.
				One must have consequently $\Phi(x_jx_k)=\Phi(x_j)\Phi(x_k)$ independently of the ordering of $x_j$ and $x_k$.
				
				\textit{Induction step}: Assume that for some some $m\in\N$, the identity $\Phi(x_{j_1}\ldots x_{j_m})=\Phi(x_{j_1})\ldots\Phi(x_{j_m})$ holds for all products of $m$ basis elements $x_{j_k}\in \g$. We consider now a product of $m+1$ basis elements and aim to show that $\Phi(x_{j_1}\ldots x_{j_m}x_{j_{m+1}})=\Phi(x_{j_1})\ldots\Phi(x_{j_m})\Phi(x_{j_{m+1}})$. To do this, we start by constructing a permutation $\pi$ from the symmetric group $S_{m+1}$ such that the reordered tuple $(x_{\pi(j_1)},\ldots,x_{\pi(j_m)},x_{\pi(j_{m+1})})$ is in non-decreasing order with respect to the canonical order on the basis of $\g$. In other words, $\pi$ satisfies the requirement $\pi(j_1)\leq\pi(j_2)\leq\ldots\leq \pi(j_m)\leq \pi(j_{m+1})$. We construct $\pi$ procedurally using Algorithm~\ref{alg:ordering:pbw}. This algorithm allows us to decompose $\pi$ into a sequence of adjacent transpositions $\tau_j=(k\;\;k+1)$, which sort the product $x_{j_1}\ldots x_{j_m} x_{j_{m+1}}$ into canonical order.
				\begin{algorithm}[htpb]
					\DontPrintSemicolon
					\KwData{A tuple $\boldsymbol{j}=(j_1,\ldots,j_m)\in\Np{m}$}
					\KwResult{A permutation $\pi\in S_m$ as a decomposition of transpositions $\tau_j$}
					\SetKwData{Left}{left}\SetKwData{This}{this}\SetKwData{Up}{up}
					\SetKwFunction{Union}{Union}\SetKwFunction{FindCompress}{FindCompress}
					\SetKwInOut{Input}{input}\SetKwInOut{Output}{output}
					\BlankLine
					$\pi\leftarrow\mathrm{id}_{S_{m}}$\;
					$\ell\leftarrow m$\;
					\While{$\ell>1$}{
						\While(\tcc*[h]{Notation: $\pi(\boldsymbol{j}):=(\pi(j_1),\ldots,\pi(j_m))$}){$\exists\,k\in\{1,\ldots,\ell-1\}$: $(\pi(\boldsymbol{j}))_k>(\pi(\boldsymbol{j}))_\ell$}{
							$k\leftarrow\min\{k\in\{1,\ldots,\ell-1\}\,\mid\,(\pi(\boldsymbol{j}))_k>(\pi(\boldsymbol{j}))_\ell\}$\;
							$\pi\leftarrow (\ell\;\;\ell-1)\circ\ldots\circ(k+1\;\;k)\circ\pi$\;
						}
						$\ell\leftarrow\ell-1$\;
					}
					\Return $\pi=\tau_o\circ\ldots\circ\tau_1$\;
					\caption{Achieving canonical ordering an $n$-tuple}
					\label{alg:ordering:pbw}
				\end{algorithm}
				If $\pi=\mathrm{id}_{S_{m+1}}$, it is clear that $\Phi(x_{j_1}\ldots x_{j_m}x_{j_{m+1}})=\Phi(x_{j_1})\ldots\Phi(x_{j_m})\Phi(x_{j_{m+1}})$, since the basis elements $x_j\in\g$ are already canonically ordered. Now suppose $\pi=\tau_{o}\circ\ldots\circ\tau_1\neq\mathrm{id}_{S_{m+1}}$ for some $o\in\N_{\geq1}$. 
				Let now $\tau_1=(k\;\;k+1)$ be a transposition swapping the $k$-th and $(k+1)$-th element. Then, analogously to the base case we have
				\begin{align*}
					\Phi\left(x_{j_1}\ldots x_{j_{k-1}} x_{j_k}x_{j_{k+1}}x_{j_{k+2}}\ldots x_{j_m}x_{j_{m+1}}\right)=\Phi\left(x_{j_1}\ldots x_{j_{k-1}}\left(x_{j_{k+1}}x_{j_k}+[x_{j_k},x_{j_{k+1}}]\right)x_{j_{k+2}}\ldots x_{j_m}x_{j_{m+1}}\right).
				\end{align*}
				The linearity of $\Phi$ allows us to write
				\begin{align*}
					\Phi\left(x_{j_1}\ldots x_{j_{k-1}} x_{j_k}x_{j_{k+1}}x_{j_{k+2}}\ldots x_{j_m}x_{j_{m+1}}\right)&=\Phi\left(x_{j_1}\ldots x_{j_{k-1}}x_{j_{k+1}}x_{j_k}x_{j_{k+2}}\ldots x_{j_m}x_{j_{m+1}}\right)\\
					&\quad+\Phi\left(x_{j_1}\ldots x_{j_{k-1}}[x_{j_k},x_{j_{k+1}}]x_{j_{k+2}}\ldots x_{j_m}x_{j_{m+1}}\right).
				\end{align*}
				We can now identify $j_p=\tau_1(j_p)$ if $p\notin\{k,k+1\}$ and $j_{k+1}=\tau_1(j_k)$, as well as $j_{k}=\tau_1(j_{k+1})$ and recall that the structure constants of $\g$ are given by $f_{jk\ell}$, i.e., $[x_j,x_k]=\sum_\ell f_{jk\ell} x_{\ell}$. Hence, we have by the linearity of $\Phi$:
				\begin{align*}
					\Phi\left(x_{j_1}\ldots x_{j_m}x_{j_{m+1}}\right)&=\Phi(x_{\tau_1(j_1)}\ldots x_{\tau_1(j_{m-1})}x_{\tau_1(j_{m})}x_{\tau_1(j_{m+1})})\\
					&\qquad+\sum_{\ell_1=1}^n f_{j_k,j_{k+1},\ell_1}\Phi(x_{\tau_1(j_1)}\ldots x_{\tau_1(j_{k-1})}x_{\ell_1}x_{\tau_1(j_{k+2})}\ldots x_{\tau_1(j_{m+1})}).
				\end{align*}
				Invoking the induction hypothesis for the last term yields finally:
				\begin{align*}
					\Phi\left(x_{j_1}\ldots x_{j_m}x_{j_{m+1}}\right)&=\Phi(x_{\tau_1(j_1)}\ldots x_{\tau_1(j_{m-1})}x_{\tau_1(j_{m})}x_{\tau_1(j_{m+1})})\\
					&\qquad+\sum_{\ell_1=1}^n f_{j_k,j_{k+1},\ell_1}\Phi(x_{\tau_1(j_1)})\ldots\Phi(x_{\tau_1(j_{k-1})})\Phi(x_{\ell_1})\Phi(x_{\tau_1(j_{k+2})})\ldots\Phi(x_{\tau_1(j_{m+1})}),
				\end{align*}
					Repeating this for all remaining transpositions $\tau_2,\ldots,\tau_o$ in the decomposition of $\pi$ yields:
					\begin{align*}
						\Phi\left(x_{j_1}\ldots x_{j_m}x_{j_{m+1}}\right)&=\Phi(x_{\pi(j_1)}\ldots x_{\pi(j_{m-1})}x_{\pi(j_{m})}x_{\pi(j_{m+1})})\\
						&\qquad+\sum_{k=1}^o\sum_{\ell_k=1}^n c_{k\ell_k}\Phi(x_{\tau_k\circ\ldots\circ\tau_1(j_1)})\ldots\Phi(x_{\ell_k})\ldots\Phi(x_{\tau_k\circ\ldots\circ\tau_1(j_{m+1})}),
					\end{align*}
					where $c_{k\ell_k}$ are the appropriate structure constants determined by the permutation $\tau_k\circ\ldots\circ\tau_1$. Now one has $\Phi(x_{\pi(j_1)}\ldots x_{\pi(j_{m})}x_{\pi(j_{m+1})})=\Phi(x_{\pi(j_1)})\ldots \Phi(x_{\pi(j_{m})})\Phi(x_{\pi(j_{m+1})})$, since the product $x_{\pi(j_1)}\ldots x_{\pi(j_{m})}x_{\pi(j_{m+1})}$ is now correctly ordered. Hence:
					\begin{align}
						\Phi\left(x_{j_1}\ldots x_{j_m}x_{j_{m+1}}\right)&=\Phi(x_{\pi(j_1)})\ldots\Phi( x_{\pi(j_{m-1})})\Phi(x_{\pi(j_{m})})\Phi(x_{\pi(j_{m+1})})\nonumber\\
						&\qquad+\sum_{k=1}^o\sum_{\ell_k=1}^n c_{k\ell_k}\Phi(x_{\tau_k\circ\ldots\circ\tau_1(j_1)})\ldots\Phi(x_{\ell_k})\ldots\Phi(x_{\tau_k\circ\ldots\circ\tau_1(j_{m+1})}).\label{eqn:help:automorphism:extension:1}
					\end{align}
					We now repeat the same process starting from $\Phi\left(x_{j_1}\right)\ldots \Phi(x_{j_m})\Phi\left(x_{j_{m+1}}\right)$. We note that $\Phi(x_{\hat{k}'})\Phi(x_{\hat{k}})+[\Phi(x_{\hat{k}}),\Phi(x_{\hat{k}'})]=\Phi(x_{\hat{k}})\Phi(x_{\hat{k}'})$ and $[\Phi(x_j),\Phi(x_k)]=[\phi(x_j),\phi(x_k)]=\sum_\ell f_{jk\ell}\phi(x_\ell)=\sum_\ell f_{jk\ell}\Phi(x_\ell)$. Therefore, following similar steps to those taken above, we obtain
					\begin{align*}
						\Phi\left(x_{j_1}\right)\ldots \Phi(x_{j_m})\Phi\left(x_{j_{m+1}}\right)&=\Phi\left(x_{j_1}\right)\ldots\left(\Phi(x_{j_{k+1}})\Phi(x_{j_k})+\left[\Phi(x_{j_k}),\Phi(x_{j_{k+1}})\right]\right)\Phi(x_{j_{k+2}})\ldots \Phi\left(x_{j_{m+1}}\right)\\
						&=\Phi\left(x_{\tau_1(j_1)}\right)\ldots \Phi(x_{\tau_1(j_m)})\Phi\left(x_{\tau_1(j_{m+1})}\right)\\
						&\qquad+\sum_{\ell_1=1}^n f_{j_k,j_{k+1},\ell_1}\Phi(x_{\tau_1(j_1)})\ldots\Phi(x_{\tau_1(j_{k-1})})\Phi(x_{\ell_1})\Phi(x_{\tau_1(j_{k+2})})\ldots\Phi(x_{\tau_1(j_{m+1})}),
					\end{align*}
					where we recalled that $\tau_1=(k\;\;k+1)$ and consequently $\tau_1(j_\ell)=j_\ell$ if $\ell\notin\{k,k+1\}$ and $j_{k+1}=\tau_1(j_k)$, $j_k=\tau_1(j_{k+1})$ otherwise. Repeating this for all remaining transpositions $\tau_2,\ldots,\tau_o$ in the decomposition of $\pi$ yields consequently:
					\begin{align}
						\Phi\left(x_{j_1}\right)\ldots \Phi(x_{j_m})\Phi\left(x_{j_{m+1}}\right)&=\Phi(x_{\pi(j_1)})\ldots\Phi( x_{\pi(j_{m-1})})\Phi(x_{\pi(j_{m})})\Phi(x_{\pi(j_{m+1})})\nonumber\\
						&\qquad+\sum_{k=1}^o\sum_{\ell_k=1}^n c_{k\ell_k}\Phi(x_{\tau_k\circ\ldots\circ\tau_1(j_1)})\ldots\Phi(x_{\ell_k})\ldots\Phi(x_{\tau_k\circ\ldots\circ\tau_1(j_{m+1})}).\label{eqn:help:automorphism:extension:2}
					\end{align}
					Comparing equations \eqref{eqn:help:automorphism:extension:1} and \eqref{eqn:help:automorphism:extension:2} finally allows us to conclude that
					\begin{align*}
						\Phi\left(x_{j_1}\ldots x_{j_m}x_{j_{m+1}}\right)=\Phi\left(x_{j_1}\right)\ldots \Phi(x_{j_m})\Phi\left(x_{j_{m+1}}\right),
					\end{align*}
					thereby completing the proof by induction.

					The linearity of $\Phi$ implies now that $\Phi(\tilde{x}_1\tilde{x}_2)=\Phi(\tilde{x}_1)\Phi(\tilde{x}_2)$ for all $\tilde{x}_1,\tilde{x}_2\in U(\g)$, since any arbitrary elements $\tilde{x}_1,\tilde{x}_2\in U(\g)$ can, by the Poincaré-Birkhoff-Witt Theorem \plainrefs{PBW:Thm}, be expressed as unique linear combinations of the form
					\begin{align*}
						\tilde{x}_1=\sum_{k\in\mathcal{K}}c_k^{(1)}\prod_{j\in\mathcal{J}}x_j^{\alpha_j^{(1)}},\quad\tilde{x}_2=\sum_{k'\in\mathcal{K}'}{c}_{k'}^{(2)}\prod_{j'\in\mathcal{J}}^nx_{j'}^{\alpha_{j'}^{(2)}},
					\end{align*}
					where the coefficients $c_k^{(1)},c_{k'}^{(2)}$ are non-zero for all $k\in\mathcal{K}\subseteq\N_{\geq1}$ and $k'\in\mathcal{K}'\subseteq\N_{\geq1}$. Hence:{\small
						\begin{align*}
							\Phi(\tilde{x}_1\tilde{x}_2)&=\Phi\left(\sum_{k\in\mathcal{K}}c_k^{(1)}\prod_{j\in\mathcal{J}}x_j^{\alpha_j^{(1)}}\sum_{k'\in\mathcal{K}'}{c}_{k'}^{(2)}\prod_{j'\in\mathcal{J}}x_{j'}^{\alpha_{j'}^{(2)}}\right)=\sum_{k\in\mathcal{K}}\sum_{k'\in\mathcal{K}'}c_k^{(1)}c_{k'}^{(2)}\Phi\left(\prod_{j\in\mathcal{J}}x_j^{\alpha_j^{(1)}}\prod_{j'\in\mathcal{J}}x_{j'}^{\alpha_{j'}^{(2)}}\right)\\
							&=\sum_{k\in\mathcal{K}}\sum_{k'\in\mathcal{K}'}c_k^{(1)}c_{k'}^{(2)}\prod_{j\in\mathcal{J}}\Phi\left(x_j^{\alpha_j^{(1)}}\right)\prod_{j'\in\mathcal{J}}\Phi\left(x_{j'}^{\alpha_{j'}^{(2)}}\right)=\sum_{k\in\mathcal{K}}c_k^{(1)}\prod_{j\in\mathcal{J}}\Phi\left(x_j^{\alpha_j^{(1)}}\right)\sum_{k'\in\mathcal{K}'}{c}_{k'}^{(2)}\prod_{j'\in\mathcal{J}}\Phi\left(x_{j'}^{\alpha_{j'}^{(2)}}\right)=\Phi(\tilde{x}_1)\Phi(\tilde{x}_2).
					\end{align*}}\noindent
					One has consequently
					\begin{align*}
						\Phi((xy)z)=\Phi(xy)\Phi(z)=(\Phi(x)\Phi(y))\Phi(z)=\Phi(x)(\Phi(y)\Phi(z))=\Phi(x)\Phi(yz)=\Phi(x(yz))
					\end{align*}
					for all $x,y,z\in U(\g)$, showing that $\Phi$ is associative.
					
					\item \textbf{Preservation of the Lie bracket}: Let $x,y\in U(\g)$. Then, by the linearity and associativity of $\Phi$, one has for the commutator of $x$ and $y$: 
					\begin{align*} 
						\Phi\left([x,y]\right)&=\Phi\left(xy-yx\right)=\Phi(x)\Phi(y)-\Phi(y)\Phi(x)=[\Phi(x),\Phi(y)],
					\end{align*}\noindent
					since the Lie bracket of the universal enveloping algebra $U(\g)$ is defined as the commutator.
					
					\item \textbf{Invertibility}: The map $\phi:\g\to\g$ is a Lie-algebra automorphism. Thus, its inverse $\phi^{-1}:\g\to\g$ exists and is also a Lie-algebra automorphism. Define the map $\Psi:U(\g)\to U(\g)$ by extending $\phi^{-1}$ in the same way $\Phi$ was defined from $\phi$. That is, for any basis element $x=\prod_{j\in\mathcal{J}}x_j^{\alpha_j}\in U(\g)$ define the action:
					\begin{align*}
						x\mapsto\Psi(x):=\prod_{j\in\mathcal{J}}(\phi^{-1}(x_j))^{\alpha_j},
					\end{align*}
					and extend $\Phi$ linearly to all of $U(\g)$. Then $\Psi$ is, by virtue of the previous discussion, also an associative Lie-algebra endomorphism that maps the universal enveloping algebra $U(\g)$ to itself. Let now $x\in U(\g)$ be an arbitrary element with the unique decomposition $x=\sum_{k\in\mathcal{K}}c_k\prod_{j\in\mathcal{J}}x_j^{\alpha_j}$. Then, one has:
					\begin{align*}
						\Phi(\Psi(x))&=\Phi\left(\sum_{k\in\mathcal{K}}\prod_{j\in\mathcal{J}}(\phi^{-1}(x_j))^{\alpha_j}\right)=\sum_{k\in\mathcal{K}}c_k\Phi\left(\prod_{j\in\mathcal{J}}(\phi^{-1}(x_j))^{\alpha_j}\right)=\sum_{k\in\mathcal{K}}c_k\prod_{j\in\mathcal{J}}\left(\Phi\left(\phi^{-1}(x_j)\right)\right)^{\alpha_j}\\
						&=\sum_{k\in\mathcal{K}}c_k\prod_{j\in\mathcal{J}}\left(\phi\left(\phi^{-1}(x_j)\right)\right)^{\alpha_j}=x=\sum_{k\in\mathcal{K}}c_k\prod_{j\in\mathcal{J}}\left(\phi^{-1}\left(\phi(x_j)\right)\right)^{\alpha_j}=\sum_{k\in\mathcal{K}}c_k\prod_{j\in\mathcal{J}}\left(\Psi\left(\phi(x_j)\right)\right)^{\alpha_j}\\
						&=\sum_{k\in\mathcal{K}}c_k\Psi\left(\prod_{j\in\mathcal{J}}\phi(x_j)^{\alpha_j}\right)=\Psi\left(\sum_{k\in\mathcal{K}}c_k\prod_{j\in\mathcal{J}}\phi(x_j)^{\alpha_j}\right)=\Psi(\Phi(x)).
					\end{align*}
					Hence, $\Phi$ is, by the existence of the inverse $\phi^{-1}$, invertible and its inverse is $\Psi$.
				\end{enumerate}
				We finally conclude that $\Phi$ is a Lie-algebra automorphism.
			\end{proof}
			
			\begin{lemma}\label{lem:Igusa:0}
				Every pair $(\boldsymbol{\sigma},\boldsymbol{\lambda})$ of a symplectic matrix $\boldsymbol{\sigma}\in\operatorname{Sp}(2,\C)$ and a vector $\boldsymbol{\lambda}\in\C^2$ induces a Lie-algebra automorphism $\Phi$ of a subalgebra of the complex Weyl algebra ${A}_{1}$.
			\end{lemma}
			
			\begin{proof}
				Let $\boldsymbol{\sigma}=(\boldsymbol{\eta}|\boldsymbol{\xi})\in\operatorname{Gl}(2,\C)$ be an invertible matrix with $\boldsymbol{\xi},\boldsymbol{\eta}\in\C^2$ and let $\boldsymbol{\lambda}\in\C^2$. Introduce the linear map $\phi:\langle\{a^\dagg,a,1\}\rangle_{\mathrm{Lie}}\to\langle\{a^\dagg,a,1\}\rangle_{\mathrm{Lie}}$ defined by the actions:
				\begin{align*}
					\phi(a^\dagg)&:=\xi_1 a^\dagg+\eta_1 a+\lambda_1,\;&\;\phi(a)&:=\xi_2a^\dagg+\eta_2a+\lambda_2,\;&\;\phi(1)&:=\det(\boldsymbol{\sigma}).
				\end{align*}
				A short calculation yields:
				\begin{align*}
					[\phi(a^\dagg),\phi(a)]&=[\xi_1 a^\dagg+\eta_1 a+\lambda_1,\xi_2a^\dagg+\eta_2a+\lambda_2]=\det(\boldsymbol{\sigma})=\phi(1),\;&\;[\phi(a^\dagg),\phi(1)]&=0,\;&\;[\phi(a),\phi(1)]&=0.
				\end{align*}
				One has furthermore $\det(\phi)=\det^2(\boldsymbol{\sigma})\neq0$. Thus, $\phi$ defines a Lie-algebra automorphism on $\langle\{a^\dagg,a,1\}\rangle_{\mathrm{Lie}}$. It is clear that this is the complex Heisenberg algebra $\gh_{1,\C}$.
				By Lemma~\ref{lem:extension:automorphism:enveloping:algebra}, $\phi$ induces a Lie-algebra automorphism of $U(\lie{\{a^\dagg,a,1\}})\cong A_1$\footnote{Note that the complex Weyl algebra $A_1$ is the universal enveloping algebra of the complex Heisenberg algebra $\gh_{1,\C}$ only if one requires that the central elements of $\gh_{1,\C}$ are scalar multiples of the unit-element.}, and consequently, by restriction to any subalgebra $\g$ of $A_1$, a Lie-algebra automorphism of $\g$. We may restrict our attention from invertible matrices $\boldsymbol{\sigma}\in\operatorname{Gl}(2,\C)$ to those with determinant one, i.e., symplectic matrices, such that $\Phi(1)=1$.
			\end{proof}

			The claim of Lemma~\ref{lem:Igusa:0} is generally not true for the skew-hermitian Weyl algebra $\hat{A}_1$, since $\Phi$ will, by construction, not be skew-hermitian for all symplectic matrices. To elaborate on this point, consider the following observations: the first issue is that $\hat{A}_1$ is not the universal enveloping algebra of any real Lie algebra, because squaring a skew-hermitian element yields a hermitian one. Thus Lemma~\ref{lem:extension:automorphism:enveloping:algebra} is not applicable in this setting. 
			The second issue is formalized in the following proposition:
			\begin{proposition}\label{prop:not:all:symplectic:matrices:generate:automorphisms:on:skew:hermitian:algebra}
				The only symplectic matrices $\boldsymbol{\sigma}\in\operatorname{Sp}(2,\C)$ and vectors $\boldsymbol{\lambda}\in\C^2$ that induce a Lie-algebra automorphism of a subalgebra of the real skew-hermitian Weyl algebra $\hat{A}_1$, constructed via the usual linear extension of a Lie-algebra automorphism $\phi$ of the complex Heisenberg algebra $\lie{\{a^\dagg,a,1\}}\cong\gh_{1,\C}$, have the form
				\begin{align*}
					\boldsymbol{\sigma}=\begin{pmatrix}
						e^{i\varphi}\cosh(s)&e^{-i\vartheta}\sinh(s)\\
						e^{i\vartheta}\sinh(s)&e^{-i\varphi}\cosh(s)
					\end{pmatrix}\quad\text{for }s,\varphi,\vartheta\in\R,\quad\text{and}\quad\boldsymbol{\lambda}=\lambda_1\begin{pmatrix}
						1\\
						1
					\end{pmatrix}+i\lambda_2\begin{pmatrix}
						1\\
						-1
					\end{pmatrix}\quad\text{for }\lambda_1,\lambda_2\in\R.
				\end{align*}
			\end{proposition}
			
			\begin{proof}
				Consider the symplectic matrix $\boldsymbol{\sigma}=(\boldsymbol{\eta}|\boldsymbol{\xi})\in\operatorname{Sp}(2,\C)$ and vector $\boldsymbol{\lambda}\in\R^2$. Suppose $\phi$ is a Lie-algebra automorphism of the complex algebra $\lie{\{a^\dagg,a,1\}}\cong\gh_{1,\C}$ that is induced by $\boldsymbol{\sigma}$ and $\boldsymbol{\lambda}$. That is $\phi$ is defined by the actions $\phi(a^\dagg):=\xi_1a^\dagg+\eta_1a+\lambda_1$, $\phi(a):=\xi_2a^\dagg+\eta_2a+\lambda_2$, $\phi(1):=1$. Let furthermore $\Phi$ be the usual extension of $\phi$ to the universal enveloping algebra $U(\lie{\{a^\dagg,a,1\}})\cong A_1$, as defined in Lemma~\ref{lem:extension:automorphism:enveloping:algebra}. It is clear that $\Phi$ is a Lie-algebra automorphism of a subalgebra of the real skew-hermitian Weyl algebra $\hat{A}_1$ only if $\phi$ is a Lie algebra automorphism on $\lie{\{i,i(a+a^\dagg),a-a^\dagg\}}\cong\gh_1$ when restricted to it. Thus, the Lie-algebra automorphism $\phi$ must satisfy the following four properties:
				\begin{enumerate}[label=(\roman*)]
					\item $\phi$ must be linear on $\hat{A}_1$, i.e., $\phi(x+\lambda y)=\phi(x)+\lambda\phi(y)$ for all $x,y\in\hat{A}_1$ and $\lambda\in\R$;
					\item $\phi$ must preserve the commutator on $\hat{A}_1$, i.e., $[\phi(x),\phi(y)]=\phi([x,y])$ for all $x,y\in\hat{A}_1$;
					\item $\phi$ must be invertible on $\hat{A}_1$;
					\item $\phi$ must, for every skew-hermitian element in the domain, yield a skew-hermitian result, i.e., $(\phi(x))^\dagg=-\phi(x)$ for all $x,y\in \gh_1\cong\spn\{i(a+a^\dagg),a-a^\dagg,i\}$.
				\end{enumerate}
				This implies that the following identities must hold:
				\begin{align*}
					\phi(i(a+a^\dagg))&=i\left((\xi_1+\xi_2)a^\dagg+(\eta_1+\eta_2)a+(\lambda_1+\lambda_2)\right)=i\left((\eta_1^*+\eta_2^*)a^\dagg+(\xi_1^*+\xi_2^*)a+(\lambda_1^*+\lambda_2^*)\right)^\dagg=-(\phi(i(a+a^\dagg)))^\dagg,\\
					\phi(a-a^\dagg)&=(\xi_2-\xi_1)a^\dagg+(\eta_2-\eta_1)a+(\lambda_2-\lambda_1)=(\eta_1^*-\eta_2^*)a^\dagg+(\xi_1^*-\xi_2^*)a+(\lambda_1^*-\lambda_2^*)^\dagg=-(\phi(a-a^\dagg))^\dagg.
				\end{align*}
				Here, we first used the linearity of $\phi$ and its definition. Then, we used that the hermitian conjugation is an involutory operation, i.e.,$(x^\dagg)^\dagg=x$ for all $x\in A_1$. The last step follows from condition (iv), as $\phi$ must yield skew-hermitian results when restricted to elements from $\spn\{i(a+a^\dagg),a-a^\dagg,i\}$.
				Comparing the prefactors of the basis elements $a^\dagg$, $a$, and 1, requires: $\xi_1=\eta_2^*$, $\xi_2=\eta_1^*$, and $\lambda_1=\lambda_2^*$. The symplectic condition $\det(\boldsymbol{\sigma})=\xi_2\eta_1-\xi_1\eta_2=1$ implies consequently $|\eta_1|^2-|\eta_2|^2=1$. Thus, one can parameterize $|\eta_1|$ and $|\eta_2|$ by the real parameter $s\in\R$ as $|\eta_1|=\cosh(s)$ and $|\eta_2|=\sinh(s)$, since the hyperbolic functions $\sinh(s)$ and $\cosh(s)$ satisfy $\cosh^2(s)-\sinh^2(s)=1$, and $\sinh(\R)=\R$, as well as $\cosh(\R)=[1,\infty)$, i.e., exhausting the possible domain of $|\eta_1|$ and $|\eta_2|$ respectively. Since $\eta_1$ and $\eta_2$ are two complex numbers, we can write $\eta_1=\cosh(s)e^{i\varphi}$ and $\eta_2=\sinh(s)e^{i\vartheta}$, where $\varphi,\vartheta\in[0,2\pi)$ are two angles. The conditions $\xi_1=\eta_2^*$ and $\xi_2=\eta_1^*$ allow us consequently to write $\xi_1=\sinh(s)e^{-i\vartheta}$ and $\xi_2=\cosh(s)e^{-i\varphi}$. Moreover, since $\lambda_1\in\C$, we can parameterize it with two real numbers such that $\lambda_1=u+iv$, where $u,v\in\R$. The requirement $\lambda_1=\lambda_2^*$ yields then $\lambda_2=u-iv$ and we can compactly write 
				\begin{align*}
					\boldsymbol{\sigma}=\begin{pmatrix}
						e^{i\varphi}\cosh(s)&e^{-i\vartheta}\sinh(s)\\
						e^{i\vartheta}\sinh(s)&e^{-i\varphi}\cosh(s)
					\end{pmatrix}\quad\text{for }s\in\R,\,\varphi,\vartheta\in[0,2\pi)\quad\text{and}\quad\boldsymbol{\lambda}=u\begin{pmatrix}
						1\\
						1
					\end{pmatrix}+iv\begin{pmatrix}
						1\\
						-1
					\end{pmatrix}\quad\text{for }u,v\in\R.
				\end{align*}
				It is straightforward to verify that the usual linear extension $\Phi$ of $\phi$ defines a Lie-algebra automorphism on $\hat{A}_1$. Indeed, the conditions $\xi_1^*=\eta_2$, $\xi_2^*=\eta_1$, $\lambda_1^*=\lambda_2$ imply:
				\begin{align*}
					-(\Phi(g_+^{\gamma}))^\dagg&=-\left(\Phi\left(i\left((a^\dagg)^\beta a^\alpha+(a^\dagg)^\alpha a^\beta\right)\right)\right)^\dagg=i\left((\phi(a^\dagg))^\beta(\phi(a))^\alpha+(\phi(a^\dagg))^\alpha (\phi(a))^\beta\right)^\dagg\\
					&=i\left(\left(\xi_1a^\dagg+\eta_1a+\lambda_1\right)^\beta\left(\xi_2a^\dagg+\eta_2a+\lambda_2\right)^\alpha+\left(\xi_1a^\dagg+\eta_1a+\lambda_1\right)^\alpha\left(\xi_2a^\dagg+\eta_2a+\lambda_2\right)^\beta\right)^\dagg\\
					&=i\left(\left(\eta_2^*a^\dagg+\xi_2^*a+\lambda_2^*\right)^\beta\left(\eta_1^*a^\dagg+\xi_1^*a+\lambda_1^*\right)^\alpha+\left(\eta_2^*a^\dagg+\xi_2^*a+\lambda_2^*\right)^\alpha\left(\eta_1^*a^\dagg+\xi_1^*a+\lambda_1^*\right)^\beta\right)^\dagg\\
					&=\Phi(g_+^\gamma),\\
					-(\Phi(g_-^{\gamma}))^\dagg&=-\left(\Phi\left((a^\dagg)^\beta a^\alpha-(a^\dagg)^\alpha a^\beta\right)\right)^\dagg=-\left((\phi(a^\dagg))^\beta(\phi(a))^\alpha-(\phi(a^\dagg))^\alpha (\phi(a))^\beta\right)^\dagg\\
					&=-\left(\left(\xi_1a^\dagg+\eta_1a+\lambda_1\right)^\beta\left(\xi_2a^\dagg+\eta_2a+\lambda_2\right)^\alpha-\left(\xi_1a^\dagg+\eta_1a+\lambda_1\right)^\alpha\left(\xi_2a^\dagg+\eta_2a+\lambda_2\right)^\beta\right)^\dagg\\
					&=\left(\left(\eta_2^*a^\dagg+\xi_2^*a+\lambda_2^*\right)^\beta\left(\eta_1^*a^\dagg+\xi_1^*a+\lambda_1^*\right)^\alpha-\left(\eta_2^*a^\dagg+\xi_2^*a+\lambda_2^*\right)^\alpha\left(\eta_1^*a^\dagg+\xi_1^*a+\lambda_1^*\right)^\beta\right)^\dagg\\
					&=\Phi(g_-^\gamma),
				\end{align*}
				showing that $\Phi$ preserves the skew-hermitian property of the monomials $g_\sigma^\gamma$. The remaining properties required for $\Phi$ to be a Lie-algebra automorphism - namely linearity, preservation of the Lie bracket, and invertibility - are satisfied as well. These follow analogously to the arguments presented in Lemmas~\ref{lem:extension:automorphism:enveloping:algebra} and~\ref{lem:Igusa:0}. Therefore, $\Phi$ defines a Lie-algebra automorphism of a subalgebra of the real skew-hermitian Weyl algebra $\hat{A}_1$ as claimed.
			\end{proof}
			
			\begin{corollary}\label{cor:extension:isomorphism:enveloping:algebra}
				Let $\boldsymbol{\sigma}=(\boldsymbol{\eta}|\boldsymbol{\xi})\in\operatorname{Sp}(2,\C)$ be a symplectic matrix with $\boldsymbol{\xi},\boldsymbol{\eta}\in\C^2$, and let $\boldsymbol{\lambda}\in\C^2$, which together define the linear map $\phi:\spn\{a,a^\dagg,1\}\cong\gh_{1,\C}\to\spn\{a^\dagg,a,1\}\cong\gh_{1,\C}$ via
				\begin{align*}
					\phi(a^\dagg):=\xi_1a^\dagg+\eta_1 a +\lambda_1,
					\qquad
					\phi(a):=\xi_2a^\dagg+\eta_2a+\lambda_2,
					\qquad
					\text{and}
					\qquad
					\phi(1):=1.
				\end{align*}
				Define the extension $\Phi:{A}_{1}\to A_1$ of $\phi$ as the linear map  given by the action on the basis elements
				\begin{align*}
					(a^\dagg)^{\alpha}a^\beta1^\gamma=(a^\dagg)^\alpha a^\beta\mapsto\Phi\left((a^\dagg)^\alpha a^\beta 1^\gamma\right):= (\phi(a^\dagg))^\alpha (\phi(a))^\beta (\phi(1))^\gamma=  (\phi(a^\dagg))^\alpha (\phi(a))^\beta, 
				\end{align*}
				for all $\alpha,\beta,\gamma\in\N_{\geq0}$.
				Then, the restriction $\Phi|_{\hat{A}_1}:\hat{A}_1\to\Phi|_{\hat{A}_1}(\hat{A}_1)$ of $\Phi$ to the skew-hermitian Weyl algebra $\hat{A}_1$ is a Lie-algebra isomorphism from $\hat{A}_1$ to another real Lie algebra. The same holds true for any subalgebra of $\hat{A}_1$.
			\end{corollary}

			\begin{proof}
				This is a direct consequence of Lemma~\ref{lem:extension:automorphism:enveloping:algebra} and~\ref{lem:Igusa:0}, since restricting the map $\Phi$ to the real skew-hermitian subalgebra $\hat{A}_1$ implies that $\Phi|_{\hat{A}_1}$ is a linear and invertible map that preserves the commutation relations. Therefore $\Phi|_{\hat{A}_1}$ is a Lie-algebra isomorphism. Moreover, the image $\Phi|_{\hat{A}_1}(\hat{A}_1)$ is, by construction, a real Lie algebra, although not necessarily skew-hermitian, as demonstrated in Proposition~\ref{prop:not:all:symplectic:matrices:generate:automorphisms:on:skew:hermitian:algebra}. Finally, restricting $\Phi|_{\hat{A}_1}$ to any subalgebra of $\hat{A}_1$ naturally extends the isomorphism to those subalgebras as well. 
			\end{proof}
			
			\begin{theorem}\label{thm:Igusa:1}
				Let $x,y\in\mathbb{S}[X,Y]$ be two polynomials of degree greater than two. We here use the notation of Lemma~\ref{lem:Igusa:3}.  If there exists a vector $\boldsymbol{\eta}\in\C^2$ such that that the following conditions apply:
				\begin{enumerate}[label =(\roman*)]
					\item $R_x(\boldsymbol{\eta})R_y(\boldsymbol{\eta})\neq0$,
					\item $\omega(\nabla R_x,\nabla R_y)(\boldsymbol{\eta})\neq0$,
					\item The set $\{\boldsymbol{J}(\nabla R_x)(\boldsymbol{\eta}),\boldsymbol{J}(\nabla R_y)(\boldsymbol{\eta}),\boldsymbol{\eta}\}$ is linearly dependent.
				\end{enumerate}
				Then, the two polynomials $x(a^\dagg,a)$ and $y(a^\dagg,a)$ generate an infinite-dimensional real Lie algebra which is explicitly given by $\g:=\lie{\{x(a^\dagg,a),y(a^\dagg,a)\}}\subseteq\hat{A}_1$.
			\end{theorem}
			
			\begin{proof}
				Let $x$ be $y$ be defined as in the statement. By Lemma~\ref{lem:Igusa:1} and Lemma~\ref{lem:Igusa:3}, there exists a symplectic matrix $\boldsymbol{\sigma}\in\operatorname{Sp}(2,\C)$ such that the transformed polynomials satisfy the following expansions:
				\begin{align*}
					\sigma R_x(a^\dagg,a)&\upto{d_x}R_{\sigma x}(a^\dagg,a)\upto{d_x}a_0a^{d_x}+a_1a^\dagg a^{d_x-1}+\mathcal{O}\left((a^\dagg)^2a^{d_x-2}+\ldots+(a^\dagg)^{d_x}\right),\\
					\sigma R_y(a^\dagg,a)&\upto{d_y}R_{\sigma y}(a^\dagg,a)\upto{d_y}b_0a^{d_x}+b_1a^\dagg a^{d_x-1}+\mathcal{O}\left((a^\dagg)^2a^{d_x-2}+\ldots+(a^\dagg)^{d_x}\right),
				\end{align*}
				where $a_0b_0\neq0$ and $d_xa_0b_1-d_ya_1b_0\neq0$. According to Corollary~\ref{cor:extension:isomorphism:enveloping:algebra}, the symplectic transformation matrix $\boldsymbol{\sigma}$ induces a Lie-algebra isomorphism. Therefore the Lie algebra generated by $x(a^\dagg,a)$ and $y(a^\dagg,a)$ is isomorphic to the Lie algebra generated by the transformed elements $\sigma x(a^\dagg,a)$ and $\sigma y(a^\dagg,a)$. Since these transformed elements satisfy the conditions of Lemma~\ref{lem:Igusa:5}, it follows that the Lie algebra they generate is infinite dimensional. Hence, the original Lie algebra $\g:=\lie{\{x(a^\dagg,a),y(a^\dagg,a)\}}$ is also infinite dimensional.
			\end{proof}

			\section{Realizability of Lie algebras with compact Lie groups\label{App:Realizability:Recap}}
			
			This section is dedicated to providing a contextualized and comprehensive proof of the claim that no semisimple Lie algebra of rank $n$ or higher, associated with a compact Lie group, can be realized within the real $n$-mode skew-hermitian Weyl algebra $\hat{A}_n$. While this result was originally proven in \plainrefs{Joseph:1972}, the presentation there may pose challenges for some readers. In particular, the use of the same symbol for both the number of modes and the dimension of the Lie algebra could lead to ambiguity.  Additionally, the argument presumes a high level of familiarity with the subject, occasionally invoking properties without detailed justification. The proof also relies on an inductive approach, though the inductive steps are not explicitly carried out.
			
			We take this as our starting point to provide here a complete and contextualized proof that can be understood even by readers who are not deeply familiar with the topic.
			
			\begin{theorem}[Theorem 4.3 from \plainrefs{Joseph:1972} (simplified)]\label{thm:Casimir:Element:identity}
				Let $\g$ be a semisimple Lie algebra of rank $n$ with realization in the $n$-mode complex Weyl algebra $A_n$. Then the Casimir elements of the enveloping algebra $U(\g)$ are all constant multiples of the identity.
			\end{theorem}

			\begin{proof}
				First, we note that Theorem 4.1 from \plainrefs{Joseph:1972} guarantees that any semisimple Lie algebra $\g$ of rank $n$ admits a subalgebra $\gh$ of dimension $2n$ with basis $\{h_1,\ldots,h_n,e_1,\ldots,e_n\}$ satisfying
				\begin{align}\label{eqn:appendix:subalgebra:semisimple:rank:n:algebra}
					[h_j,h_k]&=0,\;&\;[e_j,e_k]&=0,\;&\;[h_j,e_k]&=\alpha_{jk}e_k,
				\end{align}
				where $\det(\boldsymbol{\alpha})\neq0$ and $\alpha_{jk}\in\mathbb{Q}$ for all $j,k\in\{1,\ldots,n\}$. Let $C\in U(\g)$ be a Casimir element of the universal enveloping algebra $U(\g)$. It must, by definition, belong to the center of the enveloping algebra $U(\g)$ and consequently satisfy $[x,C]=0$ for all $x\in\g$. Assume now that $C$ is non-trivial. Then, by Theorem 3.3 from \plainrefs{Joseph:1972}, there exist complex and non-vanishing coefficients $c_{j_1,\ldots,j_{n+1}}$, such that:
				\begin{align*}
					u=\sum_{j_1,\ldots,j_{n+1}}c_{j_1,\ldots,j_{n+1}}e_1^{j_1}\ldots e_n^{j_n}C^{j_{n+1}}=0,
				\end{align*}
				It is straightforward to verify that
				\begin{align*}
					[h_k,e_1^{j_1}\ldots e_n^{j_n}]=\sum_{p=1}^n j_p\alpha_{kp} e_1^{j_1}\ldots e_n^{j_n}.
				\end{align*}
				Therefore, for any $\ell\in\N_{\geq1}$, one has 
				\begin{align*}
					\operatorname{ad}_{h_k}^\ell(u)=\sum_{j_1,\ldots,j_{n+1}}c_{j_1,\ldots,j_{n+1}}\left(\sum_{p=1}^n j_p\alpha_{ kp}\right)^\ell e_1^{j_1}\ldots e_n^{j_n}C^{j_{n+1}}=0,
				\end{align*}
				since $C$ commutes with any element in $\g$. It then follows that the quantity $\sum_{p=1}^n j_p\alpha_{ kp}$  must be a constant for all $n$-tuples of indices $(j_1,\ldots,j_n)$ appearing in the sum above, since the expression for $u$ must consist of two or more terms that are linearly independent by the Poincaré-Birkhoff-Witt Theorem \plainrefs{PBW:Thm}. Thus, there exist at least two distinct integer-$n$-tuple $(j_1,\ldots,j_n)$ and $(j_1',\ldots,j_n')$ such that
				\begin{align*}
					\sum_{p=1}^nj_p\alpha_{kp}=\sum_{p=1}^nj_p'\alpha_{kp}
				\end{align*}
				for all $k\in\{1,\ldots,n\}$. Defining, $\hat{j}_p:=j_p-j_p'$, we obtain an nonzero integer-$n$-tuple $(\hat{j}_1,\ldots,\hat{j}_n)$ satisfying:
				\begin{align*}
					\sum_{p=1}^n\hat{j}_p\alpha_{kp}=0.
				\end{align*}
				This implies that the rows of the matrix $\boldsymbol{\alpha}$ are linearly dependent. In particular, there exists at least one index $k'$ such that either $\alpha_{k'p}=0$ for all $p$, or 
				\begin{align*}
					\alpha_{k'p}=\frac{1}{\hat{j}_{k'}}\sum_{\underset{p=1}{p\neq k'}}^n\hat{j}_p\alpha_{k'p}.
				\end{align*}
				In either case, the matrix $\boldsymbol{\alpha}$ has linearly dependent column vectors, implying $\det(\boldsymbol{\alpha})=0$, which contradicts the assumption $\det(\boldsymbol{\alpha})\neq0$. The Casimir element $C$ must therefore be trivial.
			\end{proof}

			\begin{theorem}[Theorem 4.2 from \plainrefs{Joseph:1972}]\label{them:Joseph:4:2:}
				Let $\g$ be a semisimple Lie algebra of rank $r>n$. Then $\g$ has no realizations in the complex $n$-mode Weyl algebra $A_n$.
			\end{theorem}
			
			\begin{proof}
				This follows by an argument analogous to that of Theorem~\ref{thm:Casimir:Element:identity}. If a semisimple Lie algebra of rank $r>n$ were realizable in $A_n$, then it would admit a subalgebra of dimension $2r>2n$ with the same structural properties as those captured in equation \eqref{eqn:appendix:subalgebra:semisimple:rank:n:algebra}. One can therefore replace $C^{j_{n+1}}$ with the appropriate product $e_{n+1}^{j_{n+1}}\ldots e_r^{j_r}$. From that point onward, retracing the remaining steps of the proof of Theorem~\ref{thm:Casimir:Element:identity} again leads to a contradiction, as the linear independence of the monomials implies $\det(\boldsymbol{\alpha})=0$. Hence, such a realization is not possible.
			\end{proof}

			\begin{theorem}\label{thm:compact:Lie:group:negative:definite:Killing:form}
				If $G$ is a compact Lie group with Lie algebra $\g$, then the Killing form on $\g$ is negative semidefinite. If $\g$ is furthermore semisimple, then the Killing form is negative definite.
			\end{theorem}
			
			\begin{proof}
				The first statement is Corollary 4.26 from \plainrefs{Knapp:1996}. The second statement can be proven analogously. One simply recalls that, by Cartans's criterion, a Lie algebra is semisimple if and only if its Killing form is non-degenerate \plainrefs{Dieudonne:1953}. This implies that $\mathrm{Tr}((\operatorname{ad}_X)^2)\leq 0$ for all $X\in\g$, used in the final step of the proof of Corollary 4.26 in \plainrefs{Knapp:1996}, can be replaced by the stronger condition $\mathrm{Tr}((\operatorname{ad}_X)^2)<0$ for all $X\in\g$, thereby showing that the Killing form on a semisimple Lie algebra associated with a compact Lie group is negative definite.
			\end{proof}
			
			\begin{theorem}\label{thm:form:of:quadratic:Casimir:element}
				Let $\g$ be a semisimple and finite-dimensional Lie algebra with basis $\{x_j\}_{j=1}^N$. Then the element 
				\begin{align}
					C=\sum_{j,k=1}^N\kappa_{jk} x_jx_k
				\end{align}
				is a Casimir element of $\g$ in the universal enveloping algebra $U(\g)$, where $\kappa_{jk}$ is the inverse matrix of the matrix associated with the Killing form.
			\end{theorem}
			
			\begin{proof}
				Let $\g$ be a semisimple Lie algebra of dimension $N$ with basis $\{x_j\}_{j=1}^N$, and let $B$ denote its Killing form, defined on the basis elements by $B(x_j,x_k)=b_{jk}$. Since $\g$ is semisimple, the Killing form is, by Cartan's Criterion, non-degenerate \plainrefs{Dieudonne:1953}. Hence, there exists a matrix $\boldsymbol{\kappa}$ such that $\sum_{\ell=1}^N \kappa_{j\ell}b_{jk}=\delta_{jk}$. Because the Killing form is symmetric, i.e., $b_{jk}=b_{kj}$, its inverse $\kappa_{jk}$ is also symmetric. This can be verified that for any invertible symmetric $n\times n$-matrix $A$:
				\begin{align*}
					A^{-1}=\mathds{1}A^{-1}=(AA^{-1})^{\Tp}A^{-1}=(A^{-1})^{\Tp}A^{\Tp}A^{-1}=(A^{-1})^{\Tp}AA^{-1}=(A^{-1})^{\Tp}.
				\end{align*}
				Now consider the element
				\begin{align}
					C=\sum_{j,k=1}^N\kappa_{jk}x_jx_k\in U(\g).
				\end{align}
				We compute its commutator with an arbitrary basis element $x_\ell\in\g$:
				\begin{align*}
					[C,x_\ell]&=\sum_{j,k=1}^N[\kappa_{jk}x_jx_k,x_\ell]=\sum_{j,k=1}^N\kappa_{jk}\left(x_k[x_j,x_\ell]+[x_k,x_\ell]x_j\right).
				\end{align*}
				Using the symmetry of $\kappa_{jk}$ and relabeling indices, we obtain:
				\begin{align*}
					[C,x_\ell]&=\sum_{j,k=1}^N\kappa_{jk}\left(x_k[x_j,x_\ell]+[x_j,x_\ell]x_k\right)=\sum_{j,k,p=1}^N\kappa_{jk}f_{j\ell p}(x_kx_p+x_px_k),
				\end{align*}
				where $f_{j\ell p}$ are the structure constants of $\g$ determined by $[x_j,x_k]=\sum_{\ell=1}^Nf_{jk\ell}x_\ell$. Noting that $k$ and $p$ are not free indices, allows again to obtain by relabeling:
				\begin{align*}
					[C,x_\ell]&=\sum_{j,k,p=1}^N\left(\kappa_{jk}f_{j\ell p}+\kappa_{jp}f_{j\ell k}\right)x_kx_p.
				\end{align*}
				We can now use the fact that the Killing form satisfies the identity $B([a,b],c)=B(a,[b,c])$ for all $a,b,c\in\g$. This can be shown by recalling that the Killing form is typically defined as \plainrefs{Procesi:2006}:  
				\begin{align*}
					B:\g\times\g\to\C,
					\qquad
					\text{via its action}
					\qquad
					(x,y)\mapsto B(x,y):=\operatorname{Tr}(\operatorname{ad}_{x}\circ\operatorname{ad}_{y}).
				\end{align*}
				First, observe that by the Jacobi identity:
				\begin{align*}
					\operatorname{ad}_{[x,y]}(\cdot)=[[x,y],\cdot]=-[[y,\cdot],x]-[[\cdot,x],y]=[x,[y,\cdot]]-[y,[x,\cdot]]=\operatorname{ad}_x(\operatorname{ad}_y(\cdot))-\operatorname{ad}_y(\operatorname{ad}_x(\cdot)).
				\end{align*}
				Using the cyclic property of linearity of the trace, as well es the bilinearity of the commutator, one obtains:
				\begin{align*}
					B([x,y],z)&=-B([y,x],z)=-\operatorname{Tr}(\operatorname{ad}_{[y,x]}\circ\operatorname{ad}_z)=-\operatorname{Tr}(\operatorname{ad}_y\circ \operatorname{ad}_x\circ\operatorname{ad}_z)+\operatorname{Tr}(\operatorname{ad}_x\circ\operatorname{ad}_y\circ\operatorname{ad}_z)\\
					&=\operatorname{Tr}(\operatorname{ad}_x\circ\operatorname{ad}_y\circ\operatorname{ad}_z)-\operatorname{Tr}(\operatorname{ad}_x\circ\operatorname{ad}_z\circ\operatorname{ad}_y)=\operatorname{Tr}(\operatorname{ad}_x\circ\operatorname{ad}_{[y,z]})=B(x,[y,z]).
				\end{align*}
				Using this property, we find:
				\begin{align*}
					f_{j\ell k}&=\sum_{q=1}^Nf_{j\ell q}\delta_{qk}=\sum_{q=1}^N f_{j\ell q}\sum_{p=1}^N\kappa_{kp}b_{pq}=\sum_{p,q=1}^N\kappa_{j\ell q}\kappa_{kp}b_{qp}=\sum_{p,q=1}^NB(f_{j\ell q}x_q,\kappa_{kp}x_p)=\sum_{p=1}^N B([x_j,x_\ell],\kappa_{kp}x_p)\\
					&=-\sum_{p=1}^NB(x_j,[\kappa_{kp}x_p,x_\ell])=-\sum_{p,q=1}^N B(x_j,\kappa_{kp}f_{p\ell q}x_q)=-\sum_{p,q=1}^N\kappa_{kp}f_{p\ell q}b_{jq}.
				\end{align*}
				This allows us to calculate:
				\begin{align*}
					\sum_{j,k,r=1}^N\kappa_{jr}f_{j\ell k}x_k x_r&=-\sum_{j,k,r,p,q=1}^Nf_{p\ell q}\kappa_{kp}\kappa_{jr}b_{jq}x_kx_r=-\sum_{k,r,p,q=1}^Nf_{p\ell q}\kappa_{kp}\delta_{rq}x_kx_r\\
					&=-\sum_{k,q,p}^Nf_{p\ell q}\kappa_{kp}x_kx_q=-\sum_{j,k,r}^N\kappa_{jk}f_{j\ell r}x_kx_r.
				\end{align*}
				Hence, we conclude
				\begin{align*}
					[C,x_\ell]&=\sum_{j,k,r=1}^N\left(\kappa_{jk}f_{j\ell r}+\kappa_{jr}f_{j\ell k}\right)x_kx_r=0.
				\end{align*}
				Since this holds for any basis element $x_\ell\in\g$, it follows that $C$ is in the center of $U(\g)$.
			\end{proof}
			
			Note that this proof can be performed more elegantly using tensorial notation and Einstein's summation convention \plainrefs{Einstein:1916}. One simply needs to interpret $\kappa_{jk}$ as $b^{jk}$, where $b_{jk}$ is used to \textit{lower indices} and $b^{jk}$ is used to \textit{raise them}, analogously to the role of the metric tensor in differential geometry \plainrefs{Fecko:2006}.
			
			\vspace{0.5cm}
			
			\begin{tcolorbox}[breakable, colback=Cerulean!3!white,colframe=Cerulean!85!black,title=An alternative proof of Theorem~\ref{thm:form:of:quadratic:Casimir:element}]
				\begin{proof}
					Consider an $N$-dimensional semisimple Lie algebra $\g$ with basis $\{x_j\}_{j=1}^N$. Let the Killing form be denoted by $B(x_j,x_k):=\kappa_{jk}$. By Cartan's criterion, the Killing form is non-degenerate, so there exists an inverse of $\kappa_{jk}$, which we will denote by $\kappa^{jk}$. Using Einstein's summation convention, one has $\kappa^{j\ell}\kappa_{\ell k}={\delta^j}_k$. Since $\kappa_{jk}$ is symmetric, its inverse $\kappa^{jk}$ is also symmetric. We can now use these two elements to define index raising an lowering operations: $x^j:=\kappa^{jk}x_k$ and $x_j:=\kappa_{jk}x^k$. Let us compute the commutator:
					\begin{align*}
						[x^jx_j,y]&=x^j[x_j,y]+[x^j,y]x_j=x^j[x_j,y]+\kappa^{jk}[x_k,y]x_j=x^j[x_j,y]+[x_j,y]x^j.
					\end{align*}
					Now, write the commutator of $x_j$ and the arbitrary element $y\in\g$ as the linear combination $[x_j,y]={\alpha_j}^kx_k$. Substituting this into the expression above, we obtain:
					\begin{align}\label{eqn:second:proof:casimir:element:form}
						[x^jx_j,y]&={\alpha_j}^k(x^jx_k+x_kx^j)={\alpha_j}^k(\kappa^{j\ell}x_\ell x_k+\kappa^{j\ell}x_kx_\ell)=\alpha^{\ell k}(x_\ell x_k+x_k x_\ell)=(\alpha^{\ell k}+\alpha^{k\ell})x_kx_\ell.
					\end{align}
					Here, we simply relabeled contracted indices and used the symmetry of $\kappa^{jk}$. We compute now ${\alpha_j}^k$ using the Killing form:
					\begin{align*}
						{\alpha_j}^k&={\alpha_j}^\ell{\delta_\ell}^k=B({\alpha_j}^\ell x_\ell,x^k)=B([x_j,y],x^k)=-B(x_j,[x^k,y])=-B(x_j,\kappa^{k\ell}[x_\ell,y])\\
						&=-B(x_j,\kappa^{k\ell}{\alpha_\ell}^qa_q=-\kappa^{k\ell}{\alpha_\ell}^qB(x_j,\kappa_{qp}x^p)=-\kappa^{k\ell}{\alpha_\ell}^q\kappa_{qp}{\delta_j}^p=-{\alpha}^{kq}\kappa_{qj}=-{\alpha^k}_j.
					\end{align*}
					Raising the first index, we obtain: $\alpha^{jk}=\kappa^{j\ell}{\alpha_j}^k=\kappa^{j\ell}(-{\alpha^k}_\ell)=-\alpha^{kj}$. Hence $\alpha^{jk}$ is antisymmetric. In the computation above, we used the property $B([x,y],z)=B(x,[y,z])$ shown previously, as well as the observation $B(x_j,x^k)=B(x_j,\kappa^{k\ell}x_\ell)=\kappa_{k\ell}\kappa_{j\ell}={\delta_j}^k$. Since $\alpha^{jk}$ is anti-symmetric, equation \eqref{eqn:second:proof:casimir:element:form} vanishes, and this shows that $x^jx_j$ is an element of $\mathcal{Z}(U(\g))$, making it a Casimir element.
				\end{proof}
			\end{tcolorbox}
			
			It is worth noting that this proof proceeds in the opposite direction as that usually taken in the literature. Here, we begin by defining the Killing form and then demonstrate that the quadratic Casimir element lies in the center of $U(\g)$. Typically, one defines a Casimir element as an element of the center of $U(\g)$ and then derives the properties of the Killing form from it.
			
			\begin{corollary}\label{cor:quadratic:Casimir:element:form}
				Theorem~\ref{thm:form:of:quadratic:Casimir:element} holds true also for simple Lie algebras.
			\end{corollary}
			
			\begin{proof}
				This follows directly from Theorem~\ref{thm:form:of:quadratic:Casimir:element}, since any simple algebra is also semisimple.
			\end{proof}
			
			\begin{theorem}[Theorem 4.4 from \plainrefs{Joseph:1972} for single mode (reduced)]
				Let $\g$ be a finite-dimensional semisimple Lie algebra of rank one with a diagonalizable Killing form. Then $\g$ has no realizations within the skew-hermitian Weyl algebra $\hat{A}_1$.
			\end{theorem}
			
			The idea of this proof is the following: We assume that $\g$ can be realized in the complex Weyl algebra $A_1$. Using the fact that the Killing form is diagonalizable, we choose a basis in which the Casimir element $C$ is the sum of squares of all basis elements. We then express these basis elements in terms of generators of $A_1$. These generators are always three in number, of which one is the trivial unit element, while the other two have the following properties: one is hermitian while another is skew-hermitian. Since $\g$ is a semisimple Lie algebra of rank one, the Casimir element must be a scalar multiple of the identity. Expanding $C$ in terms of the basis elements and applying the Poincaré-Birkhoff-Witt Theorem leads us to a system of equations for the coefficients of the elements in the expansion . An inductive argument shows that all coefficients must vanish, contradicting the assumption that the basis elements are non-zero since they are obtained as linear combinations of the generators with such coefficients. Therefore, the desired realization in the real skew-hermitian Weyl algebra $\hat{A}_1$ is impossible.
			
			\begin{proof}
				Let $\g$ be an $N$-dimensional semisimple Lie algebra of rank one that admits a realization in the complex Weyl algebra $A_1$. To avoid a cumbersome notation, we denote the abstract Lie algebra $\g$ and its realization in $A_1$ with the same symbol. By
				Theorem~\ref{thm:form:of:quadratic:Casimir:element}, $\g$ contains the Casimir element of the form $C=\sum_{jk}\kappa_{jk}x_jx_k$, where $\{x_j\}_{j=1}^N$ is a basis of $\g$, and $\kappa_{jk}$ is the inverse matrix of the matrix associated  with the Killing form. Since the Killing form is assumed to be diagonalizable, there exists a basis $\{e_j\}_{j=1}^N$ such that $C=\sum_{j=1}^Ne_j^2$. Here, we want to mention that $N=\dim(\g)\geq 3$, as any Lie algebra of dimension two or lower is trivially solvable and therefore not semisimple.
				
				Let $m:=\max\{\deg(e_j)\}$. This must clearly be a positive integer, since otherwise $\g$ would be abelian, contradicting the non-degeneracy of the Killing form (by Cartan's criterion for semisimplicity \plainrefs{Dieudonne:1953}).
				
				We now choose the basis $x:=(a+a^\dagg)/\sqrt{2}$ and $y:=-(a-a^\dagg)/\sqrt{2}$ of $\gh_{1,\C}\cong\lie{\{a^\dagg,a,1\}}\cong\lie{\{x,y,1\}}$, which is possible, since $[x,y]=1$, establishing the Lie-algebra isomorphism $\lie{\{a^\dagg,a,1\}}\cong\lie{\{x,y,1\}}$ via the identification $x\leftrightarrow a$, $y\leftrightarrow a^\dagg$, and $1\rightarrow 1$. This basis has the advantage that $x^\dagg=x$ and $y^\dagg=-y$.
				Each $e_j\in A_1\cong U(\gh_{1,\C})$ can, by the Poincaré-Birkhoff-Witt Theorem \plainrefs{PBW:Thm}, be written as:
				\begin{align*}
					e_j\upto{m}\sum_{k=0}^m\alpha_{jk} x^{m-k}y^k,
				\end{align*}
				where $\alpha_{jk}\in\C$.
				Assuming $\g$ has a realization within the real skew-hermitian Weyl algebra $\hat{A}_1$ implies that the elements $e_j$ must be skew-hermitian. This implies:
				\begin{align*}
					e_j^\dagg\upto{m}\sum_{k=0}^m\alpha_{jk}^*(-1)^ky^kx^{m-k}\upto{m}\sum_{k=0}^m\alpha_{jk}^*(-1)^kx^{m-k}y^k\upto{m}-\sum_{k=0}^m\alpha_{jk}x^{m-k}y^k\upto{m}-e_j.
				\end{align*}
				In the first step, we used the fact that $(AB)^\dagg=B^\dagg A^\dagg$, $\alpha_{jk}\in\C$, $x^\dagg=x$, and $y^\dagg=-y$. In the second step, we commuted $x^{m-k}$ and $y^k$. This is possible as $\deg([x^\mu,y^{\nu}])\leq \mu+\nu-2$ holds (cf. Proposition 13 in~\plainrefs{Bruschi:Xuereb:2024}), which shows that $x^{m-k}y^k\upto{m}y^kx^{m-k}$, as commuting the two terms only adds terms of degree smaller than $m$. The last two steps, follow from the condition that $e_j$ is skew-hermitian, i.e., $e_j^\dagg=-e_j$. The Poincaré-Birkhoff-Witt Theorem \plainrefs{PBW:Thm} allows us now the identification $\alpha_{jk}=-(-1)^k\alpha_{jk}^*$ for all $j\in\{1,\ldots,n\}$ and $k\in\{1,\ldots,m\}$. Theorem~\ref{thm:Casimir:Element:identity} now also guarantees that $C=c\mathds{1}$ for an appropriate $c\in\C$. However, we then have:
				\begin{align*}
					C=\sum_{j=1}^Ne_j^2\upto{2m}\sum_{j=1}^N\sum_{k,\ell=0}^m\alpha_{jk}x^{m-k}y^k\alpha_{j\ell}x^{m-\ell}y^\ell\upto{2m}\sum_{j=1}^N\sum_{k,\ell=0}^m\alpha_{jk}\alpha_{j\ell}x^{2m-(k+\ell)}y^{k+\ell}\upto{2m}0.
				\end{align*}
				Here, we first recalled that the Casimir element is, in the chosen basis, simply the sum of all basis elements squared. In the next step, we expanded the basis elements $e_j$ to highest degree and noted that the degree of $e_j^2$ is twice the degree of $e_j$. In the last step, we reordered the $x$ and $y$ terms, which only adds terms of lower degree, thereby not changing the \enquote{$\upto{2m}$} relation.
				The Poincaré-Birkhoff-Witt Theorem now implies that the monomials $x^ry^s$ are linearly independent, so the above expression can only vanish if the following sum vanishes for all $z\in\{0,1,\ldots,2m\}$:
				\begin{align}\label{eqn:in:proof:no:skew:hermitian:compact:semisimple:sum}
					\sum_{k+\ell=z}\sum_{j=1}^N\alpha_{jk}\alpha_{j\ell}=0.
				\end{align}
				We now prove inductively that equation \eqref{eqn:in:proof:no:skew:hermitian:compact:semisimple:sum} implies $\alpha_{jk}=0$ for all $j\in\{1,\ldots,N\}$ and $k\in\{0,\ldots,m\}$. Before proceeding, observe that by the symmetry of $\alpha_{jk}\alpha_{j\ell}$ under the exchange $k\leftrightarrow\ell$, and the identity $\alpha_{jk}=-(-1)^k\alpha_{jk}^*$, one has:
				\begin{align*}
					\sum_{k+\ell=z}\sum_{j=1}^N\alpha_{jk}\alpha_{j\ell}=2\sum_{\underset{k<\ell}{k+\ell=z}}\sum_{j=1}^N\alpha_{jk}\alpha_{j\ell}-\sum_{2k=z}\sum_{j=1}^n(-1)^k|\alpha_{jk}|^2,
				\end{align*}
				where the first sum in the second term is empty if $z$ is odd and consists of a single term if $z$ is even.
				
				We proceed with the induction. The induction hypothesis is that for every $2z\in\{0,2,\ldots,2m\}$, equation \eqref{eqn:in:proof:no:skew:hermitian:compact:semisimple:sum} implies $\alpha_{jk}=0$ for all $k\leq z$ and all $j\in\{1,\ldots,n\}$.
				
				Consider the base case $z=0$. Here, equation \eqref{eqn:in:proof:no:skew:hermitian:compact:semisimple:sum} reduces to the equation
				\begin{align*}
					-\sum_{j=1}^N|\alpha_{j0}|^2=0,
				\end{align*}
				since the only pair of non-negative integers $j,k$ that sum to $0$ is $(j,k)=(0,0)$.
				This equation implies immediately $\alpha_{j0}=0$ for all $j\in\{1,\ldots,N\}$, as the sum of the squares of the absolute values of a set of complex numbers is positive definite. 
				
				We can consequently continue with the induction step. Assume that for one element $2z\in\{0,2,\ldots,2m\}$, one has $\alpha_{jk}=0$ for all $k\leq z$ and $j\in\{1,\ldots,N\}$. Then equation \eqref{eqn:in:proof:no:skew:hermitian:compact:semisimple:sum} reads for $z':=2(z+1)$:
				\begin{align*}
					2\sum_{\underset{k<\ell}{k+\ell=2(z+1)}}\sum_{j=1}^N\alpha_{jk}\alpha_{j\ell}-\sum_{j=1}^N(-1)^{z+1}|\alpha_{j,z+1}|^2=0.
				\end{align*}
				The condition $k+\ell=2(z+1)$ with $k<\ell$ in the first sum implies $k<(k+\ell)/2<z+1$. Hence $k\leq z$. By the induction hypothesis, $\alpha_{jk}=0$ for all such $k$ and $j\in\{1,\ldots,N\}$. So the first sum vanishes and one is left with:
				\begin{align*}
					\sum_{j=1}^N|\alpha_{j,z+1}|^2=0.
				\end{align*}
				This, together with the induction hypothesis, implies $\alpha_{jk}=0$ for all $k\leq z+1$ and $j\in\{1,\ldots,N\}$, which concludes the proof by induction.
				
				We have consequently shown that $\alpha_{jk}=0$ for all $j\in\{1,\ldots,N\}$ and $k\in\{0,\ldots,m\}$. This, however, cannot be true, as we assumed $e_j\upto{m}\sum_k\alpha_{jk} x^{m-k} y^k\not\simeq_m0$ for at least one $j\in\{1,\ldots,N\}$. Thus $\g$ has no realizations within the real skew-hermitian Weyl algebra $\hat{A}_1$.
			\end{proof}
			
			We can repeat the previous proof for the case of the $n$-mode Weyl Algebra $A_n$.
			
			\begin{theorem}[Theorem 4.4 from \plainrefs{Joseph:1972} for $n$ modes (reduced)]\label{thm:Joseph:4:4:full}
				Let $\g$ be a finite-dimensional semisimple Lie algebra of rank $n$ with diagonalizable Killing form. Then $\g$ has no realizations in the skew-hermitian Weyl algebra $\hat{A}_n$.
			\end{theorem}
			
			The idea behind this proof is similar to the previous one. However, we anticipate that the induction procedure will be more involved.
			
			\begin{proof}
				Let $\g$ be an $N$-dimensional semisimple Lie algebra of rank $n$ that admits a realization within the complex Weyl algebra $A_n$ of $n$ modes. To avoid a combersome notation, we denote the abstract Lie algebra $\g$ and its realization in $A_n$ with the same symbol. By Theorem~\ref{thm:form:of:quadratic:Casimir:element}, $\g$ contains the Casimir element of the form $C=\sum_{jk}\kappa_{jk}x_jx_k$, where $\{x_j\}_{j=1}^N$ is a basis of $\g$, and $\kappa_{jk}$ the inverse matrix of the matrix associated with the Killing form. The Killing form is, by assumption, diagonalizable. Thus, there exists a basis $\{e_j\}_{j=1}^N$ such that $C=\sum_{j=1}^Ne_j^2$. Let $m:=\max\{\deg(e_j)\}$. Clearly $m\geq 1$, since otherwise $\g$ would be abelian, contradicting Cartan's criterion for semisimplicity \plainrefs{Dieudonne:1953}.
				
				We now choose the basis $x_\mu=(a_\mu+a_\mu^\dagg)/\sqrt{2}$ and $y_\mu=-(a_\mu-a_\mu^\dagg)/\sqrt{2}$ of $\mathfrak{h}_{n,\C}\cong\lie{\{x_1,y_1,\ldots,x_n,y_n,1\}}\cong\lie{\{a_1^\dagg,a_1,\ldots,a_n^\dagg,a_n,1\}}$, which is possible, since $[x_\mu,y_\nu]=x_\mu y_\nu-y_\nu x_\mu=\delta_{\mu\nu}$ establishes the Lie-algebra isomorphism $\lie{\{x_1,y_1,\ldots,x_n,y_n,1\}}\cong\lie{\{a_1^\dagg,a_1,\ldots,a_n^\dagg,a_n,1\}}$ via the identification $x_\mu\leftrightarrow a_\mu$, $y_\mu\leftrightarrow a^\dagg_\mu$, and $1\leftrightarrow 1$, where $a_\mu$ is the bosonic annihilation operator of the $\mu$-th mode, while $a^\dagg_\mu$ is the bosonic creation operator of the $\mu$-th mode. This basis has the advantage that $x_\mu^\dagg=x_\mu$ and $y_\mu^\dagg=-y_\mu$. 
				
				Before proceeding, we want to introduce the multi-index $\gamma$ of an $n$-mode system, as a non-negative integer vector of length $2n$, i.e., $\gamma\in\N_{\geq0}^{2n}$ . This allows us to introduce the magnitude of a multi-index as:
				\begin{align*}
					|\,\cdot\,|:\N_{\geq0}^{2n}\to \N_{\geq0},\;\gamma\mapsto|\gamma|:=\sum_{\mu=1}^{2n}|\gamma_\mu|=\sum_{\mu=1}^{2n}\gamma_\mu.
				\end{align*}
				Using the Poincaré-Birkhoff-Witt Theorem \plainrefs{PBW:Thm}, we can consequently write any $e_j\in U(\mathfrak{h}_{N,\C})=A_n$ as:
				\begin{align*}
					e_j\upto{m}\sum_{\gamma\in S(m)}\alpha_{j\gamma} x^\gamma\quad\text{with}\quad x^\gamma:=\prod_{\mu=1}^n x_\mu^{\gamma_\mu}y_\mu^{\gamma_{\mu+n}}, 
				\end{align*}
				where $\alpha_{j\gamma}\in\C$ and $S(m):=\{\gamma\in \Np{2n}\,\mid\,|\gamma|=m\}$ is the set of all multi-indices of magnitude $m$.
				Assuming $\g$ has a realization within the real skew-hermitian Weyl algebra $\hat{A}_1$, i.e., $\g\subseteq\hat{A}_n$, implies that the elements $e_j$ must be skew-hermitian. Consequently:
				\begin{align}
					e_j^\dagg\upto{m}\sum_{\gamma\in S(m)}\alpha_{j\gamma}^*(-1)^{\sum_{\mu=1}^n\gamma_{\mu+n}} x^\gamma \upto{m}-\sum_{\gamma\in S(m)}\alpha_{j\gamma}x^\gamma\upto{m}-e_j.\label{eqn:help:thm;joseph:fin:1}
				\end{align}
				Let us explain the first step in detail. Note that $(AB)^\dagg = B^\dagg A^\dagg$ and consequently
				\begin{align*}
					(x^\gamma)^\dagg=\left(\prod_{\mu=1}^n x_\mu^{\gamma_\mu}y_\mu^{\gamma_{\mu+n}}\right)^\dagg=\prod_{\mu=1}^n \left(y_{n+1-\mu}^{\gamma_{2n+1-\mu}}\right)^\dagg \left(x_{n+1-\mu}^{\gamma_{n+1-\mu}}\right)^\dagg=\prod_{\mu=1}^n (-1)^{\gamma_{2n+1-\mu}}y_{n+1-\mu}^{\gamma_{2n+1-\mu}}x_{n+1-\mu}^{\gamma_{n+1-\mu}},
				\end{align*}
				since $x_\mu^\dagg=x_\mu$ and $y_\mu^\dagg=-y_\mu$. Now, we recognize that any pair of products $y_{n+1-\mu}^{\gamma_{2n+1-\mu}} x_{n+1-\mu}^{\gamma_{n+1-\mu}}$ and $y_{n+1-\nu}^{\gamma_{2n+1-\nu}} x_{n+1-\nu}^{\gamma_{n+1-\nu}}$ commute if $\mu\neq \nu$, since $[x_\mu,x_\nu]=0=[y_\mu,y_\nu]$ and $[x_\mu,y_\nu]=\delta_{\mu\nu}$ for all $\mu,\nu\in\{1,\ldots,n\}$, and $[AB,CD]=A[B,C]D+AC[B,D]+[A,C]DB+C[A,D]B$. Thus, we can write:
				\begin{align*}
					(x^\gamma)^\dagg=\prod_{\mu=1}^n (-1)^{\gamma_{\mu+n}}y_\mu^{\gamma_{\mu+n}}x_\mu^{\gamma_\mu}=(-1)^{\sum_{\mu=1}^n\gamma_{\mu+n}}\prod_{\mu=1}^ny_\mu^{\gamma_{\mu+n}}x_\mu^{\gamma_\mu}.
				\end{align*}
				Finally, we observe that by $[x_\mu,y_\nu]=\delta_{\mu\nu}$, one can commute $y_\mu^{\gamma_{\mu+n}}$ and $x_\mu^{\gamma_\mu}$ without adding terms of degree $\gamma_\mu+\gamma_{\mu+n}$, but instead only terms of lower degree. Hence
				\begin{align*}
					(x^\gamma)^\dagg\upto{|\gamma|}(-1)^{\sum_{\mu=1}^n\gamma_{\mu+n}}\prod_{\mu=1}^nx_\mu^{\gamma_\mu}y_\mu^{\gamma_{\mu+n}}=(-1)^{\sum_{\mu=1}^n\gamma_{\mu+n}}x^\gamma.
				\end{align*}
				The first step in equation \eqref{eqn:help:thm;joseph:fin:1} follows consequently from the observation that $m:=\max\{\deg(e_j)\}= |\gamma|$ for all $\gamma\in S(m)$ and the fact that, for all $u,v\in A_n$, $u\upto{f_1} v$ implies $v\upto{f_2}v$ if $f_1\leq f_2$. The remaining relations in \eqref{eqn:help:thm;joseph:fin:1} represent simply the fact that $e_j$ must be skew-hermitian.
				By means of the Poincaré-Birkhoff-Witt Theorem \plainrefs{PBW:Thm}, the monomials $x^{\gamma}$ are linearly independent, yielding the condition $\alpha_{j\gamma}=-(-1)^{\sum_\mu \gamma_{\mu+n}}\alpha_{j\gamma}^*$ for all $j\in\{1,\ldots,N\}$ and $\gamma\in S(m)$. Theorem~\ref{thm:Casimir:Element:identity} guarantees now that $C=c\mathds{1}$ for some $c\in\C$. This implies:
				\begin{align*}
					C=\sum_{j=1}^Ne_j^2\upto{2m}\sum_{j=1}^n\sum_{\gamma,\lambda\in S(m)}\alpha_{j\gamma}x^\gamma\alpha_{j\lambda}x^\lambda\upto{2m}\sum_{j=1}^n\sum_{\gamma,\lambda\in S(m)}\alpha_{j\gamma}\alpha_{j\lambda}x^{\gamma+\lambda}\upto{2m}0,
				\end{align*}
				where we started by recalling the the Casimir element, in the chosen basis, is simply the sum of all squared basis elements. In the next step, we expanded the basis elements $e_j$ to highest degree and used the fact that $\deg(e_j^2)=2\deg(e_j)$. The last step follows analogously to the discussion above from the observation that commuting $x_\mu$ and $y_\mu$ elements does only add terms of lower degree.
				The Poincaré-Birkhoff-Witt Theorem \plainrefs{PBW:Thm} now implies that the elements $x^{\gamma+\lambda}$ and $x^{\gamma'+\lambda'}$ are only linearly dependent if $\gamma+\lambda=\gamma'+\lambda'$. Thus, the following sums must vanish
				\begin{align}\label{eqn:in:proof:no:skew:hermitian:compact:semisimple:sum:nmode}
					\sum_{\underset{\gamma+\lambda=\gamma_*}{\gamma,\lambda\in S(m)}}\sum_{j=1}^N\alpha_{j\gamma}\alpha_{j\lambda}=0
				\end{align}
				for all $\gamma_*\in\{\gamma\in \Np{2n}\,\mid\,|\gamma|=2m\}=S(2m)$. 
				
				We now prove inductively that \eqref{eqn:in:proof:no:skew:hermitian:compact:semisimple:sum:nmode} already  implies $\alpha_{j\gamma}=0$ for all $j\in\{1,\ldots,N\}$ and $\gamma\in S(m)$. To be precise, the induction hypothesis is that for a given $z\in\{0,\ldots,m\}$ the coefficients $\alpha_{j\gamma}$ vanish for every $j\in\{1,\ldots,n\}$ and $\gamma\in S(m)$ for which there exists a $\mu\in\{1,\ldots,2n\}$ and $\lambda\in S(m)$ such that $\gamma+\lambda=\gamma_*=2(m-z)\iota_\mu+\hat{\gamma}^{(z)}$, where $\iota_\mu\in\Np{2n}$ such that $(\iota_\mu)_\nu=\delta_{\mu\nu}$ and $\hat{\gamma}^{(z)}\in2\Np{2n}$ such that $\hat{\gamma}_\mu^{(z)}=0$ and $|\hat{\gamma}^{(z)}|=2z$. 
				
				We start with the base case corresponding to case $z=0$ in the previous proof. In this case, we consider the multi-indices $\gamma_*=2m\iota_\mu$ for all $\mu\in\{1,2,\ldots,2n\}$ simultaneously. The only $\gamma,\lambda\in S(m)$ satisfying $\gamma+\lambda=2m\iota_\mu$ are clearly $\gamma=\lambda=m\iota_\mu$. Equation \eqref{eqn:in:proof:no:skew:hermitian:compact:semisimple:sum:nmode} then reads:
				\begin{align*}
					0=\sum_{j=1}^N\alpha_{j,m\iota_\mu}^2\;\;\;\implies\;\;\;\sum_{j=1}^N|\alpha_{j,m\iota_\mu}|^2=0,
				\end{align*}
				since $\alpha_{j,m\iota_\mu}=-(-1)^{\sum_{\nu=1}^nm(\iota_\mu)_{\nu+n}}\alpha_{j,m\iota_\mu}^*=\epsilon \alpha_{j,m\iota_\mu}^*$, where $\epsilon=-1$ if either $\mu\leq n$ or $m\in2\N_{\geq0}$, and $\epsilon=+1$ otherwise. This implies $\alpha_{j,m\iota_\mu}=0$ for all $j\in\{1,\ldots,N\}$ and $\mu\in\{1,\ldots,2n\}$. 
				
				We proceed with the induction step and suppose the induction hypothesis holds for some $z-1\in\{0,\ldots,m\}$. Consider consequently the multi-index $\gamma_*=2(m-z)\iota_\mu+\hat{\gamma}^{(z)}$ for each $\mu\in\{1,\ldots,2m\}$, where $\hat{\gamma}_\mu^{(z)}=0$, $\hat{\gamma}^{(z)}\in2\Np{2n}$, and $|\hat{\gamma}^{(z)}|=2z$. We must consider all $\gamma,\lambda\in S(m)$ satisfying $\gamma+\lambda=\gamma_*$. Since $\gamma_\mu$ and $\lambda_\mu$ are non-negative integers that must satisfy $\gamma_\mu+\lambda_\mu=2(m-z)$, we can choose, without loss of generality, that $\gamma_\mu\geq \lambda_\mu$. This allows us to conclude that $\gamma_\mu\geq (\gamma_\mu+\lambda_\mu)/2=(m-z)$. We can now consider the two cases $\gamma_\mu>\lambda_\mu$ and $\gamma_\mu=\lambda_\mu$ separately. To differentiate these two cases, we introduce the sets $M_\mu^>:=\{\gamma\in S(m)\,\mid\,\gamma_\mu> m-z\}$, $M_\mu^<:=\{\gamma\in S(m)\,\mid\,\gamma_\mu< m-z\}$ and $M^=_\mu:=\{\gamma\in S(m)\,\mid\,\gamma_\mu=m-z\}$. It is clear that they are mutually disjoint and satisfy $M_\mu^>\cup M_\mu^=\cup M_\mu^<=S(m)$, which allows us to rewrite equation \eqref{eqn:in:proof:no:skew:hermitian:compact:semisimple:sum:nmode} as:
				\begin{align*}
					\sum_{\underset{\gamma+\lambda=\gamma_*}{\gamma+\lambda=\gamma_*}}\sum_{j=1}^N\alpha_{j\gamma}\alpha_{j\lambda}=\sum_{\underset{\gamma+\lambda=\gamma_*}{\gamma\in M_\mu^>,\lambda\in M_\mu^<}}\sum_{j=1}^N\alpha_{j\gamma}\alpha_{j\lambda}+\sum_{\underset{\lambda+\gamma=\gamma_*}{\gamma,\lambda\in M_\mu^=}}\sum_{j=1}^N\alpha_{j\gamma}\alpha_{j\lambda}=0.
				\end{align*}
				By the induction hypothesis, one has that $\alpha_{j\gamma}$ vanishes for every $j\in\{1,\ldots,n\}$ and $\gamma\in M_\mu^>$. This leaves, in the equation above, only the sum over multi-indices in the set $M_\mu^=$ as the other sum vanishes. Thus, consider the case
				$\gamma_\mu=\lambda_\mu=m-z$. This will now be treated by an iterative process, and we begin with the case $\hat{\gamma}^{(z)}=2z\iota_\nu$ for some $\nu\neq\mu$. Then, one must have $\gamma=\lambda=(m-z)\iota_\mu+z\iota_\nu$, since $\gamma,\lambda$ are non-negative integer vectors that satisfy $|\gamma|=|\lambda|=m$.
				Equation \eqref{eqn:in:proof:no:skew:hermitian:compact:semisimple:sum:nmode} then becomes
				\begin{align*}
					\sum_{j=1}^N|\alpha_{j,(m-z)\iota_\mu+z\iota_\nu}|^2=0,
				\end{align*}
				implying $\alpha_{j,(m-z)\iota_\mu+z\iota_\nu}=0$ for all $\mu\neq\nu$ and $j\in\{1,\ldots,N\}$. 
				
				Next, one considers $\hat{\gamma}^{(z)}=2(z-1)\iota_\nu+2\iota_{\nu'}$, where $\nu'\notin\{\mu,\nu\}$. Then, the only $\gamma,\lambda\in S(m)$ with $\gamma+\lambda=\gamma_*$ that are, by the induction hypothesis, not associated to coefficients $\alpha_{j\eta}$ that either satisfy $\alpha_{j\gamma}=0$ or $\alpha_{j\lambda}=0$ are $\gamma=\lambda=(m-z)\iota_\mu +(z-1)\iota_\nu+\iota_{\nu'}$. It is immediate to obtain again that the associated $\alpha_{j\gamma}$ must all vanish for all $j\in\{1,\ldots,N\}$. This procedure can be repeated for all $\hat{\gamma}^{(z)}$ with $|\hat{\gamma}^{(z)}|=2z$, ultimately showing that $\alpha_{j,\gamma}=0$ for all $\gamma=(m-z)\iota_\mu+\gamma^{(r)}$, where $|\gamma^{(r)}|=z$. This completes the proof by induction. 
				
				We have consequently shown that $\alpha_{j\gamma}=0$ for all $j\in\{1,\ldots,N\}$ and all $\gamma\in S(m)$. This contradicts the assumption that $e_j\not\simeq_m0$ for at least one $j$, completing this proof, as it shows that $\g$ has no realization within the real skew-hermitian Weyl algebra $\hat{A}_n$.
			\end{proof}
			
			\begin{corollary}\label{cor:Joseph:4:4}
				Let $\g$ be a finite-dimensional semisimple Lie algebra of rank $n$ corresponding to a compact Lie group $G$. Then $\g$ has no realizations in the skew-hermitian Weyl algebra $\hat{A}_n$.
			\end{corollary}
			
			\begin{proof}
				This is a direct consequence of Theorem~\ref{thm:Joseph:4:4:full} and Theorem~\ref{thm:compact:Lie:group:negative:definite:Killing:form}. Since the Killing form of a semisimple Lie algebra associated with a compact Lie algebra is negative definite and can therefore be diagonalized. Theorem~\ref{thm:Joseph:4:4:full} then prohibits the realization of such a Lie algebra in $\hat{A}_n$.
			\end{proof}
			
			\vspace{0.2cm}
			
			\begin{tcolorbox}[breakable, colback=Cerulean!3!white,colframe=Cerulean!85!black,title=Example: Compact Lie algebras with realizations in $\hat{A}_n$]
                In Corollary~\ref{cor:Joseph:4:4}, we have established that no semisimple (and therefore no simple) Lie algebra of rank $n$ corresponding to a compact Lie group $G$ can be realized within the skew-hermitian Weyl algebra $\hat{A}_n$. However, this does not preclude the realizations of certain semisimple Lie algebras of rank $n$ corresponding to a compact Lie group $G$ within $\hat{A}_m$ with $m\geq n$. The following examples demonstrates this, as it shows that the semisimple Lie algebra $\mathfrak{su}(2)$, which is of rank one and corresponding to the compact Lie group $\operatorname{SU}(2)$, can be realized with any multi-mode skew-hermitian Weyl algebra $\hat{A}_n$ with $n\geq 2$.
				\begin{example}\label{exa:su:realization}
					Consider the two-mode skew-hermitian Weyl algebra $\hat{A}_2$ generated by the creation operators $a^\dagg$, $b^\dagg$ and the annihilation operators $a$, $b$. Define the following three operators:
					\begin{align}
						\sigma_1:=\frac{1}{2}\left(ab^\dagg-a^\dagg b\right),\quad\sigma_2:=\frac{i}{2}\left(ab^\dagg+a^\dagg b\right),\quad\sigma_3:=-\frac{i}{2}\left(a^\dagg a-b^\dagg b\right).
					\end{align}
					It is straightforward to verify: $[\sigma_1,\sigma_2]=\sigma_3$, $[\sigma_2,\sigma_3]=\sigma_1$, and $[\sigma_3,\sigma_1]=\sigma_2$. These operators are a basis for the real special unitary Lie algebra $\su{2}$ as a vector space, which corresponds to the compact Lie group $\operatorname{SU}(2)$. This Lie algebra is semisimple and of rank one \plainrefs{Pfeifer:2003}.
				\end{example}
			\end{tcolorbox}

			\section{Factorization of the dynamics induced by Hamiltonians with known algebras\label{app:factorization}}
			
			The Wei-Norman method applied to bosonic systems has been discussed in some detail in the literature \plainrefs{Qvarfort:2025}. To apply this method, we begin by considering the time evolution operator associated with a time-dependent Hamiltonian $H(t)$. This operator is given by the time-ordered exponential:
			\begin{align}\label{eqn:appendix:time:ordered:exponential:appendix}
				U(t)=\overset{\leftarrow}{\mathcal{T}}\exp\left[-i\int_0^t H(\tau)\mathrm{d}\tau\right].
			\end{align}
			To proceed, we adopt the factorization-ansatz:
			\begin{align}\label{ansatz:Wei:Norman:appendix}
				U(t)=\prod_{j=1}^n U_j(t),\quad\text{with}\quad U_j(t)=e^{-f_j(t) g_j},
			\end{align}
			where $g_j$ are linearly independent elements of a finite-dimensional Lie algebra $\g$ (in our case, $\g\subseteq\hat{A}_1$). The functions $f_j(t)$ are real-valued and encode the dynamics of the system.
			Taking the time derivative of \eqref{eqn:appendix:time:ordered:exponential:appendix} and \eqref{ansatz:Wei:Norman:appendix} yields
			\begin{align*}
				\frac{\mathrm{d}}{\mathrm{d} t}U(t)&=-i H(t)U(t)=\Dot{U}_1(t) U_2(t)\ldots U_n(t)+U_1(t)\Dot{U}_2(t) U_3(t)\ldots U_n(t)+U_1(t)\ldots U_{n-1}(t) \Dot{U}_n(t)\\
				&=-\left(\Dot{f}_1g_1+\Dot{f}_2(t)U_1(t) g_2U_1^\dagg(t)+\ldots+\Dot{f}_n(t) U_1(t)\ldots U_{n-1}(t)g_nU_{n-1}^\dagg(t)\ldots U_1^\dagg(t)\right)U(t).
			\end{align*}
			Here, we used the product rule, the identity $\Dot{U}_j(t)=-\Dot{f}_j(t)g_j U_j(t)$, and the fact that the operators $U_j(t)$ are unitary. Further manipulations of the equations above allow us to obtain the differential equation
			\begin{align}\label{eqn:app:diff:equations:general}
				iH(t)=\sum_{j=1}^nu_j(t)g_j=\Dot{f}_1g_1+\Dot{f}_2(t)U_1(t) g_2U_1^\dagg(t)+\ldots+\Dot{f}_n(t) U_1(t)\ldots U_{n-1}(t)g_nU_{n-1}^\dagg(t)\ldots U_1^\dagg(t).
			\end{align}
			The solution, as already discussed, is obtained by computing the similitude relations $U_j(t) g_kU_j^\dagg(t)$ and equating the coefficients of identical monomials $g_j$. Note that this process is successful only if the Hamiltonian Lie algebra $\g=\lie{\{g_j\}_{j=1}^n}=\spn\{g_j\}_{j=1}^n$ is finite dimensional.
			
			Below, we specialize the problem to two finite-dimensional algebras that are key to this work: namely, the Schr\"odinger algebra $\mathcal{S}$ and the Wigner-Heisenberg algebra $\wh_2$.

			\subsection{The Schr\"odinger algebra}
			
			We start by considering the Schr\"odinger algebra $\mathcal{S}$, a $6$-dimensional non-abelian Lie algebra that is spanned by the basis $\{i,ia^\dagg a,g_-^{2\iota_1},g_+^{2\iota_1},g_-^{\iota_1},g_+^{\iota_1}\}$, as shown in Proposition~\ref{prop:schroedinger:algebra}. In this case, there are many similitude relations that we need to compute since they determine the differential equations governing the dynamics of the system of interest. We list the similitude relations with all relevant elements below.

			We will use the Hadamard lemma\footnote{Note that naming this formula the Hadamard lemma can be confusing, as the term 'Hadamard lemma' often refers to a different result from the literature  \plainrefs{Nestruev:2020}.}  to compute the similitude relations $\exp(g_jx)g_k \exp(-g_j x)$ for some $g_j,g_k\in\{g_1:=ia^\dagg a,g_2:=g_-^{2\iota_1},g_3:=g_+^{2\iota_1},g_4:=g_-^{\iota_1},g_5:=g_+^{\iota_1},g_6:=i\}$ \plainrefs{Adesso:Ragy:2014}. This reads:
			\begin{align}\label{eqn:hadamard:lemma}
				e^{g_j x}g_ke^{-g_j x}=\sum_{n=0}^\infty\frac{1}{n!}\operatorname{ad}_{g_j x}^n(g_k)=\sum_{n=0}^\infty\frac{x^n}{n!}\operatorname{ad}_{g_j}^n(g_k),
			\end{align}
			where the adjoint maps $\operatorname{ad}_{g_j }^n$ are defined as $n$ nested commutators. To be precise, the maps $\operatorname{ad}_{g_j }^n:\g\to \g$ are recursively defined via the initial maps $\operatorname{ad}_{g_j}^0:\g\to\g,\, x\mapsto x$, and $\operatorname{ad}_{g_j}^1\equiv \operatorname{ad}_{g_j}:\g\to\g,\;x\mapsto[g_j,x]$, as well as the formation rule $\operatorname{ad}_{g_j}^{n+1}:=\operatorname{ad}_{g_j}^n\circ\operatorname{ad}_{g_j}:\g\to \g,\,x\mapsto \operatorname{ad}_{g_j}^{n+1}(x)=\operatorname{ad}_{g_j}^n(\operatorname{ad}_{g_j}(x))$. 
			
			\begin{itemize}
				\item \textbf{Similitude relations with} $\boldsymbol{\exp(ia^\dagg a x)}$\textbf{.}
				We begin by computing the similitude relation for the conjugation of a general monomial $g_\sigma^\gamma\in\hat{A}_1$ by the exponential of the operator $ia^\dagg a$.
				To do this, we prove now inductively that:
				\begin{align*}
					\operatorname{ad}_{ia^\dagg a }^n(g_\sigma^\gamma)&=(-1)^{\lfloor n/2\rfloor}(\sigma \chimap{\gamma})^ng_{(-1)^n\sigma}^\gamma.
				\end{align*}
				The base cases $n=0$ and $n=1$ can easily be verified by consolidating Table~\ref{tab:Full:Commutator:Algebra:(pesudo):schroedinger:algebra}. Thus, we can proceed with the induction step:
				\begin{align*}
					\operatorname{ad}_{ia^\dagg a }^{n+1}(g_\sigma^\gamma)=(-1)^{\lfloor n/2\rfloor}(\sigma \chimap{\gamma})[ia^\dagg a ,g_{(-1)^n\sigma}^\gamma]=(-1)^{\lfloor n/2\rfloor+n}(\sigma \chimap{\gamma})^{n+1}g_{(-1)^{n+1}\sigma}^\gamma.
				\end{align*}
				Observing that $(-1)^{\lfloor n/2\rfloor +n}=(-1)^{\lfloor(n+1)/2\rfloor}$ concludes this proof by induction. The resulting similitude relations are consequently:
				\begin{align}\label{eqn:app:similitude:number:operator:g:sigma:gamma}
					e^{ia^\dagg a x}g_\sigma^\gamma e^{-ia^\dagg a x}&=\sum_{n=0}^\infty\frac{1}{n!}(-1)^{\lfloor n/2\rfloor}(\sigma x\chimap{\gamma})^n g_{(-1)^n\sigma}^\gamma=\cos(x\chimap{\gamma})g_\sigma^\gamma+\sigma
					\sin(x\chimap{\gamma})g_{-\sigma}^\gamma.
				\end{align}
				\item \textbf{Similitude relations with} $\boldsymbol{\exp(g_-^{2\iota_1}x)}$\textbf{.}
				
				Here, we won't present a general formula for $\exp(g_-^{2\iota_1})g_\sigma^\gamma \exp(-g_-^{2\iota_1}x)$, as done before; instead, we focus on the $g_\sigma^\gamma$, for which the similitude relations are necessary to factorize the dynamics determined by elements from the Schrödinger algebra. That is $g_\sigma^\gamma\in \{g_+^{2\iota_1},g_-^{\iota_1},g_+^{\iota_1}\}$. Note that we do not need to consider elements $z$ from the center of $\mathcal{S}$, since those satisfy $\exp(g_-^{2\iota_1}) z \exp(-g_-^{2\iota_1}x)=z$.
				
				We start by showing inductively that
				\begin{align}
					\operatorname{ad}_{g_-^{2\iota_1}}^n(g_+^{2\iota_1})&=2^{\lfloor n/2\rfloor}8^{\lfloor(n+1)/2\rfloor}\left(\Theta_n g_+^{2\iota_1}+\Theta_{n+1} i\left(a^\dagg a+\frac{1}{2}\right)\right),
				\end{align}
				where $\Theta_n=1$ if $n$ is even and $\Theta_n=0$ if $n$ is odd. The base case is clear, and we can move on to the induction step:
				\begin{align*}
					\operatorname{ad}_{g_-^{2\iota_1}}^{n+1}(g_+^{2\iota_1})&=2^{\lfloor n/2\rfloor}8^{\lfloor(n+1)/2\rfloor}\left(\Theta_n [g_-^{2\iota_1},g_+^{2\iota_1}]+\Theta_{n+1} [g_-^{2\iota_1},ia^\dagg a]\right)=2^{\lfloor n/2\rfloor}8^{\lfloor(n+1)/2\rfloor}\left(\Theta_n 8i\left(a^\dagg a+\frac{1}{2}\right)+\Theta_{n+1} 2g_+^{2\iota_1}\right).
				\end{align*}
				Here, we used that $i\in\mathcal{Z}(\mathcal{S})$ to obtain the first equality. Suppose now $n=2n'$ is even. Then:
				\begin{align*}
					\operatorname{ad}_{g_-^{2\iota_1}}^{n+1}(g_+^{2\iota_1})&=2^{n'}8^{n'+1}i\left(a^\dagg a+\frac{1}{2}\right)=2^{\lfloor (n+1)/2\rfloor}8^{\lfloor(n+2)/2\rfloor}\left(\Theta_{n+1} g_+^{2\iota_1}+\Theta_{n+2} i\left(a^\dagg a+\frac{1}{2}\right)\right).
				\end{align*}
				One finds similarly, if $n=2n'+1$ is odd:
				\begin{align*}
					\operatorname{ad}_{g_-^{2\iota_1}}^{n+1}(g_+^{2\iota_1})&=2^{n'+1}8^{n'+1}g_+^{2\iota_1}=2^{\lfloor (n+1)/2\rfloor}8^{\lfloor(n+2)/2\rfloor}\left(\Theta_{n+1} g_+^{2\iota_1}+\Theta_{n+2} i\left(a^\dagg a+\frac{1}{2}\right)\right),
				\end{align*}
				which concludes this induction. The resulting similitude relation reads, therefore:
				\begin{align}\label{eqn:app:similitude:minus:2iota:plus:2iota}
					e^{g_-^{2\iota_1}x}g_+^{2\iota_1}e^{-g_-^{2\iota_1}}&=\sum_{n=0}^\infty\frac{x^n}{n!}2^{\lfloor n/2\rfloor}8^{\lfloor(n+1)/2\rfloor}\left(\Theta_n g_+^{2\iota_1}+\Theta_{n+1} i\left(a^\dagg a+\frac{1}{2}\right)\right)=\cosh(4x)g_+^{2\iota_1}+2i\sinh(4x)\left(a^\dagg a+\frac{1}{2}\right).
				\end{align}
				
				Next, we inductively show the identity
				\begin{align}
					\operatorname{ad}_{g_-^{2\iota_1}}^{n}(g_\sigma^{\iota_1})=\sigma^n 2^n g_\sigma^{\iota_1}.
				\end{align}
				The base cases with $n=0$ and $g_\sigma^{\iota_1}\in\{g_+^{\iota_1},g_-^{\iota_1}\}$ are clear and the induction step follows immediately from Table~\ref{tab:Full:Commutator:Algebra:(pesudo):schroedinger:algebra}:
				\begin{align*}
					\operatorname{ad}_{g_-^{2\iota_1}}^{n+1}(g_\sigma^{\iota_1})=\sigma^n 2^n [g_-^{2\iota_1},g_\sigma^{\iota_1}]=\sigma^{n+1} 2^{n+1} g_\sigma^{\iota_1}.
				\end{align*}
				The resulting similitude relations read consequently:
				\begin{align}
					e^{g_-^{2\iota_1}x}g_\sigma^{\iota_1}e^{-g_-^{2\iota_1}x}&=\sum_{n=0}^\infty\frac{x^n}{n!}\sigma^n 2^n g_\sigma^{\iota_1}=e^{2\sigma x}g_\sigma^{\iota_1}.
				\end{align}
				
				\item \textbf{Similitude relations with} $\boldsymbol{\exp(g_+^{2\iota_1}x)}$\textbf{.}
				
				Here, we focus on the similtude relation $\exp(g_+^{2\iota_1} x)g_\sigma^{\iota_1} \exp(-g_+^{2\iota_1})$. We start by showing inductively:
				\begin{align}
					\operatorname{ad}_{g_+^{2\iota_1}}^{n}(g_\sigma^{\iota_1})&=(-2)^ng_{(-1)^n\sigma}^{\iota_1}.
				\end{align}
				The base cases with $n=0$ and $g_\sigma^{\iota_1}\in\{g_+^{\iota_1},g_-^{\iota_1}\}$ are clear. For the induction step, one observes:
				\begin{align*}
					\operatorname{ad}_{g_+^{2\iota_1}}^{n+1}(g_\sigma^{\iota_1})&=(-2)^n[g_+^{2\iota_1},g_{(-1)^n\sigma}^{\iota_1}]=(-2)^{n+1}g_{(-1)^{n+1}}^{\iota_1}.
				\end{align*}
				The resulting similtude relations read consequently:
				\begin{align}
					e^{g_+^{2\iota_1}x}g_\sigma^{\iota_1}e^{-g_+^{2\iota_1}x}\sum_{n=0}^\infty\frac{x^n}{n!}(-2)^{n}g_{(-1)^n \sigma}^{\iota_1}=\cosh(2x)g_\sigma^{\iota_1}-\sinh(2x)g_{-\sigma}^{\iota_1}.
				\end{align}
				
				\item \textbf{Similitude relations with} $\boldsymbol{\exp(g_-^{\iota_1}x)}$\textbf{.}
				
				Here, the relevant similtude relations relevant for the upcoming calculations are $\exp(g_-^{\iota_1}x)g_k\exp(-g_-^{\iota_1}x)$ with $g_k\in \{ia^\dagg a, g_-^{2\iota_1},g_+^{2\iota_1},g_+^{\iota_1}\}$. These can easily be computed using Table~\ref{tab:Full:Commutator:Algebra:(pesudo):schroedinger:algebra} and recognizing that the series in equation \eqref{eqn:hadamard:lemma} terminates after at most three steps. Compute, for instance, the similtude relation $\exp(g_-^{\iota_1}x)g_-^{\iota_1}\exp(-g_-^{\iota_1}x)$ by noticing that $\operatorname{ad}_{g_-^{\iota_1}}(g_+^{\iota_1})=[g_-^{\iota_1},g_+^{\iota_1}]=2i$ and therefore $\operatorname{ad}_{g_-^{\iota_1}}^n(g_+^{\iota_1})=0$ for all $n\geq 2$. Thus, the resulting similitude relation reads:
				\begin{align}\label{eqn:app:similitude:minus:iota:1:plus:iota:1}
					e^{g_-^{\iota_1}x}g_+^{\iota_1}e^{-g_-^{\iota_1}x}&=g_+^{\iota_1}+2ix.
				\end{align}
				We also have:
				\begin{align}\label{eqn:app:similitude:minus:iota:A0}
					e^{g_-^{\iota_1}x}(ia^{\dagg}a) e^{-g_-^{\iota_1}x}=ia^{\dagg}a + xg_+^{\iota_1}+ix^2,
				\end{align} 
				since $\operatorname{ad}_{g_-^{\iota_1}}(ia^\dagg a)=[g_-^{\iota_1},ia^\dagg a]=g_+^{\iota_1}$, and thus $\operatorname{ad}^2_{g_-^{\iota_1}}(g_-^{2\iota_1})=[g_-^{\iota_1},g_+^{\iota_1}]=2i$, while $\operatorname{ad}^n_{g_-^{\iota_1}}(g_-^{2\iota_1})=0$  for all $n\geq 3$. Similarly, for $g_-^{2\iota_1}$ and $g_+^{2\iota_1}$, we obtain:
				\begin{align}
					e^{g_-^{\iota_1}x}g_-^{2\iota_1}e^{-g_-^{\iota_1}x}&=g_-^{2\iota_1}+2xg_-^{\iota_1},\label{eqn:app:similitude:minus:iota:minus:2iota}\\
					e^{g_-^{\iota_1}x}g_+^{2\iota_1}e^{-g_-^{\iota_1}x}&=g_+^{2\iota_1}+2xg_+^{\iota_1}+2ix^2\label{eqn:app:similitude:minus:iota:plus:2iota}.
				\end{align}
				
				\item \textbf{Similitude relations with $\boldsymbol{\exp(g_+^{\iota_1}x)}$.}\\
				Likewise, we find the following expressions:
				\begin{align}
					e^{g_+^{\iota_1}x}(ia^{\dagg}a) e^{-g_+^{\iota_1}x}&=ia^{\dagg}a - xg_-^{\iota_1}+ix^2,\label{eqn:app:similitude:plus:iota:A0}\\
					e^{g_+^{\iota_1}x} g_-^{2\iota_1} e^{-g_+^{\iota_1}x}&=g_-^{2\iota_1} - 2x g_+^{\iota_1}\label{eqn:app:similitude:plus:iota:minus:2iota},\\
					e^{g_+^{\iota_1}x}g_+^{2\iota_1} e^{-g_+^{\iota_1}x}&=g_+^{2\iota_1}+2xg_-^{\iota_1}-2ix^2\label{eqn:app:similitue:plus:iota:plus:2iota}.
				\end{align}
			\end{itemize}
			
			We now proceed with the similitude relations necessary to solve our problem.
			
			\begin{itemize}
				\item \textbf{First relation.}
				\begin{align*}
					U_1(t) g_2 U_1^\dagg(t)=e^{-ia^{\dagg}af_1(t)}g_-^{\iota_1}e^{ia^{\dagg}af_1(t)}=\cos(f_1(t))g_-^{\iota_1}+\sin(f_1(t))g_+^{\iota_1},
				\end{align*}
				as derived from \eqref{eqn:app:similitude:number:operator:g:sigma:gamma}.
				\item \textbf{Second relation.}
				\begin{align*}
					\prod_{j=1}^2 U_j(t) g_3\prod_{j=1}^2 U_{3-j}^\dagg(t)=e^{-ia^{\dagg}af_1(t)}e^{-g_-^{\iota_1}f_2(t)}g_+^{\iota_1}e^{g_-^{\iota_1}f_2(t)}e^{ia^{\dagg}af_1(t)}=\cos(f_1(t))g_+^{\iota_1} - \sin(f_1(t))g_-^{\iota_1}-2if_2(t),
				\end{align*}
				using \eqref{eqn:app:similitude:minus:iota:1:plus:iota:1} and \eqref{eqn:app:similitude:number:operator:g:sigma:gamma}.
				\item \textbf{Third relation.}
				\begin{align*}
					\prod_{j=1}^3 U_j(t) g_4\prod_{j=1}^3 U_{4-j}^\dagg(t)&= e^{-ia^{\dagg}af_1(t)}e^{-g_-^{\iota_1}f_2(t)}e^{-g_+^{\iota_1}f_3(t)}g_-^{2\iota_1}e^{g_+^{\iota_1}f_3(t)}e^{g_-^{\iota_1}f_2(t)}e^{ia^{\dagg}af_1(t)}\\
					&=\cos(2f_1(t))g_-^{2\iota_1}+\sin(2f_1(t))g_+^{2\iota_1}-2\left(f_2(t)\cos(f_1(t))+f_3(t)\sin(f_1(t))\right)g_-^{\iota_1}\\
					&\quad+2\left(f_3(t)\cos(f_1(t))-f_2(t)\sin(f_1(t))\right)g_+^{\iota_1}-4if_2(t)f_3(t)
				\end{align*}
				using \eqref{eqn:app:similitude:plus:iota:minus:2iota}, \eqref{eqn:app:similitude:minus:iota:minus:2iota}, \eqref{eqn:app:similitude:minus:iota:1:plus:iota:1}, and \eqref{eqn:app:similitude:number:operator:g:sigma:gamma}.
				\item \textbf{Fourth relation.}
				\begin{align*}
					\prod_{j=1}^4 U_j(t) g_5\prod_{j=1}^4 U_{5-j}^\dagg(t)&=e^{-ia^{\dagg}af_1(t)}e^{-g_-^{\iota_1}f_2(t)}e^{-g_+^{\iota}f_3(t)}e^{-g_-^{2\iota_1}f_4(t)}g_+^{2\iota_1}e^{g_-^{2\iota_1}f_4(t)}e^{g_+^{\iota}f_3(t)}e^{g_-^{\iota_1}f_2(t)}e^{ia^{\dagg}af_1(t)}\\
					&=\cosh(4f_4(t))\cos(2f_1(t))g_+^{2\iota_1}-\cosh(4f_4(t))\sin(2f_1(t))g_-^{2\iota_1}\\
					&\quad-2\left(f_2(t)e^{-4f_4(t)}\cos(f_1(t))+f_3(t)e^{4f_4(t)}\sin(f_1(t))\right)g_+^{\iota_1}\\
					&\quad+2\left(f_2(t)e^{-4f_4(t)}\sin(f_1(t))-f_3(t)e^{4f_4(t)}\cos(f_1(t))\right)g_-^{\iota_1}\\
					&\quad-2\sinh(4f_4(t))(ia^{\dagg}a) + 2i(f_2(t))^2e^{-4f_4(t)}-2i(f_3(t))^2e^{4f_4(t)}-i\sinh(4f_4(t)),
				\end{align*}
				using \eqref{eqn:app:similitude:minus:2iota:plus:2iota}, \eqref{eqn:app:similitue:plus:iota:plus:2iota}, \eqref{eqn:app:similitude:plus:iota:A0}, \eqref{eqn:app:similitude:minus:iota:plus:2iota}, \eqref{eqn:app:similitude:minus:iota:A0}, and \eqref{eqn:app:similitude:number:operator:g:sigma:gamma}.
				\item \textbf{Fith relation.}{\small
					\begin{align*}
						\prod_{j=1}^5 U_j(t) g_6\prod_{j=1}^6 U_{6-j}^\dagg(t)&=e^{-ia^{\dagg}af_1(t)}e^{-g_-^{\iota_1}f_2(t)}e^{-g_+^{\iota}f_3(t)}e^{-g_-^{2\iota_1}f_4(t)}e^{-g_+^{2\iota_1}f_5(t)}ie^{g_+^{2\iota_1}f_5(t)}e^{g_-^{2\iota_1}f_4(t)}e^{g_+^{\iota}f_3(t)}e^{g_-^{\iota_1}f_2(t)}e^{ia^{\dagg}af_1(t)}=i.
				\end{align*}}
			\end{itemize}
			\noindent
			For a Hamiltonian
			\begin{align*}
				iH(t)=\sum_{j=1}^6 u_j(t)g_j=ia^\dagg a u_1(t)+(a-a^\dagg)u_2(t)+i(a+a^\dagg)u_3(t)+\left(a^2-(a^\dagg)^2\right)u_4(t)+i\left(a^2+(a^\dagg)^2\right) u_5(t)+iu_6(t),
			\end{align*}
			which admits the factorization (up to an irrelevant phase) for the time evolution operator as:
			\begin{align*}
				U(t)=e^{-ia^{\dagg}af_1(t)}e^{-(a-a^{\dagg})f_2(t)}e^{-i(a+a^{\dagg})f_3(t)}e^{-(a^2-(a^{\dagg})^2)f_4(t)}e^{-i(a^2+(a^{\dagg})^2)f_5(t)}
			\end{align*}
			the functions $f_j(t)$ with $j\in\{1,2,3,4,5\}$ are found by (numerically) solving the set of coupled differential equations, derived using the relations above:
			\begin{subequations}
				\begin{align}
					u_1(t) &= \dot{f}_1(t) - 2\dot{f}_5(t)\sinh(4f_4(t)), \\[6pt]
					u_2(t) &= \dot{f}_2(t)\cos(f_1(t)) - \dot{f}_3(t)\sin(f_1(t))   - 2\dot{f}_4(t)\bigl(f_2(t)\cos(f_1(t)) + f_3(t)\sin(f_1(t))\bigr)  \nonumber \\
					&\quad + 2\dot{f}_5(t)\bigl(f_2(t)e^{-4f_4(t)}\sin(f_1(t)) - f_3(t)e^{4f_4(t)}\cos(f_1(t))\bigr), \\[6pt]
					u_3(t) &= \dot{f}_2(t)\sin(f_1(t)) + \dot{f}_3(t)\cos(f_1(t))   + 2\dot{f}_4(t)\bigl(f_3(t)\cos(f_1(t)) - f_2(t)\sin(f_1(t))\bigr)  \nonumber \\
					&\quad - 2\dot{f}_5(t)\bigl(f_2(t)e^{-4f_4(t)}\cos(f_1(t)) + f_3(t)e^{4f_4(t)}\sin(f_1(t))\bigr), \\[6pt]
					u_4(t) &= \dot{f}_4(t)\cos(2f_1(t)) - \dot{f}_5(t)\cosh(4f_4(t))\sin(2f_1(t)), \\[6pt]
					u_5(t) &= \dot{f}_4(t)\sin(2f_1(t)) + \dot{f}_5(t)\cosh(4f_4(t))\cos(2f_1(t)).
				\end{align}
			\end{subequations}
			This linear system of equations can be inverted:
			\begin{subequations}
				\begin{align}
					\dot{f}_1(t)&=u_1(t)+2\sinh(4f_4(t))\dot{f}_5(t),\\
					\dot{f}_2(t)&=u_2(t)\cos(f_1(t))+u_3(t)\sin(f_1(t))+2\dot{f}_4(t)f_2(t)+2\dot{f}_5(t)f_3(t)e^{4f_4(t)},\\
					\dot{f}_3(t)&=u_3(t)\cos(f_1(t))-u_2(t)\sin(f_1(t))-2\dot{f}_4(t)f_3(t)+2\dot{f}_5(t)f_2(t)e^{-4f_4(t)},\\
					\dot{f}_4(t)&=u_4(t)\cos(2f_1(t))+u_5(t)\sin(2f_1(t)),\\
					\dot{f}_5(t)&=\frac{u_5(t)\cos(2f_1(t))-u_4(t)\sin(2f_1(t))}{\cosh(4f_4(t))}.
				\end{align}
			\end{subequations}
			Substituting the implicit dependencies finally yields equations \eqref{Schroedinger:factorization:final:equations}.

			\subsection{The Wigner-Heisenberg algebra of second kind}
			We move on to the case of Wigner-Heisenberg algebra $\wh_2$.
			We choose the usual basis $\{g_1:=ia^\dagg a,g_2:=g_-^{\iota_1},g_3:=g_+^{\iota_1},g_4:=i\}$, which leaves us with the task of computing the following relations:
			\begin{itemize}
				\item \textbf{First relation.}
				\begin{align*}
					e^{-ia^\dagg a f_1(t)}g_-^{\iota_1}e^{ia^\dagg a f_1(t)}&=\cos(f_1(t))g_-^{\iota_1}+\sin(f_1(t))g_+^{\iota_1},
				\end{align*}
				as derived from \eqref{eqn:app:similitude:number:operator:g:sigma:gamma}.
				\item \textbf{Second relation.}
				\begin{align*}
					e^{-ia^\dagg a f_1(t)}e^{-g_-^{\iota_1}f_2(t)}g_+^{\iota_1}e^{g_-^{\iota_1}f_2(t)}e^{ia^\dagg a f_1(t)}&=e^{-ia^\dagg a f_1(t)}\left(g_+^{\iota_1}-2if_2(t)\right)e^{ia^\dagg a f_1(t)}\\
					&=\cos(f_1(t))g_+^{\iota_1}-\sin(f_1(t))g_-^{\iota_1}-2if_2(t),
				\end{align*}
				using equations \eqref{eqn:app:similitude:minus:iota:1:plus:iota:1} and \eqref{eqn:app:similitude:number:operator:g:sigma:gamma}.
				\item \textbf{Third relation.}
				\begin{align*}
					e^{-ia^\dagg a f_1(t)}e^{-g_-^{\iota_1}f_2(t)}e^{-g_+^{\iota_1}f_3(t)}ie^{g_+^{\iota_1}f_3(t)}e^{g_-^{\iota_1}f_2(t)}e^{ia^\dagg a f_1(t)}&=i.
				\end{align*}
			\end{itemize}
			We can therefore complete the Wei-Norman method and write the differential equation of interest as
			\begin{align*}
				u_1(t)&=\Dot{f}_1(t),\\
				u_2(t)&=\Dot{f}_2(t)\cos(f_1(t))-\Dot{f}_3(t)\sin(f_1(t)),\\
				u_3(t)&=\Dot{f}_2(t)\sin(f_1(t))+\Dot{f}_3(t)\cos(f_1(t)),\\
				u_4(t)&=\Dot{f}_4(t)-2\Dot{f}_3(t)f_2(t),
			\end{align*}
			with the initial condition $f_j(0)=0$ for all $j\in\{1,2,3,4\}$. The formal solution to this system of coupled differential equations can easily be obtained, and we present it below:
			\begin{subequations}
				\begin{align}
					f_1(t)&=\int_0^t u_1(\tau)\mathrm{d}\tau,\\
					f_2(t)&=\int_0^t\left(\cos(f_1(\tau)) u_2(\tau)+\sin(f_1(\tau)) u_3(\tau)\right) \mathrm{d}\tau,\\
					f_3(t)&=\int_0^t\left(\cos(f_1(\tau)) u_3(\tau)-\sin(f_1(\tau)) u_2(\tau)\right) \mathrm{d}\tau,\\
					f_4(t)&=\int_0^t\left(u_4(\tau)+2\Dot{f}_3(\tau)f_2(\tau)\right)\mathrm{d}\tau.
				\end{align}
			\end{subequations}
			We can now also choose a typical Hamiltonian $H(t)=\omega a^\dagg a + i\,u_2(t)(a-a^\dagg)+u_3(t)(a+a^\dagg)$ and have 
			\begin{subequations}\label{Wigner:Heisenberg:decouplig:solution:appendix}
				\begin{align}
					f_1(t)&=\omega\,t,\\
					f_2(t)&=\int_0^t\left(\cos(\omega\,\tau) u_2(\tau)+\sin(\omega\,\tau) u_3(\tau)\right) \mathrm{d}\tau,\\
					f_3(t)&=\int_0^t\left(\cos(\omega\,\tau) u_3(\tau)-\sin(\omega\,\tau) u_2(\tau)\right) \mathrm{d}\tau,\\
					f_4(t)&=2\int_0^t\Dot{f}_3(\tau)f_2(\tau)\mathrm{d}\tau,
				\end{align}
			\end{subequations}
			in terms of the control functions $u_k(t)$.
			
			
			\clearpage
			
			\bibliographystyle{apsrev4-2}
			\bibliography{main.bib}

\begin{thebibliography}{145}%
\makeatletter
\providecommand \@ifxundefined [1]{%
 \@ifx{#1\undefined}
}%
\providecommand \@ifnum [1]{%
 \ifnum #1\expandafter \@firstoftwo
 \else \expandafter \@secondoftwo
 \fi
}%
\providecommand \@ifx [1]{%
 \ifx #1\expandafter \@firstoftwo
 \else \expandafter \@secondoftwo
 \fi
}%
\providecommand \natexlab [1]{#1}%
\providecommand \enquote  [1]{``#1''}%
\providecommand \bibnamefont  [1]{#1}%
\providecommand \bibfnamefont [1]{#1}%
\providecommand \citenamefont [1]{#1}%
\providecommand \href@noop [0]{\@secondoftwo}%
\providecommand \href [0]{\begingroup \@sanitize@url \@href}%
\providecommand \@href[1]{\@@startlink{#1}\@@href}%
\providecommand \@@href[1]{\endgroup#1\@@endlink}%
\providecommand \@sanitize@url [0]{\catcode `\\12\catcode `\$12\catcode
  `\&12\catcode `\#12\catcode `\^12\catcode `\_12\catcode `\%12\relax}%
\providecommand \@@startlink[1]{}%
\providecommand \@@endlink[0]{}%
\providecommand \url  [0]{\begingroup\@sanitize@url \@url }%
\providecommand \@url [1]{\endgroup\@href {#1}{\urlprefix }}%
\providecommand \urlprefix  [0]{URL }%
\providecommand \Eprint [0]{\href }%
\providecommand \doibase [0]{https://doi.org/}%
\providecommand \selectlanguage [0]{\@gobble}%
\providecommand \bibinfo  [0]{\@secondoftwo}%
\providecommand \bibfield  [0]{\@secondoftwo}%
\providecommand \translation [1]{[#1]}%
\providecommand \BibitemOpen [0]{}%
\providecommand \bibitemStop [0]{}%
\providecommand \bibitemNoStop [0]{.\EOS\space}%
\providecommand \EOS [0]{\spacefactor3000\relax}%
\providecommand \BibitemShut  [1]{\csname bibitem#1\endcsname}%
\let\auto@bib@innerbib\@empty
\bibitem [{\citenamefont {Barnes}\ and\ \citenamefont
  {Das~Sarma}(2012)}]{Barnes:2012}%
  \BibitemOpen
  \bibfield  {author} {\bibinfo {author} {\bibfnamefont {E.}~\bibnamefont
  {Barnes}}\ and\ \bibinfo {author} {\bibfnamefont {S.}~\bibnamefont
  {Das~Sarma}},\ }\href {https://doi.org/10.1103/PhysRevLett.109.060401}
  {\bibfield  {journal} {\bibinfo  {journal} {Physical Review Letters}\
  }\textbf {\bibinfo {volume} {109}},\ \bibinfo {pages} {060401} (\bibinfo
  {year} {2012})}\BibitemShut {NoStop}%
\bibitem [{\citenamefont {LeBlanc}\ \emph {et~al.}(2015)\citenamefont
  {LeBlanc}, \citenamefont {Antipov}, \citenamefont {Becca}, \citenamefont
  {Bulik}, \citenamefont {Chan}, \citenamefont {Chung}, \citenamefont {Deng},
  \citenamefont {Ferrero}, \citenamefont {Henderson}, \citenamefont
  {Jim\'enez-Hoyos}, \citenamefont {Kozik}, \citenamefont {Liu}, \citenamefont
  {Millis}, \citenamefont {Prokof'ev}, \citenamefont {Qin}, \citenamefont
  {Scuseria}, \citenamefont {Shi}, \citenamefont {Svistunov}, \citenamefont
  {Tocchio}, \citenamefont {Tupitsyn}, \citenamefont {White}, \citenamefont
  {Zhang}, \citenamefont {Zheng}, \citenamefont {Zhu},\ and\ \citenamefont
  {Gull}}]{LeBlanc:2015}%
  \BibitemOpen
  \bibfield  {author} {\bibinfo {author} {\bibfnamefont {J.~P.~F.}\
  \bibnamefont {LeBlanc}}, \bibinfo {author} {\bibfnamefont {A.~E.}\
  \bibnamefont {Antipov}}, \bibinfo {author} {\bibfnamefont {F.}~\bibnamefont
  {Becca}}, \bibinfo {author} {\bibfnamefont {I.~W.}\ \bibnamefont {Bulik}},
  \bibinfo {author} {\bibfnamefont {G.~K.-L.}\ \bibnamefont {Chan}}, \bibinfo
  {author} {\bibfnamefont {C.-M.}\ \bibnamefont {Chung}}, \bibinfo {author}
  {\bibfnamefont {Y.}~\bibnamefont {Deng}}, \bibinfo {author} {\bibfnamefont
  {M.}~\bibnamefont {Ferrero}}, \bibinfo {author} {\bibfnamefont {T.~M.}\
  \bibnamefont {Henderson}}, \bibinfo {author} {\bibfnamefont {C.~A.}\
  \bibnamefont {Jim\'enez-Hoyos}}, \bibinfo {author} {\bibfnamefont
  {E.}~\bibnamefont {Kozik}}, \bibinfo {author} {\bibfnamefont {X.-W.}\
  \bibnamefont {Liu}}, \bibinfo {author} {\bibfnamefont {A.~J.}\ \bibnamefont
  {Millis}}, \bibinfo {author} {\bibfnamefont {N.~V.}\ \bibnamefont
  {Prokof'ev}}, \bibinfo {author} {\bibfnamefont {M.}~\bibnamefont {Qin}},
  \bibinfo {author} {\bibfnamefont {G.~E.}\ \bibnamefont {Scuseria}}, \bibinfo
  {author} {\bibfnamefont {H.}~\bibnamefont {Shi}}, \bibinfo {author}
  {\bibfnamefont {B.~V.}\ \bibnamefont {Svistunov}}, \bibinfo {author}
  {\bibfnamefont {L.~F.}\ \bibnamefont {Tocchio}}, \bibinfo {author}
  {\bibfnamefont {I.~S.}\ \bibnamefont {Tupitsyn}}, \bibinfo {author}
  {\bibfnamefont {S.~R.}\ \bibnamefont {White}}, \bibinfo {author}
  {\bibfnamefont {S.}~\bibnamefont {Zhang}}, \bibinfo {author} {\bibfnamefont
  {B.-X.}\ \bibnamefont {Zheng}}, \bibinfo {author} {\bibfnamefont
  {Z.}~\bibnamefont {Zhu}},\ and\ \bibinfo {author} {\bibfnamefont
  {E.}~\bibnamefont {Gull}} (\bibinfo {collaboration} {Simons Collaboration on
  the Many-Electron Problem}),\ }\href
  {https://doi.org/10.1103/PhysRevX.5.041041} {\bibfield  {journal} {\bibinfo
  {journal} {Physical Review X}\ }\textbf {\bibinfo {volume} {5}},\ \bibinfo
  {pages} {041041} (\bibinfo {year} {2015})}\BibitemShut {NoStop}%
\bibitem [{\citenamefont {Musielak}\ and\ \citenamefont
  {Quarles}(2014)}]{Musielak:2014}%
  \BibitemOpen
  \bibfield  {author} {\bibinfo {author} {\bibfnamefont {Z.~E.}\ \bibnamefont
  {Musielak}}\ and\ \bibinfo {author} {\bibfnamefont {B.}~\bibnamefont
  {Quarles}},\ }\href {https://doi.org/10.1088/0034-4885/77/6/065901}
  {\bibfield  {journal} {\bibinfo  {journal} {Reports on Progress in Physics}\
  }\textbf {\bibinfo {volume} {77}},\ \bibinfo {pages} {065901} (\bibinfo
  {year} {2014})}\BibitemShut {NoStop}%
\bibitem [{\citenamefont {Pretorius}(2009)}]{Pretorius:2009}%
  \BibitemOpen
  \bibfield  {author} {\bibinfo {author} {\bibfnamefont {F.}~\bibnamefont
  {Pretorius}},\ }\bibinfo {title} {Binary black hole coalescence},\ in\ \href
  {https://doi.org/10.1007/978-1-4020-9264-0_9} {\emph {\bibinfo {booktitle}
  {Physics of Relativistic Objects in Compact Binaries: From Birth to
  Coalescence}}},\ \bibinfo {editor} {edited by\ \bibinfo {editor}
  {\bibfnamefont {M.}~\bibnamefont {Colpi}}, \bibinfo {editor} {\bibfnamefont
  {P.}~\bibnamefont {Casella}}, \bibinfo {editor} {\bibfnamefont
  {V.}~\bibnamefont {Gorini}}, \bibinfo {editor} {\bibfnamefont
  {U.}~\bibnamefont {Moschella}},\ and\ \bibinfo {editor} {\bibfnamefont
  {A.}~\bibnamefont {Possenti}}}\ (\bibinfo  {publisher} {Springer
  Netherlands},\ \bibinfo {address} {Dordrecht},\ \bibinfo {year} {2009})\ pp.\
  \bibinfo {pages} {305--369}\BibitemShut {NoStop}%
\bibitem [{\citenamefont {Yokoyama}\ and\ \citenamefont
  {Shiba}(1987)}]{Yokoyama:1987}%
  \BibitemOpen
  \bibfield  {author} {\bibinfo {author} {\bibfnamefont {H.}~\bibnamefont
  {Yokoyama}}\ and\ \bibinfo {author} {\bibfnamefont {H.}~\bibnamefont
  {Shiba}},\ }\href {https://doi.org/10.1143/JPSJ.56.1490} {\bibfield
  {journal} {\bibinfo  {journal} {Journal of the Physical Society of Japan}\
  }\textbf {\bibinfo {volume} {56}},\ \bibinfo {pages} {1490} (\bibinfo {year}
  {1987})}\BibitemShut {NoStop}%
\bibitem [{\citenamefont {Shi}\ and\ \citenamefont {Zhang}(2013)}]{Shi:2013}%
  \BibitemOpen
  \bibfield  {author} {\bibinfo {author} {\bibfnamefont {H.}~\bibnamefont
  {Shi}}\ and\ \bibinfo {author} {\bibfnamefont {S.}~\bibnamefont {Zhang}},\
  }\href {https://doi.org/10.1103/PhysRevB.88.125132} {\bibfield  {journal}
  {\bibinfo  {journal} {Physical Review B}\ }\textbf {\bibinfo {volume} {88}},\
  \bibinfo {pages} {125132} (\bibinfo {year} {2013})}\BibitemShut {NoStop}%
\bibitem [{\citenamefont {ten Haaf}\ \emph {et~al.}(1995)\citenamefont {ten
  Haaf}, \citenamefont {van Bemmel}, \citenamefont {van Leeuwen}, \citenamefont
  {van Saarloos},\ and\ \citenamefont {Ceperley}}]{tenHaaf:1995}%
  \BibitemOpen
  \bibfield  {author} {\bibinfo {author} {\bibfnamefont {D.~F.~B.}\
  \bibnamefont {ten Haaf}}, \bibinfo {author} {\bibfnamefont {H.~J.~M.}\
  \bibnamefont {van Bemmel}}, \bibinfo {author} {\bibfnamefont {J.~M.~J.}\
  \bibnamefont {van Leeuwen}}, \bibinfo {author} {\bibfnamefont
  {W.}~\bibnamefont {van Saarloos}},\ and\ \bibinfo {author} {\bibfnamefont
  {D.~M.}\ \bibnamefont {Ceperley}},\ }\href
  {https://doi.org/10.1103/PhysRevB.51.13039} {\bibfield  {journal} {\bibinfo
  {journal} {Physical Review B}\ }\textbf {\bibinfo {volume} {51}},\ \bibinfo
  {pages} {13039} (\bibinfo {year} {1995})}\BibitemShut {NoStop}%
\bibitem [{\citenamefont {Aoki}\ \emph {et~al.}(2014)\citenamefont {Aoki},
  \citenamefont {Tsuji}, \citenamefont {Eckstein}, \citenamefont {Kollar},
  \citenamefont {Oka},\ and\ \citenamefont {Werner}}]{Aoki:2014}%
  \BibitemOpen
  \bibfield  {author} {\bibinfo {author} {\bibfnamefont {H.}~\bibnamefont
  {Aoki}}, \bibinfo {author} {\bibfnamefont {N.}~\bibnamefont {Tsuji}},
  \bibinfo {author} {\bibfnamefont {M.}~\bibnamefont {Eckstein}}, \bibinfo
  {author} {\bibfnamefont {M.}~\bibnamefont {Kollar}}, \bibinfo {author}
  {\bibfnamefont {T.}~\bibnamefont {Oka}},\ and\ \bibinfo {author}
  {\bibfnamefont {P.}~\bibnamefont {Werner}},\ }\href
  {https://doi.org/10.1103/RevModPhys.86.779} {\bibfield  {journal} {\bibinfo
  {journal} {Reviews of Modern Physics}\ }\textbf {\bibinfo {volume} {86}},\
  \bibinfo {pages} {779} (\bibinfo {year} {2014})}\BibitemShut {NoStop}%
\bibitem [{\citenamefont {Anders}\ \emph {et~al.}(2011)\citenamefont {Anders},
  \citenamefont {Gull}, \citenamefont {Pollet}, \citenamefont {Troyer},\ and\
  \citenamefont {Werner}}]{Anders:2011}%
  \BibitemOpen
  \bibfield  {author} {\bibinfo {author} {\bibfnamefont {P.}~\bibnamefont
  {Anders}}, \bibinfo {author} {\bibfnamefont {E.}~\bibnamefont {Gull}},
  \bibinfo {author} {\bibfnamefont {L.}~\bibnamefont {Pollet}}, \bibinfo
  {author} {\bibfnamefont {M.}~\bibnamefont {Troyer}},\ and\ \bibinfo {author}
  {\bibfnamefont {P.}~\bibnamefont {Werner}},\ }\href
  {https://doi.org/10.1088/1367-2630/13/7/075013} {\bibfield  {journal}
  {\bibinfo  {journal} {New Journal of Physics}\ }\textbf {\bibinfo {volume}
  {13}},\ \bibinfo {pages} {075013} (\bibinfo {year} {2011})}\BibitemShut
  {NoStop}%
\bibitem [{\citenamefont {Knizia}\ and\ \citenamefont
  {Chan}(2012)}]{Knizia:2012}%
  \BibitemOpen
  \bibfield  {author} {\bibinfo {author} {\bibfnamefont {G.}~\bibnamefont
  {Knizia}}\ and\ \bibinfo {author} {\bibfnamefont {G.~K.-L.}\ \bibnamefont
  {Chan}},\ }\href {https://doi.org/10.1103/PhysRevLett.109.186404} {\bibfield
  {journal} {\bibinfo  {journal} {Physical Review Letters}\ }\textbf {\bibinfo
  {volume} {109}},\ \bibinfo {pages} {186404} (\bibinfo {year}
  {2012})}\BibitemShut {NoStop}%
\bibitem [{\citenamefont {Blanes}\ and\ \citenamefont
  {Moan}(2002)}]{Blanes:2002}%
  \BibitemOpen
  \bibfield  {author} {\bibinfo {author} {\bibfnamefont {S.}~\bibnamefont
  {Blanes}}\ and\ \bibinfo {author} {\bibfnamefont {P.}~\bibnamefont {Moan}},\
  }\href {https://doi.org/https://doi.org/10.1016/S0377-0427(01)00492-7}
  {\bibfield  {journal} {\bibinfo  {journal} {Journal of Computational and
  Applied Mathematics}\ }\textbf {\bibinfo {volume} {142}},\ \bibinfo {pages}
  {313} (\bibinfo {year} {2002})}\BibitemShut {NoStop}%
\bibitem [{\citenamefont {Blakie}\ \emph {et~al.}(2008)\citenamefont {Blakie},
  \citenamefont {Bradley}, \citenamefont {Davis}, \citenamefont {Ballagh},\
  and\ \citenamefont {Gardiner}}]{Blakie:2008}%
  \BibitemOpen
  \bibfield  {author} {\bibinfo {author} {\bibfnamefont {P.}~\bibnamefont
  {Blakie}}, \bibinfo {author} {\bibfnamefont {A.}~\bibnamefont {Bradley}},
  \bibinfo {author} {\bibfnamefont {M.}~\bibnamefont {Davis}}, \bibinfo
  {author} {\bibfnamefont {R.}~\bibnamefont {Ballagh}},\ and\ \bibinfo {author}
  {\bibfnamefont {C.}~\bibnamefont {Gardiner}},\ }\href
  {https://doi.org/10.1080/00018730802564254} {\bibfield  {journal} {\bibinfo
  {journal} {Advances in Physics}\ }\textbf {\bibinfo {volume} {57}},\ \bibinfo
  {pages} {363} (\bibinfo {year} {2008})}\BibitemShut {NoStop}%
\bibitem [{\citenamefont {Metzner}\ and\ \citenamefont
  {Vollhardt}(1989)}]{Metzner:1989}%
  \BibitemOpen
  \bibfield  {author} {\bibinfo {author} {\bibfnamefont {W.}~\bibnamefont
  {Metzner}}\ and\ \bibinfo {author} {\bibfnamefont {D.}~\bibnamefont
  {Vollhardt}},\ }\href {https://doi.org/10.1103/PhysRevLett.62.324} {\bibfield
   {journal} {\bibinfo  {journal} {Physical Review Letters}\ }\textbf {\bibinfo
  {volume} {62}},\ \bibinfo {pages} {324} (\bibinfo {year} {1989})}\BibitemShut
  {NoStop}%
\bibitem [{\citenamefont {Georges}\ \emph {et~al.}(1996)\citenamefont
  {Georges}, \citenamefont {Kotliar}, \citenamefont {Krauth},\ and\
  \citenamefont {Rozenberg}}]{Georges:1996}%
  \BibitemOpen
  \bibfield  {author} {\bibinfo {author} {\bibfnamefont {A.}~\bibnamefont
  {Georges}}, \bibinfo {author} {\bibfnamefont {G.}~\bibnamefont {Kotliar}},
  \bibinfo {author} {\bibfnamefont {W.}~\bibnamefont {Krauth}},\ and\ \bibinfo
  {author} {\bibfnamefont {M.~J.}\ \bibnamefont {Rozenberg}},\ }\href
  {https://doi.org/10.1103/RevModPhys.68.13} {\bibfield  {journal} {\bibinfo
  {journal} {Reviews of Modern Physics}\ }\textbf {\bibinfo {volume} {68}},\
  \bibinfo {pages} {13} (\bibinfo {year} {1996})}\BibitemShut {NoStop}%
\bibitem [{\citenamefont {Dyson}(1949)}]{Dyson:1949}%
  \BibitemOpen
  \bibfield  {author} {\bibinfo {author} {\bibfnamefont {F.~J.}\ \bibnamefont
  {Dyson}},\ }\href {https://doi.org/10.1103/PhysRev.75.486} {\bibfield
  {journal} {\bibinfo  {journal} {Physical Review}\ }\textbf {\bibinfo {volume}
  {75}},\ \bibinfo {pages} {486} (\bibinfo {year} {1949})}\BibitemShut
  {NoStop}%
\bibitem [{\citenamefont {Blanes}\ \emph {et~al.}(2009)\citenamefont {Blanes},
  \citenamefont {Casas}, \citenamefont {Oteo},\ and\ \citenamefont
  {Ros}}]{Blanes:2009}%
  \BibitemOpen
  \bibfield  {author} {\bibinfo {author} {\bibfnamefont {S.}~\bibnamefont
  {Blanes}}, \bibinfo {author} {\bibfnamefont {F.}~\bibnamefont {Casas}},
  \bibinfo {author} {\bibfnamefont {J.}~\bibnamefont {Oteo}},\ and\ \bibinfo
  {author} {\bibfnamefont {J.}~\bibnamefont {Ros}},\ }\href
  {https://doi.org/https://doi.org/10.1016/j.physrep.2008.11.001} {\bibfield
  {journal} {\bibinfo  {journal} {Physics Reports}\ }\textbf {\bibinfo {volume}
  {470}},\ \bibinfo {pages} {151} (\bibinfo {year} {2009})}\BibitemShut
  {NoStop}%
\bibitem [{\citenamefont {Kirchhoff}\ \emph {et~al.}(2025)\citenamefont
  {Kirchhoff}, \citenamefont {Wilhelm},\ and\ \citenamefont
  {Motzoi}}]{Kirchhoff:2025}%
  \BibitemOpen
  \bibfield  {author} {\bibinfo {author} {\bibfnamefont {S.}~\bibnamefont
  {Kirchhoff}}, \bibinfo {author} {\bibfnamefont {F.~K.}\ \bibnamefont
  {Wilhelm}},\ and\ \bibinfo {author} {\bibfnamefont {F.}~\bibnamefont
  {Motzoi}},\ }\href {https://doi.org/10.1103/PRXQuantum.6.010328} {\bibfield
  {journal} {\bibinfo  {journal} {PRX Quantum}\ }\textbf {\bibinfo {volume}
  {6}},\ \bibinfo {pages} {010328} (\bibinfo {year} {2025})}\BibitemShut
  {NoStop}%
\bibitem [{\citenamefont {Wei}\ and\ \citenamefont
  {Norman}(1963)}]{Wei:Norman:1963}%
  \BibitemOpen
  \bibfield  {author} {\bibinfo {author} {\bibfnamefont {J.}~\bibnamefont
  {Wei}}\ and\ \bibinfo {author} {\bibfnamefont {E.}~\bibnamefont {Norman}},\
  }\href {https://doi.org/10.1063/1.1703993} {\bibfield  {journal} {\bibinfo
  {journal} {Journal of Mathematical Physics}\ }\textbf {\bibinfo {volume}
  {4}},\ \bibinfo {pages} {575} (\bibinfo {year} {1963})}\BibitemShut {NoStop}%
\bibitem [{\citenamefont {Wei}\ and\ \citenamefont
  {Norman}(1964)}]{Wei:Norman:1964}%
  \BibitemOpen
  \bibfield  {author} {\bibinfo {author} {\bibfnamefont {J.}~\bibnamefont
  {Wei}}\ and\ \bibinfo {author} {\bibfnamefont {E.}~\bibnamefont {Norman}},\
  }\href {https://doi.org/10.2307/2034065} {\bibfield  {journal} {\bibinfo
  {journal} {Proc. Am. Mat. Soc.}\ }\textbf {\bibinfo {volume} {15}},\ \bibinfo
  {pages} {327} (\bibinfo {year} {1964})}\BibitemShut {NoStop}%
\bibitem [{\citenamefont {Qvarfort}\ \emph {et~al.}(2020)\citenamefont
  {Qvarfort}, \citenamefont {Serafini}, \citenamefont {Xuereb}, \citenamefont
  {Braun}, \citenamefont {Rätzel},\ and\ \citenamefont
  {Bruschi}}]{Qvarfort:Serafini:2020}%
  \BibitemOpen
  \bibfield  {author} {\bibinfo {author} {\bibfnamefont {S.}~\bibnamefont
  {Qvarfort}}, \bibinfo {author} {\bibfnamefont {A.}~\bibnamefont {Serafini}},
  \bibinfo {author} {\bibfnamefont {A.}~\bibnamefont {Xuereb}}, \bibinfo
  {author} {\bibfnamefont {D.}~\bibnamefont {Braun}}, \bibinfo {author}
  {\bibfnamefont {D.}~\bibnamefont {Rätzel}},\ and\ \bibinfo {author}
  {\bibfnamefont {D.~E.}\ \bibnamefont {Bruschi}},\ }\href
  {https://doi.org/10.1088/1751-8121/ab64d5} {\bibfield  {journal} {\bibinfo
  {journal} {Journal of Physics A: Mathematical and Theoretical}\ }\textbf
  {\bibinfo {volume} {53}},\ \bibinfo {pages} {075304} (\bibinfo {year}
  {2020})}\BibitemShut {NoStop}%
\bibitem [{\citenamefont {Schneiter}\ \emph {et~al.}(2020)\citenamefont
  {Schneiter}, \citenamefont {Qvarfort}, \citenamefont {Serafini},
  \citenamefont {Xuereb}, \citenamefont {Braun}, \citenamefont {R\"atzel},\
  and\ \citenamefont {Bruschi}}]{Schneiter:Qvarfort:2020}%
  \BibitemOpen
  \bibfield  {author} {\bibinfo {author} {\bibfnamefont {F.}~\bibnamefont
  {Schneiter}}, \bibinfo {author} {\bibfnamefont {S.}~\bibnamefont {Qvarfort}},
  \bibinfo {author} {\bibfnamefont {A.}~\bibnamefont {Serafini}}, \bibinfo
  {author} {\bibfnamefont {A.}~\bibnamefont {Xuereb}}, \bibinfo {author}
  {\bibfnamefont {D.}~\bibnamefont {Braun}}, \bibinfo {author} {\bibfnamefont
  {D.}~\bibnamefont {R\"atzel}},\ and\ \bibinfo {author} {\bibfnamefont
  {D.~E.}\ \bibnamefont {Bruschi}},\ }\href
  {https://doi.org/10.1103/PhysRevA.101.033834} {\bibfield  {journal} {\bibinfo
   {journal} {Physical Review A}\ }\textbf {\bibinfo {volume} {101}},\ \bibinfo
  {pages} {033834} (\bibinfo {year} {2020})}\BibitemShut {NoStop}%
\bibitem [{\citenamefont {Martínez-Tibaduiza}\ \emph
  {et~al.}(2025)\citenamefont {Martínez-Tibaduiza}, \citenamefont
  {Vargas-Calderón}, \citenamefont {Dueñas}, \citenamefont
  {Flórez-Jiménez},\ and\ \citenamefont {Khoury}}]{martíneztibaduiza:2025}%
  \BibitemOpen
  \bibfield  {author} {\bibinfo {author} {\bibfnamefont {D.}~\bibnamefont
  {Martínez-Tibaduiza}}, \bibinfo {author} {\bibfnamefont {V.}~\bibnamefont
  {Vargas-Calderón}}, \bibinfo {author} {\bibfnamefont {J.~G.}\ \bibnamefont
  {Dueñas}}, \bibinfo {author} {\bibfnamefont {J.}~\bibnamefont
  {Flórez-Jiménez}},\ and\ \bibinfo {author} {\bibfnamefont {A.~Z.}\
  \bibnamefont {Khoury}},\ }\href@noop {} {\bibinfo {title}
  {{$\texttt{Symdyn}$: an automated algebraic solution for high-order quantum
  systems}}} (\bibinfo {year} {2025}),\ \Eprint
  {https://arxiv.org/abs/2503.22061} {arXiv:2503.22061 [quant-ph]} \BibitemShut
  {NoStop}%
\bibitem [{\citenamefont {Qvarfort}\ and\ \citenamefont
  {Pikovski}(2025)}]{Qvarfort:2025}%
  \BibitemOpen
  \bibfield  {author} {\bibinfo {author} {\bibfnamefont {S.}~\bibnamefont
  {Qvarfort}}\ and\ \bibinfo {author} {\bibfnamefont {I.}~\bibnamefont
  {Pikovski}},\ }\href {https://doi.org/10.1103/PRXQuantum.6.010201} {\bibfield
   {journal} {\bibinfo  {journal} {PRX Quantum}\ }\textbf {\bibinfo {volume}
  {6}},\ \bibinfo {pages} {010201} (\bibinfo {year} {2025})}\BibitemShut
  {NoStop}%
\bibitem [{\citenamefont {Bruschi}\ \emph {et~al.}(2013)\citenamefont
  {Bruschi}, \citenamefont {Lee},\ and\ \citenamefont
  {Fuentes}}]{Bruschi:Lee:2013}%
  \BibitemOpen
  \bibfield  {author} {\bibinfo {author} {\bibfnamefont {D.~E.}\ \bibnamefont
  {Bruschi}}, \bibinfo {author} {\bibfnamefont {A.~R.}\ \bibnamefont {Lee}},\
  and\ \bibinfo {author} {\bibfnamefont {I.}~\bibnamefont {Fuentes}},\ }\href
  {https://doi.org/10.1088/1751-8113/46/16/165303} {\bibfield  {journal}
  {\bibinfo  {journal} {Journal of Physics A: Mathematical and Theoretical}\
  }\textbf {\bibinfo {volume} {46}},\ \bibinfo {pages} {165303} (\bibinfo
  {year} {2013})}\BibitemShut {NoStop}%
\bibitem [{\citenamefont {Bruschi}\ \emph {et~al.}(2024)\citenamefont
  {Bruschi}, \citenamefont {Xuereb},\ and\ \citenamefont
  {Zeier}}]{Bruschi:Xuereb:2024}%
  \BibitemOpen
  \bibfield  {author} {\bibinfo {author} {\bibfnamefont {D.~E.}\ \bibnamefont
  {Bruschi}}, \bibinfo {author} {\bibfnamefont {A.}~\bibnamefont {Xuereb}},\
  and\ \bibinfo {author} {\bibfnamefont {R.}~\bibnamefont {Zeier}},\ }\href
  {https://doi.org/10.1088/1751-8121/ad91fc} {\bibfield  {journal} {\bibinfo
  {journal} {Journal of Physics A: Mathematical and Theoretical}\ }\textbf
  {\bibinfo {volume} {58}},\ \bibinfo {pages} {025204} (\bibinfo {year}
  {2024})}\BibitemShut {NoStop}%
\bibitem [{\citenamefont {Tanasa}(2005)}]{Tanasa:2005}%
  \BibitemOpen
  \bibfield  {author} {\bibinfo {author} {\bibfnamefont {A.}~\bibnamefont
  {Tanasa}},\ }\emph {\bibinfo {title} {Sous-algèbres de Lie de l'algèbre de
  Weyl}},\ \href {https://arxiv.org/abs/hep-th/0509174} {\bibinfo {type} {Phd
  thesis}},\ \bibinfo  {school} {L'Universite de Haute Alsace} (\bibinfo {year}
  {2005})\BibitemShut {NoStop}%
\bibitem [{\citenamefont {Rausch~de Traubenberg}\ \emph
  {et~al.}(2006)\citenamefont {Rausch~de Traubenberg}, \citenamefont
  {Slupinski},\ and\ \citenamefont {Tanasa}}]{TST:2006}%
  \BibitemOpen
  \bibfield  {author} {\bibinfo {author} {\bibfnamefont {M.}~\bibnamefont
  {Rausch~de Traubenberg}}, \bibinfo {author} {\bibfnamefont {M.~J.}\
  \bibnamefont {Slupinski}},\ and\ \bibinfo {author} {\bibfnamefont
  {A.}~\bibnamefont {Tanasa}},\ }\href@noop {} {\bibfield  {journal} {\bibinfo
  {journal} {Journal of Lie Theory}\ }\textbf {\bibinfo {volume} {16}},\
  \bibinfo {pages} {427} (\bibinfo {year} {2006})}\BibitemShut {NoStop}%
\bibitem [{\citenamefont {Boscain}\ \emph {et~al.}(2021)\citenamefont
  {Boscain}, \citenamefont {Sigalotti},\ and\ \citenamefont
  {Sugny}}]{Boscain:2021}%
  \BibitemOpen
  \bibfield  {author} {\bibinfo {author} {\bibfnamefont {U.}~\bibnamefont
  {Boscain}}, \bibinfo {author} {\bibfnamefont {M.}~\bibnamefont {Sigalotti}},\
  and\ \bibinfo {author} {\bibfnamefont {D.}~\bibnamefont {Sugny}},\ }\href
  {https://doi.org/10.1103/PRXQuantum.2.030203} {\bibfield  {journal} {\bibinfo
   {journal} {PRX Quantum}\ }\textbf {\bibinfo {volume} {2}},\ \bibinfo {pages}
  {030203} (\bibinfo {year} {2021})}\BibitemShut {NoStop}%
\bibitem [{\citenamefont {Huang}\ \emph
  {et~al.}(1983{\natexlab{a}})\citenamefont {Huang}, \citenamefont {Tarn},\
  and\ \citenamefont {Clark}}]{Huang:1983}%
  \BibitemOpen
  \bibfield  {author} {\bibinfo {author} {\bibfnamefont {G.~M.}\ \bibnamefont
  {Huang}}, \bibinfo {author} {\bibfnamefont {T.~J.}\ \bibnamefont {Tarn}},\
  and\ \bibinfo {author} {\bibfnamefont {J.~W.}\ \bibnamefont {Clark}},\ }\href
  {https://doi.org/10.1063/1.525634} {\bibfield  {journal} {\bibinfo  {journal}
  {Journal of Mathematical Physics}\ }\textbf {\bibinfo {volume} {24}},\
  \bibinfo {pages} {2608} (\bibinfo {year} {1983}{\natexlab{a}})}\BibitemShut
  {NoStop}%
\bibitem [{\citenamefont {Acharya}\ \emph {et~al.}(2024)\citenamefont
  {Acharya}, \citenamefont {Abanin}, \citenamefont {Aghababaie-Beni},
  \citenamefont {Aleiner}, \citenamefont {Andersen}, \citenamefont {Ansmann},
  \citenamefont {Arute}, \citenamefont {Arya}, \citenamefont {Asfaw},
  \citenamefont {Astrakhantsev}, \citenamefont {Atalaya}, \citenamefont
  {Babbush}, \citenamefont {Bacon}, \citenamefont {Ballard}, \citenamefont
  {Bardin},\ and\ \citenamefont {Bausch}}]{Google:2024}%
  \BibitemOpen
  \bibfield  {author} {\bibinfo {author} {\bibfnamefont {R.}~\bibnamefont
  {Acharya}}, \bibinfo {author} {\bibfnamefont {D.~A.}\ \bibnamefont {Abanin}},
  \bibinfo {author} {\bibfnamefont {L.}~\bibnamefont {Aghababaie-Beni}},
  \bibinfo {author} {\bibfnamefont {I.}~\bibnamefont {Aleiner}}, \bibinfo
  {author} {\bibfnamefont {T.~I.}\ \bibnamefont {Andersen}}, \bibinfo {author}
  {\bibfnamefont {M.}~\bibnamefont {Ansmann}}, \bibinfo {author} {\bibfnamefont
  {F.}~\bibnamefont {Arute}}, \bibinfo {author} {\bibfnamefont
  {K.}~\bibnamefont {Arya}}, \bibinfo {author} {\bibfnamefont {A.}~\bibnamefont
  {Asfaw}}, \bibinfo {author} {\bibfnamefont {N.}~\bibnamefont
  {Astrakhantsev}}, \bibinfo {author} {\bibfnamefont {J.}~\bibnamefont
  {Atalaya}}, \bibinfo {author} {\bibfnamefont {R.}~\bibnamefont {Babbush}},
  \bibinfo {author} {\bibfnamefont {D.}~\bibnamefont {Bacon}}, \bibinfo
  {author} {\bibfnamefont {B.}~\bibnamefont {Ballard}}, \bibinfo {author}
  {\bibfnamefont {J.~C.}\ \bibnamefont {Bardin}},\ and\ \bibinfo {author}
  {\bibfnamefont {J.~e.~a.}\ \bibnamefont {Bausch}},\ }\href
  {https://doi.org/10.1038/s41586-024-08449-y} {\bibfield  {journal} {\bibinfo
  {journal} {Nature}\ }\textbf {\bibinfo {volume} {638}},\ \bibinfo {pages}
  {920–926} (\bibinfo {year} {2024})}\BibitemShut {NoStop}%
\bibitem [{\citenamefont {Werninghaus}\ \emph {et~al.}(2021)\citenamefont
  {Werninghaus}, \citenamefont {Egger}, \citenamefont {Roy}, \citenamefont
  {Machnes}, \citenamefont {Wilhelm},\ and\ \citenamefont
  {Filipp}}]{Werninghaus:2021}%
  \BibitemOpen
  \bibfield  {author} {\bibinfo {author} {\bibfnamefont {M.}~\bibnamefont
  {Werninghaus}}, \bibinfo {author} {\bibfnamefont {D.~J.}\ \bibnamefont
  {Egger}}, \bibinfo {author} {\bibfnamefont {F.}~\bibnamefont {Roy}}, \bibinfo
  {author} {\bibfnamefont {S.}~\bibnamefont {Machnes}}, \bibinfo {author}
  {\bibfnamefont {F.~K.}\ \bibnamefont {Wilhelm}},\ and\ \bibinfo {author}
  {\bibfnamefont {S.}~\bibnamefont {Filipp}},\ }\href
  {https://doi.org/10.1038/s41534-020-00346-2} {\bibfield  {journal} {\bibinfo
  {journal} {npj Quantum Information}\ }\textbf {\bibinfo {volume} {7}},\
  \bibinfo {pages} {14} (\bibinfo {year} {2021})}\BibitemShut {NoStop}%
\bibitem [{\citenamefont {Kim}\ \emph {et~al.}(2023)\citenamefont {Kim},
  \citenamefont {Eddins}, \citenamefont {Anand}, \citenamefont {Wei},
  \citenamefont {van~den Berg}, \citenamefont {Rosenblatt}, \citenamefont
  {Nayfeh}, \citenamefont {Wu}, \citenamefont {Zaletel}, \citenamefont
  {Temme},\ and\ \citenamefont {Kandala}}]{Kim:2023}%
  \BibitemOpen
  \bibfield  {author} {\bibinfo {author} {\bibfnamefont {Y.}~\bibnamefont
  {Kim}}, \bibinfo {author} {\bibfnamefont {A.}~\bibnamefont {Eddins}},
  \bibinfo {author} {\bibfnamefont {S.}~\bibnamefont {Anand}}, \bibinfo
  {author} {\bibfnamefont {K.~X.}\ \bibnamefont {Wei}}, \bibinfo {author}
  {\bibfnamefont {E.}~\bibnamefont {van~den Berg}}, \bibinfo {author}
  {\bibfnamefont {S.}~\bibnamefont {Rosenblatt}}, \bibinfo {author}
  {\bibfnamefont {H.}~\bibnamefont {Nayfeh}}, \bibinfo {author} {\bibfnamefont
  {Y.}~\bibnamefont {Wu}}, \bibinfo {author} {\bibfnamefont {M.~P.}\
  \bibnamefont {Zaletel}}, \bibinfo {author} {\bibfnamefont {K.}~\bibnamefont
  {Temme}},\ and\ \bibinfo {author} {\bibfnamefont {A.}~\bibnamefont
  {Kandala}},\ }\href {https://api.semanticscholar.org/CorpusID:259157188}
  {\bibfield  {journal} {\bibinfo  {journal} {Nature}\ }\textbf {\bibinfo
  {volume} {618}},\ \bibinfo {pages} {500 } (\bibinfo {year}
  {2023})}\BibitemShut {NoStop}%
\bibitem [{\citenamefont {Adesso}\ \emph {et~al.}(2014)\citenamefont {Adesso},
  \citenamefont {Ragy},\ and\ \citenamefont {Lee}}]{Adesso:Ragy:2014}%
  \BibitemOpen
  \bibfield  {author} {\bibinfo {author} {\bibfnamefont {G.}~\bibnamefont
  {Adesso}}, \bibinfo {author} {\bibfnamefont {S.}~\bibnamefont {Ragy}},\ and\
  \bibinfo {author} {\bibfnamefont {A.~R.}\ \bibnamefont {Lee}},\ }\href
  {https://doi.org/10.1142/S1230161214400010} {\bibfield  {journal} {\bibinfo
  {journal} {Open Systems \& Information Dynamics}\ }\textbf {\bibinfo {volume}
  {21}},\ \bibinfo {pages} {1440001} (\bibinfo {year} {2014})}\BibitemShut
  {NoStop}%
\bibitem [{\citenamefont {Cartan}(1894)}]{Cartan:1894}%
  \BibitemOpen
  \bibfield  {author} {\bibinfo {author} {\bibfnamefont {E.}~\bibnamefont
  {Cartan}},\ }\href {https://books.google.de/books?id=JY8LAAAAYAAJ} {\emph
  {\bibinfo {title} {Sur la structure des groupes de transformations finis et
  continus}}},\ Th{\`e}ses pr{\'e}sent{\'e}es a la Facult{\'e} des Sciences de
  Paris pour obtenir le grade de docteur {\`e}s sciences math{\'e}matiques\
  (\bibinfo  {publisher} {Nony},\ \bibinfo {year} {1894})\BibitemShut {NoStop}%
\bibitem [{\citenamefont
  {b.~Killing}(1888{\natexlab{a}})}]{Killing:erster:1888}%
  \BibitemOpen
  \bibfield  {author} {\bibinfo {author} {\bibfnamefont {W.}~\bibnamefont
  {b.~Killing}},\ }\href {https://doi.org/10.1007/BF01211904} {\bibfield
  {journal} {\bibinfo  {journal} {Mathematische Annalen}\ }\textbf {\bibinfo
  {volume} {31}},\ \bibinfo {pages} {252} (\bibinfo {year}
  {1888}{\natexlab{a}})}\BibitemShut {NoStop}%
\bibitem [{\citenamefont
  {b.~Killing}(1888{\natexlab{b}})}]{Killing:zweiter:1888}%
  \BibitemOpen
  \bibfield  {author} {\bibinfo {author} {\bibfnamefont {W.}~\bibnamefont
  {b.~Killing}},\ }\href {https://doi.org/10.1007/BF01444109} {\bibfield
  {journal} {\bibinfo  {journal} {Mathematische Annalen}\ }\textbf {\bibinfo
  {volume} {33}},\ \bibinfo {pages} {1} (\bibinfo {year}
  {1888}{\natexlab{b}})}\BibitemShut {NoStop}%
\bibitem [{\citenamefont
  {b.~Killing}(1889{\natexlab{a}})}]{Killing:dritter:1889}%
  \BibitemOpen
  \bibfield  {author} {\bibinfo {author} {\bibfnamefont {W.}~\bibnamefont
  {b.~Killing}},\ }\href {https://doi.org/10.1007/BF01446792} {\bibfield
  {journal} {\bibinfo  {journal} {Mathematische Annalen}\ }\textbf {\bibinfo
  {volume} {34}},\ \bibinfo {pages} {57} (\bibinfo {year}
  {1889}{\natexlab{a}})}\BibitemShut {NoStop}%
\bibitem [{\citenamefont
  {b.~Killing}(1889{\natexlab{b}})}]{Killing:vierter:1890}%
  \BibitemOpen
  \bibfield  {author} {\bibinfo {author} {\bibfnamefont {W.}~\bibnamefont
  {b.~Killing}},\ }\href {https://doi.org/10.1007/BF01207837} {\bibfield
  {journal} {\bibinfo  {journal} {Mathematische Annalen}\ }\textbf {\bibinfo
  {volume} {36}},\ \bibinfo {pages} {161} (\bibinfo {year}
  {1889}{\natexlab{b}})}\BibitemShut {NoStop}%
\bibitem [{\citenamefont {Bianchi}(1903)}]{Bianchi:1903}%
  \BibitemOpen
  \bibfield  {author} {\bibinfo {author} {\bibfnamefont {L.}~\bibnamefont
  {Bianchi}},\ }\href@noop {} {\emph {\bibinfo {title} {Lezoni sulla teoria dei
  gruppi continui finiti di trasformazioni}}}\ (\bibinfo  {publisher} {Pisa},\
  \bibinfo {year} {1903})\ Chap.\ \bibinfo {chapter} {§199 I due tipi dei G3
  non integrabili (semplici)}, pp.\ \bibinfo {pages} {555--557}\BibitemShut
  {NoStop}%
\bibitem [{\citenamefont {Levi}(1905)}]{Levi:1905}%
  \BibitemOpen
  \bibfield  {author} {\bibinfo {author} {\bibfnamefont {E.~E.}\ \bibnamefont
  {Levi}},\ }\href@noop {} {\bibfield  {journal} {\bibinfo  {journal} {Atti
  della Reale Accademia delle Scienze di Torino}\ }\textbf {\bibinfo {volume}
  {4}} (\bibinfo {year} {1905})}\BibitemShut {NoStop}%
\bibitem [{\citenamefont {Mal'tsev}(1942)}]{Matltsev:1942}%
  \BibitemOpen
  \bibfield  {author} {\bibinfo {author} {\bibfnamefont {A.~I.}\ \bibnamefont
  {Mal'tsev}},\ }\href@noop {} {\bibfield  {journal} {\bibinfo  {journal}
  {Comptes Rendus (Doklady) de l'Acad{\'e}mie des Sciences de l'URSS, Nouvelle
  S{\'e}rie}\ }\textbf {\bibinfo {volume} {36}},\ \bibinfo {pages} {42}
  (\bibinfo {year} {1942})}\BibitemShut {NoStop}%
\bibitem [{\citenamefont {Qi}(2019)}]{Qi:2019}%
  \BibitemOpen
  \bibfield  {author} {\bibinfo {author} {\bibfnamefont {L.}~\bibnamefont
  {Qi}},\ }\href@noop {} {\bibinfo {title} {{Nilpotent Decomposition of
  Solvable Lie Algebras}}} (\bibinfo {year} {2019}),\ \Eprint
  {https://arxiv.org/abs/1903.03776} {arXiv:1903.03776 [math.RA]} \BibitemShut
  {NoStop}%
\bibitem [{\citenamefont {del Barco}\ \emph {et~al.}(2025)\citenamefont {del
  Barco}, \citenamefont {Infanti}, \citenamefont {Rivas},\ and\ \citenamefont
  {Schwahn}}]{DelBarco:2025}%
  \BibitemOpen
  \bibfield  {author} {\bibinfo {author} {\bibfnamefont {V.}~\bibnamefont {del
  Barco}}, \bibinfo {author} {\bibfnamefont {G.}~\bibnamefont {Infanti}},
  \bibinfo {author} {\bibfnamefont {E.}~\bibnamefont {Rivas}},\ and\ \bibinfo
  {author} {\bibfnamefont {P.}~\bibnamefont {Schwahn}},\ }\href@noop {}
  {\bibinfo {title} {{Formalizing a classification theorem for low-dimensional
  solvable Lie algebras in Lean}}} (\bibinfo {year} {2025}),\ \Eprint
  {https://arxiv.org/abs/2505.19975} {arXiv:2505.19975 [cs.LO]} \BibitemShut
  {NoStop}%
\bibitem [{\citenamefont {Oeh}(2023)}]{Oeh:2023}%
  \BibitemOpen
  \bibfield  {author} {\bibinfo {author} {\bibfnamefont {D.}~\bibnamefont
  {Oeh}},\ }\href {https://doi.org/10.1007/s10711-023-00785-z} {\bibfield
  {journal} {\bibinfo  {journal} {Geometriae Dedicata}\ }\textbf {\bibinfo
  {volume} {217}},\ \bibinfo {pages} {50} (\bibinfo {year} {2023})}\BibitemShut
  {NoStop}%
\bibitem [{\citenamefont {{de Graaf}}(2007)}]{DeGraaf:2007}%
  \BibitemOpen
  \bibfield  {author} {\bibinfo {author} {\bibfnamefont {W.~A.}\ \bibnamefont
  {{de Graaf}}},\ }\href
  {https://doi.org/https://doi.org/10.1016/j.jalgebra.2006.08.006} {\bibfield
  {journal} {\bibinfo  {journal} {Journal of Algebra}\ }\textbf {\bibinfo
  {volume} {309}},\ \bibinfo {pages} {640} (\bibinfo {year} {2007})},\ \bibinfo
  {note} {computational Algebra}\BibitemShut {NoStop}%
\bibitem [{\citenamefont {{Ancochea Bermúdez}}\ and\ \citenamefont
  {Campoamor-Stursberg}(2015)}]{Ancocha:2015}%
  \BibitemOpen
  \bibfield  {author} {\bibinfo {author} {\bibfnamefont {J.}~\bibnamefont
  {{Ancochea Bermúdez}}}\ and\ \bibinfo {author} {\bibfnamefont
  {R.}~\bibnamefont {Campoamor-Stursberg}},\ }\href
  {https://doi.org/https://doi.org/10.1016/j.laa.2014.12.019} {\bibfield
  {journal} {\bibinfo  {journal} {Linear Algebra and its Applications}\
  }\textbf {\bibinfo {volume} {471}},\ \bibinfo {pages} {54} (\bibinfo {year}
  {2015})}\BibitemShut {NoStop}%
\bibitem [{\citenamefont {Choriyeva}\ and\ \citenamefont
  {Khudoyberdiyev}(2024)}]{Choriyeva:2024}%
  \BibitemOpen
  \bibfield  {author} {\bibinfo {author} {\bibfnamefont {I.}~\bibnamefont
  {Choriyeva}}\ and\ \bibinfo {author} {\bibfnamefont {A.}~\bibnamefont
  {Khudoyberdiyev}},\ }\href
  {https://doi.org/https://doi.org/10.1007/s41980-024-00874-z} {\bibfield
  {journal} {\bibinfo  {journal} {Bulletin of the Iranian Mathematical
  Society}\ }\textbf {\bibinfo {volume} {50}},\ \bibinfo {pages} {33 }
  (\bibinfo {year} {2024})}\BibitemShut {NoStop}%
\bibitem [{\citenamefont {Igusa}(1981)}]{Igusa1981}%
  \BibitemOpen
  \bibfield  {author} {\bibinfo {author} {\bibfnamefont {J.-i.}\ \bibnamefont
  {Igusa}},\ }\bibinfo {title} {On lie algebras generated by two differential
  operators},\ in\ \href {https://doi.org/10.1007/978-1-4612-5987-9_9} {\emph
  {\bibinfo {booktitle} {Manifolds and Lie Groups: Papers in Honor of Yoz{\^o}
  Matsushima}}},\ \bibinfo {editor} {edited by\ \bibinfo {editor}
  {\bibfnamefont {J.-i.}\ \bibnamefont {Hano}}, \bibinfo {editor}
  {\bibfnamefont {A.}~\bibnamefont {Morimoto}}, \bibinfo {editor}
  {\bibfnamefont {S.}~\bibnamefont {Murakami}}, \bibinfo {editor}
  {\bibfnamefont {K.}~\bibnamefont {Okamoto}},\ and\ \bibinfo {editor}
  {\bibfnamefont {H.}~\bibnamefont {Ozeki}}}\ (\bibinfo  {publisher}
  {Birkh{\"a}user Boston},\ \bibinfo {address} {Boston, MA},\ \bibinfo {year}
  {1981})\ pp.\ \bibinfo {pages} {187--195}\BibitemShut {NoStop}%
\bibitem [{\citenamefont {Srednicki}(2007)}]{Srednicki:2007}%
  \BibitemOpen
  \bibfield  {author} {\bibinfo {author} {\bibfnamefont {M.}~\bibnamefont
  {Srednicki}},\ }\href {https://doi.org/10.1017/CBO9780511813917} {\emph
  {\bibinfo {title} {Quantum Field Theory}}}\ (\bibinfo  {publisher} {Cambridge
  University Press},\ \bibinfo {address} {Cambridge},\ \bibinfo {year}
  {2007})\BibitemShut {NoStop}%
\bibitem [{\citenamefont {Ashcroft}\ and\ \citenamefont
  {Mermin}(1976)}]{Ashcroft:Mermin:1976}%
  \BibitemOpen
  \bibfield  {author} {\bibinfo {author} {\bibfnamefont {N.}~\bibnamefont
  {Ashcroft}}\ and\ \bibinfo {author} {\bibfnamefont {N.}~\bibnamefont
  {Mermin}},\ }\href {https://books.google.de/books?id=1C9HAQAAIAAJ} {\emph
  {\bibinfo {title} {Solid State Physics}}},\ HRW international editions\
  (\bibinfo  {publisher} {Holt, Rinehart and Winston},\ \bibinfo {year}
  {1976})\BibitemShut {NoStop}%
\bibitem [{\citenamefont {Huttner}\ and\ \citenamefont
  {Barnett}(1992)}]{Huttner:1992}%
  \BibitemOpen
  \bibfield  {author} {\bibinfo {author} {\bibfnamefont {B.}~\bibnamefont
  {Huttner}}\ and\ \bibinfo {author} {\bibfnamefont {S.~M.}\ \bibnamefont
  {Barnett}},\ }\href {https://doi.org/10.1103/PhysRevA.46.4306} {\bibfield
  {journal} {\bibinfo  {journal} {Physical Review A}\ }\textbf {\bibinfo
  {volume} {46}},\ \bibinfo {pages} {4306} (\bibinfo {year}
  {1992})}\BibitemShut {NoStop}%
\bibitem [{\citenamefont {Bozhevolnyi}\ and\ \citenamefont
  {Khurgin}(2017)}]{Bozhevolnyi:2017}%
  \BibitemOpen
  \bibfield  {author} {\bibinfo {author} {\bibfnamefont {S.~I.}\ \bibnamefont
  {Bozhevolnyi}}\ and\ \bibinfo {author} {\bibfnamefont {J.~B.}\ \bibnamefont
  {Khurgin}},\ }\href
  {https://doi.org/https://doi.org/10.1038/nphoton.2017.103} {\bibfield
  {journal} {\bibinfo  {journal} {Nature Photonics}\ }\textbf {\bibinfo
  {volume} {11}},\ \bibinfo {pages} {398 } (\bibinfo {year}
  {2017})}\BibitemShut {NoStop}%
\bibitem [{\citenamefont {Haroche}\ and\ \citenamefont
  {Raimond}(2006)}]{Haroche:Raimond:2006}%
  \BibitemOpen
  \bibfield  {author} {\bibinfo {author} {\bibfnamefont {S.}~\bibnamefont
  {Haroche}}\ and\ \bibinfo {author} {\bibfnamefont {J.-M.}\ \bibnamefont
  {Raimond}},\ }\href
  {https://doi.org/10.1093/acprof:oso/9780198509141.001.0001} {\emph {\bibinfo
  {title} {Exploring the Quantum: Atoms, Cavities, and Photons}}}\ (\bibinfo
  {publisher} {Oxford University Press},\ \bibinfo {year} {2006})\BibitemShut
  {NoStop}%
\bibitem [{\citenamefont {Aspelmeyer}\ \emph {et~al.}(2014)\citenamefont
  {Aspelmeyer}, \citenamefont {Kippenberg},\ and\ \citenamefont
  {Marquardt}}]{Aspelmeyer:Kippenberg:2014}%
  \BibitemOpen
  \bibfield  {author} {\bibinfo {author} {\bibfnamefont {M.}~\bibnamefont
  {Aspelmeyer}}, \bibinfo {author} {\bibfnamefont {T.~J.}\ \bibnamefont
  {Kippenberg}},\ and\ \bibinfo {author} {\bibfnamefont {F.}~\bibnamefont
  {Marquardt}},\ }\href {https://doi.org/10.1103/RevModPhys.86.1391} {\bibfield
   {journal} {\bibinfo  {journal} {Reviews of Modern Physics}\ }\textbf
  {\bibinfo {volume} {86}},\ \bibinfo {pages} {1391} (\bibinfo {year}
  {2014})}\BibitemShut {NoStop}%
\bibitem [{\citenamefont {A.}\ and\ \citenamefont
  {I.}(2007)}]{Adesso:Illuminati:2007:v2}%
  \BibitemOpen
  \bibfield  {author} {\bibinfo {author} {\bibfnamefont {G.}~\bibnamefont
  {A.}}\ and\ \bibinfo {author} {\bibfnamefont {F.}~\bibnamefont {I.}},\
  }\bibinfo {title} {Bipartite and multipartite entanglement of gaussian
  states},\ in\ \href {https://doi.org/10.1142/9781860948169_0001} {\emph
  {\bibinfo {booktitle} {Quantum Information with Continuous Variables of Atoms
  and Light}}}\ (\bibinfo  {publisher} {World Scientific},\ \bibinfo {year}
  {2007})\ pp.\ \bibinfo {pages} {1--21}\BibitemShut {NoStop}%
\bibitem [{\citenamefont {Weedbrook}\ \emph {et~al.}(2012)\citenamefont
  {Weedbrook}, \citenamefont {Pirandola}, \citenamefont {Garc\'{\i}a-Patr\'on},
  \citenamefont {Cerf}, \citenamefont {Ralph}, \citenamefont {Shapiro},\ and\
  \citenamefont {Lloyd}}]{Weedbrook:Pirandola:2012}%
  \BibitemOpen
  \bibfield  {author} {\bibinfo {author} {\bibfnamefont {C.}~\bibnamefont
  {Weedbrook}}, \bibinfo {author} {\bibfnamefont {S.}~\bibnamefont
  {Pirandola}}, \bibinfo {author} {\bibfnamefont {R.}~\bibnamefont
  {Garc\'{\i}a-Patr\'on}}, \bibinfo {author} {\bibfnamefont {N.~J.}\
  \bibnamefont {Cerf}}, \bibinfo {author} {\bibfnamefont {T.~C.}\ \bibnamefont
  {Ralph}}, \bibinfo {author} {\bibfnamefont {J.~H.}\ \bibnamefont {Shapiro}},\
  and\ \bibinfo {author} {\bibfnamefont {S.}~\bibnamefont {Lloyd}},\ }\href
  {https://doi.org/10.1103/RevModPhys.84.621} {\bibfield  {journal} {\bibinfo
  {journal} {Reviews of Modern Physics}\ }\textbf {\bibinfo {volume} {84}},\
  \bibinfo {pages} {621} (\bibinfo {year} {2012})}\BibitemShut {NoStop}%
\bibitem [{\citenamefont {Mann}\ and\ \citenamefont
  {Ralph}(2012)}]{Mann:Ralph:2012}%
  \BibitemOpen
  \bibfield  {author} {\bibinfo {author} {\bibfnamefont {R.~B.}\ \bibnamefont
  {Mann}}\ and\ \bibinfo {author} {\bibfnamefont {T.~C.}\ \bibnamefont
  {Ralph}},\ }\href {https://doi.org/10.1088/0264-9381/29/22/220301} {\bibfield
   {journal} {\bibinfo  {journal} {Classical and Quantum Gravity}\ }\textbf
  {\bibinfo {volume} {29}},\ \bibinfo {pages} {220301} (\bibinfo {year}
  {2012})}\BibitemShut {NoStop}%
\bibitem [{\citenamefont {Georgi}(1999)}]{Georgi:1999}%
  \BibitemOpen
  \bibfield  {author} {\bibinfo {author} {\bibfnamefont {H.}~\bibnamefont
  {Georgi}},\ }\href {https://books.google.de/books?id=g4yEuH5rBMUC} {\emph
  {\bibinfo {title} {Lie Algebras In Particle Physics: from Isospin To Unified
  Theories}}},\ Frontiers in Physics\ (\bibinfo  {publisher} {Avalon
  Publishing},\ \bibinfo {year} {1999})\BibitemShut {NoStop}%
\bibitem [{\citenamefont {Woit}(2017{\natexlab{a}})}]{Woit:2017}%
  \BibitemOpen
  \bibfield  {author} {\bibinfo {author} {\bibfnamefont {P.}~\bibnamefont
  {Woit}},\ }\href {https://doi.org/https://doi.org/10.1007/978-3-319-64612-1}
  {\emph {\bibinfo {title} {Quantum Theory, Groups and Representations}}},\
  \bibinfo {edition} {1st}\ ed.\ (\bibinfo  {publisher} {Springer},\ \bibinfo
  {address} {Cham},\ \bibinfo {year} {2017})\BibitemShut {NoStop}%
\bibitem [{\citenamefont {Knapp}(1996)}]{Knapp:1996}%
  \BibitemOpen
  \bibfield  {author} {\bibinfo {author} {\bibfnamefont {A.~W.}\ \bibnamefont
  {Knapp}},\ }\href {https://doi.org/https://doi.org/10.1007/978-1-4757-2453-0}
  {\emph {\bibinfo {title} {Lie Groups Beyond an Introduction}}},\ \bibinfo
  {edition} {1st}\ ed.\ (\bibinfo  {publisher} {Birkhäuser Boston, MA},\
  \bibinfo {year} {1996})\BibitemShut {NoStop}%
\bibitem [{\citenamefont {Shore}\ \emph {et~al.}(2020)\citenamefont {Shore},
  \citenamefont {Sumner},\ and\ \citenamefont {Holland}}]{Shore:2020}%
  \BibitemOpen
  \bibfield  {author} {\bibinfo {author} {\bibfnamefont {J.~A.}\ \bibnamefont
  {Shore}}, \bibinfo {author} {\bibfnamefont {J.~G.}\ \bibnamefont {Sumner}},\
  and\ \bibinfo {author} {\bibfnamefont {B.~R.}\ \bibnamefont {Holland}},\
  }\href@noop {} {\bibinfo {title} {The impracticalities of
  multiplicatively-closed codon models: a retreat to linear alternatives}}
  (\bibinfo {year} {2020}),\ \Eprint {https://arxiv.org/abs/1804.11249}
  {arXiv:1804.11249 [q-bio.PE]} \BibitemShut {NoStop}%
\bibitem [{\citenamefont {Che}\ \emph {et~al.}(2021)\citenamefont {Che},
  \citenamefont {Wei}, \citenamefont {Huang}, \citenamefont {Zhao},
  \citenamefont {Xue}, \citenamefont {Nie}, \citenamefont {Li}, \citenamefont
  {Lu},\ and\ \citenamefont {Xin}}]{Liangyu:2021}%
  \BibitemOpen
  \bibfield  {author} {\bibinfo {author} {\bibfnamefont {L.}~\bibnamefont
  {Che}}, \bibinfo {author} {\bibfnamefont {C.}~\bibnamefont {Wei}}, \bibinfo
  {author} {\bibfnamefont {Y.}~\bibnamefont {Huang}}, \bibinfo {author}
  {\bibfnamefont {D.}~\bibnamefont {Zhao}}, \bibinfo {author} {\bibfnamefont
  {S.}~\bibnamefont {Xue}}, \bibinfo {author} {\bibfnamefont {X.}~\bibnamefont
  {Nie}}, \bibinfo {author} {\bibfnamefont {J.}~\bibnamefont {Li}}, \bibinfo
  {author} {\bibfnamefont {D.}~\bibnamefont {Lu}},\ and\ \bibinfo {author}
  {\bibfnamefont {T.}~\bibnamefont {Xin}},\ }\href
  {https://doi.org/10.1103/PhysRevResearch.3.023246} {\bibfield  {journal}
  {\bibinfo  {journal} {Physical Review Research}\ }\textbf {\bibinfo {volume}
  {3}},\ \bibinfo {pages} {023246} (\bibinfo {year} {2021})}\BibitemShut
  {NoStop}%
\bibitem [{\citenamefont {{Acosta Coden}}\ \emph {et~al.}(2021)\citenamefont
  {{Acosta Coden}}, \citenamefont {Gómez}, \citenamefont {Ferrón},\ and\
  \citenamefont {Osenda}}]{Coden:2021}%
  \BibitemOpen
  \bibfield  {author} {\bibinfo {author} {\bibfnamefont {D.}~\bibnamefont
  {{Acosta Coden}}}, \bibinfo {author} {\bibfnamefont {S.}~\bibnamefont
  {Gómez}}, \bibinfo {author} {\bibfnamefont {A.}~\bibnamefont {Ferrón}},\
  and\ \bibinfo {author} {\bibfnamefont {O.}~\bibnamefont {Osenda}},\ }\href
  {https://doi.org/https://doi.org/10.1016/j.physleta.2020.127009} {\bibfield
  {journal} {\bibinfo  {journal} {Physics Letters A}\ }\textbf {\bibinfo
  {volume} {387}},\ \bibinfo {pages} {127009} (\bibinfo {year}
  {2021})}\BibitemShut {NoStop}%
\bibitem [{\citenamefont {Dimo}\ and\ \citenamefont
  {Faribault}(2022)}]{Dimo:2022}%
  \BibitemOpen
  \bibfield  {author} {\bibinfo {author} {\bibfnamefont {C.}~\bibnamefont
  {Dimo}}\ and\ \bibinfo {author} {\bibfnamefont {A.}~\bibnamefont
  {Faribault}},\ }\href {https://doi.org/10.1103/PhysRevB.105.L121404}
  {\bibfield  {journal} {\bibinfo  {journal} {Physical Review B}\ }\textbf
  {\bibinfo {volume} {105}},\ \bibinfo {pages} {L121404} (\bibinfo {year}
  {2022})}\BibitemShut {NoStop}%
\bibitem [{\citenamefont
  {Hall}(2013{\natexlab{a}})}]{Hall:Lie:quantum:theory:16}%
  \BibitemOpen
  \bibfield  {author} {\bibinfo {author} {\bibfnamefont {B.~C.}\ \bibnamefont
  {Hall}},\ }\href {https://doi.org/https://doi.org/10.1007/978-1-4614-7116-5}
  {\emph {\bibinfo {title} {Quantum Theory for Mathematicians}}},\ \bibinfo
  {edition} {1st}\ ed.\ (\bibinfo  {publisher} {Springer New York, NY},\
  \bibinfo {year} {2013})\BibitemShut {NoStop}%
\bibitem [{\citenamefont {Bogachev}\ and\ \citenamefont
  {Smolyanov}(2020)}]{Bogachev:2020}%
  \BibitemOpen
  \bibfield  {author} {\bibinfo {author} {\bibfnamefont {V.~I.}\ \bibnamefont
  {Bogachev}}\ and\ \bibinfo {author} {\bibfnamefont {O.~G.}\ \bibnamefont
  {Smolyanov}},\ }\bibinfo {title} {{Linear Operators and Functionals}},\ in\
  \href {https://doi.org/10.1007/978-3-030-38219-3_6} {\emph {\bibinfo
  {booktitle} {Real and Functional Analysis}}}\ (\bibinfo  {publisher}
  {Springer International Publishing},\ \bibinfo {address} {Cham},\ \bibinfo
  {year} {2020})\ pp.\ \bibinfo {pages} {187--278}\BibitemShut {NoStop}%
\bibitem [{\citenamefont {Bagarello}(2007)}]{Bagarello:2007}%
  \BibitemOpen
  \bibfield  {author} {\bibinfo {author} {\bibfnamefont {F.}~\bibnamefont
  {Bagarello}},\ }\href {https://doi.org/10.1063/1.2423230} {\bibfield
  {journal} {\bibinfo  {journal} {Journal of Mathematical Physics}\ }\textbf
  {\bibinfo {volume} {48}},\ \bibinfo {pages} {013511} (\bibinfo {year}
  {2007})}\BibitemShut {NoStop}%
\bibitem [{\citenamefont {Agust\'{\i}}\ \emph {et~al.}(2020)\citenamefont
  {Agust\'{\i}}, \citenamefont {Chang}, \citenamefont {Quijandr\'{\i}a},
  \citenamefont {Johansson}, \citenamefont {Wilson},\ and\ \citenamefont
  {Sab\'{\i}n}}]{Agust:SPDC:2020}%
  \BibitemOpen
  \bibfield  {author} {\bibinfo {author} {\bibfnamefont {A.}~\bibnamefont
  {Agust\'{\i}}}, \bibinfo {author} {\bibfnamefont {C.~W.~S.}\ \bibnamefont
  {Chang}}, \bibinfo {author} {\bibfnamefont {F.}~\bibnamefont
  {Quijandr\'{\i}a}}, \bibinfo {author} {\bibfnamefont {G.}~\bibnamefont
  {Johansson}}, \bibinfo {author} {\bibfnamefont {C.~M.}\ \bibnamefont
  {Wilson}},\ and\ \bibinfo {author} {\bibfnamefont {C.}~\bibnamefont
  {Sab\'{\i}n}},\ }\href {https://doi.org/10.1103/PhysRevLett.125.020502}
  {\bibfield  {journal} {\bibinfo  {journal} {Physical Review Letters}\
  }\textbf {\bibinfo {volume} {125}},\ \bibinfo {pages} {020502} (\bibinfo
  {year} {2020})}\BibitemShut {NoStop}%
\bibitem [{\citenamefont {Huang}\ \emph
  {et~al.}(1983{\natexlab{b}})\citenamefont {Huang}, \citenamefont {Tarn},\
  and\ \citenamefont {Clark}}]{Huang:Tarn:1983}%
  \BibitemOpen
  \bibfield  {author} {\bibinfo {author} {\bibfnamefont {G.~M.}\ \bibnamefont
  {Huang}}, \bibinfo {author} {\bibfnamefont {T.~J.}\ \bibnamefont {Tarn}},\
  and\ \bibinfo {author} {\bibfnamefont {J.~W.}\ \bibnamefont {Clark}},\ }\href
  {https://doi.org/10.1063/1.525634} {\bibfield  {journal} {\bibinfo  {journal}
  {Journal of Mathematical Physics}\ }\textbf {\bibinfo {volume} {24}},\
  \bibinfo {pages} {2608} (\bibinfo {year} {1983}{\natexlab{b}})}\BibitemShut
  {NoStop}%
\bibitem [{\citenamefont {Lan}\ \emph {et~al.}(2005)\citenamefont {Lan},
  \citenamefont {Tarn}, \citenamefont {Chi},\ and\ \citenamefont
  {Clark}}]{Lan:2005}%
  \BibitemOpen
  \bibfield  {author} {\bibinfo {author} {\bibfnamefont {C.}~\bibnamefont
  {Lan}}, \bibinfo {author} {\bibfnamefont {T.-J.}\ \bibnamefont {Tarn}},
  \bibinfo {author} {\bibfnamefont {Q.-S.}\ \bibnamefont {Chi}},\ and\ \bibinfo
  {author} {\bibfnamefont {J.~W.}\ \bibnamefont {Clark}},\ }\href
  {https://doi.org/10.1063/1.1867979} {\bibfield  {journal} {\bibinfo
  {journal} {Journal of Mathematical Physics}\ }\textbf {\bibinfo {volume}
  {46}},\ \bibinfo {pages} {052102} (\bibinfo {year} {2005})}\BibitemShut
  {NoStop}%
\bibitem [{\citenamefont {Wu}\ \emph {et~al.}(2006)\citenamefont {Wu},
  \citenamefont {Tarn},\ and\ \citenamefont {Li}}]{Wu:Tarn:2006}%
  \BibitemOpen
  \bibfield  {author} {\bibinfo {author} {\bibfnamefont {R.-B.}\ \bibnamefont
  {Wu}}, \bibinfo {author} {\bibfnamefont {T.-J.}\ \bibnamefont {Tarn}},\ and\
  \bibinfo {author} {\bibfnamefont {C.-W.}\ \bibnamefont {Li}},\ }\href
  {https://doi.org/10.1103/PhysRevA.73.012719} {\bibfield  {journal} {\bibinfo
  {journal} {Physical Review A}\ }\textbf {\bibinfo {volume} {73}},\ \bibinfo
  {pages} {012719} (\bibinfo {year} {2006})}\BibitemShut {NoStop}%
\bibitem [{\citenamefont {Nelson}(1959)}]{Nelson:1959}%
  \BibitemOpen
  \bibfield  {author} {\bibinfo {author} {\bibfnamefont {E.}~\bibnamefont
  {Nelson}},\ }\href {http://www.jstor.org/stable/1970331} {\bibfield
  {journal} {\bibinfo  {journal} {Annals of Mathematics}\ }\textbf {\bibinfo
  {volume} {70}},\ \bibinfo {pages} {572} (\bibinfo {year} {1959})}\BibitemShut
  {NoStop}%
\bibitem [{\citenamefont {Bloch}\ \emph {et~al.}(2010)\citenamefont {Bloch},
  \citenamefont {Brockett},\ and\ \citenamefont
  {Rangan}}]{Bloch:finite:controlability:2010}%
  \BibitemOpen
  \bibfield  {author} {\bibinfo {author} {\bibfnamefont {A.~M.}\ \bibnamefont
  {Bloch}}, \bibinfo {author} {\bibfnamefont {R.~W.}\ \bibnamefont
  {Brockett}},\ and\ \bibinfo {author} {\bibfnamefont {C.}~\bibnamefont
  {Rangan}},\ }\href {https://doi.org/10.1109/TAC.2010.2044273} {\bibfield
  {journal} {\bibinfo  {journal} {IEEE Transactions on Automatic Control}\
  }\textbf {\bibinfo {volume} {55}},\ \bibinfo {pages} {1797} (\bibinfo {year}
  {2010})}\BibitemShut {NoStop}%
\bibitem [{\citenamefont {Jurdjevic}\ and\ \citenamefont
  {Sussmann}(1972)}]{JS:1972}%
  \BibitemOpen
  \bibfield  {author} {\bibinfo {author} {\bibfnamefont {V.}~\bibnamefont
  {Jurdjevic}}\ and\ \bibinfo {author} {\bibfnamefont {H.~J.}\ \bibnamefont
  {Sussmann}},\ }\href {https://doi.org/10.1016/0022-0396(72)90035-6}
  {\bibfield  {journal} {\bibinfo  {journal} {Journal of Differential
  Equations}\ }\textbf {\bibinfo {volume} {12}},\ \bibinfo {pages} {313}
  (\bibinfo {year} {1972})}\BibitemShut {NoStop}%
\bibitem [{\citenamefont {Cohen-Tannoudji}\ \emph {et~al.}(2020)\citenamefont
  {Cohen-Tannoudji}, \citenamefont {Diu},\ and\ \citenamefont
  {Laloë}}]{Cohen:Tannoudji:2020}%
  \BibitemOpen
  \bibfield  {author} {\bibinfo {author} {\bibfnamefont {C.}~\bibnamefont
  {Cohen-Tannoudji}}, \bibinfo {author} {\bibfnamefont {B.}~\bibnamefont
  {Diu}},\ and\ \bibinfo {author} {\bibfnamefont {F.}~\bibnamefont {Laloë}},\
  }\href@noop {} {\emph {\bibinfo {title} {Quantenmechanik}}},\ \bibinfo
  {edition} {5th}\ ed.\ (\bibinfo  {publisher} {De Gruyter},\ \bibinfo {year}
  {2020})\BibitemShut {NoStop}%
\bibitem [{\citenamefont {Graustein}(1930)}]{Graustein:1930}%
  \BibitemOpen
  \bibfield  {author} {\bibinfo {author} {\bibfnamefont {W.}~\bibnamefont
  {Graustein}},\ }\href {https://books.google.de/books?id=bTE6AAAAMAAJ} {\emph
  {\bibinfo {title} {Introduction to Higher Geometry}}},\ Series of
  mathematical texts ed. by E. R. Hedrick\ (\bibinfo  {publisher} {Macmillan},\
  \bibinfo {year} {1930})\BibitemShut {NoStop}%
\bibitem [{\citenamefont {Popovych}\ \emph {et~al.}(2003)\citenamefont
  {Popovych}, \citenamefont {Boyko}, \citenamefont {Nesterenko},\ and\
  \citenamefont {Lutfullin}}]{Popovych:2003}%
  \BibitemOpen
  \bibfield  {author} {\bibinfo {author} {\bibfnamefont {R.~O.}\ \bibnamefont
  {Popovych}}, \bibinfo {author} {\bibfnamefont {V.~M.}\ \bibnamefont {Boyko}},
  \bibinfo {author} {\bibfnamefont {M.~O.}\ \bibnamefont {Nesterenko}},\ and\
  \bibinfo {author} {\bibfnamefont {M.~W.}\ \bibnamefont {Lutfullin}},\ }\href
  {https://doi.org/10.1088/0305-4470/36/26/309} {\bibfield  {journal} {\bibinfo
   {journal} {Journal of Physics A: Mathematical and General}\ }\textbf
  {\bibinfo {volume} {36}},\ \bibinfo {pages} {7337–7360} (\bibinfo {year}
  {2003})}\BibitemShut {NoStop}%
\bibitem [{\citenamefont {Gorbatsevich}(2017)}]{Gorbatsevich:2017}%
  \BibitemOpen
  \bibfield  {author} {\bibinfo {author} {\bibfnamefont {V.~V.}\ \bibnamefont
  {Gorbatsevich}},\ }\href {https://doi.org/10.1134/S0001434617050054}
  {\bibfield  {journal} {\bibinfo  {journal} {Mathematical Notes}\ }\textbf
  {\bibinfo {volume} {101}},\ \bibinfo {pages} {795 } (\bibinfo {year}
  {2017})}\BibitemShut {NoStop}%
\bibitem [{\citenamefont {Joseph}(1972)}]{Joseph:1972}%
  \BibitemOpen
  \bibfield  {author} {\bibinfo {author} {\bibfnamefont {A.}~\bibnamefont
  {Joseph}},\ }\href {https://doi.org/10.1063/1.1665983} {\bibfield  {journal}
  {\bibinfo  {journal} {Journal of Mathematical Physics}\ }\textbf {\bibinfo
  {volume} {13}},\ \bibinfo {pages} {351} (\bibinfo {year} {1972})}\BibitemShut
  {NoStop}%
\bibitem [{\citenamefont {Kuz'min}(1977)}]{Kuzmin:1977}%
  \BibitemOpen
  \bibfield  {author} {\bibinfo {author} {\bibfnamefont {E.}~\bibnamefont
  {Kuz'min}},\ }\href {https://doi.org/https://doi.org/10.1007/BF01669280}
  {\bibfield  {journal} {\bibinfo  {journal} {Algebra and Logic}\ }\textbf
  {\bibinfo {volume} {16}},\ \bibinfo {pages} {286} (\bibinfo {year}
  {1977})}\BibitemShut {NoStop}%
\bibitem [{\citenamefont {Björk}(1979)}]{Bjoerk:1979}%
  \BibitemOpen
  \bibfield  {author} {\bibinfo {author} {\bibfnamefont {J.-E.}\ \bibnamefont
  {Björk}},\ }\href@noop {} {\emph {\bibinfo {title} {Rings of Differential
  Operators}}}\ (\bibinfo  {publisher} {North-Holland},\ \bibinfo {address}
  {Amsterdam},\ \bibinfo {year} {1979})\BibitemShut {NoStop}%
\bibitem [{\citenamefont {Dixmier}(1968)}]{Dixmier:1968}%
  \BibitemOpen
  \bibfield  {author} {\bibinfo {author} {\bibfnamefont {J.}~\bibnamefont
  {Dixmier}},\ }\href {https://doi.org/10.24033/bsmf.1667} {\bibfield
  {journal} {\bibinfo  {journal} {Bulletin de la Société Mathématique de
  France}\ }\textbf {\bibinfo {volume} {96}},\ \bibinfo {pages} {209} (\bibinfo
  {year} {1968})}\BibitemShut {NoStop}%
\bibitem [{\citenamefont {Dixmier}(1977)}]{Dixmier:1977}%
  \BibitemOpen
  \bibfield  {author} {\bibinfo {author} {\bibfnamefont {J.}~\bibnamefont
  {Dixmier}},\ }\href@noop {} {\emph {\bibinfo {title} {Enveloping Algebras}}}\
  (\bibinfo  {publisher} {North-Holland},\ \bibinfo {address} {Amsterdam},\
  \bibinfo {year} {1977})\BibitemShut {NoStop}%
\bibitem [{\citenamefont {Birkhoff}(1937)}]{PBW:Thm}%
  \BibitemOpen
  \bibfield  {author} {\bibinfo {author} {\bibfnamefont {G.}~\bibnamefont
  {Birkhoff}},\ }\href {http://www.jstor.org/stable/1968569} {\bibfield
  {journal} {\bibinfo  {journal} {Annals of Mathematics}\ }\textbf {\bibinfo
  {volume} {38}},\ \bibinfo {pages} {526} (\bibinfo {year} {1937})}\BibitemShut
  {NoStop}%
\bibitem [{\citenamefont {Weinberg}(1995)}]{Weinberg:1995}%
  \BibitemOpen
  \bibfield  {author} {\bibinfo {author} {\bibfnamefont {S.}~\bibnamefont
  {Weinberg}},\ }\bibinfo {title} {Quantum fields and antiparticles},\ in\
  \href@noop {} {\emph {\bibinfo {booktitle} {The Quantum Theory of Fields}}}\
  (\bibinfo  {publisher} {Cambridge University Press},\ \bibinfo {year}
  {1995})\ p.\ \bibinfo {pages} {191–258}\BibitemShut {NoStop}%
\bibitem [{\citenamefont {Fischer}(2005)}]{Fischer:2005}%
  \BibitemOpen
  \bibfield  {author} {\bibinfo {author} {\bibfnamefont {G.}~\bibnamefont
  {Fischer}},\ }\href@noop {} {\emph {\bibinfo {title} {Lineare Algebra}}}\
  (\bibinfo  {publisher} {Friedr. Vieweg \& Sohn Verlag},\ \bibinfo {year}
  {2005})\BibitemShut {NoStop}%
\bibitem [{\citenamefont {Koczor}\ \emph {et~al.}(2023)\citenamefont {Koczor},
  \citenamefont {vom Ende}, \citenamefont {de~Gosson}, \citenamefont {Glaser},\
  and\ \citenamefont {Zeier}}]{KEGGZ:2023}%
  \BibitemOpen
  \bibfield  {author} {\bibinfo {author} {\bibfnamefont {B.}~\bibnamefont
  {Koczor}}, \bibinfo {author} {\bibfnamefont {F.}~\bibnamefont {vom Ende}},
  \bibinfo {author} {\bibfnamefont {M.}~\bibnamefont {de~Gosson}}, \bibinfo
  {author} {\bibfnamefont {S.}~\bibnamefont {Glaser}},\ and\ \bibinfo {author}
  {\bibfnamefont {R.}~\bibnamefont {Zeier}},\ }\href
  {https://doi.org/10.1007/s00023-023-01338-6} {\bibfield  {journal} {\bibinfo
  {journal} {Annales Henri Poincaré}\ }\textbf {\bibinfo {volume} {24}},\
  \bibinfo {pages} {4169} (\bibinfo {year} {2023})}\BibitemShut {NoStop}%
\bibitem [{\citenamefont {Wei\ss{}l}\ \emph {et~al.}(2015)\citenamefont
  {Wei\ss{}l}, \citenamefont {K\"ung}, \citenamefont {Dumur}, \citenamefont
  {Feofanov}, \citenamefont {Matei}, \citenamefont {Naud}, \citenamefont
  {Buisson}, \citenamefont {Hekking},\ and\ \citenamefont
  {Guichard}}]{Weiss:2015}%
  \BibitemOpen
  \bibfield  {author} {\bibinfo {author} {\bibfnamefont {T.}~\bibnamefont
  {Wei\ss{}l}}, \bibinfo {author} {\bibfnamefont {B.}~\bibnamefont {K\"ung}},
  \bibinfo {author} {\bibfnamefont {E.}~\bibnamefont {Dumur}}, \bibinfo
  {author} {\bibfnamefont {A.~K.}\ \bibnamefont {Feofanov}}, \bibinfo {author}
  {\bibfnamefont {I.}~\bibnamefont {Matei}}, \bibinfo {author} {\bibfnamefont
  {C.}~\bibnamefont {Naud}}, \bibinfo {author} {\bibfnamefont {O.}~\bibnamefont
  {Buisson}}, \bibinfo {author} {\bibfnamefont {F.~W.~J.}\ \bibnamefont
  {Hekking}},\ and\ \bibinfo {author} {\bibfnamefont {W.}~\bibnamefont
  {Guichard}},\ }\href {https://doi.org/10.1103/PhysRevB.92.104508} {\bibfield
  {journal} {\bibinfo  {journal} {Physical Review B}\ }\textbf {\bibinfo
  {volume} {92}},\ \bibinfo {pages} {104508} (\bibinfo {year}
  {2015})}\BibitemShut {NoStop}%
\bibitem [{\citenamefont {Larson}\ and\ \citenamefont
  {Mavrogordatos}(2021)}]{Larson:2021}%
  \BibitemOpen
  \bibfield  {author} {\bibinfo {author} {\bibfnamefont {J.}~\bibnamefont
  {Larson}}\ and\ \bibinfo {author} {\bibfnamefont {T.}~\bibnamefont
  {Mavrogordatos}},\ }\href {https://doi.org/10.1088/978-0-7503-3447-1} {\emph
  {\bibinfo {title} {The Jaynes–Cummings Model and Its Descendants}}},\
  2053-2563\ (\bibinfo  {publisher} {IOP Publishing},\ \bibinfo {year}
  {2021})\BibitemShut {NoStop}%
\bibitem [{\citenamefont {Gong}\ \emph {et~al.}(2009)\citenamefont {Gong},
  \citenamefont {Ian}, \citenamefont {Liu}, \citenamefont {Sun},\ and\
  \citenamefont {Nori}}]{Gong:2009}%
  \BibitemOpen
  \bibfield  {author} {\bibinfo {author} {\bibfnamefont {Z.~R.}\ \bibnamefont
  {Gong}}, \bibinfo {author} {\bibfnamefont {H.}~\bibnamefont {Ian}}, \bibinfo
  {author} {\bibfnamefont {Y.-x.}\ \bibnamefont {Liu}}, \bibinfo {author}
  {\bibfnamefont {C.~P.}\ \bibnamefont {Sun}},\ and\ \bibinfo {author}
  {\bibfnamefont {F.}~\bibnamefont {Nori}},\ }\href
  {https://doi.org/10.1103/PhysRevA.80.065801} {\bibfield  {journal} {\bibinfo
  {journal} {Physical Review A}\ }\textbf {\bibinfo {volume} {80}},\ \bibinfo
  {pages} {065801} (\bibinfo {year} {2009})}\BibitemShut {NoStop}%
\bibitem [{\citenamefont {Kac}(1990)}]{Kac:1990}%
  \BibitemOpen
  \bibfield  {author} {\bibinfo {author} {\bibfnamefont {V.~G.}\ \bibnamefont
  {Kac}},\ }\bibinfo {title} {The invariant bilinear form and the generalized
  casimir operator},\ in\ \href@noop {} {\emph {\bibinfo {booktitle}
  {Infinite-Dimensional Lie Algebras}}}\ (\bibinfo  {publisher} {Cambridge
  University Press},\ \bibinfo {year} {1990})\ p.\ \bibinfo {pages}
  {16–29}\BibitemShut {NoStop}%
\bibitem [{\citenamefont {Yang}(1951)}]{Yang:1951}%
  \BibitemOpen
  \bibfield  {author} {\bibinfo {author} {\bibfnamefont {L.~M.}\ \bibnamefont
  {Yang}},\ }\href {https://doi.org/10.1103/PhysRev.84.788} {\bibfield
  {journal} {\bibinfo  {journal} {Physical Review}\ }\textbf {\bibinfo {volume}
  {84}},\ \bibinfo {pages} {788} (\bibinfo {year} {1951})}\BibitemShut
  {NoStop}%
\bibitem [{\citenamefont {Wigner}(1950)}]{Wigner:1950}%
  \BibitemOpen
  \bibfield  {author} {\bibinfo {author} {\bibfnamefont {E.~P.}\ \bibnamefont
  {Wigner}},\ }\href {https://doi.org/10.1103/PhysRev.77.711} {\bibfield
  {journal} {\bibinfo  {journal} {Physical Review.}\ }\textbf {\bibinfo
  {volume} {77}},\ \bibinfo {pages} {711} (\bibinfo {year} {1950})}\BibitemShut
  {NoStop}%
\bibitem [{\citenamefont {Chung}\ \emph {et~al.}(2023)\citenamefont {Chung},
  \citenamefont {Schulze-Halberg},\ and\ \citenamefont
  {Hassanabadi}}]{Chung:2023}%
  \BibitemOpen
  \bibfield  {author} {\bibinfo {author} {\bibfnamefont {W.~S.}\ \bibnamefont
  {Chung}}, \bibinfo {author} {\bibfnamefont {A.}~\bibnamefont
  {Schulze-Halberg}},\ and\ \bibinfo {author} {\bibfnamefont {H.}~\bibnamefont
  {Hassanabadi}},\ }\href {https://doi.org/10.1140/epjp/s13360-023-03703-0}
  {\bibfield  {journal} {\bibinfo  {journal} {The European Physical Journal
  Plus}\ }\textbf {\bibinfo {volume} {138}},\ \bibinfo {pages} {66} (\bibinfo
  {year} {2023})}\BibitemShut {NoStop}%
\bibitem [{\citenamefont {Helgason}(2024)}]{helgason2024differential}%
  \BibitemOpen
  \bibfield  {author} {\bibinfo {author} {\bibfnamefont {S.}~\bibnamefont
  {Helgason}},\ }\href {https://books.google.de/books?id=_N48EQAAQBAJ} {\emph
  {\bibinfo {title} {Differential Geometry, Lie Groups, and Symmetric
  Spaces}}},\ Graduate Studies in Mathematics\ (\bibinfo  {publisher} {American
  Mathematical Society},\ \bibinfo {year} {2024})\BibitemShut {NoStop}%
\bibitem [{\citenamefont {Knapp}(1986)}]{Knapp:1986}%
  \BibitemOpen
  \bibfield  {author} {\bibinfo {author} {\bibfnamefont {A.~W.}\ \bibnamefont
  {Knapp}},\ }\href {http://www.jstor.org/stable/j.ctt1bpm9sn} {\emph {\bibinfo
  {title} {Representation Theory of Semisimple Groups: An Overview Based on
  Examples (PMS-36)}}},\ \bibinfo {edition} {rev - revised}\ ed.\ (\bibinfo
  {publisher} {Princeton University Press},\ \bibinfo {year}
  {1986})\BibitemShut {NoStop}%
\bibitem [{\citenamefont {Sattinger}\ and\ \citenamefont
  {Weaver}(1986)}]{Sattinger:1986}%
  \BibitemOpen
  \bibfield  {author} {\bibinfo {author} {\bibfnamefont {D.~H.}\ \bibnamefont
  {Sattinger}}\ and\ \bibinfo {author} {\bibfnamefont {O.~L.}\ \bibnamefont
  {Weaver}},\ }\href
  {https://doi.org/https://doi.org/10.1007/978-1-4757-1910-9} {\emph {\bibinfo
  {title} {Lie Groups and Algebras with Applications to Physics, Geometry, and
  Mechanics}}},\ \bibinfo {edition} {1st}\ ed.,\ Applied Mathematical Sciences\
  (\bibinfo  {publisher} {Springer New York, NY},\ \bibinfo {year}
  {1986})\BibitemShut {NoStop}%
\bibitem [{\citenamefont {Ripka}\ \emph {et~al.}(1986)\citenamefont {Ripka},
  \citenamefont {Blaizot},\ and\ \citenamefont {Ripka}}]{Ripka:1986}%
  \BibitemOpen
  \bibfield  {author} {\bibinfo {author} {\bibfnamefont {S.}~\bibnamefont
  {Ripka}}, \bibinfo {author} {\bibfnamefont {J.}~\bibnamefont {Blaizot}},\
  and\ \bibinfo {author} {\bibfnamefont {G.}~\bibnamefont {Ripka}},\
  }\href@noop {} {\emph {\bibinfo {title} {Quantum Theory of Finite Systems}}}\
  (\bibinfo  {publisher} {MIT Press},\ \bibinfo {year} {1986})\BibitemShut
  {NoStop}%
\bibitem [{\citenamefont {Kazi}\ \emph {et~al.}(2025)\citenamefont {Kazi},
  \citenamefont {Larocca}, \citenamefont {Farinati}, \citenamefont {Coles},
  \citenamefont {Cerezo},\ and\ \citenamefont {Zeier}}]{Kazi:2025}%
  \BibitemOpen
  \bibfield  {author} {\bibinfo {author} {\bibfnamefont {S.}~\bibnamefont
  {Kazi}}, \bibinfo {author} {\bibfnamefont {M.}~\bibnamefont {Larocca}},
  \bibinfo {author} {\bibfnamefont {M.}~\bibnamefont {Farinati}}, \bibinfo
  {author} {\bibfnamefont {P.~J.}\ \bibnamefont {Coles}}, \bibinfo {author}
  {\bibfnamefont {M.}~\bibnamefont {Cerezo}},\ and\ \bibinfo {author}
  {\bibfnamefont {R.}~\bibnamefont {Zeier}},\ }\href@noop {} {\bibinfo {title}
  {{Analyzing the quantum approximate optimization algorithm: ans\"atze,
  symmetries, and Lie algebras}}} (\bibinfo {year} {2025}),\ \Eprint
  {https://arxiv.org/abs/2410.05187} {arXiv:2410.05187 [quant-ph]} \BibitemShut
  {NoStop}%
\bibitem [{\citenamefont {Yang}\ and\ \citenamefont {Tang}(2018)}]{Yang:2018}%
  \BibitemOpen
  \bibfield  {author} {\bibinfo {author} {\bibfnamefont {Y.}~\bibnamefont
  {Yang}}\ and\ \bibinfo {author} {\bibfnamefont {X.}~\bibnamefont {Tang}},\
  }\href {https://doi.org/https://doi.org/10.1007/s12190-017-1157-5} {\bibfield
   {journal} {\bibinfo  {journal} {Journal of Applied Mathematics and
  Computing}\ }\textbf {\bibinfo {volume} {58}},\ \bibinfo {pages} {567 }
  (\bibinfo {year} {2018})}\BibitemShut {NoStop}%
\bibitem [{\citenamefont {Dobrev}\ \emph {et~al.}(1997)\citenamefont {Dobrev},
  \citenamefont {Doebner},\ and\ \citenamefont {Mrugalla}}]{Dobrev:1997}%
  \BibitemOpen
  \bibfield  {author} {\bibinfo {author} {\bibfnamefont {V.}~\bibnamefont
  {Dobrev}}, \bibinfo {author} {\bibfnamefont {H.-D.}\ \bibnamefont
  {Doebner}},\ and\ \bibinfo {author} {\bibfnamefont {C.}~\bibnamefont
  {Mrugalla}},\ }\href
  {https://doi.org/https://doi.org/10.1016/S0034-4877(97)88001-9} {\bibfield
  {journal} {\bibinfo  {journal} {Reports on Mathematical Physics}\ }\textbf
  {\bibinfo {volume} {39}},\ \bibinfo {pages} {201} (\bibinfo {year}
  {1997})}\BibitemShut {NoStop}%
\bibitem [{\citenamefont {Nikitin}(2020)}]{Nikitin:2020}%
  \BibitemOpen
  \bibfield  {author} {\bibinfo {author} {\bibfnamefont {A.~G.}\ \bibnamefont
  {Nikitin}},\ }\href {https://doi.org/10.1088/1751-8121/abb956} {\bibfield
  {journal} {\bibinfo  {journal} {Journal of Physics A: Mathematical and
  Theoretical}\ }\textbf {\bibinfo {volume} {53}},\ \bibinfo {pages} {455202}
  (\bibinfo {year} {2020})}\BibitemShut {NoStop}%
\bibitem [{\citenamefont {Duval}\ and\ \citenamefont
  {Horváthy}(1994)}]{Duval:1994}%
  \BibitemOpen
  \bibfield  {author} {\bibinfo {author} {\bibfnamefont {C.}~\bibnamefont
  {Duval}}\ and\ \bibinfo {author} {\bibfnamefont {P.~A.}\ \bibnamefont
  {Horváthy}},\ }\href {https://doi.org/10.1063/1.530521} {\bibfield
  {journal} {\bibinfo  {journal} {Journal of Mathematical Physics}\ }\textbf
  {\bibinfo {volume} {35}},\ \bibinfo {pages} {2516–2538} (\bibinfo {year}
  {1994})}\BibitemShut {NoStop}%
\bibitem [{\citenamefont {Aizawa}(2011)}]{Aizawa:2011}%
  \BibitemOpen
  \bibfield  {author} {\bibinfo {author} {\bibfnamefont {N.}~\bibnamefont
  {Aizawa}},\ }\href {https://doi.org/10.1088/1742-6596/284/1/012007}
  {\bibfield  {journal} {\bibinfo  {journal} {Journal of Physics: Conference
  Series}\ }\textbf {\bibinfo {volume} {284}},\ \bibinfo {pages} {012007}
  (\bibinfo {year} {2011})}\BibitemShut {NoStop}%
\bibitem [{\citenamefont {Tao}(2022)}]{Tao:2022}%
  \BibitemOpen
  \bibfield  {author} {\bibinfo {author} {\bibfnamefont {W.-Q.}\ \bibnamefont
  {Tao}},\ }\href {https://doi.org/https://doi.org/10.1016/j.jpaa.2021.106944}
  {\bibfield  {journal} {\bibinfo  {journal} {Journal of Pure and Applied
  Algebra}\ }\textbf {\bibinfo {volume} {226}},\ \bibinfo {pages} {106944}
  (\bibinfo {year} {2022})}\BibitemShut {NoStop}%
\bibitem [{\citenamefont {Andrada}\ \emph {et~al.}(2005)\citenamefont
  {Andrada}, \citenamefont {Barberis}, \citenamefont {Dotti},\ and\
  \citenamefont {Ovando}}]{Andrada:2005}%
  \BibitemOpen
  \bibfield  {author} {\bibinfo {author} {\bibfnamefont {A.}~\bibnamefont
  {Andrada}}, \bibinfo {author} {\bibfnamefont {M.~L.}\ \bibnamefont
  {Barberis}}, \bibinfo {author} {\bibfnamefont {I.~G.}\ \bibnamefont
  {Dotti}},\ and\ \bibinfo {author} {\bibfnamefont {G.~P.}\ \bibnamefont
  {Ovando}},\ }\href@noop {} {\bibfield  {journal} {\bibinfo  {journal}
  {Homology, Homotopy and Applications}\ }\textbf {\bibinfo {volume} {7}},\
  \bibinfo {pages} {9 } (\bibinfo {year} {2005})}\BibitemShut {NoStop}%
\bibitem [{\citenamefont {Nesterenko}\ \emph {et~al.}(2016)\citenamefont
  {Nesterenko}, \citenamefont {Pošta},\ and\ \citenamefont
  {Vaneeva}}]{Nesterenko:2016}%
  \BibitemOpen
  \bibfield  {author} {\bibinfo {author} {\bibfnamefont {M.}~\bibnamefont
  {Nesterenko}}, \bibinfo {author} {\bibfnamefont {S.}~\bibnamefont {Pošta}},\
  and\ \bibinfo {author} {\bibfnamefont {O.}~\bibnamefont {Vaneeva}},\ }\href
  {https://doi.org/10.1088/1751-8113/49/11/115203} {\bibfield  {journal}
  {\bibinfo  {journal} {Journal of Physics A: Mathematical and Theoretical}\
  }\textbf {\bibinfo {volume} {49}},\ \bibinfo {pages} {115203} (\bibinfo
  {year} {2016})}\BibitemShut {NoStop}%
\bibitem [{\citenamefont {Mubarakzyanov}(1963)}]{Mubarakzyanov:1963}%
  \BibitemOpen
  \bibfield  {author} {\bibinfo {author} {\bibfnamefont {G.~M.}\ \bibnamefont
  {Mubarakzyanov}},\ }\href@noop {} {\bibfield  {journal} {\bibinfo  {journal}
  {Izvestiya Vysshikh Uchebnykh Zavedenii. Matematika}\ ,\ \bibinfo {pages}
  {114}} (\bibinfo {year} {1963})}\BibitemShut {NoStop}%
\bibitem [{\citenamefont {de~Graaf}(2004)}]{DeGraaf:2004}%
  \BibitemOpen
  \bibfield  {author} {\bibinfo {author} {\bibfnamefont {W.~A.}\ \bibnamefont
  {de~Graaf}},\ }\href@noop {} {\bibinfo {title} {{Classification of solvable
  Lie algebras}}} (\bibinfo {year} {2004}),\ \Eprint
  {https://arxiv.org/abs/math/0404071} {arXiv:math/0404071 [math.RA]}
  \BibitemShut {NoStop}%
\bibitem [{\citenamefont {Vinet}\ and\ \citenamefont
  {Zhedanov}(2011)}]{Vinet:2011}%
  \BibitemOpen
  \bibfield  {author} {\bibinfo {author} {\bibfnamefont {L.}~\bibnamefont
  {Vinet}}\ and\ \bibinfo {author} {\bibfnamefont {A.}~\bibnamefont
  {Zhedanov}},\ }\href@noop {} {\bibinfo {title} {{Representations of the
  Schr\"odinger group and matrix orthogonal polynomials}}} (\bibinfo {year}
  {2011}),\ \Eprint {https://arxiv.org/abs/1105.0701} {arXiv:1105.0701
  [math-ph]} \BibitemShut {NoStop}%
\bibitem [{\citenamefont {Fuchs}\ and\ \citenamefont
  {Schweigert}(2003)}]{fuchs2003symmetries}%
  \BibitemOpen
  \bibfield  {author} {\bibinfo {author} {\bibfnamefont {J.}~\bibnamefont
  {Fuchs}}\ and\ \bibinfo {author} {\bibfnamefont {C.}~\bibnamefont
  {Schweigert}},\ }\href {https://books.google.de/books?id=B_JQryjNYyAC} {\emph
  {\bibinfo {title} {Symmetries, Lie Algebras and Representations: A Graduate
  Course for Physicists}}},\ Cambridge Monographs on Mathematical Physics\
  (\bibinfo  {publisher} {Cambridge University Press},\ \bibinfo {year}
  {2003})\BibitemShut {NoStop}%
\bibitem [{\citenamefont {Barati}\ \emph {et~al.}(2020)\citenamefont {Barati},
  \citenamefont {Saeedi},\ and\ \citenamefont {Alemi}}]{Barati:2020}%
  \BibitemOpen
  \bibfield  {author} {\bibinfo {author} {\bibfnamefont {M.~H.}\ \bibnamefont
  {Barati}}, \bibinfo {author} {\bibfnamefont {F.}~\bibnamefont {Saeedi}},\
  and\ \bibinfo {author} {\bibfnamefont {M.~R.}\ \bibnamefont {Alemi}},\ }\href
  {https://api.semanticscholar.org/CorpusID:186226905} {\bibfield  {journal}
  {\bibinfo  {journal} {Rendiconti del Circolo Matematico di Palermo Series 2}\
  }\textbf {\bibinfo {volume} {69}},\ \bibinfo {pages} {653} (\bibinfo {year}
  {2020})}\BibitemShut {NoStop}%
\bibitem [{\citenamefont {Woit}(2017{\natexlab{b}})}]{Woit:2017:Semidirect}%
  \BibitemOpen
  \bibfield  {author} {\bibinfo {author} {\bibfnamefont {P.}~\bibnamefont
  {Woit}},\ }\bibinfo {title} {Representations of semi-direct products},\ in\
  \href {https://doi.org/10.1007/978-3-319-64612-1_20} {\emph {\bibinfo
  {booktitle} {Quantum Theory, Groups and Representations: An Introduction}}}\
  (\bibinfo  {publisher} {Springer International Publishing},\ \bibinfo
  {address} {Cham},\ \bibinfo {year} {2017})\ pp.\ \bibinfo {pages}
  {259--273}\BibitemShut {NoStop}%
\bibitem [{\citenamefont {Amitsur}(1958)}]{Amitsur1958CommutativeLD}%
  \BibitemOpen
  \bibfield  {author} {\bibinfo {author} {\bibfnamefont {S.~A.}\ \bibnamefont
  {Amitsur}},\ }\href {https://api.semanticscholar.org/CorpusID:59505007}
  {\bibfield  {journal} {\bibinfo  {journal} {Pacific Journal of Mathematics}\
  }\textbf {\bibinfo {volume} {8}},\ \bibinfo {pages} {1} (\bibinfo {year}
  {1958})}\BibitemShut {NoStop}%
\bibitem [{\citenamefont {Hall}(2016)}]{Hall:Lie:groups:16}%
  \BibitemOpen
  \bibfield  {author} {\bibinfo {author} {\bibfnamefont {B.~C.}\ \bibnamefont
  {Hall}},\ }\href {https://doi.org/https://doi.org/10.1007/978-3-319-13467-3}
  {\emph {\bibinfo {title} {Lie Groups, Lie Algebras, and Representations}}},\
  \bibinfo {edition} {2nd}\ ed.\ (\bibinfo  {publisher} {Springer Cham},\
  \bibinfo {year} {2016})\BibitemShut {NoStop}%
\bibitem [{\citenamefont {Jacobson}(1979)}]{Jacobson:1979}%
  \BibitemOpen
  \bibfield  {author} {\bibinfo {author} {\bibfnamefont {N.}~\bibnamefont
  {Jacobson}},\ }\href {https://books.google.de/books?id=hPE1Mmm7SFMC} {\emph
  {\bibinfo {title} {Lie Algebras}}},\ Dover books on advanced mathematics\
  (\bibinfo  {publisher} {Dover},\ \bibinfo {year} {1979})\BibitemShut
  {NoStop}%
\bibitem [{\citenamefont {Klep}\ and\ \citenamefont
  {Moravec}(2010)}]{Klep:2010}%
  \BibitemOpen
  \bibfield  {author} {\bibinfo {author} {\bibfnamefont {I.}~\bibnamefont
  {Klep}}\ and\ \bibinfo {author} {\bibfnamefont {P.}~\bibnamefont {Moravec}},\
  }\href {https://api.semanticscholar.org/CorpusID:120651221} {\bibfield
  {journal} {\bibinfo  {journal} {Algebra Colloquium}\ }\textbf {\bibinfo
  {volume} {17}},\ \bibinfo {pages} {629} (\bibinfo {year} {2010})}\BibitemShut
  {NoStop}%
\bibitem [{\citenamefont {Simoni}\ and\ \citenamefont
  {Zaccaria}(1969)}]{simoni:1969}%
  \BibitemOpen
  \bibfield  {author} {\bibinfo {author} {\bibfnamefont {A.}~\bibnamefont
  {Simoni}}\ and\ \bibinfo {author} {\bibfnamefont {F.}~\bibnamefont
  {Zaccaria}},\ }\href {https://doi.org/https://doi.org/10.1007/BF02754988}
  {\bibfield  {journal} {\bibinfo  {journal} {Il Nuovo Cimento A (1965-1970)}\
  }\textbf {\bibinfo {volume} {59}},\ \bibinfo {pages} {280} (\bibinfo {year}
  {1969})}\BibitemShut {NoStop}%
\bibitem [{\citenamefont {Knebelman}(1935)}]{Knebelman:1935}%
  \BibitemOpen
  \bibfield  {author} {\bibinfo {author} {\bibfnamefont {M.~S.}\ \bibnamefont
  {Knebelman}},\ }\href {http://www.jstor.org/stable/1968663} {\bibfield
  {journal} {\bibinfo  {journal} {Annals of Mathematics}\ }\textbf {\bibinfo
  {volume} {36}},\ \bibinfo {pages} {46} (\bibinfo {year} {1935})}\BibitemShut
  {NoStop}%
\bibitem [{\citenamefont {Albuquerque}\ and\ \citenamefont
  {Elduque}(1992)}]{Alburquerque:1992}%
  \BibitemOpen
  \bibfield  {author} {\bibinfo {author} {\bibfnamefont {H.}~\bibnamefont
  {Albuquerque}}\ and\ \bibinfo {author} {\bibfnamefont {A.}~\bibnamefont
  {Elduque}},\ }\href
  {https://doi.org/https://doi.org/10.1016/0024-3795(92)90277-H} {\bibfield
  {journal} {\bibinfo  {journal} {Linear Algebra and its Applications}\
  }\textbf {\bibinfo {volume} {166}},\ \bibinfo {pages} {195} (\bibinfo {year}
  {1992})}\BibitemShut {NoStop}%
\bibitem [{\citenamefont {Kuranishi}(1951)}]{Kuranishi:1951}%
  \BibitemOpen
  \bibfield  {author} {\bibinfo {author} {\bibfnamefont {M.}~\bibnamefont
  {Kuranishi}},\ }\href {https://doi.org/10.1017/S0027763000010059} {\bibfield
  {journal} {\bibinfo  {journal} {Nagoya Mathematical Journal}\ }\textbf
  {\bibinfo {volume} {2}},\ \bibinfo {pages} {63–71} (\bibinfo {year}
  {1951})}\BibitemShut {NoStop}%
\bibitem [{\citenamefont {Macon}\ and\ \citenamefont
  {Spitzbart}(1958)}]{Macon:1958}%
  \BibitemOpen
  \bibfield  {author} {\bibinfo {author} {\bibfnamefont {N.}~\bibnamefont
  {Macon}}\ and\ \bibinfo {author} {\bibfnamefont {A.}~\bibnamefont
  {Spitzbart}},\ }\href {http://www.jstor.org/stable/2308881} {\bibfield
  {journal} {\bibinfo  {journal} {The American Mathematical Monthly}\ }\textbf
  {\bibinfo {volume} {65}},\ \bibinfo {pages} {95} (\bibinfo {year}
  {1958})}\BibitemShut {NoStop}%
\bibitem [{\citenamefont {Sat{\^o}}(1974)}]{Sato:1974}%
  \BibitemOpen
  \bibfield  {author} {\bibinfo {author} {\bibfnamefont {T.}~\bibnamefont
  {Sat{\^o}}},\ }\href {https://api.semanticscholar.org/CorpusID:118594128}
  {\bibfield  {journal} {\bibinfo  {journal} {Hiroshima Mathematical Journal}\
  }\textbf {\bibinfo {volume} {4}},\ \bibinfo {pages} {29} (\bibinfo {year}
  {1974})}\BibitemShut {NoStop}%
\bibitem [{\citenamefont {Ball}\ \emph {et~al.}(1982)\citenamefont {Ball},
  \citenamefont {Marsden},\ and\ \citenamefont {Slemrod}}]{Ball1982}%
  \BibitemOpen
  \bibfield  {author} {\bibinfo {author} {\bibfnamefont {J.~M.}\ \bibnamefont
  {Ball}}, \bibinfo {author} {\bibfnamefont {J.~E.}\ \bibnamefont {Marsden}},\
  and\ \bibinfo {author} {\bibfnamefont {M.}~\bibnamefont {Slemrod}},\ }\href
  {https://doi.org/10.1137/0320042} {\bibfield  {journal} {\bibinfo  {journal}
  {SIAM Journal on Control and Optimization}\ }\textbf {\bibinfo {volume}
  {20}},\ \bibinfo {pages} {575} (\bibinfo {year} {1982})}\BibitemShut
  {NoStop}%
\bibitem [{\citenamefont {Boussaïd}\ \emph {et~al.}(2019)\citenamefont
  {Boussaïd}, \citenamefont {Caponigro},\ and\ \citenamefont
  {Chambrion}}]{BCC:2019}%
  \BibitemOpen
  \bibfield  {author} {\bibinfo {author} {\bibfnamefont {N.}~\bibnamefont
  {Boussaïd}}, \bibinfo {author} {\bibfnamefont {M.}~\bibnamefont
  {Caponigro}},\ and\ \bibinfo {author} {\bibfnamefont {T.}~\bibnamefont
  {Chambrion}},\ }in\ \href {https://doi.org/10.1109/CDC40024.2019.9029511}
  {\emph {\bibinfo {booktitle} {IEEE 58th Conference on Decision and
  Control}}}\ (\bibinfo {year} {2019})\ pp.\ \bibinfo {pages}
  {4971--4976}\BibitemShut {NoStop}%
\bibitem [{\citenamefont {Dirr}(2022)}]{Dirr2022}%
  \BibitemOpen
  \bibfield  {author} {\bibinfo {author} {\bibfnamefont {G.}~\bibnamefont
  {Dirr}},\ }\href {https://doi.org/10.1016/j.ifacol.2022.11.063} {\bibfield
  {journal} {\bibinfo  {journal} {IFAC-PapersOnLine}\ }\textbf {\bibinfo
  {volume} {55}},\ \bibinfo {pages} {266} (\bibinfo {year} {2022})},\ \bibinfo
  {note} {25th International Symposium on Mathematical Theory of Networks and
  Systems MTNS 2022}\BibitemShut {NoStop}%
\bibitem [{\citenamefont {Keyl}\ \emph {et~al.}(2014)\citenamefont {Keyl},
  \citenamefont {Zeier},\ and\ \citenamefont
  {Schulte-Herbr{\"u}ggen}}]{KZSH:2014}%
  \BibitemOpen
  \bibfield  {author} {\bibinfo {author} {\bibfnamefont {M.}~\bibnamefont
  {Keyl}}, \bibinfo {author} {\bibfnamefont {R.}~\bibnamefont {Zeier}},\ and\
  \bibinfo {author} {\bibfnamefont {T.}~\bibnamefont
  {Schulte-Herbr{\"u}ggen}},\ }\href
  {https://doi.org/10.1088/1367-2630/16/6/065010} {\bibfield  {journal}
  {\bibinfo  {journal} {New J. Phys.}\ }\textbf {\bibinfo {volume} {16}},\
  \bibinfo {pages} {065010} (\bibinfo {year} {2014})}\BibitemShut {NoStop}%
\bibitem [{\citenamefont {Keyl}(2019)}]{Keyl:2019}%
  \BibitemOpen
  \bibfield  {author} {\bibinfo {author} {\bibfnamefont {M.}~\bibnamefont
  {Keyl}},\ }in\ \href {https://doi.org/10.1109/cdc40024.2019.9029317} {\emph
  {\bibinfo {booktitle} {IEEE 58th Conference on Decision and Control}}}\
  (\bibinfo {year} {2019})\ pp.\ \bibinfo {pages} {2298--2303}\BibitemShut
  {NoStop}%
\bibitem [{\citenamefont {Arenz}\ \emph {et~al.}(2018)\citenamefont {Arenz},
  \citenamefont {Burgarth}, \citenamefont {Facchi},\ and\ \citenamefont
  {Hiller}}]{ABFH:2018}%
  \BibitemOpen
  \bibfield  {author} {\bibinfo {author} {\bibfnamefont {C.}~\bibnamefont
  {Arenz}}, \bibinfo {author} {\bibfnamefont {D.}~\bibnamefont {Burgarth}},
  \bibinfo {author} {\bibfnamefont {P.}~\bibnamefont {Facchi}},\ and\ \bibinfo
  {author} {\bibfnamefont {R.}~\bibnamefont {Hiller}},\ }\href
  {https://doi.org/10.1063/1.5016495} {\bibfield  {journal} {\bibinfo
  {journal} {Journal of Mathematical Physics}\ }\textbf {\bibinfo {volume}
  {59}},\ \bibinfo {pages} {032203} (\bibinfo {year} {2018})}\BibitemShut
  {NoStop}%
\bibitem [{\citenamefont {Bloch}(2015)}]{Bloch2015}%
  \BibitemOpen
  \bibfield  {author} {\bibinfo {author} {\bibfnamefont {A.~M.}\ \bibnamefont
  {Bloch}},\ }\href {https://doi.org/10.1007/978-1-4939-3017-3} {\emph
  {\bibinfo {title} {{Nonholonomic Mechanics and Control}}}},\ \bibinfo
  {edition} {2nd}\ ed.\ (\bibinfo  {publisher} {Springer, New York},\ \bibinfo
  {year} {2015})\BibitemShut {NoStop}%
\bibitem [{\citenamefont {Messiah}(1999)}]{Messiah:1999}%
  \BibitemOpen
  \bibfield  {author} {\bibinfo {author} {\bibfnamefont {A.}~\bibnamefont
  {Messiah}},\ }\href@noop {} {\emph {\bibinfo {title} {Quantum Mechanics}}}\
  (\bibinfo  {publisher} {Dover},\ \bibinfo {address} {Mineola},\ \bibinfo
  {year} {1999})\BibitemShut {NoStop}%
\bibitem [{\citenamefont {Meise}\ and\ \citenamefont {Vogt}(1997)}]{meise1997}%
  \BibitemOpen
  \bibfield  {author} {\bibinfo {author} {\bibfnamefont {R.}~\bibnamefont
  {Meise}}\ and\ \bibinfo {author} {\bibfnamefont {D.}~\bibnamefont {Vogt}},\
  }\href@noop {} {\emph {\bibinfo {title} {Introduction to Functional
  Analysis}}}\ (\bibinfo  {publisher} {Oxford University Press},\ \bibinfo
  {address} {Oxford},\ \bibinfo {year} {1997})\BibitemShut {NoStop}%
\bibitem [{\citenamefont {Hall}(2013{\natexlab{b}})}]{hall2013quantum}%
  \BibitemOpen
  \bibfield  {author} {\bibinfo {author} {\bibfnamefont {B.~C.}\ \bibnamefont
  {Hall}},\ }\href {https://doi.org/10.1007/978-1-4614-7116-5} {\emph {\bibinfo
  {title} {Quantum theory for mathematicians}}}\ (\bibinfo  {publisher}
  {Springer, New York},\ \bibinfo {year} {2013})\BibitemShut {NoStop}%
\bibitem [{\citenamefont {Schm{\"u}dgen}(2012)}]{schmuedgen2012}%
  \BibitemOpen
  \bibfield  {author} {\bibinfo {author} {\bibfnamefont {K.}~\bibnamefont
  {Schm{\"u}dgen}},\ }\href {https://doi.org/10.1007/978-94-007-4753-1} {\emph
  {\bibinfo {title} {{Unbounded Self-adjoint Operators on Hilbert Space}}}}\
  (\bibinfo  {publisher} {Springer, Dordrecht},\ \bibinfo {year}
  {2012})\BibitemShut {NoStop}%
\bibitem [{\citenamefont {BOSMA}\ \emph {et~al.}(1997)\citenamefont {BOSMA},
  \citenamefont {CANNON},\ and\ \citenamefont {PLAYOUST}}]{Bosma:MAGMA:1997}%
  \BibitemOpen
  \bibfield  {author} {\bibinfo {author} {\bibfnamefont {W.}~\bibnamefont
  {BOSMA}}, \bibinfo {author} {\bibfnamefont {J.}~\bibnamefont {CANNON}},\ and\
  \bibinfo {author} {\bibfnamefont {C.}~\bibnamefont {PLAYOUST}},\ }\href
  {https://doi.org/https://doi.org/10.1006/jsco.1996.0125} {\bibfield
  {journal} {\bibinfo  {journal} {Journal of Symbolic Computation}\ }\textbf
  {\bibinfo {volume} {24}},\ \bibinfo {pages} {235} (\bibinfo {year}
  {1997})}\BibitemShut {NoStop}%
\bibitem [{\citenamefont {Cox}\ \emph {et~al.}(2005)\citenamefont {Cox},
  \citenamefont {Little},\ and\ \citenamefont {O’shea}}]{Cox:2013}%
  \BibitemOpen
  \bibfield  {author} {\bibinfo {author} {\bibfnamefont {D.~A.}\ \bibnamefont
  {Cox}}, \bibinfo {author} {\bibfnamefont {J.}~\bibnamefont {Little}},\ and\
  \bibinfo {author} {\bibfnamefont {D.}~\bibnamefont {O’shea}},\ }\href
  {https://doi.org/https://doi.org/10.1007/b138611} {\emph {\bibinfo {title}
  {Using Algebraic Geometry}}},\ \bibinfo {edition} {2nd}\ ed.,\ Graduate Texts
  in Mathematics\ (\bibinfo  {publisher} {Springer New York},\ \bibinfo {year}
  {2005})\BibitemShut {NoStop}%
\bibitem [{\citenamefont {Coutinho}(1995)}]{Coutinho:1995}%
  \BibitemOpen
  \bibfield  {author} {\bibinfo {author} {\bibfnamefont {S.~C.}\ \bibnamefont
  {Coutinho}},\ }\href@noop {} {\emph {\bibinfo {title} {A Primer of Algebraic
  D-modules}}}\ (\bibinfo  {publisher} {Cambridge University Press},\ \bibinfo
  {address} {Cambridge},\ \bibinfo {year} {1995})\BibitemShut {NoStop}%
\bibitem [{\citenamefont {Bourbaki}(1974)}]{Bourbaki:algebra:1974}%
  \BibitemOpen
  \bibfield  {author} {\bibinfo {author} {\bibfnamefont {N.}~\bibnamefont
  {Bourbaki}},\ }\href@noop {} {\emph {\bibinfo {title} {Algebra I:}}},\
  Actualit{\'e}s scientifiques et industrielles\ (\bibinfo  {publisher}
  {Addison-Wesley Publishing Company},\ \bibinfo {year} {1974})\BibitemShut
  {NoStop}%
\bibitem [{\citenamefont {Bourbaki}(1971)}]{Bourbaki:Lie:algebra:1971}%
  \BibitemOpen
  \bibfield  {author} {\bibinfo {author} {\bibfnamefont {N.}~\bibnamefont
  {Bourbaki}},\ }\href@noop {} {\emph {\bibinfo {title} {Lie Groups and Lie
  Algebras}}},\ \bibinfo {edition} {1st}\ ed.\ (\bibinfo  {publisher} {Springer
  Berlin, Heidelberg},\ \bibinfo {year} {1971})\BibitemShut {NoStop}%
\bibitem [{\citenamefont {Dieudonné}(1953)}]{Dieudonne:1953}%
  \BibitemOpen
  \bibfield  {author} {\bibinfo {author} {\bibfnamefont {J.}~\bibnamefont
  {Dieudonné}},\ }\href {http://www.jstor.org/stable/2031832} {\bibfield
  {journal} {\bibinfo  {journal} {Proceedings of the American Mathematical
  Society}\ }\textbf {\bibinfo {volume} {4}},\ \bibinfo {pages} {931} (\bibinfo
  {year} {1953})}\BibitemShut {NoStop}%
\bibitem [{\citenamefont {Procesi}(2006)}]{Procesi:2006}%
  \BibitemOpen
  \bibfield  {author} {\bibinfo {author} {\bibfnamefont {C.}~\bibnamefont
  {Procesi}},\ }\href
  {https://doi.org/https://doi.org/10.1007/978-0-387-28929-8} {\emph {\bibinfo
  {title} {Lie Groups}}},\ \bibinfo {edition} {1st}\ ed.\ (\bibinfo
  {publisher} {Springer New York, NY},\ \bibinfo {year} {2006})\BibitemShut
  {NoStop}%
\bibitem [{\citenamefont {Einstein}(1916)}]{Einstein:1916}%
  \BibitemOpen
  \bibfield  {author} {\bibinfo {author} {\bibfnamefont {A.}~\bibnamefont
  {Einstein}},\ }\href
  {https://doi.org/https://doi.org/10.1002/andp.19163540702} {\bibfield
  {journal} {\bibinfo  {journal} {Annalen der Physik}\ }\textbf {\bibinfo
  {volume} {354}},\ \bibinfo {pages} {769} (\bibinfo {year}
  {1916})}\BibitemShut {NoStop}%
\bibitem [{\citenamefont {Fecko}(2006)}]{Fecko:2006}%
  \BibitemOpen
  \bibfield  {author} {\bibinfo {author} {\bibfnamefont {M.}~\bibnamefont
  {Fecko}},\ }\href@noop {} {\emph {\bibinfo {title} {Differential Geometry and
  Lie Groups for Physicists}}}\ (\bibinfo  {publisher} {Cambridge University
  Press},\ \bibinfo {year} {2006})\BibitemShut {NoStop}%
\bibitem [{\citenamefont {Pfeifer}(2006)}]{Pfeifer:2003}%
  \BibitemOpen
  \bibfield  {author} {\bibinfo {author} {\bibfnamefont {W.}~\bibnamefont
  {Pfeifer}},\ }\href
  {https://doi.org/https://doi.org/10.1007/978-3-0348-8097-8} {\emph {\bibinfo
  {title} {The Lie Algebras su(N)}}},\ \bibinfo {edition} {1st}\ ed.\ (\bibinfo
   {publisher} {Birkhäuser Basel},\ \bibinfo {year} {2006})\BibitemShut
  {NoStop}%
\bibitem [{\citenamefont {Nestruev}(2020)}]{Nestruev:2020}%
  \BibitemOpen
  \bibfield  {author} {\bibinfo {author} {\bibfnamefont {J.}~\bibnamefont
  {Nestruev}},\ }\href
  {https://doi.org/https://doi.org/10.1007/978-3-030-45650-4} {\emph {\bibinfo
  {title} {Smooth Manifolds and Observables}}},\ Graduate Texts in Mathematics\
  (\bibinfo  {publisher} {Springer Cham},\ \bibinfo {year} {2020})\BibitemShut
  {NoStop}%
\end{thebibliography}%
			
		\end{document}